%% file: aaagadt.tex
\def\OPTIONSmallFigs{0}
\def\OPTIONArxiv{0}
\input{flags.tex}
\def\OPTIONConf{1}
\ifnum\OPTIONConf=1%
  \ifnum\OPTIONArxiv=0%
     \documentclass[
       acmsmall,
       screen
       acmthm=true%
       ]{acmart}
  \else
     \documentclass[
       acmsmall,
       acmthm=true%
       ]{acmart}
     \input{arxivmagic.tex}
  \fi
     \def~{\nobreakspace{}}   %
\usepackage{etex}

\ifdefined\synctex
\else
\usepackage{srcltx}
\fi

\usepackage{jdunfield}
\usepackage{llproof}
\usepackage{rulelinks}
\usepackage[OT2,T1]{fontenc}

\usepackage{relsize}

\usepackage{amsmath,amsthm,amssymb}
\usepackage{thmtools,thm-restate}

\usepackage{enumerate}

\usepackage{pgf,tikz}
\usetikzlibrary{shapes,arrows,matrix,fit,backgrounds,calc,decorations,decorations.pathmorphing,decorations.markings}

\usepackage{breakurl}
\else
\documentclass[10pt,leqno]{article}
\fi

\ifnum\OPTIONConf=1%
  \setcitestyle{nosort}    %
\else
  \usepackage[authoryear]{natbib}
  \bibpunct{(}{)}{;}{a}{}{,}
\fi

\ifnum\OPTIONConf=1%
\else%
\fi

\usepackage{graphics,pstricks,color}

\usepackage{listings}

\usepackage{pgf,tikz}
\usetikzlibrary{shapes,arrows,matrix,fit,backgrounds,calc,decorations,decorations.pathmorphing}
\usepackage{xspace}
\ifnum\OPTIONConf=1%
\else
    \usepackage{goodcharter}
\fi
\usepackage{multirow}
\usepackage{comment}
\usepackage{mathpartir}
\usepackage{framed}
\usepackage{rotating}

\let\MathRightArrow\Rightarrow %
\usepackage{marvosym} %
\def\Rightarrow{\MathRightArrow}

\usepackage{srcltx}
\usepackage{fancyhdr}
\usepackage{enumerate}
\usepackage{relsize}

\usepackage{booktabs}

\ifnum\OPTIONArxiv=0
  \usepackage{xr}  %
\else
  \usepackage{resilientxr}  %
\fi
\def\MathparLineskip{\lineskiplimit=0.8em\lineskip=0.6em plus 0.1em}  %

\input{local}

\renewcommand{\FLabel}[1]{\label{#1}}

\setcopyright{rightsretained}
\acmPrice{}
\acmDOI{10.1145/3290322}
\acmYear{2019}
\copyrightyear{2019}
\acmJournal{PACMPL}
\acmVolume{3}
\acmNumber{POPL}
\acmArticle{9}
\acmMonth{1}



\begin{document}
\title[Sound and Complete Bidirectional Typechecking]
{%
  Sound and Complete Bidirectional Typechecking
  for Higher-Rank Polymorphism
  with Existentials and Indexed~Types}

\ifnum\OPTIONConf=1     %
  \ifnum\OPTIONArxiv=0%
  \else
   \fi
    \iffalse
    \else
    \author[Jana Dunfield]{Jana Dunfield}
        \orcid{0000-0002-3718-3395}
        \affiliation{%
          \institution{Queen's University}
          \streetaddress{Goodwin Hall 557}
          \city{Kingston, ON}
          \postcode{K7L 3N6}
          \country{Canada}
        }
        \email{jd169@queensu.ca}

    \author[Neel\ Krishnaswami]{Neelakantan R. Krishnaswami}
        \affiliation{%
          \institution{University of Cambridge}
          \streetaddress{Computer Laboratory, William Gates Building}
          \city{Cambridge}
          \postcode{CB3 0FD}
          \country{United Kingdom}}
        \email{nk480@cl.cam.ac.uk}
    \fi%
\fi

\setlength{\pdfpageheight}{\paperheight}
\setlength{\pdfpagewidth}{\paperwidth}

\newcommand{\XXXACM}[1]{\vspace{#1}}  %

\input{abstract}

\maketitle

\setcounter{footnote}{0}

\input{introduction}
\input{overview.tex}

\input{examples.tex}

\input{declarative.tex}

\input{algorithmic.tex}

\input{soundness.tex}

\input{completeness.tex}
\input{related.tex}

\input{acks.tex}
{%
\bibliography{local}
}

\end{document}

\maketitle

\makeatletter
\ifdefined\r@mainfigLAST
\setcounter{figure}{\eatlabel\r@mainfigLAST}
\else
\setcounter{figure}{100}
\fi
\makeatother

The first part (Sections \ref{sec:figures}--\ref{sec:last-section-before-proofs})
of this supplementary material contains rules, figures and definitions omitted
in the main paper for space reasons, and a list of judgment forms (Section \ref{sec:list-of-judgments}).

The remainder (Sections \ref{sec:first-proof}--\ref{sec:last-section})
includes statements of all lemmas and theorems, along with full proofs.
as well as statements of theorems and a few selected lemmas.

\tableofcontents

\clearpage

\def\OPTIONInAppendix{1}

\input{appendix.tex}

\def\OPTIONAppendix{1}
\def\OPTIONSmallFigs{0}

\section{List of Judgments}
\label{sec:list-of-judgments}
\label{sec:last-section-before-proofs}

For convenience, we list all the judgment forms:

\bigskip

\input{fig-judgments.tex}

\clearpage

\appendix

\section{Properties of the Declarative System}
\label{sec:first-proof}

\begin{restatable}[Declarative Well-foundedness]{lemma}{declarativewellfounded}
\XLabel{lem:declarative-well-founded}
 ~\\
 The inductive definition of the following judgments is well-founded:

 \begin{enumerate}[(i)]
   \item %
       synthesis $\declsynjudg{p}{\Psi}{e}{B}$
   \item %
       checking $\declchkjudg{p}{\Psi}{e}{A}$
   \item %
       checking, equality elimination $\declchkjudg{p}{\Psi \ctxsep P}{e}{C}$
   \item %
       ordinary spine $\declspinejudg{\Psi}{s}{A}{p}{B}{q}$
   \item %
       recovery spine $\declrecspinejudg{\Psi}{s}{A}{p}{B}{q}$
   \item %
       pattern matching $\declmatchjudg p \Psi \Pi \Avec C$
   \item %
       pattern matching, equality elimination $\declmatchelimjudg p \Psi \Pi P \Avec C$
 \end{enumerate}
\end{restatable}

\begin{restatable}[Declarative Weakening]{lemma}{declarativeweakening}
\XLabel{lem:declarative-weakening}
~
  \begin{enumerate}[(i)]
  \item If $\judge{\Psi_0, \Psi_1}{t}{\sort}$ then $\judge{\Psi_0, \Psi, \Psi_1}{t}{\sort}$.
  \item If $\judgeprop{\Psi_0, \Psi_1}{P}$ then $\judgeprop{\Psi_0, \Psi, \Psi_1}{P}$.
  \item If $\declcheckprop{\Psi_0, \Psi_1}{P}$ then $\declcheckprop{\Psi_0, \Psi, \Psi_1}{P}$.
  \item If $\judgetp{\Psi_0, \Psi_1}{A}$ then $\judgetp{\Psi_0, \Psi, \Psi_1}{A}$.
  \end{enumerate}
\end{restatable}

\begin{restatable}[Declarative Term Substitution]{lemma}{declarativetermsubstitution}
\XLabel{lem:declarative-term-substitution}

  Suppose $\judge{\Psi}{t}{\sort}$. Then:

  \begin{enumerate}
  \item If $\judge{\Psi_0, \alpha:\sort, \Psi_1}{t'}{\sort}$
    then $\judge{\Psi_0, [t/\alpha]\Psi_1}{[t/\alpha]t'}{\sort}$.
  \item If $\judgeprop{\Psi_0, \alpha:\sort, \Psi_1}{P}$
    then $\judgeprop{\Psi_0, [t/\alpha]\Psi_1}{[t/\alpha]P}$.
  \item If $\judgetp{\Psi_0, \alpha:\sort, \Psi_1}{A}$ then $\judgetp{\Psi_0, [t/\alpha]\Psi_1}{[t/\alpha]A}$.
  \item If $\declsubjudg[\polvar]{\Psi_0, \alpha:\sort, \Psi_1}{A}{B}$ then 
           $\declsubjudg[\polvar]{\Psi_0, [t/\alpha]\Psi_1}{[t/\alpha]A}{[t/\alpha]B}$.
  \item If $\declcheckprop{\Psi_0, \alpha:\sort, \Psi_1}{P}$
          then $\declcheckprop{\Psi_0, [t/\alpha]\Psi_1}{[t/\alpha]P}$.
  \end{enumerate}
\end{restatable}

\begin{restatable}[Reflexivity of Declarative Subtyping]{lemma}{declarativereflexivity}
\XLabel{lem:declarative-reflexivity}
~\\
  Given $\judgetp{\Psi}{A}$, we have that $\declsubjudg[\polvar]{\Psi}{A}{A}$. 

\end{restatable}

\begin{restatable}[Subtyping Inversion]{lemma}{declarativesubtypinginversion}
  \XLabel{lem:declarative-subtyping-inversion}
~
  \begin{itemize}
  \item If $\declsubjudg[+]{\Psi}{\extype{\alpha:\sort}{A}}{B}$ then 
           $\declsubjudg[+]{\Psi, \alpha:\sort}{A}{B}$. 
  \item If $\declsubjudg[-]{\Psi}{A}{\alltype{\beta:\sort}{B}}$ then 
           $\declsubjudg[-]{\Psi, \beta:\sort}{A}{B}$. 
  \end{itemize}
\end{restatable}

\begin{restatable}[Subtyping Polarity Flip]{lemma}{declarativepolarityflip}
  \XLabel{lem:declarative-polarity-flip} 
~
  \begin{itemize}
  \item If $\nonPos{A}$ and $\nonPos{B}$ and 
        $\declsubjudg[+]{\Psi}{A}{B}$
        \\ then 
        $\declsubjudg[-]{\Psi}{A}{B}$ by a derivation of the same or smaller size.
  \item If $\nonNeg{A}$ and $\nonNeg{B}$ 
        and $\declsubjudg[-]{\Psi}{A}{B}$
        \\ then 
        $\declsubjudg[+]{\Psi}{A}{B}$ by a derivation of the same or smaller size.
  \item If $\nonPos{A}$ and $\nonNeg{A}$ 
        and $\nonPos{B}$ and $\nonNeg{B}$ 
        and $\declsubjudg[\polvar]{\Psi}{A}{B}$ 
        \\
        then $A = B$.
  \end{itemize}
\end{restatable}

\begin{restatable}[Transitivity of Declarative Subtyping]{lemma}{declarativetransitivity}
\XLabel{lem:declarative-transitivity}
~\\
  Given $\judgetp{\Psi}{A}$ and $\judgetp{\Psi}{B}$ and $\judgetp{\Psi}{C}$:
  \begin{enumerate}[(i)]
    \item   If
     $\Dee_1 \derives \declsubjudg[\polvar]{\Psi}{A}{B}$ and $\Dee_2 \derives \declsubjudg[\polvar]{\Psi}{B}{C}$
     \\
     then $\declsubjudg[\polvar]{\Psi}{A}{C}$. 

  \end{enumerate}
\end{restatable}

\begin{property}
\Label{prop:exp-no-evars}
   We assume that all types mentioned in annotations in expressions have no free existential variables.
   By the grammar, it follows that all expressions have no free existential variables,
   that is, $\FEV{e} = \emptyset$.
\end{property}

\section{Substitution and Well-formedness Properties}

\begin{definition}[Softness]
\Label{def:soft}
  A context $\Theta$ is \emph{soft} iff it consists only of $\ahat : \sort$ and $\hypeq{\ahat : \sort}{\tau}$
  declarations.
\end{definition}

\begin{restatable}[Substitution---Well-formedness]{lemma}{substitutionwf}
\XLabel{lem:substitution-wf}
~
\begin{enumerate}[(i)]
\item  %
  If $\judgetp{\Gamma}{A\;p}$
  and $\judgetp{\Gamma}{\tau\;p}$
  then $\judgetp{\Gamma}{[\tau/\alpha]A\;p}$.

\item  %
  If $\judgeprop{\Gamma}{P}$
  and $\judgetp{\Gamma}{\tau\;p}$
  then $\judgeprop{\Gamma}{[\tau/\alpha]P}$.
  \\
  Moreover, if $p = \p$ and $\FEV{[\Gamma]P} = \emptyset$
  then $\FEV{[\Gamma][\tau/\alpha]P} = \emptyset$.
\end{enumerate}
\end{restatable}

\begin{restatable}[Uvar Preservation]{lemma}{ctxappuvarpreservation}
\XLabel{lem:ctxapp-uvar-preservation} ~\\
  If $\substextend{\Delta}{\Omega}$ then:

  \begin{enumerate}[(i)]
  \item   If $(\alpha : \sort) \in \Omega$
    then
    $(\alpha : \sort) \in [\Omega]\Delta$.
  \item   If $(\hyp{p}{x}{A}) \in \Omega$ then $(\hyp{p}{x}{[\Omega]A}) \in [\Omega]\Delta$.
  \end{enumerate}
\end{restatable}

\begin{restatable}[Sorting Implies Typing]{lemma}{sortingimpliestyping}
\XLabel{lem:sorting-implies-typing}
  If $\judge{\Gamma}{t}{\type}$
  then
  $\judgetp{\Gamma}{t}$.
\end{restatable}

\begin{restatable}[Right-Hand Substitution for Sorting]{lemma}{substitutionsort}
\XLabel{lem:substitution-sort}
  If $\judge{\Gamma}{t}{\sort}$
  then
  $\judge{\Gamma}{[\Gamma]t}{\sort}$. 
\end{restatable}

\begin{restatable}[Right-Hand Substitution for Propositions]{lemma}{substitutionprop}
\XLabel{lem:substitution-prop}
  If $\judgeprop{\Gamma}{P}$
  then
  $\judgeprop{\Gamma}{[\Gamma]P}$. 
\end{restatable}

\begin{restatable}[Right-Hand Substitution for Typing]{lemma}{substitutiontp}
\XLabel{lem:substitution-tp}
  If $\judgetp{\Gamma}{A}$
  then
  $\judgetp{\Gamma}{[\Gamma]A}$. 
\end{restatable}

\begin{restatable}[Substitution for Sorting]{lemma}{completionsort}
\XLabel{lem:completion-sort}
  If $\judge{\Omega}{t}{\sort}$
  then
  $\judge{[\Omega]\Omega}{[\Omega]t}{\sort}$.
\end{restatable}

\begin{restatable}[Substitution for Prop Well-Formedness]{lemma}{completionprop}
\XLabel{lem:completion-prop}
~\\
  If $\judgeprop{\Omega}{P}$
  then
  $\judgeprop{[\Omega]\Omega}{[\Omega]P}$.
\end{restatable}

\begin{restatable}[Substitution for Type Well-Formedness]{lemma}{completionwf}
\XLabel{lem:completion-wf}
  If $\judgetp{\Omega}{A}$
  then
  $\judgetp{[\Omega]\Omega}{[\Omega]A}$.
\end{restatable}

\begin{restatable}[Substitution Stability]{lemma}{substitutionstability}
\XLabel{lem:substitution-stability}
~\\
  If $(\Omega, \Omega_Z)$ is well-formed and $\Omega_Z$ is soft
  and $\judgetp{\Omega}{A}$
  then $[\Omega]A = [\Omega, \Omega_Z]A$.
\end{restatable}

\begin{restatable}[Equal Domains]{lemma}{equaldomainswftp}
\XLabel{lem:equal-domains-wf-tp}
~\\
  If $\judgetp{\Omega_1}{A}$
  and
  $\dom{\Omega_1} = \dom{\Omega_2}$
  then
  $\judgetp{\Omega_2}{A}$.
\end{restatable}

\section{Properties of Extension}

\begin{restatable}[Declaration Preservation]{lemma}{declarationpreservation}
\XLabel{lem:declaration-preservation}
  If $\substextend{\Gamma}{\Delta}$ 
  and $u$ is declared in $\Gamma$,
  then $u$ is declared in $\Delta$. 
\end{restatable}

\begin{restatable}[Declaration Order Preservation]{lemma}{declarationorderpreservation}
\XLabel{lem:declaration-order-preservation}
  If $\substextend{\Gamma}{\Delta}$ 
  and $u$ is declared to the left of $v$ in $\Gamma$,
  then $u$ is declared to the left of $v$ in $\Delta$. 
\end{restatable}

\begin{restatable}[Reverse Declaration Order Preservation]{lemma}{reversedeclarationorderpreservation}
\XLabel{lem:reverse-declaration-order-preservation}
  If $\substextend{\Gamma}{\Delta}$ and $u$ and $v$ are both declared in $\Gamma$
  and $u$ is declared to the left of $v$ in $\Delta$,
  then $u$ is declared to the left of $v$ in $\Gamma$. 
\end{restatable}

An older paper had a lemma
\begin{quote} ``Substitution Extension Invariance''  %
\\
   If $\judgetp{\Theta}{A}$ and $\substextend{\Theta}{\Gamma}$
   then
   $[\Gamma]A = [\Gamma]([\Theta]A)$
   and $[\Gamma]A = [\Theta]([\Gamma]A)$.
\end{quote}

For the second part, $[\Gamma]A = [\Theta]([\Gamma]A)$,  use \Lemmaref{lem:substitution-monotonicity} (i) or (iii) instead.
The first part $[\Gamma]A = [\Gamma][\Theta]A$ hasn't been proved in this system.

\begin{restatable}[Extension Inversion]{lemma}{extensioninversion}
  \XLabel{lem:extension-inversion}
~
\begin{enumerate}[(i)]
\item   %
  If $\Dee \derives \substextend{\Gamma_0, \alpha:\sort, \Gamma_1}{\Delta}$
  \\ then there exist unique $\Delta_0$ and $\Delta_1$
  \\ such that
  $\Delta = (\Delta_0, \alpha:\sort, \Delta_1)$ and
  $\Dee' \derives \substextend{\Gamma_0}{\Delta_0}$
  where $\Dee' < \Dee$.

  Moreover, if $\Gamma_1$ is soft, then $\Delta_1$ is soft. 

\item   %
  If $\Dee \derives \substextend{\Gamma_0, \MonnierComma{u}, \Gamma_1}{\Delta}$
  \\ then there exist unique $\Delta_0$ and $\Delta_1$
  \\ such that
  $\Delta = (\Delta_0, \MonnierComma{u}, \Delta_1)$
  and $\Dee' \derives \substextend{\Gamma_0}{\Delta_0}$
  where $\Dee' < \Dee$.

  Moreover, if $\Gamma_1$ is soft, then $\Delta_1$ is soft. 

  Moreover, if $\dom{\Gamma_0, \MonnierComma{u}, \Gamma_1} = \dom{\Delta}$ then
  $\dom{\Gamma_0} = \dom{\Delta_0}$.   

\item    %
  If $\Dee \derives \substextend{\Gamma_0, \hypeq{\alpha}{\tau}, \Gamma_1}{\Delta}$
  \\ then there exist unique $\Delta_0$, $\tau'$, and $\Delta_1$
  \\ such that
  $\Delta = (\Delta_0, \hypeq{\alpha}{\tau'}, \Delta_1)$
  and $\Dee' \derives \substextend{\Gamma_0}{\Delta_0}$
  and $[\Delta_0]\tau = [\Delta_0]\tau'$
  where $\Dee' < \Dee$.

\item   %
  If $\Dee \derives \substextend{\Gamma_0, \hypeq{\ahat:\sort}{\tau}, \Gamma_1}{\Delta}$
  \\ then there exist unique $\Delta_0$, $\tau'$, and $\Delta_1$
  \\ such that
  $\Delta = (\Delta_0, \hypeq{\ahat:\sort}{\tau'}, \Delta_1)$
  and $\Dee' \derives \substextend{\Gamma_0}{\Delta_0}$
  and $[\Delta_0]\tau = [\Delta_0]\tau'$
  where $\Dee' < \Dee$.

\item   %
  If $\Dee \derives \substextend{\Gamma_0, x:A, \Gamma_1}{\Delta}$
  \\ then there exist unique $\Delta_0$, $A'$, and $\Delta_1$
  \\ such that
  $\Delta = (\Delta_0, x:A', \Delta_1)$
  and $\Dee' \derives \substextend{\Gamma_0}{\Delta_0}$
  and $[\Delta_0]A = [\Delta_0]A'$
  where $\Dee' < \Dee$.

  Moreover, if $\Gamma_1$ is soft, then $\Delta_1$ is soft. 

  Moreover, if $\dom{\Gamma_0, x:A, \Gamma_1} = \dom{\Delta}$ then
  $\dom{\Gamma_0} = \dom{\Delta_0}$.

\item   %
  If $\Dee \derives \substextend{\Gamma_0, \ahat:\sort, \Gamma_1}{\Delta}$
  then either 

  \begin{itemize}
  \item there exist unique $\Delta_0$, $\tau'$, and $\Delta_1$
    \\ such that
    $\Delta = (\Delta_0, \hypeq{\ahat:\sort}{\tau'}, \Delta_1)$ and
    $\Dee' \derives \substextend{\Gamma_0}{\Delta_0}$
    where $\Dee' < \Dee$,
    \\ or
  \item there exist unique $\Delta_0$ and $\Delta_1$
    \\ such that
    $\Delta = (\Delta_0, \ahat:\sort, \Delta_1)$
    and $\Dee' \derives \substextend{\Gamma_0}{\Delta_0}$
    where $\Dee' < \Dee$.
  \end{itemize}
\end{enumerate}
\end{restatable}

\begin{restatable}[Deep Evar Introduction]{lemma}{deepexistential}
\XLabel{lem:deep-existential} %
  \begin{enumerate}[(i)]
  \item If $\Gamma_0, \Gamma_1$ is well-formed
    and $\ahat$ is not declared in $\Gamma_0, \Gamma_1$
    then $\substextend{\Gamma_0, \Gamma_1}{\Gamma_0, \ahat:\sort, \Gamma_1}$. 

  \item If $\Gamma_0, \ahat:\sort, \Gamma_1$ is well-formed
    and $\judge{\Gamma}{t}{\sort}$
    then $\substextend{\Gamma_0, \ahat:\sort, \Gamma_1}{\Gamma_0, \hypeq{\ahat:\sort}{t}, \Gamma_1}$.    

  \item If $\Gamma_0, \Gamma_1$ is well-formed
    and $\judge{\Gamma}{t}{\sort}$
    then $\substextend{\Gamma_0, \Gamma_1}{\Gamma_0, \hypeq{\ahat:\sort}{t}, \Gamma_1}$.    
  \end{enumerate}
\end{restatable}

\begin{restatable}[Soft Extension]{lemma}{softextension}
\XLabel{lem:soft-extension}
~\\
  If $\substextend{\Gamma}{\Delta}$
  and $\judgectx{\Gamma, \Theta}$
  and $\Theta$ is soft,
  then
  there exists $\Omega$ such that $\dom{\Theta} = \dom{\Omega}$
  and $\substextend{\Gamma, \Theta}{\Delta, \Omega}$. 
\end{restatable}

\begin{definition}[Filling]
  The \emph{filling} of a context $\soln{\Gamma}$ solves all unsolved variables:

  \begin{displ}
        \begin{array}[t]{lcll}
            \soln{\cdot} & = &   \cdot
            \\
            \soln{\Gamma, x : A} & = &   \soln{\Gamma}, x : A
            \\
            \soln{\Gamma, \alpha : \sort} & = &   \soln{\Gamma}, \alpha : \sort
            \\
            \soln{\Gamma, \hypeq{\alpha}{t}} & = &   \soln{\Gamma}, \hypeq{\alpha}{t}
            \\
            \soln{\Gamma, \hypeq{\ahat : \sort}{t}} & = &   \soln{\Gamma}, \hypeq{\ahat : \sort}{t}
            \\
            \soln{\Gamma, \MonnierComma{\ahat}} & = &   \soln{\Gamma}, \MonnierComma{\ahat}
            \\
            \soln{\Gamma, \ahat : \type} & = &   \soln{\Gamma}, \hypeq{\ahat : \type}{\unitty}
            \\
            \soln{\Gamma, \ahat : \ind} & = &   \soln{\Gamma}, \hypeq{\ahat : \ind}{\zero}
        \end{array}
  \end{displ}
\end{definition}

\begin{lemma}[Filling Completes]  \Label{lem:soln-completes}
  If $\substextend{\Gamma}{\Omega}$
  and $(\Gamma, \Theta)$ is well-formed,
  then
  $\substextend{\Gamma, \Theta}{\Omega, \soln{\Theta}}$.
\end{lemma}
\begin{proof}%
  By induction on $\Theta$, following the definition of $\soln{-}$ and applying the rules
  for $\extendssym$.
\end{proof}

\begin{restatable}[Parallel Admissibility]{lemma}{parallelevaradmissibility}
\XLabel{lem:parallel-admissibility} ~\\
    If $\substextend{\Gamma_L}{\Delta_L}$ and $\substextend{\Gamma_L, \Gamma_R}{\Delta_L, \Delta_R}$ then:
    
    \begin{enumerate}[(i)]
    \item $\substextend{\Gamma_L, \ahat:\sort, \Gamma_R}{\Delta_L, \ahat:\sort, \Delta_R}$
    \item If $\judge{\Delta_L}{\tau'}{\sort}$ then 
          $\substextend{\Gamma_L, \ahat:\sort, \Gamma_R}{\Delta_L, \hypeq{\ahat:\sort}{\tau'}, \Delta_R}$.
    \item If $\judge{\Gamma_L}{\tau}{\sort}$ and $\judgetp{\Delta_L}{\tau'}$ and $[\Delta_L]\tau = [\Delta_L]\tau'$,
          then $\substextend{\Gamma_L, \hypeq{\ahat:\sort}{\tau}, \Gamma_R}{\Delta_L, \hypeq{\ahat:\sort}{\tau'}, \Delta_R}$. 
    \end{enumerate}
\end{restatable}

\begin{restatable}[Parallel Extension Solution]{lemma}{parallelextensionsolve}
\XLabel{lem:parallel-extension-solve} ~\\
  If $\substextend{\Gamma_L, \ahat : \sort, \Gamma_R}{\Delta_L, \hypeq{\ahat : \sort}{\tau'}, \Delta_R}$
  and $\judge{\Gamma_L}{\tau}{\sort}$
  and $[\Delta_L]\tau = [\Delta_L]\tau'$
  \\
  then $\substextend{\Gamma_L, \hypeq{\ahat:\sort}{\tau}, \Gamma_R}{\Delta_L, \hypeq{\ahat:\sort}{\tau'}, \Delta_R}$. 
\end{restatable}

\begin{restatable}[Parallel Variable Update]{lemma}{parallelextensionupdate}
\XLabel{lem:parallel-extension-update} ~\\
  If ~~~\, $\substextend{\Gamma_L, \ahat : \sort, \Gamma_R}{\Delta_L, \hypeq{\ahat:\sort}{\tau_0}, \Delta_R}$
  and $\judge{\Gamma_L}{\tau_1}{\sort}$
  and $\judge{\Delta_L}{\tau_2}{\sort}$
  and $[\Delta_L]\tau_0 = [\Delta_L]\tau_1 = [\Delta_L]\tau_2$
  \\
  then $\substextend{\Gamma_L, \hypeq{\ahat:\sort}{\tau_1}, \Gamma_R}{\Delta_L, \hypeq{\ahat:\sort}{\tau_2}, \Delta_R}$. 
\end{restatable}

\begin{restatable}[Substitution Monotonicity]{lemma}{substitutionmonotonicity}
~\XLabel{lem:substitution-monotonicity}
\begin{enumerate}[(i)]
  \item If $\substextend{\Gamma}{\Delta}$
  and $\judge{\Gamma}{t}{\sort}$
  then $[\Delta][\Gamma]t = [\Delta]t$. 
  \item If $\substextend{\Gamma}{\Delta}$
  and $\judgeprop{\Gamma}{P}$
  then $[\Delta][\Gamma]P = [\Delta]P$. 
  \item If $\substextend{\Gamma}{\Delta}$
  and $\judgetp{\Gamma}{A}$
  then $[\Delta][\Gamma]A = [\Delta]A$. 
\end{enumerate}
\end{restatable}

\begin{restatable}[Substitution Invariance]{lemma}{substitutioninvariance}
~\XLabel{lem:substitution-invariance}
\begin{enumerate}[(i)]
  \item If $\substextend{\Gamma}{\Delta}$
  and $\judge{\Gamma}{t}{\sort}$ and $\FEV{[\Gamma]t} = \emptyset$
  then $[\Delta][\Gamma]t = [\Gamma]t$. 
  \item If $\substextend{\Gamma}{\Delta}$ 
  and $\judgeprop{\Gamma}{P}$ and $\FEV{[\Gamma]P} = \emptyset$
  then $[\Delta][\Gamma]P = [\Gamma]P$. 
  \item If $\substextend{\Gamma}{\Delta}$ 
  and $\judgetp{\Gamma}{A}$ and $\FEV{[\Gamma]A} = \emptyset$
  then $[\Delta][\Gamma]A = [\Gamma]A$. 
\end{enumerate}
\end{restatable}

\begin{definition}[Canonical Contexts]
\Label{def:canonical}
  A (complete) context $\Omega$ is \emph{canonical}
  iff, for all $(\hypeq{\ahat:\sort}{t})$ and $(\hypeq{\alpha}{t}) \in \Omega$, the solution $t$ is
  ground ($\FEV{t} = \emptyset$).
\end{definition}

\begin{restatable}[Split Extension]{lemma}{splitextension}
\XLabel{lem:split-extension} ~\\
  If $\substextend{\Delta}{\Omega}$
  \\
  and $\ahat \in \unsolved{\Delta}$
  \\
  and $\Omega = \Omega_1[\hypeq{\ahat:\sort}{t_1}]$
  \\
  and $\Omega$ is canonical (\Definitionref{def:canonical})
  \\
  and $\judge{\Omega}{t_2}{\sort}$
  \\
  then $\substextend{\Delta}{\Omega_1[\hypeq{\ahat:\sort}{t_2}]}$.
\end{restatable}

\subsection{Reflexivity and Transitivity}

\begin{restatable}[Extension Reflexivity]{lemma}{extensionreflexivity}
  \XLabel{lem:extension-reflexivity}
~\\
  If $\judgectx{\Gamma}$ then $\substextend{\Gamma}{\Gamma}$. 
\end{restatable}

\begin{restatable}[Extension Transitivity]{lemma}{extensiontransitivity}
  \XLabel{lem:extension-transitivity}
~\\
  If $\Dee \derives \substextend{\Gamma}{\Theta}$
  and $\Dee' \derives \substextend{\Theta}{\Delta}$ 
  then $\substextend{\Gamma}{\Delta}$. 
\end{restatable}

\subsection{Weakening}

The ``suffix weakening'' lemmas take a judgment under $\Gamma$ and produce
a judgment under $(\Gamma, \Theta)$.  They do \emph{not} require $\substextend{\Gamma}{\Gamma, \Theta}$.

\begin{restatable}[Suffix Weakening]{lemma}{suffixweakeningjudge}
  \XLabel{lem:suffix-weakening-judge}
  If $\judge{\Gamma}{t}{\sort}$
  then $\judge{\Gamma, \Theta}{t}{\sort}$.
\end{restatable}

\begin{restatable}[Suffix Weakening]{lemma}{suffixweakeningjudgetp}
  \XLabel{lem:suffix-weakening-judgetp}
  If $\judgetp{\Gamma}{A}$
  then $\judgetp{\Gamma, \Theta}{A}$.
\end{restatable}

The following proposed lemma is false.

\begin{quote}   %
  ``Extension Weakening (Truth)''
  \\
  If $\checkprop{\Gamma}{P}{\Delta}$
  and $\substextend{\Gamma}{\Gamma'}$
  then
  there exists $\Delta'$ such that
  $\substextend{\Delta}{\Delta'}$
  and $\checkprop{\Gamma'}{P}{\Delta'}$.
\end{quote}

Counterexample: Suppose $\checkprop{\ahat}{\ahat = \unitty}{\ahat = \unitty}$
and $\substextend{\ahat}{(\hypeq{\ahat}{(\unitty{\arr}\unitty)})}$.
Then there does \emph{not} exist such a $\Delta'$.

\begin{restatable}[Extension Weakening (Sorts)]{lemma}{extensionweakening}
  \XLabel{lem:extension-weakening}
  If $\judge{\Gamma}{t}{\sort}$ 
  and $\substextend{\Gamma}{\Delta}$
  then $\judge{\Delta}{t}{\sort}$. 
\end{restatable}

\begin{restatable}[Extension Weakening (Props)]{lemma}{extensionweakeningprop}
  \XLabel{lem:extension-weakening-prop}
  If $\judgeprop{\Gamma}{P}$ 
  and $\substextend{\Gamma}{\Delta}$
  then $\judgeprop{\Delta}{P}$. 
\end{restatable}

\begin{restatable}[Extension Weakening (Types)]{lemma}{extensionweakeningtp}
  \XLabel{lem:extension-weakening-tp}
  If $\judgetp{\Gamma}{A}$
  and $\substextend{\Gamma}{\Delta}$
  then $\judgetp{\Delta}{A}$. 
\end{restatable}

\subsection{Principal Typing Properties}

\begin{restatable}[Principal Agreement]{lemma}{substitutiontppstable}
\XLabel{lem:substitution-tpp-stable}
~
\begin{enumerate}[(i)]
\item 
  If $\judgetp{\Gamma}{A\;\p}$    %
  and $\substextend{\Gamma}{\Delta}$
  then
  $[\Delta]A = [\Gamma]A$.

\item
  If $\judgeprop{\Gamma}{P}$ and $\FEV{P} = \emptyset$
  and $\substextend{\Gamma}{\Delta}$
  then
  $[\Delta]P = [\Gamma]P$.
\end{enumerate}
\end{restatable}

\begin{restatable}[Right-Hand Subst.{} for Principal Typing]{lemma}{substitutiontpp}
\XLabel{lem:substitution-tpp}
  If $\judgetp{\Gamma}{A\;p}$
  then
  $\judgetp{\Gamma}{[\Gamma]A\;p}$. 
\end{restatable}

\begin{restatable}[Extension Weakening for Principal Typing]{lemma}{extensionweakeningtpp}
  \XLabel{lem:extension-weakening-tpp}
  If $\judgetp{\Gamma}{A\;p}$
  and $\substextend{\Gamma}{\Delta}$
  then $\judgetp{\Delta}{A\;p}$. 
\end{restatable}

\begin{restatable}[Inversion of Principal Typing]{lemma}{principalinversion}
  \XLabel{lem:principal-inversion}
  ~
  \begin{enumerate}[(1)]
  \item If $\judgetp{\Gamma}{(A \arr B)\;p}$ then $\judgetp{\Gamma}{A\;p}$ and $\judgetp{\Gamma}{B\;p}$. 
  \item If $\judgetp{\Gamma}{(P \implies A)\;p}$ then $\judgeprop{\Gamma}{P}$ and $\judgetp{\Gamma}{A\;p}$. 
  \item If $\judgetp{\Gamma}{(A \with P)\;p}$ then $\judgeprop{\Gamma}{P}$ and $\judgetp{\Gamma}{A\;p}$. 
\end{enumerate}
\end{restatable}

\stopproofheading

\subsection{Instantiation Extends}

\begin{restatable}[Instantiation Extension]{lemma}{instantiationextension}
\XLabel{lem:instantiation-extension}
~\\
   If $\instjudg{\Gamma}{\ahat}{\tau}{\sort}{\Delta}$
   then
   $\substextend{\Gamma}{\Delta}$.
\end{restatable}

\subsection{Equivalence Extends}

\begin{restatable}[Elimeq Extension]{lemma}{elimeqextension}
\XLabel{lem:elimeq-extension}
~\\
   If $\elimeq{\Gamma}{s}{t}{\sort}{\Delta}$ then
   there exists $\Theta$ such that
   $\substextend{\Gamma, \Theta}{\Delta}$.
\end{restatable}

\begin{restatable}[Elimprop Extension]{lemma}{elimpropextension}
\XLabel{lem:elimprop-extension}
~\\
   If $\elimprop{\Gamma}{P}{\Delta}$ then
   there exists $\Theta$ such that
   $\substextend{\Gamma, \Theta}{\Delta}$.
\end{restatable}

\begin{restatable}[Checkeq Extension]{lemma}{checkeqextension}
\XLabel{lem:checkeq-extension}
~\\
   If $\equivjudg{\Gamma}{A}{B}{\Delta}$ then
   $\substextend{\Gamma}{\Delta}$.
\end{restatable}

\begin{restatable}[Checkprop Extension]{lemma}{checkpropextension}
\XLabel{lem:checkprop-extension}
~\\
   If $\checkprop{\Gamma}{P}{\Delta}$ then
   $\substextend{\Gamma}{\Delta}$.
\end{restatable}

\begin{restatable}[Prop Equivalence Extension]{lemma}{propequivextension}
\XLabel{lem:propequiv-extension}
~\\
   If $\propequivjudg{\Gamma}{P}{Q}{\Delta}$ then
   $\substextend{\Gamma}{\Delta}$.
\end{restatable}

\begin{restatable}[Equivalence Extension]{lemma}{equivextension}
\XLabel{lem:equiv-extension}
~\\
   If $\equivjudg{\Gamma}{A}{B}{\Delta}$ then
   $\substextend{\Gamma}{\Delta}$.
\end{restatable}

\subsection{Subtyping Extends}

\begin{restatable}[Subtyping Extension]{lemma}{subtypingextension}
\XLabel{lem:subtyping-extension}   %
~
If $\subjudg[\mp]{\Gamma}{A}{B}{\Delta}$ then
$\substextend{\Gamma}{\Delta}$.
\end{restatable}

\subsection{Typing Extends}

\begin{restatable}[Typing Extension]{lemma}{typingextension}
\XLabel{lem:typing-extension}
 ~\\
  If $\chkjudg{p}{\Gamma}{e}{A}{\Delta}$
  \\
  or $\synjudg{p}{\Gamma}{e}{A}{\Delta}$
  \\
  or $\spinejudg{\Gamma}{s}{A}{p}{B}{q}{\Delta}$
  \\
  or $\matchjudg{q}{p}{\Gamma}{\Pi}{\Avec}{C}{\Delta}$
  \\
  or $\matchelimjudg{p}{\Gamma}{\Pi}{P}{\Avec}{C}{\Delta}$
  \\
  then
  $\substextend{\Gamma}{\Delta}$.
\end{restatable}

\subsection{Unfiled}

\begin{restatable}[Context Partitioning]{lemma}{contextpartitioning}
\XLabel{lem:context-partitioning}
  ~\\
  If $\substextend{\Delta, \MonnierComma{\ahat}, \Theta}{\Omega, \MonnierComma{\ahat}, \Omega_Z}$
  then
  there is a $\Psi$ such that
  $[\Omega, \MonnierComma{\ahat}, \Omega_Z](\Delta, \MonnierComma{\ahat}, \Theta)
   =
   [\Omega]\Delta, \Psi$. 
\end{restatable}

\begin{lemma}[Softness Goes Away]
\Label{lem:softness-goes-away}
  ~\\
  If $\substextend{\Delta, \Theta}{\Omega, \Omega_Z}$
  where $\substextend{\Delta}{\Omega}$ and $\Theta$ is soft,
  then
  $[\Omega, \Omega_Z](\Delta, \Theta) = [\Omega]\Delta$.
\end{lemma}
\stopproofheading
\begin{proof}%
  By induction on $\Theta$, following the definition of $[\Omega]\Gamma$.
\end{proof}

\begin{restatable}[Completing Stability]{lemma}{completesstability}
\XLabel{lem:completes-stability}
~\\
  If $\substextend{\Gamma}{\Omega}$
  then $[\Omega]\Gamma = [\Omega]\Omega$. 
\end{restatable}

\begin{restatable}[Completing Completeness]{lemma}{completingcompleteness}
\XLabel{lem:completing-completeness}
~
  \begin{enumerate}[(i)]
  \item %
      If $\substextend{\Omega}{\Omega'}$ and $\judge{\Omega}{t}{\sort}$ then $[\Omega]t = [\Omega']t$. 
  \item %
      If $\substextend{\Omega}{\Omega'}$ and $\judgetp{\Omega}{A}$ then $[\Omega]A = [\Omega']A$. 
  \item %
      If $\substextend{\Omega}{\Omega'}$ then $[\Omega]\Omega = [\Omega']\Omega'$. 
  \end{enumerate}
\end{restatable}

\begin{restatable}[Confluence of Completeness]{lemma}{confluenceofcompleteness}    \XLabel{lem:completes-confluence}  ~\\
  If $\substextend{\Delta_1}{\Omega}$
  and $\substextend{\Delta_2}{\Omega}$
  then
  $[\Omega]\Delta_1 = [\Omega]\Delta_2$.
\end{restatable}

\begin{restatable}[Multiple Confluence]{lemma}{multipleconfluence}    \XLabel{lem:multiple-confluence}  ~\\
  If $\substextend{\Delta}{\Omega}$
  and $\substextend{\Omega}{\Omega'}$
  and $\substextend{\Delta'}{\Omega'}$
  then
  $[\Omega]\Delta = [\Omega']\Delta'$.
\end{restatable}

\begin{lemma}[Bundled Substitution for Sorting]
\Label{lem:bundled-subst-sort}
  If $\judge{\Gamma}{t}{\sort}$
  and $\substextend{\Gamma}{\Omega}$
  then
  $\judge{[\Omega]\Gamma}{[\Omega]t}{\sort}$.
\end{lemma}
\begin{proof}%
 ~\\
  \begin{llproof}
    \judgePf{\Gamma}{t}{\sort}  {Given}
    \judgePf{\Omega}{t}{\sort}  {By \Lemmaref{lem:extension-weakening}}
    \judgePf{[\Omega]\Omega}{[\Omega]t}{\sort}   {By \Lemmaref{lem:completion-sort}}
    \substextendPf{\Omega}{\Omega}   {By \Lemmaref{lem:extension-reflexivity}}
    \eqPf{[\Omega]\Omega}{[\Omega]\Gamma}   {By \Lemmaref{lem:completes-confluence}}
  \Hand  \judgePf{[\Omega]\Gamma}{[\Omega]t}{\sort}   {By above equality}
  \end{llproof}
\qedhere
\end{proof}

\newcommand{\OmegaC}{\Omega_{\textit{canon}}}
\begin{restatable}[Canonical Completion]{lemma}{canonicalcompletion}    \XLabel{lem:canonical-completion}  ~\\
  If $\substextend{\Gamma}{\Omega}$
  \\
  then there exists $\OmegaC$
  such that $\substextend{\Gamma}{\OmegaC}$
  and $\substextend{\OmegaC}{\Omega}$
  and $\dom{\OmegaC} = \dom{\Gamma}$
  and, for all $\hypeq{\ahat:\sort}{\tau}$ and $\hypeq{\alpha}{\tau}$ in $\OmegaC$, we have $\FEV{\tau} = \emptyset$.
\end{restatable}

The completion $\OmegaC$ is ``canonical'' because (1) its domain exactly matches $\Gamma$ and (2) its solutions
$\tau$ have no evars.
Note that it follows from \Lemmaref{lem:multiple-confluence} that $[\OmegaC]\Gamma = [\Omega]\Gamma$.

\begin{restatable}[Split Solutions]{lemma}{splitsolutions}
\XLabel{lem:split-solutions} ~\\
  If $\substextend{\Delta}{\Omega}$
  and $\ahat \in \unsolved{\Delta}$
  \\
  then
  there exists $\Omega_1 = \Omega_1'[\hypeq{\ahat:\sort}{t_1}]$
  such that
  $\substextend{\Omega_1}{\Omega}$
  and
  $\Omega_2 = \Omega_1'[\hypeq{\ahat:\sort}{t_2}]$
  where
  $\substextend{\Delta}{\Omega_2}$
  and $t_2 \neq t_1$
  and $\Omega_2$ is canonical.
\end{restatable}

\section{Internal Properties of the Declarative System}

\begin{restatable}[Interpolating With and Exists]{lemma}{matchwithexists}
\XLabel{lem:match-with-exists}
~
    \begin{enumerate}[(1)]
    \item 
       If $\Dee \derives \declmatchjudg{p}{\Psi}{\Pi}{\Avec}{C}$
       and $\declcheckprop{\Psi}{P_0}$
       \\
       then  $\Dee' \derives \declmatchjudg{p}{\Psi}{\Pi}{\Avec}{C \with P_0}$.

    \item 
       If $\Dee \derives \declmatchjudg{p}{\Psi}{\Pi}{\Avec}{[\tau/\alpha]C_0}$
       and $\judge{\Psi}{\tau}{\sort}$
       \\
       then  $\Dee' \derives \declmatchjudg{p}{\Psi}{\Pi}{\Avec}{(\extype{\alpha:\sort} C_0)}$.
    \end{enumerate}

{\noindent  In both cases, the height of $\Dee'$ is one greater than the height of $\Dee$.
\\
    Moreover, similar properties hold for the eliminating judgment $\declmatchelimjudg{p}{\Psi}{\Pi}{P}{\Avec}{C}$.
}
\end{restatable}

\begin{restatable}[Case Invertibility]{lemma}{caseinvertibility}
\XLabel{lem:case-invertibility}
~\\
   If $\declchkjudg{p}{\Psi}{\case{e_0}{\Pi}}{C}$
   \\
   then
   $\declsynjudg{\p}{\Psi}{e_0}{A}$
   and
   $\declmatchjudg{p}{\Psi}{\Pi}{A}{C}$
   and
   $\declcovers{\Psi}{\Pi}{A}$
   \\
   where the height of each resulting derivation is strictly less than the height of the given derivation.
\end{restatable}

\section{Miscellaneous Properties of the Algorithmic System}

\begin{restatable}[Well-Formed Outputs of Typing]{lemma}{wfoutputs}
\XLabel{lem:wf-outputs}
~
    \begin{enumerate}[(1)]
    \item[(Spines)]
        If\; $\spinejudg{\Gamma}{s}{A}{q}{C}{p}{\Delta}$
        or  $\recspinejudg{\Gamma}{s}{A}{q}{C}{p}{\Delta}$
        \\
        and $\judgetp{\Gamma}{A\;q}$
        \\
        then $\judgetp{\Delta}{C\;p}$.

    \item[(Synthesis)]
        If $\synjudg{p}{\Gamma}{e}{A}{\Delta}$
        \\
        then $\judgetp{A}{p}$.

    \end{enumerate}
\end{restatable}

\section{Decidability of Instantiation}

\begin{restatable}[Left Unsolvedness Preservation]{lemma}{leftunsolvednesspreservation}
\XLabel{lem:left-unsolvedness-preservation}
~\\
  If $\instjudg{\underbrace{\Gamma_0, \ahat, \Gamma_1}_\Gamma}{\ahat}{A}{\sort}{\Delta}$
  and $\bhat \in \unsolved{\Gamma_0}$
  then
  $\bhat \in \unsolved{\Delta}$.
\end{restatable}

\begin{restatable}[Left Free Variable Preservation]{lemma}{leftfreevariablepreservation}
\XLabel{lem:left-free-variable-preservation}
  If $\instjudg{\overbrace{\Gamma_0, \ahat:\sort, \Gamma_1}^\Gamma}{\ahat}{t}{\sort}{\Delta}$
  and $\judge{\Gamma}{s}{\sort'}$
  and $\ahat \notin \FV{[\Gamma]s}$
  and $\bhat \in \unsolved{\Gamma_0}$
  and $\bhat \notin \FV{[\Gamma]s}$,
  then $\bhat \notin \FV{[\Delta]s}$. 
\end{restatable}

\begin{restatable}[Instantiation Size Preservation]{lemma}{instantiationsizepreservation}
\XLabel{lem:instantiation-size-preservation}
  If $\instjudg{\overbrace{\Gamma_0, \ahat, \Gamma_1}^\Gamma}{\ahat}{\tau}{\sort}{\Delta}$
  and $\judge{\Gamma}{s}{\sort'}$
  and
  $\ahat \notin \FV{[\Gamma]s}$, 
  then $|[\Gamma]s| = |[\Delta]s|$,
  where $|C|$ is the plain size of the term $C$.
\end{restatable}

\begin{restatable}[Decidability of Instantiation]{lemma}{instantiationdecidability}
\XLabel{lem:instantiation-decidable}
    If $\Gamma = \Gamma_0[\ahat:\sort']$
    and
    $\judge{\Gamma}{t}{\sort}$
    such that $[\Gamma]t = t$
    and $\ahat \notin \FV{t}$, then:
    
    \begin{enumerate}[(1)]
    \item Either there exists $\Delta$ such that $\instjudg{\Gamma_0[\ahat:\sort']}{\ahat}{t}{\sort}{\Delta}$,
      or not.
    \end{enumerate}
\end{restatable}

\section{Separation}

\begin{definition}[Separation]
\Label{def:separation}
~\\
  An algorithmic context $\Gamma$ is \emph{separable} and written $\Gamma_L \sep \Gamma_R$
  if (1) $\Gamma = (\Gamma_L, \Gamma_R)$
  and (2) for all $(\hypeq{\ahat:\sort}{\tau}) \in \Gamma_R$
               it is the case that $\FEV{\tau} \subseteq \dom{\Gamma_R}$.
\end{definition}

Any context $\Gamma$ is separable into, at least, $\cdot \sep \Gamma$
and $\Gamma \sep \cdot$.

\begin{definition}[Separation-Preserving Extension]
\Label{def:separation-extension}
~\\
  The separated context $\Gamma_L \sep \Gamma_R$ extends to $\Delta_L \sep \Gamma_R$,
  written
  \[
      \sepextend{\Gamma_L}{\Gamma_R}{\Delta_L}{\Delta_R}
  \]
  if $\substextend{(\Gamma_L, \Gamma_R)}{(\Delta_L, \Delta_R)}$
  and $\dom{\Gamma_L} \subseteq \dom{\Delta_L}$
  and $\dom{\Gamma_R} \subseteq \dom{\Delta_R}$.
\end{definition}

Separation-preserving extension says that variables from one half don't ``cross'' into the
other half.  Thus, $\Delta_L$ may add existential variables to $\Gamma_L$, and $\Delta_R$
may add existential variables to $\Gamma_R$, but no variable from $\Gamma_L$ ends up
in $\Delta_R$ and no variable from $\Gamma_R$ ends up in $\Delta_L$.

It is necessary to write $\sepextend{\Gamma_L}{\Gamma_R}{\Delta_L}{\Delta_R}$
rather than
$\substextend{(\Gamma_L \sep \Gamma_R)}{(\Delta_L \sep \Delta_R)}$,
because only $\sepextendsym$ includes the domain conditions.  For example,
$\substextend{(\ahat \sep \bhat)}{(\ahat, \hypeq{\bhat}{\ahat}) \sep \cdot}$,
but the variable $\bhat$ has ``crossed over'' to the left of $\sep$ in
the context $(\ahat, \hypeq{\bhat}{\ahat}) \sep \cdot$.

\begin{restatable}[Transitivity of Separation]{lemma}{separationtransitivity}
\XLabel{lem:separation-transitivity}
~\\
   If $\sepextend{\Gamma_L}{\Gamma_R}{\Theta_L}{\Theta_R}$
   and $\sepextend{\Theta_L}{\Theta_R}{\Delta_L}{\Delta_R}$
   \\
   then
   $\sepextend{\Gamma_L}{\Gamma_R}{\Delta_L}{\Delta_R}$.
\end{restatable}

\begin{restatable}[Separation Truncation]{lemma}{separationtruncation}
\XLabel{lem:separation-truncation}
~\\
   If $H$ has the form
   $\alpha : \sort$
   or $\MonnierComma{\ahat}$
   or $\MonnierComma{P}$
   or $\hyp{p}{x}{A}$
   \\
   and 
   $\sepextend{\Gamma_L}{(\Gamma_R, H)}{\Delta_L}{\Delta_R}$
   \\
   then
   $\sepextend{\Gamma_L}{\Gamma_R}{\Delta_L}{\Delta_0}$
   where $\Delta_R = (\Delta_0, H, \Theta)$.
\end{restatable}

\begin{restatable}[Separation for Auxiliary Judgments]{lemma}{separationaux}
\XLabel{lem:separation-aux}
~
\begin{enumerate}[(i)]
    \item %
        If $\checkeq{\Gamma_L \sep \Gamma_R}{\sigma}{\tau}{\sort}{\Delta}$
        \\
        and $\exbasisTwo{\Gamma_R}{\sigma}{\tau}$
        \\
        then $\Delta = (\Delta_L \sep \Delta_R)$
        and $\sepextend{\Gamma_L}{\Gamma_R}{\Delta_L}{\Delta_R}$.

    \item %
        If $\checkprop{\Gamma_L \sep \Gamma_R}{P}{\Delta}$
        \\
        and $\exbasis{\Gamma_R}{P}$
        \\
        then $\Delta = (\Delta_L \sep \Delta_R)$
        and $\sepextend{\Gamma_L}{\Gamma_R}{\Delta_L}{\Delta_R}$.

    \item %
        If $\elimeq{\Gamma_L \sep \Gamma_R}{\sigma}{\tau}{\sort}{\Delta}$
        \\
        and $\FEV{\sigma} \union \FEV{\tau} = \emptyset$
        \\
        then $\Delta = (\Delta_L \sep (\Delta_R, \Theta))$
        and $\sepextend{\Gamma_L}{(\Gamma_R, \Theta)}{\Delta_L}{\Delta_R}$.

    \item %
        If $\elimprop{\Gamma_L \sep \Gamma_R}{P}{\Delta}$
        \\
        and $\FEV{P} = \emptyset$
        \\
        then $\Delta = (\Delta_L \sep (\Delta_R, \Theta))$
        and $\sepextend{\Gamma_L}{(\Gamma_R, \Theta)}{\Delta_L}{\Delta_R}$.

    \item %
        If $\instjudg{\Gamma_L \sep \Gamma_R}{\ahat}{\tau}{\sort}{\Delta}$
        \\ and $(\FEV{\tau} \union \setof{\ahat}) \subseteq \dom{\Gamma_R}$
        \\
        then $\Delta = (\Delta_L \sep \Delta_R)$
        and $\sepextend{\Gamma_L}{\Gamma_R}{\Delta_L}{\Delta_R}$.

    \item %
        If $\propequivjudg{\Gamma_L \sep \Gamma_R}{P}{Q}{\Delta}$
        \\ and $\exbasisTwo{\Gamma_R}{P}{Q}$
        \\
        then $\Delta = (\Delta_L \sep \Delta_R)$
        and $\sepextend{\Gamma_L}{\Gamma_R}{\Delta_L}{\Delta_R}$.

    \item
        If $\equivjudg{\Gamma_L \sep \Gamma_R}{A}{B}{\Delta}$
        \\ and $\exbasisTwo{\Gamma_R}{A}{B}$
        \\
        then $\Delta = (\Delta_L \sep \Delta_R)$
        and $\sepextend{\Gamma_L}{\Gamma_R}{\Delta_L}{\Delta_R}$.
\end{enumerate}
\end{restatable}

\begin{restatable}[Separation for Subtyping]{lemma}{separationsubtyping}
\XLabel{lem:separation-subtyping}

        If $\subjudg[\polvar]{\Gamma_L \sep \Gamma_R}{A}{B}{\Delta}$
        \\ and $\exbasis{\Gamma_R}{A}$
        \\ and $\exbasis{\Gamma_R}{B}$
        \\
        then $\Delta = (\Delta_L \sep \Delta_R)$
        and $\sepextend{\Gamma_L}{\Gamma_R}{\Delta_L}{\Delta_R}$.   
\end{restatable}

\begin{restatable}[Separation---Main]{lemma}{separationmain}
\XLabel{lem:separation-main}
~
\begin{enumerate}
    \item[(Spines)]
        If\; $\spinejudg{\Gamma_L \sep \Gamma_R}{s}{A}{p}{C}{q}{\Delta}$
        \\
        or  $\recspinejudg{\Gamma_L \sep \Gamma_R}{s}{A}{p}{C}{q}{\Delta}$
        \\
        and $\judgetp{\Gamma_L \sep \Gamma_R}{A\;p}$
        \\
        and $\exbasis{\Gamma_R}{A}$
        \\
        then $\Delta = (\Delta_L \sep \Delta_R)$
        and $\sepextend{\Gamma_L}{\Gamma_R}{\Delta_L}{\Delta_R}$
        and $\exbasis{\Delta_R}{C}$.

    \item[(Checking)]
        If $\chkjudg{p}{\Gamma_L \sep \Gamma_R}{e}{C}{\Delta}$
        \\
        and $\judgetp{\Gamma_L \sep \Gamma_R}{C\;p}$
        \\
        and $\exbasis{\Gamma_R}{C}$
        \\
        then $\Delta = (\Delta_L \sep \Delta_R)$
        and $\sepextend{\Gamma_L}{\Gamma_R}{\Delta_L}{\Delta_R}$.

    \item[(Synthesis)]
        If $\synjudg{p}{\Gamma_L \sep \Gamma_R}{e}{A}{\Delta}$
        \\
        then $\Delta = (\Delta_L \sep \Delta_R)$
        and $\sepextend{\Gamma_L}{\Gamma_R}{\Delta_L}{\Delta_R}$.

    \item[(Match)]
        If $\matchjudg{q}{p}{\Gamma_L \sep \Gamma_R}{\Pi}{\Avec}{C}{\Delta}$
        \\
        and $\FEV{\Avec} = \emptyset$   %
        \\
        and $\exbasis{\Gamma_R}{C}$   %
        \\
        then $\Delta = (\Delta_L \sep \Delta_R)$
        and $\sepextend{\Gamma_L}{\Gamma_R}{\Delta_L}{\Delta_R}$.

    \item[(Match Elim.)]
        If $\matchelimjudg{p}{\Gamma_L \sep \Gamma_R}{\Pi}{P}{\Avec}{C}{\Delta}$
        \\
        and $\FEV{P} = \emptyset$
        \\
        and $\FEV{\Avec} = \emptyset$
        \\
        and $\exbasis{\Gamma_R}{C}$
        \\
        then $\Delta = (\Delta_L \sep \Delta_R)$
        and $\sepextend{\Gamma_L}{\Gamma_R}{\Delta_L}{\Delta_R}$.
\end{enumerate}
\end{restatable}

\section{Decidability of Algorithmic Subtyping}

\begin{definition}
  The following connectives are \emph{large}:
  \[
     \AllSym
     ~~~~
     \implies
     ~~~~
     \with
  \]
  A type is \emph{large} iff its head connective is large.
  (Note that a non-large type may contain large connectives, provided they are not in head position.)

  The number of these connectives in a type $A$ is denoted by $\largecount{A}$.
\end{definition}

\subsection{Lemmas for Decidability of Subtyping}

\begin{restatable}[Substitution Isn't Large]{lemma}{substitutionnotlarge}
\XLabel{lem:substitution-not-large}
~\\
   For all contexts $\Theta$, we have $\largecount{[\Theta]A} = \largecount{A}$.
\end{restatable}

\begin{restatable}[Instantiation Solves]{lemma}{instantiationsolves}
\XLabel{lem:instantiation-solves}
~\\
  If $\instjudg{\Gamma}{\ahat}{\tau}{\sort}{\Delta}$
  and $[\Gamma]\tau = \tau$
  and $\ahat \notin \FV{[\Gamma]\tau}$
  then
  $|\unsolved{\Gamma}| = |\unsolved{\Delta}| + 1$. 
\end{restatable}

\begin{restatable}[Checkeq Solving]{lemma}{checkeqsolving}
\XLabel{lem:checkeq-solving}
  If $\checkeq{\Gamma}{s}{t}{\sort}{\Delta}$
  then
  either $\Delta = \Gamma$
  or $|\unsolved{\Delta}| < |\unsolved{\Gamma}|$.    %
\end{restatable}

\begin{restatable}[Prop Equiv Solving]{lemma}{propequivsolving}
\XLabel{lem:propequiv-solving}
~\\
  If $\propequivjudg{\Gamma}{P}{Q}{\Delta}$
  then
  either $\Delta = \Gamma$
  or $|\unsolved{\Delta}| < |\unsolved{\Gamma}|$.
\end{restatable}

\begin{restatable}[Equiv Solving]{lemma}{equivsolving}
\XLabel{lem:equiv-solving}
~\\
  If $\equivjudg{\Gamma}{A}{B}{\Delta}$
  then
  either $\Delta = \Gamma$
  or $|\unsolved{\Delta}| < |\unsolved{\Gamma}|$.
\end{restatable}

\begin{restatable}[Decidability of Propositional Judgments]{lemma}{propdecidable}
\XLabel{lem:prop-decidable}
~\\
  The following judgments are decidable, with $\Delta$ as output in (1)--(3), and $\Deltabot$ as output
  in (4) and (5).

  We assume $\sigma = [\Gamma]\sigma$ and $t = [\Gamma]t$ in (1) and (4).
  Similarly, in the other parts we assume $P = [\Gamma]P$
  and (in part (3)) $Q = [\Gamma]Q$.

  \begin{enumerate}[(1)]
  \item $\checkeq{\Gamma}{\sigma}{t}{\sort}{\Delta}$    %
  \item $\checkprop{\Gamma}{P}{\Delta}$    %
  \item $\propequivjudg{\Gamma}{P}{Q}{\Delta}$    %
  \item $\elimeq{\Gamma}{\sigma}{t}{\sort}{\Deltabot}$    %
  \item $\elimprop{\Gamma}{P}{\Deltabot}$    %
  \end{enumerate}
\end{restatable}

\begin{restatable}[Decidability of Equivalence]{lemma}{equivdecidable}
\XLabel{lem:equiv-decidable}
~\\
  Given a context $\Gamma$ and types $A$, $B$ such that $\judgetp{\Gamma}{A}$ and
  $\judgetp{\Gamma}{B}$ and $[\Gamma]A = A$ and $[\Gamma]B = B$,
  it is decidable whether there exists $\Delta$ such that $\equivjudg{\Gamma}{A}{B}{\Delta}$.
\end{restatable}

\subsection{Decidability of Subtyping}

\begin{restatable}[Decidability of Subtyping]{theorem}{subtypingdecidable}
\XLabel{thm:subtyping-decidable}
~\\
  Given a context $\Gamma$ and types $A$, $B$ such that $\judgetp{\Gamma}{A}$ and
  $\judgetp{\Gamma}{B}$ and $[\Gamma]A = A$ and $[\Gamma]B = B$,
  it is decidable whether there exists $\Delta$ such that $\subjudg[\polvar]{\Gamma}{A}{B}{\Delta}$.
\end{restatable}

\subsection{Decidability of Matching and Coverage}

\begin{restatable}[Decidability of Guardedness Judgment]{lemma}{guardeddecidable}
  \XLabel{lem:guarded-decidable}
  ~\\
  For any set of branches $\Pi$, the relation $\guarded{\Pi}$ is decidable. 
\end{restatable}

\begin{restatable}[Decidability of Expansion Judgments]{lemma}{expansiondecidable}
\XLabel{lem:expansion-decidable}
~\\
    Given branches $\Pi$, it is decidable whether:

    \begin{enumerate}[(1)]
    \item there exists a unique $\Pi'$ such that $\expandpair{\Pi}{\Pi'}$;
    \item there exist unique $\Pi_L$ and $\Pi_R$ such that $\expandsum{\Pi}{\Pi_L}{\Pi_R}$;
    \item there exists a unique $\Pi'$ such that $\expandvar{\Pi}{\Pi'}$;
    \item there exists a unique $\Pi'$ such that $\expandunit{\Pi}{\Pi'}$.
    \item there exist unique $\Pi_{\vecnil}$ and $\Pi_{\vecconssym}$ such that $\expandvec{\Pi}{\Pi_{\vecnil}}{\Pi_{\vecconssym}}$.
    \end{enumerate}
\end{restatable}

\begin{restatable}[Expansion Shrinks Size]{lemma}{expansionshrinkssize}
\XLabel{lem:expansion-shrinks-size}
~\\
We define the \emph{size} of a pattern $|p|$ as follows:
\begin{displaymath}
  \begin{array}{lcl}
    |x| & = & 0 \\
    |\wild| & = & 0 \\
    |\pair{p}{p'}| & = & 1 + |p| + |p'| \\
    |\unitpat| & = & 0 \\
    |\inl{p}| & = & 1 + |p| \\
    |\inr{p}| & = & 1 + |p| \\
    |\vecnil| & = & 1 \\
    |\veccons{p}{p'}| & = & 1 + |p| + |p'| \\
  \end{array}
\end{displaymath}
We lift \emph{size} to branches $\pi = \branch{\vec{p}}{e}$ as follows:
\begin{displaymath}
  |\branch{p_1, \ldots, p_n}{e}| = |p_1| + \ldots + |p_n|
\end{displaymath}
We lift \emph{size} to branch lists $\Pi = \pi_1 \alt \ldots \alt \pi_n$ as follows:
\begin{displaymath}
  |\pi_1 \alt \ldots \alt \pi_n| = |\pi_1| + \ldots + |\pi_n|
\end{displaymath}

Now, the following properties hold:

\begin{enumerate}
\item If $\expandvar{\Pi}{\Pi'}$ then $|\Pi| = |\Pi'|$.
\item If $\expandunit{\Pi}{\Pi'}$ then $|\Pi| = |\Pi'|$.
\item If $\expandpair{\Pi}{\Pi'}$ then $|\Pi| \leq |\Pi'|$. 
\item If $\expandsum{\Pi}{\Pi_L}{\Pi_R}$ then $|\Pi| \leq |\Pi_1|$ and $|\Pi| \leq |\Pi_2|$.
\item If $\expandvec{\Pi}{\Pi_{\vecnil}}{\Pi_{\vecconssym}}$ then $|\Pi_{\vecnil}| \leq |\Pi|$ and  $|\Pi_{\vecconssym}| \leq |\Pi|$.
\item If $\guarded{\Pi}$ and $\expandvec{\Pi}{\Pi_{\vecnil}}{\Pi_{\vecconssym}}$ then $|\Pi_{\vecnil}| < |\Pi|$ and  $|\Pi_{\vecconssym}| < |\Pi|$.
\end{enumerate}
\end{restatable}

\begin{restatable}[Decidability of Coverage]{theorem}{coveragedecidable}
\XLabel{thm:coverage-decidable}
~\\
   Given a context $\Gamma$, branches $\Pi$ and types $\Avec$,
   it is decidable whether $\covers{q}{\Gamma}{\Pi}{\Avec}$ is derivable.
\end{restatable}

\subsection{Decidability of Typing}

\begin{restatable}[Decidability of Typing]{theorem}{typingdecidable}
\XLabel{thm:typing-decidable}
~\\[-2ex]
  \begin{enumerate}[(i)]
      \item \emph{Synthesis:}  
        Given a context $\Gamma$, a principality $p$, 
        and a term $e$, 
        \\
        it is decidable whether there exist a type $A$ and a context $\Delta$ such that \\
        $\synjudg{p}{\Gamma}{e}{A}{\Delta}$.

      \item \emph{Spines:}
        Given a context $\Gamma$,
        a spine $s$, a principality $p$, 
        and a type $A$ such that $\judgetp{\Gamma}{A}$,
        \\
        it is decidable whether there exist a type $B$, a principality $q$ and a context $\Delta$ such that \\
        $\spinejudg{\Gamma}{s}{A}{p}{B}{q}{\Delta}$.

      \item \emph{Checking:}
        Given a context $\Gamma$, a principality $p$, 
        a term $e$,
        and a type $B$ such that $\judgetp{\Gamma}{B}$,
        \\
        it is decidable whether there is a context $\Delta$ such that \\
        $\chkjudg{p}{\Gamma}{e}{B}{\Delta}$.

      \item \emph{Matching:}
           Given a context $\Gamma$, branches $\Pi$, a list of types $\Avec$, a type $C$, and a principality $p$,
           it is decidable whether there exists $\Delta$ such that
           $\matchjudg{q}{p}{\Gamma}{\Pi}{\Avec}{C}{\Delta}$.

           Also, if given a proposition $P$ as well, it is decidable whether there exists $\Delta$ such that
           $\matchelimjudg{p}{\Gamma}{\Pi}{P}{\Avec}{C}{\Delta}$.
  \end{enumerate}
\end{restatable}

\section{Determinacy}

\begin{restatable}[Determinacy of Auxiliary Judgments]{lemma}{auxdet}
\XLabel{lem:aux-det}
~
\begin{enumerate}[(1)]
    \item \emph{Elimeq:}   %
      Given $\Gamma$, $\sigma$, $t$, $\sort$
      such that
      $\FEV{\sigma} \union \FEV{t} = \emptyset$
      and
      $\Dee_1 \derives \elimeq{\Gamma}{\sigma}{t}{\sort}{\Deltabot_1}$
      and
      $\Dee_2 \derives \elimeq{\Gamma}{\sigma}{t}{\sort}{\Deltabot_2}$,
      \\
      it is the case that $\Deltabot_1 = \Deltabot_2$.

    \item \emph{Instantiation:}    %
      Given $\Gamma$, $\ahat$, $t$, $\sort$
      such that $\ahat \in \unsolved{\Gamma}$
      and $\judge{\Gamma}{t}{\sort}$
      and $\ahat \notin \FV{t}$
      \\
      and
      $\Dee_1 \derives \instjudg{\Gamma}{\ahat}{t}{\sort}{\Delta_1}$
      and
      $\Dee_2 \derives \instjudg{\Gamma}{\ahat}{t}{\sort}{\Delta_2}$
      \\
      it is the case that $\Delta_1 = \Delta_2$.

    \item \emph{Symmetric instantiation:}  %

      Given $\Gamma$, $\ahat$, $\bhat$, $\sort$
      such that $\ahat, \bhat \in \unsolved{\Gamma}$
      and $\ahat \neq \bhat$
      \\
      and
      $\Dee_1 \derives \instjudg{\Gamma}{\ahat}{\bhat}{\sort}{\Delta_1}$
      and
      $\Dee_2 \derives \instjudg{\Gamma}{\bhat}{\ahat}{\sort}{\Delta_2}$
      \\
      it is the case that $\Delta_1 = \Delta_2$.

    \item \emph{Checkeq:}   %
      Given $\Gamma$, $\sigma$, $t$, $\sort$
      such that
      $\Dee_1 \derives \checkeq{\Gamma}{\sigma}{t}{\sort}{\Delta_1}$
      and
      $\Dee_2 \derives \checkeq{\Gamma}{\sigma}{t}{\sort}{\Delta_2}$
      \\
      it is the case that $\Delta_1 = \Delta_2$.

    \item \emph{Elimprop:}   %
      Given $\Gamma$, $P$
      such that
      $\Dee_1 \derives \elimprop{\Gamma}{P}{\Deltabot_1}$
      and
      $\Dee_2 \derives \elimprop{\Gamma}{P}{\Deltabot_2}$
      \\
      it is the case that $\Delta_1 = \Delta_2$.

    \item \emph{Checkprop:}   %
      Given $\Gamma$, $P$
      such that
      $\Dee_1 \derives \checkprop{\Gamma}{P}{\Delta_1}$
      and
      $\Dee_2 \derives \checkprop{\Gamma}{P}{\Delta_2}$,
      \\
      it is the case that $\Delta_1 = \Delta_2$.
\end{enumerate}
\end{restatable}

\begin{restatable}[Determinacy of Equivalence]{lemma}{equivdet}
\XLabel{lem:equiv-det}
~
\begin{enumerate}[(1)]
    \item \emph{Propositional equivalence:}   %
      Given $\Gamma$, $P$, $Q$
      such that
      $\Dee_1 \derives \propequivjudg{\Gamma}{P}{Q}{\Delta_1}$
      and
      $\Dee_2 \derives \propequivjudg{\Gamma}{P}{Q}{\Delta_2}$,
      \\
      it is the case that $\Delta_1 = \Delta_2$.

    \item \emph{Type equivalence:}   %
      Given $\Gamma$, $A$, $B$
      such that
      $\Dee_1 \derives \propequivjudg{\Gamma}{A}{B}{\Delta_1}$
      and
      $\Dee_2 \derives \propequivjudg{\Gamma}{A}{B}{\Delta_2}$,
      \\
      it is the case that $\Delta_1 = \Delta_2$.
\end{enumerate}
\end{restatable}

\begin{restatable}[Determinacy of Subtyping]{theorem}{subtypingdet}
\XLabel{thm:subtyping-det}
~
\begin{enumerate}[(1)]
    \item \emph{Subtyping:}   %
      Given $\Gamma$, $e$, $A$, $B$
      such that
      $\Dee_1 \derives \subjudg[\polvar]{\Gamma}{A}{B}{\Delta_1}$
      and
      $\Dee_2 \derives \subjudg[\polvar]{\Gamma}{A}{B}{\Delta_2}$,
      \\
      it is the case that $\Delta_1 = \Delta_2$.
\end{enumerate}
\end{restatable}

\begin{restatable}[Determinacy of Typing]{theorem}{typingdet}
\XLabel{thm:typing-det}
~
\begin{enumerate}[(1)]
    \item \emph{Checking:}   %
      Given $\Gamma$, $e$, $A$, $p$
      such that
      $\Dee_1 \derives \chkjudg{p}{\Gamma}{e}{A}{\Delta_1}$
      and
      $\Dee_2 \derives \chkjudg{p}{\Gamma}{e}{A}{\Delta_2}$,
      \\
      it is the case that $\Delta_1 = \Delta_2$.

   \item \emph{Synthesis:}      %
      Given $\Gamma$, $e$
      such that
      $\Dee_1 \derives \synjudg{p_1}{\Gamma}{e}{B_1}{\Delta_1}$
      and
      $\Dee_2 \derives \synjudg{p_2}{\Gamma}{e}{B_2}{\Delta_2}$,
      \\
      it is the case that $B_1 = B_2$ and $p_1 = p_2$ and $\Delta_1 = \Delta_2$.

   \item \emph{Spine judgments:}      %

      Given $\Gamma$, $e$, $A$, $p$
      such that
      $\Dee_1 \derives \spinejudg{\Gamma}{e}{A}{p}{C_1}{q_1}{\Delta_1}$
      and
      $\Dee_2 \derives \spinejudg{\Gamma}{e}{A}{p}{C_2}{q_2}{\Delta_2}$,
      \\
      it is the case that
      $C_1 = C_2$ and $q_1 = q_2$ and $\Delta_1 = \Delta_2$.

      The same applies for derivations of the principality-recovering judgments
      $\recspinejudg{\Gamma}{e}{A}{p}{C_k}{q_k}{\Delta_k}$.

   \item \emph{Match judgments:}      %

      Given $\Gamma$, $\Pi$, $\Avec$, $p$, $C$
      such that
      $\Dee_1 \derives \matchjudg{q}{p}{\Gamma}{\Pi}{\Avec}{C}{\Delta_1}$
      and
      $\Dee_2 \derives \matchjudg{q}{p}{\Gamma}{\Pi}{\Avec}{C}{\Delta_2}$,
      \\
      it is the case that
      $\Delta_1 = \Delta_2$.

      Given $\Gamma$, $P$, $\Pi$, $\Avec$, $p$, $C$
      \\
      such that
      $\Dee_1 \derives \matchelimjudg{p}{\Gamma}{\Pi}{P}{\Avec}{C}{\Delta_1}$
      and
      $\Dee_2 \derives \matchelimjudg{p}{\Gamma}{\Pi}{P}{\Avec}{C}{\Delta_2}$,
      \\
      it is the case that
      $\Delta_1 = \Delta_2$.
\end{enumerate}
\end{restatable}

\section{Soundness}
\label{apx:soundness}

\subsection{Soundness of Instantiation}

\begin{restatable}[Soundness of Instantiation]{lemma}{instantiationsoundness}  \XLabel{lem:instantiation-soundness}
~\\
    If $\instjudg{\Gamma}{\ahat}{\tau}{\sort}{\Delta}$
    and $\ahat \notin \FV{[\Gamma]\tau}$
    and $[\Gamma]\tau = \tau$
    and $\substextend{\Delta}{\Omega}$
    then
    $[\Omega]\ahat = [\Omega]\tau$.
\end{restatable}

\subsection{Soundness of Checkeq}

\begin{restatable}[Soundness of Checkeq]{lemma}{checkeqsoundness}
\XLabel{lem:checkeq-soundness} ~\\
  If $\checkeq{\Gamma}{\sigma}{t}{\sort}{\Delta}$
  where
  $\substextend{\Delta}{\Omega}$
  then
  $[\Omega]\sigma = [\Omega]t$.
\end{restatable}

\subsection{Soundness of Equivalence (Propositions and Types)}

\begin{restatable}[Soundness of Propositional Equivalence]{lemma}{propequivsoundness}
\XLabel{lem:propequiv-soundness} ~\\
  If $\propequivjudg{\Gamma}{P}{Q}{\Delta}$
  where
  $\substextend{\Delta}{\Omega}$
  then
  $[\Omega]P = [\Omega]Q$.
\end{restatable}

\begin{restatable}[Soundness of Algorithmic Equivalence]{lemma}{equivsoundness}
\XLabel{lem:equiv-soundness} ~\\
  If $\equivjudg{\Gamma}{A}{B}{\Delta}$ 
  where
  $\substextend{\Delta}{\Omega}$
  then
  $[\Omega]A = [\Omega]B$.

\end{restatable}

\subsection{Soundness of Checkprop}

\begin{restatable}[Soundness of Checkprop]{lemma}{checkpropsoundness}
\XLabel{lem:checkprop-soundness}
~\\
  If $\checkprop{\Gamma}{P}{\Delta}$ and $\substextend{\Delta}{\Omega}$
  then
  $\declcheckprop{\Psi}{[\Omega]P}$.
\end{restatable}

\subsection{Soundness of Eliminations (Equality and Proposition)}

\begin{restatable}[Soundness of Equality Elimination]{lemma}{elimeqsoundness}
\XLabel{lem:elimeq-soundness} ~\\
  If
  $[\Gamma]\sigma = \sigma$
  and $[\Gamma]t = t$
  and $\judge{\Gamma}{\sigma}{\sort}$
  and $\judge{\Gamma}{t}{\sort}$
  and $\FEV{\sigma} \union \FEV{t} = \emptyset$,
  then:

  \begin{enumerate}[(1)]
    \item 
        If $\elimeq{\Gamma}{\sigma}{t}{\sort}{\Delta}$
        \\
        then
        $\Delta = (\Gamma, \Theta)$
        where $\Theta = (\hypeq{\alpha_1}{t_1}, \dots, \hypeq{\alpha_n}{t_n})$
        and
        \\
        for all $\Omega$ such that $\substextend{\Gamma}{\Omega}$
        \\
        and all $t'$ such that
        $\judge{\Omega}{t'}{\sort'}$,
        \\
            it is the case that
            $[\Omega, \Theta]t' = [\theta][\Omega]t'$,
            where $\theta = \mgu{\sigma}{t}$.

    \item  
        If $\elimeq{\Gamma}{\sigma}{t}{\sort}{\bot}$
        then $\mgu{\sigma}{t} = \bot$ (that is, no most general unifier exists).
  \end{enumerate}
\end{restatable}

\subsection{Soundness of Subtyping}

\begin{restatable}[Soundness of Algorithmic Subtyping]{theorem}{subtypingsoundness} 
\XLabel{thm:subtyping-soundness} ~\\
  If
  $[\Gamma]A = A$
  and $[\Gamma]B = B$
  and $\judgetp{\Gamma}{A}$
  and $\judgetp{\Gamma}{B}$
  and $\substextend{\Delta}{\Omega}$
  and $\subjudg[\polvar]{\Gamma}{A}{B}{\Delta}$ then 
    $\declsubjudg[\polvar]{[\Omega]\Delta}{[\Omega]A}{[\Omega]B}$.
\end{restatable}

\subsection{Soundness of Typing}

\begin{restatable}[Soundness of Match Coverage]{theorem}{coveragesoundness}
\XLabel{thm:coverage-soundness} ~
\begin{enumerate}
\item If $\covers{q}{\Gamma}{\Pi}{\Avec}$
    and $\judgetpvec{\Gamma}{\Avec\;q}$
    and $[\Gamma]\Avec = \Avec$
    and $\substextend{\Gamma}{\Omega}$
    then 
    $\covers{q}{[\Omega]\Gamma}{\Pi}{\Avec}$. 
\item If $\coverseq{\Gamma}{P}{\Pi}{\Avec}$
    and $\substextend{\Gamma}{\Omega}$
    and $\judgetpvec{\Gamma}{\Avec\;\p}$
    and $[\Gamma]\Avec = \Avec$
    and $[\Gamma]P = P$
    then 
    $\coverseq{[\Omega]\Gamma}{P}{\Pi}{\Avec}$. 
\end{enumerate}
\end{restatable}

\begin{restatable}[Well-formedness of Algorithmic Typing]{lemma}{typingwf}
\XLabel{lem:typing-wf} ~\\
Given $\judgectx{\Gamma}$: \\[-2.5ex]
    \begin{enumerate}[(i)]
      \item
          If $\synjudg{p}{\Gamma}{e}{A}{\Delta}$ 
          then
          $\judgetp{\Delta}{A\;p}$. 

      \item
        If $\spinejudg{\Gamma}{s}{A}{p}{B}{q}{\Delta}$ and $\judgetp{\Gamma}{A\;p}$ then 
           $\judgetp{\Delta}{B\;q}$. 
    \end{enumerate}
\end{restatable}

\newcommand{\mtuple}[3]{\ensuremath{\langle #1, #2, #3 \rangle}}

\begin{definition}[Measure]
\Label{def:soundness-completeness-measure}
  Let measure $\mathcal{M}$ on
  typing judgments
  be a lexicographic ordering:

  \begin{enumerate}
  \item first, the subject expression $e$, spine $s$, or matches $\Pi$---regarding
    all types in annotations as equal in size;
  \item second, the partial order on judgment forms
    where an ordinary spine judgment is \emph{smaller} than a
    principality-recovering spine judgment---and with all other judgment
    forms considered equal in size; and,
  \item third, the derivation height.
  \end{enumerate}
\end{definition}
\[
      \left\langle
        e / s / \Pi,~~~
        \begin{array}[c]{c}
            \text{ordinary spine judgment}
          \\[0.3ex]
          <
          \\[0.3ex]
            \text{recovering spine judgment}
        \end{array}
        ,
        ~~~
        \text{height}(\Dee)
      \right\rangle
\]
Note that this definition doesn't take notice of whether
a spine judgment is declarative or algorithmic.

This measure works to show soundness and completeness.
We list each rule below, along with a 3-tuple.  For example,
for \Sub we write \mtuple==<,
meaning that each judgment to which we need to apply the i.h.\
has a subject of the same size ($=$),
a judgment form of the same size ($=$),
and a smaller derivation height ($<$).
We write ``$-$'' when a part of the measure
need not be considered because a lexicographically more significant part is
smaller, as in the \Anno rule, where the premise has a smaller subject:
\mtuple<--.

Algorithmic rules (soundness cases):

\begin{itemize}
\item \Var, \UnitIntro, \UnitIntroSolve, \EmptySpine
  and \Nil
  have no premises, or only auxiliary judgments as premises.

\item \Sub: \mtuple==<

\item \Anno: \mtuple<--

\item \AllIntro, \AllSpine, \ExIntro, \WithIntro: \mtuple==<
\item \ImpliesIntro: \mtuple==<
\item \ImpliesIntroBot has only an auxiliary judgment, to which we need not apply the i.h.,
  putting it in the same class as the rules with no premises.
\item \ImpliesSpine: \mtuple==<
\item \ArrIntro, \ArrIntroSolve, \ArrElim, \Rec: \mtuple<--
\item \Recover: \mtuple=<-
\item \Pass: \mtuple=<-
\item \ArrSpine, \SumIntro{k}, \SumIntroSolve{k},
      \PairIntro, \PairIntroSolve, \Cons:
        \mtuple<--
\item \SolveSpine: \mtuple==<
\item \Case: \mtuple<--
\end{itemize}

Declarative rules (completeness cases):

\begin{itemize}
\item \DeclVar, \DeclUnitIntro, \DeclEmptySpine
  and \DeclNil
  have no premises, or only auxiliary judgments as premises.
\item \DeclSub: \mtuple==<
\item \DeclAnno: \mtuple<--
\item \DeclAllIntro, \DeclAllSpine, \DeclExIntro, \DeclWithIntro,
      \DeclImpliesIntro, \DeclImpliesSpine:
        \mtuple==<
\item \DeclArrIntro, \DeclArrElim, \DeclRec: \mtuple<--
\item \DeclRecover: \mtuple=<-
\item \DeclPass: \mtuple=<-
\item \DeclArrSpine, \DeclSumIntro{k},
      \DeclPairIntro, \DeclCase,
      \DeclCons,
        \mtuple<--
\end{itemize}

\begin{definition}[Eagerness]
\Label{def:eager}
~\\
  A derivation $\Dee$ whose conclusion is $\J$ is \emph{eager} if:

    \begin{enumerate}[(i)]
      \item %
        $\J = \chkjudg{p}{\Gamma}{e}{A}{\Delta}$

          if $\judgetp{\Gamma}{A\;p}$
          and $A = [\Gamma]A$   %
          \\
          implies that
          \\
          every subderivation of $\Dee$ is eager.
      
      \item %
          $\J = \synjudg{p}{\Gamma}{e}{A}{\Delta}$ 

          if $A = [\Delta]A$
          \\
          and
          every subderivation of $\Dee$ is eager.
          
      \item %
          $\J = \spinejudg{\Gamma}{s}{A}{p}{B}{q}{\Delta}$
          
          if $\judgetp{\Gamma}{A\;p}$
          and $A = [\Gamma]A$   %
          \\
          implies
          that
          \\
          $B = [\Delta]B$
          \\
          and
          every subderivation of $\Dee$ is eager.

      \item %
          $\J = \recspinejudg{\Gamma}{s}{A}{p}{B}{q}{\Delta}$

          if $\judgetp{\Gamma}{A\;p}$
          and $A = [\Gamma]A$   %
          \\
          implies that
          \\
          $B = [\Delta]B$
          \\
          and
          every subderivation of $\Dee$ is eager.

      \item %
          $\J = \matchjudg{q}{p}{\Gamma}{\Pi}{\Avec}{C}{\Delta}$

          if
          $\judgetpvec{\Gamma}{\Avec\;q}$
          and $[\Gamma]\Avec = \Avec$  %
          and $\judgetp{\Gamma}{C\;p}$
          and $C = [\Gamma]C$
          \\
          implies that
          \\
          every subderivation of $\Dee$ is eager.

      \item %
          $\J = \matchelimjudg{p}{\Gamma}{\Pi}{P}{\Avec}{C}{\Delta}$

          if
          $\judgetpvec{\Gamma}{\Avec\;\p}$
          and $\judgeprop{\Gamma}{P}$
          and $[\Gamma]\Avec = \Avec$  %
          and $\judgetp{\Gamma}{C\;p}$
          and $C = [\Gamma]C$
          \\
          implies that
          \\
          every subderivation of $\Dee$ is eager.
    \end{enumerate}  
\end{definition}

\begin{restatable}[Eagerness of Types]{theorem}{eagerness}
\XLabel{thm:eagerness}
~
    \begin{enumerate}[(i)]
      \item %
          If $\Dee$ derives $\chkjudg{p}{\Gamma}{e}{A}{\Delta}$
          and $\judgetp{\Gamma}{A\;p}$
          and $A = [\Gamma]A$   %
          then
          $\Dee$ is eager.
      
      \item %
          If $\Dee$ derives $\synjudg{p}{\Gamma}{e}{A}{\Delta}$ 
          then
          $\Dee$ is eager.

      \item %
          If $\Dee$ derives  $\spinejudg{\Gamma}{s}{A}{p}{B}{q}{\Delta}$
          and $\judgetp{\Gamma}{A\;p}$
          and $A = [\Gamma]A$   %
          then 
          $\Dee$ is eager.

      \item %
          If $\Dee$ derives  $\recspinejudg{\Gamma}{s}{A}{p}{B}{q}{\Delta}$
          and $\judgetp{\Gamma}{A\;p}$
          and $A = [\Gamma]A$   %
          then 
          $\Dee$ is eager.

      \item %
          If $\Dee$ derives  $\matchjudg{q}{p}{\Gamma}{\Pi}{\Avec}{C}{\Delta}$
          and $\judgetpvec{\Gamma}{\Avec\;q}$
          and $[\Gamma]\Avec = \Avec$  %
          and $\judgetp{\Gamma}{C\;p}$
          \\
          then
          $\Dee$ is eager.

      \item %
          If $\Dee$ derives $\matchelimjudg{p}{\Gamma}{\Pi}{P}{\Avec}{C}{\Delta}$
          and $\judgeprop{\Gamma}{P}$
          and $\FEV{P} = \emptyset$
          and $[\Gamma]P = P$
          \\
          and $\judgetpvec{\Gamma}{\Avec\;\p}$  %
          and $\judgetp{\Gamma}{C\;p}$
          \\
          then
          $\Dee$ is eager.
    \end{enumerate}
\end{restatable}

\begin{restatable}[Soundness of Algorithmic Typing]{theorem}{typingsoundness}
\XLabel{thm:typing-soundness} ~\\
Given $\substextend{\Delta}{\Omega}$: \\[-2.5ex]
    \begin{enumerate}[(i)]
      \item %
          If $\chkjudg{p}{\Gamma}{e}{A}{\Delta}$
          and $\judgetp{\Gamma}{A\;p}$
          and $A = [\Gamma]A$   %
          then
          $\declchkjudg{p}{[\Omega]\Delta}{[\Omega]e}{[\Omega]A}$. 
      
      \item %
          If $\synjudg{p}{\Gamma}{e}{A}{\Delta}$ 
          then
          $\declsynjudg{p}{[\Omega]\Delta}{[\Omega]e}{[\Omega]A}$. 

      \item %
          If $\spinejudg{\Gamma}{s}{A}{p}{B}{q}{\Delta}$
          and $\judgetp{\Gamma}{A\;p}$
          and $A = [\Gamma]A$   %
          then 
          $\declspinejudg{[\Omega]\Delta}{[\Omega]s}{[\Omega]A}{p}{[\Omega]B}{q}$.

      \item %
          If $\recspinejudg{\Gamma}{s}{A}{p}{B}{q}{\Delta}$
          and $\judgetp{\Gamma}{A\;p}$
          and $A = [\Gamma]A$   %
          then 
          $\declrecspinejudg{[\Omega]\Delta}{[\Omega]s}{[\Omega]A}{p}{[\Omega]B}{q}$.

      \item %
          If $\matchjudg{q}{p}{\Gamma}{\Pi}{\Avec}{C}{\Delta}$
          and $\judgetpvec{\Gamma}{\Avec\;\p}$
          and $[\Gamma]\Avec = \Avec$  %
          and $\judgetp{\Gamma}{C\;p}$
          \\
          then
          $\declmatchjudg{q}{p}{[\Omega]\Delta}{[\Omega]\Pi}{[\Omega]\Avec}{[\Omega]C}$. 

      \item %
          If $\matchelimjudg{p}{\Gamma}{\Pi}{P}{\Avec}{C}{\Delta}$
          and $\judgeprop{\Gamma}{P}$
          and $\FEV{P} = \emptyset$
          and $[\Gamma]P = P$
          \\
          and $\judgetpvec{\Gamma}{\Avec\;\p}$  %
          and $\judgetp{\Gamma}{C\;p}$
          \\
          then
          $\declmatchelimjudg{p}{[\Omega]\Delta}{[\Omega]\Pi}{[\Omega]P}{[\Omega]\Avec}{[\Omega]C}$. 
    \end{enumerate}
\end{restatable}

\section{Completeness}

\subsection{Completeness of Auxiliary Judgments}

\begin{restatable}[Completeness of Instantiation]{lemma}{instantiationcompleteness}
   \XLabel{lem:instantiation-completeness}
   ~\\
   Given $\substextend{\Gamma}{\Omega}$ and $\dom{\Gamma} = \dom{\Omega}$
   and $\judge{\Gamma}{\tau}{\sort}$
   and $\tau = [\Gamma]\tau$ 
   and $\ahat \in \unsolved{\Gamma}$  %
   and $\ahat \notin \FV{\tau}$:
 \\   
   If %
   $[\Omega]\ahat = [\Omega]\tau$
                 \\
                 then there are $\Delta$, $\Omega'$ such that
                 $\substextend{\Omega}{\Omega'}$
                 and $\substextend{\Delta}{\Omega'}$
                 and $\dom{\Delta} = \dom{\Omega'}$
                 and $\instjudg{\Gamma}{\ahat}{\tau}{\sort}{\Delta}$. 
\end{restatable}

\begin{restatable}[Completeness of Checkeq]{lemma}{checkeqcompleteness}
  \XLabel{lem:checkeq-completeness}
  ~\\
  Given $\substextend{\Gamma}{\Omega}$ and $\dom{\Gamma} = \dom{\Omega}$
  \\
  and $\judge{\Gamma}{\sigma}{\sort}$
  and $\judge{\Gamma}{\tau}{\sort}$
  \\
  and $[\Omega]\sigma = [\Omega]\tau$
  \\
  then $\checkeq{\Gamma}{[\Gamma]\sigma}{[\Gamma]\tau}{\sort}{\Delta}$
  \\
  where $\substextend{\Delta}{\Omega'}$
  and $\dom{\Delta} = \dom{\Omega'}$
  and $\substextend{\Omega}{\Omega'}$.
\end{restatable}

\begin{restatable}[Completeness of Elimeq]{lemma}{elimeqcompleteness}
  \XLabel{lem:elimeq-completeness}
  ~\\
  If
  $[\Gamma]\sigma = \sigma$
  and $[\Gamma]t = t$
  and
  $\judge{\Gamma}{\sigma}{\sort}$
  and $\judge{\Gamma}{t}{\sort}$
  and
  $\FEV{\sigma} \union \FEV{t} = \emptyset$
  then:

  \begin{enumerate}[(1)]
  \item  %
         If $\mgu{\sigma}{t} = \theta$
         \\
          then
          $\elimeq{\Gamma}{\sigma}{t}{\sort}{(\Gamma, \Delta)}$
          \\
          where
          $\Delta$ has the form $\hypeq{\alpha_1}{t_1}, \dots, \hypeq{\alpha_n}{t_n}$
          \\
          and
          for all $u$ such that $\judge{\Gamma}{u}{\sort}$,
          it is the case that $[\Gamma, \Delta]u = \theta([\Gamma]u)$.

  \item  %
         If $\mgu{\sigma}{t} = \bot$ (that is, no most general unifier exists)
         then $\elimeq{\Gamma}{\sigma}{t}{\sort}{\bot}$.
  \end{enumerate}
\end{restatable}

\begin{restatable}[Substitution Upgrade]{lemma}{substupgrade}
  \XLabel{lem:subst-upgrade}
~\\
  If $\Delta$ has the form $\hypeq{\alpha_1}{t_1}, \dots, \hypeq{\alpha_n}{t_n}$
  \\
  and, for all $u$ such that $\judge{\Gamma}{u}{\sort}$, it is the case that
  $[\Gamma, \Delta]u = \theta([\Gamma]u)$,
  \\
  then:

  \begin{enumerate}[(i)]
  \item %
    If $\judgetp{\Gamma}{A}$
    then $[\Gamma, \Delta]A = \theta([\Gamma]A)$.

  \item %
    If $\substextend{\Gamma}{\Omega}$
    then $[\Omega]\Gamma = \theta([\Omega]\Gamma)$.

  \item %
    If $\substextend{\Gamma}{\Omega}$
    then $[\Omega, \Delta](\Gamma, \Delta) = \theta([\Omega]\Gamma)$.

  \item %
    If $\substextend{\Gamma}{\Omega}$ then
    $[\Omega, \Delta]e = \theta([\Omega]e)$.
  \end{enumerate}
\end{restatable}

\begin{restatable}[Completeness of Propequiv]{lemma}{propequivcompleteness}
  \XLabel{lem:propequiv-completeness}
  ~\\
  Given $\substextend{\Gamma}{\Omega}$
  \\
  and $\judgeprop{\Gamma}{P}$
  and $\judgeprop{\Gamma}{Q}$
  \\
  and $[\Omega]P = [\Omega]Q$
  \\
  then $\propequivjudg{\Gamma}{[\Gamma]P}{[\Gamma]Q}{\Delta}$
  \\
  where $\substextend{\Delta}{\Omega'}$
  and $\substextend{\Omega}{\Omega'}$.
\end{restatable}

\begin{restatable}[Completeness of Checkprop]{lemma}{checkpropcompleteness}
  \XLabel{lem:checkprop-completeness}
  ~\\
  If $\substextend{\Gamma}{\Omega}$ and $\dom{\Gamma} = \dom{\Omega}$
  \\
  and $\judgeprop{\Gamma}{P}$
  \\
  and $[\Gamma]P = P$
  \\
  and $\declcheckprop{[\Omega]\Gamma}{[\Omega]P}$
  \\
  then $\checkprop{\Gamma}{P}{\Delta}$
  \\
  where $\substextend{\Delta}{\Omega'}$
  and $\substextend{\Omega}{\Omega'}$
  and $\dom{\Delta} = \dom{\Omega'}$.
\end{restatable}

\subsection{Completeness of Equivalence and Subtyping}

\begin{restatable}[Completeness of Equiv]{lemma}{equivcompleteness}
   \XLabel{lem:equiv-completeness}
  ~\\
  If $\substextend{\Gamma}{\Omega}$
  and $\judgetp{\Gamma}{A}$
  and $\judgetp{\Gamma}{B}$
  \\
  and $[\Omega]A = [\Omega]B$
  \\
  then
  there exist $\Delta$ and $\Omega'$ such that 
  $\substextend{\Delta}{\Omega'}$ and 
  $\substextend{\Omega}{\Omega'}$ and 
  $\equivjudg{\Gamma}{[\Gamma]A}{[\Gamma]B}{\Delta}$.
\end{restatable}

\begin{restatable}[Completeness of Subtyping]{theorem}{subtypingcompleteness}
  \XLabel{thm:subtyping-completeness}
  ~\\
  If $\substextend{\Gamma}{\Omega}$
  and $\dom{\Gamma} = \dom{\Omega}$
  and $\judgetp{\Gamma}{A}$
  and $\judgetp{\Gamma}{B}$
  \\
  and
   $\declsubjudg[\polvar]{[\Omega]\Gamma}{[\Omega]A}{[\Omega]B}$
      \\
      then
      there exist $\Delta$ and $\Omega'$ such that 
      $\substextend{\Delta}{\Omega'}$
  \\ and $\dom{\Delta} = \dom{\Omega'}$
  \\ and $\substextend{\Omega}{\Omega'}$
  \\ and $\subjudg[\polvar]{\Gamma}{[\Gamma]A}{[\Gamma]B}{\Delta}$.
\end{restatable}

\subsection{Completeness of Typing}

\begin{restatable}[Variable Decomposition]{lemma}{variabledecomposition}
\XLabel{thm:variable-decomposition} ~
If $\expandvar{\Pi}{\Pi'}$, then
\begin{enumerate}
\item if $\expandunit{\Pi}{\Pi''}$ then $\Pi'' = \Pi'$. 
\item if $\expandpair{\Pi}{\Pi'''}$ then there exists $\Pi''$ such that $\expandvar{\Pi'''}{\Pi''}$ and $\expandvar{\Pi''}{\Pi'}$,
\item if $\expandsum{\Pi}{\Pi_L}{\Pi_R}$ then $\expandvar{\Pi_L}{\Pi'}$ and $\expandvar{\Pi_R}{\Pi'}$,
\item if $\expandvec{\Pi}{\Pi_{[]}}{\Pi_{::}}$ then $\Pi' = \Pi_{[]}$. 
\end{enumerate}
\end{restatable}

\begin{restatable}[Pattern Decomposition and Substitution]{lemma}{patterndecompositionsubstitution}
  \XLabel{thm:pattern-decomposition-substitution} ~
  \begin{enumerate}
  \item If $\expandvar{\Pi}{\Pi'}$ then $\expandvar{[\Omega]\Pi}{[\Omega]\Pi'}$.
  \item If $\expandunit{\Pi}{\Pi'}$ then $\expandunit{[\Omega]\Pi}{[\Omega]\Pi'}$.
  \item If $\expandpair{\Pi}{\Pi'}$ then $\expandpair{[\Omega]\Pi}{[\Omega]\Pi'}$.

  \item If $\expandsum{\Pi}{\Pi_1}{\Pi_2}$ then $\expandsum{[\Omega]\Pi}{[\Omega]\Pi_1}{[\Omega]\Pi_2}$.
  \item If $\expandvec{\Pi}{\Pi_1}{\Pi_2}$ then $\expandvec{[\Omega]\Pi}{[\Omega]\Pi_1}{[\Omega]\Pi_2}$.
  \item If $\expandvar{[\Omega]\Pi}{\Pi'}$ then there is $\Pi''$ such that $[\Omega]\Pi'' = \Pi'$ and $\expandvar{\Pi}{\Pi''}$.
  \item If $\expandunit{[\Omega]\Pi}{\Pi'}$ then there is $\Pi''$ such that $[\Omega]\Pi'' = \Pi'$ and $\expandunit{\Pi}{\Pi''}$.
  \item If $\expandpair{[\Omega]\Pi}{\Pi'}$ then there is $\Pi''$ such that $[\Omega]\Pi'' = \Pi'$ and $\expandpair{\Pi}{\Pi''}$.
  \item If $\expandsum{[\Omega]\Pi}{\Pi'_1}{\Pi'_2}$ then there are $\Pi_1$ and $\Pi_2$ such that $[\Omega]\Pi_1 = \Pi'_1$ and $[\Omega]\Pi_2 = \Pi'_2$ and $\expandsum{\Pi}{\Pi_1}{\Pi_2}$.

  \item If $\expandvec{[\Omega]\Pi}{\Pi'_1}{\Pi'_2}$ then there are $\Pi_1$ and $\Pi_2$ such that $[\Omega]\Pi_1 = \Pi'_1$ and $[\Omega]\Pi_2 = \Pi'_2$ and $\expandvec{\Pi}{\Pi_1}{\Pi_2}$.
  \end{enumerate}

\end{restatable}

\begin{restatable}[Pattern Decomposition Functionality]{lemma}{patterndecompositionfunctionality}
  \XLabel{thm:pattern-decomposition-functionality} ~
  \begin{enumerate}
    \item If $\expandvar{\Pi}{\Pi'}$ and $\expandvar{\Pi}{\Pi''}$ then $\Pi' = \Pi''$. 
    \item If $\expandunit{\Pi}{\Pi'}$ and $\expandunit{\Pi}{\Pi''}$ then $\Pi' = \Pi''$. 
    \item If $\expandpair{\Pi}{\Pi'}$ and $\expandpair{\Pi}{\Pi''}$ then $\Pi' = \Pi''$.
    \item If $\expandsum{\Pi}{\Pi_1}{\Pi_2}$ and $\expandsum{\Pi}{\Pi'_1}{\Pi'_2}$  then $\Pi_1 = \Pi'_1$ and $\Pi_2 = \Pi'_2$. 
    \item If $\expandvec{\Pi}{\Pi_1}{\Pi_2}$ and $\expandvec{\Pi}{\Pi_1}{\Pi_2}$ then $\Pi_1 = \Pi'_1$ and $\Pi_2 = \Pi'_2$. 
   \end{enumerate}
 \end{restatable}

\begin{restatable}[Decidability of Variable Removal]{lemma}{decidabilityvariableremoval}
\XLabel{thm:variable-removal} ~
For all $\Pi$, either there exists a $\Pi'$ such that $\expandvar{\Pi}{\Pi'}$ or there does not. 
\end{restatable}

\begin{restatable}[Variable Inversion]{lemma}{variableinversion}
\XLabel{thm:variable-inversion} ~
\begin{enumerate}
\item If $\expandvar{\Pi}{\Pi'}$ and $\declcovers[q]{\Psi}{\Pi}{A,\Avec}$ then $\declcovers[q]{\Psi}{\Pi'}{\Avec}$.
\item If $\expandvar{\Pi}{\Pi'}$ and $\covers{q}{\Gamma}{\Pi}{A,\Avec}$ then $\covers{q}{\Gamma}{\Pi'}{\Avec}$.
\end{enumerate}
\end{restatable}

\begin{restatable}[Completeness of Match Coverage]{theorem}{coveragecompleteness}
\XLabel{thm:coverage-completeness} ~
\begin{enumerate}
\item 
    If $\judgetpvec{\Gamma}{\Avec\;q}$
    and $[\Gamma]\Avec = \Avec$
    and (for all $\Omega$ such that $\substextend{\Gamma}{\Omega}$,  we have $\covers{q}{[\Omega]\Gamma}{[\Omega]\Pi}{[\Omega]\Avec}$)
    \\
    then 
    $\covers{q}{\Gamma}{\Pi}{\Avec}$. 
\item 
    If  $[\Gamma]\Avec = \Avec$
    and $[\Gamma]P = P$
    and $\judgetpvec{\Gamma}{\Avec\;\p}$
    and (for all $\Omega$ such that $\substextend{\Gamma}{\Omega}$, we have $\coverseq{[\Omega]\Gamma}{[\Omega]P}{[\Omega]\Pi}{[\Omega]\Avec}$) \\
    then 
    $\coverseq{\Gamma}{P}{\Pi}{\Avec}$. 
\end{enumerate}
\end{restatable}

\begin{restatable}[Completeness of Algorithmic Typing]{theorem}{typingcompleteness}
  \XLabel{thm:typing-completeness}
  Given $\substextend{\Gamma}{\Omega}$ such that $\dom{\Gamma} = \dom{\Omega}$:
  \\[-2ex]
    \begin{enumerate}[(i)]
      \item %
           If $\judgetp{\Gamma}{A\;p}$
           and $\declchkjudg{p}{[\Omega]\Gamma}{[\Omega]e}{[\Omega]A}$
           and $p' \moreprincipal p$
           \\
           then there exist $\Delta$ and $\Omega'$
           \\
           such that
           $\substextend{\Delta}{\Omega'}$
           and $\dom{\Delta} = \dom{\Omega'}$
           and $\substextend{\Omega}{\Omega'}$ 
           \\ and 
           $\chkjudg{p'}{\Gamma}{e}{[\Gamma]A}{\Delta}$.
      
      \item %
           If  $\judgetp{\Gamma}{A\;p}$
           and $\declsynjudg{p}{[\Omega]\Gamma}{[\Omega]e}{A}$
           \\
           then there exist $\Delta$, $\Omega'$, $A'$, and $p' \moreprincipal p$
           \\
           such that
           $\substextend{\Delta}{\Omega'}$
           and $\dom{\Delta} = \dom{\Omega'}$
           and $\substextend{\Omega}{\Omega'}$
           \\
           and $\synjudg{p'}{\Gamma}{e}{A'}{\Delta}$
           and $A' = [\Delta]A'$
           and $A = [\Omega']A'$. 

      \item %
           If $\judgetp{\Gamma}{A\;p}$
           and $\declspinejudg{[\Omega]\Gamma}{[\Omega]s}{[\Omega]A}{p}{B}{q}$
           and $p' \moreprincipal p$
           \\
           then there exist $\Delta$, $\Omega'$, $B'$ and $q' \moreprincipal q$
           \\
           such that
           $\substextend{\Delta}{\Omega'}$
           and $\dom{\Delta} = \dom{\Omega'}$
           and $\substextend{\Omega}{\Omega'}$
           \\
           and $\spinejudg{\Gamma}{s}{[\Gamma]A}{p'}{B'}{q'}{\Delta}$
           and $B' = [\Delta]B'$
           and $B = [\Omega']B'$. 

      \item %
           If $\judgetp{\Gamma}{A\;p}$
           and $\declrecspinejudg{[\Omega]\Gamma}{[\Omega]s}{[\Omega]A}{p}{B}{q}$
           and $p' \moreprincipal p$
           \\
           then there exist $\Delta$, $\Omega'$, $B'$, and $q' \moreprincipal q$
           \\
           such that
           $\substextend{\Delta}{\Omega'}$
           and $\dom{\Delta} = \dom{\Omega'}$
           and $\substextend{\Omega}{\Omega'}$ \\
           and $\recspinejudg{\Gamma}{s}{[\Gamma]A}{p'}{B'}{q'}{\Delta}$
           and $B' = [\Delta]B'$
           and $B = [\Omega']B'$.       

      \item %
           If $\judgetpvec{\Gamma}{\Avec\;\p}$
           and $\judgetp{\Gamma}{C\;p}$
           and $\declmatchjudg[q]{p}{[\Omega]\Gamma}{[\Omega]\Pi}{[\Omega]\Avec}{[\Omega]C}$
           and $p' \moreprincipal p$
           \\
           then there exist $\Delta$, $\Omega'$, and $C$
           \\
           such that
           $\substextend{\Delta}{\Omega'}$
           and $\dom{\Delta} = \dom{\Omega'}$
           and $\substextend{\Omega}{\Omega'}$ \\
           and $\matchjudg{q}{p'}{\Gamma}{\Pi}{[\Gamma]\Avec}{[\Gamma]C}{\Delta}$.

      \item %
           If $\judgetpvec{\Gamma}{\Avec\;\p}$
           and $\judgeprop{\Gamma}{P}$
           and $\FEV{P} = \emptyset$
           and $\judgetp{\Gamma}{C\;p}$
           \\
           and $\declmatchelimjudg{p}{[\Omega]\Gamma}{[\Omega]\Pi}{[\Omega]P}{[\Omega]\Avec}{[\Omega]C}$
           \\
           and $p' \moreprincipal p$
           \\
           then there exist $\Delta$, $\Omega'$, and $C$
           \\
           such that
           $\substextend{\Delta}{\Omega'}$
           and $\dom{\Delta} = \dom{\Omega'}$
           and $\substextend{\Omega}{\Omega'}$ \\
           and $\matchelimjudg{p'}{\Gamma}{\Pi}{[\Gamma]P}{[\Gamma]\Avec}{[\Gamma]C}{\Delta}$.
     \end{enumerate}
\end{restatable}

\clearpage

\section*{Proofs}
\addcontentsline{toc}{section}{\fontsz{14pt}{Proofs}}

\newcounter{savedsection}
\setcounter{savedsection}{\value{section}}

\renewcommand{\thesection}{\XAlph{11}$'$}    %

In the rest of this document, we prove the results stated above, with the same sectioning.

\section{Properties of the Declarative System}

\declarativewellfounded*
\begin{proof}\proofheading
  Let $|e|$ be the size of the expression $e$.
  Let $|s|$ be the size of the spine $s$.
  Let $|\Pi|$ be the size of the branch list $\Pi$.
  Let $\largecount{A}$ be the number of ``large'' connectives $\AllSym$,
  $\ExistsSym$, $\implies$, $\with$ in $A$.

  First, stratify judgments by the size of the term (expression, spine, or branches),
  and say that a judgment is \emph{at} $n$ if it types a term of size $n$.
  Order the main judgment forms as follows:
  \[
  \begin{array}{r@{~~}l@{~~}l}
    & \text{synthesis judgment at $n$}
    \\
    <& \text{checking judgments at $n$}
    \\
    <& \text{ordinary spine judgment at $n$}
    \\
    <& \text{recovery spine judgment at $n$}
    \\
    <& \text{match judgments at $n$}
    \\[1ex]
    <& \text{synthesis judgment at $n + 1$}
    \\
    \vdots
  \end{array}
  \]

  Within the checking judgment forms at $n$, we compare types lexicographically, first
  by the number of large connectives, and then by the ordinary size.
  Within the match judgment forms at $n$, we compare using a lexicographic order of,
  first, $\largecount{\Avec}$; second, the judgment form, considering the match
  judgment to be smaller than the matchelim judgment; third, the size of $\Avec$.
  These criteria order the judgments as follows:
  \[
  \begin{array}{r@{~~}r@{~~}ll}
    && \text{synthesis judgment at $n$}
    \\
    <&& \!\!\big(\text{checking judgment at $n$ with $\largecount{A} = 1$}
    \\
    &<& \text{checkelim judgment at $n$ with $\largecount{A} = 1$}
    \\
    &<& \text{checking judgment at $n$ with $\largecount{A} = 2$}
    \\
    &<& \text{checkelim judgment at $n$ with $\largecount{A} = 2$}
    \\
    &<& \dots\big)
    \\[1ex]
    <&& \!\!\big(\text{match judgment at $n$ with $\largecount{\Avec} = 1$ and $\Avec$ of size 1}
    \\
    &<& \text{match judgment at $n$ with $\largecount{\Avec} = 1$ and $\Avec$ of size 2}
    \\
    &<& \text{matchelim judgment at $n$ with $\largecount{\Avec} = 1$}
    \\
    &<& \text{match judgment at $n$ with $\largecount{\Avec} = 2$ and $\Avec$ of size 1}
    \\
    &<& \text{match judgment at $n$ with $\largecount{\Avec} = 2$ and $\Avec$ of size 2}
    \\
    &<& \text{matchelim judgment at $n$ with $\largecount{\Avec} = 2$}
    \\
    &<& \dots\big)
  \end{array}
\]
The class of ordinary spine judgments at 1 need not be refined, because
the only ordinary spine rule applicable to a spine of size 1 is \DeclEmptySpine,
which has no premises; rules \DeclAllSpine, \DeclImpliesSpine, and \DeclArrSpine
are restricted to non-empty spines and can only apply to larger terms.

Similarly, the class of match judgments at 1 need not be refined,
because only \DeclMatchEmpty is applicable.

Note that we distinguish the ``checkelim'' form
 $\declchkjudg p {\Psi \ctxsep P} e A$ of the checking judgment.
We also define the size of an expression $e$ to consider all types in annotations
to be of the same size, that is,
\[
   |(e : A)|  ~=~  |e| + 1
\]
Thus, $|\theta(e)| = |e|$, even when $e$ has annotations.
This is used for \DeclCheckUnify; see below.

We assume that coverage, which does not depend on any other
typing judgments, is well-founded.  We likewise assume that subtyping,
$\judgetp \Psi A$,
$\judge \Psi \tau \sort$,
and
$\judgeprop \Psi P$
are well-founded.

We now show that, for each class of judgments, every judgment in that class
depends only on smaller judgments.

\begin{itemize}
\newcommand{\CLAIM}{\textbf{Claim:}~}
\newcommand{\PROOF}{\textit{Proof.}~~}

\item \textbf{Synthesis judgments}
  
       \CLAIM
       For all $n$, synthesis at $n$ depends only on judgments at $n - 1$ or less.

       \PROOF
          Rule \DeclVar has no premises.
          \\
          Rule \DeclAnno depends on a premise at a strictly smaller term.
          \\
          Rule \DeclArrElim depends on (1) a synthesis premise at a strictly smaller
          term, and (2) a recovery spine judgment at a strictly smaller term.

\item \textbf{Checking judgments}

        \CLAIM
         For all $n \geq 1$,
           the checking judgment over terms of size $n$ with type of size $m$
           depends only on

           \begin{enumerate}[(1)]
           \item %
               synthesis judgments at size $n$ or smaller, and
           \item %
               checking judgments at size $n - 1$ or smaller, and
           \item %
               checking judgments at size $n$ with fewer large connectives, and
           \item %
               checkelim judgments at size $n$ with fewer large connectives, and
           \item %
               match judgments at size $n - 1$ or smaller.
           \end{enumerate}

        \PROOF
          Rule \DeclSub depends on a synthesis judgment of size $n$. (1) \\
          Rule \DeclUnitIntro has no premises. \\
          Rule \DeclAllIntro depends on a checking judgment at $n$ with fewer large connectives. (3) \\
          Rule \DeclExIntro depends on a checking judgment at $n$ with fewer large connectives. (3) \\             Rule \DeclWithIntro depends on a checking judgment at $n$ with fewer large connectives. (3) \\
          Rule \DeclImpliesIntro
            depends on a checkelim judgment at $n$ with fewer large connectives. (4) \\
          Rules \DeclArrIntro, \DeclRec, \DeclSumIntro{k}, \DeclPairIntro,
          and \DeclCons depend on checking judgments at size $< n$. (2) \\
          Rule \DeclNil depends only on an auxiliary judgment. \\
          Rule \DeclCase depends on:

          \begin{itemize}
          \item a synthesis judgment at size $n$ (1),
          \item a match judgment at size $< n$ (5), and
          \item a coverage judgment.
          \end{itemize}

\item \textbf{Checkelim judgments}

        \CLAIM
           For all $n \geq 1$,
             the checkelim judgment $\declchkjudg p {\Psi \ctxsep P} e A$ over terms of size $n$
             depends only on
             checking judgments at size $n$,
               with a type $A'$ such that $\largecount{A'} = \largecount{A}$.

        \PROOF
          Rule \DeclCheckBot has no nontrivial premises. \\
          Rule \DeclCheckUnify depends on a checking judgment:
             Since $|\theta(e)| = |e|$, this checking judgment is at $n$.
             Since the mgu $\theta$ is over monotypes, $\largecount{\theta(A)} = \largecount{A}$.

\item \textbf{Ordinary spine judgments}

        An ordinary spine judgment at 1 depends on no other judgments:
        the only spine of size 1 is the empty spine, so only \DeclEmptySpine applies,
        and it has no premises.

        \CLAIM
          For all $n \geq 2$,
            the ordinary spine judgment $\declspinejudg \Psi s A p C q$
            over spines of size $n$
            depends only on

            \begin{enumerate}[(a)]
            \item %
              checking judgments at size $n - 1$ or smaller, and
            \item %
              ordinary spine judgments at size $n - 1$ or smaller, and
            \item %
              ordinary spine judgments at size $n$ with strictly smaller \largecount{A}.
            \end{enumerate}

        \PROOF
          Rule \DeclAllSpine depends on an ordinary spine judgment of size $n$,
              with a type that has fewer large connectives. (c) \\
          Rule \DeclImpliesSpine depends on an ordinary spine judgment of size $n$,
              with a type that has fewer large connectives. (c) \\
          Rule \DeclEmptySpine has no premises. \\
          Rule \DeclArrSpine depends on a checking judgment of size $n - 1$ or smaller (a)
            and an ordinary spine judgment of size $n - 1$ or smaller (b).

\item \textbf{Recovery spine judgments}

      \CLAIM For all $n$, the recovery spine judgment at $n$ depends only on ordinary
      spine judgments at $n$.

      \PROOF Rules \DeclRecover and \DeclPass depend only on ordinary
      spine judgments at $n$.

\item \textbf{Match judgments}

        \CLAIM
          For all $n~ \geq 1$,
            the match judgment $\declmatchjudg p \Psi \Pi \Avec C$ over $\Pi$ of size $n$
            depends only on

            \begin{enumerate}[(a)]
            \item %
              checking judgments at size $n - 1$ or smaller, and
            \item %
              match judgments at size $n - 1$ or smaller, and
            \item %
              match judgments at size $n$ with smaller $\Avec$, and
            \item %
              matchelim judgments at size $n$ with fewer large connectives in $\Avec$.
            \end{enumerate}

        \PROOF
          Rule \DeclMatchEmpty has no premises. \\
          Rule \DeclMatchSeq depends on match judgments at $n - 1$ or smaller (b). \\
          Rule \DeclMatchBase depends on a checking judgment at $n - 1$ or smaller (a). \\
          Rules \DeclMatchUnit, \DeclMatchPair, \DeclMatchSum{k},
          \DeclMatchNeg, and \DeclMatchWild
             depend on match judgments at $n - 1$ or smaller (b). \\
          Rule \DeclMatchExists depends on a match judgment at size $n$ with smaller $\Avec$ (c). \\
          Rule \DeclMatchWith depends on an matchelim judgment at n,
           with fewer large connectives in $\Avec$. (d)

\item \textbf{Matchelim judgments}

        \CLAIM
          For all $n \geq 1$,
            the matchelim judgment $\declmatchelimjudg p \Psi P \Pi \Avec C$ over $\Psi$ of size $n$
            depends only on
            match judgments with the same number of large connectives in $\Avec$.

        \PROOF
          Rule \DeclMatchBot has no nontrivial premises. \\
          Rule \DeclMatchUnify depends on a match judgment with the same number of
            large connectives (similar to \DeclCheckUnify, considered above).
\qedhere
\end{itemize}
\end{proof}

\declarativeweakening*
\begin{proof}\proofheading
  By induction on the derivation.
\end{proof}

\declarativetermsubstitution*
\begin{proof}\proofheading
  By induction on the derivation of the substitutee.
\end{proof}

\declarativereflexivity*
\begin{proof}\proofheading
  By induction on $A$, writing $p$ for the sign of the subtyping judgment.

  Our induction metric is the number of quantifiers on the outside of $A$,
  plus one if the polarity of $A$ and the subtyping judgment do not match up
  (that is, if $\Neg{A}$ and $\polvar = +$, or $\Pos{A}$ and $\polvar = -$).

  \begin{itemize}
    \ProofCaseThing{\nonPos{A}, \nonNeg{A}}

    By rule $\DsubReflPm$.

    \ProofCaseThing{A = \extype{b:\sort}{B}  \AND  \polvar = +}
    
    \begin{llproof}
      \declsubjudgPf{+}{\Psi, b:\sort}{B}{B}{By i.h.\ (one less quantifier)}
      \judgePf{\Psi, b:\sort}{b}{\sort}{By rule $\DeclUvarSort$}
      \declsubjudgPf{+}{\Psi, b:\sort}{B}{\extype{b:\sort}{B}}{By rule $\DsubExistsR$}
      \declsubjudgPf{+}{\Psi}{\extype{b:\sort}{B}}{\extype{b:\sort}{B}}{By rule $\DsubExistsL$}
    \end{llproof}

    \ProofCaseThing{A = \extype{b:\sort}{B}  \AND  \polvar = -}
    
    \begin{llproof}
      \declsubjudgPf{+}{\Psi}{\extype{b:\sort}{B}}{\extype{b:\sort}{B}}
           {By i.h.\ (polarities match)}
      \declsubjudgPf{-}{\Psi}{\extype{b:\sort}{B}}{\extype{b:\sort}{B}}{By \DsubNegPos}
    \end{llproof}

    \ProofCaseThing{A = \alltype{b:\sort}{B}  \AND  \polvar = +}
    
    \begin{llproof}
      \declsubjudgPf{-}{\Psi}{\alltype{b:\sort}{B}}{\alltype{b:\sort}{B}}
           {By i.h.\ (polarities match)}
      \declsubjudgPf{+}{\Psi}{\alltype{b:\sort}{B}}{\alltype{b:\sort}{B}}{By \DsubPosNeg}
    \end{llproof}

    \ProofCaseThing{A = \alltype{b:\sort}{B}  \AND  \polvar = -}

    \begin{llproof}
      \declsubjudgPf{-}{\Psi, b:\sort}{B}{B}{By i.h.\ (one less quantifier)}
      \judgePf{\Psi, b:\sort}{b}{\sort}{By rule \DeclUvarSort}
      \declsubjudgPf{-}{\Psi, b:\sort}{\alltype{b:\sort}{B}}{B}{By rule \DsubAllL}
      \declsubjudgPf{-}{\Psi}{\alltype{b:\sort}{B}}{\alltype{b:\sort}{B}}{By rule \DsubAllR}
    \end{llproof}
  \qedhere
  \end{itemize}
\end{proof}

\declarativesubtypinginversion*
\begin{proof}\proofheading
By a routine induction on the subtyping derivations.
\end{proof}

\declarativepolarityflip*
\begin{proof}\proofheading
  By a routine induction on the subtyping derivations.
\end{proof}

\declarativetransitivity*
\begin{proof}\proofheading
  By lexicographic induction on (1) the sum of head quantifiers in $A$, $B$, and $C$,
  and (2) the size of the derivation. 

  We begin by case analysis on the shape of $B$, and the polarity of subtyping:

  \begin{itemize}
  \item Case $B = \alltype{\beta:\sort_2}{B'}$, polarity = $-$:

    We case-analyze $\Dee_1$:

    \begin{itemize}
      \DerivationProofCase{\DsubAllL}
              {\judge{\Psi}{\tau}{\sort_1}
                \\
                \declsubjudg[-]{\Psi}{[\tau/\alpha]A'}{B}}
              {\declsubjudg[-]{\Psi}{\alltype{\alpha : \sort_1}{A'}}{B}}

         \begin{llproof}
           \decljudgePf{\Psi}{\tau}{\sort_1}{Subderivation}
           \declsubjudgPf{-}{\Psi}{[\tau/\alpha]A'}{B}{Subderivation}
           \declsubjudgPf{-}{\Psi}{B}{C}{Given}
           \declsubjudgPf{-}{\Psi}{[\tau/\alpha]A'}{C}{By i.h.\ ($A$ lost a quantifier)}
           \declsubjudgPf{-}{\Psi}{A}{C}{By rule \DsubAllL}
         \end{llproof}

    \DerivationProofCase{\DsubAllR}
          {\declsubjudg[-]{\Psi, \beta : \sort_2}{A}{B'}}
          {\declsubjudg[-]{\Psi}{A}{\alltype{\beta : \sort_2}{B'}}}

      We case-analyze $\Dee_2$:

      \begin{itemize}
      \DerivationProofCase{\DsubAllL}
              {\judge{\Psi}{\tau}{\sort_2}
                \\
                \declsubjudg[-]{\Psi}{[\tau/\beta]B'}{C}}
              {\declsubjudg[-]{\Psi}{\alltype{\beta : \sort_2}{B'}}{C}}

       \begin{llproof}
         \declsubjudgPf{-}{\Psi, \beta : \sort_2}{A}{B'}
                       {By \Lemmaref{lem:declarative-subtyping-inversion} on $\Dee_1$}
         \judgePf{\Psi}{\tau}{\sort_2}{Subderivation}
         \declsubjudgPf{-}{\Psi}{[\tau/\beta]B'}{C}{Subderivation of $\Dee_2$}
         \declsubjudgPf{-}{\Psi}{A}{[\tau/\beta]B'}{By \Lemmaref{lem:declarative-term-substitution}}
         \declsubjudgPf{-}{\Psi}{A}{C}{By i.h.\ ($B$ lost a quantifier)}
       \end{llproof}

       \DerivationProofCase{\DsubAllR}
          {\declsubjudg[-]{\Psi, c : \sort_3}{B}{C'}}
          {\declsubjudg[-]{\Psi}{B}{\alltype{c : \sort_3}{C'}}}
          
       \begin{llproof}
         \declsubjudgPf{-}{\Psi}{A}{B}{Given}
         \declsubjudgPf{-}{\Psi, c : \sort_3}{A}{B}{By \Lemmaref{lem:declarative-weakening}}
         \proofsep
         \declsubjudgPf{-}{\Psi, c : \sort_3}{B}{C'}{Subderivation}
         \proofsep
         \declsubjudgPf{-}{\Psi, c : \sort_3}{A}{C'}{By i.h.\ ($C$ lost a quantifier)}
         \declsubjudgPf{-}{\Psi}{B}{\alltype{c : \sort_3}{C'}}{By \DsubAllR}
       \end{llproof}
     \end{itemize}
    \end{itemize}

  \item Case $\nonPos{B}$, polarity = $+$:

    Now we case-analyze $\Dee_1$:
    
    \begin{itemize}
    \DerivationProofCase{\DsubExistsL}
              {\declsubjudg[+]{\Psi, \alpha:\tau}{A'}{B}}
              {\declsubjudg[+]{\Psi}{\underbrace{\extype{\alpha : \sort_1}{A'}}_{A}}{B}}

      \begin{llproof}
        \declsubjudgPf{+}{\Psi, \alpha:\tau}{A'}{B}{Subderivation}
        \declsubjudgPf{+}{\Psi, \alpha:\tau}{B}{C}{By \Lemmaref{lem:declarative-weakening} ($\Dee_2$)}
        \declsubjudgPf{+}{\Psi, \alpha:\tau}{A'}{C}{By i.h.\ ($A$ lost a quantifier)}
        \declsubjudgPf{+}{\Psi}{\extype{\alpha : \sort_1}{A'}}{C}{By \DsubExistsL}
      \end{llproof}

    \DerivationProofCase{\DsubPosNeg}
         {\declsubjudg[-]{\Psi}{A}{B} \\ \nonPos{A} \\ \nonPos{B}}
         {\declsubjudg[+]{\Psi}{A}{B}}

     Now we case-analyze $\Dee_2$:
     
     \begin{itemize}
       \DerivationProofCase{\DsubExistsR}
          {\judge{\Psi}{\tau}{\sort_3} \\
           \declsubjudg[+]{\Psi}{B}{[\tau/c]C'}}
          {\declsubjudg[+]{\Psi}{B}{\underbrace{\extype{c : \sort_3}{C'}}_C}}

          \begin{llproof}
            \declsubjudgPf{+}{\Psi}{A}{B}{Given}
            \judgePf{\Psi}{\tau}{\sort_3}{Subderivation of $\Dee_2$}
            \declsubjudgPf{+}{\Psi}{B}{[\tau/c]C'}{Subderivation of $\Dee_2$}
            \declsubjudgPf{+}{\Psi}{A}{[\tau/c]C'}{By i.h.\ ($C$ lost a quantifier)}
            \declsubjudgPf{+}{\Psi}{A}{\extype{c : \sort_3}{C'}}{By \DsubExistsR}
          \end{llproof}

      \DerivationProofCase{\DsubPosNeg}
         {\declsubjudg[-]{\Psi}{B}{C} \\ \nonPos{B} \\ \nonPos{C}}
         {\declsubjudg[+]{\Psi}{B}{C}}

         \begin{llproof}
           \declsubjudgPf{-}{\Psi}{A}{B}{Subderivation of $\Dee_1$}
           \declsubjudgPf{-}{\Psi}{B}{C}{Subderivation of $\Dee_2$}
           \declsubjudgPf{-}{\Psi}{A}{C}{By i.h.\ ($\Dee_1$ and $\Dee_2$ smaller)}
           \SincePf{\nonPos{A}}{Subderivation of $\Dee_1$}
           \SincePf{\nonPos{C}}{Subderivation of $\Dee_2$}
           \declsubjudgPf{+}{\Psi}{A}{C}{By \DsubPosNeg}
         \end{llproof}
     \end{itemize}
    \end{itemize}

  \item Case $B = \extype{\beta:\sort_2}{B'}$, polarity = $+$:

    Now we case-analyze $\Dee_2$:

    \begin{itemize}
      \DerivationProofCase{\DsubExistsR}
              {\judge{\Psi}{\tau}{\sort_3}
                \\
                \declsubjudg[+]{\Psi}{B}{[\tau/\alpha]C'}}
              {\declsubjudg[+]{\Psi}{B}{\underbrace{\extype{\alpha : \sort_3}{C'}}_C}}

         \begin{llproof}
           \decljudgePf{\Psi}{\tau}{\sort_3}{Subderivation of $\Dee_2$}
           \declsubjudgPf{+}{\Psi}{B}{[\tau/\alpha]C'}{Subderivation of $\Dee_2$}
           \declsubjudgPf{+}{\Psi}{A}{B}{Given}
           \declsubjudgPf{+}{\Psi}{A}{[\tau/\alpha]C'}{By i.h.\ ($C$ lost a quantifier)}
           \declsubjudgPf{+}{\Psi}{A}{C}{By rule \DsubExistsR}
         \end{llproof}

    \DerivationProofCase{\DsubExistsL}
          {\declsubjudg[+]{\Psi, \beta : \sort_2}{B'}{C}}
          {\declsubjudg[+]{\Psi}{\extype{\beta : \sort_2}{B'}}{C}}
          
      Now we case-analyze $\Dee_1$:

      \begin{itemize}
      \DerivationProofCase{\DsubExistsR}
              {\judge{\Psi}{\tau}{\sort_2}
                \\
                \declsubjudg[+]{\Psi}{A}{[\tau/\beta]B'}}
              {\declsubjudg[+]{\Psi}{A}{\underbrace{\extype{\beta : \sort_2}{B'}}_B}}

       \begin{llproof}
         \declsubjudgPf{+}{\Psi, \beta : \sort_2}{B'}{C}{Subderivation of $\Dee_2$}
         \judgePf{\Psi}{\tau}{\sort_2}{Subderivation of $\Dee_1$}
         \declsubjudgPf{+}{\Psi}{A}{[\tau/\beta]B'}{Subderivation of $\Dee_1$}
         \declsubjudgPf{+}{\Psi}{[\tau/\beta]B'}{C}{By \Lemmaref{lem:declarative-term-substitution}}
         \declsubjudgPf{+}{\Psi}{A}{C}{By i.h.\ ($B$ lost a quantifier)}
       \end{llproof}

       \DerivationProofCase{\DsubExistsL}
          {\declsubjudg[+]{\Psi, \alpha : \sort_1}{A}{B}}
          {\declsubjudg[+]{\Psi}{\underbrace{\extype{\alpha : \sort_1}{A'}}_A}{B}}
          
       \begin{llproof}
         \declsubjudgPf{+}{\Psi}{B}{C}{Given}
         \declsubjudgPf{+}{\Psi, \alpha : \sort_1}{A'}{B}{Subderivation of $\Dee_1$}
         \declsubjudgPf{+}{\Psi, \alpha : \sort_1}{A'}{B}{By \Lemmaref{lem:declarative-weakening}}
         \declsubjudgPf{+}{\Psi, \alpha : \sort_1}{A'}{C}{By i.h.\ ($A$ lost a quantifier)}
         \declsubjudgPf{+}{\Psi}{\extype{\alpha : \sort_1}{A'}}{C}{By \DsubExistsL}
       \end{llproof}
     \end{itemize}
    \end{itemize}

  \item Case $\nonNeg{B}$, polarity = $-$:

    We case-analyze $\Dee_2$:

    \begin{itemize}
    \DerivationProofCase{\DsubAllR} %
              {\declsubjudg[+]{\Psi, c:\sort_3}{B}{C'}}
              {\declsubjudg[+]{\Psi}{B}{\underbrace{\extype{c : \sort_3}{C'}}_{C}}}

      \begin{llproof}
        \declsubjudgPf{+}{\Psi, c:\sort_3}{B}{C'}{Subderivation of $\Dee_2$}
        \declsubjudgPf{+}{\Psi, c:\sort_3}{A}{B}{By \Lemmaref{lem:declarative-weakening}}
        \declsubjudgPf{+}{\Psi, c:\sort_3}{A}{C'}{By i.h.\ ($C$ lost a quantifier)}
        \declsubjudgPf{+}{\Psi}{A}{\alltype{c : \sort_3}{C'}}{By \DsubAllR}
      \end{llproof}

    \DerivationProofCase{\DsubNegPos} %
         {\declsubjudg[+]{\Psi}{B}{C} \\ \nonNeg{B} \\ \nonNeg{C}}
         {\declsubjudg[-]{\Psi}{B}{C}}

     We case-analyze $\Dee_1$:

     \begin{itemize}
       \DerivationProofCase{\DsubAllL}    %
          {\judge{\Psi}{\tau}{\sort_1} \\
           \declsubjudg[-]{\Psi}{[\tau/\alpha]A'}{B} }
          {\declsubjudg[-]{\Psi}{\underbrace{\alltype{\alpha : \sort_1}{A'}}_A}{B} }

          \begin{llproof}
            \declsubjudgPf{-}{\Psi}{B}{C}{Given}
            \judgePf{\Psi}{\tau}{\sort_1}{Subderivation of $\Dee_1$}
            \declsubjudgPf{-}{\Psi}{[\tau/\alpha]A'}{B}{Subderivation of $\Dee_1$}
            \declsubjudgPf{-}{\Psi}{[\tau/\alpha]A'}{C}{By i.h.\ ($A$ lost a quantifier)}
            \declsubjudgPf{-}{\Psi}{\alltype{\alpha : \sort_1}{A'}}{C}{By \DsubAllL}
          \end{llproof}

      \DerivationProofCase{\DsubNegPos}    %
         {\declsubjudg[+]{\Psi}{A}{B} \\ \nonPos{A} \\ \nonPos{B}}
         {\declsubjudg[-]{\Psi}{A}{B}}

         \begin{llproof}
           \declsubjudgPf{+}{\Psi}{A}{B}{Subderivation of $\Dee_1$}
           \declsubjudgPf{+}{\Psi}{B}{C}{Subderivation of $\Dee_2$}
           \declsubjudgPf{+}{\Psi}{A}{C}{By i.h.\ ($\Dee_1$ and $\Dee_2$ smaller)}
           \SincePf{\nonNeg{A}}{Subderivation of $\Dee_2$}
           \SincePf{\nonNeg{C}}{Subderivation of $\Dee_2$}
           \declsubjudgPf{-}{\Psi}{A}{C}{By \DsubNegPos}
         \end{llproof}
     \end{itemize}
    \end{itemize}

  \end{itemize}
\end{proof}

\section{Substitution and Well-formedness Properties}

\substitutionwf*
\begin{proof}\proofheading
  By induction on the derivations of $\judgetp{\Gamma}{A\;p}$ and
  $\judgeprop{\Gamma}{P}$.
\end{proof}

\ctxappuvarpreservation*
\begin{proof}\proofheading
  By induction on $\Omega$, following the definition of context application (\Figureref{mainfig:context-substitution}).
\end{proof}

\sortingimpliestyping*
\begin{proof}\proofheading

  By induction on the given derivation.  All cases are straightforward.
\end{proof}

\substitutionsort*
\begin{proof}\proofheading
  By induction on $\typesize{\Gamma}{t}$ (the size of $t$ under $\Gamma$).
  
  \begin{itemize}
      \ProofCasesRules{\UnitSort}
           Here $t = \unitty$, so applying $\Gamma$ to $t$ does not change it: $t = [\Gamma]t$.
           Since $\judge{\Gamma}{t}{\sort}$, we have $\judge{\Gamma}{[\Gamma]t}{\sort}$, which was to be shown.
      
      \ProofCaseRule{\VarSort}
           If $t$ is an existential variable $\ahat$, then $\Gamma = \Gamma_0[\ahat]$, so applying $\Gamma$ to $t$ does not change it,
           and we proceed as in the \UnitSort case above.
           
           If $t$ is a universal variable $\alpha$ and $\Gamma$ has no equation for it, then proceed as in the \UnitSort case.
           
           Otherwise, $t = \alpha$ and $(\hypeq{\alpha}{\tau}) \in \Gamma$:
           \[
              \Gamma ~~=~~ (\Gamma_L, \alpha : \sort, \Gamma_M, \hypeq{\alpha}{\tau}, \Gamma_R)
           \]
           By the implicit assumption that $\Gamma$ is well-formed, $\judge{\Gamma_L, \alpha : \sort, \Gamma_M}{\tau}{\sort}$. \\
           By \Lemmaref{lem:suffix-weakening-judge}, $\judge{\Gamma}{\tau}{\sort}$.
           Since $\typesize{\Gamma}{\tau} < \typesize{\Gamma}{\alpha}$,
           we can apply the i.h., giving
           \[
              \judge{\Gamma}{[\Gamma]\tau}{\sort}
           \]
           By the definition of substitution, $[\Gamma]\tau = [\Gamma] \alpha$,
           so we have $\judge{\Gamma}{[\Gamma]\alpha}{\sort}$.

      \ProofCaseRule{\SolvedVarSort}
           In this case $t = \ahat$ and $\Gamma = (\Gamma_L, \hypeq{\ahat}{\tau}, \Gamma_R)$.
           Thus $[\Gamma]t = [\Gamma]\ahat = [\Gamma_L]\tau$.  We assume contexts are well-formed,
           so all free variables in $\tau$ are declared in $\Gamma_L$.  Consequently,
           $\typesize{\Gamma_L}{\tau} = \typesize{\Gamma}{\tau}$,
           which is less than $\typesize{\Gamma}{\ahat}$.  We can therefore apply the i.h.\ to $\tau$, yielding
           $\judge{\Gamma}{[\Gamma]\tau}{\sort}$.  By the definition of substitution, $[\Gamma]\tau = [\Gamma]\ahat$,
           so we have $\judge{\Gamma}{[\Gamma]\ahat}{\sort}$.

      \ProofCaseRule{\BinSort}
           In this case $t = t_1 \binc t_2$.
           By i.h., $\judge{\Gamma}{[\Gamma]t_1}{\sort}$ and $\judge{\Gamma}{[\Gamma]t_2}{\sort}$.
           By \BinSort, $\judge{\Gamma}{([\Gamma]t_1) \binc ([\Gamma]t_2)}{\sort}$,
           which by the definition of substitution is
           $\judge{\Gamma}{[\Gamma](t_1 \binc t_2)}{\sort}$.
  \qedhere
  \end{itemize}
\end{proof}

\substitutionprop*
\begin{proof}\proofheading
  Use inversion (\EqProp), apply \Lemmaref{lem:substitution-sort} to each premise, and apply \EqProp again.
\end{proof}

\substitutiontp*
\begin{proof}\proofheading
  By induction on $\typesize{\Gamma}{A}$ (the size of $A$ under $\Gamma$).  

  Several cases correspond to cases in the proof of \Lemmaref{lem:substitution-sort}:

  \begin{itemize}
  \item the case for \UnitWF is like the case for \UnitSort;
  \item the case for \SolvedVarSort is like the cases for \VarWF and \SolvedVarWF,
  \item the case for \VarSort is like the case for \VarWF,
    but in the last subcase, apply \Lemmaref{lem:sorting-implies-typing} to move from a sorting judgment
    to a typing judgment.
  \item the case for \BinWF is like the case for \BinSort.
  \end{itemize}

  Now, the new cases:

  \begin{itemize}
      \ProofCaseRule{\ForallWF}
          In this case $A = \alltype{\alpha:\sort} A_0$.  By i.h., $\judgetp{\Gamma, \alpha:\sort}{[\Gamma, \alpha:\sort]A_0}$.
          By the definition of substitution, $[\Gamma, \alpha:\sort]A_0 = [\Gamma]A_0$,
          so by \ForallWF, $\judgetp{\Gamma}{\alltype{\alpha} [\Gamma]A_0}$,
          which by the definition of substitution is
          $\judgetp{\Gamma}{[\Gamma](\alltype{\alpha} A_0)}$.
    
     \ProofCaseRule{\ExistsWF}  Similar to the \ForallWF case.

     \ProofCaseRule{\ImpliesWF, \WithWF}  Use the i.h.\ and \Lemmaref{lem:substitution-prop}, then
        apply \ImpliesWF or \WithWF.
     \qedhere
  \end{itemize}
\end{proof}

\completionsort*
\begin{proof}\proofheading
  By induction on $\typesize{\Omega}{t}$ (the size of $t$ under $\Omega$).
  
  \begin{itemize}
      \DerivationProofCase{\VarSort}
            {u:\sort \in \Omega}
            {\judge{\Omega}{u}{\sort}}

            We have a complete context $\Omega$, so $u$ cannot be an existential variable:
            it must be some universal variable $\alpha$.

            If $\Omega$ lacks an equation for $\alpha$,
            use \Lemmaref{lem:ctxapp-uvar-preservation}
            and apply rule \DeclUvarSort.

            Otherwise, $(\hypeq{\alpha}{\tau} \in \Omega$, so we need to show
            $\judge{\Omega}{[\Omega]\tau}{\sort}$.  By the implicit assumption that
            $\Omega$ is well-formed, plus \Lemmaref{lem:suffix-weakening-judge},
            $\judge{\Omega}{\tau}{\sort}$.  By \Lemmaref{lem:substitution-sort},
            $\judge{\Omega}{[\Omega]\tau}{\sort}$.
     
      \DerivationProofCase{\SolvedVarSort}
            {\hypeq{\ahat:\sort}{\tau} \in \Omega }
            { \judge{\Omega}{\ahat}{\sort} }
            
            \begin{llproof}
              \inPf{\hypeq{\ahat:\sort}{\tau}} {\Omega}   {Subderivation}
              \eqPf{\Omega}  {(\Omega_L, \hypeq{\ahat:\sort}{\tau}, \Omega_R)}   {Decomposing $\Omega$}
              \PfLEFT{3}
              \judgePf{\Omega_L} {\tau} {\sort}   {By implicit assumption that $\Omega$ is well-formed}
              \judgePf{\Omega_L, \hypeq{\ahat:\sort}{\tau}, \Omega_R} {\tau} {\sort}   {By \Lemmaref{lem:suffix-weakening-judge}}
              \judgePf{\Omega} {[\Omega]\tau} {\sort}   {By \Lemmaref{lem:substitution-sort}}
          \Hand  \judgePf{[\Omega]\Omega} {[\Omega]\ahat} {\sort}   {$[\Omega]\tau = [\Omega]\ahat$}
            \end{llproof}
      
      \DerivationProofCase{\UnitSort}
            { }
            { \judge{\Omega}{\unitty}{\type} }
            
            Since $\unitty = [\Omega]\unitty$, applying \DeclUnitSort gives the result.

      \DerivationProofCase{\BinSort}
            { \judge{\Omega}{\tau_1}{\type}
              \\
              \judge{\Omega}{\tau_2}{\type}
            }
            { \judge{\Omega}{\tau_1 \binc \tau_2}{\type} }
            
            By i.h.\ on each premise, rule \DeclBinSort, and the definition of substitution.

      \DerivationProofCase{\ZeroSort}
            { }
            {\judge{\Omega}{\zero}{\ind}}

            Since $\zero = [\Omega]\zero$, applying \DeclZeroSort gives the result.
      
      \DerivationProofCase{\SuccSort}
            {\judge{\Omega}{t}{\ind} }
            {\judge{\Omega}{\succ{t}}{\ind}}

            By i.h., rule \DeclSuccSort, and the definition of substitution.
      \qedhere
  \end{itemize}
\end{proof}

\completionprop*
\begin{proof}\proofheading
  Only one rule derives this judgment form:

  \begin{itemize}
    \DerivationProofCase{\EqProp}
          {\judge{\Omega}{t}{\ind} \\
           \judge{\Omega}{t'}{\ind} 
          }
          {\judgeprop{\Omega}{t = t'}}

          \begin{llproof}
            \judgePf{\Omega}{t}{\ind}    {Subderivation}
            \judgePf{[\Omega]\Omega}{[\Omega]t}{\ind}    {By \Lemmaref{lem:completion-sort}}
            \judgePf{\Omega}{t'}{\ind}    {Subderivation}
            \judgePf{[\Omega]\Omega}{[\Omega]t'}{\ind}    {By \Lemmaref{lem:completion-sort}}
            \judgepropPf{[\Omega]\Omega}{([\Omega]t) = ([\Omega]t')}    {By \EqProp}
\Hand       \judgepropPf{[\Omega]\Omega}{[\Omega](t = t')}    {By def.\ of subst.} 
          \end{llproof}
  \qedhere
  \end{itemize}
\end{proof}

\completionwf*
\begin{proof}\proofheading
  By induction on $\typesize{\Omega}{A}$.

  Several cases correspond to those in the proof of \Lemmaref{lem:completion-sort}:

  \begin{itemize}
  \item the \UnitWF case is like the \UnitSort case (using \DeclUnitWF instead of \DeclUnitSort);
  \item the \VarWF case is like the \VarSort case (using \DeclUvarWF instead of \DeclUvarSort);
  \item the \SolvedVarWF case is like the \SolvedVarSort case.
  \end{itemize}
  
  However, uses of \Lemmaref{lem:substitution-sort} are replaced by uses of \Lemmaref{lem:substitution-tp}.
  
  Now, the new cases:
  
  \begin{itemize}
      \DerivationProofCase{\ForallWF}
            {\judgetp{\Omega, \alpha : \sort} {A_0}}
            {\judgetp{\Omega}{\alltype{\alpha : \sort} A_0}}

            \begin{llproof}
              \judgePf{\Omega, \alpha : \sort} {A_0} {\sort'}   {Subderivation}
              \judgePf{[\Omega, \alpha : \sort](\Omega, \alpha : \sort)} {[\Omega]A_0} {\sort'}   {By i.h.}
              \judgePf{[\Omega]\Omega, \alpha : \sort} {[\Omega]A_0} {\sort'}   {By definition of completion}
              \judgePf{[\Omega]\Omega} {\alltype{\alpha : \sort} [\Omega]A_0} {\sort'}   {By \DeclForallWF}
\Hand         \judgePf{[\Omega]\Omega} {[\Omega](\alltype{\alpha : \sort} A_0)} {\sort'}   {By def.\ of subst.}
            \end{llproof}

      \ProofCaseRule{\ExistsWF}
           Similar to the \ForallWF case, using \DeclExistsWF instead of \DeclForallWF.
      
      \DerivationProofCase{\BinWF}
            { \judgetp{\Omega}{A_1}
              \\
              \judgetp{\Omega}{A_2}
            }
            { \judgetp{\Omega}{A_1 \binc A_2} }

            By i.h.\ on each premise, rule \DeclBinWF, and the definition of substitution.

       \ProofCaseRule{\VecWF}
            Similar to the \BinWF case.

      \DerivationProofCase{\ImpliesWF}
            {\judgeprop{\Omega}{P} \\ 
             \judgetp{\Omega}{A_0}}
            {\judgetp{\Omega}{P \implies A_0}}

            \begin{llproof}
              \judgepropPf{\Omega}{P}  {Subderivation}
              \judgepropPf{[\Omega]\Omega}{[\Omega]P}  {By \Lemmaref{lem:completion-prop}}
              \proofsep
              \judgetpPf{\Omega}{A_0}  {Subderivation}
              \judgetpPf{[\Omega]\Omega}{[\Omega]A_0}  {By i.h.}
              \proofsep
              \judgetpPf{[\Omega]\Omega}{([\Omega]P) \implies ([\Omega]A_0)}  {By \DeclImpliesWF}
            \Hand \judgetpPf{[\Omega]\Omega}{[\Omega](P \implies A_0)}  {By def.\ of subst.}
            \end{llproof}

      \DerivationProofCase{\WithWF}
            {\judgeprop{\Omega}{P} \\ 
             \judgetp{\Omega}{A_0}}
            {\judgetp{\Omega}{A_0 \with P}}

            Similar to the \ImpliesWF case.
  \qedhere
  \end{itemize}
\end{proof}

\substitutionstability*
\begin{proof}\proofheading
  By induction on $\Omega_Z$.

  Since $\Omega_Z$ is soft, either (1) $\Omega_Z = \cdot$ (and the result is immediate)
  or (2) $\Omega_Z = (\Omega', \ahat : \sort)$
  or (3) $\Omega_Z = (\Omega', \hypeq{\ahat : \sort}{t})$.
  However, according to the grammar for complete contexts such as $\Omega_Z$, (2) is impossible.
  Only case (3) remains.
  
  By i.h., $[\Omega]A = [\Omega, \Omega']A$.
  Use the fact that
  $\judgetp{\Omega}{A}$ implies $\FV{A} \sect \dom{\Omega_Z} = \emptyset$.
\end{proof}

\equaldomainswftp*
\begin{proof}
  By induction on the given derivation.
\end{proof}

\section{Properties of Extension}

\declarationpreservation*
\begin{proof}\proofheading
  By induction on the derivation of $\substextend{\Gamma}{\Delta}$. 

\begin{itemize}
    \DerivationProofCase{\substextendId}
        { }
        { \substextend{\emptyctx}{\emptyctx} }

        This case is impossible, since by hypothesis $u$ is declared in $\Gamma$. 
    
    \DerivationProofCase{\substextendVV}
          {\substextend{\Gamma}{\Delta}
           \\
           [\Delta]A = [\Delta]A'
          }
          {\substextend{\Gamma, x:A}{\Delta, x:A'}}
    
          \begin{itemize}
          \item Case $u = x$:
            Immediate. 

          \item Case $u \neq x$:  Since $u$ is declared in $(\Gamma, x:A)$,
            it is declared in $\Gamma$.  By i.h., $u$ is declared in $\Delta$,
            and therefore declared in $(\Delta, x:A')$.
          \end{itemize}

    \DerivationProofCase{\substextendUU}
        { \substextend{\Gamma}{\Delta} }
        { \substextend{\Gamma, \alpha:\sort}{\Delta, \alpha:\sort} }

        Similar to the \substextendVV case.
    
    \DerivationProofCase{\substextendEE}
        { \substextend{\Gamma}{\Delta} }
        { \substextend{\Gamma, \ahat:\sort}{\Delta, \ahat:\sort} }

        Similar to the \substextendVV case.
    
    \DerivationProofCase{\substextendSolved}
        { \substextend{\Gamma}{\Delta} \\ [\Delta]t = [\Delta]t'}
        { \substextend{\Gamma, \hypeq{\ahat:\sort}{t}}{\Delta, \hypeq{\ahat:\sort}{t'}}
        }

        Similar to the \substextendVV case.
    
    \DerivationProofCase{\substextendEqn}
        { \substextend{\Gamma}{\Delta} \\ [\Delta]t = [\Delta]t'}
        { \substextend{\Gamma, \hypeq{\alpha}{t}}{\Delta, \hypeq{\alpha}{t'}}
        }

        It is given that $u$ is declared in $(\Gamma, \hypeq{\alpha}{t})$.
        Since $\hypeq{\alpha}{t}$ is not a declaration, $u$ is declared in $\Gamma$. \\
        By i.h., $u$ is declared in $\Delta$, and therefore declared in $(\Delta, \hypeq{\alpha}{t'}$.

    \DerivationProofCase{\substextendMonMon}
        { \substextend{\Gamma}{\Delta} }
        { \substextend{\Gamma, \MonnierComma{\ahat}}{\Delta, \MonnierComma{\ahat}} }

        Similar to the \substextendEqn case.

    \DerivationProofCase{\substextendSolve}
        { \substextend{\Gamma}{\Delta} }
        { \substextend{\Gamma, \bhat:\sort'}{\Delta, \hypeq{\bhat:\sort'}{t}} }

        Similar to the \substextendVV case.

    \DerivationProofCase{\substextendAdd}
        { \substextend{\Gamma}{\Delta} }
        { \substextend{\Gamma}{\Delta, \ahat:\sort} }

        It is given that $u$ is declared in $\Gamma$.
        By i.h., $u$ is declared in $\Delta$, and therefore declared in
        $(\Delta, \ahat:\sort)$.
    
    \DerivationProofCase{\substextendAddSolved}
        { \substextend{\Gamma}{\Delta} }
        { \substextend{\Gamma}{\Delta, \hypeq{\ahat:\sort}{t}} }

        Similar to the \substextendAdd case.
  \qedhere
  \end{itemize}
\end{proof}

\declarationorderpreservation*
\begin{proof}\proofheading
   By induction on the derivation of $\substextend{\Gamma}{\Delta}$.

  \begin{itemize}
    \DerivationProofCase{\substextendId}
        { }
        { \substextend{\emptyctx}{\emptyctx} }

        This case is impossible, since by hypothesis $u$ and $v$ are declared in $\Gamma$. 
    
    \DerivationProofCase{\substextendVV}
          {\substextend{\Gamma}{\Delta}
           \\
           [\Delta]A = [\Delta]A'
          }
          {\substextend{\Gamma, x:A}{\Delta, x:A'}}

          Consider whether $v = x$: 
          \begin{itemize}
          \item Case $v = x$:

                It is given that $u$ is declared to the left of $v$ in $(\Gamma, x:A)$,
                so $u$ is declared in $\Gamma$. \\ 
                By \Lemmaref{lem:declaration-preservation}, $u$ is declared in $\Delta$. \\
                Therefore $u$ is declared to the left of $v$ in $(\Delta, x:A')$. 

          \item Case $v \neq x$: 

                Here, $v$ is declared in $\Gamma$.  By i.h.,
                $u$ is declared to the left of $v$ in $\Delta$. \\ 
                Therefore $u$ is declared to the left of $v$ in $(\Delta, x:A')$. 
          \end{itemize}
    
    \DerivationProofCase{\substextendUU}
        { \substextend{\Gamma}{\Delta} }
        { \substextend{\Gamma, \alpha:\sort}{\Delta, \alpha:\sort} }

        Similar to the \substextendVV case.

    \DerivationProofCase{\substextendEE}
        { \substextend{\Gamma}{\Delta} }
        { \substextend{\Gamma, \ahat:\sort}{\Delta, \ahat:\sort} }

        Similar to the \substextendVV case.
    
    \DerivationProofCase{\substextendSolved}
        { \substextend{\Gamma}{\Delta} \\ [\Delta]t = [\Delta]t'}
        { \substextend{\Gamma, \hypeq{\ahat:\sort}{t}}{\Delta, \hypeq{\ahat:\sort}{t'}}
        }

        Similar to the \substextendVV case.

    \DerivationProofCase{\substextendSolve}
        { \substextend{\Gamma}{\Delta} }
        { \substextend{\Gamma, \bhat:\sort'}{\Delta, \hypeq{\bhat:\sort'}{t}} }

        Similar to the \substextendVV case.

    \DerivationProofCase{\substextendEqn}
        { \substextend{\Gamma}{\Delta} \\ [\Delta]t = [\Delta]t'}
        { \substextend{\Gamma, \hypeq{\alpha}{t}}{\Delta, \hypeq{\alpha}{t'}}
        }

        The equation $\hypeq{\ahat}{t}$ does not declare any variables, so
        $u$ and $v$ must be declared in $\Gamma$. \\
        By i.h., $u$ is declared to the left of $v$ in $\Delta$. \\
        Therefore $u$ is declared to the left of $v$ in $\Delta, \hypeq{\ahat:\sort}{t'}$. 

    \DerivationProofCase{\substextendMonMon}
        { \substextend{\Gamma}{\Delta} }
        { \substextend{\Gamma, \MonnierComma{\ahat}}{\Delta, \MonnierComma{\ahat}} }

        Similar to the \substextendEqn case.

    \DerivationProofCase{\substextendAdd}
        { \substextend{\Gamma}{\Delta} }
        { \substextend{\Gamma}{\Delta, \ahat:\sort} }

        By i.h., $u$ is declared to the left of $v$ in $\Delta$. \\ 
        Therefore $u$ is declared to the left of $v$ in $(\Delta, \ahat:\sort)$. 
    
    \DerivationProofCase{\substextendAddSolved}
        { \substextend{\Gamma}{\Delta} }
        { \substextend{\Gamma}{\Delta, \hypeq{\ahat:\sort}{t}} }

        Similar to the \substextendAdd case.
  \qedhere
  \end{itemize}
\end{proof}

\reversedeclarationorderpreservation*
\begin{proof}\proofheading
  It is given that $u$ and $v$ are declared in $\Gamma$.  Either $u$ is declared to the left of $v$
  in $\Gamma$, or $v$ is declared to the left of $u$.  Suppose the latter (for a contradiction).
  By \Lemmaref{lem:declaration-order-preservation}, $v$ is declared to
  the left of $u$ in $\Delta$.  But we know that $u$ is declared to the left of $v$ in $\Delta$:
  contradiction.   Therefore $u$ is declared to the left of $v$ in $\Gamma$.
\end{proof}

\extensioninversion*
\begin{proof}\proofheading
~
In each part, we proceed by induction on the derivation of $\substextend{\Gamma_0, \dots, \Gamma_1}{\Delta}$.

Note that in each part, the \substextendId case is impossible.

Throughout this proof, we shadow $\Delta$ so that it refers to the \emph{largest proper prefix} of
the $\Delta$ in the statement of the lemma.  For example, in the \substextendVV case of part (i),
we really have $\Delta = (\Delta_{00}, x : A')$, but we call $\Delta_{00}$ ``$\Delta$''.

\begin{enumerate}[(i)]
\item We have $\substextend{\Gamma_0, \alpha:\sort, \Gamma_1}{\Delta}$.

  \begin{itemize}
    \DerivationProofCase{\substextendVV}
          {\substextend{\Gamma}{\Delta}
           \\
           [\Delta]A = [\Delta]A'
          }
          {\substextend{\underbrace{\Gamma, x:A}_{\Gamma_0, \alpha:\sort, \Gamma_1}}{\Delta, x:A'}}
    
          \begin{llproof}
            \eqPfParen{\Gamma, x:A}{\Gamma_0, \alpha:\sort, \Gamma_1}{Given}
            \continueeqPfParen{\Gamma_0, \alpha:\sort,\Gamma'_1, x:A}{Since the last element must be equal }
            \eqPfParen{\Gamma, x:A}{\Gamma_0, \alpha:\sort,\Gamma'_1, x:A}{By transitivity}
            \eqPfParenR{\Gamma}{\Gamma_0, \alpha:\sort,\Gamma'_1}{By injectivity of syntax}

            \proofsep

            \substextendPf{\Gamma}{\Delta}{Subderivation}
            \substextendPf{\Gamma_0, \alpha:\sort,\Gamma'_1}{\Delta}{By equality}
            \eqPfParenR{\Delta}{\Delta_0, \alpha:\sort, \Delta_1}{By i.h.}
            \Hand \substextendPf{\Gamma_0}{\Delta_0}{\ditto}
            \impliesPf{\Gamma'_1\mbox{ soft}}{\Delta_1\mbox{ soft}}{\ditto}

            \proofsep

            \Hand \eqPfParen{\Delta, x:A'}{\Delta_0, \alpha:\sort, \Delta_1, x:A'}{By congruence}
            
            \Hand \impliesPf{\Gamma'_1, x:A\mbox{ soft}}{\Delta_1, x:A'\mbox{ soft}}{Since $\Gamma'_1, x:A$ is not  soft}
          \end{llproof}

    \DerivationProofCase{\substextendUU}
        { \substextend{\Gamma}{\Delta} }
        { \substextend{\underbrace{\Gamma, \beta:\sort'}_{\Gamma_0, \alpha:\sort, \Gamma_1}}{\Delta, \beta:\sort'} }

        There are two cases:
        \begin{itemize}
          \ProofCaseRule{$\alpha:\sort = \beta:\sort'$}

          \begin{llproof}
            \Hand \eqPfParen{\Gamma, \alpha:\sort}{\Gamma_0, \alpha:\sort, \Gamma_1}{where $\Gamma_0 = \Gamma$ and $\Gamma_1 = \cdot$}
            \Hand \eqPfParen{\Delta, \alpha:\sort}{\Delta_0, \alpha:\sort, \Delta_1}{where $\Delta_0 = \Delta$ and $\Delta_1 = \cdot$}
            \Hand \impliesPf{\Gamma_1\mbox{ soft}}{\Delta_1\mbox{ soft}}{since $\cdot$ is soft}
          \end{llproof}

          \ProofCaseRule{$\alpha \neq \beta$}

          \begin{llproof}
            \eqPfParen{\Gamma, \beta:\sort'}{\Gamma_0, \alpha:\sort, \Gamma_1}{Given}
            \continueeqPfParen{\Gamma_0, \alpha:\sort,\Gamma'_1, \beta:\sort'}{Since the last element must be equal }
            \eqPfParenR{\Gamma}{\Gamma_0, \alpha:\sort,\Gamma'_1}{By injectivity of syntax}

            \proofsep

            \substextendPf{\Gamma}{\Delta}{Subderivation}
            \substextendPf{\Gamma_0, \alpha:\sort,\Gamma'_1}{\Delta}{By equality}
            \eqPfParenR{\Delta}{\Delta_0, \alpha:\sort, \Delta_1}{By i.h.}
            \Hand \substextendPf{\Gamma_0}{\Delta_0}{\ditto}
            \impliesPf{\Gamma'_1\mbox{ soft}}{\Delta_1\mbox{ soft}}{\ditto}

            \proofsep

            \Hand \eqPfParen{\Delta, \beta:\sort'}{\Delta_0, \alpha:\sort, \Delta_1, \beta:\sort'}{By congruence}
            
            \Hand \impliesPf{\Gamma'_1, \beta:\sort'\mbox{ soft}}{\Delta_1, \beta:\sort'\mbox{ soft}}{Since $\Gamma'_1, \beta:\sort'$ is not  soft}
          \end{llproof}
          
        \end{itemize}
    
    \DerivationProofCase{\substextendEE}
        { \substextend{\Gamma}{\Delta} }
        { \substextend{\underbrace{\Gamma, \ahat:\sort'}_{\Gamma_0, \alpha:\sort, \Gamma_1}}{\Delta, \ahat:\sort'} }
    
          \begin{llproof}
            \eqPfParen{\Gamma, \ahat:\sort'}{\Gamma_0, \alpha:\sort, \Gamma_1}{Given}
            \continueeqPfParen{\Gamma_0, \alpha:\sort,\Gamma'_1, \ahat:\sort'}{Since the last element must be equal }
            \eqPfParenR{\Gamma}{\Gamma_0, \alpha:\sort,\Gamma'_1}{By injectivity of syntax}

            \proofsep

            \substextendPf{\Gamma}{\Delta}{Subderivation}
            \substextendPf{\Gamma_0, \alpha:\sort,\Gamma'_1}{\Delta}{By equality}
            \eqPfParenR{\Delta}{\Delta_0, \alpha:\sort, \Delta_1}{By i.h.}
            \Hand \substextendPf{\Gamma_0}{\Delta_0}{\ditto}
            \impliesPf{\Gamma'_1\mbox{ soft}}{\Delta_1\mbox{ soft}}{\ditto}

            \proofsep

            \Hand \eqPfParen{\Delta, \ahat:\sort'}{\Delta_0, \alpha:\sort, \Delta_1, \ahat:\sort'}{By congruence}

            \proofsep
            \supposePf{\Gamma'_1, \ahat:\sort'\mbox{ soft}}
            \softPf{\Gamma'_1}{By definition of softness}
            \softPf{\Delta_1}{By induction}
            \softPf{\Delta_1}{By definition of softness}
            \Hand \impliesPf{\Gamma'_1, \ahat:\sort'\mbox{ soft}}{\Delta_1, \ahat:\sort'\mbox{ soft}}{Implication introduction}
          \end{llproof}
    
    \DerivationProofCase{\substextendSolved}
        { \substextend{\Gamma}{\Delta} \\ [\Delta]t = [\Delta]t'}
        { \substextend{\underbrace{\Gamma, \hypeq{\ahat:\sort}{t}}_{\Gamma_0, \alpha:\sort, \Gamma_1}}{\Delta, \hypeq{\ahat:\sort}{t'}}
        }

        Similar to the \substextendEE case.

    \DerivationProofCase{\substextendEqn}
        { \substextend{\Gamma}{\Delta} \\ [\Delta]t = [\Delta]t'}
        { \substextend{\underbrace{\Gamma, \hypeq{\beta}{t}}_{\Gamma_0, \alpha:\sort, \Gamma_1}}{\Delta, \hypeq{\beta}{t'}}
        }

          \begin{llproof}
            \eqPfParen{\Gamma, \hypeq{\beta}{t}}{\Gamma_0, \alpha:\sort, \Gamma_1}{Given}
            \continueeqPfParen{\Gamma_0, \alpha:\sort,\Gamma'_1, \hypeq{\beta}{t}}{Since the last element must be equal }
            \eqPfParenR{\Gamma}{\Gamma_0, \alpha:\sort,\Gamma'_1}{By injectivity of syntax}

            \proofsep

            \substextendPf{\Gamma}{\Delta}{Subderivation}
            \substextendPf{\Gamma_0, \alpha:\sort,\Gamma'_1}{\Delta}{By equality}
            \eqPfParenR{\Delta}{\Delta_0, \alpha:\sort, \Delta_1}{By i.h.}
            \Hand \substextendPf{\Gamma_0}{\Delta_0}{\ditto}
            \impliesPf{\Gamma'_1\mbox{ soft}}{\Delta_1\mbox{ soft}}{\ditto}

            \proofsep

            \Hand \eqPfParen{\Delta, \hypeq{\beta}{t'}}{\Delta_0, \alpha:\sort, \Delta_1, \hypeq{\beta}{t'}}{By congruence}
            
            \Hand \impliesPf{\Gamma'_1, \hypeq{\beta}{t}\mbox{ soft}}{\Delta_1, \hypeq{\beta}{t'}\mbox{ soft}}{Since $\Gamma'_1, \hypeq{\beta}{t}$ is not  soft}
          \end{llproof}

    \DerivationProofCase{\substextendMonMon}
        { \substextend{\Gamma}{\Delta} }
        { \substextend{\underbrace{\Gamma, \MonnierComma{\ahat}}_{\Gamma_0, \alpha:\sort, \Gamma_1}}{\Delta, \MonnierComma{\ahat}} }

          \begin{llproof}
            \eqPfParen{\Gamma, \MonnierComma{\ahat}}{\Gamma_0, \alpha:\sort, \Gamma_1}{Given}
            \continueeqPfParen{\Gamma_0, \alpha:\sort,\Gamma'_1, \MonnierComma{\ahat}}{Since the last element must be equal }
            \eqPfParenR{\Gamma}{\Gamma_0, \alpha:\sort,\Gamma'_1}{By injectivity of syntax}

            \proofsep

            \substextendPf{\Gamma}{\Delta}{Subderivation}
            \substextendPf{\Gamma_0, \alpha:\sort,\Gamma'_1}{\Delta}{By equality}
            \eqPfParenR{\Delta}{\Delta_0, \alpha:\sort, \Delta_1}{By i.h.}
            \Hand \substextendPf{\Gamma_0}{\Delta_0}{\ditto}
            \impliesPf{\Gamma'_1\mbox{ soft}}{\Delta_1\mbox{ soft}}{\ditto}

            \proofsep

            \Hand \eqPfParenR{\Delta, \MonnierComma{\ahat}}{\Delta_0, \alpha:\sort, \Delta_1, \MonnierComma{\ahat}}{By congruence}
            
            \Hand \impliesPf{\Gamma'_1, \MonnierComma{\ahat}\mbox{ soft}}{\Delta_1, \MonnierComma{\ahat}\mbox{ soft}}{Since $\Gamma'_1, \MonnierComma{\ahat}$ is not  soft}
          \end{llproof}

    \DerivationProofCase{\substextendAdd}
        { \substextend{\Gamma}{\Delta} }
        { \substextend{\underbrace{\Gamma}_{\Gamma_0, \alpha:\sort', \Gamma_1}}{\Delta, \ahat:\sort} }

        \begin{llproof}
          \eqPfParenR{\Delta}{\Delta_0, \alpha:\sort, \Delta_1}{By i.h.}
          \Hand \substextendPf{\Gamma_0}{\Delta_0}{\ditto}
          \impliesPf{\Gamma_1 \mbox{ soft}}{\Delta_1 \mbox{ soft}}{\ditto}
          \Hand \eqPfParenR{\Delta, \ahat:\sort'}{\Delta_0, \alpha:\sort, \Delta_1, \ahat:\sort'}{By congruence of equality}
          
          \proofsep
          \supposePf{\Gamma_1 \mbox{ soft}}
          \softPf{\Delta_1}{By i.h.}
          \softPf{\Delta_1, \ahat:\sort'}{By definition of softnesss}
          \Hand \impliesPf{\Gamma_1 \mbox{ soft}}{\Delta_1, \ahat:\sort' \mbox{ soft}}{Implication introduction}

        \end{llproof}
    
    \DerivationProofCase{\substextendAddSolved}
        { \substextend{\Gamma}{\Delta} }
        { \substextend{\underbrace{\Gamma}_{\Gamma_0, \alpha:\sort, \Gamma_1}}{\Delta, \hypeq{\ahat:\sort'}{t}} }

        \begin{llproof}
          \eqPfParenR{\Delta}{\Delta_0, \alpha:\sort, \Delta_1}{By i.h.}
          \Hand \substextendPf{\Gamma_0}{\Delta_0}{\ditto}
          \impliesPf{\Gamma_1 \mbox{ soft}}{\Delta_1 \mbox{ soft}}{\ditto}
          \Hand \eqPfParenR{(\Delta, \hypeq{\ahat:\sort'}{t})}{\Delta_0, \alpha:\sort, \Delta_1, \hypeq{\ahat:\sort'}{t}}{By congruence of equality}
          
          \proofsep
          \supposePf{\Gamma_1 \mbox{ soft}}
          \softPf{\Delta_1}{By i.h.}
          \softPf{(\Delta_1, \hypeq{\ahat:\sort'}{t})}{By definition of softnesss}
          \Hand \impliesPf{\Gamma_1 \mbox{ soft}}{\Delta_1, \hypeq{\ahat:\sort'}{t} \mbox{ soft}}{Implication introduction}

        \end{llproof}

    \DerivationProofCase{\substextendSolve}
           { \substextend{\Gamma}{\Delta} }
           { \substextend{\underbrace{\Gamma, \bhat:\sort'}_{\Gamma_0, \alpha:\sort, \Gamma_1}}{\Delta, \hypeq{\bhat:\sort'}{t}} }

          \begin{llproof}
            \eqPfParen{\Gamma, \bhat:\sort'}{\Gamma_0, \alpha:\sort, \Gamma_1}{Given}
            \eqPfParenR{}{\Gamma_0, \alpha:\sort, \Gamma'_1, \bhat:\sort'}{Since the final elements are equal}
            \eqPfParenR{\Gamma}{\Gamma_0, \alpha:\sort, \Gamma'_1}{By injectivity of context syntax}

            \proofsep

            \substextendPf{\Gamma}{\Delta}{Subderivation}
            \substextendPf{\Gamma_0, \alpha:\sort, \Gamma'_1}{\Delta}{By equality} 
            \eqPfParenR{\Delta}{\Delta_0, \alpha:\sort, \Delta_1}{By i.h.}
            \Hand \substextendPf{\Gamma_0}{\Delta_0}{\ditto}
            \impliesPf{\Gamma'_1 \mbox{ soft}}{\Delta_1 \mbox{ soft}}{\ditto}

            \proofsep
            \Hand \eqPf{\Delta, \bhat:\sort'}{\Delta_0, \alpha:\sort, \Delta_1, \bhat:\sort'}{By congruence}

            \proofsep
            \supposePf{\Gamma'_1, \bhat:\sort'\mbox{ soft}} 
            \softPf{\Gamma'_1}{By definition of softness} 
            \softPf{\Delta_1}{Using i.h.}
            \softPf{\Delta_1, \hypeq{\bhat:\sort'}{t}}{By definition of softness}
            \Hand \impliesPf{\Gamma'_1, \bhat:\sort'\mbox{ soft}}{\Delta_1, \hypeq{\bhat:\sort'}{t} \mbox{ soft}}{Implication intro}
          \end{llproof}
  \end{itemize}

\item %
  We have $\substextend{\Gamma_0, \MonnierComma{u}, \Gamma_1}{\Delta}$.
  This part is similar to part (i) above, except for
  ``if $\dom{\Gamma_0, \MonnierComma{u}, \Gamma_1} = \dom{\Delta}$ then
  $\dom{\Gamma_0} = \dom{\Delta_0}$'', which follows by i.h.\ in most cases.
  In the \substextendMonMon case, either we have $\dots, \MonnierComma{u'}$
  where $u' = u$---in which case the i.h.\ gives us what we need---or
  we have a matching $\MonnierComma{u}$.  In this latter case, 
  we have $\Gamma_1 = \cdot$.
  We know that $\dom{\Gamma_0, \MonnierComma{u}, \Gamma_1} = \dom{\Delta}$
  and $\Delta = (\Delta_0, \MonnierComma{u})$.  Since $\Gamma_1 = \cdot$,
  we have $\dom{\Gamma_0, \MonnierComma{u}} = \dom{\Delta_0, \MonnierComma{u}}$.
  Therefore $\dom{\Gamma_0} = \dom{\Delta_0}$.

\item %
  We have $\substextend{\Gamma_0, \hypeq{\alpha}{\tau}, \Gamma_1}{\Delta}$.

  \begin{itemize}
    \DerivationProofCase{\substextendUU}
        { \substextend{\Gamma}{\Delta} }
        { \substextend{\underbrace{\Gamma, \beta:\sort'}_{\Gamma_0, \hypeq{\alpha}{\tau}, \Gamma_1}}{\Delta, \beta:\sort'} }

        \begin{llproof}
          \eqPfParen
                       {\Gamma_0, \hypeq{\alpha}{\tau}, \Gamma_1}
                       {\Gamma, \beta:\sort'}
                       {Given}
          \continueeqPf{(\Gamma_0, \hypeq{\alpha}{\tau}, \Gamma'_1, \beta:\sort')}{Since the final elements must be equal}
          \eqPf{\Gamma}{(\Gamma_0, \hypeq{\alpha}{\tau}, \Gamma'_1)}{By injectivity of context syntax}

          \proofsep

          \eqPf{\Delta}{(\Delta_0, \hypeq{\alpha}{\tau'}, \Delta_1)}{By i.h.}
          \Hand \eqPf{[\Delta_0]\tau}{[\Delta_0]\tau'}{\ditto}
          \Hand \substextendPf{\Gamma_0}{\Delta_0}{\ditto}
          \Hand \eqPfParen{\Delta, \beta:\sort'}{\Delta_0, \hypeq{\alpha}{\tau'}, \Delta_1, \beta:\sort'}{By congruence of equality}
        \end{llproof}

    \DerivationProofCase{\substextendVV}
          {\substextend{\Gamma}{\Delta}
           \\
           [\Delta]A = [\Delta]A'
          }
          {\substextend{\underbrace{\Gamma, x:A}_{\Gamma_0, \hypeq{\alpha}{\tau}, \Gamma_1}}{\Delta, x:A'}}

          Similar to the \substextendUU case.

    \DerivationProofCase{\substextendMonMon}
        { \substextend{\Gamma}{\Delta} }
        { \substextend{\Gamma, \MonnierComma{\ahat}}{\Delta, \MonnierComma{\ahat}} }

        Similar to the \substextendUU case.

    \DerivationProofCase{\substextendEE}
        { \substextend{\Gamma}{\Delta} }
        { \substextend{\Gamma, \ahat:\sort'}{\Delta, \ahat:\sort'} }

        Similar to the \substextendUU case.

    \DerivationProofCase{\substextendSolved}
        { \substextend{\Gamma}{\Delta} \\ [\Delta]t = [\Delta]t'}
        { \substextend{\underbrace{\Gamma, \hypeq{\ahat:\sort'}{t}}_{\Gamma_0, \hypeq{\alpha}{\tau}, \Gamma_1}}{\Delta, \hypeq{\ahat:\sort'}{t'}}
        }

        Similar to the \substextendUU case.

   \DerivationProofCase{\substextendSolve}
      { \substextend{\Gamma}{\Delta} }
      { \substextend{\underbrace{\Gamma, \bhat:\sort'}_{\Gamma_0, \hypeq{\alpha}{\tau}, \Gamma_1}}{\Delta, \hypeq{\bhat:\sort'}{t}} }

      Similar to the \substextendUU case.

    \DerivationProofCase{\substextendEqn}
        { \substextend{\Gamma}{\Delta} \\ [\Delta]t = [\Delta]t'}
        { \substextend{\underbrace{\Gamma, \hypeq{\beta}{t}}_{\Gamma_0, \hypeq{\alpha}{\tau}, \Gamma_1}}{\Delta, \hypeq{\beta}{t'}}
        }

        There are two cases:

        \begin{itemize}
              \ProofCaseRule{$\alpha = \beta$}

                  \begin{llproof}
                    \eqPf{\tau}{t \AND \Gamma_1 = \cdot \AND \Gamma_0 = \Gamma}{By injectivity of syntax}
                    \Hand \substextendPf{\Gamma_0}{\Delta_0}  {Subderivation ($\Gamma_0 = \Gamma$ and let $\Delta_0 = \Delta$)}
                    \Hand \eqPfParen{\Delta, \hypeq{\alpha}{t'}}{\Delta_0, \hypeq{\alpha}{t'}, \Delta_1}{where $\Delta_1 = \cdot$}
                    \Hand \eqPf{[\Delta_0]t}{[\Delta_0]t'}{By premise $[\Delta]t = [\Delta]t'$}
                  \end{llproof}

              \ProofCaseRule{$\alpha \neq \beta$}

                  \begin{llproof}
                    \eqPfParen{\Gamma_0, \hypeq{\alpha}{\tau}, \Gamma_1}{\Gamma, \hypeq{\beta}{t}}{Given}
                    \continueeqPf{(\Gamma_0, \hypeq{\alpha}{\tau}, \Gamma'_1, \hypeq{\beta}{t})}{Since the final elements must be equal}
                    \eqPfParenR{\Gamma}{\Gamma_0, \hypeq{\alpha}{\tau}, \Gamma'_1}{By injectivity of context syntax}

                    \proofsep

                    \eqPfParenR{\Delta}{\Delta_0, \hypeq{\alpha}{\tau'}, \Delta_1}{By i.h.}
                    \Hand \eqPf{[\Delta_0]\tau}{[\Delta_0]\tau'}{\ditto}
                    \Hand \substextendPf{\Gamma_0}{\Delta_0}{\ditto}
                    \Hand \eqPfParen{\Delta, \hypeq{\beta}{t'}}{\Delta_0, \hypeq{\alpha}{\tau'}, \Delta_1, \hypeq{\beta}{t'}}{By congruence of equality}
                  \end{llproof}
        \end{itemize}    

    \DerivationProofCase{\substextendAdd}
        { \substextend{\Gamma}{\Delta} }
        { \substextend{\underbrace{\Gamma}_{\Gamma_0, \hypeq{\alpha}{\tau}, \Gamma_1}}{\Delta, \ahat:\sort'} }

        \begin{llproof}
          \eqPfParenR{\Delta}{\Delta_0, \hypeq{\alpha}{\tau'}, \Delta_1}{By i.h.}
          \Hand \eqPf{[\Delta_0]\tau}{[\Delta_0]\tau'}{\ditto}
          \Hand \substextendPf{\Gamma_0}{\Delta_0}{\ditto}
          \Hand \eqPfParen{\Delta, \ahat:\sort'}{\Delta_0, \hypeq{\alpha}{\tau'}, \Delta_1, \ahat:\sort'}{By congruence of equality}
        \end{llproof}
    
    \DerivationProofCase{\substextendAddSolved}
        { \substextend{\Gamma}{\Delta} }
        { \substextend{\underbrace{\Gamma}_{\Gamma_0, \hypeq{\alpha}{\tau}, \Gamma_1}}{\Delta, \hypeq{\ahat:\sort'}{t}} }

        \begin{llproof}
          \eqPfParenR{\Delta}{\Delta_0, \hypeq{\alpha}{\tau'}, \Delta_1}{By i.h.}
          \Hand \eqPf{[\Delta_0]\tau}{[\Delta_0]\tau'}{\ditto}
          \Hand \substextendPf{\Gamma_0}{\Delta_0}{\ditto}
          \Hand \eqPf{(\Delta, \hypeq{\ahat:\sort'}{t})}{(\Delta_0, \hypeq{\alpha}{\tau'}, \Delta_1, \hypeq{\ahat:\sort'}{t})}{By congruence of equality}
        \end{llproof}
  \end{itemize}

\item %
  We have $\substextend{\Gamma_0, \hypeq{\ahat:\sort}{\tau}, \Gamma_1}{\Delta}$.

  \begin{itemize}
    \DerivationProofCase{\substextendUU}
        { \substextend{\Gamma}{\Delta} }
        { \substextend{\underbrace{\Gamma, \beta:\sort'}_{\Gamma_0, \hypeq{\ahat:\sort}{\tau}, \Gamma_1}}{\Delta, \beta:\sort'} }

        \begin{llproof}
          \eqPf
               {(\Gamma_0, \hypeq{\ahat:\sort}{\tau}, \Gamma_1)}
               {(\Gamma, \beta:\sort')}
               {Given}
          \continueeqPf{(\Gamma_0, \hypeq{\ahat:\sort}{\tau}, \Gamma'_1, \beta:\sort')}{Since the final elements must be equal}
          \eqPf{\Gamma}{(\Gamma_0, \hypeq{\ahat:\sort}{\tau}, \Gamma'_1)}{By injectivity of context syntax}

          \proofsep

          \eqPf{\Delta}{(\Delta_0, \hypeq{\ahat:\sort}{\tau'}, \Delta_1)}{By i.h.}
          \Hand \eqPf{[\Delta_0]\tau}{[\Delta_0]\tau'}{\ditto}
          \Hand \substextendPf{\Gamma_0}{\Delta_0}{\ditto}
          \Hand \eqPf{(\Delta, \beta:\sort')}{(\Delta_0, \hypeq{\ahat:\sort}{\tau'}, \Delta_1, \beta:\sort')}{By congruence of equality}
        \end{llproof}

    \DerivationProofCase{\substextendVV}
          {\substextend{\Gamma}{\Delta}
           \\
           [\Delta]A = [\Delta]A'
          }
          {\substextend{\underbrace{\Gamma, x:A}_{\Gamma_0, \hypeq{\ahat:\sort}{\tau}, \Gamma_1}}{\Delta, x:A'}}

          Similar to the \substextendUU case.

    \DerivationProofCase{\substextendMonMon}
        { \substextend{\Gamma}{\Delta} }
        { \substextend{\Gamma, \MonnierComma{\bhat}}{\Delta, \MonnierComma{\bhat}} }

        Similar to the \substextendUU case.

    \DerivationProofCase{\substextendEE}
        { \substextend{\Gamma}{\Delta} }
        { \substextend{\Gamma, \bhat:\sort'}{\Delta, \bhat:\sort'} }

        Similar to the \substextendUU case.

    \DerivationProofCase{\substextendSolved}
        { \substextend{\Gamma}{\Delta} \\ [\Delta]t = [\Delta]t'}
        { \substextend{\underbrace{\Gamma, \hypeq{\bhat:\sort'}{t}}_{\Gamma_0, \hypeq{\ahat:\sort}{\tau}, \Gamma_1}}{\Delta, \hypeq{\bhat:\sort'}{t'}}
        }

        There are two cases.
        
            \begin{itemize}
                \ProofCaseRule{$\ahat = \bhat$}

                  \begin{llproof}
                    \eqPf{\sort' = \sort \AND t}{\tau \AND \Gamma_1 = \cdot \AND \Gamma = \Gamma_0}{By injectivity of syntax}
                    \Hand \eqPf{(\Delta, \hypeq{\bhat:\sort'}{t'})}{(\Delta_0, \hypeq{\bhat:\sort'}{\tau'}, \Delta_1)}
                                   {where $\tau' = t'$ and $\Delta_1 = \cdot$ and $\Delta = \Delta_0$}
                    \Hand \substextendPf{\Gamma_0}{\Delta_0}{From subderivation $\substextend{\Gamma}{\Delta}$}
                    \Hand \eqPf{[\Delta_0]\tau}{[\Delta_0]\tau'}{From premise $[\Delta]t = [\Delta]t'$ and $x$}
                  \end{llproof}

                \ProofCaseRule{$\ahat \neq \bhat$}

                  \begin{llproof}
                    \eqPf{(\Gamma_0, \hypeq{\ahat:\sort}{\tau}, \Gamma_1)}{(\Gamma, \hypeq{\bhat:\sort'}{t})}{Given}
                    \continueeqPf{(\Gamma_0, \hypeq{\ahat:\sort}{\tau}, \Gamma'_1, \hypeq{\bhat:\sort'}{t})}{Since the final elements must be equal}
                    \eqPf{\Gamma}{(\Gamma_0, \hypeq{\ahat:\sort}{\tau}, \Gamma'_1)}{By injectivity of context syntax}
                    \decolumnizePf
                    \eqPfParenR{\Delta}{\Delta_0, \hypeq{\ahat:\sort}{\tau'}, \Delta_1}{By i.h.}
                    \Hand \eqPf{[\Delta_0]\tau}{[\Delta_0]\tau'}{\ditto}
                    \Hand \substextendPf{\Gamma_0}{\Delta_0}{\ditto}
                    \Hand \eqPf{(\Delta, \hypeq{\bhat:\sort'}{t'})}{(\Delta_0, \hypeq{\ahat:\sort}{\tau'}, \Delta_1, \hypeq{\bhat:\sort'}{t'})}{By congruence of equality}
                  \end{llproof}
            \end{itemize}
        
    \DerivationProofCase{\substextendEqn}
        { \substextend{\Gamma}{\Delta} \\ [\Delta]t = [\Delta]t'}
        { \substextend{\underbrace{\Gamma, \hypeq{\beta}{t}}_{\Gamma_0, \hypeq{\ahat:\sort}{\tau}, \Gamma_1}}{\Delta, \hypeq{\beta}{t'}}
        }

          \begin{llproof}
            \eqPf{(\Gamma_0, \hypeq{\ahat:\sort}{\tau}, \Gamma_1)}{(\Gamma, \hypeq{\beta}{t})}{Given}
            \continueeqPf{(\Gamma_0, \hypeq{\ahat:\sort}{\tau}, \Gamma'_1, \hypeq{\beta}{t})}{Since the final elements must be equal}
            \eqPf{\Gamma}{(\Gamma_0, \hypeq{\ahat:\sort}{\tau}, \Gamma'_1)}{By injectivity of context syntax}

            \proofsep

            \eqPfParenR{\Delta}{\Delta_0, \hypeq{\ahat:\sort}{\tau'}, \Delta_1}{By i.h.}
            \Hand \eqPf{[\Delta_0]\tau}{[\Delta_0]\tau'}{\ditto}
            \Hand \substextendPf{\Gamma_0}{\Delta_0}{\ditto}
            \Hand \eqPfParen{\Delta, \hypeq{\beta}{t'}}{\Delta_0, \hypeq{\ahat:\sort}{\tau'}, \Delta_1, \hypeq{\beta}{t'}}{By congruence of equality}
          \end{llproof}

    \DerivationProofCase{\substextendAdd}
        { \substextend{\Gamma}{\Delta} }
        { \substextend{\underbrace{\Gamma}_{\Gamma_0, \hypeq{\ahat:\sort}{\tau}, \Gamma_1}}{\Delta, \bhat:\sort'} }

        \begin{llproof}
          \eqPfParenR{\Delta}{\Delta_0, \hypeq{\ahat:\sort}{\tau'}, \Delta_1}{By i.h.}
          \Hand \eqPf{[\Delta_0]\tau}{[\Delta_0]\tau'}{\ditto}
          \Hand \substextendPf{\Gamma_0}{\Delta_0}{\ditto}
          \Hand \eqPfParen{\Delta, \bhat:\sort'}{\Delta_0, \hypeq{\ahat:\sort}{\tau'}, \Delta_1, \bhat:\sort'}{By congruence of equality}
        \end{llproof}
    
    \DerivationProofCase{\substextendAddSolved}
        { \substextend{\Gamma}{\Delta} }
        { \substextend{\underbrace{\Gamma}_{\Gamma_0, \hypeq{\ahat:\sort}{\tau}, \Gamma_1}}{\Delta, \hypeq{\bhat:\sort'}{t}} }

        \begin{llproof}
          \eqPfParenR{\Delta}{\Delta_0, \hypeq{\ahat:\sort}{\tau'}, \Delta_1}{By i.h.}
          \Hand \eqPf{[\Delta_0]\tau}{[\Delta_0]\tau'}{\ditto}
          \Hand \substextendPf{\Gamma_0}{\Delta_0}{\ditto}
          \Hand \eqPf{(\Delta, \hypeq{\bhat:\sort'}{t})}{(\Delta_0, \hypeq{\ahat:\sort}{\tau'}, \Delta_1, \hypeq{\bhat:\sort'}{t})}{By congruence of equality}
        \end{llproof}

   \DerivationProofCase{\substextendSolve}
      { \substextend{\Gamma}{\Delta} }
      { \substextend{\underbrace{\Gamma, \bhat:\sort'}_{\Gamma_0, \hypeq{\ahat:\sort}{\tau}, \Gamma_1}}{\Delta, \hypeq{\bhat:\sort'}{t}} }

      \begin{llproof}
        \eqPfParen{\Gamma, \bhat:\sort'}{\Gamma_0, \hypeq{\ahat:\sort}{\tau}, \Gamma_1}{Given}
        \continueeqPfParen{\Gamma_0, \hypeq{\ahat:\sort}{\tau}, \Gamma'_1, \bhat:\sort'}{Since the last elements must be equal}
        \eqPfParenR{\Gamma}{\Gamma_0, \hypeq{\ahat:\sort}{\tau}, \Gamma'_1}{By injectivity of syntax}

        \proofsep
        \substextendPf{\Gamma}{\Delta}{Subderivation}
        \substextendPf{\Gamma_0, \hypeq{\ahat:\sort}{\tau}, \Gamma'_1}{\Delta}{By equality}
        \eqPfParenR{\Delta}{\Delta_0, \hypeq{\ahat:\sort}{\tau'}, \Delta_1}{By i.h.}
        \Hand \eqPf{[\Delta_0]\tau}{[\Delta_0]\tau'}{\ditto}
        \Hand \substextendPf{\Gamma_0}{\Delta_0}{\ditto}
        \Hand \eqPfParen{\Delta, \bhat:\sort'}{\Delta_0, \hypeq{\ahat:\sort}{\tau'}, \Delta_1, \bhat:\sort'}{By congruence of equality}
      \end{llproof}

  \end{itemize}

\item %
  We have $\substextend{\Gamma_0, x:A, \Gamma_1}{\Delta}$.
  This proof is similar to the proof of part (i), except for the domain condition, which we handle
  similarly to part (ii).

\item %
  We have $\substextend{\Gamma_0, \ahat:\sort, \Gamma_1}{\Delta}$.

  \begin{itemize}
    \DerivationProofCase{\substextendUU}
        { \substextend{\Gamma}{\Delta} }
        { \substextend{\underbrace{\Gamma, \beta:\sort'}_{\Gamma_0, \ahat:\sort, \Gamma_1}}{\Delta, \beta:\sort'} }

             \begin{llproof}
                    \eqPf{(\Gamma_0, \ahat:\sort, \Gamma_1)}{(\Gamma, \beta:\sort')}{Given}
                    \continueeqPf{(\Gamma_0, \ahat:\sort, \Gamma'_1, \beta:\sort')}{Since the final elements must be equal}
                    \eqPf{\Gamma}{(\Gamma_0, \ahat:\sort, \Gamma'_1)}{By injectivity of context syntax}
             \end{llproof}

             By induction, there are two possibilities:

                 \begin{itemize}
                 \item $\ahat$ is not solved:

                   \begin{llproof}
                     \eqPfParenR{\Delta}{\Delta_0, \ahat:\sort, \Delta_1}{By i.h.}
                     \Hand \substextendPf{\Gamma_0}{\Delta_0}{\ditto}
                     \Hand \eqPfParen{\Delta, \beta:\sort'}{\Delta_0, \ahat:\sort, \Delta_1, \beta:\sort'}{By congruence of equality}
                   \end{llproof}

                 \item $\ahat$ is solved:

                   \begin{llproof}
                     \eqPfParenR{\Delta}{\Delta_0, \hypeq{\ahat:\sort}{\tau'}, \Delta_1}{By i.h.}
                     \Hand \substextendPf{\Gamma_0}{\Delta_0}{\ditto}
                     \Hand \eqPfParen{\Delta, \beta:\sort'}{\Delta_0, \hypeq{\ahat:\sort}{\tau'}, \Delta_1, \beta:\sort'}{By congruence of equality}
                   \end{llproof}
                 \end{itemize}

    \DerivationProofCase{\substextendVV}
          {\substextend{\Gamma}{\Delta}
           \\
           [\Delta]A = [\Delta]A'
          }
          {\substextend{\underbrace{\Gamma, x:A}_{\Gamma_0, \ahat:\sort, \Gamma_1}}{\Delta, x:A'}}

          Similar to the \substextendUU case.

    \DerivationProofCase{\substextendMonMon}
        { \substextend{\Gamma}{\Delta} }
        { \substextend{\Gamma, \MonnierComma{\bhat}}{\Delta, \MonnierComma{\bhat}} }

        Similar to the \substextendUU case.

    \DerivationProofCase{\substextendEqn}
        { \substextend{\Gamma}{\Delta} \\ [\Delta]t = [\Delta]t'}
        { \substextend{\Gamma, \hypeq{\beta}{t}}{\Delta, \hypeq{\beta}{t'}}
        }

        Similar to the \substextendUU case.

    \DerivationProofCase{\substextendSolved}
        { \substextend{\Gamma}{\Delta} \\ [\Delta]t = [\Delta]t'}
        { \substextend{\underbrace{\Gamma, \hypeq{\bhat:\sort'}{t}}_{\Gamma_0, \ahat:\sort, \Gamma_1}}{\Delta, \hypeq{\bhat:\sort'}{t'}}
        }

        Similar to the \substextendUU case.        

    \DerivationProofCase{\substextendEE}
        { \substextend{\Gamma}{\Delta} }
        { \substextend{\underbrace{\Gamma, \bhat:\sort'}_{\Gamma_0, \ahat:\sort, \Gamma_1}}{\Delta, \bhat:\sort'} }

    \begin{itemize}
        \ProofCaseRule{$\ahat \neq \bhat$}

             \begin{llproof}
                    \eqPf{(\Gamma_0, \ahat:\sort, \Gamma_1)}{(\Gamma, \bhat:\sort')}{Given}
                    \continueeqPf{(\Gamma_0, \ahat:\sort, \Gamma'_1, \bhat:\sort')}{Since the final elements must be equal}
                    \eqPf{\Gamma}{(\Gamma_0, \ahat:\sort, \Gamma'_1)}{By injectivity of context syntax}
             \end{llproof}

             By induction, there are two possibilities:
             \begin{itemize}
             \item $\ahat$ is not solved:

               \begin{llproof}
                 \eqPfParenR{\Delta}{\Delta_0, \ahat:\sort, \Delta_1}{By i.h.}
                 \Hand \substextendPf{\Gamma_0}{\Delta_0}{\ditto}
                 \Hand \eqPfParen{\Delta, \bhat:\sort'}{\Delta_0, \ahat:\sort, \Delta_1, \bhat:\sort'}{By congruence of equality}
               \end{llproof}

             \item $\ahat$ is solved:

               \begin{llproof}
                 \eqPfParenR{\Delta}{\Delta_0, \hypeq{\ahat:\sort}{\tau'}, \Delta_1}{By i.h.}
                 \Hand \substextendPf{\Gamma_0}{\Delta_0}{\ditto}
                 \Hand \eqPfParen{\Delta, \bhat:\sort'}{\Delta_0, \hypeq{\ahat:\sort}{\tau'}, \Delta_1, \bhat:\sort'}{By congruence of equality}
               \end{llproof}
             \end{itemize}

        \ProofCaseRule{$\ahat = \bhat$}

            \begin{llproof}
              \eqPf{\sort'}{\sort \AND \Gamma_0 = \Gamma \AND \Gamma_1 = \cdot}{By injectivity of syntax}
              \Hand \eqPfParen{\Delta, \bhat:\sort'}{\Delta_0, \ahat:\sort, \Delta_1}
                         {where $\Delta_0 = \Delta$ and $\Delta_1 = \cdot$}
              \Hand \substextendPf{\Gamma_0}{\Delta_0}{From premise $\substextend{\Gamma}{\Delta}$}
            \end{llproof}
  \end{itemize}

    \DerivationProofCase{\substextendAdd}
        { \substextend{\Gamma}{\Delta} }
        { \substextend{\underbrace{\Gamma}_{\Gamma_0, \ahat:\sort, \Gamma_1}}{\Delta, \bhat:\sort'} }

             By induction, there are two possibilities:

                 \begin{itemize}
                 \item $\ahat$ is not solved:

                   \begin{llproof}
                     \eqPfParenR{\Delta}{\Delta_0, \ahat:\sort, \Delta_1}{By i.h.}
                     \Hand \substextendPf{\Gamma_0}{\Delta_0}{\ditto}
                     \Hand \eqPfParen{\Delta, \bhat:\sort'}{\Delta_0, \ahat:\sort, \Delta_1, \bhat:\sort'}{By congruence of equality}
                   \end{llproof}

                 \item $\ahat$ is solved:

                   \begin{llproof}
                     \eqPfParenR{\Delta}{\Delta_0, \hypeq{\ahat:\sort}{\tau'}, \Delta_1}{By i.h.}
                     \Hand \substextendPf{\Gamma_0}{\Delta_0}{\ditto}
                     \Hand \eqPfParen{\Delta, \bhat:\sort'}{\Delta_0, \hypeq{\ahat:\sort}{\tau'}, \Delta_1, \bhat:\sort'}{By congruence of equality}
                   \end{llproof}
                 \end{itemize}

    \DerivationProofCase{\substextendAddSolved}
        { \substextend{\Gamma}{\Delta} }
        { \substextend{\underbrace{\Gamma}_{\Gamma_0, \ahat:\sort, \Gamma_1}}{\Delta, \hypeq{\bhat:\sort'}{t}} }

             By induction, there are two possibilities:

             \begin{itemize}
             \item $\ahat$ is not solved:

               \begin{llproof}
                 \eqPfParenR{\Delta}{\Delta_0, \ahat:\sort, \Delta_1}{By i.h.}
                 \Hand \substextendPf{\Gamma_0}{\Delta_0}{\ditto}
                 \Hand \eqPfParen{\Delta, \hypeq{\bhat:\sort'}{t}}{\Delta_0, \ahat:\sort, \Delta_1, \hypeq{\bhat:\sort'}{t}}{By congruence of equality}
               \end{llproof}

             \item $\ahat$ is solved:

               \begin{llproof}
                 \eqPfParenR{\Delta}{\Delta_0, \hypeq{\ahat:\sort}{\tau'}, \Delta_1}{By i.h.}
                 \Hand \substextendPf{\Gamma_0}{\Delta_0}{\ditto}
                 \Hand \eqPfParen{\Delta, \hypeq{\bhat:\sort'}{t}}{\Delta_0, \hypeq{\ahat:\sort}{\tau'}, \Delta_1, \hypeq{\bhat:\sort'}{t}}{By congruence of equality}
               \end{llproof}
             \end{itemize}

    \DerivationProofCase{\substextendSolve}
        { \substextend{\Gamma}{\Delta} }
        { \substextend{\underbrace{\Gamma, \bhat:\sort'}_{\Gamma_0, \ahat:\sort, \Gamma_1}}{\Delta, \hypeq{\bhat:\sort'}{t}} }

            \begin{itemize}
            \ProofCaseRule{$\ahat \neq\bhat$}

                     \begin{llproof}
                            \eqPf{(\Gamma_0, \ahat:\sort, \Gamma_1)}{(\Gamma, \bhat:\sort')}{Given}
                            \continueeqPf{(\Gamma_0, \ahat:\sort, \Gamma'_1, \bhat:\sort')}{Since the final elements must be equal}
                            \eqPf{\Gamma}{(\Gamma_0, \ahat:\sort, \Gamma'_1)}{By injectivity of context syntax}
                     \end{llproof}

                     By induction, there are two possibilities:

                         \begin{itemize}
                         \item $\ahat$ is not solved:

                           \begin{llproof}
                             \eqPfParenR{\Delta}{\Delta_0, \ahat:\sort, \Delta_1}{By i.h.}
                             \Hand \substextendPf{\Gamma_0}{\Delta_0}{\ditto}
                             \Hand \eqPfParen{\Delta, \hypeq{\bhat:\sort'}{t}}{\Delta_0, \ahat:\sort, \Delta_1, \hypeq{\bhat:\sort'}{t}}{By congruence of equality}
                           \end{llproof}

                         \item $\ahat$ is solved:

                           \begin{llproof}
                             \eqPfParenR{\Delta}{\Delta_0, \hypeq{\ahat:\sort}{\tau'}, \Delta_1}{By i.h.}
                             \Hand \substextendPf{\Gamma_0}{\Delta_0}{\ditto}
                             \Hand \eqPfParen{\Delta, \hypeq{\bhat:\sort'}{t}}{\Delta_0, \hypeq{\ahat:\sort}{\tau'}, \Delta_1, \hypeq{\bhat:\sort'}{t}}{By congruence of equality}
                           \end{llproof}
                         \end{itemize}

            \ProofCaseRule{$\ahat = \bhat$}

                    \begin{llproof}
                      \eqPf{\Gamma =\Gamma_0 \AND \sort}{\sort' \AND \Gamma_1 = \cdot}{By injectivity of syntax}
                      \Hand \eqPfParen{\Delta, \hypeq{\bhat:\sort'}{t}}{\Delta_0, \hypeq{\ahat:\sort}{\tau'}, \Delta_1}
                                      {where $\Delta_0 = \Delta$ and $\tau' = t$ and $\Delta_1 = \cdot$}
                      \Hand \substextendPf{\Gamma_0}{\Delta_0}{From premise $\substextend{\Gamma}{\Delta}$}
                    \end{llproof}
            \qedhere
            \end{itemize}
  \end{itemize} %
\end{enumerate} %
\end{proof}

\deepexistential*
\begin{proof}\proofheading
~
  \begin{enumerate}[(i)]
  \item Assume that $\Gamma_0, \Gamma_1$ is well-formed. We proceed by induction on $\Gamma_1$. 

    \begin{itemize}
    \item Case $\Gamma_1 = \cdot$: 

      \begin{llproof}
        \judgectxPf{\Gamma_0}{Given}
        \Pf{}{}{\ahat \notin \dom{\Gamma_0}}{Given}
        \judgectxPf{\Gamma_0, \ahat:\sort}{By rule \VarCWF}
        \substextendPf{\Gamma_0}{\Gamma_0}{By \Lemmaref{lem:extension-reflexivity}}
        \Hand \substextendPf{\Gamma_0}{\Gamma_0, \ahat:\sort}{By rule \substextendAdd}
      \end{llproof}

    \item Case $\Gamma_1 = \Gamma'_1, x:A$: 

      \begin{llproof}
        \judgectxPf{\Gamma_0, \Gamma'_1, x:A}{Given}
        \judgectxPf{\Gamma_0, \Gamma'_1}{By inversion} 
        \Pf{}{}{x \notin \dom{\Gamma_0, \Gamma'_1}}{By inversion (1)}
        \judgetpPf{\Gamma_0, \Gamma'_1}{A}{By inversion}
        \Pf{}{}{\ahat \notin \dom{\Gamma_0, \Gamma'_1, x:A}}{Given}
        \Pf{}{}{\ahat \neq x}{By inversion (2)}
        \judgectxPf{\Gamma_0, \ahat:\sort, \Gamma'_1}{By i.h.}
        \substextendPf{\Gamma_0, \Gamma'_1}{\Gamma_0, \ahat:\sort, \Gamma'_1}{\ditto}
        \judgetpPf{\Gamma_0, \ahat:\sort, \Gamma'_1}{A}{By \Lemmaref{lem:extension-weakening}}
        \Pf{}{}{x \notin \dom{\Gamma_0, \ahat:\sort, \Gamma'_1}}{By (1) and (2) }
        \Hand \substextendPf{\Gamma_0, \Gamma'_1, x:A}{\Gamma_0, \ahat:\sort, \Gamma'_1, x:A}{By \substextendVV}
      \end{llproof}
      
    \item Case $\Gamma_1 = \Gamma'_1, \beta:\sort'$: 

      \begin{llproof}
        \judgectxPf{\Gamma_0, \Gamma'_1, \beta:\sort'}{Given}
        \judgectxPf{\Gamma_0, \Gamma'_1}{By inversion} 
        \Pf{}{}{\beta \notin \dom{\Gamma_0, \Gamma'_1}}{By inversion (1)}
        \Pf{}{}{\ahat \notin \dom{\Gamma_0, \Gamma'_1, \beta:\sort'}}{Given}
        \Pf{}{}{\ahat \neq \beta}{By inversion (2)}
        \judgectxPf{\Gamma_0, \ahat:\sort, \Gamma'_1}{By i.h.}
        \substextendPf{\Gamma_0, \Gamma'_1}{\Gamma_0, \ahat:\sort, \Gamma'_1}{\ditto}
        \Pf{}{}{\beta \notin \dom{\Gamma_0, \ahat:\sort, \Gamma'_1}}{By (1) and (2) }
        \Hand \substextendPf{\Gamma_0, \Gamma'_1, \beta:\sort'}{\Gamma_0, \ahat:\sort, \Gamma'_1, \beta:\sort'}{By \substextendUU}
      \end{llproof}

    \item Case $\Gamma_1 = \Gamma'_1, \bhat:\sort'$: 

      \begin{llproof}
        \judgectxPf{\Gamma_0, \Gamma'_1, \bhat:\sort'}{Given}
        \judgectxPf{\Gamma_0, \Gamma'_1}{By inversion} 
        \Pf{}{}{\bhat \notin \dom{\Gamma_0, \Gamma'_1}}{By inversion (1)}
        \Pf{}{}{\ahat \notin \dom{\Gamma_0, \Gamma'_1, \bhat:\sort'}}{Given}
        \Pf{}{}{\ahat \neq \bhat}{By inversion (2)}
        \judgectxPf{\Gamma_0, \ahat:\sort, \Gamma'_1}{By i.h.}
        \substextendPf{\Gamma_0, \Gamma'_1}{\Gamma_0, \ahat:\sort, \Gamma'_1}{\ditto}
        \Pf{}{}{\bhat \notin \dom{\Gamma_0, \ahat:\sort, \Gamma'_1}}{By (1) and (2) }
        \Hand \substextendPf{\Gamma_0, \Gamma'_1, \bhat:\sort'}{\Gamma_0, \ahat:\sort, \Gamma'_1, \bhat:\sort'}{By \substextendEE}
      \end{llproof}

    \item Case $\Gamma_1 = (\Gamma'_1, \hypeq{\bhat:\sort'}{t})$: 

      \begin{llproof}
        \judgectxPf{\Gamma_0, \Gamma'_1, \hypeq{\bhat:\sort'}{t}}{Given}
        \judgectxPf{\Gamma_0, \Gamma'_1}{By inversion} 
        \Pf{}{}{\bhat \notin \dom{\Gamma_0, \Gamma'_1}}{By inversion (1)}
        \judgePf{\Gamma_0, \Gamma'_1}{t}{\sort'}{By inversion}
        \Pf{}{}{\ahat \notin \dom{\Gamma_0, \Gamma'_1, \hypeq{\bhat:\sort'}{t}}}{Given}
        \Pf{}{}{\ahat \neq \bhat}{By inversion (2)}
        \judgectxPf{\Gamma_0, \ahat:\sort, \Gamma'_1}{By i.h.}
        \substextendPf{\Gamma_0, \Gamma'_1}{\Gamma_0, \ahat:\sort, \Gamma'_1}{\ditto}
        \PfLEFT{6}
        \judgePf{\Gamma_0, \ahat:\sort, \Gamma'_1}{t}{\sort'}{By \Lemmaref{lem:extension-weakening}}
        \PfRESTORE
        \Pf{}{}{\bhat \notin \dom{\Gamma_0, \ahat:\sort, \Gamma'_1}}{By (1) and (2) }
        \Hand \substextendPf{\Gamma_0, \Gamma'_1, \hypeq{\bhat:\sort'}{t}}{\Gamma_0, \ahat:\sort, \Gamma'_1, \hypeq{\bhat:\sort'}{t}}{By \substextendSolved}
      \end{llproof}

    \item Case $\Gamma_1 = (\Gamma'_1, \hypeq{\beta}{t})$:

      \begin{llproof}
        \judgectxPf{\Gamma_0, \Gamma'_1, \hypeq{\beta}{t}}{Given}
        \judgectxPf{\Gamma_0, \Gamma'_1}{By inversion} 
        \Pf{}{}{\beta \notin \dom{\Gamma_0, \Gamma'_1}}{By inversion (1)}
        \judgePf{\Gamma_0, \Gamma'_1}{t}{\ind}{By inversion}
        \Pf{}{}{\ahat \notin \dom{\Gamma_0, \Gamma'_1, \hypeq{\beta}{t}}}{Given}
        \Pf{}{}{\ahat \neq \beta}{By inversion (2)}
        \judgectxPf{\Gamma_0, \ahat:\sort, \Gamma'_1}{By i.h.}
        \substextendPf{\Gamma_0, \Gamma'_1}{\Gamma_0, \ahat:\sort, \Gamma'_1}{\ditto}
        \PfLEFT{3}
        \judgePf{\Gamma_0, \ahat:\sort, \Gamma'_1}{t}{\ind}{By \Lemmaref{lem:extension-weakening}}
        \PfRESTORE
        \Pf{}{}{\beta \notin \dom{\Gamma_0, \ahat:\sort, \Gamma'_1}}{By (1) and (2) }
        \Hand \substextendPf{\Gamma_0, \Gamma'_1, \hypeq{\beta}{t}}{\Gamma_0, \ahat:\sort, \Gamma'_1, \hypeq{\beta}{t}}{By \substextendSolved}
      \end{llproof}

    \item Case $\Gamma_1 = (\Gamma'_1, \MonnierComma{\bhat})$: 

      \begin{llproof}
        \judgectxPf{\Gamma_0, \Gamma'_1, \MonnierComma{\bhat}}{Given}
        \judgectxPf{\Gamma_0, \Gamma'_1}{By inversion} 
        \Pf{}{}{\bhat \notin \dom{\Gamma_0, \Gamma'_1}}{By inversion (1)}
        \Pf{}{}{\ahat \notin \dom{\Gamma_0, \Gamma'_1, \MonnierComma{\bhat}}}{Given}
        \Pf{}{}{\ahat \neq \bhat}{By inversion (2)}
        \judgectxPf{\Gamma_0, \ahat:\sort, \Gamma'_1}{By i.h.}
        \substextendPf{\Gamma_0, \Gamma'_1}{\Gamma_0, \ahat:\sort, \Gamma'_1}{\ditto}
        \Pf{}{}{\bhat \notin \dom{\Gamma_0, \ahat:\sort, \Gamma'_1}}{By (1) and (2) }
        \Hand \substextendPf{\Gamma_0, \Gamma'_1, \MonnierComma{\bhat}}{\Gamma_0, \ahat:\sort, \Gamma'_1, \MonnierComma{\bhat}}{By \substextendMonMon}
      \end{llproof}
    \end{itemize}

  \item Assume $\judgectx{\Gamma_0, \ahat:\sort, \Gamma_1}$. We proceed by induction on $\Gamma_1$: 

    \begin{itemize}
    \item Case $\Gamma_1 = \cdot$:

      \begin{llproof}
        \judgePf{\Gamma_0}{t}{\sort}{Given}
        \judgectxPf{\Gamma_0, \Gamma_1}{Given}
        \judgectxPf{\Gamma_0}{Since $\Gamma_1 = \cdot$} 
        \substextendPf{\Gamma_0}{\Gamma_0}{By \Lemmaref{lem:extension-reflexivity}}
         \substextendPf{\Gamma_0, \ahat:\sort}{\Gamma_0, \hypeq{\ahat:\sort}{t}}{By rule \substextendSolve}
         \Hand \substextendPf{\Gamma_0, \ahat:\sort, \Gamma_1}{\Gamma_0, \hypeq{\ahat:\sort}{t}, \Gamma_1}{Since $\Gamma_1 = \cdot$}
      \end{llproof}
      
    \item Case $\Gamma_1 = (\Gamma'_1, x:A)$:

      \begin{llproof}
        \judgePf{\Gamma_0}{t}{\sort}{Given}
        \judgectxPf{\Gamma_0, \ahat:\sort, \Gamma'_1, x:A}{Given}
        \judgectxPf{\Gamma_0, \ahat:\sort, \Gamma'_1}{By inversion}
        \judgetpPf{\Gamma_0, \ahat:\sort, \Gamma'_1}{A}{By inversion}
        \notinPf{x}{\dom{\Gamma_0, \ahat:\sort, \Gamma'_1}}{By inversion (1)}
        \substextendPf{\Gamma_0, \ahat:\sort, \Gamma'_1}{\Gamma_0, \hypeq{\ahat:\sort}{t}, \Gamma_1}{By i.h.}
        \judgetpPf{\Gamma_0, \hypeq{\ahat:\sort}{t}, \Gamma_1}{A}{By \Lemmaref{lem:extension-weakening}}
        \notinPf{x}{\dom{\Gamma_0, \hypeq{\ahat:\sort}{t}, \Gamma'_1}}{since this is the same domain as (1)}
        \substextendPf{\Gamma_0, \ahat:\sort, \Gamma'_1, x:A}{\Gamma_0, \hypeq{\ahat:\sort}{t}, \Gamma_1, x:A}{By rule \substextendVV}
      \end{llproof}
      
    \item Case $\Gamma_1 = (\Gamma'_1, \beta:\sort')$: 

      \begin{llproof}
        \judgePf{\Gamma_0}{t}{\sort}{Given}
        \judgectxPf{\Gamma_0, \ahat:\sort, \Gamma'_1, \beta:\sort'}{Given}
        \judgectxPf{\Gamma_0, \ahat:\sort, \Gamma'_1}{By inversion}
        \notinPf{\beta}{\dom{\Gamma_0, \ahat:\sort, \Gamma'_1}}{By inversion (1)}
        \substextendPf{\Gamma_0, \ahat:\sort, \Gamma'_1}{\Gamma_0, \hypeq{\ahat:\sort}{t}, \Gamma_1}{By i.h.}
        \notinPf{\beta}{\dom{\Gamma_0, \hypeq{\ahat:\sort}{t}, \Gamma'_1}}{since this is the same domain as (1)}
        \substextendPf{\Gamma_0, \ahat:\sort, \Gamma'_1, \beta:\sort'}{\Gamma_0, \hypeq{\ahat:\sort}{t}, \Gamma_1, \beta:\sort'}{By rule \substextendUU}
      \end{llproof}

    \item Case $\Gamma_1 = (\Gamma'_1, \bhat:\sort')$: 

      \begin{llproof}
        \judgePf{\Gamma_0}{t}{\sort}{Given}
        \judgectxPf{\Gamma_0, \ahat:\sort, \Gamma'_1, \bhat:\sort'}{Given}
        \judgectxPf{\Gamma_0, \ahat:\sort, \Gamma'_1}{By inversion}
        \notinPf{\bhat}{\dom{\Gamma_0, \ahat:\sort, \Gamma'_1}}{By inversion (1)}
        \substextendPf{\Gamma_0, \ahat:\sort, \Gamma'_1}{\Gamma_0, \hypeq{\ahat:\sort}{t}, \Gamma_1}{By i.h.}
        \notinPf{\bhat}{\dom{\Gamma_0, \hypeq{\ahat:\sort}{t}, \Gamma'_1}}{since this is the same domain as (1)}
        \substextendPf{\Gamma_0, \ahat:\sort, \Gamma'_1, \bhat:\sort'}{\Gamma_0, \hypeq{\ahat:\sort}{t}, \Gamma_1, \bhat:\sort'}{By rule \substextendEE}
      \end{llproof}

    \item Case $\Gamma_1 = (\Gamma'_1, \hypeq{\bhat:\sort'}{t'})$: 

      \begin{llproof}
        \judgePf{\Gamma_0}{t'}{\sort}{Given}
        \judgectxPf{\Gamma_0, \ahat:\sort, \Gamma'_1, \hypeq{\bhat:\sort'}{t'}}{Given}
        \judgectxPf{\Gamma_0, \ahat:\sort, \Gamma'_1}{By inversion}
        \judgePf{\Gamma_0, \ahat:\sort, \Gamma'_1}{t'}{\sort'}{By inversion}
        \notinPf{\bhat}{\dom{\Gamma_0, \ahat:\sort, \Gamma'_1}}{By inversion (1)}
        \substextendPf{\Gamma_0, \ahat:\sort, \Gamma'_1}{\Gamma_0, \hypeq{\ahat:\sort}{t}, \Gamma_1}{By i.h.}
        \notinPf{\bhat}{\dom{\Gamma_0, \hypeq{\ahat:\sort}{t}, \Gamma'_1}}{since this is the same domain as (1)}
        \proofsep
        \judgePf{\Gamma_0, \hypeq{\ahat:\sort}{t}, \Gamma_1}{t'}{\sort'}{\hspace{-8ex}By \Lemmaref{lem:extension-weakening}}
        \decolumnizePf
        \substextendPf{\Gamma_0, \ahat:\sort, \Gamma'_1, \hypeq{\bhat:\sort'}{t'}}{\Gamma_0, \hypeq{\ahat:\sort}{t'}, \Gamma_1, \hypeq{\bhat:\sort'}{t'}}{By rule \substextendSolved}
      \end{llproof}

    \item Case $\Gamma_1 = (\Gamma'_1, \hypeq{\beta}{t'})$: 

      \begin{llproof}
        \judgePf{\Gamma_0}{t'}{\sort}{Given}
        \judgectxPf{\Gamma_0, \ahat:\sort, \Gamma'_1, \hypeq{\beta}{t'}}{Given}
        \judgectxPf{\Gamma_0, \ahat:\sort, \Gamma'_1}{By inversion}
        \judgePf{\Gamma_0, \ahat:\sort, \Gamma'_1}{t'}{\ind}{By inversion}
        \notinPf{\beta}{\dom{\Gamma_0, \ahat:\sort, \Gamma'_1}}{By inversion (1)}
        \substextendPf{\Gamma_0, \ahat:\sort, \Gamma'_1}{\Gamma_0, \hypeq{\ahat:\sort}{t}, \Gamma_1}{By i.h.}
        \notinPf{\beta}{\dom{\Gamma_0, \hypeq{\ahat:\sort}{t}, \Gamma'_1}}{since this is the same domain as (1)}
        \judgePf{\Gamma_0, \hypeq{\ahat:\sort}{t}, \Gamma_1}{t'}{\ind}{\hspace{-6ex}By \Lemmaref{lem:extension-weakening}}
        \substextendPf{\Gamma_0, \ahat:\sort, \Gamma'_1, \hypeq{\beta}{t'}}{\Gamma_0, \hypeq{\ahat:\sort}{t'}, \Gamma_1, \hypeq{\beta}{t'}}{By rule \substextendEqn}
      \end{llproof}
    \item Case $\Gamma_1 = (\Gamma'_1, \MonnierComma{\bhat})$: 

      \begin{llproof}
        \judgePf{\Gamma_0}{t}{\sort}{Given}
        \judgectxPf{\Gamma_0, \ahat:\sort, \Gamma'_1, \MonnierComma{\bhat}}{Given}
        \judgectxPf{\Gamma_0, \ahat:\sort, \Gamma'_1}{By inversion}
        \notinPf{\bhat}{\dom{\Gamma_0, \ahat:\sort, \Gamma'_1}}{By inversion (1)}
        \substextendPf{\Gamma_0, \ahat:\sort, \Gamma'_1}{\Gamma_0, \hypeq{\ahat:\sort}{t}, \Gamma_1}{By i.h.}
        \notinPf{\bhat}{\dom{\Gamma_0, \hypeq{\ahat:\sort}{t}, \Gamma'_1}}{since this is the same domain as (1)}
        \substextendPf{\Gamma_0, \ahat:\sort, \Gamma'_1, \MonnierComma{\bhat}}{\Gamma_0, \hypeq{\ahat:\sort}{t}, \Gamma_1, \MonnierComma{\bhat}}{By rule \substextendEE}
      \end{llproof}
    \end{itemize}

  \item Apply parts (i) and (ii) as lemmas, then \Lemmaref{lem:extension-transitivity}.
  \qedhere
  \end{enumerate}
\end{proof}

\parallelevaradmissibility*
\begin{proof}\proofheading
  By induction on $\Delta_R$.  As always, we assume that all contexts mentioned in the
  statement of the lemma are well-formed.  Hence, $\ahat \notin \dom{\Gamma_L} \union \dom{\Gamma_R}
  \union \dom{\Delta_L} \union \dom{\Delta_R}$.
  
  \begin{enumerate}[(i)]
  \item We proceed by cases of $\Delta_R$.  Observe that in all the extension rules, the right-hand context
    gets smaller, so as we enter subderivations of
    $\substextend{\Gamma_L, \Gamma_R}{\Delta_L, \Delta_R}$, the context $\Delta_R$ becomes smaller.
    
    The only tricky part of the proof is that to apply the i.h., we need $\substextend{\Gamma_L}{\Delta_L}$.
    So we need to make sure that as we drop items from the right of $\Gamma_R$ and $\Delta_R$, we don't
    go too far and start decomposing $\Gamma_L$ or $\Delta_L$!  It's easy to avoid decomposing $\Delta_L$:
    when $\Delta_R = \cdot$, we don't need to apply the i.h.\ anyway.  To avoid decomposing $\Gamma_L$,
    we need to reason by contradiction, using \Lemmaref{lem:declaration-preservation}.
    
        \begin{itemize}
        \ProofCaseRule{$\Delta_R = \cdot$} \\
            We have $\substextend{\Gamma_L}{\Delta_L}$.
            Applying \substextendEE to that derivation gives the result.
    
        \ProofCaseRule{$\Delta_R = (\Delta_R', \bhat)$}
            We have $\bhat \neq \ahat$ by the well-formedness assumption.

            The concluding rule of $\substextend{\Gamma_L, \Gamma_R}{\Delta_L, \Delta_R', \bhat}$
            must have been \substextendEE or \substextendAdd.  In both cases, the result follows by i.h.
            and applying \substextendEE or \substextendAdd.

            Note:  In \substextendAdd, the left-hand context doesn't change, so we clearly maintain
            $\substextend{\Gamma_L}{\Delta_L}$.  In \substextendEE, we can correctly apply the
            i.h.\ because $\Gamma_R \neq \cdot$.  Suppose, for
            a contradiction, that $\Gamma_R = \cdot$.  Then $\Gamma_L = (\Gamma_L', \bhat)$.
            It was given that $\substextend{\Gamma_L}{\Delta_L}$, that is, $\substextend{\Gamma_L', \bhat}{\Delta_L}$.
            By \Lemmaref{lem:declaration-preservation}, $\Delta_L$ has a declaration of $\bhat$.
            But then $\Delta = (\Delta_L, \Delta_R', \bhat)$ is not well-formed: contradiction.  Therefore
            $\Gamma_R \neq \cdot$.

        \ProofCaseRule{$\Delta_R = (\Delta_R', \hypeq{\bhat:\sort}{t})$}
            We have $\bhat \neq \ahat$ by the well-formedness assumption.

            The concluding rule must have been \substextendSolved, \substextendSolve or \substextendAddSolved.
            In each case, apply the i.h.\ and then the corresponding rule.  (In \substextendSolved and \substextendSolve,
            use \Lemmaref{lem:declaration-preservation} to show $\Gamma_R \neq \cdot$.)

        \ProofCaseRule{$\Delta_R = (\Delta_R', \alpha)$}
           The concluding rule must have been \substextendUU.  The result follows by i.h.\ and
           applying \substextendUU.

        \ProofCaseRule{$\Delta_R = (\Delta_R', \hypeq{\alpha}{\tau})$}
           The concluding rule must have been \substextendEqn.  The result follows by i.h.\ and
           applying \substextendEqn.

        \ProofCaseRule{$\Delta_R = (\Delta_R', \MonnierComma{\bhat})$}  Similar to the previous
           case, with rule \substextendMonMon.
        
        \ProofCaseRule{$\Delta_R = (\Delta_R', x : A)$}  Similar to the previous case,
           with rule \substextendVV.
        \end{itemize}

  \item Similar to part (i), except that when $\Delta_R = \cdot$, apply rule \substextendSolve.

  \item Similar to part (i), except that when $\Delta_R = \cdot$, apply rule \substextendSolved,
     using the given equality to satisfy the second premise.
  \qedhere
  \end{enumerate}
\end{proof}

\parallelextensionsolve*
\begin{proof}\proofheading
  By induction on $\Delta_R$.

  In the case where $\Delta_R = \cdot$,
  we know that rule \substextendSolve must have concluded the derivation (we can use
  \Lemmaref{lem:declaration-preservation} to get a contradiction that rules out \substextendAddSolved);
  then we have a subderivation $\substextend{\Gamma_L}{\Delta_L}$, to which we can apply
  \substextendSolved.
\end{proof}

\parallelextensionupdate*
\begin{proof}\proofheading
  By induction on $\Delta_R$.  Similar to the proof of \Lemmaref{lem:parallel-extension-solve},
  but applying \substextendSolved at the end.
\end{proof}

\substitutionmonotonicity*
\begin{proof}\proofheading
We prove each part in turn; part (i) does not depend on parts (ii) or (iii), so we can use part (i) as a lemma
in the proofs of parts (ii) and (iii).

\begin{itemize}
\item \textbf{Proof of Part (i):}
  By lexicographic induction on the derivation of $\Dee \derives \substextend{\Gamma}{\Delta}$ and $\judge{\Gamma}{t}{\sort}$.
  We proceed by cases on the derivation of $\judge{\Gamma}{t}{\sort}$. 

      \begin{itemize}
          \DerivationProofCase{\VarSort}
                {\ahat:\sort \in \Gamma}
                {\judge{\Gamma}{\ahat}{\sort}}

            \begin{llproof}
              \eqPf{[\Gamma]\ahat}{\ahat}{Since $\ahat$ is not solved in $\Gamma$}
              \eqPf{[\Delta]\ahat}{[\Delta]\ahat}{Reflexivity}
              \continueeqPf{[\Delta][\Gamma]\ahat}{By above equality}
            \end{llproof}

          \DerivationProofCase{\VarSort}
                {(\alpha:\sort) \in \Gamma}
                {\judge{\Gamma}{\alpha}{\sort}}

               Consider whether or not there is a binding of the form $(\hypeq{\alpha}{\tau}) \in \Gamma$. 

               \begin{itemize}
                 \ProofCaseRule{$(\hypeq{\alpha}{\tau}) \in \Gamma$}

                     \begin{llproof}
                       \eqPf{\Delta}{(\Delta_0, \hypeq{\alpha}{\tau'}, \Delta_1)}{By \Lemmaref{lem:extension-inversion} (i)}
                       \derivesPf{\Dee'} \substextendPf{\Gamma_0}{\Delta_0}{\ditto}
                       \ltPf{\Dee'}{\Dee}{\ditto}
                       (1)~~\eqPf{[\Delta_0]\tau'}{[\Delta_0]\tau}{\ditto}
                       (2)~~\eqPf{[\Delta_0][\Gamma_0]\tau}{[\Delta_0]\tau}{By i.h.}
                       \eqPf{[\Delta][\Gamma]\alpha}{[\Delta_0, \hypeq{\alpha}{\tau'}, \Delta_1][\Gamma_0, \hypeq{\alpha}{\tau}, \Gamma_1]\alpha}{By definition}
                       \continueeqPf{[\Delta_0, \hypeq{\alpha}{\tau'}, \Delta_1][\Gamma_0, \hypeq{\alpha}{\tau}]\alpha}{Since $\alpha \notin \dom{\Gamma_1}$}
                       \continueeqPf{[\Delta_0, \hypeq{\alpha}{\tau'}, \Delta_1][\Gamma_0]\tau}{By definition of substitution}
                       \continueeqPf{[\Delta_0][\Gamma_0]\tau}{Since $\FV{[\Gamma_0]\tau} \sect \dom{\Delta_1} = \emptyset$}
                       \continueeqPf{[\Delta_0]\tau'}{By (2) and (1)}
                       \continueeqPf{[\Delta_0, \hypeq{\alpha}{\tau'}]\alpha}{By definition of substitution}
                       \continueeqPf{[\Delta_0, \hypeq{\alpha}{\tau'}, \Delta_1]\alpha}{Since $\FV{[\Delta_0]\tau} \sect \dom{\Delta_1} = \emptyset$}
                       \continueeqPf{[\Delta]\alpha}{By definition of $\Delta$}
                     \end{llproof}

                 \ProofCaseRule{$(\hypeq{\alpha}{\tau}) \notin \Gamma$}

                     \begin{llproof}
                       \eqPf{[\Gamma]\alpha}{\alpha}{By definition of substitution}
                       \eqPf{[\Delta][\Gamma]\alpha}{[\Delta]\alpha}{Apply $[\Delta]$ to both sides}
                     \end{llproof}
               \end{itemize}

          \DerivationProofCase{\SolvedVarSort}
                {  }
                { \judge{\Gamma_0, \hypeq{\ahat:\sort}{\tau}, \Gamma_1}{\ahat}{\sort} }

               Similar to the \VarSort case. 

          \DerivationProofCase{\UnitSort}
                { }
                { \judge{\Gamma}{\unitty}{\type} }

                \begin{llproof}
                  \eqPf{[\Delta]\unitty}{\unitty = [\Delta][\Gamma]\unitty}{Since $\FV{\unitty} = \emptyset$}
                \end{llproof}

          \DerivationProofCase{\BinSort}
                { \judge{\Gamma}{\tau_1}{\type}
                  \\ 
                  \judge{\Gamma}{\tau_2}{\type} }
                { \judge{\Gamma}{\tau_1 \binc \tau_2}{\type} }

                \begin{llproof}
                  \eqPf{[\Delta][\Gamma]\tau_1}{[\Delta]\tau_1}{By i.h.}
                  \eqPf{[\Delta][\Gamma]\tau_2}{[\Delta]\tau_2}{By i.h.}
                  \eqPf{[\Delta][\Gamma]\tau_1 \binc [\Delta][\Gamma]\tau_2}
                       {[\Delta]\tau_1 \binc [\Delta]\tau_2}
                       {By congruence of equality}
                  \eqPf{[\Delta][\Gamma](\tau_1 \binc \tau_2)}
                       {[\Delta](\tau_1 \binc \tau_2)}
                       {Definition of substitution}
                \end{llproof}

          \DerivationProofCase{\ZeroSort}
                { }
                {\judge{\Gamma}{\zero}{\ind}}

                \begin{llproof}
                  \eqPf{[\Delta]\zero}{\zero  =  [\Delta][\Gamma]\zero}{Since $\FV{\zero} = \emptyset$}
                \end{llproof}

          \DerivationProofCase{\SuccSort}
                {\judge{\Gamma}{t}{\ind} }
                {\judge{\Gamma}{\succ{t}}{\ind}}

                \begin{llproof}
                  \eqPf{[\Delta][\Gamma]t}{[\Delta]t}{By i.h.}
                  \eqPf{\succ{{[\Delta][\Gamma]t}}}{\succ{{[\Delta]t}}}{By congruence of equality}
                  \eqPf{[\Delta][\Gamma]\succ{t}}{[\Delta]\succ{t}}{By definition of substitution}
                \end{llproof}
      \end{itemize}

 \item \textbf{Proof of Part (ii):} We have a derivation of $\judgeprop{\Gamma}{P}$, and will use the previous part as a lemma.

      \begin{itemize}
          \DerivationProofCase{\EqProp}
                {\judge{\Gamma}{t}{\ind} \\
                 \judge{\Gamma}{t'}{\ind} 
                }
                {\judgeprop{\Gamma}{t = t'}}

              \begin{llproof}
                \eqPf{[\Delta][\Gamma]t}{[\Delta]t}{By part (i)}
                \eqPf{[\Delta][\Gamma]t'}{[\Delta]t'}{By part (i)}
                \eqPfParen{[\Delta][\Gamma]t = [\Delta][\Gamma]t'}{[\Delta]t = [\Delta]t'}{By congruence of equality}
                \eqPf{[\Delta][\Gamma](t = t')}{[\Delta](t = t')}{Definition of substitution}
              \end{llproof}
      \end{itemize}

 \item
   \textbf{Proof of Part (iii):}
     By induction on the derivation of $\judgetp{\Gamma}{A}$, using the previous parts as lemmas.

     \begin{itemize}
          \DerivationProofCase{\VarWF}
                { (u : \type) \in \Gamma }
                { \judgetp{\Gamma}{u} }

                \begin{llproof}
                  \judgePf{\Gamma}{u}{\type}{By rule \VarSort}
                  \eqPf{[\Delta][\Gamma]u}{[\Delta]u}{By part (i)}
                \end{llproof}

          \DerivationProofCase{\SolvedVarWF}
                { (\hypeq{\ahat:\type}{\tau}) \in \Gamma }
                { \judgetp{\Gamma}{\ahat} }

                \begin{llproof}
                  \judgePf{\Gamma}{\ahat}{\type}{By rule \SolvedVarSort}
                  \eqPf{[\Delta][\Gamma]\ahat}{[\Delta]\ahat}{By part (i)}
                \end{llproof}

          \DerivationProofCase{\UnitWF}
                { }
                { \judgetp{\Gamma}{\unitty} }

                \begin{llproof}
                  \judgePf{\Gamma}{\unitty}{\type}{By rule \UnitSort}
                  \eqPf{[\Delta][\Gamma]\unitty}{[\Delta]\unitty}{By part (i)}
                \end{llproof}

          \DerivationProofCase{\BinWF}
                { \judgetp{\Gamma}{A_1}
                  \\
                  \judgetp{\Gamma}{A_2}
                }
                { \judgetp{\Gamma}{A_1 \binc A_2} }

                \begin{llproof}
                  \eqPf{[\Delta][\Gamma]A_1}{[\Delta]A_1}{By i.h.}
                  \eqPf{[\Delta][\Gamma]A_2}{[\Delta]A_2}{By i.h.}
                  \eqPf{[\Delta][\Gamma]A_1 \binc [\Delta][\Gamma]A_2}
                       {[\Delta]A_1 \binc [\Delta]A_2}
                       {By congruence of equality}
                  \eqPf{[\Delta][\Gamma](A_1 \binc A_2)}
                       {[\Delta](A_1 \binc A_2)}
                       {Definition of substitution}
                \end{llproof}

          \ProofCaseRule{\VecWF}  Similar to the \BinWF case.

          \DerivationProofCase{\ForallWF}
                { \judgetp{\Gamma, \alpha:\sort}{A_0} }
                { \judgetp{\Gamma}{\alltype{\alpha:\sort} A_0} }

            \begin{llproof}
              \substextendPf{\Gamma}{\Delta}{Given}
              \substextendPf{\Gamma, \alpha:\sort}{\Delta, \alpha:\sort}{By rule \substextendUU}
              \eqPf{[\Delta, \alpha:\sort][\Gamma, \alpha:\sort]A_0}{[\Delta, \alpha:\sort]A_0}{By i.h.}
              \eqPf{[\Delta][\Gamma]A_0}{[\Delta]A_0}{By definition of substitution}
              \eqPf{\alltype{\alpha:\sort}{[\Delta][\Gamma]A_0}}{\alltype{\alpha:\sort}{[\Delta]A_0}}{By congruence of equality}
              \eqPf{[\Delta][\Gamma](\alltype{\alpha:\sort}{A_0})}{[\Delta](\alltype{\alpha:\sort}{A_0})}{By definition of substitution}
            \end{llproof}

         \ProofCaseRule{\ExistsWF}  Similar to the \ForallWF case.

         \DerivationProofCase{\ImpliesWF}
                {\judgeprop{\Gamma}{P} \\ 
                 \judgetp{\Gamma}{A_0}}
                {\judgetp{\Gamma}{P \implies A_0}}

                \begin{llproof}
                  \eqPf{[\Delta][\Gamma]P}{[\Delta]P}{By part (ii)}
                  \eqPf{[\Delta][\Gamma]A_0}{[\Delta]A_0}{By i.h.}
                  \eqPf{[\Delta][\Gamma]P \implies [\Delta][\Gamma]A_0}
                       {[\Delta]P \implies [\Delta]A_0}
                       {By congruence of equality}
                  \eqPf{[\Delta][\Gamma](P \implies A_0)}
                       {[\Delta](P \implies A_0)}
                       {Definition of substitution}
                \end{llproof}

          \DerivationProofCase{\WithWF}
                {\judgeprop{\Gamma}{P} \\ 
                 \judgetp{\Gamma}{A_0}}
                {\judgetp{\Gamma}{A_0 \with P}}

            Similar to the \ImpliesWF case.
     \qedhere
     \end{itemize}  %
  \end{itemize}  %
\end{proof}

\substitutioninvariance*
\begin{proof}\proofheading
  Each part is a separate induction, relying on the proofs of the
  earlier parts.  In each part, the result follows by an induction on
  the derivation of $\substextend{\Gamma}{\Delta}$.

  The main observation is that $\Delta$ adds no equations for any 
  variable of $t$, $P$, and $A$ that $\Gamma$ does not already contain,
  and as a result applying $\Delta$ as a substitution to $[\Gamma]t$ 
  does nothing. 
\end{proof}

\softextension*
\begin{proof}\proofheading
  By induction on $\Theta$.
  
  \begin{itemize}
    \ProofCaseRule{$\Theta = \cdot$}
         We have $\substextend{\Gamma}{\Delta}$.
         Let $\Omega = \cdot$.  Then $\substextend{\Gamma, \Theta}{\Delta, \Omega}$.

    \ProofCaseRule{$\Theta = (\Theta', \hypeq{\ahat:\sort}{t})$}

        \begin{llproof}
          \substextendPf{\Gamma, \Theta'}{\Gamma, \Omega'}{By i.h.}
          \Hand \substextendPf{\Gamma, \underbrace{\Theta', \hypeq{\ahat:\sort}{t}}_{\Theta}}{\Delta, \underbrace{\Omega', \hypeq{\ahat:\sort}{t}}_{\Omega}}{By rule \substextendSolved}
        \end{llproof}

    \ProofCaseRule{$\Theta = (\Theta', \ahat:\sort)$}
    
        If $\sort = \type$, let $t = \unitty$; if $\sort = \ind$, let $t = \zero$. 
        
        \begin{llproof}
          \substextendPf{\Gamma, \Theta'}{\Gamma, \Omega'}{By i.h.}
          \Hand \substextendPf{\Gamma, \underbrace{\Theta', \ahat:\sort}_{\Theta}}{\Delta, \underbrace{\Omega', \hypeq{\ahat:\sort}{t}}_{\Omega}}{By rule \substextendSolve}
        \end{llproof}
  \qedhere
  \end{itemize}
\end{proof}

\splitextension*
\begin{proof}\proofheading
  By induction on the derivation of $\substextend{\Delta}{\Omega}$.
  Use the fact that $\Omega_1[\hypeq{\ahat:\sort}{t_1}]$
  and $\Omega_1[\hypeq{\ahat:\sort}{t_2}]$
  agree on all solutions \emph{except} the solution for $\ahat$.
  In the \substextendSolve case where
  the existential variable is $\ahat$, use $\judge{\Omega}{t_2}{\sort}$.
\end{proof}

\subsection{Reflexivity and Transitivity}

\extensionreflexivity*
\begin{proof}\proofheading
  By induction on the derivation of $\judgectx{\Gamma}$.

  \begin{itemize}
      \DerivationProofCase{\EmptyCWF}
            { }
            { \judgectx{\cdot} }

      \begin{llproof}
        \substextendPf{\cdot}{\cdot}{By rule \substextendId}
      \end{llproof}
      
      \DerivationProofCase{\HypCWF}
            { \judgectx{\Gamma}
              \\
              x \notin \dom{\Gamma}
               \\
              \judgetp{\Gamma}{A}
            }
            { \judgectx{\Gamma, x : A} }

     \begin{llproof}
       \substextendPf{\Gamma}{\Gamma}{By i.h.}
       \eqPf{[\Gamma]A}{[\Gamma]A}{By reflexivity}
       \substextendPf{\Gamma, x:A}{\Gamma, x:A}{By rule \substextendVV}
     \end{llproof}

      \DerivationProofCase{\VarCWF}
            { \judgectx{\Gamma}
              \\
              u:\sort \notin \dom{\Gamma}}
            { \judgectx{\Gamma, u:\sort} }

     \begin{llproof}
       \substextendPf{\Gamma}{\Gamma}{By i.h.}
       \substextendPf{\Gamma, u:\sort}{\Gamma, u:\sort}{By rule \substextendUU or \substextendEE}
     \end{llproof}
      
      \DerivationProofCase{\SolvedCWF }
            { \judgectx{\Gamma}
              \\
              \ahat \notin \dom{\Gamma} 
              \\
              \judge{\Gamma}{t}{\sort} 
            }
            { \judgectx{\Gamma, \hypeq{\ahat:\sort}{t}}
            }

     \begin{llproof}
       \substextendPf{\Gamma}{\Gamma}{By i.h.}
       \eqPf{[\Gamma]t}{[\Gamma]t}{By reflexivity}
       \substextendPf{\Gamma, \hypeq{\ahat:\sort}{t}}{\Gamma, \hypeq{\ahat:\sort}{t}}{By rule \substextendSolved}
     \end{llproof}
      
      \DerivationProofCase{\EqnCWF}
            { \judgectx{\Gamma}
              \\
              \alpha:\sort \in \Gamma
              \\
              (\hypeq{\alpha}{-}) \notin \Gamma
              \\
              \judge{\Gamma}{\tau}{\sort}
            }
            { \judgectx{\Gamma, \hypeq{\alpha}{\tau}} }
      
     \begin{llproof}
       \substextendPf{\Gamma}{\Gamma}{By i.h.}
       \eqPf{[\Gamma]t}{[\Gamma]t}{By reflexivity}
       \substextendPf{\Gamma, \hypeq{\alpha}{t}}{\Gamma, \hypeq{\alpha}{t}}{By rule \substextendEqn}
     \end{llproof}

      \DerivationProofCase{\MarkerCWF}
            { \judgectx{\Gamma}
              \\
              \MonnierComma{u} \notin \Gamma
            }
            { \judgectx{\Gamma, \MonnierComma{u}} }

     \begin{llproof}
       \substextendPf{\Gamma}{\Gamma}{By i.h.}
       \substextendPf{\Gamma, \MonnierComma{u}}{\Gamma, \MonnierComma{u}}{By rule \substextendMonMon}
     \end{llproof}
  \qedhere
  \end{itemize}
\end{proof}

\extensiontransitivity*
\begin{proof}\proofheading
  By induction on $\Dee'$. 

  \begin{itemize}
    \DerivationProofCase{\substextendId}
        { }
        { \substextend{\underbrace{\emptyctx}_{\Theta}}{\underbrace{\emptyctx}_{\Delta}} }

    \begin{llproof}
      \eqPf{\Gamma}{\cdot}{By inversion on $\Dee$}
      \substextendPf{\cdot}{\cdot}{By rule \substextendId}
      \substextendPf{\Gamma}{\Delta}{Since $\Gamma = \Delta = \cdot$}
    \end{llproof}
    
    \DerivationProofCase{\substextendVV}
          {\substextend{\Theta'}{\Delta'}
           \\
           [\Delta']A = [\Delta']A'
          }
          {\substextend{\underbrace{\Theta', x:A}_{\Theta}}{\underbrace{\Delta', x:A'}_{\Delta}}}

     \begin{llproof}
       \eqPfParenR{\Gamma}{\Gamma',x:A''}{By inversion on $\Dee$}
       \eqPf{[\Theta]A''}{[\Theta]A}{By inversion on $\Dee$}
       \substextendPf{\Gamma'}{\Theta'}{By inversion on $\Dee$}
       \substextendPf{\Gamma'}{\Delta'}{By i.h.} 
       \eqPf{[\Delta'][\Theta']A''}{[\Delta'][\Theta']A}{By congruence of equality}
       \eqPf{[\Delta']A''}{[\Delta']A}{By \Lemmaref{lem:substitution-monotonicity}}
       \continueeqPf{[\Delta']A'}{By premise $[\Delta']A = [\Delta']A'$}
       \substextendPf{\Gamma', x:A''}{\Delta', x:A'}{By \substextendVV}
     \end{llproof}
    
    \DerivationProofCase{\substextendUU}
        { \substextend{\Theta'}{\Delta'} }
        { \substextend{\underbrace{\Theta', \alpha:\sort}_{\Theta}}{\underbrace{\Delta', \alpha:\sort}_{\Delta}} }
    
     \begin{llproof}
       \eqPfParenR{\Gamma}{\Gamma', \alpha:\sort}{By inversion on $\Dee$}
       \substextendPf{\Gamma'}{\Theta'}{By inversion on $\Dee$}
       \substextendPf{\Gamma'}{\Delta'}{By i.h.} 
       \substextendPf{\Gamma', \alpha:\sort}{\Delta', \alpha:\sort}{By \substextendUU}
     \end{llproof}

    \DerivationProofCase{\substextendEE}
        { \substextend{\Theta'}{\Delta'} }
        { \substextend{\underbrace{\Theta', \ahat:\sort}_{\Theta}}{\underbrace{\Delta', \ahat:\sort}_{\Delta}} }

    Two rules could have concluded $\Dee \derives \substextend{\Gamma}{(\Theta', \ahat:\sort)}$:
    \begin{itemize}
    \DerivationProofCase{\substextendEE}
                        { \substextend{\Gamma'}{\Theta'} }
                        { \substextend{\underbrace{\Gamma', \ahat:\sort}_{\Gamma}}{\Theta', \ahat:\sort} }

      \begin{llproof}
        \substextendPf{\Gamma'}{\Delta'}{By i.h.}
        \substextendPf{\Gamma', \ahat:\sort}{\Delta', \ahat:\sort}{By rule \substextendAdd}
      \end{llproof}

    \DerivationProofCase{\substextendAdd}
                        { \substextend{\Gamma}{\Theta'} }
                        { \substextend{\Gamma}{\Theta', \ahat:\sort} }

      \begin{llproof}
        \substextendPf{\Gamma}{\Delta'}{By i.h.}
        \substextendPf{\Gamma}{\Delta', \ahat:\sort}{By rule \substextendAdd}
      \end{llproof}
    \end{itemize}

    \DerivationProofCase{\substextendSolved}
        { \substextend{\Theta'}{\Delta'} \\ [\Delta']t = [\Delta']t'}
        { \substextend{\underbrace{\Theta', \hypeq{\ahat:\sort}{t}}_{\Theta}}{\underbrace{\Delta', \hypeq{\ahat:\sort}{t'}}_{\Delta}}
        }

    Two rules could have concluded $\Dee \derives \substextend{\Gamma}{(\Theta', \hypeq{\ahat:\sort}{t})}$:

        \begin{itemize}
            \DerivationProofCase{\substextendSolved}
                    { \substextend{\Gamma'}{\Theta'} \\ [\Theta']t'' = [\Theta']t }
                    { \substextend{\underbrace{\Gamma', \hypeq{\ahat:\sort}{t''}}_{\Gamma}}{\Theta', \hypeq{\ahat:\sort}{t}} }

                \begin{llproof}
                  \substextendPf{\Gamma'}{\Delta'}{By i.h.}
                  \eqPf{[\Theta']t''} {[\Theta']t}  {Premise}
                  \eqPf{[\Delta'][\Theta']t''}{[\Delta'][\Theta']t}{Applying $\Delta'$ to both sides}
                  \eqPf{[\Delta']t''}{[\Delta']t}{By \Lemmaref{lem:substitution-monotonicity}}
                  \continueeqPf{[\Delta']t'}{By premise $[\Delta']t = [\Delta']t'$}
                  \substextendPf{\Gamma', \hypeq{\ahat:\sort}{t''}}{\Delta', \hypeq{\ahat:\sort}{t'}}{By rule \substextendSolved}
                \end{llproof}                        

            \DerivationProofCase{\substextendAddSolved}
                   { \substextend{\Gamma}{\Theta'} }
                   { \substextend{\Gamma}{\Theta', \hypeq{\ahat:\sort}{t}} }

                \begin{llproof}
                  \substextendPf{\Gamma}{\Delta'}{By i.h.}
                  \substextendPf{\Gamma}{\Delta', \hypeq{\ahat:\sort}{t'}}{By rule \substextendAddSolved}
                \end{llproof}
      \end{itemize}    
  
  \DerivationProofCase{\substextendEqn}
        { \substextend{\Theta'}{\Delta'} \\ [\Delta']t = [\Delta']t'}
        { \substextend{\underbrace{\Theta', \hypeq{\alpha}{t}}_{\Theta}}{\underbrace{\Delta', \hypeq{\alpha}{t'}}_{\Delta}}
        }

    \begin{llproof}
      \eqPf{\Gamma}{(\Gamma', \hypeq{\alpha}{t''})}{By inversion on $\Dee$}
      \substextendPf{\Gamma'}{\Theta'}{By inversion on $\Dee$}
      \eqPf{[\Theta']t''}{[\Theta']t}{By inversion on $\Dee$}
      \eqPf{[\Delta'][\Theta']t''}{[\Delta'][\Theta']t}{Applying $\Delta'$ to both sides}
      \substextendPf{\Gamma'}{\Delta'}{By i.h.}
      \eqPf{[\Delta']t''}{[\Delta']t}{By \Lemmaref{lem:substitution-monotonicity}}
      \continueeqPf{[\Delta']t'}{By premise $[\Delta']t = [\Delta']t'$}
      \substextendPf{\Gamma', \hypeq{\alpha}{t''}}{\Delta', \hypeq{\alpha}{t'}}{By rule \substextendEqn}
    \end{llproof}

    \DerivationProofCase{\substextendAdd}
        { \substextend{\Theta}{\Delta'} }
        { \substextend{\Theta}{\underbrace{\Delta', \ahat:\sort}_{\Delta}} }

        \begin{llproof}
          \substextendPf{\Gamma}{\Delta'}{By i.h.}
          \substextendPf{\Gamma}{\Delta', \ahat:\sort}{By rule \substextendAdd}
        \end{llproof}
    
    \DerivationProofCase{\substextendAddSolved}
        { \substextend{\Theta}{\Delta'} }
        { \substextend{\Theta}{\underbrace{\Delta', \hypeq{\ahat:\sort}{t}}_{\Delta}} }
    
        \begin{llproof}
          \substextendPf{\Gamma}{\Delta'}{By i.h.}
          \substextendPf{\Gamma}{\Delta', \hypeq{\ahat:\sort}{t}}{By rule \substextendAddSolved}
        \end{llproof}

    \DerivationProofCase{\substextendMonMon}
        { \substextend{\Theta'}{\Delta'} }
        { \substextend{\underbrace{\Theta', \MonnierComma{u}}_{\Theta}}{\underbrace{\Delta', \MonnierComma{u}}_{\Delta}} }

         \begin{llproof}
           \eqPf{\Gamma}{\Gamma',\MonnierComma{u}}{By inversion on $\Dee$}
           \substextendPf{\Gamma'}{\Theta'}{By inversion on $\Dee$}
           \substextendPf{\Gamma'}{\Delta'}{By i.h.} 
           \substextendPf{\Gamma', \MonnierComma{u}}{\Delta', \MonnierComma{u}}{By \substextendUU}
         \end{llproof}
  \qedhere
  \end{itemize}
\end{proof}

\subsection{Weakening}

\suffixweakeningjudge*
\begin{proof}\proofheading
  By induction on the given derivation.  All cases are straightforward.
\end{proof}

\suffixweakeningjudgetp*
\begin{proof}\proofheading
  By induction on the given derivation.  All cases are straightforward.
\end{proof}

\extensionweakening*
\begin{proof}\proofheading
  By a straightforward induction on $\judge{\Gamma}{t}{\sort}$.

  In the \VarSort case, use \Lemmaref{lem:extension-inversion} (i) or (v).
  In the \SolvedVarSort case, use \Lemmaref{lem:extension-inversion} (iv).
  In the other cases, apply the i.h.\ to all subderivations, then apply the rule.
\end{proof}

\extensionweakeningprop*
\begin{proof}\proofheading
  By inversion on rule \EqProp, and \Lemmaref{lem:extension-weakening} twice.
\end{proof}

\extensionweakeningtp*
\begin{proof}\proofheading
  By a straightforward induction on $\judgetp{\Gamma}{A}$.

  In the \VarWF case, use \Lemmaref{lem:extension-inversion} (i) or (v).
  In the \SolvedVarWF case, use \Lemmaref{lem:extension-inversion} (iv).
  
  In the other cases, apply the i.h.\ and/or (for \ImpliesWF and \WithWF)
  \Lemmaref{lem:extension-weakening-prop} to all subderivations, then apply the rule.
\end{proof}

\subsection{Principal Typing Properties}

\substitutiontppstable*
\begin{proof}\proofheading
  By induction on the derivation of $\substextend{\Gamma}{\Delta}$.

  Part (i):
  
  \begin{itemize}
       \DerivationProofCase{\substextendEqn}
            {\substextend{\Gamma_0}{\Delta_0}
              \\
              [\Delta_0]t = [\Delta_0]t'
            }
            { \substextend{\Gamma_0, \hypeq{\alpha}{t}}{\underbrace{\Delta_0, \hypeq{\alpha}{t'}}_{\Delta}}
            }

            If $\alpha \notin \FV{A}$, then:

            \begin{llproof}
              \eqPf   {[\Gamma_0, \hypeq{\alpha}{t}]A}
                           {[\Gamma_0]A}   {By def.\ of subst.}
              \continueeqPf{[\Delta_0]A}   {By i.h.}
              \continueeqPf{[\Delta_0, \hypeq{\alpha}{t'}]A}   {By def.\ of subst.}
            \end{llproof}

            Otherwise, $\alpha \in \FV{A}$.

            \begin{llproof}
              \judgetpPf{\Gamma_0}{t}   {$\Gamma$ is well-formed}
              \judgetpPf{\Gamma_0}{[\Gamma_0]t}   {By \Lemmaref{lem:substitution-tp}}
            \end{llproof}

            Suppose, for a contradiction, that $\FEV{[\Gamma_0]t} \neq \emptyset$.

            Since $\alpha \in \FV{A}$, we also have $\FEV{[\Gamma]A} \neq \emptyset$, a contradiction.

            \begin{llproof}
              \neqPf {\FEV{[\Gamma_0]t}} {\emptyset}  {Assumption (for contradiction)}
              \eqPf {[\Gamma_0]t} {[\Gamma]\alpha}   {By def.\ of subst.}
              \neqPf {\FEV{[\Gamma]\alpha}} {\emptyset}   {By above equality}
              \inPf{\alpha}{\FV{A}}   {Above}
              \neqPf {\FEV{[\Gamma]A}} {\emptyset}   {By a property of subst.}
              \judgetpPf{\Gamma}{A\;\p}   {Given}
              \eqPf {\FEV{[\Gamma]A}} {\emptyset}   {By inversion}
              \contraPf{\FEV{[\Gamma_0]t} = \emptyset}
              \judgetpPf{\Gamma_0}{t\;\p}   {By \PrincipalWF}
              \eqPf{[\Gamma_0]t}{[\Delta_0]t}  {By i.h.}
              \proofsep
              \judgetpPf{\Gamma_0}{[\Delta_0]t}   {By above equality}
              \eqPf{\FEV{[\Delta_0]t}}{\emptyset}   {By above equality}
              \judgetpPf{\Gamma_0}{\big[[\Delta_0]t/\alpha\big]A\;\p} {By \Lemmaref{lem:substitution-wf} (i)}
              \eqPf{[\Gamma_0]\big[[\Delta_0]t/\alpha\big]A}
                       {[\Delta_0]\big[[\Delta_0]t/\alpha\big]A}   {By i.h.\ (at $\big[[\Delta_0]t/\alpha\big]A$)}
              \proofsep
              \eqPf   {[\Gamma_0, \hypeq{\alpha}{t}]A}
                      {\big[\Gamma_0\big]\big[[\Gamma_0]t / \alpha\big]A}   {By def.\ of subst.}
              \continueeqPf{\big[\Gamma_0\big]\big[[\Delta_0]t / \alpha\big]A}   {By above equality}
              \continueeqPf{\big[\Delta_0\big]\big[[\Delta_0]t / \alpha\big]A}   {By above equality}
              \continueeqPf{\big[\Delta_0\big]\big[[\Delta_0]t' / \alpha\big]A}   {By $[\Delta_0]t = [\Delta_0]t'$}
              \continueeqPf{[\Delta]A}   {By def.\ of subst.}
            \end{llproof}

  \ProofCaseRule{\substextendSolved, \substextendSolve, \substextendAdd, \substextendSolved}
      Similar to the \substextendEqn case.

  \ProofCaseRule{\substextendId, \substextendVV, \substextendUU, \substextendEE, \substextendMonMon} \\
      Straightforward, using the i.h.\ and the definition of substitution.
  \end{itemize}

  Part (ii):  Similar to part (i), using part (ii) of \Lemmaref{lem:substitution-wf}.
\end{proof}

\substitutiontpp*
\begin{proof}\proofheading
  By cases of $p$:

  \begin{itemize}
  \item Case $p = \p$:
    
    \begin{llproof}
      \judgetpPf{\Gamma}{A}{By inversion}
      \eqPf {\FEV{[\Gamma]A}} {\emptyset} {By inversion}
      \judgetpPf{\Gamma}{[\Gamma]A}{By \Lemmaref{lem:substitution-tp}}
      \substextendPf{\Gamma}{\Gamma}{By \Lemmaref{lem:extension-reflexivity}}
      \eqPf{[\Gamma][\Gamma]A}{[\Gamma]A}{By \Lemmaref{lem:substitution-monotonicity}}
      \eqPf{\FEV{[\Gamma][\Gamma]A}}{\emptyset}{By inversion}
      \judgetpPf{\Gamma}{[\Gamma]A\;\p}{By rule \PrincipalWF}
    \end{llproof}

  \item Case $p = \OK$:

    \begin{llproof}
      \judgetpPf{\Gamma}{A}{By inversion}
      \judgetpPf{\Gamma}{[\Gamma]A}{By \Lemmaref{lem:substitution-tp}}
      \judgetpPf{\Gamma}{A\;\OK}{By rule \NonPrincipalWF}
    \end{llproof}
  \qedhere
  \end{itemize}
\end{proof}

\extensionweakeningtpp*
\begin{proof}\proofheading
  By cases of $p$:

  \begin{itemize}
  \item Case $p = \OK$:

    \begin{llproof}
      \judgetpPf{\Gamma}{A}{By inversion}
      \judgetpPf{\Delta}{A}{By \Lemmaref{lem:extension-weakening-tp}}
      \judgetpPf{\Delta}{A\;\OK}{By rule \NonPrincipalWF}
    \end{llproof}
    
  \item Case $p = \p$:

    \begin{llproof}
      \judgetpPf{\Gamma}{A}{By inversion}
      \eqPf{\FEV{[\Gamma]A}}{\emptyset}{By inversion}
      \judgetpPf{\Delta}{A}{By \Lemmaref{lem:extension-weakening-tp}}
      \judgetpPf{\Delta}{[\Delta]A}{By \Lemmaref{lem:substitution-tp}}
      \eqPf{[\Delta]A}{[\Gamma]A}{By \Lemmaref{lem:substitution-invariance}}
      \eqPf{\FEV{[\Delta]A}}{\emptyset}{By congruence of equality}
      \judgetpPf{\Delta}{[\Delta]A\;\p}{By rule \PrincipalWF}
    \end{llproof}
  \qedhere
  \end{itemize}
\end{proof}

\principalinversion*
\begin{proof}\proofheading
  Proof of part 1:

    We have $\judgetp{\Gamma}{A \arr B\;p}$.

        \begin{itemize}
        \item Case $p = \OK$:

              \begin{llproof}
              1~~  \judgetpPf{\Gamma}{A \arr B}{By inversion}
                \judgetpPf{\Gamma}{A}{By inversion on 1}
                \judgetpPf{\Gamma}{B}{By inversion on 1}
                \judgetpPf{\Gamma}{A\;\OK}{By rule \NonPrincipalWF}
                \judgetpPf{\Gamma}{B\;\OK}{By rule \NonPrincipalWF}      
              \end{llproof}

        \item Case $p = \p$:

              \begin{llproof}
              1~~  \judgetpPf{\Gamma}{A \arr B}{By inversion on $\judgetp{\Gamma}{A \arr B\;\p}$}
                \eqPf{\emptyset}{\FEV{[\Gamma](A \arr B)}}{\ditto}
                \continueeqPf{\FEV{[\Gamma]A \arr [\Gamma]B}}{By definition of substitution}
                \continueeqPf{\FEV{[\Gamma]A} \cup \FEV{[\Gamma]B}}{By definition of $\FEV{-}$}
                \eqPf{\FEV{[\Gamma]A}}{\FEV{[\Gamma]B} = \emptyset}{By properties of empty sets and unions}
                \judgetpPf{\Gamma}{A}{By inversion on 1}
                \judgetpPf{\Gamma}{B}{By inversion on 1}
                \judgetpPf{\Gamma}{A\;\p}{By rule \PrincipalWF}
                \judgetpPf{\Gamma}{B\;\p}{By rule \PrincipalWF}      
              \end{llproof}
        \end{itemize}

    \smallskip

    Part 2: We have $\judgetp{\Gamma}{P \implies A\;p}$.      Similar to Part 1.

    \smallskip

    Part 3: We have $\judgetp{\Gamma}{A \with P\;p}$.      Similar to Part 2.
\end{proof}

\subsection{Instantiation Extends}

\instantiationextension*
\begin{proof}\proofheading  By induction on the given derivation.
  
  \begin{itemize}
        \DerivationProofCase{\InstSolve}
                { \judge{\Gamma_L}{\tau}{\sort} }
                { \instjudg{\underbrace{\Gamma_L, \ahat:\sort, \Gamma_R}_{\Gamma}}
                            {\ahat}
                            {\tau}
                            {\sort}
                            {\Gamma_L, \hypeq{\ahat:\sort}{\tau}, \Gamma_R}
                 }

              Follows by \Lemmaref{lem:deep-existential} (ii).
        
        \DerivationProofCase{\InstReach}
                {\bhat \in \unsolved{\Gamma_0[\ahat:\sort][\bhat:\sort]} }
                {\instjudg{\underbrace{\Gamma_0[\ahat:\sort][\bhat:\sort]}_{\Gamma}}
                            {\ahat}
                            {\bhat}
                            {\sort}
                            {\Gamma_0[\ahat:\sort][\hypeq{\bhat:\sort}{\ahat}]}}

              Follows by \Lemmaref{lem:deep-existential} (ii).
        
        \DerivationProofCase{\InstBin}
                {\instjudg{\Gamma_0[\ahat_2:\type, \ahat_1:\type, \hypeq{\ahat:\type}{\ahat_1 \binc \ahat_2}]}
                            {\ahat_1}
                            {\tau_1}
                            {\type}
                            {\Theta} \\
                 \instjudg{\Theta}
                            {\ahat_2}
                            {[\Theta]\tau_2}
                            {\type}
                            {\Delta}}
                {\instjudg{\Gamma_0[\ahat:\type]}
                            {\ahat}
                            {\tau_1 \binc \tau_2}
                            {\type}
                            {\Delta}}

            \begin{llproof}
              \instjudgPf{\Gamma_0[\ahat_2:\type, \ahat_1:\type, \hypeq{\ahat:\type}{\ahat_1 \binc \ahat_2}]}
                            {\ahat_1}
                            {\tau_1}
                            {\type}
                            {\Theta}
                      {Subderivation}
              \substextendPf{\Gamma_0[\ahat_2:\type, \ahat_1:\type, \hypeq{\ahat:\type}{\ahat_1 \binc \ahat_2}]}
                            {\Theta}    {By i.h.}
              \instjudgPf{\Theta}
                            {\ahat_2}
                            {[\Theta]\tau_2}
                            {\type}
                            {\Delta}
                      {Subderivation}
              \substextendPf{\Theta}{\Delta}   {By i.h.}
              \substextendPf{\Gamma_0[\ahat_2:\type, \ahat_1:\type, \hypeq{\ahat:\type}{\ahat_1 \binc \ahat_2}]}
                            {\Delta}   {\hspace{-6ex}By \Lemmaref{lem:extension-transitivity}}
              \decolumnizePf
              \substextendPf{\Gamma_0[\ahat:\type]} {\Gamma_0[\ahat_2:\type, \ahat_1:\type, \hypeq{\ahat:\type}{\ahat_1 \binc \ahat_2}]}
                        {By \Lemmaref{lem:deep-existential} }
              \Pf{}{}{}{~~(parts (i), (i), and (ii),}
              \Pf{}{}{}{~~using \Lemmaref{lem:extension-transitivity})}
              \decolumnizePf
              \substextendPf{\Gamma_0[\ahat:\type]}
                            {\Delta}   {By \Lemmaref{lem:extension-transitivity}}
            \end{llproof}
        
        \DerivationProofCase{\InstZero}
              { }
              {\instjudg{\Gamma_0[\ahat:\ind]}{\ahat}{\zero}{\ind}{\Gamma_0[\hypeq{\ahat:\ind}{\zero}]}}

              Follows by \Lemmaref{lem:deep-existential} (ii).
        
        \DerivationProofCase{\InstSucc}
                {\instjudg{\Gamma[\ahat_1:\ind, \hypeq{\ahat:\ind}{\succ{\ahat_1}}]}
                            {\ahat_1}
                            {t_1}
                            {\ind}
                            {\Delta} 
                }
                {\instjudg{\Gamma[\ahat:\ind]}
                            {\ahat}
                            {\succ{t_1}}
                            {\ind}
                            {\Delta}}

              By reasoning similar to the \InstBin case.
  \qedhere
  \end{itemize}
\end{proof}

\subsection{Equivalence Extends}

\elimeqextension*
\begin{proof}\proofheading  By induction on the given derivation.
    Note that the statement restricts the output to be a (consistent) context $\Delta$.

    \begin{itemize}
      \DerivationProofCase{\ElimeqUvarRefl}
            { }
            {\elimeq{\Gamma}{\alpha}{\alpha}{\sort}{\Gamma}}

            Since $\Delta = \Gamma$, applying \Lemmaref{lem:extension-reflexivity} suffices (let $\Theta = \cdot$).

      \DerivationProofCase{\ElimeqZero}
            { }
            {\elimeq{\Gamma}{\zero}{\zero}{\ind}{\Gamma}}

            Similar to the \ElimeqUvarRefl case.
      
      \DerivationProofCase{\ElimeqSucc}
            {\elimeq{\Gamma}{\sigma}{t}{\ind}{\Delta}}
            {\elimeq{\Gamma}{\succ{\sigma}}{\succ{t}}{\ind}{\Delta}}

            Follows by i.h.

     \DerivationProofCase{\ElimeqInstL}
           {\instjudg{\Gamma_0[\ahat:\sort]}{\ahat}{t}{\sort}{\Delta}}
           {\elimeq{\underbrace{\Gamma_0[\ahat:\sort]}_{\Gamma}}{\ahat}{t}{\sort}{\Delta}}

           \begin{llproof}
             \instjudgPf{\Gamma}{\ahat}{t}{\sort}{\Delta}  {Subderivation}
             \substextendPf{\Gamma} {\Delta}   {By \Lemmaref{lem:instantiation-extension}}
             \LetPf{\Theta}{\cdot}  {}
           \Hand  \substextendPf{\Gamma, \Theta} {\Delta}   {By $\Theta = \cdot$}
           \end{llproof}
     
     \DerivationProofCase{\ElimeqUvarL}
          { \alpha \notin \FV{[\Gamma]t}
            \\
            (\hypeq{\alpha}{-}) \notin \Gamma
          }
          {\elimeq{\Gamma}{\alpha}{t}{\sort}{\Gamma, \hypeq{\alpha}{t}}}

          Let $\Theta$ be $(\alpha = t)$.
           
          \begin{llproof}
          \Hand   \substextendPf{\Gamma, \underbrace{\hypeq{\alpha}{t}}_{\Theta}}{\Gamma, \hypeq{\alpha}{t}}   {By \Lemmaref{lem:extension-reflexivity}}
          \end{llproof}

     \ProofCasesRules{\ElimeqInstR, \ElimeqUvarR}
     
           Similar to the respective \rulename{L} cases.

    \DerivationProofCase{\ElimeqClash}
          { \clash{\sigma}{t} }
          {\elimeq{\Gamma}{\sigma}{t}{\sort}{\bot}}

          The statement says that the output is a (consistent) context $\Delta$, so this case is impossible.
  \qedhere
  \end{itemize}
\end{proof}

\elimpropextension*
\begin{proof}\proofheading  By induction on the given derivation.
    Note that the statement restricts the output to be a (consistent) context $\Delta$.
  
  \begin{itemize}
    \DerivationProofCase{\ElimpropEq}
          {\elimeq{\Gamma}{\sigma}{t}{\ind}{\Delta} }
          {\elimprop{\Gamma}{\sigma = t}{\Delta}}

          \begin{llproof}
            \elimeqPf{\Gamma}{\sigma}{t}{\ind}{\Delta}  {Subderivation}
          \Hand  \substextendPf{\Gamma, \Theta}{\Delta}  {By \Lemmaref{lem:elimeq-extension}}
          \end{llproof}
  \qedhere
  \end{itemize}
\end{proof}

\checkeqextension*
\begin{proof}\proofheading  By induction on the given derivation.

  \begin{itemize}
    \DerivationProofCase{\CheckeqVar}
         { }
         {\checkeq{\Gamma}{u}{u}{\sort}{\Gamma}}
         
         Since $\Delta = \Gamma$, applying \Lemmaref{lem:extension-reflexivity} suffices.
    
    \ProofCasesRules{\CheckeqUnit, \CheckeqZero}  Similar to the \CheckeqVar case.

    \DerivationProofCase{\CheckeqBin}
          {\checkeq{\Gamma}{\tau_1}{\tau'_1}{\type}{\Theta} \\
           \checkeq{\Theta}{[\Theta]\tau_2}{[\Theta]\tau'_2}{\type}{\Delta}
          }
          {\checkeq{\Gamma}{\tau_1 \binc \tau_2}{\tau'_1 \binc \tau'_2}{\type}{\Delta}}

          \begin{llproof}
            \substextendPf{\Gamma}{\Theta}{By i.h.}
            \substextendPf{\Theta}{\Delta}{By i.h.}
            \Hand \substextendPf{\Gamma}{\Delta}{By \Lemmaref{lem:extension-transitivity}}
          \end{llproof}
          
    \DerivationProofCase{\CheckeqSucc}
          {\checkeq{\Gamma}{\sigma}{t}{\ind}{\Delta}}
          {\checkeq{\Gamma}{\succ{\sigma}}{\succ{t}}{\ind}{\Delta}}

          \begin{llproof}
            \checkeqPf{\Gamma}{\sigma}{t}{\ind}{\Delta}  {Subderivation}
          \Hand \substextendPf{\Gamma}{\Delta}  {By i.h.}
          \end{llproof}

    \DerivationProofCase{\CheckeqInstL}
         {\instjudg{\Gamma_0[\ahat]}{\ahat}{t}{\sort}{\Delta}
          \\
          \ahat \notin \FV{[\Gamma_0[\ahat]]t}
         }
         {\checkeq{\Gamma_0[\ahat]}{\ahat}{t}{\sort}{\Delta}}

         \begin{llproof}
           \instjudgPf{\Gamma_0[\ahat]}{\ahat}{t}{\sort}{\Delta}      {Subderivation}
         \Hand  \substextendPf{\underbrace{\Gamma_0[\ahat]}_{\Gamma}}{\Delta}      {By \Lemmaref{lem:instantiation-extension}}
         \end{llproof}

    \ProofCaseRule{\CheckeqInstR}   Similar to the \CheckeqInstL case.
  \qedhere
  \end{itemize}
\end{proof}

\checkpropextension*
\begin{proof}\proofheading  By induction on the given derivation.

  \begin{itemize}
     \DerivationProofCase{\CheckpropEq}
          { \checkeq{\Gamma}{\sigma}{t}{\ind}{\Delta} }
          {\checkprop{\Gamma}{\sigma = t}{\Delta}}

          \begin{llproof}
            \checkeqPf{\Gamma}{\sigma}{t}{\ind}{\Delta}    {Subderivation}
          \Hand  \substextendPf{\Gamma}{\Delta}   {By \Lemmaref{lem:checkeq-extension}}
          \end{llproof}
  \qedhere
  \end{itemize} 
\end{proof}

\propequivextension*
\begin{proof}\proofheading  By induction on the given derivation.

    \begin{itemize}
        \DerivationProofCase{\PropequivEq}
             {\checkeq{\Gamma}{\sigma_1}{\tau_1}{\ind}{\Theta}
              \\
              \checkeq{\Theta}{\sigma_2}{\tau_2}{\ind}{\Delta}}
             {\propequivjudg{\Gamma}{(\sigma_1 = \sigma_2)}{(\tau_1 = \tau_2)}{\Delta}}

             \begin{llproof}
               \checkeqPf{\Gamma}{\sigma_1}{\tau_1}{\ind}{\Theta}   {Subderivation}
               \substextendPf{\Gamma}{\Theta}   {By \Lemmaref{lem:checkeq-extension}}
               \checkeqPf{\Theta}{\sigma_2}{\tau_2}{\ind}{\Delta}   {Subderivation}
               \substextendPf{\Theta}{\Delta}   {By \Lemmaref{lem:checkeq-extension}}
             \Hand  \substextendPf{\Gamma}{\Delta}   {By \Lemmaref{lem:extension-transitivity}}
             \end{llproof}
    \qedhere
    \end{itemize}
\end{proof}

\equivextension*
\begin{proof}\proofheading  By induction on the given derivation.

     \begin{itemize}
        \DerivationProofCase{\EquivVar}
             { }
             {\equivjudg{\Gamma}{\alpha}{\alpha}{\Gamma}}

             Here $\Delta = \Gamma$, so \Lemmaref{lem:extension-reflexivity} suffices.
        
        \DerivationProofCase{\EquivExvar}
             { }
             {\equivjudg{\Gamma}{\ahat}{\ahat}{\Gamma}}

             Similar to the \EquivVar case.
        
        \DerivationProofCase{\EquivUnit}
             { }
             {\equivjudg{\Gamma}{\unitty}{\unitty}{\Gamma}}

             Similar to the \EquivVar case.
        
        \DerivationProofCase{\EquivBin}
             {\equivjudg{\Gamma}{A_1}{B_1}{\Theta}
               \\
               \equivjudg{\Theta}{[\Theta]A_2}{[\Theta]B_2}{\Delta}
             }
             {\equivjudg{\Gamma}{(A_1 \binc A_2)}{(B_1 \binc B_2)}{\Delta}}

             \begin{llproof}
               \equivjudgPf{\Gamma}{A_1}{B_1}{\Theta}   {Subderivation}
               \substextendPf{\Gamma}{\Theta}   {By i.h.}
               \equivjudgPf{\Theta}{[\Theta]A_2}{[\Theta]B_2}{\Delta}   {Subderivation}
               \substextendPf{\Theta}{\Delta}   {By i.h.}
             \Hand  \substextendPf{\Gamma}{\Delta}   {By \Lemmaref{lem:extension-transitivity}}
             \end{llproof}

        \ProofCaseRule{\EquivVec}
             Similar to the \EquivBin case.

        \ProofCasesRules{\EquivImplies, \EquivWith}
             Similar to the \EquivBin case, but with \Lemmaref{lem:propequiv-extension} on the first premise.

        \DerivationProofCase{\EquivAll}
             {\equivjudg{\Gamma, \alpha:\sort}{A_0}{B}{\Delta, \alpha:\sort, \Delta'}}
             {\equivjudg{\Gamma}{\alltype{\alpha:\sort} A_0}{\alltype{\alpha:\sort} B}{\Delta}}

             \begin{llproof}
               \equivjudgPf{\Gamma, \alpha:\sort}{A_0}{B}{\Delta, \alpha:\sort, \Delta'}   {Subderivation}
               \substextendPf{\Gamma, \alpha:\sort}{\Delta, \alpha:\sort, \Delta'}   {By i.h.}
             \Hand  \substextendPf{\Gamma}{\Delta}   {By \Lemmaref{lem:extension-inversion} (i)}
             \end{llproof}             
        
        \DerivationProofCase{\EquivInstL}
            {\instjudg{\Gamma_0[\ahat]}{\ahat}{\tau}{\type}{\Delta}
             \\
             \ahat \notin \FV{[\Gamma_0[\ahat]]\tau}
            }
            {\equivjudg{\Gamma_0[\ahat]}{\ahat}{\tau}{\Delta}}

            \begin{llproof}
              \instjudgPf{\Gamma_0[\ahat]}{\ahat}{\tau}{\type}{\Delta}   {Subderivation}
            \Hand   \substextendPf{\underbrace{\Gamma_0[\ahat]}_{\Gamma}}{\Delta}   {By \Lemmaref{lem:instantiation-extension}}
            \end{llproof}
        
        \ProofCaseRule{\EquivInstR}
            Similar to the \EquivInstL case.        
     \qedhere
     \end{itemize}
\end{proof}

\subsection{Subtyping Extends}

\subtypingextension*
\begin{proof}\proofheading  By induction on the given derivation.
  
  \begin{itemize}
      \DerivationProofCase{\SubAllL}
                {\subjudg[-]{\Gamma, \MonnierComma{\ahat}, \ahat:\sort}
                         {[\ahat/\alpha]A}
                         {B}
                         {\Delta, \MonnierComma{\ahat}, \Theta}}
                {\subjudg[-]{\Gamma}{\alltype{\alpha:\sort}{A}}{B}{\Delta}}

                \begin{llproof}
                    \subjudgPf{-}{\Gamma, \MonnierComma{\ahat}, \ahat:\sort}
                            {[\ahat/\alpha]A}
                            {B}
                            {\Delta, \MonnierComma{\ahat}, \Theta}
                        {Subderivation}
                    \substextendPf{\Gamma, \MonnierComma{\ahat}, \ahat:\sort}{\Delta, \MonnierComma{\ahat}, \Theta}  {By i.h.\ (i)}
                \Hand    \substextendPf{\Gamma}{\Delta}  {By \Lemmaref{lem:extension-inversion} (ii)}
                \end{llproof}
  
      \ProofCaseRule{\SubExistsR}  Similar to the \SubAllL case.

      \DerivationProofCase{\SubAllR}
                {\subjudg[-]{\Gamma, \alpha:\sort}{A}{B}{\Delta, \alpha:\sort, \Theta}}
                {\subjudg[-]{\Gamma}{A}{\alltype{\alpha:\sort}{B}}{\Delta}}

                Similar to the \SubAllL case, but using part (i) of \Lemmaref{lem:extension-inversion}.

      \ProofCaseRule{\SubExistsL}  Similar to the \SubAllR case.

      \DerivationProofCase{\SubEquiv}
                {
                  \equivjudg{\Gamma}{A}{B}{\Delta}
                }
                {\subjudg[\polvar]{\Gamma}{A}{B}{\Delta}}

                \begin{llproof}
                  \equivjudgPf{\Gamma}{A}{B}{\Delta}  {Subderivation}
                \Hand \substextendPf{\Gamma}{\Delta}   {By \Lemmaref{lem:equiv-extension}}
                \end{llproof}\qedhere

  \qedhere
  \end{itemize}  
\end{proof}

\subsection{Typing Extends}

\typingextension*
\begin{proof}\proofheading
  By induction on the given derivation.

  \begin{itemize}
  \item \textbf{Match judgments:}
    
    In rule \MatchEmpty, $\Delta = \Gamma$, so the result follows by \Lemmaref{lem:extension-reflexivity}.

    Rules \MatchBase, \MatchPair, \MatchSum{k} and \MatchWild each have a single premise in which
    the contexts match the conclusion (input $\Gamma$ and output $\Delta$), so the result follows by i.h.
    For rule \MatchSeq, \Lemmaref{lem:extension-transitivity} is also needed.
    
    In rule \MatchExists, apply the i.h., then use \Lemmaref{lem:extension-inversion} (i).

    \MatchWith: Use the i.h.
    
    \MatchNeg: Use the i.h.\ and \Lemmaref{lem:extension-inversion} (v).

    \MatchBot: Immediate by \Lemmaref{lem:extension-reflexivity}.

    \MatchUnify:

        \begin{llproof}
          \substextendPf{\Gamma, \MonnierComma{P}, \Theta'}{\Theta}   {    By \Lemmaref{lem:elimeq-extension}}
          \substextendPf{\Theta}{\Delta, \MonnierComma{P}, \Delta'}   {    By i.h.}
          \substextendPf{\Gamma, \MonnierComma{P}, \Theta'}{\Delta, \MonnierComma{P}, \Delta'}   {    By \Lemmaref{lem:extension-transitivity}}
        \Hand \substextendPf{\Gamma}{\Delta}  {By \Lemmaref{lem:extension-inversion} (ii)}
        \end{llproof}

  \item \textbf{Synthesis, checking, and spine judgments:}
    In rules \Var, \UnitIntro, \EmptySpine, and \ImpliesIntroBot,
    the output context $\Delta$ is exactly $\Gamma$,
    so the result follows by \Lemmaref{lem:extension-reflexivity}.

    \begin{itemize}
        \ProofCaseRule{\AllIntro}  Use the i.h.\ and \Lemmaref{lem:extension-transitivity}.

        \ProofCaseRule{\AllSpine, \ExIntro}
           By \substextendAdd, $\substextend{\Gamma}{\Gamma, \ahat : \sort}$. \\
           The result follows by i.h.\ and \Lemmaref{lem:extension-transitivity}.

        \ProofCasesRules{\WithIntro, \ImpliesSpine}
        Use \Lemmaref{lem:checkprop-extension}, the i.h., and \Lemmaref{lem:extension-transitivity}.

        \ProofCasesRules{\Nil, \Cons}
            Using reasoning found in the \WithIntro and \ImpliesIntro cases.

        \ProofCaseRule{\ImpliesIntro}

            \begin{llproof}
              \substextendPf{\Gamma, \MonnierComma{P}, \Theta'}{\Theta}          {By \Lemmaref{lem:elimprop-extension}}
              \substextendPf{\Theta}{\Delta, \MonnierComma{P}, \Delta}     {By i.h.}
              \substextendPf{\Gamma, \MonnierComma{P}, \Theta'}{\Delta, \MonnierComma{P}, \Delta}     {By \Lemmaref{lem:extension-transitivity}}
            \Hand  \substextendPf{\Gamma}{\Delta}     {By \Lemmaref{lem:extension-inversion}}
            \end{llproof}

        \ProofCasesRules{\ArrIntro, \Rec}
             Use the i.h.\ and \Lemmaref{lem:extension-inversion}.

        \ProofCasesRules{\Sub, \Anno,
                         \ArrElim,  %
                         \ArrSpine,
                         \SumIntro{k}, %
                         \PairIntro %
                       }

                   Use the i.h., and \Lemmaref{lem:extension-transitivity} as needed.

      \ProofCaseRule{\UnitIntroSolve}
              By \Lemmaref{lem:deep-existential} (ii).

      \ProofCaseRule{\SolveSpine, \SumIntroSolve{k}, \PairIntroSolve}

            Use \Lemmaref{lem:deep-existential} (i) twice, \Lemmaref{lem:deep-existential} (ii),
            the i.h., and \Lemmaref{lem:extension-transitivity}.

        \ProofCaseRule{\ArrIntroSolve}
             Use \Lemmaref{lem:deep-existential} (i) twice, \Lemmaref{lem:deep-existential} (ii),
             the i.h.\ and \Lemmaref{lem:extension-inversion} (v).
 
      \ProofCaseRule{\Case}  Use the i.h.\ on the synthesis premise and the match premise, and then \Lemmaref{lem:extension-transitivity}.
    \qedhere
    \end{itemize}
  \end{itemize}
\end{proof}

\subsection{Unfiled}

\contextpartitioning*
\begin{proof}\proofheading
  By induction on the given derivation.

  \begin{itemize}
      \ProofCaseRule{\substextendId}   Impossible: $\Delta, \MonnierComma{\ahat}, \Theta$ cannot have the form $\cdot$.

      \ProofCaseRule{\substextendVV}  We have $\Omega_Z = (\Omega_Z', x : A)$ and $\Theta = (\Theta', x : A')$.
        By i.h., there is $\Psi'$ such that 
        $[\Omega, \MonnierComma{\ahat}, \Omega_Z'](\Delta, \MonnierComma{\ahat}, \Theta')
       =
       [\Omega]\Delta, \Psi'$.  Then by the definition of context application,
       $[\Omega, \MonnierComma{\ahat}, \Omega_Z', x : A](\Delta, \MonnierComma{\ahat}, \Theta', x : A')
      = [\Omega]\Delta, \Psi', x : [\Omega']A$.  Let $\Psi = (\Psi', x : [\Omega']A)$.

      \ProofCaseRule{\substextendUU}  Similar to the \substextendVV case, with $\Psi = (\Psi', \alpha : \sort)$.

      \ProofCasesRules{\substextendEqn,
                       \substextendEE,                        \substextendSolved,
                       \substextendSolve, \substextendAdd, \substextendAddSolved,
                       \substextendMonMon}

        Broadly similar to the \substextendUU case, but the rightmost context element disappears in context application,
        so we let $\Psi = \Psi'$.
  \qedhere
  \end{itemize}
\end{proof}

\completesstability*
\begin{proof}\proofheading
  By induction on the derivation of $\substextend{\Gamma}{\Omega}$.
  
    \begin{itemize}
    \DerivationProofCase{\substextendId}
        { }
        { \substextend{\emptyctx}{\emptyctx} }

        Immediate. 
    
    \DerivationProofCase{\substextendVV}
          {\substextend{\Gamma_0}{\Omega_0}
           \\
           [\Omega_0]A = [\Omega_0]A'
          }
          {\substextend{\Gamma_0, x:A}{\Omega_0, x:A'}}

          \begin{llproof}
            \substextendPf{\Gamma_0}{\Omega_0}{Subderivation}
            \eqPf{[\Omega_0]\Gamma_0}{[\Omega_0]\Omega_0}{By i.h.}
            \eqPf{[\Omega_0]A}{[\Omega_0]A'}{Subderivation} 
            \eqPf{[\Omega_0]\Gamma_0, x:[\Omega_0]A}{[\Omega_0]\Omega_0, x:[\Omega_0]A'}{By congruence of equality}
            \eqPf{[\Omega_0, x:A'](\Gamma_0, x:A)}{[\Omega_0, x:A'](\Omega_0, x:A')}{By definition of substitution}
            
          \end{llproof}
    
    \DerivationProofCase{\substextendUU}
        { \substextend{\Gamma_0}{\Omega_0} }
        { \substextend{\Gamma_0, \alpha:\sort}{\Omega_0, \alpha:\sort} }
    
        Similar to \substextendVV.

    \DerivationProofCase{\substextendEE}
        { \substextend{\Gamma_0}{\Omega_0} }
        { \substextend{\Gamma_0, \ahat:\sort}{\Omega_0, \ahat:\sort} }

        Similar to \substextendVV.
    
    \DerivationProofCase{\substextendSolved}
        { \substextend{\Gamma_0}{\Omega_0} \\ [\Omega_0]t = [\Omega_0]t'}
        { \substextend{\Gamma_0, \hypeq{\ahat:\sort}{t}}{\Omega_0, \hypeq{\ahat:\sort}{t'}}
        }
    
        Similar to \substextendVV.
    
    \DerivationProofCase{\substextendMonMon}
        { \substextend{\Gamma_0}{\Omega_0} }
        { \substextend{\Gamma_0, \MonnierComma{\ahat}}{\Omega_0, \MonnierComma{\ahat}} }

        Similar to \substextendVV.

    \DerivationProofCase{\substextendSolve}
        { \substextend{\Gamma_0}{\Omega_0} }
        { \substextend{\Gamma_0, \bhat:\sort'}{\Omega_0, \hypeq{\bhat:\sort'}{t}} }
   
        Similar to \substextendVV.

    \DerivationProofCase{\substextendEqn}
        { \substextend{\Gamma_0}{\Omega_0} \\ [\Omega_0]t' = [\Omega_0]t}
        { \substextend{\Gamma_0, \hypeq{\alpha}{t'}}{\Omega_0, \hypeq{\alpha}{t}}
        }

        \begin{llproof}
          \substextendPf{\Gamma_0}{\Omega_0}{Subderivation}
          \eqPf{[\Omega_0]t'}{[\Omega_0]t}{Subderivation}
          \eqPf{[\Omega_0]\Gamma_0}{[\Omega_0]\Omega_0}{By i.h.}
          \eqPf{[[\Omega_0]t/\alpha]([\Omega_0]\Gamma_0)}{[[\Omega_0]t/\alpha]([\Omega_0]\Omega_0)}{By congruence of equality}
          \eqPf{[\Omega_0, \hypeq{\alpha}{t}](\Gamma_0, \hypeq{\alpha}{t'})}{[\Omega_0, \hypeq{\alpha}{t}](\Omega_0, \hypeq{\alpha}{t})}{By definition of context substitution}
        \end{llproof}

    \DerivationProofCase{\substextendAdd}
        { \substextend{\Gamma}{\Omega_0} }
        { \substextend{\Gamma}{\Omega_0, \ahat:\sort} }

        \begin{llproof}
          \substextendPf{\Gamma}{\Omega_0}{Subderivation}
          \eqPf{[\Omega_0]\Gamma}{[\Omega_0]\Omega_0}{By i.h.}
          \eqPf{[\Omega_0, \ahat:\sort]\Gamma}{[\Omega_0, \ahat:\sort](\Omega_0, \ahat:\sort)}{By definition of context substitution}
        \end{llproof}
    
    \DerivationProofCase{\substextendAddSolved}
        { \substextend{\Gamma}{\Omega_0} }
        { \substextend{\Gamma}{\Omega_0, \hypeq{\ahat:\sort}{t}} }

        Similar to the \substextendAdd case. 
  \qedhere
  \end{itemize}
\end{proof}

\completingcompleteness*
\begin{proof}\proofheading
  ~
  \begin{itemize}
  \item \textbf{Part (i):}

    By \Lemmaref{lem:substitution-monotonicity} (i), $[\Omega']t = [\Omega'][\Omega]t$.
    Now we need to show $[\Omega'][\Omega]t = [\Omega]t$.
    Considered as a substitution, $\Omega'$ is the identity everywhere except existential
    variables $\ahat$ and universal variables $\alpha$.  First, since $\Omega$ is complete,
    $[\Omega]t$ has no free existentials.  Second, universal variables free in $[\Omega]t$
    have no equations in $\Omega$ (if they had, their occurrences would have been replaced).
    But if $\Omega$ has no equation for $\alpha$, it follows from $\substextend{\Omega}{\Omega'}$
    and the definition of context extension in \Figureref{mainfig:substextend}
    that $\Omega'$ also lacks an equation, so applying $\Omega'$ also leaves $\alpha$ alone.
    
    Transitivity of equality gives $[\Omega']t = [\Omega]t$.

  \item \textbf{Part (ii):}
    Similar to part (i), using \Lemmaref{lem:substitution-monotonicity} (iii) instead of (i).
    
  \item \textbf{Part (iii):}
      By induction on the given derivation of $\substextend{\Omega}{\Omega'}$.

      Only cases \substextendId, \substextendVV, \substextendUU, \substextendEqn,
      \substextendSolved,
      \substextendAddSolved and \substextendMonMon are possible.
      In all of these cases, we use the i.h.\ and the definition of context application;
      in cases \substextendVV, \substextendEqn and \substextendSolved, we also use the equality in the
      premise of the respective rule.
  \qedhere
  \end{itemize}
\end{proof}

\confluenceofcompleteness*
\begin{proof}\proofheading ~\\
  \begin{llproof}
    \substextendPf{\Delta_1}{\Omega}    {Given}
    \eqPf{[\Omega]\Delta_1}{[\Omega]\Omega}  {By \Lemmaref{lem:completes-stability}}
    \substextendPf{\Delta_2}{\Omega}    {Given}
    \eqPf{[\Omega]\Delta_2}{[\Omega]\Omega}  {By \Lemmaref{lem:completes-stability}}
    \eqPf{[\Omega]\Delta_1}{[\Omega]\Delta_2}  {By transitivity of equality}
  \end{llproof}
\end{proof}

\multipleconfluence*
\begin{proof}\proofheading ~\\
  \begin{llproof}
    \substextendPf{\Delta}{\Omega}    {Given}
    \eqPf{[\Omega]\Delta}{[\Omega]\Omega}  {By \Lemmaref{lem:completes-stability}}
    \substextendPf{\Omega}{\Omega'}    {Given}
    \eqPf{[\Omega]\Omega}{[\Omega']\Omega'}    {By \Lemmaref{lem:completing-completeness} (iii)}
    \continueeqPf{[\Omega']\Delta'}  {By \Lemmaref{lem:completes-stability} ($\substextend{\Delta'}{\Omega'}$ given)}
  \end{llproof}
\end{proof}

\canonicalcompletion*
\begin{proof}\proofheading
  By induction on $\Omega$.  In $\OmegaC$, make all solutions (for evars and uvars) canonical by
  applying $\Omega$ to them, dropping declarations of existential variables that aren't in $\dom{\Gamma}$.
\end{proof}

\splitsolutions*
\begin{proof}\proofheading
  Use \Lemmaref{lem:canonical-completion} to get $\OmegaC$ such that
  $\substextend{\Delta}{\OmegaC}$ and $\substextend{\OmegaC}{\Omega}$,
  where for all solutions $t$ in $\OmegaC$ we have $\FEV{t} = \emptyset$.

  We have $\OmegaC = \Omega_1'[\hypeq{\ahat:\sort}{t_1}]$, where $\FEV{t_1} = \emptyset$.
  Therefore \Hand$\substextend{\Omega_1'[\hypeq{\ahat:\sort}{t_1}]}{\Omega}$.

  Now choose $t_2$ as follows:

  \begin{itemize}
  \item   If $\sort = \type$, let $t_2 = t_1 \arr t_1$.
  \item   If $\sort = \ind$, let $t_2 = \succ{t_1}$.
  \end{itemize}

  Thus, \Hand$t_2 \neq t_1$.
  Let $\Omega_2 = \Omega_1'[\hypeq{\ahat:\sort}{t_2}]$.  
  
  \begin{llproof}
  \Hand  \substextendPf{\Delta}{\Omega_2}    {By \Lemmaref{lem:split-extension}}
  \end{llproof}
\qedhere
\end{proof}

\section{Internal Properties of the Declarative System}

\matchwithexists*
\begin{proof}\proofheading
  By induction on the given match derivation.

  In the \DeclMatchBase case, for part (1), apply rule \WithIntro.
  For part (2), apply rule \ExIntro.

  In the \DeclMatchNeg case, part (1), use \Lemmaref{lem:declarative-weakening} (iii).
  In part (2), use \Lemmaref{lem:declarative-weakening} (i).
\end{proof}

\caseinvertibility*
\begin{proof}\proofheading
  By induction on the given derivation.

  \begin{itemize}
     \DerivationProofCase{\DeclSub}
          {\declsynjudg{q}{\Psi}{\case{e_0}{\Pi}}{A}
            \\
            \declsubjudg
                [\joinpolarity{\polarity B}{\polarity A}]
                {\Psi}{A}{B}
          }
          {\declchkjudg{p}{\Psi}{\case{e_0}{\Pi}}{B}}

          Impossible, because $\declsynjudg{q}{\Psi}{\case{e_0}{\Pi}}{A}$ is not derivable.
     
     \ProofCasesRules{\DeclAllIntro, \DeclImpliesIntro}
           Impossible: these rules have a value restriction, but a case expression is not a value.
     
     \DerivationProofCase{\DeclWithIntro}
           {\declcheckprop{\Psi}{P}
             \\
             \declchkjudg{p}{\Psi}{\case{e_0}{\Pi}}{C_0}
           }
           {\declchkjudg{p}{\Psi}{\case{e_0}{\Pi}}{C_0 \with P}}

           \begin{llproof}
           \Hand  $< n - 1~~$   \declsynjudgPf{\p}{\Psi}{e_0}{A}   {By i.h.}
           $< n - 1~~$       \declmatchjudgPf{p}{\Psi}{\Pi}{A}{C_0}   {\ditto}
           \Hand $< n - 1~~$   \declcoversPf{\Psi}{\Pi}{A}   {\ditto}
           $\leq n - 1~~$       \declcheckpropPf{\Psi}{P}   {Subderivation}
           \Hand $< n~~$ \declmatchjudgPf{p}{\Psi}{\Pi}{A}{C_0 \with P}   {By \Lemmaref{lem:match-with-exists} (1)}
           \end{llproof}

      \DerivationProofCase{\DeclExIntro}
            {\judge{\Psi}{\tau}{\sort} \\
             \declchkjudg{\ok}{\Psi}{\case{e_0}{\Pi}}{[\tau/\alpha]C_0}}
            {\declchkjudg{p}{\Psi}{\case{e_0}{\Pi}}{\extype{\alpha:\sort} C_0}}

            \begin{llproof}
            \Hand   \declsynjudgPf{\p}{\Psi}{e_0}{A}   {By i.h.}
                   \declmatchjudgPf{p}{\Psi}{\Pi}{A}{C_0}   {\ditto}
            \Hand   \declcoversPf{\Psi}{\Pi}{A}   {\ditto}
                   \judgePf{\Psi}{\tau}{\sort}   {Subderivation}
            \Hand   \declmatchjudgPf{p}{\Psi}{\Pi}{A}{\extype{\alpha:\sort} C_0}
                            {By \Lemmaref{lem:match-with-exists} (2)}
            \end{llproof}

            The heights of the derivations are similar to those in the \DeclWithIntro case.

            \smallskip
     
     \ProofCasesRules{\DeclUnitIntro, \DeclArrIntro, \DeclRec, \DeclSumIntro{k}, \DeclPairIntro, \DeclNil, \DeclCons}

        Impossible, because in these rules $e$ cannot have the form $\case{e_0}{\Pi}$.
      
     \DerivationProofCase{\DeclCase}
           {  \declsynjudg{\p}{\Psi}{\case{e_0}{\Pi}}{A}
              \\
              \declmatchjudg{p}{\Psi}{\Pi}{A}{C}
              \\
              \declcovers{\Psi}{\Pi}{A}
           }
           {\declchkjudg{p}{\Psi}{\case{e_0}{\Pi}}{C}}

           Immediate.
  \qedhere
  \end{itemize}
\end{proof}

\section{Miscellaneous Properties of the Algorithmic System}

\wfoutputs*
\begin{proof}\proofheading
  By induction on the given derivation.

  \begin{itemize}
      \ProofCaseRule{\Anno}
          Use \Lemmaref{lem:typing-extension} and \Lemmaref{lem:extension-weakening-tpp}.

      \ProofCaseRule{\AllSpine}
          We have $\judgetp{\Gamma}{(\alltype{\alpha:\sort}A_0)\;q}$. \\
          By inversion, $\judgetp{\Gamma, \alpha:\sort}{A_0\;q}$. \\
          By properties of substitution, $\judgetp{\Gamma, \ahat:\sort}{[\ahat/\alpha]A_0\;q}$. \\
          Now apply the i.h.

      \ProofCaseRule{\ImpliesSpine}
          Use \Lemmaref{lem:principal-inversion} (2),
          \Lemmaref{lem:checkprop-extension}, 
          and \Lemmaref{lem:extension-weakening-tpp}.

      \ProofCaseRule{\Recover}

         By i.h., $\judgetp{\Delta}{C\;\OK}$. \\
         We have as premise $\FEV{C} = \emptyset$. \\
         Therefore $\judgetp{\Delta}{C\;\p}$.
         
      \ProofCaseRule{\Pass}  By i.h.

      \ProofCaseRule{\EmptySpine}  Immediate.

      \ProofCaseRule{\ArrSpine}
          Use \Lemmaref{lem:principal-inversion} (1),
          \Lemmaref{lem:typing-extension}, 
          and \Lemmaref{lem:extension-weakening-tpp}.

      \ProofCaseRule{\SolveSpine}
          Show that $\ahat_1 \arr \ahat_2$ is well-formed, then use the i.h.
  \qedhere
  \end{itemize}
\end{proof}

\section{Decidability of Instantiation}

\leftunsolvednesspreservation*
\begin{proof}\proofheading
  By induction on the given derivation.
  
    \begin{itemize}
        \DerivationProofCase{\InstSolve}
                { \judge{\Gamma_0}{\tau}{\sort} }
                { \instjudg{\underbrace{\Gamma_0, \ahat:\sort, \Gamma_1}_{\Gamma}}
                            {\ahat}
                            {\tau}
                            {\sort}
                            {\underbrace{\Gamma_0, \hypeq{\ahat:\sort}{\tau}, \Gamma_1}_{\Delta}}
                 }
        
            Immediate, since to the left of $\ahat$, the contexts $\Delta$ and $\Gamma$ are the same.

        \DerivationProofCase{\InstReach}
                {\bhat \in \unsolved{\Gamma'[\ahat:\sort][\bhat:\sort]} }
                {\instjudg{\underbrace{\Gamma'[\ahat:\sort][\bhat:\sort]}_{\Gamma}}
                            {\ahat}
                            {\bhat}
                            {\sort}
                            {\underbrace{\Gamma'[\ahat:\sort][\hypeq{\bhat:\sort}{\ahat}]}_{\Delta}}}

            Immediate, since to the left of $\ahat$, the contexts $\Delta$ and $\Gamma$ are the same.
        
        \DerivationProofCase{\InstBin}
                {\instjudg{\Gamma_0, \ahat_2:\type, \ahat_1:\type, \hypeq{\ahat:\type}{\ahat_1 \binc \ahat_2}, \Gamma_1}
                            {\ahat_1}
                            {\tau_1}
                            {\type}
                            {\Theta} \\
                 \instjudg{\Theta}
                            {\ahat_2}
                            {[\Theta]\tau_2}
                            {\type}
                            {\Delta}}
                {\instjudg{\Gamma_0, \ahat:\type, \Gamma_1}
                            {\ahat}
                            {\tau_1 \binc \tau_2}
                            {\type}
                            {\Delta}}

            We have $\bhat \in \unsolved{\Gamma_0}$.  Therefore $\bhat \in \unsolved{\Gamma_0, \ahat_2:\type}$. \\
            Clearly, $\ahat_2 \in \unsolved{\Gamma_0, \ahat_2:\type}$. \\
            We have two subderivations:
            \begin{align*}
                \instjudg{\Gamma_0, \ahat_2:\type, \ahat_1:\type, \hypeq{\ahat:\type}{\ahat_1 \binc \ahat_2}, \Gamma_1}
                         {\ahat_1}
                         {A_1}
                         {\type}
                         {\Theta}
                && \text{(1)}
                \\
                \instjudg{\Theta}{\ahat_2}{[\Theta]A_2}{\type}{\Delta}
                && \text{(2)}
            \end{align*}
            By induction on (1), $\bhat \in \unsolved{\Theta}$. \\
            Also by induction on (1), with $\ahat_2$ playing the role of $\bhat$, we get $\ahat_2 \in \unsolved{\Theta}$. \\
            Since $\bhat \in \Gamma_0$, it is declared to the left of $\ahat_2$
              in $\Gamma_0, \ahat_2:\type, \ahat_1:\type, \hypeq{\ahat}{\ahat_1 \binc \ahat_2}, \Gamma_1$. \\
            Hence by \Lemmaref{lem:declaration-order-preservation}, $\bhat$ is declared to the left of $\ahat_2$ in $\Theta$.
            That is, $\Theta = (\Theta_0, \ahat_2:\type, \Theta_1)$, where $\bhat \in \unsolved{\Theta_0}$. \\
            By induction on (2), $\bhat \in \unsolved{\Delta}$. 
        
        \DerivationProofCase{\InstZero}
              { }
              {\instjudg{\underbrace{\Gamma'[\ahat:\ind]}_{\Gamma}}{\ahat}{\zero}{\ind}{\underbrace{\Gamma'[\hypeq{\ahat:\ind}{\zero}]}_{\Delta}}}

            Immediate, since to the left of $\ahat$, the contexts $\Delta$ and $\Gamma$ are the same.
        
        \DerivationProofCase{\InstSucc}
                {\instjudg{\Gamma[\ahat_1:\ind, \hypeq{\ahat:\ind}{\succ{\ahat_1}}]}
                            {\ahat_1}
                            {t_1}
                            {\ind}
                            {\Delta} 
                }
                {\instjudg{\Gamma[\ahat:\ind]}
                            {\ahat}
                            {\succ{t_1}}
                            {\ind}
                            {\Delta}}

                We have $\bhat \in \unsolved{\Gamma_0}$.  Therefore $\bhat \in \unsolved{\Gamma_0, \ahat_1:\ind}$.
                By i.h., $\bhat \in \unsolved{\Delta}$.
   \qedhere
   \end{itemize}    
\end{proof}

\leftfreevariablepreservation*
\begin{proof}\proofheading
  By induction on the given instantiation derivation.

    \begin{itemize}
        \DerivationProofCase{\InstSolve}
                { \judge{\Gamma_0}{\tau}{\sort} }
                { \instjudg{\Gamma_0, \ahat:\sort, \Gamma_1}
                            {\ahat}
                            {\tau}
                            {\sort}
                            {\underbrace{\Gamma_0, \hypeq{\ahat:\sort}{\tau}, \Gamma_1}_{\Delta}}
                 }

            We have $\ahat \notin \FV{[\Gamma]\sigma}$.  Since $\Delta$ differs from $\Gamma$
            only in $\ahat$, it must be the case that $[\Gamma]\sigma = [\Delta]\sigma$.
            It is given that $\bhat \notin \FV{[\Gamma]\sigma}$, so $\bhat \notin \FV{[\Delta]\sigma}$.

        \DerivationProofCase{\InstReach}
                {\chat \in \unsolved{\Gamma[\ahat:\sort][\chat:\sort]} }
                {\instjudg{\Gamma[\ahat:\sort][\chat:\sort]}
                            {\ahat}
                            {\chat}
                            {\sort}
                            {\underbrace{\Gamma[\ahat:\sort][\hypeq{\chat:\sort}{\ahat}]}_{\Delta}}}
            
            Since $\Delta$ differs from $\Gamma$ only
            in solving $\chat$ to $\ahat$, applying $\Delta$ to a type will not introduce a $\bhat$.
            We have $\bhat \notin \FV{[\Gamma]\sigma}$, so $\bhat \notin \FV{[\Delta]\sigma}$.
        
        \DerivationProofCase{\InstBin}
                {\instjudg{\overbrace{\Gamma[\ahat_2:\type, \ahat_1:\type, \hypeq{\ahat:\type}{\ahat_1 \binc \ahat_2}]}^{\Gamma'}}
                            {\ahat_1}
                            {\tau_1}
                            {\type}
                            {\Theta} \\
                 \instjudg{\Theta}
                            {\ahat_2}
                            {[\Theta]\tau_2}
                            {\type}
                            {\Delta}}
                {\instjudg{\Gamma[\ahat:\type]}
                            {\ahat}
                            {\tau_1 \binc \tau_2}
                            {\type}
                            {\Delta}}

            We have $\judgetp{\Gamma}{\sigma}$
            and $\ahat \notin \FV{[\Gamma]\sigma}$
            and $\bhat \notin \FV{[\Gamma]\sigma}$. \\
            By weakening, we get $\judge{\Gamma'}{\sigma}{\sort'}$; since $\ahat \notin \FV{[\Gamma]\sigma}$ and $\Gamma'$ only adds
            a solution for $\ahat$, it follows that $[\Gamma']\sigma = [\Gamma]\sigma$. \\
            Therefore $\ahat_1 \notin \FV{[\Gamma']\sigma}$ and $\ahat_2 \notin \FV{[\Gamma']\sigma}$ and $\bhat \notin \FV{[\Gamma']\sigma}$. \\
            Since we have $\bhat \in \Gamma_0$, we also have $\bhat \in (\Gamma_0, \ahat_2:\type)$. \\
            By induction on the first premise, $\bhat \notin \FV{[\Theta]\sigma}$. \\
            Also by induction on the first premise, with $\ahat_2$ playing the role of $\bhat$, we have $\ahat_2 \notin \FV{[\Theta]\sigma}$. \\
            Note that $\ahat_2 \in \unsolved{\Gamma_0, \ahat_2:\type}$. \\
            By \Lemmaref{lem:left-unsolvedness-preservation}, $\ahat_2 \in \unsolved{\Theta}$. \\
            Therefore $\Theta$ has the form $(\Theta_0, \ahat_2:\type, \Theta_1)$. \\
            Since $\bhat \neq \ahat_2$, we know that $\bhat$ is declared to the left of $\ahat_2$ in $(\Gamma_0, \ahat_2:\type)$,
            so by \Lemmaref{lem:declaration-order-preservation}, $\bhat$ is declared to the left of $\ahat_2$ in $\Theta$.
            Hence $\bhat \in \Theta_0$. \\
            Furthermore, by \Lemmaref{lem:instantiation-extension}, we have $\substextend{\Gamma'}{\Theta}$. \\
            Then by \Lemmaref{lem:extension-weakening}, we have $\judge{\Delta}{\sigma}{\sort'}$. \\
            Using induction on the second premise, $\bhat \notin \FV{[\Delta]\sigma}$. 

        \DerivationProofCase{\InstZero}
              { }
              {\instjudg{\underbrace{\Gamma'[\ahat:\ind]}_{\Gamma}}{\ahat}{\zero}{\ind}{\underbrace{\Gamma'[\hypeq{\ahat:\ind}{\zero}]}_{\Delta}}}

            We have $\ahat \notin \FV{[\Gamma]\sigma}$.  Since $\Delta$ differs from $\Gamma$
            only in $\ahat$, it must be the case that $[\Gamma]\sigma = [\Delta]\sigma$.
            It is given that $\bhat \notin \FV{[\Gamma]\sigma}$, so $\bhat \notin \FV{[\Delta]\sigma}$.
        
        \DerivationProofCase{\InstSucc}
                {\instjudg{\overbrace{\Gamma'[\ahat_1:\ind, \hypeq{\ahat:\ind}{\succ{\ahat_1}}]}^{\Theta}}
                            {\ahat_1}
                            {t_1}
                            {\ind}
                            {\Delta} 
                }
                {\instjudg{\underbrace{\Gamma'[\ahat:\ind]}_{\Gamma}}
                            {\ahat}
                            {\succ{t_1}}
                            {\ind}
                            {\Delta}}
           
           \begin{llproof}
             \judgePf{\Gamma}{\sigma}{\sort'}   {Given}
             \judgePf{\Theta}{\sigma}{\sort'}   {By weakening}
             \proofsep
             \notinPf{\ahat}{\FV{[\Gamma]\sigma}}   {Given}
             \notinPf{\ahat}{\FV{[\Theta]\sigma}}   {$\ahat \notin \FV{[\Gamma]\sigma}$ and $\Theta$ only solves $\ahat$}
             \proofsep
             \eqPf{\Theta}{(\Gamma_0, \ahat_1:\ind, \hypeq{\ahat:\ind}{\succ{\ahat_1}, \Gamma_1})}   {Given}
             \notinPf{\bhat}{\unsolved{\Gamma_0}}   {Given}
             \notinPf{\bhat}{\unsolved{\Gamma_0, \ahat_1:\ind}}   {$\ahat_1$ fresh}
             \proofsep
             \notinPf{\bhat}{\FV{[\Gamma]\sigma}}   {Given}
             \notinPf{\bhat}{\FV{[\Theta]\sigma}}   {$\ahat_1$ fresh}
             \proofsep
           \Hand  \notinPf{\bhat}{\FV{[\Delta]\sigma}}   {By i.h.}
           \end{llproof}    \Pfqedhere
    \end{itemize}
\end{proof}

\instantiationsizepreservation*
\begin{proof}\proofheading
  By induction on the given derivation.
  
  \begin{itemize}
      \DerivationProofCase{\InstSolve}
                { \judge{\Gamma_0}{\tau}{\sort} }
                { \instjudg{\underbrace{\Gamma_0, \ahat:\sort, \Gamma_1}_{\Gamma}}
                            {\ahat}
                            {\tau}
                            {\sort}
                            {\underbrace{\Gamma_0, \hypeq{\ahat:\sort}{\tau}, \Gamma_1}_{\Delta}}
                 }

              Since $\Delta$ differs from $\Gamma$ only in solving $\ahat$, and we know $\ahat \notin \FV{[\Gamma]\sigma}$,
              we have $[\Delta]\sigma = [\Gamma]\sigma$; therefore $|[\Delta]\sigma = [\Gamma]\sigma|$.

     \DerivationProofCase{\InstZero}
              { }
              {\instjudg{\underbrace{\Gamma'[\ahat:\ind]}_{\Gamma}}{\ahat}{\zero}{\ind}{\underbrace{\Gamma'[\hypeq{\ahat:\ind}{\zero}]}_{\Delta}}}

              Similar to the \InstSolve case.
        
      \DerivationProofCase{\InstReach}
                {\bhat \in \unsolved{\Gamma'[\ahat:\sort][\bhat:\sort]} }
                {\instjudg{\underbrace{\Gamma'[\ahat:\sort][\bhat:\sort]}_{\Gamma}}
                            {\ahat}
                            {\bhat}
                            {\sort}
                            {\underbrace{\Gamma'[\ahat:\sort][\hypeq{\bhat:\sort}{\ahat}]}_{\Delta}}
               }

              Here, $\Delta$ differs from $\Gamma$ only in solving $\bhat$ to $\ahat$.  However, $\ahat$ has the same
              size as $\bhat$, so even if $\bhat \in \FV{[\Gamma]\sigma}$, we have $|[\Delta]\sigma = [\Gamma]\sigma|$.

      \DerivationProofCase{\InstBin}
                {\instjudg{\overbrace{\Gamma[\ahat_2:\type, \ahat_1:\type, \hypeq{\ahat:\type}{\ahat_1 \binc \ahat_2}]}^{\Gamma'}}
                            {\ahat_1}
                            {\tau_1}
                            {\type}
                            {\Theta} \\
                 \instjudg{\Theta}
                            {\ahat_2}
                            {[\Theta]\tau_2}
                            {\type}
                            {\Delta}}
                {\instjudg{\Gamma[\ahat:\type]}
                            {\ahat}
                            {\tau_1 \binc \tau_2}
                            {\type}
                            {\Delta}}

          We have $\judge{\Gamma}{\sigma}{\sort'}$ and $\ahat \notin \FV{[\Gamma]\sigma}$. \\
          Since $\ahat_1, \ahat_2 \notin \dom{\Gamma}$, we have $\ahat, \ahat_1, \ahat_2 \notin \FV{[\Gamma]\sigma}$. \\
          By \Lemmaref{lem:deep-existential}, $\substextend{\Gamma[\ahat:\type]}{\Gamma'}$. \\
          By \Lemmaref{lem:extension-weakening}, $\judge{\Gamma'}{\sigma}{\sort'}$. \\ 
          Since $\ahat \notin\FV{\sigma}$, it follows that $[\Gamma']\sigma = [\Gamma]\sigma$, and so $|[\Gamma']\sigma| = |[\Gamma]\sigma|$. \\
          By induction on the first premise, $|[\Gamma']\sigma| = |[\Theta]\sigma|$. \\
          By \Lemmaref{lem:declaration-order-preservation},
            since $\ahat_2$ is declared to the left of $\ahat_1$ in $\Gamma'$,
            we have that $\ahat_2$ is declared to the left of $\ahat_1$ in $\Theta$. \\
          By \Lemmaref{lem:left-unsolvedness-preservation}, since $\ahat_2 \in \unsolved{\Gamma'}$, it is unsolved in $\Theta$:
          that is, $\Theta = (\Theta_0, \ahat_2 : \type, \Theta_1)$. \\
          By \Lemmaref{lem:instantiation-extension}, we have $\substextend{\Gamma'}{\Theta}$. \\
          By \Lemmaref{lem:extension-weakening}, $\judge{\Theta}{\sigma}{\sort'}$. \\
          Since $\ahat_2 \notin \FV{[\Gamma']\sigma}$,
            \Lemmaref{lem:left-free-variable-preservation} gives $\ahat_2 \notin \FV{[\Theta]\sigma}$. \\
          By induction on the second premise, $|[\Theta]\sigma| = |[\Delta]\sigma|$,
            and by transitivity of equality, $|[\Gamma]\sigma| = |[\Delta]\sigma|$. 

        \DerivationProofCase{\InstSucc}
                {\instjudg{\overbrace{\Gamma[\ahat_1:\ind, \hypeq{\ahat:\ind}{\succ{\ahat_1}}]}^{\Gamma'}}
                            {\ahat_1}
                            {t_1}
                            {\ind}
                            {\Delta} 
                }
                {\instjudg{\Gamma[\ahat:\ind]}
                            {\ahat}
                            {\succ{t_1}}
                            {\ind}
                            {\Delta}}

       \begin{llproof}
         \judgePf{\Gamma[\ahat:\type]}{\sigma}{\sort'}{Given}
         \notinPf{\ahat}{[\Gamma[\ahat:\type]]\sigma}{Given}
         \substextendPf{\Gamma[\ahat:\type]}{\Gamma'}{By \Lemmaref{lem:deep-existential}}
         \judgePf{\Gamma'}{\sigma}{\sort'}{By \Lemmaref{lem:extension-weakening}}
         \eqPf{[\Gamma']\sigma}{[\Gamma[\ahat:\type]]\sigma}{Since $\ahat \notin \FV{[\Gamma[\ahat:\type]]\sigma}$}
         \eqPf{|[\Gamma']\sigma|}{|[\Gamma[\ahat:\type]]\sigma|}{By congruence of equality}
         \notinPf{\ahat_1}{[\Gamma']\sigma}{Since $[\Gamma']\sigma = [\Gamma[\ahat:\type]]\sigma$, and $\ahat_1 \notin \dom{\Gamma[\ahat:\type]}$}
         \eqPf{|[\Gamma']\sigma|}{|[\Theta]\sigma|}{By i.h.}
         \eqPf{|[\Gamma[\ahat:\type]]\sigma|}{|[\Theta]\sigma|}{By transitivity of equality}
       \end{llproof}
  \qedhere
  \end{itemize}
\end{proof}

\instantiationdecidability*
\begin{proof}\proofheading
  By induction on the derivation of $\judge{\Gamma}{t}{\sort}$.

      \begin{itemize}
       \DerivationProofCase{\VarSort}
            {(u:\sort) \in \Gamma}
            {\judge{\Gamma_L, \ahat:\sort', \Gamma_R}{u}{\sort}}

          If $\sort \neq \sort'$, no rule matches and no derivation exists. \\
          Otherwise:

          \begin{itemize}
          \item If $(u:\sort) \in \Gamma_L$, we can apply rule \InstSolve.
          \item If $u$ is some unsolved existential variable $\bhat$ and $(\bhat:\sort) \in \Gamma_R$, then
            we can apply rule \InstReach.
          \item Otherwise, $u$ is declared in $\Gamma_R$ and is a universal variable;
            no rule matches and no derivation exists.
          \end{itemize}

       \DerivationProofCase{\SolvedVarSort}
            {(\hypeq{\bhat:\sort}{\tau}) \in \Gamma }
            { \judge{\Gamma}{\bhat}{\sort} }

           By inversion, $(\hypeq{\bhat:\sort}{\tau}) \in \Gamma$,
           but $[\Gamma]\bhat = \bhat$ is given, so this case is impossible.

       \ProofCaseRule{\UnitSort} \\
          If $\sort' = \type$, then apply rule \InstSolve.  Otherwise, no rule matches and no derivation exists.

      \DerivationProofCase{\BinSort}
            { \judge{\Gamma}{\tau_1}{\type}
              \\
              \judge{\Gamma}{\tau_2}{\type}
            }
            { \judge{\underbrace{\Gamma_L, \ahat:\sort', \Gamma_R}_\Gamma}{\tau_1 \binc \tau_2}{\type} }

           If $\sort' \neq \type$, then no rule matches and no derivation exists.  Otherwise:\\
           Given, $[\Gamma](\tau_1 \binc \tau_2) = \tau_1 \binc \tau_2$
           and $\ahat \notin \FV{[\Gamma](\tau_1 \binc \tau_2)}$. \\
           If $\judge{\Gamma_L}{\tau_1 \binc \tau_2}{\type}$, then we have a derivation by \InstSolve. \\
           If not, the only other rule whose conclusion matches $\tau_1 \binc \tau_2$ is \InstBin. \\
           First, consider whether $\instjudg{\Gamma_L, \ahat_2:\type, \ahat_1:\type, \hypeq{\ahat:\type}{\ahat_1 \binc \ahat_2}, \Gamma_R}{\ahat_1}{t}{\type}{-}$
           is decidable. \\[1ex]
           By definition of substitution,
             $[\Gamma](\tau_1 \binc \tau_2) = ([\Gamma]\tau_1) \binc ([\Gamma]\tau_2)$.
           Since $[\Gamma](\tau_1 \binc \tau_2) = \tau_1 \binc \tau_2$,
             we have $[\Gamma]\tau_1 = \tau_1$ and $[\Gamma]\tau_2 = \tau_2$. \\
           By weakening,
             $\judge{\Gamma_L, \ahat_2:\type, \ahat_1:\type, \hypeq{\ahat:\type}{\ahat_1 \binc \ahat_2}, \Gamma_R}{\tau_1 \binc \tau_2}{\type}$. \\
           Since $\judge{\Gamma}{\tau_1}{\type}$ and $\judge{\Gamma}{\tau_2}{\type}$, we have
           $\ahat_1, \ahat_2 \notin \FV{\tau_1} \union \FV{\tau_2}$. \\
           Since $\ahat \notin \FV{t} \supseteq \FV{\tau_1}$, it follows that $[\Gamma']\tau_1 = \tau_1$. \\
           By i.h., either there exists $\Theta$ s.t. $\instjudg{\Gamma_L, \ahat_2:\type, \ahat_1:\type, \hypeq{\ahat:\type}{\ahat_1 \binc \ahat_2}, \Gamma_R}{\ahat_1}{\tau_1}{\type}{\Theta}$, or not. \\
           If not, then no derivation by \InstBin exists. \\
           Otherwise, there exists such a $\Theta$.  By \Lemmaref{lem:left-unsolvedness-preservation},
                we have $\ahat_2 \in \unsolved{\Theta}$. \\
           By \Lemmaref{lem:left-free-variable-preservation}, we know that $\ahat_2 \notin \FV{[\Theta]\tau_2}$. \\
           Substitution is idempotent, so $[\Theta][\Theta]\tau_2 = [\Theta]\tau_2$. \\
           By i.h., either there exists $\Delta$ such that $\instjudg{\Theta}{\ahat_2}{[\Theta]\tau_2}{\sort}{\Delta}$, or not. \\
           If not, no derivation by \InstBin exists. \\
           Otherwise, there exists such a $\Delta$.  By rule \InstBin, we have $\instjudg{\Gamma}{\ahat}{t}{\sort}{\Delta}$.

      \DerivationProofCase{\ZeroSort}
            { }
            {\judge{\Gamma}{\zero}{\ind}}

           If $\sort' \neq \ind$, then no rule matches and no derivation exists.  Otherwise, apply rule \InstSolve.
      
      \DerivationProofCase{\SuccSort}
            {\judge{\Gamma}{t_0}{\ind} }
            {\judge{\Gamma}{\succ{t_0}}{\ind}}

           If $\sort' \neq \ind$, then no rule matches and no derivation exists.  Otherwise:\\
           If $\judge{\Gamma_L}{\succ{t_0}}{\ind}$, then we have a derivation by \InstSolve. \\
           If not, the only other rule whose conclusion matches $\succ{t_0}$ is \InstSucc. \\
           The remainder of this case is similar to the \BinSort case, but shorter.
      \qedhere
      \end{itemize}
\end{proof}

\section{Separation}

\separationtransitivity*
\begin{proof}\proofheading ~

  \begin{llproof}
    \sepextendPf{\Gamma_L}{\Gamma_R}{\Theta_L}{\Theta_R}   {Given}
    \substextendPf{(\Gamma_L, \Gamma_R)}{(\Theta_L, \Theta_R)}   {By \Definitionref{def:separation-extension}}
    \subseteqPf{\Gamma_L}{\Theta_L \AND \Gamma_R \subseteq \Theta_R}   {\ditto}
    \proofsep
    \sepextendPf{\Theta_L}{\Theta_R}{\Delta_L}{\Delta_R}   {Given}
    \substextendPf{(\Theta_L, \Theta_R)}{(\Delta_L, \Delta_R)}   {By \Definitionref{def:separation-extension}}
    \subseteqPf{\Theta_L}{\Delta_L \AND \Theta_R \subseteq \Delta_R}   {\ditto}
    \proofsep
    \substextendPf{(\Gamma_L, \Gamma_R)}{(\Delta_L, \Delta_R)}   {By \Lemmaref{lem:extension-transitivity}}
    \subseteqPf{\Gamma_L}{\Delta_L  \AND \Gamma_R \subseteq \Delta_R}   {By transitivity of $\subseteq$}
    \proofsep
  \Hand  \sepextendPf{\Gamma_L}{\Gamma_R}{\Delta_L}{\Delta_R}   {By \Definitionref{def:separation-extension}}
  \end{llproof}
  \qedhere
\end{proof}

\separationtruncation*
\begin{proof}\proofheading
  By induction on $\Delta_R$.

  If $\Delta_R = (\dots, H)$, we have
  $\sepextend{\Gamma_L}{\Gamma_R, H}{\Delta_L}{(\Delta, H)}$,
  and inversion on \substextendUU (if $H$ is $(\alpha:\sort)$, or the corresponding rule for other forms)
  gives the result (with $\Theta = \cdot$).

  Otherwise, proceed into the subderivation of
  $\substextend{(\Gamma_L, \Gamma_R, \alpha:\sort)}{(\Delta_L, \Delta_R)}$,
  with $\Delta_R = (\Delta_R', \Delta')$ where $\Delta'$ is a single declaration.
  Use the i.h.\ on $\Delta_R'$, producing some $\Theta'$.
  Finally, let $\Theta = (\Theta', \Delta')$.
\end{proof}

\separationaux*
\begin{proof}\proofheading
  Part (i): By induction on the derivation of the given checkeq judgment.
     Cases \CheckeqVar, \CheckeqUnit and \CheckeqZero are immediate ($\Delta_L = \Gamma_L$
     and $\Delta_R = \Gamma_R$).
     For case \CheckeqSucc, apply the i.h.
     For cases \CheckeqInstL and \CheckeqInstR, use the i.h.\ (v).
     For case \CheckeqBin, use reasoning similar to that in the \WithIntro case of
     \Lemmaref{lem:separation-main}
     (transitivity
     of separation, and applying $\Theta$ in the second premise).

  Part (ii), checkprop: Use the i.h.\ (i).

  Part (iii), elimeq:
     Cases \ElimeqUvarRefl, \ElimeqUnit and \CheckeqZero are immediate ($\Delta_L = \Gamma_L$
     and $\Delta_R = \Gamma_R$).  Cases \ElimeqUvarLBot, \ElimeqUvarRBot, \ElimeqBinBot and
     \ElimeqClash are impossible (we have $\Delta$, not $\bot$).
     For case \ElimeqSucc, apply the i.h.
     The case for \ElimeqBin is similar to the case \CheckeqBin in part (i).
     For cases \ElimeqUvarL and \ElimeqUvarR, $\Delta = (\Gamma_L, \Gamma_R, \hypeq{\alpha}{\tau})$
     which, since $\exbasis{\Gamma_R}{\tau}$,
     ensures that $\sepextend{\Gamma_L}{(\Gamma_R, \hypeq{\alpha}{\tau})}{\Delta_L}{(\Delta_R, \hypeq{\alpha}{\tau})}$.
     
  Part (iv), elimprop: Use the i.h.\ (iii).

  Part (v), instjudg:

  \begin{itemize}
      \ProofCaseRule{\InstSolve}
         Here, $\Gamma = (\Gamma_0, \ahat:\sort, \Gamma_1)$
         and $\Delta = (\Gamma_0, \hypeq{\ahat:\sort}{\tau}, \Gamma_1)$.
         We have $\ahat \in \dom{\Gamma_R}$, so the declaration $\ahat:\sort$ is in $\Gamma_R$.
         Since $\exbasis{\Gamma_R}{\tau}$, the context $\Delta$ maintains the separation.

      \ProofCaseRule{\InstReach}
         Here, $\Gamma = \Gamma_0[\ahat:\sort][\bhat:\sort]$
         and $\Delta = \Gamma_0[\ahat:\sort][\hypeq{\bhat:\sort}{\ahat}]$.
         We have $\ahat \in \dom{\Gamma_R}$, so the declaration $\ahat:\sort$ is in $\Gamma_R$.
         Since $\bhat$ is declared to the right of $\ahat$, it too must be in $\Gamma_R$, which
         can also be shown from $\exbasis{\Gamma_R}{\bhat}$.
         Both declarations are in $\Gamma_R$, so the context $\Delta$ maintains the separation.
         
      \ProofCaseRule{\InstZero}  In this rule, $\Delta$ is the same as $\Gamma$ except for
          a solution $\zero$, which doesn't violate separation.

      \ProofCaseRule{\InstSucc}
         The result follows by i.h., taking care to keep the declaration $\ahat_1:\ind$ on the right
         when applying the i.h., even if $\ahat : \ind$ is the leftmost declaration in $\Gamma_R$,
         ensuring that $\succ{\ahat_1}$ does not violate separation.

      \ProofCaseRule{\InstBin}
         As in the \InstSucc case, the new declarations should be kept on the right-hand side
         of the separator.  Otherwise the case is straightforward (using the i.h.\ twice and transitivity).
  \end{itemize}

  Part (vi), propequivjudg:  Similar to the \CheckeqBin case of part (i), using
  the i.h.\ (i).

  Part (vii), equivjudg:

  \begin{itemize}
      \ProofCasesRules{\EquivVar, \EquivExvar, \EquivUnit}
          Immediate ($\Delta_L = \Gamma_L$ and $\Delta_R = \Gamma_R$).

      \ProofCaseRule{\EquivBin}
           Similar to the case \CheckeqBin in part (i).

      \ProofCaseRule{\EquivVec}
           Similar to the case \CheckeqBin in part (i).

      \ProofCasesRules{\EquivAll, \EquivExists}
           Similar to the case \CheckeqBin in part (i).

      \ProofCasesRules{\EquivImplies, \EquivWith}
           Similar to the case \CheckeqBin in part (i), using the i.h.\ (vi).

      \ProofCasesRules{\EquivInstL, \EquivInstR}
          Use the i.h.\ (v).
  \qedhere
  \end{itemize}
\end{proof}

\separationsubtyping*
\begin{proof}\proofheading
  By induction on the given derivation.
  In the \SubEquiv case, use \Lemmaref{lem:separation-aux} (vii).
  Otherwise, the reasoning needed follows that used in the proof
  of \Lemmaref{lem:separation-main}.
\end{proof}

\separationmain*
\begin{proof}\proofheading  By induction on the given derivation.

  First, the (Match) judgment part, giving only the cases that motivate the side conditions:

  \begin{itemize}
      \ProofCaseRule{\MatchBase}
          Here we use the i.h.\ (Checking),
          for which we need $\exbasis{\Gamma_R}{C}$.

      \ProofCaseRule{\MatchWith}
          Here we use the i.h.\ (Match Elim.),
          which requires that $\FEV{P} = \emptyset$,
          which motivates $\FEV{\Avec} = \emptyset$.

      \ProofCaseRule{\MatchNeg}
          In its premise, this rule appends a type $A \in \Avec$ to $\Gamma_R$ and claims it is principal
          ($z : A \p$), which motivates $\FEV{\Avec = \emptyset}$.
  \end{itemize}

  Similarly, (Match Elim.):

  \begin{itemize}
      \ProofCaseRule{\MatchUnify}
          Here we use \Lemmaref{lem:separation-aux} (iii),
          for which we need $\FEV{\sigma} \union \FEV{\tau} = \emptyset$,
          which motivates $\FEV{P} = \emptyset$.
  \end{itemize}

  Now, we show the cases for the (Spine), (Checking), and (Synthesis) parts.

  \begin{itemize}
     \ProofCasesRules{\Var, \UnitIntro, \ImpliesIntroBot}
          In all of these rules,
          the output context is the same as the input context, so just let $\Delta_L = \Gamma_L$ and $\Delta_R = \Gamma_R$.
     
     \DerivationProofCase{\EmptySpine}
         {}
         {\spinejudg{\Gamma_L \sep \Gamma_R}{\emptyspine}{A}{p}{\underbrace{A}_{C}}{\underbrace{p}_{q}}{\Gamma_L \sep \Gamma_R}}

         Let $\Delta_L = \Gamma_L$ and $\Delta_R = \Gamma_R$. \\
         We have $\exbasis{\Gamma_R}{A}$.  Since $\Delta_R = \Gamma_R$ and $C = A$,
         it is immediate that $\exbasis{\Delta_R}{C}$.
     
     \DerivationProofCase{\Sub}
          {
            \synjudg{q}{\Gamma_L \sep \Gamma_R}{e}{A}{\Theta}
            \\
            \subjudg[\polvar]{\Theta}{A}{B}{\Delta}
          }
          {\chkjudg{p}{\Gamma_L \sep \Gamma_R}{e}{B}{\Delta}}

          By i.h., $\Theta = (\Theta_L \sep \Theta_R)$ and $\sepextend{\Gamma_L}{\Gamma_R}{\Theta_L}{\Theta_R}$. \\
          By \Lemmaref{lem:separation-subtyping}, $\Delta = (\Delta_L \sep \Delta_R)$ and $\sepextend{\Theta_L}{\Theta_R}{\Delta_L}{\Delta_R}$. \\
          By \Lemmaref{lem:separation-transitivity}, $\sepextend{\Gamma_L}{\Gamma_R}{\Delta_L}{\Delta_R}$.
     
     \DerivationProofCase{\Anno}
          {\judgetp{\Gamma}{A \p}
            \\
            \chkjudg{\p}{\Gamma}{e}{[\Gamma]A}{\Delta}
          }
          {\synjudg{\p}{\Gamma}{(e : A)}{[\Delta]A}{\Delta}}

          By i.h.; since $\FEV{A} = \emptyset$, the condition on the (Checking) part is trivial.

     \DerivationProofCase{\UnitIntroSolve}
            { }
            {\chkjudg{\ok}{\Gamma[\ahat:\type]}{\unitexp}{\ahat}{\Gamma[\hypeq{\ahat:\type}{\unitty}]}}

            Adding a solution with a ground type cannot destroy separation.

     \DerivationProofCase{\AllIntro}
           {\chkintro{v}
            \\
            \chkjudg{p}{\Gamma_L, \Gamma_R, \alpha:\sort}{v}{A_0}{\Delta, \alpha:\sort, \Theta}
           }
           {\chkjudg{p}{\Gamma_L, \Gamma_R}{v}{\alltype{\alpha:\sort} A_0}{\Delta}}

           \begin{llproof}
             \exbasisPf{\Gamma_R}{\alltype{\alpha:\sort} A_0}   {Given}
             \exbasisPf{\Gamma_R, \alpha:\sort}{A_0}   {From definition of $\xfev$}
             \eqPfParen{\Delta, \alpha:\sort, \Theta}{\Delta_L \sep \Delta'_R}   {By i.h.}
             \sepextendPf{\Gamma_L}{(\Gamma_R, \alpha:\sort)}{\Delta_L}{\Delta'_R}  {\ditto}
           \Hand  \sepextendPf{\Gamma_L}{\Gamma_R}{\Delta_L}{\Delta_R}   {By \Lemmaref{lem:separation-truncation}}
             \eqPfParenR{\Delta'_R}{\Delta_R, \alpha:\sort, \Theta}   {\ditto}
             \eqPfParen{\Delta, \alpha:\sort, \Theta}{\Delta_L \sep \Delta'_R}   {Above}
             \continueeqPfParen{\Delta_L, \Delta'_R}   {Definition of $\sep$}
             \continueeqPfParen{\Delta_L, \Delta_R, \alpha:\sort, \Theta}   {By above equation}
           \Hand  \eqPfParenR{\Delta}{\Delta_L, \Delta_R}   {$\alpha$ not multiply declared}
           \end{llproof}

     \DerivationProofCase{\AllSpine}
         {\spinejudg{\Gamma_L, \Gamma_R, \ahat:\sort}{\appspine{e}{s}}{[\ahat/\alpha]A_0}{\ok}{C}{q}{\Delta}}
         {\spinejudg{\Gamma_L, \Gamma_R}{\appspine{e}{s}}{\alltype{\alpha:\sort} A_0}{p}{C}{q}{\Delta}}

           \begin{llproof}
             \exbasisPf{\Gamma_R}{\alltype{\alpha:\sort A_0}}  {Given}
             \exbasisPf{\Gamma_R, \ahat:\sort}{[\ahat/\alpha]A_0}   {From definition of $\xfev$}
           \Hand  \eqPfParenR{\Delta}{\Delta_L \sep \Delta_R}   {By i.h.}
             \sepextendPf{\Gamma_L}{(\Gamma_R, \ahat:\sort)}{\Delta_L}{\Delta_R}  {\ditto}
           \Hand  \exbasisPf{\Delta_R}{C}  {\ditto}
             \subseteqPf{\dom{\Gamma_L}}{\dom{\Delta_L}}   {By \Definitionref{def:separation-extension}}
             \subseteqPf{\dom{\Gamma_R, \ahat:\sort}} {\dom{\Delta_R}}   {By \Definitionref{def:separation-extension}}
             \proofsep
             \subseteqPf{\dom{\Gamma_R} \union \{\ahat\}}{\dom{\Delta_R}}   {By definition of $\dom{-}$}
             \subseteqPf{\dom{\Gamma_R}}{\dom{\Delta_R}}   {Property of $\subseteq$}
             \proofsep
             \substextendPf{(\Gamma_L, \Gamma_R)}{(\Delta_L, \Delta_R)}   {By \Lemmaref{lem:typing-extension}}
             \proofsep
           \Hand  \sepextendPf{\Gamma_L}{\Gamma_R}{\Delta_L}{\Delta_R}   {By \Definitionref{def:separation-extension}}
           \end{llproof}

     \DerivationProofCase{\WithIntro}
          {
            \notcase{e} \\     
            \checkprop{\Gamma_L \sep \Gamma_R}{P}{\Theta}
            \\
            \chkjudg{p}{\Theta}{e}{[\Theta]A_0}{\Delta}
          }
          {\chkjudg{p}{\Gamma_L \sep \Gamma_R}{e}{(A_0 \with P)}{\Delta}}

          \begin{llproof}
            \judgetpPf{\Gamma_L \sep \Gamma_R}{(A_0 \with P)\;p}    {Given}
            \judgepropPf{\Gamma_L \sep \Gamma_R}{P}    {By inversion}
            \judgetpPf{\Gamma_L \sep \Gamma_R}{A_0\;p}    {By inversion}
            \exbasisPf{\Gamma_R}{A_0 \with P}    {Given}
            \exbasisPf{\Gamma_R}{P}   {By def.\ of $\xfev$}
            \exbasisPf{\Gamma_R}{A_0}   {\ditto}
            \eqPfParenR{\Theta}{\Theta_L \sep \Theta_R}   {By \Lemmaref{lem:separation-aux} (i)}
            \sepextendPf{\Gamma_L}{\Gamma_R}{\Theta_L}{\Theta_R}   {\ditto}
            \decolumnizePf
            \exbasisPf{\Gamma_R}{A_0}   {Above}
            \subseteqPf{\dom{\Gamma_R}}{\dom{\Theta_R}}   {By \Definitionref{def:separation-extension}}
            \exbasisPf{\Theta_R}{A_0}   {By previous line}
            \exbasisPf{\Theta_R}{[\Theta]A_0}   {Previous line and $\sepextend{\Gamma_L}{\Gamma_R}{\Theta_L}{\Theta_R}$}
            \judgetpPf{\Gamma_L \sep \Gamma_R}{(A_0 \with P)\;p}    {Given}
            \judgetpPf{\Gamma_L \sep \Gamma_R}{A_0\;p}    {By inversion}
            \judgetpPf{\Theta}{A_0\;p}    {By \Lemmaref{lem:extension-weakening-tpp}}
            \judgetpPf{\Theta}{[\Theta]A_0\;p}    {By \Lemmaref{lem:substitution-tp}}
         \Hand   \eqPfParenR{\Delta}{\Delta_L \sep \Delta_R}   {By i.h.}
            \sepextendPf{\Theta_L}{\Theta_R}{\Delta_L}{\Delta_R}   {\ditto}
          \Hand   \sepextendPf{\Gamma_L}{\Gamma_R}{\Delta_L}{\Delta_R}    {By \Lemmaref{lem:separation-transitivity}}
          \end{llproof}

     \ProofCaseRule{\Nil}
         Similar to a section of the \WithIntro case.

     \ProofCaseRule{\Cons}
         Similar to the \WithIntro case, with an extra use of the i.h.\ for the additional
         second premise.
          
     \DerivationProofCase{\ImpliesIntro}
           {
             \chkintro{v}
             \\
             \elimprop{\Gamma_L \sep (\Gamma_R, \MonnierComma{P})}{P}{\Theta}
             \\
             \chkjudg{\p}{\Theta}{v}{[\Theta]A_0}{\Delta, \MonnierComma{P}, \Delta'}
           }
           {\chkjudg{\p}{\Gamma_L \sep \Gamma_R}{v}{P \implies A_0}{\Delta}}

           \begin{llproof}
             \judgetpPf{\Gamma_L \sep \Gamma_R}{(P \implies A_0)\;\p}   {Given}
             \judgepropPf{\Gamma_L \sep \Gamma_R}{P \implies A_0}   {By inversion}
             \eqPf{\FEV{P \implies A_0}}{\emptyset}   {\ditto}
             \eqPf{\FEV{P}}{\emptyset}   {By def.\ of $\xfev$}
             \proofsep
             \elimpropPf{\Gamma_L \sep (\Gamma_R, \MonnierComma{P})}{P}{\Theta}       {Subderivation}
             \eqPfParenR{\Theta}{\Theta_L \sep (\Theta_R, \Theta_Z)}   {By \Lemmaref{lem:separation-aux} (iv)}
             \sepextendPf{\Gamma_L}{(\Gamma_R, \MonnierComma{P}, \Theta_Z)}{\Theta_L}{(\Theta_R, \Theta_Z)}   {\ditto}
             \proofsep
             \judgetpPf{\Gamma_L \sep \Gamma_R}{(P \implies A_0)\;\p}   {Given}
             \judgetpPf{\Gamma_L, \Gamma_R}{A_0\;\p}   {By \Lemmaref{lem:principal-inversion} (2)}
             \judgetpPf{\Gamma_L, \Gamma_R, \MonnierComma{P}, \Theta_Z}{A_0\;\p}   {By \Lemmaref{lem:suffix-weakening-judgetp}}
             \judgetpPf{\Theta}{[\Theta]A_0\;\p}   {By Lemmas \ref{lem:extension-weakening-tpp} and \ref{lem:substitution-tpp}}
             \eqPf{\FEV{A_0}}{\emptyset}   {Above and def.\ of $\xfev$}
             \exbasisPf{\Theta_R, \Theta_Z}{A_0}   {Immediate}
             \eqPfParen{\Delta, \MonnierComma{P}, \Delta'}{\Delta_L \sep \Delta'_R}   {By i.h.}
             \sepextendPf{\Theta_L}{(\Theta_R, \Theta_Z)}{\Delta_L}{\Delta'_R}   {\ditto}
             \proofsep
             \sepextendPf{\Gamma_L}{(\Gamma_R, \MonnierComma{P})}{\Delta_L}{\Delta'_R}    {By \Lemmaref{lem:separation-transitivity}}
           \Hand  \sepextendPf{\Gamma_L}{\Gamma_R}{\Delta_L}{\Delta_R}    {By \Lemmaref{lem:separation-truncation}}
             \eqPfParenR{\Delta'_R}{\Delta_R, \MonnierComma{P}, \dots}  {\ditto}
           \Hand  \eqPfParenR{\Delta}{\Delta_L, \Delta_R}   {Similar to the \AllIntro case}
           \end{llproof}

     \ProofCaseRule{\ExIntro}
         Similar to the \AllSpine case.

     \DerivationProofCase{\ImpliesSpine}
         {\checkprop{\Gamma_L \sep \Gamma_R}{P}{\Theta}
          \\
          \spinejudg{\Theta}{\appspine{e}{s}}{[\Theta]A_0}{p}{C}{q}{\Delta}
        }
        {
           \spinejudg{\Gamma_L \sep \Gamma_R}{\appspine{e}{s}}{P \implies A_0}{p}{C}{q}{\Delta}
        }

         \begin{llproof}
            \judgetpPf{\Gamma_L \sep \Gamma_R}{(P \implies A_0)\;p}   {Given}
            \judgepropPf{\Gamma_L \sep \Gamma_R}{P}   {By inversion}
            \checkpropPf{\Gamma_L, \Gamma_R}{P}{\Theta}  {Subderivation}
            \eqPfParenR{\Theta}{\Theta_L \sep \Theta_R}   {By \Lemmaref{lem:separation-aux} (i)}
            \sepextendPf{\Gamma_L}{\Gamma_R}{\Theta_L}{\Theta_R}   {\ditto}
            \decolumnizePf
            \spinejudgPf{\Theta}{\appspine{e}{s}}{[\Theta]A_0}{p}{C}{q}{\Delta}  {Subderivation}
            \eqPfParen{\Delta, \MonnierComma{P}, \Delta'}{\Delta_L \sep \Delta'_R}   {By i.h.}
            \sepextendPf{\Theta_L}{\Theta_R}{\Delta_L}{\Delta'_R}   {\ditto}
         \Hand  \exbasisPf{\Delta_R}{C}  {\ditto}
          \Hand   \sepextendPf{\Gamma_L}{\Gamma_R}{\Delta_L}{\Delta_R}    {By \Lemmaref{lem:separation-transitivity}}
         \end{llproof}

     \DerivationProofCase{\Rec}
          {\chkjudg{p}{\Gamma_L, \Gamma_R, \hyp{p}{x}{C}}{v}{C}{\Delta, \hyp{p}{x}{C}, \Theta}
          }
          {\chkjudg{p}{\Gamma_L, \Gamma_R}{\rec{x} v}{C}{\Delta}}

          \begin{llproof}
            \judgetpPf{\Gamma_L \sep \Gamma_R}{C\;p}   {Given}
            \exbasisPf{\Gamma_R}{C}   {Given}
            \judgetpPf{\Gamma_L \sep (\Gamma_R, \hyp{p}{x}{C})}{C\;p}   {By weakening and \Definitionref{def:separation}}
            \chkjudgPf{p}{\Gamma_L, \Gamma_R, \hyp{p}{x}{C}}{v}{C}{\Delta, \hyp{p}{x}{C}, \Theta}   {Subderivation}
            \eqPfParen{\Delta, \hyp{p}{x}{C}, \Theta}{\Delta_L, \Delta_R'}   {By i.h.}
            \sepextendPf{\Gamma_L}{\Gamma_R}{\Delta_L}{\Delta_R'}   {\ditto}
          \Hand  \sepextendPf{\Gamma_L}{\Gamma_R}{\Delta_L}{\Delta_R}   {By \Lemmaref{lem:separation-truncation}}
             \eqPfParenR{\Delta'_R}{\Delta_R, \hyp{p}{x}{C}, \dots}  {\ditto}
           \Hand  \eqPfParenR{\Delta}{\Delta_L, \Delta_R}   {Similar to the \AllIntro case}
          \end{llproof}

     \DerivationProofCase{\ArrIntro}
          {\chkjudg{p}{\Gamma_L, \Gamma_R, \hyp{p}{x}{A}}{e}{B}{\Delta, \hyp{p}{x}{A}, \Theta}
          }
          {\chkjudg{p}{\Gamma_L, \Gamma_R}{\lam{x} e}{A \arr B}{\Delta}}

          \begin{llproof}
            \judgetpPf{\Gamma_L \sep \Gamma_R}{(A \arr B)\;p}   {Given}
            \judgetpPf{\Gamma_L \sep \Gamma_R}{B\;p}   {By inversion}
            \exbasisPf{\Gamma_R}{A \arr B}   {Given}
            \exbasisPf{\Gamma_R}{A}   {By def.\ of $\xfev$}
            \judgetpPf{\Gamma_L \sep (\Gamma_R, \hyp{p}{x}{A})}{B\;p}   {By weakening and \Definitionref{def:separation}}
            \chkjudgPf{p}{\Gamma_L, \Gamma_R, \hyp{p}{x}{A}}{e}{B}{\Delta, \hyp{p}{x}{A}, \Theta}   {Subderivation}
            \eqPfParen{\Delta, \hyp{p}{x}{A}, \Theta}{\Delta_L, \Delta_R'}   {By i.h.}
            \sepextendPf{\Gamma_L}{\Gamma_R}{\Delta_L}{\Delta_R'}   {\ditto}
          \Hand  \sepextendPf{\Gamma_L}{\Gamma_R}{\Delta_L}{\Delta_R}   {By \Lemmaref{lem:separation-truncation}}
             \eqPfParenR{\Delta'_R}{\Delta_R, \hyp{p}{x}{A}, \dots}  {\ditto}
           \Hand  \eqPfParenR{\Delta}{\Delta_L, \Delta_R}   {Similar to the \AllIntro case}
          \end{llproof}

     \DerivationProofCase{\ArrIntroSolve}
           {\chkjudg{\ok}{
                 \Gamma_0[\ahat_1{:}\type, \ahat_2{:}\type, \hypeq{\ahat{:}\type}{\ahat_1{\arr}\ahat_2}], \hyp{\ok}{x}{\ahat_1}
              }
              {e_0}
              {\ahat_2}
              {\Delta, \hyp{\ok}{x}{\ahat_1}, \Delta'}
           }
           {\chkjudg{\ok}{\underbrace{\Gamma_0[\ahat : \type]}_{\Gamma_L \sep \Gamma_R}}{\lam{x} e_0}{\ahat}{\Delta}}

           We have $(\Gamma_L \sep \Gamma_R) = \Gamma_0[\ahat : \type]$.  We also have $\exbasis{\Gamma_R}{\ahat}$.
           Therefore $\ahat \in \dom{\Gamma_R}$ and
           \[
              \Gamma_0[\ahat : \type]
              ~~=~~
              \Gamma_L, \Gamma_2, \ahat : \type, \Gamma_3
           \]
           where $\Gamma_R = (\Gamma_2, \ahat : \type, \Gamma_3)$.

           Then the input context in the premise has the following form:
           \[
               \Gamma_0[\ahat_1{:}\type, \ahat_2{:}\type, \hypeq{\ahat{:}\type}{\ahat_1{\arr}\ahat_2}], \hyp{\ok}{x}{\ahat_1}
               ~~=~~
               \Gamma_L,
               \Gamma_2,
               \ahat_1{:}\type,
               \ahat_2{:}\type,
               \hypeq{\ahat{:}\type}{\ahat_1{\arr}\ahat_2},
               \Gamma_3, 
               \hyp{\ok}{x}{\ahat_1}
           \]
           Let us separate this context at the same point as $\Gamma_0[\ahat : \type]$, that is, after $\Gamma_L$ and before $\Gamma_2$,
and call the resulting right-hand context $\Gamma'_R$.
           That is, 
           \[
               \Gamma_0[\ahat_1{:}\type, \ahat_2{:}\type, \hypeq{\ahat{:}\type}{\ahat_1{\arr}\ahat_2}], \hyp{\ok}{x}{\ahat_1}
               ~~=~~
               \Gamma_L
               ~\sep~
               \big(\underbrace{
                   \Gamma_2,
                   \ahat_1{:}\type,
                   \ahat_2{:}\type,
                   \hypeq{\ahat{:}\type}{\ahat_1{\arr}\ahat_2},
                   \Gamma_3, 
                   \hyp{\ok}{x}{\ahat_1}
               }_{\Gamma'_R}
               \big)
           \]
           \begin{llproof}
             \exbasisPf{\Gamma_R}{\ahat}   {Given}
             \chkjudgPf{\ok}{
                   \Gamma_L
                   \sep
                   \Gamma'_R
                  }
                      {e_0}
                      {\ahat_2}
                      {\Delta, \hyp{\ok}{x}{\ahat_1}, \Delta'}
                      {Subderivation}
             \judgetpPf{\Gamma_L \sep \Gamma'_R}{\ahat_2\;\OK}   {$\ahat_2 \in \dom{\Gamma'_R}$}
             \exbasisPf{\Gamma'_R}{\ahat_2}   {$\ahat_2 \in \dom{\Gamma'_R}$}
             \eqPfParen{\Delta, x : \ahat_1, \Delta'}{\Delta_L, \Delta'_R}   {By i.h.}
             \sepextendPf{\Gamma_L}{\Gamma'_R}{\Delta_L}{\Delta'_R}   {\ditto}
           \Hand  \eqPfParenR{\Delta}{\Delta_L, \Delta_R}   {Similar to the \AllIntro case}
           \Hand  \sepextendPf{\Gamma_L}{\Gamma_R}{\Delta_L}{\Delta_R}   {\ditto}
           \end{llproof}
  
     \DerivationProofCase{\ArrElim}
           {\synjudg{p}{\Gamma}{e}{A}{\Theta}
             \\
             \recspinejudg{\Theta}{s}{[\Theta]A}{p}{C}{q}{\Delta}
           }
           {\synjudg{q}{\Gamma}{e\,s}{C}{\Delta}}

           Use the i.h.\ and \Lemmaref{lem:separation-transitivity}, with
           \Lemmaref{lem:typing-wf} and \Lemmaref{lem:substitution-tp}.
      
     \DerivationProofCase{\Recover}
           {
                 \spinejudg{\Gamma}{s}{A}{\p}{C}{\OK}{\Delta}
                 \\
                 \FEV{[\Delta]C} = \emptyset
           }
           {\recspinejudg{\Gamma}{s}{A}{\p}{C}{\p}{\Delta}}

           Use the i.h.

     \DerivationProofCase{\Pass}
           {
                 \spinejudg{\Gamma}{s}{A}{p}{C}{q}{\Delta}
                 \\
                 \big(
                     (p = \OK)
                     \OR (q = \p)
                     \OR (\FEV{[\Delta]C} \neq \emptyset)
                 \big)
           }
           {\recspinejudg{\Gamma}{s}{A}{p}{C}{q}{\Delta}}

           Use the i.h.

     \DerivationProofCase{\ArrSpine}
         {\chkjudg{p}{\Gamma_L \sep \Gamma_R}{e}{A_1}{\Theta}
          \\
          \spinejudg{\Theta}{s}{[\Theta]A_2}{p}{C}{q}{\Delta}
         }
         {\spinejudg{\Gamma_L \sep \Gamma_R}{\appspine{e}{s}}{A_1 \arr A_2}{p}{C}{q}{\Delta}}

         \begin{llproof}
           \judgetpPf{\Gamma}{(A_1 \arr A_2)\;p}   {Given}
           \judgetpPf{\Gamma}{A_1}   {By inversion}
           \exbasisPf{\Gamma_R}{A_1 \arr A_2}   {Given}
           \exbasisPf{\Gamma_R}{A_1}   {By def.\ of $\xfev$}
           \eqPfParenR{\Theta}{\Theta_L, \Theta_R}   {By i.h.}
           \sepextendPf{\Gamma_L}{\Gamma_R}{\Theta_L}{\Theta_R}  {\ditto}
           \proofsep
           \judgetpPf{\Gamma}{A_2}   {By inversion}
           \judgetpPf{\Gamma}{[\Theta]A_2}   {By \Lemmaref{lem:substitution-tp}}
           \exbasisPf{\Gamma_R}{A_2}  {By def.\ of $\xfev$}
         \Hand  \eqPfParenR{\Delta}{\Delta_L, \Delta_R}   {By i.h.}
           \sepextendPf{\Theta_L}{\Theta_R}{\Delta_L}{\Delta_R}  {\ditto}
         \Hand  \exbasisPf{\Delta_R}{C}  {\ditto}
         \Hand  \sepextendPf{\Gamma_L}{\Gamma_R}{\Delta_L}{\Delta_R}  {By \Lemmaref{lem:separation-transitivity}}
         \end{llproof}

     \DerivationProofCase{\SumIntro{k}}
           {\chkjudg{p}{\Gamma}{e}{A_k}{\Delta}}
           {\chkjudg{p}{\Gamma}{\inj{k} e}{A_1 + A_2}{\Delta}}

           Use the i.h.\ (inverting $\judgetp{\Gamma}{(A_1 + A_2)\;p}$).
     
     \DerivationProofCase{\PairIntro}
           {\chkjudg{p}{\Gamma}{e_1}{A_1}{\Theta}
             \\
             \chkjudg{p}{\Theta}{e_2}{[\Theta]A_2}{\Delta}
           }
           {\chkjudg{p}{\Gamma}{\pair{e_1}{e_2}}{A_1 \times A_2}{\Delta}}

           \begin{llproof}
             \judgetpPf{\Gamma}{(A_1 \times A_2)\;p}   {Given}
             \judgetpPf{\Gamma}{A_1\;p}   {By inversion}
             \chkjudgPf{p}{\Gamma}{e_1}{A_1}{\Theta}    {Subderivation}
             \eqPfParenR{\Theta}{\Theta_L, \Theta_R}   {By i.h.}
             \sepextendPf{\Gamma_L}{\Gamma_R}{\Theta_L}{\Theta_R}   {\ditto}
             \proofsep
             \judgetpPf{\Gamma}{A_2}   {By inversion}
             \substextendPf{\Gamma}{\Theta}  {By \Lemmaref{lem:typing-extension}}
             \judgetpPf{\Theta}{A_2}   {By \Lemmaref{lem:extension-weakening}}
             \judgetpPf{\Theta}{[\Theta]A_2}   {By \Lemmaref{lem:substitution-tp}}
             \chkjudgPf{p}{\Theta}{e_2}{[\Theta]A_2}{\Delta}   {Subderivation}
             \proofsep
             \eqPfParenR{\Delta}{\Delta_L, \Delta_R}   {By i.h.}
             \sepextendPf{\Theta_L}{\Theta_R}{\Delta_L}{\Delta_R}   {\ditto}
           \Hand  \sepextendPf{\Gamma_L}{\Gamma_R}{\Delta_L}{\Delta_R}   {By \Lemmaref{lem:separation-transitivity}}
           \end{llproof}

\clearpage
     \DerivationProofCase{\PairIntroSolve}
           {
               \chkjudg{\ok}{\Gamma[\ahat_2{:}\type, \ahat_1{:}\type, \hypeq{\ahat{:}\type}{\ahat_1{\times}\ahat_2}]}{e_1}{\ahat_1}{\Theta}
               \\
               \chkjudg{\ok}{\Theta}{e_2}{[\Theta]\ahat_2}{\Delta}
           }
           {\chkjudg{\ok}{\Gamma[\ahat : \type]}{\pair{e_1}{e_2}}{\ahat}{\Delta}}

           We have $(\Gamma_L \sep \Gamma_R) = \Gamma_0[\ahat : \type]$.  We also have $\exbasis{\Gamma_R}{\ahat}$.
           Therefore $\ahat \in \dom{\Gamma_R}$ and
           \[
              \Gamma_0[\ahat : \type]
              ~~=~~
              \Gamma_L, \Gamma_2, \ahat : \type, \Gamma_3
           \]
           where $\Gamma_R = (\Gamma_2, \ahat : \type, \Gamma_3)$.

           Then the input context in the premise has the following form:
           \[
               \Gamma_0[\ahat_1{:}\type, \ahat_2{:}\type, \hypeq{\ahat{:}\type}{\ahat_1{\times}\ahat_2}]
               ~~=~~
               (\Gamma_L,
               \Gamma_2,
               \ahat_1{:}\type,
               \ahat_2{:}\type,
               \hypeq{\ahat{:}\type}{\ahat_1{\times}\ahat_2},
               \Gamma_3
               )
           \]
           Let us separate this context at the same point as $\Gamma_0[\ahat : \type]$,
           that is, after $\Gamma_L$ and before $\Gamma_2$,
           and call the resulting right-hand context $\Gamma'_R$:
           \[
               \Gamma_0[\ahat_1{:}\type, \ahat_2{:}\type, \hypeq{\ahat{:}\type}{\ahat_1{\times}\ahat_2}]
               ~~=~~
               \Gamma_L
               ~\sep~
               \big(\underbrace{
                   \Gamma_2,
                   \ahat_1{:}\type,
                   \ahat_2{:}\type,
                   \hypeq{\ahat{:}\type}{\ahat_1{\times}\ahat_2},
                   \Gamma_3 
               }_{\Gamma'_R}
               \big)
           \]
           \begin{llproof}
             \exbasisPf{\Gamma_R}{\ahat}   {Given}
             \chkjudgPf{\ok}{
                   \Gamma_L
                   \sep
                   \Gamma'_R
                  }
                      {e_1}
                      {\ahat_1}
                      {\Theta}
                      {Subderivation}
             \exbasisPf{\Gamma'_R}{\ahat_2}   {$\ahat_2 \in \dom{\Gamma'_R}$}
             \eqPfParenR{\Theta}{\Theta_L, \Theta_R}   {By i.h.}
             \sepextendPf{\Gamma_L}{\Gamma'_R}{\Theta_L}{\Theta_R}   {\ditto}
             \chkjudgPf{\ok}{\Theta}{e_2}{[\Theta]\ahat_2}{\Delta}   {Subderivation}
             \subseteqPf{\dom{\Gamma'_R}}{\dom{\Theta_R}}   {By \Definitionref{def:separation-extension}}
             \exbasisPf{\Theta_R}{\ahat_2} {By above $\subseteq$}
             \exbasisPf{\Theta_R}{[\Theta_R]\ahat_2}   {By \Definitionref{def:separation}}
           \Hand  \eqPfParenR{\Delta}{\Delta_L, \Delta_R}   {By i.h.}
             \sepextendPf{\Theta_L}{\Theta_R}{\Delta_L}{\Delta_R}   {\ditto}
             \proofsep
             \eqPfParenR{\Gamma_R}{\Gamma_2, \ahat:\type, \Gamma_3}   {Above}
             \eqPfParenR{\Gamma'_R}{\Gamma_2, \ahat_1{:}\type, \ahat_2{:}\type, \hypeq{\ahat:\type}{\ahat_1{\times}\ahat_2}, \Gamma_3}   {Above}
          \end{llproof}

          By \Lemmaref{lem:deep-existential} (i), (i), (ii) and the definition of
          separation, we can show

          \begin{llproof}
             \sepextendPf{\Gamma_L}{(\Gamma_2, \ahat:\type, \Gamma_3)}
                         {\Gamma_L}{(\Gamma_2, \ahat_1{:}\type, \ahat_2{:}\type, \hypeq{\ahat{:}\type}{\ahat_1{\times}\ahat_2}, \Gamma_3)}
                         {}
             \decolumnizePf
             \sepextendPf{\Gamma_L}{\Gamma_R}{\Gamma_L}{\Gamma'_R}   {By above equalities}
           \Hand  \sepextendPf{\Gamma_L}{\Gamma_R}{\Delta_L}{\Delta_R}
                              {By \Lemmaref{lem:separation-transitivity} twice}
          \end{llproof}

     \DerivationProofCase{\SumIntroSolve{k}}
            {\chkjudg{\ok}{\Gamma[\ahat_1 : \type, \ahat_2 : \type, \hypeq{\ahat:\type}{\ahat_1{+}\ahat_2}]}{e}{\ahat_k}{\Delta}}
            {\chkjudg{\ok}{\Gamma[\ahat : \type]}{\inj{k} e}{\ahat}{\Delta}}

            Similar to the \PairIntroSolve case, but simpler.

      \DerivationProofCase{\SolveSpine}
            {\spinejudg{\Gamma[\ahat_2 : \type, \ahat_1 : \type, \hypeq{\ahat:\type}{\ahat_1{\arr}\ahat_2}]}{\appspine{e}{s_0}}{(\ahat_1 \arr \ahat_2)}{\ok}{C}{\ok}{\Delta}}
            {\spinejudg{\Gamma[\ahat:\type]}{\appspine{e}{s_0}}{\ahat}{\ok}{C}{\ok}{\Delta}}
            
            Similar to the \PairIntroSolve and \SumIntroSolve{k} cases, except that (because we're in the spine
            part of the lemma) we have to show that $\exbasis{\Delta_R}{C}$.  But we have the same $C$ in the premise
            and conclusion, so we get that by applying the i.h.
            
      \DerivationProofCase{\Case}
           {  \synjudg{\p}{\Gamma}{e}{A}{\Theta}
              \\
              \matchjudg{q}{p}{\Theta}{\Pi}{A}{[\Theta]C}{\Delta}
              \\
              \covers{\Delta}{\Pi}{[\Delta]A}
           }
           {\chkjudg{p}{\Gamma}{\case{e}{\Pi}}{C}{\Delta}}
           
           Use the i.h.\ and \Lemmaref{lem:separation-transitivity}.
  \qedhere
  \end{itemize}
\end{proof}

\section{Decidability of Algorithmic Subtyping}

\subsection{Lemmas for Decidability of Subtyping}

\substitutionnotlarge*
\begin{proof}\proofheading
  By induction on $A$, following the definition of substitution.
\end{proof}

\instantiationsolves*
\begin{proof}\proofheading  By induction on the given derivation.

      \begin{itemize}
        \DerivationProofCase{\InstSolve}
                { \judge{\Gamma_L}{\tau}{\sort} }
                { \instjudg{\Gamma_L, \ahat:\sort, \Gamma_R}
                            {\ahat}
                            {\tau}
                            {\sort}
                            {\Gamma_L, \hypeq{\ahat:\sort}{\tau}, \Gamma_R}
                 }

           It is evident that
            $|\unsolved{\Gamma_L, \ahat:\sort, \Gamma_R}|
           = |\unsolved{\Gamma_L, \hypeq{\ahat:\sort}{\tau}, \Gamma_R}|
           + 1$.
        
        \DerivationProofCase{\InstReach}
                {\bhat \in \unsolved{\Gamma[\ahat:\sort][\bhat:\sort]} }
                {\instjudg{\Gamma[\ahat:\sort][\bhat:\sort]}
                            {\ahat}
                            {\underbrace{\bhat}_{\tau}}
                            {\sort}
                            {\Gamma[\ahat:\sort][\hypeq{\bhat:\sort}{\ahat}]}
                }

            Similar to the previous case.
        
        \DerivationProofCase{\InstBin}
                {\instjudg{\Gamma_0[\ahat_2:\type, \ahat_1:\type, \hypeq{\ahat:\type}{\ahat_1 \binc \ahat_2}]}
                            {\ahat_1}
                            {\tau_1}
                            {\type}
                            {\Theta} \\
                 \instjudg{\Theta}
                            {\ahat_2}
                            {[\Theta]\tau_2}
                            {\type}
                            {\Delta}}
                {\instjudg{\Gamma_0[\ahat:\type]}
                            {\ahat}
                            {\tau_1 \binc \tau_2}
                            {\type}
                            {\Delta}}

        \smallskip

        \begin{llproof}
            \eqPf{|\unsolved{\Gamma_0[\ahat_2:\type, \ahat_1:\type, \hypeq{\ahat}{\ahat_1 \binc \ahat_2}]}|}
                    {|\unsolved{\Gamma_0[\ahat]}| + 1}
                    {Immediate}
            \eqPf{|\unsolved{\Gamma_0[\ahat_2:\type, \ahat_1:\type, \hypeq{\ahat}{\ahat_1 \binc \ahat_2}]}|}
                   {|\unsolved{\Theta}| + 1}  {By i.h.}
            \eqPf{|\unsolved{\Gamma}|}{|\unsolved{\Theta}|}   {Subtracting $1$}
        \Hand    \continueeqPf{|\unsolved{\Delta}| + 1}   {By i.h.}
        \end{llproof}

        \DerivationProofCase{\InstZero}
              { }
              {\instjudg{\Gamma[\ahat:\ind]}{\ahat}{\zero}{\ind}{\Gamma[\hypeq{\ahat:\ind}{\zero}]}}

           Similar to the \InstSolve case.
        
        \DerivationProofCase{\InstSucc}
                {\instjudg{\Gamma_0[\ahat_1:\ind, \hypeq{\ahat:\ind}{\succ{\ahat_1}}]}
                            {\ahat_1}
                            {t_1}
                            {\ind}
                            {\Delta} 
                }
                {\instjudg{\Gamma_0[\ahat:\ind]}
                            {\ahat}
                            {\succ{t_1}}
                            {\ind}
                            {\Delta}}

             \begin{llproof}
                  \eqPf{|\unsolved{\Delta}| + 1}{|\unsolved{\Gamma_0[\ahat_1:\ind, \hypeq{\ahat:\ind}{\succ{\ahat_1}}]}|}   {By i.h.}
             \Hand     \continueeqPf{|\unsolved{\Gamma_0[\ahat:\ind]}|}   {By definition of $\unsolved{-}$}
             \end{llproof}
      \qedhere
      \end{itemize}    
\end{proof}

\checkeqsolving*
\begin{proof}\proofheading  By induction on the given derivation.

  \begin{itemize}
    \DerivationProofCase{\CheckeqVar}
         { }
         {\checkeq{\Gamma}{u}{u}{\sort}{\underbrace{\Gamma}_{\Delta}}}

         Here $\Delta = \Gamma$.
    
    \ProofCasesRules{\CheckeqUnit, \CheckeqZero}
        Similar to the \CheckeqVar case.
    
    \DerivationProofCase{\CheckeqSucc}
          {\checkeq{\Gamma}{\sigma}{t}{\ind}{\Delta}}
          {\checkeq{\Gamma}{\succ{\sigma}}{\succ{t}}{\ind}{\Delta}}

          Follows by i.h.

    \DerivationProofCase{\CheckeqInstL}
         {\instjudg{\Gamma_0[\ahat]}{\ahat}{t}{\sort}{\Delta}
          \\
          \ahat \notin \FV{t}
         }
         {\checkeq{\underbrace{\Gamma_0[\ahat]}_{\Gamma}}{\ahat}{t}{\sort}{\Delta}}

         \begin{llproof}
           \instjudgPf{\Gamma_0[\ahat]}{\ahat}{t}{\sort}{\Delta}    {Subderivation}
           \instjudgPf{\Gamma}{\ahat}{t}{\sort}{\Delta}    {$\Gamma = \Gamma_0[\ahat]$}
           \Pf{}{}{\Delta = \Gamma  \OR  |\unsolved{\Delta}| = |\unsolved{\Gamma}| - 1}   {By \Lemmaref{lem:instantiation-solves}}
         \Hand  \Pf{}{}{\Delta = \Gamma  \OR  |\unsolved{\Delta}| < |\unsolved{\Gamma}|}   {}
         \end{llproof}
        
    \DerivationProofCase{\CheckeqInstR}
         {\instjudg{\Gamma[\ahat:\sort]}{\ahat}{t}{\sort}{\Delta}
          \\
          \ahat \notin \FV{t}
         }
         {\checkeq{\Gamma[\ahat:\sort]}{t}{\ahat}{\sort}{\Delta}}

         Similar to the \CheckeqInstL case.

    \DerivationProofCase{\CheckeqBin}
         {\checkeq{\Gamma}{\sigma_1}{\tau_1}{\type}{\Theta} \\
          \checkeq{\Theta}{[\Theta]\sigma_2}{[\Theta]\tau_2}{\type}{\Delta}
         }
         {\checkeq{\Gamma}{\underbrace{\sigma_1 \binc \sigma_2}_{\sigma}}{\underbrace{\tau_1 \binc \tau_2}_{t}}{\type}{\Delta}}
         
         \begin{llproof}
           \checkeqPf{\Gamma}{\sigma_1}{\tau_1}{\type}{\Theta}   {Subderivation}
           \Pf{}{}{\Theta = \Gamma  \OR   |\unsolved{\Theta}| < |\unsolved{\Gamma}|}    {By i.h.}
         \end{llproof}

         \begin{itemize}
           \item $\Theta = \Gamma$:

              \begin{llproof}
                \checkeqPf{\Theta}{[\Theta]\sigma_2}{[\Theta]\tau_2}{\type}{\Delta}   {Subderivation}
                \checkeqPf{\Gamma}{[\Gamma]\sigma_2}{[\Gamma]\tau_2}{\type}{\Delta}   {By $\Theta = \Gamma$}
              \Hand   \Pf{}{}{\Delta = \Gamma  \OR   |\unsolved{\Gamma}| = |\unsolved{\Delta}| + 1}    {By i.h.}              
              \end{llproof}

           \item $|\unsolved{\Theta}| < |\unsolved{\Gamma}|$:
           
             \begin{llproof}
               \checkeqPf{\Theta}{[\Theta]\sigma_2}{[\Theta]\tau_2}{\type}{\Delta}   {Subderivation}
               \Pf{}{}{\Delta = \Theta  \OR  |\unsolved{\Delta}| < |\unsolved{\Theta}|}   {By i.h.}
             \end{llproof}

             If $\Delta = \Theta$ then substituting $\Delta$ for $\Theta$ in
             $|\unsolved{\Theta}| < |\unsolved{\Gamma}|$
             gives
             $|\unsolved{\Delta}| < |\unsolved{\Gamma}|$.
             
             If $|\unsolved{\Delta}| < |\unsolved{\Theta}|$ then transitivity of $<$ gives
             $|\unsolved{\Delta}| < |\unsolved{\Gamma}$.
         \qedhere
         \end{itemize}
  \end{itemize}
\end{proof}

\propequivsolving*
\begin{proof}\proofheading
     Only one rule can derive the judgment:

     \begin{itemize}
        \DerivationProofCase{\PropequivEq}
             {\checkeq{\Gamma}{\sigma_1}{t_1}{\ind}{\Theta}
              \\
              \checkeq{\Theta}{[\Theta]\sigma2}{[\Theta]t_2}{\ind}{\Delta}}
             {\propequivjudg{\Gamma}{(\sigma_1 = \sigma2)}{(t_1 = t_2)}{\Delta}}

             By \Lemmaref{lem:checkeq-solving} on the first premise, \\
             either $\Theta = \Gamma$ or $|\unsolved{\Theta}| < |\unsolved{\Gamma}|$.

             In the former case, the result follows from \Lemmaref{lem:checkeq-solving}
             on the second premise.

             In the latter case, applying \Lemmaref{lem:checkeq-solving} to the second premise
             either gives $\Delta = \Theta$, and therefore
             \[
                 |\unsolved{\Delta}| < |\unsolved{\Gamma}|
             \]
             or gives $|\unsolved{\Delta}| < |\unsolved{\Theta}|$, which also leads to
             $|\unsolved{\Delta}| < |\unsolved{\Gamma}|$.
      \qedhere
      \end{itemize}
\end{proof}

\equivsolving*
\begin{proof}\proofheading  By induction on the given derivation.
  
       \begin{itemize}
        \DerivationProofCase{\EquivVar}
             { }
             {\equivjudg{\Gamma}{\alpha}{\alpha}{\Gamma}}

             Here $\Delta = \Gamma$.
        
        \ProofCasesRules{\EquivExvar, \EquivUnit}
             Similar to the \EquivVar case.

        \DerivationProofCase{\EquivBin}
             {\equivjudg{\Gamma}{A_1}{B_1}{\Theta}
               \\
               \equivjudg{\Theta}{[\Theta]A_2}{[\Theta]B_2}{\Delta}
             }
             {\equivjudg{\Gamma}{(A_1 \binc A_2)}{(B_1 \binc B_2)}{\Delta}}

             By i.h., either $\Theta = \Gamma$ or $|\unsolved{\Theta}| < |\unsolved{\Gamma}|$.

             In the former case, apply the i.h.\ to the second premise.  Now either $\Delta = \Theta$---and therefore
             $\Delta = \Gamma$---or $|\unsolved{\Delta}| < |\unsolved{\Theta}|$.  Since $\Theta = \Gamma$,
             we have $|\unsolved{\Delta}| < |\unsolved{\Gamma}|$.

             In the latter case, we have $|\unsolved{\Theta}| < |\unsolved{\Gamma}|$.
             By i.h.\ on the second premise, either $\Delta = \Theta$, and substituting $\Delta$ for $\Theta$
             gives $|\unsolved{\Delta}| < |\unsolved{\Gamma}|$---or
             $|\unsolved{\Delta}| < |\unsolved{\Theta}|$, which combined with  $|\unsolved{\Theta}| < |\unsolved{\Gamma}|$
             gives $|\unsolved{\Delta}| < |\unsolved{\Gamma}|$.

        \ProofCaseRule{\EquivVec}
            Similar to the \EquivBin case.
        
        \DerivationProofCase{\EquivAll}
             {\equivjudg{\Gamma, \alpha:\sort}{A_0}{B_0}{\Delta, \alpha:\sort, \Delta'}}
             {\equivjudg{\Gamma}{\alltype{\alpha:\sort} A_0}{\alltype{\alpha:\sort} B_0}{\Delta}}

             By i.h., either $(\Delta, \alpha : \sort, \Delta') = (\Gamma, \alpha : \sort)$,
             or $|\unsolved{\Delta, \alpha : \sort, \Delta'}| < |\unsolved{\Gamma, \alpha : \sort}|$.
             
             In the former case, \Lemmaref{lem:extension-inversion} (i) tells us that $\Delta' = \cdot$.
             Thus, $(\Delta, \alpha : \sort) = (\Gamma, \alpha:\sort)$, and so $\Delta = \Gamma$.
             
             In the latter case, we have $|\unsolved{\Delta, \alpha : \sort, \Delta'}| < |\unsolved{\Gamma, \alpha : \sort}|$,
             that is:
             \[
                       |\unsolved{\Delta}| + 0 + |\unsolved{\Delta'}|
                       ~~<~~
                       |\unsolved{\Gamma}| + 0
             \]
             Since $|\unsolved{\Delta'}|$ cannot be negative, we have $|\unsolved{\Delta}| < |\unsolved{\Gamma}|$.
        
        \DerivationProofCase{\EquivImplies}
             {\propequivjudg{\Gamma}{P}{Q}{\Theta}
              \\
              \equivjudg{\Theta}{[\Theta]A_0}{[\Theta]B_0}{\Delta}}
             {\equivjudg{\Gamma}{P \implies A_0}{Q \implies B_0}{\Delta}}
        
             Similar to the \EquivBin case, but using \Lemmaref{lem:propequiv-solving} on the first premise
             instead of the i.h.

        \DerivationProofCase{\EquivWith}
             {\propequivjudg{\Gamma}{P}{Q}{\Theta}
              \\
              \equivjudg{\Theta}{[\Theta]A_0}{[\Theta]B_0}{\Delta}}
             {\equivjudg{\Gamma}{A_0 \with P}{B_0 \with Q}{\Delta}}

             Similar to the \EquivWith case.
        
        \DerivationProofCase{\EquivInstL}
            {\instjudg{\Gamma_0[\ahat]}{\ahat}{\tau}{\type}{\Delta}
             \\
             \ahat \notin \FV{\tau}
            }
            {\equivjudg{\underbrace{\Gamma_0[\ahat]}_{\Gamma}}{\ahat}{\tau}{\Delta}}
            
            By \Lemmaref{lem:instantiation-solves}, $|\unsolved{\Delta}| = |\unsolved{\Gamma}| - 1$.
        
        \DerivationProofCase{\EquivInstR}
            {\instjudg{\Gamma_0[\ahat]}{\ahat}{\tau}{\type}{\Delta}
             \\
             \ahat \notin \FV{\tau}
            }
            {\equivjudg{\Gamma_0[\ahat]}{\tau}{\ahat}{\Delta}}

            Similar to the \EquivInstL case.
     \qedhere
     \end{itemize}
\end{proof}

\propdecidable*
\begin{proof}\proofheading
  Since there is no mutual recursion between the judgments, we can prove their decidability in order, separately.

  \newcommand{\propdecidableheading}[1]{\item \textsl{Decidability of {#1}}:}

  \begin{enumerate}[(1)]
  \propdecidableheading{$\checkeq{\Gamma}{\sigma}{t}{\sort}{\Delta}$}
    By induction on the sizes of $\sigma$ and $t$.

    \begin{itemize}
      \ProofCasesRules{\CheckeqVar, \CheckeqUnit, \CheckeqZero} No premises.

      \ProofCaseRule{\CheckeqSucc}
           Both $\sigma$ and $t$ get smaller in the premise.

      \ProofCasesRules{\CheckeqInstL, \CheckeqInstR}
           Follows from \Lemmaref{lem:instantiation-decidable}.
    \end{itemize}
  
  \propdecidableheading{$\checkprop{\Gamma}{P}{\Delta}$}      
     By induction on $\sigma$ and $t$.  But we have only one rule deriving this
     judgment form, \CheckpropEq, which has the judgment in (1) as a premise,
     so decidability follows from part (1).
  
  \propdecidableheading{ $\propequivjudg{\Gamma}{P}{Q}{\Delta}$}
     By induction on $P$ and $Q$.  But we have only one rule deriving this
     judgment form, \PropequivEq, which has two premises of the form (1),
     so decidability follows from part (1).

   \propdecidableheading{ $\elimeq{\Gamma}{\sigma}{t}{\sort}{\Deltabot}$}
     By lexicographic induction, first on the number of unsolved variables (both universal
     and existential) in $\Gamma$, then on $\sigma$ and $t$.
     We also show that the number of unsolved variables is nonincreasing in
     the output context (if it exists). 
     
     \begin{itemize}
       \ProofCasesRules{\ElimeqUvarRefl, \ElimeqZero}
           No premises, and the output is the same as the input. 

       \ProofCaseRule{\ElimeqClash}
           The only premise is the clash judgment, which is clearly decidable. 
           There is no output. 
       
       \ProofCaseRule{\ElimeqBin}
           In the first premise, we have the same $\Gamma$ but
           both $\sigma$ and $t$ are smaller.  By i.h., the first premise is decidable;
           moreover, either some variables in $\Theta$ were solved, or no
           additional variables were solved.

           If some variables in $\Theta$ were solved, the second premise is smaller
           than the conclusion according to our lexicographic measure, so by i.h.,
           the second premise is decidable.

           If no additional variables were solved, then $\Theta = \Gamma$.
           Therefore $[\Theta]\tau_2 = [\Gamma]\tau_2$.  It is given that
           $\sigma = [\Gamma]\sigma$ and $t = [\Gamma]t$, so $[\Gamma]\tau_2 = \tau_2$.
           Likewise, $[\Theta]\tau'_2 = [\Gamma]\tau'_2 = \tau'_2$,
           so we aremaking a recursive call on a strictly smaller subterm.

           Regardless, $\Deltabot$ is either $\bot$, or is a $\Delta$
           which has no more unsolved variables than $\Theta$, which 
           in turn has no more unsolved variables than $\Gamma$. 

       \ProofCaseRule{\ElimeqBinBot}

           The premise is invoked on subterms, and does not yield an output
           context. 

       \ProofCaseRule{\ElimeqSucc}  Both $\sigma$ and $t$ get smaller.
           By i.h., the output context has fewer unsolved variables, if it exists. 

       \ProofCasesRules{\ElimeqInstL, \ElimeqInstR}
           Follows from
           \Lemmaref{lem:instantiation-decidable}.
           Furthermore, by \Lemmaref{lem:instantiation-solves},
           instantiation solves a variable in the output.

       \ProofCasesRules{\ElimeqUvarL, \ElimeqUvarR}
           These rules have no nontrivial premises, and $\alpha$ is solved in the output context.

       \ProofCasesRules{\ElimeqUvarLBot, \ElimeqUvarRBot}
           These rules have no nontrivial premises, and produce the output context $\bot$.
     \end{itemize}

  \propdecidableheading{ $\elimprop{\Gamma}{P}{\Deltabot}$}
       By induction on $P$.  But we have only one rule deriving this
       judgment form, \ElimpropEq, for which decidability follows from part (4).
  \qedhere
  \end{enumerate}
\end{proof}

\equivdecidable*
\begin{proof}\proofheading
  Let the judgment $\equivjudg{\Gamma}{A}{B}{\Delta}$ be measured lexicographically by

  \begin{enumerate}[(E1)]
  \item $\largecount{A} + \largecount{B}$;
  \item $|\unsolved{\Gamma}|$, the number of unsolved existential variables in $\Gamma$;
  \item $|A| + |B|$.
  \end{enumerate}

     \begin{itemize}
        \ProofCasesRules{\EquivVar, \EquivExvar, \EquivUnit} 
           No premises.

        \DerivationProofCase{\EquivBin}
             {\equivjudg{\Gamma}{A_1}{B_1}{\Theta}
               \\
               \equivjudg{\Theta}{[\Theta]A_2}{[\Theta]B_2}{\Delta}
             }
             {\equivjudg{\Gamma}{A_1 \binc A_2}{B_1 \binc B_2}{\Delta}}

             In the first premise, part (E1) either gets smaller (if $A_2$ or $B_2$ have large connectives)
             or stays the same.  Since the first premise has the same input context, part (E2) remains the same.
             However, part (E3) gets smaller.

             In the second premise, part (E1) either gets smaller (if $A_1$ or $B_1$ have large connectives)
             or stays the same.  

        \ProofCaseRule{\EquivVec}
             Similar to a special case of \EquivBin, where two of the types are monotypes.

        \DerivationProofCase{\EquivAll}
             {\equivjudg{\Gamma, \alpha:\sort}{A_0}{B_0}{\Delta, \alpha:\sort, \Delta'}}
             {\equivjudg{\Gamma}{\underbrace{\alltype{\alpha:\sort} A_0}_{A}}{\underbrace{\alltype{\alpha:\sort} B_0}_{B}}{\Delta}}

             Since $\largecount{A_0} + \largecount{B_0} = \largecount{A} + \largecount{B} - 2$,
             the first part of the measure gets smaller.
        
        \DerivationProofCase{\EquivImplies}
             {\propequivjudg{\Gamma}{P}{Q}{\Theta}
              \\
              \equivjudg{\Theta}{[\Theta]A_0}{[\Theta]B_0}{\Delta}}
             {\equivjudg{\Gamma}{\underbrace{P \implies A_0}_{A}}{\underbrace{Q \implies B_0}_{B}}{\Delta}}
             
             The first premise is decidable by \Lemmaref{lem:prop-decidable} (3).

             For the second premise, by \Lemmaref{lem:substitution-not-large},
             $\largecount{[\Theta]A_0} = \largecount{A_0}$ and 
             $\largecount{[\Theta]B_0} = \largecount{B_0}$.
             Since $\largecount{A} = \largecount{A_0} + 1$ and $\largecount{B} = \largecount{B_0} + 1$,
             we have
             \[
                     \largecount{[\Theta]A_0} + \largecount{[\Theta]B_0}
                     ~<~
                     \largecount{A} + \largecount{B}
             \]
             which makes the first part of the measure smaller.
        
        \DerivationProofCase{\EquivWith}
             {\propequivjudg{\Gamma}{P}{Q}{\Theta}
              \\
              \equivjudg{\Theta}{[\Theta]A_0}{[\Theta]B_0}{\Delta}}
             {\equivjudg{\Gamma}{A_0 \with P}{B_0 \with Q}{\Delta}}

             Similar to the \EquivImplies case.

        \DerivationProofCase{\EquivInstL}
            {\instjudg{\Gamma[\ahat]}{\ahat}{\tau}{\type}{\Delta}
             \\
             \ahat \notin \FV{\tau}
            }
            {\equivjudg{\Gamma[\ahat]}{\ahat}{\tau}{\Delta}}

            Follows from \Lemmaref{lem:instantiation-decidable}.
        
        \ProofCaseRule{\EquivInstR}
           Similar to the \EquivInstL case.
     \qedhere
     \end{itemize}
\end{proof}

\subsection{Decidability of Subtyping}

\subtypingdecidable*
\begin{proof}\proofheading
  Let the judgments be measured lexicographically by
  $\largecount{A} + \largecount{B}$.

  For each subtyping rule, we show that every premise is smaller than the conclusion, or already known to be decidable. 
  The condition that $[\Gamma]A = A$ and $[\Gamma]B = B$ is easily satisfied at each inductive
  step, using the definition of substitution.

  Now, we consider the rules deriving $\subjudg[\polvar]{\Gamma}{A}{B}{\Delta}$.

  \begin{itemize}
      \DerivationProofCase{\SubEquiv}
                {
                  \arrayenvbl{\notAllExists{A} \\ \notAllExists{B}}
                  \\
                  \equivjudg{\Gamma}{A}{B}{\Delta}
                }
                {\subjudg[\polvar]{\Gamma}{A}{B}{\Delta}}

      In this case, we appeal to \Lemmaref{lem:equiv-decidable}.           

      \DerivationProofCase{\SubAllL}
                {
                  \arrayenvbl{\notAll{B}
                      \\
                      \subjudg[-]{\Gamma, \MonnierComma{\ahat}, \ahat:\sort}
                             {[\ahat/\alpha]A}
                             {B}
                             {\Delta, \MonnierComma{\ahat}, \Theta}
                   }
                }
                {\subjudg[-]{\Gamma}{\alltype{\alpha:\sort}{A}}{B}{\Delta}}

      The premise has one fewer quantifier.

      \DerivationProofCase{\SubAllR}
                {\subjudg[-]{\Gamma, \beta :\sort}{A}{B}{\Delta, \beta:\sort, \Theta}}
                {\subjudg[-]{\Gamma}{A}{\alltype{\beta:\sort}{B}}{\Delta}}

      The premise has one fewer quantifier.

      \DerivationProofCase{\SubExistsL}
                {
                 \subjudg[+]{\Gamma, \alpha:\sort}
                         {A}
                         {B}
                         {\Delta, \alpha:\sort, \Theta}}
                {\subjudg[+]{\Gamma}
                         {\extype{\alpha:\sort}{A}}
                         {B}
                         {\Delta}}

      The premise has one fewer quantifier.

      \DerivationProofCase{\SubExistsR}
                {
                  \arrayenvbl{
                      \notExists{A}
                      \\
                      \subjudg[+]{\Gamma, \MonnierComma{\bhat}, \bhat:\sort}
                             {A}
                             {[\bhat/\beta]B}
                             {\Delta, \MonnierComma{\bhat}, \Theta}
                  }
                }
                {\subjudg[+]{\Gamma}
                         {A}
                         {\extype{\beta:\sort}{B}}
                         {\Delta}}

      The premise has one fewer quantifier.

       \DerivationProofCase{\SubPosNegL}
             {\subjudg[-]{\Gamma}{A}{B}{\Delta}
               \\
               \arrayenvbl{\Neg{A}
               \\
               \nonPos{B}}}
             {\subjudg[+]{\Gamma}{A}{B}{\Delta}}

       Consider whether $B$ is negative. 
       \begin{itemize}
       \item Case $\Neg{B}$:

         \begin{llproof}
           \eqPf{B}{\alltype{\beta:\sort}{B'}}{Definition of $\Neg{B}$}
           \subjudgPf{-}{\Gamma, \beta:\sort}{A}{B'}{\Delta, \beta:\sort,\Theta}{Inversion on the premise}
         \end{llproof}

       There is one fewer quantifier in the subderivation.

       \item Case $\nonNeg{B}$:

         In this case, $B$ is not headed by a $\forall$. 

         \begin{llproof}
           \eqPf{A}{\alltype{\alpha:\sort}{A'}}{Definition of $\Neg{A}$}
           \subjudgPf{-}{\Gamma, \MonnierComma{\ahat}, \ahat:\sort}{[\ahat/\alpha]A'}{'}{\Delta, \MonnierComma{\ahat},\Theta}{Inversion on the premise}
         \end{llproof}
       \end{itemize}

       There is one fewer quantifier in the subderivation. 

       \DerivationProofCase{\SubPosNegR}
             {\subjudg[-]{\Gamma}{A}{B}{\Delta}
              \\
              \arrayenvbl{\nonPos{A}
              \\
              \Neg{B}}}
             {\subjudg[+]{\Gamma}{A}{B}{\Delta}}

         \begin{llproof}
           \eqPf{B}{\alltype{\beta:\sort}{B'}}{Definition of $\Neg{B}$}
           \subjudgPf{-}{\Gamma, \beta:\sort}{A}{B'}{\Delta, \beta:\sort,\Theta}{Inversion on the premise}
         \end{llproof}

       There is one fewer quantifier in the subderivation.

       \DerivationProofCase{\SubNegPosL}
             {\subjudg[+]{\Gamma}{A}{B}{\Delta}
              \\
              \arrayenvbl{\Pos{A}
              \\
              \nonNeg{B}}}
             {\subjudg[-]{\Gamma}{A}{B}{\Delta}}

       This case is similar to the $\SubPosNegR$ case. 

       \DerivationProofCase{\SubNegPosR}
             {\subjudg[+]{\Gamma}{A}{B}{\Delta}
              \\
              \arrayenvbl{\nonNeg{A}
              \\
              \Pos{B}}}
             {\subjudg[-]{\Gamma}{A}{B}{\Delta}}

       This case is similar to the $\SubPosNegL$ case. 

     \qedhere
     \end{itemize}

\end{proof}

\subsection{Decidability of Matching and Coverage}

\guardeddecidable*
\begin{proof}
  This follows via a routine induction on $\Pi$, counting the
  number of branch lists. 
\end{proof}

\expansiondecidable*
\begin{proof}\proofheading
  In each part, by induction on $\Pi$: Every rule either has no premises,
  or breaks down $\Pi$ in its nontrivial premise.
\end{proof}

\expansionshrinkssize*
\begin{proof}
  Properties 1-5 follow by a routine induction on the derivation of the expansion judgement. Since expansion only ever removes pattern constructors, and only ever adds wildcards, it never increases the size of the resulting branch list.

  Case 6 for vectors proceeds by induction on the derivation of $\guarded{\Pi}$, and furthermore depends upon the proof for case 5.
  \begin{enumerate}
 \DerivationProofCase{ }
        { }
        {\guarded{\branch{\vecnil, \vec{p}}{e} \alt \Pi}}

        By inversion on the expansion derivation, we know $\expandvec{\Pi}{\Pi_{\vecnil}}{\Pi_{\vecconssym}}$. \\
        By part 5, we know that $|\Pi_{\vecnil}| \leq |\Pi|$ and $|\Pi_{\vecconssym}| \leq |\Pi|$. \\
      By the definition of size, we know that $|\branch{\vec{p}}{e}| < |\branch{\vecnil, \vec{p}}{e}|$. \\
      \Hand Hence $|\branch{\vec{p}}{e} \alt \Pi_{\vecnil}| < |\branch{\vecnil, \vec{p}}{e} \alt \Pi|$. \\
     By the definition of size, we know that $|\Pi| < |\branch{\vecnil, \vec{p}}{e} \alt \Pi|$. \\
     \Hand Hence $|\Pi_{\vecconssym}| < |\branch{\vecnil, \vec{p}}{e} \alt \Pi|$. \\

  \DerivationProofCase{ }
        { }
        {\guarded{\branch{\veccons{p}{p'}, \vec{p}}{e} \alt \Pi}}

        By inversion on the expansion derivation, we know $\expandvec{\Pi}{\Pi_{\vecnil}}{\Pi_{\vecconssym}}$. \\
        By part 5, we know that $|\Pi_{\vecnil}| \leq |\Pi|$ and $|\Pi_{\vecconssym}| \leq |\Pi|$. \\
      By the definition of size, we know that $|\branch{p, p', \vec{p}}{e}| < |\branch{\veccons{p}{p'}, \vec{p}}{e}|$. \\
      \Hand Hence $|\branch{p, p', \vec{p}}{e} \alt \Pi_{\vecconssym}| < |\branch{\veccons{p}{p'}, \vec{p}}{e} \alt \Pi|$. \\
     By the definition of size, we know that $|\Pi| < |\branch{\veccons{p}{p'}, \vec{p}}{e} \alt \Pi|$. \\
     \Hand Hence $|\Pi_{\vecnil}| < |\branch{\vecnil, \vec{p}}{e} \alt \Pi|$. \\

  \DerivationProofCase{ }
        {\guarded{\Pi} }
        {\guarded{\branch{\wild, \vec{p}}{e} \alt \Pi}}

        By inversion on the expansion derivation, we know $\expandvec{\Pi}{\Pi_{\vecnil}}{\Pi_{\vecconssym}}$. \\ 
        By induction, $|\Pi_{\vecnil}| < |\Pi|$ and $|\Pi_{\vecconssym}| < |\Pi|$. \\
        \Hand By the definition of size, $|\branch{\wild, \vec{p}}{e} \alt \Pi_{\vecnil}| < |\branch{\wild, \vec{p}}{e} \alt \Pi|$  \\
        \Hand By the definition of size, $|\branch{\wild, \vec{p}}{e} \alt \Pi_{\vecconssym}| < |\branch{\wild, \vec{p}}{e} \alt \Pi|$ 
        
  \DerivationProofCase{ }
        {\guarded{\Pi} }
        {\guarded{\branch{x, \vec{p}}{e} \alt \Pi}}

        Similar to previous case.
        
  \end{enumerate}

\end{proof}

\coveragedecidable*
\begin{proof}\proofheading
  By induction on, lexicographically,
  (1) the size $|\Pi|$ of the branch list $\Pi$ 
  and then
  (2) the number of $\with$ connectives in $\Avec$,
  and then 
  (3) the size of $\Avec$, considered to be the sum of the sizes $\plaintypesize{A}$
  of each type $A$ in $\Avec$ (treating constraints $s = t$ as size 1). 

  (For \CoversVar, \CoversTimes, \CoversVec, \CoversVecOK, and \CoversSum, we also use the appropriate
  part of \Lemmaref{lem:expansion-decidable}, as well as \Lemmaref{lem:expansion-shrinks-size}.)
  
  \begin{itemize}
  \ProofCaseRule{\CoversEmpty}   No premises.

  \ProofCaseRule{\CoversVar}  The number of $\with$ connectives does not grow, and the size of the branch list stays the same, and 
  $\Avec$ gets smaller.

  \ProofCaseRule{\CoversUnit} The number of $\with$ connectives and the size of the branch list stays the same, and 
  $\Avec$ gets smaller.

  \ProofCaseRule{\CoversWith} The size of the branch list stays the same, and the number of $\with$ connectives in $\Avec$ goes down. This lets us analyze the two possibilities for the coverage-with-assumptions judgement:

  \begin{itemize}
  \ProofCaseRule{\CoversEq}
     The first premise is decidable by \Lemmaref{lem:prop-decidable} (4).
     The number of $\with$ connectives in $\Avec$ gets smaller (note that applying
     $\Delta$ as a substitution cannot add $\with$ connectives).
  
  \ProofCaseRule{\CoversEqBot} The premise is decidable by \Lemmaref{lem:prop-decidable} (4).
  \end{itemize}

  \ProofCaseRule{\CoversWithOK} The size of the branch list stays the same, and the  number of $\with$ connectives in $\Avec$ goes down.

  \ProofCaseRule{\CoversTimes}  The size of the branch list does not grow, the number of $\with$ connectives stays the same, 
  and $\Avec$ gets smaller, since
      $\plaintypesize{A_1} + \plaintypesize{A_2}
       <
       \plaintypesize{A_1 \times A_2}$.

  \ProofCaseRule{\CoversSum}  Here we have $\Avec = (A_1 + A_2, \Bvec)$.
       In the first premise, we have $(A_1, \Bvec)$, which is smaller than $\Avec$,
       and in the second premise we have $(A_2, \Bvec)$, which is likewise smaller.
       (In both premises, the size of the branch list does not grow, and the number of $\with$ connectives stays the same.)

  \ProofCaseRule{\CoversVec} 

  Since $\guarded{\Pi}$ is decidable, and $\expandvec{\Pi}{\Pi_{\vecnil}}{\Pi_{\vecconssym}}$ is decidable, then we know that the size of the branch lists ${\Pi_{\vecnil}}$ and ${\Pi_{\vecconssym}}$ is strictly smaller than $\Pi$. 

  This lets us analyze the two cases for each premise, noting that the assumption is
  decidable by \Lemmaref{lem:prop-decidable} (4). 
  \begin{itemize}
  \ProofCaseRule{\CoversEq}
  The first premise (that $t = \zero$) is decidable by \Lemmaref{lem:prop-decidable} (4).
  The size of $\Pi_{\vecnil}$ is strictly smaller than $\Pi$'s size, so we can still appeal to
  induction (note $\Delta$ as a substitution cannot add change the size of a branch list).
  
  \ProofCaseRule{\CoversEqBot}  Decidable by \Lemmaref{lem:prop-decidable} (4).
  \end{itemize}

  The cons case is nearly identical:
  \begin{itemize}
  \ProofCaseRule{\CoversEq}
  The first premise (that $t = \succ{n}$) is decidable by \Lemmaref{lem:prop-decidable} (4).
  The size of $\Pi_{\vecnil}$ is strictly smaller than $\Pi$'s size, so we can still appeal to
  induction (note $\Delta$ as a substitution cannot add change the size of a branch list).
  
  \ProofCaseRule{\CoversEqBot}  Decidable by \Lemmaref{lem:prop-decidable} (4).
  \end{itemize}

  \ProofCaseRule{\CoversVecOK}

  Since $\guarded{\Pi}$ is decidable, and $\expandvec{\Pi}{\Pi_{\vecnil}}{\Pi_{\vecconssym}}$ is decidable, then we know that the size of the branch lists ${\Pi_{\vecnil}}$ and ${\Pi_{\vecconssym}}$ is strictly smaller than $\Pi$. 

  \ProofCaseRule{\CoversEx}  The size of the branch list does not grow, 
        and $\Avec$ gets smaller.

  \ProofCaseRule{\CoversEq}
     The first premise is decidable by \Lemmaref{lem:prop-decidable} (4).
     The number of $\with$ connectives in $\Avec$ gets smaller (note that applying
     $\Delta$ as a substitution cannot add $\with$ connectives).
  
  \ProofCaseRule{\CoversEqBot}  Decidable by \Lemmaref{lem:prop-decidable} (4).
  \qedhere
  \end{itemize}
\end{proof}

\subsection{Decidability of Typing}

\typingdecidable*
\begin{proof}\proofheading
    For rules deriving judgments of the form
    \[
        \begin{array}[t]{lll}
               \synjudg{-}{\Gamma}{e}{-}{-}
            \\ \chkjudg{p}{\Gamma}{e}{B}{-}
            \\ \spinejudg{\Gamma}{s}{B}{p}{-}{-}{-}
            \\ \matchjudg{q}{p}{\Gamma}{\Pi}{\Avec}{C}{-}
        \end{array}
    \]
    (where we write ``$-$'' for parts of the judgments that are outputs),
    the following induction measure on such judgments is adequate to prove decidability:
    \[
      \left\langle
        e / s / \Pi,~~~
        \begin{array}[c]{llllll}
            \syn
          \\[0.3ex]
            \chk/\spinejudgsym,
            &
            \largecount{B},
            &
            B
         \\[0.3ex]
            \textsf{Match},
            &
            \Avec,
            &
            \text{match judgment form}
        \end{array}
      \right\rangle
    \]
    where $\langle \dots \rangle$ denotes lexicographic order, and
    where (when comparing two judgments typing terms of the same size)
    the synthesis judgment (top line) is considered smaller than
    the checking judgment (second line).  That is,
    \[
               \syn ~~\bigprec~~ \chk/\spinejudgsym/~\textsf{Match}
    \]
    Two match judgments are compared according to, first, the list of branches $\Pi$
    (which is a subterm of the containing case expression, allowing us to invoke
    the i.h.\ for the \Case rule), then the size of the list of types $\Avec$
    (considered to be the sum of the sizes $\plaintypesize{A}$ of each type
    $A$ in $\Avec$), and then, finally, 
    whether the judgment is $\Gamma / P \entails \dots$ or $\Gamma \entails \dots$,
    considering the former judgment ($\Gamma / P \entails \dots$) to be larger.

    Note that this measure only uses the input parts of the judgments, leading to a
    straightforward decidability argument.

    We will show that in each rule deriving a synthesis, checking, spine or
    match judgment, every premise is smaller than the conclusion.

    \begin{itemize}
      \ProofCaseRule{\EmptySpine}  No premises.

      \ProofCaseRule{\ArrSpine}  In each premise, the expression/spine gets smaller
         (we have $\appspine{e}{s}$ in the conclusion, $e$ in the first premise,
         and $s$ in the second premise).

      \ProofCaseRule{\Var}  No nontrivial premises.

      \ProofCaseRule{\Sub}
          The first premise has the same subject term $e$ as the conclusion, but
          the judgment is smaller because our measure considers synthesis
          to be smaller than checking.

          The second premise is a subtyping judgment, which by \Theoremref{thm:subtyping-decidable}
          is decidable.

      \ProofCaseRule{\Anno}
          It is easy to show that the judgment $\judgetp{\Gamma}{A\,\p}$ is decidable.
          The second premise types $e$, but the conclusion types $(e : A)$,
          so the first part of the measure gets smaller.

      \ProofCasesRules{\UnitIntro, \UnitIntroSolve}  No premises.
      
      \ProofCaseRule{\AllIntro}  Both the premise and conclusion type $e$, and both are checking;
         however, $\largecount{A_0} < \largecount{\alltype{\alpha:\sort} A_0}$,
         so the premise is smaller.

      \ProofCaseRule{\AllSpine}  Both the premise and conclusion type $\appspine{e}{s}$, and both are spine judgments; however, $\largecount{-}$ decreases.
      
      \ProofCaseRule{\WithIntro}
         By \Lemmaref{lem:prop-decidable} (2), the first premise is decidable.
         For the second premise,
         $\largecount{[\Theta]A_0} = \largecount{A_0} < \largecount{A_0 \with P}$.

       \ProofCaseRule{\ExIntro}
          Both the premise and conclusion type $e$, and both are checking;
          however, %
          $\largecount{-}$ decreases
          so the premise is smaller.

      \ProofCaseRule{\ImpliesIntro}
         For the first premise, use \Lemmaref{lem:prop-decidable} (5).
         In the second premise, $\largecount{-}$ gets smaller (similar to the \WithIntro case).
      \ProofCaseRule{\ImpliesIntroBot}
         The premise is decidable by \Lemmaref{lem:prop-decidable} (5).

      \ProofCaseRule{\ImpliesSpine}
         Similar to the \WithIntro case.
      
      \ProofCasesRules{\ArrIntro, \ArrIntroSolve}  In the premise, the term is smaller.

      \ProofCasesRules{\ArrElim}%
      In all premises, the term is smaller.
      
      \ProofCasesRules{\SumIntro{k}, \SumIntroSolve{k}, \PairIntro, \PairIntroSolve}  In all premises, the term is smaller.
      
      \ProofCaseRule{\Case}
         In the first premise, the term is smaller.
         In the second premise, we have a list of branches that is a proper subterm of the
         case expression.
         The third premise is decidable by \Theoremref{thm:coverage-decidable}.
  \end{itemize}

  We now consider the match rules:

  \begin{itemize}
      \ProofCaseRule{\MatchEmpty}  No premises.

      \ProofCaseRule{\MatchSeq}  In each premise, the list of branches is properly contained
         in $\Pi$, making each premise smaller by the first part (``$e/s/\Pi$'') of the measure.

      \ProofCaseRule{\MatchBase}
          The term $e$ in the premise is properly contained in $\Pi$.

      \ProofCasesRules{\MatchExists, \MatchPair, \MatchSum{k}, \MatchNeg, \MatchWild}
          Smaller by part (2) of the measure.

      \ProofCaseRule{\MatchWith}
          The premise has a smaller $\Avec$, so it is smaller by the $\Avec$ part of the measure.
          (The premise is the other judgment form, so it is \emph{larger} by the ``match judgment form''
          part, but $\Avec$ lexicographically dominates.)

      \ProofCaseRule{\MatchBot}
          For the premise, use \Lemmaref{lem:prop-decidable} (4).

      \ProofCaseRule{\MatchUnify}

          \Lemmaref{lem:prop-decidable} (4) shows that the first premise is decidable. \\
          The second premise has the same (single) branch and list of types, but is smaller
          by the ``match judgment form'' part of the measure.
  \qedhere
  \end{itemize}
\end{proof}

\section{Determinacy}

\newcommand{\MAINLABEL}[1]{\mathsz{14pt}{#1}}
\newcommand{\NORULE}{\mathcolor{dRed}{\emptyset}}
\newcommand{\TABLENOTE}[1]{\text{\it \textcolor{dGreen}{#1}}}

\auxdet*
\begin{proof}\proofheading ~\\
  \textbf{Proof of Part (1) (Elimeq).} ~\\
     Rule \ElimeqZero applies if and only if $\sigma = t = \zero$.

     Rule \ElimeqSucc applies if and only if $\sigma$ and $t$ are headed by $\xsucc$.
     
     Now suppose $\sigma = \alpha$.
     
     \begin{itemize}
     \item 
         Rule \ElimeqUvarRefl applies if and only if $t = \alpha$.
         (Rule \ElimeqClash cannot apply; rules \ElimeqUvarL and \ElimeqUvarR
         have a free variable condition; rules \ElimeqUvarLBot and \ElimeqUvarRBot
         have a condition that $\sigma \neq t$.)

         In the remainder, assume $t \neq alpha$.

     \item
         If $\alpha \in \FV{t}$, then rule \ElimeqUvarLBot applies,
         and no other rule applies (including \ElimeqUvarRBot and \ElimeqClash).

         In the remainder, assume $\alpha \notin \FV{t}$.
 
     \item
         Consider whether \ElimeqUvarRBot applies.
         The conclusion matches if we have $t = \beta$ for some $\beta \neq \alpha$
         (that is, $\sigma = \alpha$ and $t = \beta$).  But \ElimeqUvarRBot has a condition
         that $\beta \in \FV{\sigma}$, and $\sigma = \alpha$, so the condition is not satisfied.
     \end{itemize}

     In the symmetric case, use the reasoning above, exchanging \rulename{L}'s
     and \rulename{R}'s in the rule names.

  \paragraph{Proof of Part (2) (Instantiation).} ~\\
     Rule \InstBin applies if and only if $t$ has the form $t_1 \binc t_2$.

     Rule \InstZero applies if and only if $t$ has the form $\zero$.

     Rule \InstSucc applies if and only if $t$ has the form $\succ{t_0}$.
  
     If $t$ has the form $\bhat$, then consider whether $\bhat$ is declared to
     the left of $\ahat$ in the given context:

     \begin{itemize}
     \item  If $\bhat$ is declared to the left
       of $\ahat$, then rule \InstReach cannot be used, which leaves only \InstSolve.
     \item  If $\bhat$ is declared to the right of $\ahat$, then \InstSolve cannot be used
     because $\bhat$ is not well-formed under $\Gamma_0$ (the context to the left of $\ahat$ in \InstSolve).  That leaves only \InstReach.
     \item  $\ahat$ cannot be $\bhat$, because it is given that $\ahat \notin \FV{t} = \FV{\bhat} = \{\bhat\}$.
     \end{itemize}

  \paragraph{Proof of Part (3) (Symmetric instantiation).}  ~\\
      \InstBin, \InstZero and \InstSucc cannot have been used in either derivation.
      
      Suppose that \InstSolve concluded $\Dee_1$.  Then $\Delta_1$ is the same
      as $\Gamma$ with $\ahat$ solved to $\bhat$.  Moreover, $\bhat$ is declared
      to the left of $\ahat$ in $\Gamma$.  Thus, \InstSolve cannot conclude $\Dee_2$.
      However, \InstReach can conclude $\Dee_2$, but produces a context $\Delta_2$
      which is the same as $\Gamma$ but with $\ahat$ solved to $\bhat$.  Therefore
      $\Delta_1 = \Delta_2$.

      The other possibility is that \InstReach concluded $\Dee_1$.  Then $\Delta_1$ is
      the same as $\Gamma$ with $\bhat$ solved to $\ahat$, with $\ahat$ declared
      to the left of $\bhat$ in $\Gamma$.  Thus, \InstReach cannot conclude $\Dee_2$.
      However, \InstSolve can conclude $\Dee_2$, producing a context $\Delta_2$
      which is the same as $\Gamma$ but with $\bhat$ solved to $\ahat$.  Therefore
      $\Delta_1 = \Delta_2$.
      
  \paragraph{Proof of Part (4) (Checkeq).} ~\\
     Rule \CheckeqVar applies if and only if $\sigma = t = \ahat$ or $\sigma = t = \alpha$
     (note the free variable conditions in \CheckeqInstL and \CheckeqInstR).

     Rule \CheckeqUnit applies if and only if $\sigma = t = \unitty$.

     Rule \CheckeqBin applies if and only if $\sigma$ and $t$ are both headed by
     the same binary connective.

     Rule \CheckeqZero applies if and only if $\sigma = t = \zero$.

     Rule \CheckeqSucc applies if and only if $\sigma$ and $t$ are headed by $\xsucc$.

     Now suppose $\sigma = \ahat$.  If $t$ is not an existential variable, then \CheckeqInstR
     cannot be used, which leaves only \CheckeqInstL.  If $t$ is an existential variable,
     that is, some $\bhat$ (distinct from $\ahat$), and is unsolved, then both
     \CheckeqInstL and \CheckeqInstR apply, but by part (3), we get the same
     output context from each.

     The $t = \ahat$ subcase is similar.

  \paragraph{Proof of Part (5) (Elimprop).}
     There is only one rule deriving this judgment; the result follows by part (1).

  \paragraph{Proof of Part (6) (Checkprop).}
     There is only one rule deriving this judgment; the result follows by part (4).
\qedhere
\end{proof}

\equivdet*
\begin{proof}\proofheading ~\\
  \textbf{Proof of Part (1) (propositional equivalence).}
     Only one rule derives judgments of this form; the result follows
     from \Lemmaref{lem:aux-det} (4).

  \paragraph{Proof of Part (2) (type equivalence).}
  
     If neither $A$ nor $B$ is an existential variable, they must have the same
     head connectives, and the same rule must conclude both derivations.
     
     If $A$ and $B$ are the same existential variable, then only \EquivExvar
     applies (due to the free variable conditions in \EquivInstL and \EquivInstR).
     
     If $A$ and $B$ are different unsolved existential variables,
     the judgment matches the conclusion of both \EquivInstL and \EquivInstR,
     but by part (3) of \Lemmaref{lem:aux-det}, we get the same output context
     regardless of which rule we choose.
\end{proof}

\subtypingdet*
\begin{proof}\proofheading
  First, we consider whether we are looking at positive or negative subtyping, and then
  consider the outermost connective of $A$ and $B$:
  \begin{itemize}
  \item  If $\subjudg[+]{\Gamma}{A}{B}{\Delta_1}$ and $\subjudg[+]{\Gamma}{A}{B}{\Delta_2}$, then we
    know the last rule ending the derivation of $\Dee_1$ and $\Dee_2$ must be:

      \[
      \renewcommand{\arraystretch}{1.2}
         \begin{array}[t]{rlcccccc}
         &\multicolumn{4}{c}{\MAINLABEL{B}}
        \\[2pt]
        &
                    &   \AllSym & \ExistsSym & \text{other}
        \\\cmidrule[1pt]{3-5}
        & \AllSym
          &   \SubPosNegR, \SubPosNegL & \SubExistsR & \SubPosNegL
        \\
 \MAINLABEL{A}
        & \ExistsSym
          &   \SubExistsL & \SubExistsL & \SubExistsL
        \\
        & \text{other}
          &   \SubPosNegR & \SubExistsR & \SubEquiv
         \end{array}
      \]

      The only case in which there are two possible final rules is in
      the $\forall/\forall$ case. In this case, regardless of the
      choice of rule, by inversion we get subderivations $\subjudg[-]{\Gamma}{A}{B}{\Delta_1}$ and 
      $\subjudg[-]{\Gamma}{A}{B}{\Delta_2}$.

    \item If $\subjudg[-]{\Gamma}{A}{B}{\Delta_1}$ and $\subjudg[-]{\Gamma}{A}{B}{\Delta_2}$, then we
    know the last rule ending the derivation of $\Dee_1$ and $\Dee_2$ must be:

      \[
      \renewcommand{\arraystretch}{1.2}
         \begin{array}[t]{rlcccccc}
         &\multicolumn{4}{c}{\MAINLABEL{B}}
        \\[2pt]
        &
                    &   \AllSym & \ExistsSym & \text{other}
        \\\cmidrule[1pt]{3-5}
        & \AllSym
          &   \SubAllR & \SubAllL & \SubAllL
        \\
 \MAINLABEL{A}
        & \ExistsSym
          &   \SubAllR & \SubNegPosL, \SubNegPosR & \SubNegPosL
        \\
        & \text{other}
          &   \SubAllR & \SubNegPosR & \SubEquiv
         \end{array}
      \]

      The only case in which there are two possible final rules is in
      the $\forall/\forall$ case. In this case, regardless of the
      choice of rule, by inversion we get subderivations $\subjudg[+]{\Gamma}{A}{B}{\Delta_1}$ and 
      $\subjudg[+]{\Gamma}{A}{B}{\Delta_2}$.
    \end{itemize}

    As a result, the result follows by a routine induction. \qedhere

\end{proof}

\typingdet*
\begin{proof}\proofheading ~\\
  \textbf{Proof of Part (1) (checking).}
      
  The rules with a checking judgment in the conclusion are:
      \UnitIntro, \UnitIntroSolve,
      \AllIntro,
      \WithIntro,
      \ExIntro,
      \ImpliesIntro,
      \ImpliesIntroBot,
      \ArrIntro,
      \ArrIntroSolve,
      \Rec,
      \SumIntro{k},
      \SumIntroSolve{k},
      \PairIntro,
      \PairIntroSolve,
      \Case,
      \Nil, \Cons.

      The table below shows which rules apply for given $e$ and $A$.
      The extra ``${\xchkintro}\text{?}$'' column highlights the role of the ``\xchkintro''
      (``check-intro'') category of syntactic forms: we restrict the introduction rules
      for $\AllSym$ and $\implies$   %
      to type only these forms.
      For example, given $e = x$ and $A = (\alltype{\alpha:\sort}A_0)$,
      we need not choose between \Sub and \AllIntro:
      the latter is ruled out by its $\xchkintro$ premise.
      \[
\newcommand{\SubOrExIntro}{\ExIntro}  %
      \renewcommand{\arraystretch}{1.2}
         \begin{array}[t]{r@{\hspace*{3.3ex}}llllllllllllllll}
         &\multicolumn{12}{c}{\MAINLABEL{A}}
        \\[-2pt]
        &
        &
          &  &  \TABLENOTE{Note 1} &  &
          &&&&&&&
        \\[-4pt]
        &
        & {\xchkintro}\text{?}
                    &   \AllSym & \implies & \ExistsSym & \with
                    & \arr
                    & + &  \times   & \unitty
                    & \ahat
                    & \alpha
                    & \vecop
        \\\cmidrule[1pt]{4-14}
        & \lam{x} e_0
        & \xchkintro
          &   \AllIntro & \ImpliesIntro/\ImpliesIntroBot & \SubOrExIntro &    \WithIntro
                                  &    \ArrIntro
                                  &    \NORULE &    \NORULE & \NORULE
                                  &    \ArrIntroSolve
                                  &    \NORULE
                    & \NORULE
        \\
        & \rec{x} v
        & \TABLENOTE{Note 2}
          &   \Rec & \Rec & \Rec &    \Rec
                                  &    \Rec
                                  &    \Rec &    \Rec & \Rec
                                  &    \Rec
                                  &    \Rec
                    & \NORULE
        \\
        & \inj{k}{e_0}
        & \xchkintro
          &   \AllIntro & \ImpliesIntro/\ImpliesIntroBot & \SubOrExIntro &    \WithIntro
                                  & \NORULE
                                  &    \SumIntro{k} &    \NORULE & \NORULE
                                  &    \SumIntroSolve{k}
                                  &    \NORULE
                    & \NORULE
        \\
        & \pair{e_1}{e_2}
        & \xchkintro
          &   \AllIntro & \ImpliesIntro/\ImpliesIntroBot & \SubOrExIntro &    \WithIntro
                                  &    \NORULE
                                  &    \NORULE &    \PairIntro & \NORULE
                                  &    \PairIntroSolve
                                  &    \NORULE
                    & \NORULE
        \\
        & \unitexp
        & \xchkintro
          &   \AllIntro & \ImpliesIntro/\ImpliesIntroBot & \SubOrExIntro &    \WithIntro
                                  &    \NORULE
                                  &    \NORULE &    \NORULE & \UnitIntro
                                  &    \UnitIntroSolve
                                  &    \NORULE
                      &    \NORULE
        \\
 \MAINLABEL{e}
        & \vecnil
        & \xchkintro
          &   \AllIntro & \ImpliesIntro/\ImpliesIntroBot & \SubOrExIntro &    \WithIntro
                                  &    \NORULE
                                  &    \NORULE &    \NORULE & \NORULE
                                  &    \NORULE
                                  &    \NORULE
                    & \Nil
        \\
        & \veccons{e_1}{e_2}
        & \xchkintro
          &   \AllIntro & \ImpliesIntro/\ImpliesIntroBot & \SubOrExIntro &    \WithIntro
                                  &    \NORULE
                                  &    \NORULE &    \NORULE & \NORULE
                                  &    \NORULE
                                  &    \NORULE
                     & \Cons
        \\
        & \case{e_0}{\Pi}
        & \TABLENOTE{Note 3}
          &   \Case & \Case & \Case &    \Case
                                  &    \Case
                                  &    \Case &    \Case & \Case
                                  &    \Case
                                  &    \Case
                      &    \Case
        \\
        & x
        &
          &   \Sub & \Sub & \Sub &    \Sub
                                  &    \Sub
                                  &    \Sub &    \Sub & \Sub
                                  &    \Sub
                                  &    \Sub
                      &    \Sub
        \\
        & (e_0 : A)
        &
          &   \Sub & \Sub & \Sub &    \Sub
                                  &    \Sub
                                  &    \Sub &    \Sub & \Sub
                                  &    \Sub
                                  &    \Sub
                       &     \Sub
        \\
        & e_1\,s
        &
          &   \Sub & \Sub & \Sub &    \Sub
                                  &    \Sub
                                  &    \Sub &    \Sub & \Sub
                                  &    \Sub
                                  &    \Sub
                                  &    \Sub
         \end{array}
      \]
      Notes:
      \begin{itemize}
          \item \TABLENOTE{Note 1:}
                 The choice between \ImpliesIntro and \ImpliesIntroBot is
                 resolved by \Lemmaref{lem:aux-det} (5).   %

          \item \TABLENOTE{Note 2:}
                Fixed points are a checking form, but not an introduction form.
                So if $e$ is $\rec{x} v$, we need not choose between an
                introduction rule for a large connective and the \Rec rule:\ only the
                \Rec rule is viable.  Large connectives must, therefore, be introduced
                \emph{inside} the typing of the body $v$.

          \item \TABLENOTE{Note 3:}
                Case expressions are a checking form, but not an introduction form.
                So if $e$ is a case expression, we need not choose between an
                introduction rule for a large connective and the \Case rule:\ only the
                \Case rule is viable.  Large connectives must, therefore, be introduced
                \emph{inside} the branches.               
      \end{itemize}
  
  \paragraph{Proof of Part (2) (synthesis).}
      Only four rules have a synthesis judgment in the conclusion:
      \Var, \Anno and \ArrElim%
      Rule \Var applies if and only if $e$ has the form $x$.
      Rule \Anno applies if and only if $e$ has the form $(e_0 : A)$.

      Otherwise, the judgment can be derived only if $e$ has the form
      $e_1\,e_2$, by \ArrElim. %

  \paragraph{Proof of Part (3) (spine judgments).}
     For the ordinary spine judgment,
     rule \EmptySpine applies if and only if the given spine is empty.
     Otherwise, the choice of rule is determined by the head constructor of the
     input type: $\arr$/\ArrSpine; $\AllSym$/\AllSpine; $\implies$/\ImpliesSpine; $\ahat$/\SolveSpine.

     For the principality-recovering spine judgment:
     If $p = \OK$, only rule \Pass applies.  If $p = \p$ and $q = \p$, only rule \Pass applies.
     If $p = \p$ and $q = \OK$, then the rule is determined by $\FEV{C}$:
     if $\FEV{C} = \emptyset$ then only \Recover applies; otherwise, $\FEV{C} \neq \emptyset$
     and only \Pass applies.

  \paragraph{Proof of Part (4) (matching).}
     First, the elimination judgment form $\Gamma\;/\;P \entails \dots$:
     It cannot be the case that both
     $\elimeq{\Gamma}{\sigma}{t}{\sort}{\bot}$
     and
     $\elimeq{\Gamma}{\sigma}{t}{\sort}{\Theta}$,
     so either \MatchBot concludes both $\Dee_1$ and $\Dee_2$ (and the result
     follows), or \MatchUnify concludes both $\Dee_1$ and $\Dee_2$ (in which case,
     apply the i.h.).
     
     Now the main judgment form, without ``$/\;P$'':
     either $\Pi$ is empty, or has length one, or has length greater than one.
     \MatchEmpty applies if and only if $\Pi$ is empty,
     and \MatchSeq applies if and only if $\Pi$ has length greater than one.
     So in the rest of this part, we assume $\Pi$ has length one.
     
     Moreover, \MatchBase applies if and only if $\Avec$ has length zero.
     So in the rest of this part, we assume the length of $\Avec$ is at least one.

     Let $A$ be the first type in $\Avec$.  Inspection of the rules shows that
     given particular $A$ and $\pat$, where $\pat$ is the first pattern,
     only a single rule can apply, or no rule (``$\NORULE$'') can apply,
     as shown in the following table:     
     \[
         \begin{array}[b]{rlllllllll}
         &\multicolumn{6}{c}{\MAINLABEL{A}}
        \\[5pt]
               &&    \exists   &  \with & +  &  \times   &   \vecop  &  \text{other}
        \\\cmidrule[1pt]{3-8}
        & \inj{k}{\pat_0} &   \MatchExists & \MatchWith & \MatchSum{k} & \NORULE & \NORULE & \NORULE
        \\
 \MAINLABEL{\pat}          & \pair{\pat_1}{\pat_2} & \MatchExists & \MatchWith & \NORULE & \MatchPair & \NORULE & \NORULE
        \\
              &  z & \MatchExists & \MatchWith & \MatchNeg & \MatchNeg & \MatchNeg & \MatchNeg
        \\
              & \wild & \MatchExists & \MatchWith & \MatchWild & \MatchWild & \MatchWild & \MatchWild
        \\
              & \vecnil & \MatchExists & \MatchWith & \NORULE & \NORULE & \MatchNil & \NORULE
        \\
              & \veccons{\rho_1}{\rho_2} & \MatchExists & \MatchWith & \NORULE & \NORULE & \MatchCons & \NORULE
         \end{array}
    \qedhere
    \]
\end{proof}

\section{Soundness}

\subsection{Instantiation}

\instantiationsoundness*
\begin{proof}\proofheading  By induction on the derivation of $\instjudg{\Gamma}{\ahat}{\tau}{\sort}{\Delta}$.

      \begin{itemize}
        \DerivationProofCase{\InstSolve}
                { \judge{\Gamma_0}{\tau}{\sort} }
                { \instjudg{\Gamma_0, \ahat:\sort, \Gamma_1}
                            {\ahat}
                            {\tau}
                            {\sort}
                            {\underbrace{\Gamma_0, \hypeq{\ahat:\sort}{\tau}, \Gamma_1}_{\Delta}}
                 }

                 \begin{llproof}
                  \eqPf{[\Delta]\ahat}{[\Delta]\tau}{By definition}

                 \Hand \eqPf{[\Omega]\ahat}{[\Omega]\tau}    {By \Lemmaref{lem:substitution-monotonicity} to each side}
                 \end{llproof}
        
        \DerivationProofCase{\InstReach}
                {\bhat \in \unsolved{\Gamma[\ahat:\sort][\bhat:\sort]} }
                {\instjudg{\Gamma[\ahat:\sort][\bhat:\sort]}
                            {\ahat}
                            {\underbrace{\bhat}_{\tau}}
                            {\sort}
                            {\underbrace{\Gamma[\ahat:\sort][\hypeq{\bhat:\sort}{\ahat}]}_\Delta}
                }

             \begin{llproof}
               \eqPf{[\Delta]\bhat}{[\Delta]\ahat}   {By definition}
               \eqPf{[\Omega][\Delta]\bhat}{[\Omega][\Delta]\ahat}   {Applying $\Omega$ to each side}
             \Hand   \eqPf{[\Omega]\underbrace{\bhat}_{\tau}}{[\Omega]\ahat}    {By \Lemmaref{lem:substitution-monotonicity} to each side}
             \end{llproof}

        \DerivationProofCase{\InstBin}
                {\instjudg{\overbrace{\Gamma_0[\ahat_2:\type, \ahat_1:\type, \hypeq{\ahat:\type}{\ahat_1 \binc \ahat_2}]}^{\Gamma'}}
                            {\ahat_1}
                            {\tau_1}
                            {\type}
                            {\Theta} \\
                 \instjudg{\Theta}
                            {\ahat_2}
                            {[\Theta]\tau_2}
                            {\type}
                            {\Delta}}
                {\instjudg{\Gamma_0[\ahat:\type]}
                            {\ahat}
                            {\tau_1 \binc \tau_2}
                            {\type}
                            {\Delta}}

            \begin{llproof}
               \substextendPf{\Delta}{\Omega}   {Given}
               \instjudgPf{\Gamma'}
                            {\ahat_1}
                            {\tau_1}
                            {\type}
                            {\Theta}   {Subderivation}
               \substextendPf{\Theta}{\Delta}   {By \Lemmaref{lem:instantiation-extension}}
               \substextendPf{\Theta}{\Omega}   {By \Lemmaref{lem:extension-transitivity}}
               \eqPf{[\Omega]\ahat_1}{[\Omega]\tau_1}   {By i.h.}
               \proofsep
               \instjudgPf{\Theta}
                            {\ahat_2}
                            {[\Theta]\tau_2}
                            {\type}
                            {\Delta}   {Subderivation}
               \eqPf{[\Omega]\ahat_2}{[\Omega][\Theta]\tau_2}   {By i.h.}
               \continueeqPf{[\Omega]\tau_2}   {By \Lemmaref{lem:substitution-monotonicity}}
               \proofsep
               \eqPf{([\Omega]\tau_1) \binc ([\Omega]\tau_2)}
                    {([\Omega]\ahat_1) \binc ([\Omega]\ahat_2)}    {By above equalities}
               \continueeqPf{[\Omega](\ahat_1 \binc \ahat_2)}{By definition of substitution}
               \continueeqPf{[\Omega]([\Gamma']\ahat)}   {By definition of substitution}
               \continueeqPf{[\Omega]\ahat}   {By \Lemmaref{lem:substitution-monotonicity}}
               \proofsep
            \Hand \eqPf{[\Omega]\underbrace{(\tau_1 \binc \tau_2)}_{\tau}}{[\Omega]\ahat}   {By definition of substitution}
            \end{llproof}
        
        \DerivationProofCase{\InstZero}
              { }
              {\instjudg{\Gamma_0[\ahat:\ind]}{\ahat}{\zero}{\ind}{\Gamma_0[\hypeq{\ahat:\ind}{\zero}]}}
              
              Similar to the \InstSolve case.
        
        \DerivationProofCase{\InstSucc}
                {\instjudg{\Gamma_0[\ahat_1:\ind, \hypeq{\ahat:\ind}{\succ{\ahat_1}}]}
                            {\ahat_1}
                            {t_1}
                            {\ind}
                            {\Delta} 
                }
                {\instjudg{\Gamma_0[\ahat:\ind]}
                            {\ahat}
                            {\succ{t_1}}
                            {\ind}
                            {\Delta}}

              Similar to the \InstBin case, but simpler.
      \qedhere
      \end{itemize}    
\end{proof}

\checkeqsoundness*
\begin{proof}\proofheading  By induction on the given derivation.

    \begin{itemize}  
      \DerivationProofCase{\CheckeqVar}
           { }
           {\checkeq{\Gamma}{u}{u}{\sort}{\Gamma}}

           \begin{llproof}
           \Hand  \eqPf{[\Omega]u}{[\Omega]u}  {By reflexivity of equality}
           \end{llproof}

      \ProofCasesRules{\CheckeqUnit, \CheckeqZero}
           Similar to the \CheckeqVar case.

      \DerivationProofCase{\CheckeqSucc}
            {\checkeq{\Gamma}{\sigma_0}{t_0}{\ind}{\Delta}}
            {\checkeq{\Gamma}{\succ{\sigma_0}}{\succ{t_0}}{\ind}{\Delta}}

            \begin{llproof}
              \checkeqPf{\Gamma}{\sigma_0}{t_0}{\ind}{\Delta}   {Subderivation}
              \eqPf{[\Omega]\sigma_0}{[\Omega]t_0}   {By i.h.}
              \eqPf{\succ{[\Omega]\sigma_0}}{\succ{[\Omega]t_0}}   {By congruence}
            \Hand  \eqPf{[\Omega](\succ{\sigma_0})}{[\Omega](\succ{t_0})}   {By definition of substitution}
            \end{llproof}

      \DerivationProofCase{\CheckeqBin}
            {\checkeq{\Gamma}{\sigma_0}{t_0}{\type}{\Theta} \\
             \checkeq{\Theta}{[\Theta]\sigma_1}{[\Theta]t_1}{\type}{\Delta}
            }
            {\checkeq{\Gamma}{\sigma_0 \binc \sigma_1}{t_0 \binc t_1}{\type}{\Delta}}

            \begin{llproof}
              \checkeqPf{\Gamma}{\sigma_0}{t_0}{\ind}{\Delta}   {Subderivation}
              \checkeqPf{\Theta}{[\Theta]\sigma_1}{[\Theta]t_1}{\type}{\Delta}{Subderivation}
              \substextendPf{\Delta}{\Omega}{Given}
              \substextendPf{\Theta}{\Delta}{By \Lemmaref{lem:checkeq-extension}}
              \substextendPf{\Theta}{\Omega}{By \Lemmaref{lem:extension-transitivity}}
              \eqPf{[\Omega]\sigma_0}{[\Omega]t_0}   {By i.h.\ on first subderivation}
              \eqPf{[\Omega][\Theta]\sigma_1}{[\Omega][\Theta]t_1}{By i.h.\ on second subderivation}
              \eqPf{[\Omega][\Theta]\sigma_1}{[\Omega]\sigma_1}{By \Lemmaref{lem:substitution-monotonicity}}
              \eqPf{[\Omega][\Theta]t_1}{[\Omega]t_1}{By \Lemmaref{lem:substitution-monotonicity}}
              \eqPf{[\Omega]\sigma_1}{[\Omega]t_1}{By transitivity of equality}
              \eqPf{[\Omega]\sigma_0 \binc [\Omega]\sigma_1}{[\Omega]t_0 \binc [\Omega]t_1}{By congruence of equality}
              \Hand \eqPf{[\Omega](\sigma_0 \binc \sigma_1)}{[\Omega](t_0 \binc t_1)}{By definition of substitution}

            \end{llproof}

      \DerivationProofCase{\CheckeqInstL}
           {\instjudg{\Gamma[\ahat]}{\ahat}{t}{\sort}{\Delta}
            \\
            \ahat \notin \FV{t}
           }
           {\checkeq{\Gamma[\ahat]}{\ahat}{t}{\sort}{\Delta}}

           \begin{llproof}
             \instjudgPf{\Gamma[\ahat]}{\ahat}{t}{\sort}{\Delta}   {Subderivation}
             \notinPf{\ahat}{\FV{t}}   {Premise}
           \Hand  \eqPf{[\Omega]\ahat}{[\Omega]t}   {By \Lemmaref{lem:instantiation-soundness}}
           \end{llproof}

      \DerivationProofCase{\CheckeqInstR}
           {\instjudg{\Gamma[\ahat:\sort]}{\ahat}{\sigma}{\sort}{\Delta}
            \\
            \ahat \notin \FV{t}
           }
           {\checkeq{\Gamma[\ahat:\sort]}{\sigma}{\ahat}{\sort}{\Delta}}

           Similar to the \CheckeqInstL case.
   \qedhere
   \end{itemize}
\end{proof}

\propequivsoundness*
\begin{proof}\proofheading  By induction on the given derivation.
  
     \begin{itemize}
         \DerivationProofCase{\PropequivEq}
             {\checkeq{\Gamma}{\sigma_1}{t_1}{\ind}{\Theta}
              \\
              \checkeq{\Theta}{[\Theta]\sigma_2}{[\Theta]t_2}{\ind}{\Delta}}
             {\propequivjudg{\Gamma}{(\sigma_1 = \sigma_2)}{(t_1 = t_2)}{\Delta}}

        \begin{llproof}
          \substextendPf{\Delta}{\Omega}{Given}
          \substextendPf{\Theta}{\Delta}{By \Lemmaref{lem:checkeq-extension} (on 2nd premise)}
          \substextendPf{\Theta}{\Omega}{By \Lemmaref{lem:extension-transitivity}}
          
          \checkeqPf{\Gamma}{\sigma_1}{t_1}{\ind}{\Theta}{Given}
          \eqPf{[\Omega]\sigma_1}{[\Omega]t_1}{By \Lemmaref{lem:checkeq-soundness}}

          \checkeqPf{\Theta}{[\Theta]\sigma_2}{[\Theta]t_2}{\ind}{\Delta}{Given}
          \eqPf{[\Omega][\Theta]\sigma_2}{[\Omega][\Theta]t_2}{By \Lemmaref{lem:checkeq-soundness}}

          \eqPf{[\Omega][\Theta]\sigma_2}{[\Omega]\sigma_2}{By \Lemmaref{lem:substitution-monotonicity}}
          \eqPf{[\Omega][\Theta]t_2}{[\Omega]t_2}{By \Lemmaref{lem:substitution-monotonicity}}
          \eqPf{[\Omega]\sigma_2}{[\Omega]t_2}{By transitivity of equality}
          \eqPfParen{[\Omega]\sigma_1 = [\Omega]\sigma_2}{[\Omega]t_1 = [\Omega]t_2}{By congruence of equality}

          \Hand \eqPf{[\Omega](\sigma_1 = \sigma_2)}{[\Omega](t_1 = t_2)}{By definition of substitution \qedhere}
        \end{llproof}
     \end{itemize}
\end{proof}

\equivsoundness*
\begin{proof}\proofheading  By induction on the given derivation.

     \begin{itemize}
        \DerivationProofCase{\EquivVar}
             { }
             {\equivjudg{\Gamma}{\alpha}{\alpha}{\Gamma}}

             \begin{llproof}
             \Hand  \eqPf{[\Omega]\alpha}{[\Omega]\alpha}   {By reflexivity of equality}
             \end{llproof}

        \ProofCasesRules{\EquivExvar, \EquivUnit}
            Similar to the \EquivVar case.

        \DerivationProofCase{\EquivBin}
             {\equivjudg{\Gamma}{A_1}{B_1}{\Theta}
               \\
               \equivjudg{\Theta}{[\Theta]A_2}{[\Theta]B_2}{\Delta}
             }
             {\equivjudg{\Gamma}{A_1 \binc A_2}{B_1 \binc B_2}{\Delta}}

             \begin{llproof}
               \substextendPf{\Delta}{\Omega}   {Given}
               \equivjudgPf{\Theta}{[\Theta]A_2}{[\Theta]B_2}{\Delta}  {Subderivation}
               \substextendPf{\Theta}{\Delta}   {By \Lemmaref{lem:equiv-extension}}
               \substextendPf{\Theta}{\Omega}   {By \Lemmaref{lem:extension-transitivity}}
               \proofsep
               \equivjudgPf{\Gamma}{A_1}{B_1}{\Theta}  {Subderivation}
               \eqPf{[\Omega]A_1}{[\Omega]B_1}  {By i.h.}
               \proofsep
               \substextendPf{\Delta}{\Omega}   {Given}
               \eqPf{[\Omega][\Theta]A_2}{[\Omega][\Theta]B_2}  {By i.h.}
               \eqPf{[\Omega]A_2}{[\Omega]B_2}  {By \Lemmaref{lem:substitution-monotonicity}}
               \decolumnizePf
\Hand               \eqPf{([\Omega]A_1) \binc ([\Omega]A_2)} {([\Omega]B_1) \binc ([\Omega]B_2)}  {By above equations}
             \end{llproof}

        \DerivationProofCase{\EquivAll}
             {\equivjudg{\Gamma, \alpha:\sort}{A_0}{B_0}{\Delta, \alpha:\sort, \Delta'}}
             {\equivjudg{\Gamma}{\alltype{\alpha:\sort} A_0}{\alltype{\alpha:\sort} B_0}{\Delta}}
             
             \begin{llproof}
               \equivjudgPf{\Gamma, \alpha:\sort}{A_0}{B_0}{\Delta, \alpha:\sort, \Delta'}  {Subderivation}
               \substextendPf{\Delta}{\Omega}   {Given}
               \substextendPf{\Gamma, \alpha:\sort, \cdot}{\Delta, \alpha:\sort, \Delta'}{By \Lemmaref{lem:equiv-extension}}
               \decolumnizePf
               \softPf{\Delta'}{Since $\cdot$ is soft}
               \substextendPf{\Delta, \alpha:\sort, \Delta'}{\Omega, \alpha:\sort, \Omega_Z}   {By \Lemmaref{lem:soft-extension}}
               \judgetpPf{\Gamma, \alpha:\sort}{A_0}{By validity on subderivation}
               \judgetpPf{\Gamma, \alpha:\sort}{B_0}{By validity on subderivation}
               \subseteqPf{\FV{A_0}}{\dom{\Gamma, \alpha:\sort}}{By well-typing of $A_0$}
               \subseteqPf{\FV{B_0}}{\dom{\Gamma, \alpha:\sort}}{By well-typing of $B_0$}
               \substextendPf{\Gamma, \alpha:\sort}{\Omega, \alpha:\sort}{By \substextendUU}
               \subseteqPf{\FV{A_0}}{\dom{\Omega, \alpha:\sort}}{By \Lemmaref{lem:declaration-order-preservation}}
               \subseteqPf{\FV{B_0}}{\dom{\Omega, \alpha:\sort}}{By \Lemmaref{lem:declaration-order-preservation}}
               \eqPf{[\Omega, \alpha:\sort, \Omega_Z]A_0}{[\Omega, \alpha:\sort]A_0} {By definition of substitution, since $\FV{A_0} \sect \dom{\Omega_Z} = \emptyset$}
               \eqPf{[\Omega, \alpha:\sort, \Omega_Z]B_0}{[\Omega, \alpha:\sort]B_0} {By definition of substitution, since $\FV{B_0} \sect \dom{\Omega_Z} = \emptyset$}
               \eqPf{[\Omega, \alpha:\sort]A_0}{[\Omega, \alpha:\sort]B_0}  {By transitivity of equality}
               \eqPf{[\Omega]A_0}{[\Omega]B_0}  {From definition of substitution}
               \eqPf{\alltype{\alpha:\sort} [\Omega]A_0}{\alltype{\alpha:\sort} [\Omega]B_0}  {Adding quantifier to each side}
               \eqPf{[\Omega](\alltype{\alpha:\sort} A_0)}{[\Omega](\alltype{\alpha:\sort} B_0)}  {By definition of subsitution}
             \end{llproof}
        
        \DerivationProofCase{\EquivImplies}
             {\propequivjudg{\Gamma}{P}{Q}{\Theta}
              \\
              \equivjudg{\Theta}{[\Theta]A_0}{[\Theta]B_0}{\Delta}}
             {\equivjudg{\Gamma}{P \implies A_0}{Q \implies B_0}{\Delta}}

             \begin{llproof}
               \substextendPf{\Delta}{\Omega}   {Given}
               \equivjudgPf{\Theta}{[\Theta]A_0}{[\Theta]B_0}{\Delta}   {Subderivation}
               \substextendPf{\Theta}{\Delta}   {By \Lemmaref{lem:equiv-extension}}
               \substextendPf{\Theta}{\Omega}   {By \Lemmaref{lem:extension-transitivity}}
               \propequivjudgPf{\Gamma}{P}{Q}{\Theta}   {Subderivation}
               \eqPf{[\Omega]P}{[\Omega]Q}   {By \Lemmaref{lem:propequiv-soundness}}
               \proofsep
               \equivjudgPf{\Theta}{[\Theta]A_0}{[\Theta]B_0}{\Delta}   {Subderivation}
               \eqPf{[\Omega][\Theta]A_0}{[\Omega][\Theta]B_0}    {By i.h.}
               \eqPf{[\Omega]A_0}{[\Omega]B_0}    {By \Lemmaref{lem:substitution-monotonicity}}
             \end{llproof}
        
        \DerivationProofCase{\EquivWith}
             {\propequivjudg{\Gamma}{P}{Q}{\Theta}
              \\
              \equivjudg{\Theta}{[\Theta]A_0}{[\Theta]B_0}{\Delta}}
             {\equivjudg{\Gamma}{A_0 \with P}{B_0 \with Q}{\Delta}}

             Similar to the \EquivImplies case.

        \DerivationProofCase{\EquivInstL}
            {\instjudg{\Gamma[\ahat]}{\ahat}{\tau}{\type}{\Delta}
             \\
             \ahat \notin \FV{\tau}
            }
            {\equivjudg{\Gamma[\ahat]}{\ahat}{\underbrace{\tau}_{A}}{\Delta}}

            \begin{llproof}
              \instjudgPf{\Gamma[\ahat]}{\ahat}{\tau}{\type}{\Delta}  {Subderivation}
            \Hand  \eqPf{[\Omega]\ahat}{[\Omega]\tau}{By \Lemmaref{lem:instantiation-soundness}}
            \end{llproof}
        
        \ProofCaseRule{\EquivInstR}
            Similar to the \EquivInstL case.
        \qedhere
     \end{itemize}
  
\end{proof}

\subsection{Soundness of Checkprop}

\checkpropsoundness*
\begin{proof}\proofheading
  By induction on the derivation of $\checkprop{\Gamma}{P}{\Delta}$.

  \begin{itemize}
    \DerivationProofCase{\CheckpropEq}
          {\checkeq{\Gamma}{\sigma}{t}{\ind}{\Delta} }
          {\checkprop{\Gamma}{\underbrace{\sigma = t}_{P}}{\Delta}}

          \begin{llproof}
            \checkeqPf{\Gamma}{\sigma}{t}{\ind}{\Delta}   {Subderivation}
            \eqPf{[\Omega]\sigma}{[\Omega]t}   {By \Lemmaref{lem:checkeq-soundness}}
            \declcheckpropPf{\Psi}{[\Omega]\sigma = [\Omega]t}   {By \DeclCheckpropEq}
            \declcheckpropPf{\Psi}{[\Omega](\sigma = t)}   {By def.\ of subst.}
\Hand       \declcheckpropPf{\Psi}{[\Omega]P}   {By $P = (\sigma = t)$}
          \end{llproof}
  \qedhere
  \end{itemize}
\end{proof}

\subsection{Soundness of Eliminations (Equality and Proposition)}

\elimeqsoundness*
\begin{proof}\proofheading
  First, we need to recall a few properties of term unification.
  \begin{enumerate}[(i)]
  \item If $\sigma$ is a term, then $\mgu{\sigma}{\sigma} = \id$. 
  \item If $f$ is a unary constructor, then $\mgu{f(\sigma)}{f(t)} = \mgu{\sigma}{t}$, supposing that $\mgu{\sigma}{t}$ exists. 
  \item If $f$ is a binary constructor,
    and $\sigma = \mgu{f(\sigma_1, \sigma_2)}{f(t_1, t_2)}$
    and $\sigma_1 = \mgu{\sigma_1}{t_1}$ \\
    and $\sigma_2 = \mgu{[\sigma_1]\sigma_2}{[\sigma_1]t_2}$,
    then
    $\sigma = \sigma_2 \circ \sigma_1 = \sigma_1 \circ \sigma_2$. 
  \item If $\alpha \notin \FV{t}$, then $\mgu{\alpha}{t} = (\hypeq{\alpha}{t})$.
  \item If $f$ is an $n$-ary constructor, and $\sigma_i$ and $t_i$ (for $i \leq n$) have no unifier, then $f(\sigma_1, \ldots, \sigma_n)$ and $f(t_1, \ldots, t_n)$ have no unifier. 
  \end{enumerate}

  We proceed by induction on the derivation of $\elimeq{\Gamma}{\sigma}{t}{\sort}{\Deltabot}$, proving
  both parts with a single induction. 

  \begin{itemize}
      \DerivationProofCase{\ElimeqUvarRefl}
            { }
            {\elimeq{\Gamma}{\alpha}{\alpha}{\sort}{\Gamma}}

            Here we have $\Delta = \Gamma$, so we are in part (1).  \\
            Let $\theta = \id$ (which is $\mgu{\sigma}{\sigma}$). \\
            We can easily show $[\id][\Omega]\alpha = [\Omega, \alpha] = [\Omega, \cdot]\alpha$.

      \DerivationProofCase{\ElimeqZero}
            { }
            {\elimeq{\Gamma}{\zero}{\zero}{\ind}{\Gamma}}
            Similar to the \ElimeqUvarRefl case. 

      \DerivationProofCase{\ElimeqSucc}
            {\elimeq{\Gamma}{t_1}{t_2}{\ind}{\Deltabot}}
            {\elimeq{\Gamma}{\succ{t_1}}{\succ{t_2}}{\ind}{\Deltabot}}

            We distinguish two subcases:

            \begin{itemize}
            \ProofCaseRule{$\Deltabot = \Delta$}
            
                 Since we have the same output context in the conclusion and premise,
                 the ``for all $t'\dots$'' part follows immediately from the i.h.\ (1). \\
                 The i.h.\ also gives us $\theta_0 = \mgu{t_1}{t_2}$. \\
                 Let $\theta = \theta_0$.  By property (ii), $\mgu{t_1}{t_2} = \mgu{\succ{t_1}}{\succ{t_2}}
                   = \theta$.

            \ProofCaseRule{$\Deltabot = \bot$}

                 \begin{llproof}
                   \elimeqPf{\Gamma}{t_1}{t_2}{\ind}{\bot}{Subderivation}
                   \eqPf{\mgu{t_1}{t_2}}{\bot}{By i.h.\ (2)}
                   \Hand \eqPf{\mgu{\succ{t_1}}{\succ{t_2}}}{\bot}{By contrapositive of property (ii)}
                 \end{llproof}
            \end{itemize}

    \DerivationProofCase{\ElimeqUvarL}
          { \alpha \notin \FV{t}
            \\
            (\hypeq{\alpha}{-}) \notin \Gamma
          }
          {\elimeq{\Gamma}{\alpha}{t}{\sort}{\Gamma, \hypeq{\alpha}{t}}}

          Here $\Delta \neq \bot$, so we are in part (1).

          \begin{llproof}
            \eqPf{[\Omega, \hypeq{\alpha}{t}]t'}{\big[[\Omega]t/\alpha\big][\Omega]t'}
                {By a property of substitution}
            \continueeqPf{[\Omega][t/\alpha][\Omega]t'}
                {By a property of substitution}
            \continueeqPf{[\Omega][\theta][\Omega]t'}  {By $\mgu{\alpha}{t} = (\alpha/t)$}
          \Hand  \continueeqPf{[\theta][\Omega]t'}  {By a property of substitution ($\theta$ creates no evars)}
          \end{llproof}

    \DerivationProofCase{\ElimeqUvarR}
          { \alpha \notin \FV{t}
            \\
            (\hypeq{\alpha}{-}) \notin \Gamma
          }
          {\elimeq{\Gamma}{t}{\alpha}{\sort}{\Gamma, \hypeq{\alpha}{t}}}

          Similar to the \ElimeqUvarL case. 

    \DerivationProofCase{\ElimeqUnit}
          { }
          {\elimeq{\Gamma}{\unitty}{\unitty}{\type}{\Gamma}}
    
          Similar to the \ElimeqUvarRefl case. 

    \DerivationProofCase{\ElimeqBin}
          {\elimeq{\Gamma}{\tau_1}{\tau'_1}{\type}{\Theta}
           \\ 
          \elimeq{\Theta}{[\Theta]\tau_1}{[\Theta]\tau'_2}{\type}{\Deltabot}
          }
          {\elimeq{\Gamma}{\tau_1 \binc \tau_2}{\tau'_1 \binc \tau'_2}{\type}{\Deltabot}}

          \medskip

          Either $\Deltabot$ is some $\Delta$, or it is $\bot$.

          \begin{itemize}
              \ProofCaseRule{$\Deltabot = \Delta$}

                  \begin{llproof}
                    \elimeqPf{\Gamma}{\tau_1}{\tau'_1}{\type}{\Theta}  {Subderivation}
                    \eqPfParenR{\Theta}{\Gamma, \Delta_1}  {By i.h.\ (1)}
                    (IH-1st) \eqPf{[\Omega, \Delta_1]u_1}{[\theta_1][\Omega]u_1}  {\ditto~~for all $\judge{\Omega}{u_1}{\sort'}$}
                    \eqPf{\theta_1}{\mgu{\tau_1}{\tau'_1}}  {\ditto}
                    \proofsep
                    \elimeqPf{\Theta}{[\Theta]\tau_1}{[\Theta]\tau'_2}{\type}{\Delta} {Subderivation}
                    \eqPfParenR{\Delta}{\Theta, \Delta_2}  {By i.h.\ (1)}
                    (IH-2nd)~~~ \eqPf{[\Omega, \Delta_1, \Delta_2]u_2}{[\theta_2][\Omega, \Delta_1]u_2}  {\ditto~~for all $\judge{\Omega}{u_2}{\sort'}$}
                    \eqPf{\theta_2}{\mgu{\tau_2}{\tau'_2}}  {\ditto}
                  \end{llproof}

                  Suppose $\judge{\Omega}{u}{\sort'}$.
                  
                  \begin{llproof}
                    \eqPf{[\Omega, \Delta_1, \Delta_2]u} {[\theta_2][\Omega, \Delta_1]u}   {By (IH-2nd), with $u_2 = u$}
                    \continueeqPf {[\theta_2][\theta_1][\Omega]u}   {By (IH-1st), with $u_1 = u$}
                  \Hand  \continueeqPf {[\Omega][\theta_2 \circ \theta_1]u}   {By a property of substitution}
                    \proofsep
                  \Hand  \eqPf{\theta_2 \circ \theta_1}
                         {\mgu{(\tau_1 \binc \tau_2)}
                              {(\tau'_1 \binc \tau'_2)}
                         }
                         {By property (iii) of substitution}
                  \end{llproof}

              \ProofCaseRule{$\Deltabot = \bot$}

                 Use the i.h.\ (2) on the second premise to show $\mgu{\tau_2}{\tau_2'} = \bot$,
                 then use property (v) of unification to show $\mgu{(\tau_1 \binc \tau_2)}
                              {(\tau'_1 \binc \tau'_2)} = \bot$.

          \end{itemize}
    
    \DerivationProofCase{\ElimeqBinBot}
          {\elimeq{\Gamma}{\tau_1}{\tau'_1}{\type}{\bot} }
          {\elimeq{\Gamma}{\tau_1 \binc \tau_2}{\tau'_1 \binc \tau'_2}{\type}{\bot}}

          Similar to the $\bot$ subcase for \ElimeqSucc, but using property (v) instead of property (ii).

    \DerivationProofCase{\ElimeqClash}
          { \clash{\sigma}{t} }
          { \elimeq{\Gamma}{\sigma}{t}{\sort}{\bot} }

          Since $\clash{\sigma}{t}$, we know $\sigma$ and $t$ have different head constructors, and thus no unifier.
  \qedhere
  \end{itemize}
\end{proof}

\clearpage

\subtypingsoundness*
\begin{proof}\proofheading
  By induction on the given derivation.

  \begin{itemize}
      \DerivationProofCase{\SubAllL}
            {\notAll{B} \\
             \subjudg[-]{\Gamma, \MonnierComma{\ahat}, \ahat:\sort}
                     {[\ahat/\alpha]A_0}
                     {B}
                     {\Delta, \MonnierComma{\ahat}, \Theta}}
            {\subjudg[-]{\Gamma}{\alltype{\alpha:\sort} A_0}{B}{\Delta}}

            Let $\Omega' = (\Omega, \soln{\MonnierComma{\ahat}, \Theta})$.

            \begin{llproof}
                \subjudgPf{-}{\Gamma, \MonnierComma{\ahat}, \ahat:\sort}
                       {[\ahat/\alpha]A_0}
                       {B}
                       {\Delta, \MonnierComma{\ahat}, \Theta} 
                           {Subderivation}
                \proofsep
                \substextendPf{\Delta}{\Omega}   {Given}
                \substextendPf{(\Delta, \MonnierComma{\ahat}, \Theta)}{\Omega'}   {By \Lemmaref{lem:soln-completes}}
                \proofsep
                \judgetpPf{\Gamma}{\alltype{\alpha:\sort} A_0}   {Given}
                \judgetpPf{\Gamma, \alpha:\sort}{A_0}   {By inversion (\ForallWF)}
                \judgetpPf{\Gamma, \MonnierComma{\ahat}, \ahat : \sort}{[\ahat/\alpha]A_0}   {By a property of substitution}
                \judgetpPf{\Gamma}{B}   {Given}
                \proofsep
                \declsubjudgPf{-}{[\Omega'](\Delta, \MonnierComma{\ahat}, \Theta)}
                              {[\Omega'][\ahat/\alpha]A_0}
                              {[\Omega']B}
                           {By i.h.}
                \judgetpPf{\Omega}{B}  {By \Lemmaref{lem:extension-weakening}}
                \eqPf{[\Omega']B} {[\Omega]B}  {By \Lemmaref{lem:substitution-stability}}
                \declsubjudgPf{-}{[\Omega'](\Delta, \MonnierComma{\ahat}, \Theta)}
                              {[\Omega'][\ahat/\alpha]A_0}
                              {[\Omega]B}
                           {By above equality}
                \declsubjudgPf{-}{[\Omega'](\Delta, \MonnierComma{\ahat}, \Theta)}
                              {\big[[\Omega']\ahat/\alpha\big][\Omega']A_0}
                              {[\Omega]B}
                           {By distributivity of substitution}
                \proofsep
                
                \judgePf{\Gamma, \MonnierComma{\ahat}, \ahat:\sort} {\ahat} {\sort} {By \VarSort}
                \substextendPf{\Gamma, \MonnierComma{\ahat}, \ahat:\sort}  {\Delta, \MonnierComma{\ahat}, \Theta}  {By \Lemmaref{lem:subtyping-extension}}
                \Pf{}{}{\Theta\mbox{ is soft}}{By \Lemmaref{lem:extension-inversion} (ii)}
                \judgePf{\Delta, \MonnierComma{\ahat}, \Theta} {\ahat} {\sort}  {By \Lemmaref{lem:extension-weakening}}
                \substextendPf{(\Delta, \MonnierComma{\ahat}, \Theta)}{\Omega'}   {Above}
                \judgePf{[\Omega']\Omega'} {[\Omega']\ahat} {\sort}   {By \Lemmaref{lem:completion-sort}}
                \judgePf{[\Omega'](\Delta, \MonnierComma{\ahat}, \Theta)} {[\Omega']\ahat} {\sort}
                        {By \Lemmaref{lem:completes-stability}}
                \proofsep
                \declsubjudgPf{-}{[\Omega'](\Delta, \MonnierComma{\ahat}, \Theta)}
                              {\alltype{\alpha:\sort}{[\Omega']A_0}}
                              {[\Omega]B}   {By \DsubAllL}
                \declsubjudgPf{-}{[\Omega'](\Delta, \MonnierComma{\ahat}, \Theta)}
                              {\alltype{\alpha:\sort}{[\Omega, \alpha:\sort]A_0}}
                              {[\Omega]B}   {By \Lemmaref{lem:substitution-stability}}
                \declsubjudgPf{-}{[\Omega]\Delta}
                              {\alltype{\alpha:\sort}{[\Omega, \alpha:\sort]A_0}}
                              {[\Omega]B}   {\runonfontsz{9pt}By \Lemmaref{lem:context-partitioning} + thinning}
                \declsubjudgPf{-}{[\Omega]\Delta}
                              {\alltype{\alpha:\sort}{[\Omega]A_0}}
                              {[\Omega]B}   {By def.\ of substitution}
                \declsubjudgPf{-}{[\Omega]\Delta}
                              {[\Omega](\alltype{\alpha:\sort}{A_0})}
                              {[\Omega]B}   {By def.\ of substitution}
            \end{llproof}

      \ProofCaseRule{\SubExistsR}  Similar to the \SubAllL case.
  
      \DerivationProofCase{\SubAllR}
            {
             \subjudg[-]{\Gamma, \beta:\sort}{A}{B_0}{\Delta, \beta:\sort, \Theta}}
            {\subjudg[-]{\Gamma}{A}{\alltype{\beta:\sort} B_0}{\Delta}}

            \begin{llproof}
              \subjudgPf{-}{\Gamma, \beta:\sort}{A}{B_0}{\Delta, \beta:\sort, \Theta}  {Subderivation}
              \LetPf{\Omega_Z}{\soln{\Theta}}   {}
              \LetPf{\Omega'}{(\Omega, \beta:\sort, \Omega_Z)} {}
              \substextendPf{(\Delta, \beta:\sort, \Theta)}{\Omega'}   {By \Lemmaref{lem:soln-completes}}
              \proofsep
              \judgetpPf{\Gamma}{A}   {Given}
              \judgetpPf{\Gamma, \beta:\sort}{A}   {By \Lemmaref{lem:suffix-weakening-judgetp}}
              \judgetpPf{\Gamma}{\alltype{\beta:\sort} B_0}   {Given}
              \judgetpPf{\Gamma, \beta:\sort}{B_0}   {By inversion (\ForallWF)}
              \proofsep
              \declsubjudgPf{-}{[\Omega'](\Delta, \beta:\sort, \Theta)} {[\Omega']A} {[\Omega']B_0}  {By i.h.}
              \substextendPf{\Gamma, \beta:\sort}{\Delta, \beta:\sort, \Theta}{By \Lemmaref{lem:subtyping-extension}}
              \Pf{}{}{\Theta \mbox{ is soft}}{By \Lemmaref{lem:extension-inversion} (i)}
              \declsubjudgPf{-}{[\Omega, \beta:\sort](\Delta, \beta:\sort)} {[\Omega, \beta:\sort]A} {[\Omega, \beta:\sort]B_0}  {By \Lemmaref{lem:substitution-stability}}
              \declsubjudgPf{-}{[\Omega, \beta:\sort](\Delta, \beta:\sort)} {[\Omega]A} {[\Omega]B_0}  {By def.\ of substitution}
              \declsubjudgPf{-}{[\Omega]\Delta} {[\Omega]A} {\alltype{\beta:\sort}{[\Omega]B_0}}  {By \DsubAllR}
              \declsubjudgPf{-}{[\Omega]\Delta} {[\Omega]A} {[\Omega](\alltype{\beta:\sort}{B_0})}  {By def.\ of substitution}
            \end{llproof}

      \ProofCaseRule{\SubExistsL}  Similar to the \SubAllR case.

      \DerivationProofCase{\SubEquiv}
            {
              \equivjudg{\Gamma}{A}{B}{\Delta}
            }
            {\subjudg[\polvar]{\Gamma}{A}{B}{\Delta}}

            \begin{llproof}
              \equivjudgPf{\Gamma}{A}{B}{\Delta}  {Subderivation}
              \substextendPf{\Delta}{\Omega}   {Given}
              \eqPf{[\Omega]A} {[\Omega]B}  {By \Lemmaref{lem:equiv-soundness}}
              \substextendPf{\Gamma}{\Delta}  {By \Lemmaref{lem:equiv-extension}}
              \judgetpPf{\Gamma} {A}  {Given}
              \judgetpPf{[\Omega]\Omega} {[\Omega]A}  {By \Lemmaref{lem:completion-wf}}
              \judgetpPf{[\Omega]\Delta} {[\Omega]A}  {By \Lemmaref{lem:completes-stability}}
            \Hand  \declsubjudgPf{\polvar}{[\Omega]\Delta} {[\Omega]A} {[\Omega]B}  {By \DsubRefl{\polvar}}
            \end{llproof}

       \DerivationProofCase{\SubPosNegL}
             {\subjudg[-]{\Gamma}{A}{B}{\Delta}
               \\
               \arrayenvbl{\Neg{A}
               \\
               \nonPos{B}}}
           {\subjudg[+]{\Gamma}{A}{B}{\Delta}}

           \begin{llproof}
             \subjudgPf{-}{\Gamma}{A}{B}{\Delta}{By inversion}
             \SincePf{\Neg{A}}{By inversion}
             \SincePf{\nonPos{B}}{By inversion}
             \SincePf{\nonPos{A}}{since $\Neg{A}$}
             \declsubjudgPf{-}{[\Omega]\Gamma}{[\Omega]A}{[\Omega]B}{By induction}
             \Hand \declsubjudgPf{+}{[\Omega]\Gamma}{[\Omega]A}{[\Omega]B}{By \DsubPosNeg}
           \end{llproof}

       \DerivationProofCase{\SubPosNegR}
             {\subjudg[-]{\Gamma}{A}{B}{\Delta}
              \\
              \arrayenvbl{\nonPos{A}
              \\
              \Neg{B}}}
             {\subjudg[+]{\Gamma}{A}{B}{\Delta}}

       Similar to the \SubPosNegL case.

       \DerivationProofCase{\SubNegPosL}
             {\subjudg[+]{\Gamma}{A}{B}{\Delta}
              \\
              \arrayenvbl{\Pos{A}
              \\
              \nonNeg{B}}}
             {\subjudg[-]{\Gamma}{A}{B}{\Delta}}

       Similar to the \SubPosNegL case.

       \DerivationProofCase{\SubNegPosR}
             {\subjudg[+]{\Gamma}{A}{B}{\Delta}
              \\
              \arrayenvbl{\nonNeg{A}
              \\
              \Pos{B}}}
             {\subjudg[-]{\Gamma}{A}{B}{\Delta}}

       Similar to the \SubPosNegL case.

  \qedhere
  \end{itemize}
\end{proof}

\subsection{Soundness of Typing}

\coveragesoundness*
\begin{proof}\proofheading

  By mutual induction on the given algorithmic coverage derivation.

\begin{enumerate}
\item 
  \begin{itemize}
    \DerivationProofCase{\CoversEmpty}
          { }
          {\covers{\Gamma}{\branch{\cdot}{e_1} \alt \ldots}{\cdot}}

          \begin{llproof}
            \declcoversPf{[\Omega]\Gamma}{\branch{\cdot}{e_1} \alt \ldots}{\cdot}   {By \DeclCoversEmpty}
          \end{llproof}
    
    \ProofCasesRules{\CoversVar, \CoversUnit, \CoversTimes, \CoversSum,
                     \CoversEx, \CoversWith, \CoversVec, \CoversWithOK, \CoversVecOK}

        Use the i.h. and apply the corresponding declarative rule.
  \end{itemize}          

\item 
  \begin{itemize}
      \DerivationProofCase{\CoversEq}
          { \elimeq{\Gamma}{[\Gamma]t_1}{[\Gamma]t_2}{\sort}{\Delta}
            \\
            \covers{\Delta}{[\Delta]\Pi}{[\Delta]\Avec}
          }
          { \coverseq{\Gamma}{t_1 = t_2}{\Pi}{\Avec} 
          }

          \begin{llproof}
            \elimeqPf{\Gamma}{[\Gamma]t_1}{[\Gamma]t_2}{\sort}{\Delta}  {Subderivation}
            \proofsep
            \coversPf{\Delta}{[\Delta]\Pi}{[\Delta]\Avec}   {Subderivation}
            \declcoversPf{[\Omega]\Delta}{[\Delta]\Pi}{[\Delta]A_0, [\Delta]\Avec)}
                   {By i.h.}
            \proofsep
            \eqPfParenR{\Delta}{\Gamma, \Theta}    {\hspace{-12ex}By \Lemmaref{lem:elimeq-soundness} (1)}
            \eqPf{\mgu{t_1}{t_2}} {\theta}   {\ditto}
            \Pf{}{}{\dots}   {\ditto}
            \proofsep
            \eqPf{[\Omega]\Delta}{[\theta][\Omega]\Gamma}   {By \Lemmaref{lem:subst-upgrade} (iii)}
            \eqPf{[\Delta]\Pi}{[\theta]\Pi}   {By \Lemmaref{lem:subst-upgrade} (iv)}
            \eqPfParen{[\Delta]\Avec}{[\theta]A_0, [\theta]\Avec}   {By \Lemmaref{lem:subst-upgrade} (i)}
            \proofsep
            \declcoversPf{[\theta][\Omega]\Gamma}{[\theta]\Pi}{[\theta]\Avec}
                    {By above equalities}
            \proofsep
          \Hand
            \declcoverseqPf{[\Omega]\Gamma}{t_1 = t_2}{\Pi}{\Avec}   {By \DeclCoversEq}
          \end{llproof}
    
    \DerivationProofCase{\CoversEqBot}
          {\elimeq{\Gamma}{[\Gamma]t_1}{[\Gamma]t_2}{\sort}{\bot}}
          {\coverseq{\Gamma}{t_1 = t_2}{\Pi}{\Avec}}

          \begin{llproof}
            \elimeqPf{\Gamma}{[\Gamma]t_1}{[\Gamma]t_2}{\sort}{\bot}  {Subderivation}
            \decolumnizePf
            \eqPf{\mgu{[\Gamma]t_1}{[\Gamma]t_2}}{\bot}   {By \Lemmaref{lem:elimeq-soundness} (2)}
            \eqPf{\mgu{t_1}{t_2}}{\bot}   {By given equality}
            \decolumnizePf
          \Hand
            \declcoverseqPf{[\Omega]\Gamma}{t_1 = t_2}{\Pi}{\Avec}   {By \DeclCoversEqBot}
          \end{llproof}
  \qedhere
  \end{itemize}
\end{enumerate}
\end{proof}

\typingwf*
\begin{proof}\proofheading
  \begin{enumerate}
  \item Suppose $\synjudg{p}{\Gamma}{e}{A}{\Delta}$:

    \begin{itemize}
    \DerivationProofCase{\Var}
                        {(\hyp{p}{x}{A}) \in \Gamma}
                        {\synjudg{p}{\Gamma}{x}{[\Gamma]A}{\Gamma}}

      \begin{llproof}
        \eqPfParenR{\Gamma}{\Gamma_0, \hyp{p}{x}{A}, \Gamma_1}{$(\hyp{p}{x}{A}) \in \Gamma$}
        \judgetpPf{\Gamma}{A\;p}{Follows from $\judgectx{\Gamma}$}
      \end{llproof}

     \DerivationProofCase{\Anno}
          {\judgetp{\Gamma}{A \p}
            \\
            \chkjudg{\p}{\Gamma}{e}{[\Gamma]A}{\Delta}
          }
          {\synjudg{\p}{\Gamma}{(e : A)}{[\Delta]A}{\Delta}}

       \begin{llproof}
         \judgetpPf{\Gamma}{A\;\p}{By inversion}
         \substextendPf{\Gamma}{\Delta}  {By \Lemmaref{lem:typing-extension}}
         \judgetpPf{\Delta}{A\;\p}{By \Lemmaref{lem:extension-weakening-tpp}}
       \Hand  \judgetpPf{\Delta}{[\Delta]A\;\p}{By \Lemmaref{lem:substitution-tpp-stable} (i)}
       \end{llproof}
       
     \DerivationProofCase{\ArrElim}
           {\synjudg{p}{\Gamma}{e}{A}{\Theta}
             \\
             \spinejudg{\Theta}{s}{[\Theta]A}{p}{C}{q}{\Delta}
             \\
             \arrayenvbl{
               \text{$p = \OK$ or $q = \p$} \\
               \text{or $\FEV{[\Delta]C} \neq \emptyset$}
             }
           }
           {\synjudg{q}{\Gamma}{e\,s}{C}{\Delta}}

      \begin{llproof}
        \synjudgPf{p}{\Gamma}{e}{A}{\Theta}{By inversion}
        \judgetpPf{\Theta}{A\;p}{By induction}
        \judgetpPf{\Theta}{[\Theta]A\;p}{By \Lemmaref{lem:substitution-tpp}}
        \judgectxPf{\Theta}{By implicit assumption}
        \spinejudgPf{\Theta}{s}{[\Theta]A}{p}{C}{q}{\Delta}{By inversion}
        \Hand \judgetpPf{\Delta}{C\;q}{By mutual induction}
      \end{llproof}

    \end{itemize}

  \item Suppose $\spinejudg{\Gamma}{s}{A}{p}{B}{q}{\Delta}$ and $\judgetp{\Gamma}{A\;p}$:

    \begin{itemize}
     \DerivationProofCase{\EmptySpine}
         {}
         {\spinejudg{\Gamma}{\emptyspine}{A}{p}{A}{p}{\Gamma}}

         \begin{llproof}
           \judgetpPf{\Gamma}{A\;p}{Given}
         \end{llproof}

     \DerivationProofCase{\ArrSpine}
         {\chkjudg{p}{\Gamma}{e}{A}{\Theta}
          \\
          \spinejudg{\Theta}{s}{[\Theta]B}{p}{C}{q}{\Delta}
         }
         {\spinejudg{\Gamma}{\appspine{e}{s}}{A \arr B}{p}{C}{q}{\Delta}}

         \begin{llproof}
           \judgetpPf{\Gamma}{A\arr B \;p}{Given}
           \judgetpPf{\Gamma}{B\;p}{By \Lemmaref{lem:principal-inversion}}
           \judgetpPf{\Theta}{B\;p}{By \Lemmaref{lem:extension-weakening-tpp}}
           \judgetpPf{\Theta}{[\Theta]B\;p}{By \Lemmaref{lem:substitution-tpp}}
           \judgetpPf{\Delta}{C\;q}{By induction}
         \end{llproof}

     \DerivationProofCase{\AllSpine}
         {\spinejudg{\Gamma, \ahat:\sort}{\appspine{e}{s}}{[\ahat/\alpha]A}{\ok}{C}{q}{\Delta}}
         {\spinejudg{\Gamma}{\appspine{e}{s}}{\alltype{\alpha:\sort} A}{p}{C}{q}{\Delta}}

         \begin{llproof}
           \judgetpPf{\Gamma}{\alltype{\alpha:\sort} A\;p}{Given}
           \judgetpPf{\Gamma}{\alltype{\alpha:\sort} A}{By inversion}
           \judgetpPf{\Gamma, \alpha:\sort}{A}{By inversion}
           \judgetpPf{\Gamma, \ahat:\sort, \alpha:\sort}{A}{By weakening}
           \judgetpPf{\Gamma, \ahat:\sort}{[\ahat/\alpha]A}{By substitution}
           \Hand \judgetpPf{\Delta}{C\;q}{By induction}
         \end{llproof}

     \DerivationProofCase{\ImpliesSpine}
         {\checkprop{\Gamma}{P}{\Theta}
          \\
          \spinejudg{\Theta}{\appspine{e}{s}}{[\Theta]A}{p}{C}{q}{\Delta}}
         {\spinejudg{\Gamma}{\appspine{e}{s}}{P \implies A}{p}{C}{q}{\Delta}}

         \begin{llproof}
           \judgetpPf{\Gamma}{P \implies A\;p}{Given}
           \judgepropPf{\Gamma}{P}{By \Lemmaref{lem:principal-inversion}}
           \judgetpPf{\Gamma}{A\;p}{\ditto}
           \substextendPf{\Gamma}{\Theta}{By \Lemmaref{lem:checkprop-extension}}
           \judgetpPf{\Theta}{A\;p}{By \Lemmaref{lem:extension-weakening-tpp}}
           \judgetpPf{\Theta}{[\Theta]A\;p}{By \Lemmaref{lem:substitution-tpp}}
           \Hand \judgetpPf{\Delta}{C\;q}{By induction}
         \end{llproof}

      \DerivationProofCase{\SolveSpine}
            {\spinejudg{\overbrace{\Gamma[\ahat_2 : \type, \ahat_1 : \type, \hypeq{\ahat:\type}{\ahat_1 \arr \ahat_2}]}^{\Theta}}{\appspine{e}{s}}{(\ahat_1 \arr \ahat_2)}{\ok}{C}{\ok}{\Delta}}
            {\spinejudg{\Gamma[\ahat:\type]}{\appspine{e}{s}}{\ahat}{\ok}{C}{\ok}{\Delta}}

        \begin{llproof}
          \judgetpPf{\Theta}
                    {\ahat_1 \arr \ahat_2}
                    {By rules}
          \Hand \judgetpPf{\Delta}{C\;q}{By induction}
        \end{llproof}

    \end{itemize}
  \end{enumerate}
\end{proof}

\eagerness*
\begin{proof}\proofheading
  By induction on the given derivation.

  \textbf{Part (i), checking}
  
  \begin{itemize}
    \ProofCaseRule{\Rec}  By i.h.\ (i).

    \ProofCaseRule{\Sub}  By i.h.\ (ii) and (i).

    \ProofCaseRule{\AllIntro, \ExIntro}  By i.h.\ (i).

    \ProofCaseRule{\WithIntro}

    Substitution is idempotent, so in the last premise
    $[\Theta][\Theta]A_0 = [\Theta]A_0$
    and we can use the i.h.\ (i).

    \ProofCaseRule{\ImpliesIntro}    Similar to the \WithIntro case.

    \ProofCaseRule{\ImpliesIntroBot}
    This rule has no subderivations of the relevant form, so the case is trivial.

    \ProofCaseRule{\ArrIntro}
    By i.h.\ (i).

    \ProofCaseRule{\ArrIntroSolve}

    In the premise,
    $
    \big[
      \Gamma_0[\ahat_1{:}\type, \ahat_2{:}\type, \hypeq{\ahat{:}\type}{\ahat_1{\arr}\ahat_2}], \hyp{\ok}{x}{\ahat_1}\big]
    =
    \ahat_2
    $
    so we can use the i.h.\ (i).

    \ProofCaseRule{\SumIntro{k}}  By i.h.\ (i).

    \ProofCaseRule{\SumIntroSolve{k}}  Similar to the \ArrIntroSolve case.

    \ProofCaseRule{\PairIntro}

    By i.h.\ (i) on the first subderivation,
    then i.h.\ (i) on the second subderivation
    (using the fact that $[\Theta][\Theta]A_2 = [\Theta]A_2$).

    \ProofCaseRule{\PairIntroSolve}  Similar to the \ArrIntroSolve case.

    \ProofCaseRule{\Nil}     
    This rule has no subderivations of the relevant form, so the case is trivial.

    \ProofCaseRule{\Cons}

    By i.h.\ (i) on the subderivations typing $e_1$ and $e_2$, using 
    $[\Gamma'][\Gamma']A_0 = [\Gamma']A_0$
    and
    $[\Theta][\Theta](\vectype{\ahat}{A_0}) = [\Theta](\vectype{\ahat}{A_0})$.

    \DerivationProofCase{\Case}
           {
             \synjudg{q}{\Gamma}{e}{B}{\Theta}
             \\
             \arrayenvbl{
                \matchjudg{q}{p}{\Theta}{\Pi}{[\Theta]B}{[\Theta]A}{\Delta}
                \\
                \covers{q}{\Delta}{\Pi}{[\Delta]B}
             }
           }
           {\chkjudg{p}{\Gamma}{\case{e}{\Pi}}{A}{\Delta}}

           \begin{llproof}
             \derivesPf{\Dee_1} \synjudgPf{\p}{\Gamma}{e}{B}{\Theta}  {Subderivation}
             \eqPf{[\Theta]B}{B \AND \Dee_1~\text{eager}}   {By i.h.\ (ii)}
             \derivesPf{\Dee_2} \matchjudgPf{p}{\Theta}{\Pi}{[\Theta]B}{[\Theta]A}{\Delta}   {Subderivation}
             \Pf{}{}{\Dee_2~\text{eager}}   {By i.h.\ (v)}
           \end{llproof}

           By \Definitionref{def:eager}, the given derivation is eager.
  \end{itemize}

  \textbf{Part (ii), synthesis}

  \begin{itemize}
    \ProofCaseRule{\Var}
    Substitution is idempotent: $[\Gamma][\Gamma]A_0 = [\Gamma]A_0$.

    By inversion, $\Delta = \Gamma$ and $A = [\Gamma]A_0$ where $(x : A_0 p) \in \Gamma$.

    Using the above equations, we have
    
    \begin{llproof}
      \eqPf{[\Gamma][\Gamma]A_0}{[\Gamma]A_0} {}
      \eqPf{[\Gamma]A}{A} {}
      \eqPf{[\Delta]A}{A} {}
    \end{llproof}
    
    This rule has no subderivations, so there is nothing else to show.

    \ProofCaseRule{\Anno}
    By inversion, $A = [\Delta]A_0$.

    Substitution is idempotent, so $[\Gamma][\Gamma]A_0 = [\Gamma]A_0$
    and we can use the i.h.\ (i) to show that the checking subderivation is eager.

    The type in the conclusion is $[\Delta]A_0$, which by idempotence is equal to
    $[\Delta][\Delta]A_0$.  Since $A = [\Delta]A_0$, we have $A = [\Delta]A$.

    \DerivationProofCase{\ArrElim}
           {
               \synjudg{p}{\Gamma}{e}{B}{\Theta}
               \\
               \recspinejudg{\Theta}{s}{B}{p}{A}{q}{\Delta}
           }
           {\synjudg{q}{\Gamma}{e\,s}{A}{\Delta}}

           \begin{llproof}
               \derivesPf{\Dee_1}
               \synjudgPf{p}{\Gamma}{e}{B}{\Theta}      {Subderivation}
               \eqPf{B}{[\Theta]B \AND \Dee_1\text{~eager}}  {By i.h.\ (ii) on $\Dee_1$}
               \proofsep
               \derivesPf{\Dee_2}
               \recspinejudgPf{\Theta}{s}{B}{p}{A}{q}{\Delta}  {Subderivation} 
               \eqPf{B}{[\Theta]B}   {Above}
               \eqPf{A}{[\Theta]A \AND \Dee_2\text{~eager}}  {By i.h.\ (iv) on $\Dee_2$}
               \proofsep
           \Hand    \eqPf{A}{[\Theta]A}  {Above}
           \Hand    \Pf{}{}{\Dee_1\text{~eager}}  {Above}
           \Hand    \Pf{}{}{\Dee_2\text{~eager}}  {Above}
           \end{llproof}
  \end{itemize}

  \textbf{Parts (iii) and (iv), spines}

  \begin{itemize}
     \DerivationProofCase{\AllSpine}
         {\spinejudg{\Gamma, \ahat:\sort}{\appspine{e}{s_0}}{[\ahat/\alpha]A_0}{\OK}{C}{q}{\Delta}}
         {\spinejudg{\Gamma}{\appspine{e}{s_0}}{\alltype{\alpha:\sort} A_0}{p}{C}{q}{\Delta}}

         It is given that
         $
         [\Gamma]
         \big(\alltype{\alpha:\sort} A_0\big)
         =
         \big(\alltype{\alpha:\sort} A_0\big)$.
     
         Therefore, $[\Gamma]A_0 = A_0$.

         Since $\ahat$ is not solved in $\Gamma, \ahat:\sort$, we also have
         \[
           [\Gamma, \ahat:\sort][\ahat/\alpha]A_0
           =
           [\ahat/\alpha]A_0
         \]
         By i.h., $C = [\Delta]C$ and all subderivations are eager.
         Since the output type and output context of the conclusion are $C$ and $\Delta$, 
         the same as the premise, we have $C = [\Delta]C$.
         
     \DerivationProofCase{\ImpliesSpine}
         {\checkprop{\Gamma}{P}{\Theta}
          \\
          \spinejudg{\Theta}{\appspine{e}{s_0}}{[\Theta]A_0}{p}{C}{q}{\Delta}}
         {\spinejudg{\Gamma}{\appspine{e}{s_0}}{P \implies A_0}{p}{C}{q}{\Delta}}

         Substitution is idempotent, so $[\Theta][\Theta]A_0 = [\Theta]A_0$,
         and we can apply the i.h.\ showing $C = [\Delta]C$ and that all subderivations are eager.
         Since the output type and output context of the conclusion are $C$ and $\Delta$, 
         the same as the premise, we have $C = [\Delta]C$.

    \ProofCaseRule{\Recover}
      By i.h.\ (iii).

    \ProofCaseRule{\Pass}
      By i.h.\ (iii).

    \DerivationProofCase{\EmptySpine}
         {}
         {\spinejudg{\Gamma}{\emptyspine}{A}{p}{\underbrace{A}_{C}}{\underbrace{p}_{q}}{\underbrace{\Gamma}{\Delta}}}

         We have $[\Gamma]A = A$.  Since $C = A$, we also have $[\Gamma]C = C$;
         since $\Gamma = \Delta$, we also have $[\Delta]C = C$, which was to be shown.
    
     \DerivationProofCase{\ArrSpine}
         {\chkjudg{p}{\Gamma}{e}{A_1}{\Theta}
          \\
          \spinejudg{\Theta}{s}{[\Theta]A_2}{p}{C}{q}{\Delta}
         }
         {\spinejudg{\Gamma}{\appspine{e}{s}}{A_1 \arr A_2}{p}{C}{q}{\Delta}}

         We have $[\Gamma](A_1 \arr A_2) = A_1 \arr A_2$.
         Therefore, $[\Gamma]A_1 = A_1$.
         By i.h.\ on the first subderivation, its subderivations are eager.

         Substitution is idempotent, so $[\Theta][\Theta]A_2 = [\Theta]A_2$.
         By i.h.\ on the second subderivation, $[\Delta]C = C$ (and its subderivations are eager).
         
         Since the output type and output context of the conclusion are $C$ and $\Delta$, 
         the same as the premise, we have $C = [\Delta]C$; we also showed that all subderivations are eager.

     \DerivationProofCase{\SolveSpine}
            {\spinejudg{
                \Gamma_0[\ahat_2 : \type, \ahat_1 : \type, \hypeq{\ahat:\type}{\ahat_1{\arr}\ahat_2}]
              }%
              {\appspine{e}{s_0}}{(\ahat_1 \arr \ahat_2)}{\ok}{C}{\ok}{\Delta}
            }
            {\spinejudg{\Gamma_0[\ahat:\type]}{\appspine{e}{s_0}}{\ahat}{\ok}{C}{\ok}{\Delta}}

            By definition of substitution,
            \[
            \big[\Gamma_0[\ahat_2 : \type, \ahat_1 : \type, \hypeq{\ahat:\type}{\ahat_1{\arr}\ahat_2}]\big]
            (\ahat_1 \arr \ahat_2)
            =
            (\ahat_1 \arr \ahat_2)
            \]
            Therefore, we can apply the i.h.

            Since the output type and output context of the conclusion are $C$ and $\Delta$, 
            the same as the premise, we have $C = [\Delta]C$; we also showed that all subderivations are eager.
  \end{itemize}

  \textbf{Parts (v) and (vi), pattern matching}

  Part (v), rules \MatchEmpty, etc.: By i.h.\ (v) and, in \MatchBase, i.h.\ (i).
  \MatchSeq: By i.h.\ (v), using idempotency of substitution for $\Avec$.

  Part (vi), rule \MatchBot: trivial.
  Part (vi), rule \MatchUnify: by the assumption
  $\judgetpvec{\Gamma}{\Avec\;\p}$,
  the vector $\Avec$ has no existential variables at all,
  so in the second premise, $\Avec = [\Gamma]\Avec$
  and we can apply the i.h.\ (v).
\end{proof}

\typingsoundness*
\begin{proof}\proofheading
  By induction, using the measure in \Definitionref{def:soundness-completeness-measure}.

  Where the i.h.\ is used, we elide the reasoning establishing
  the condition $[\Gamma]A = A$ for parts (i), (iii), (iv), (v) and (vi):
  this condition follows from \Theoremref{thm:eagerness}, which ensures
  that the appropriate condition holds for all subderivations.

  \begin{itemize}
     \DerivationProofCase{\Var}
          {(\hyp{p}{x}{A}) \in \Gamma}
          {\synjudg{p}{\Gamma}{x}{[\Gamma]A}{\Gamma}}
          
          \begin{llproof}
            \inPf{(\hyp{p}{x}{A})}{\Gamma}  {Premise}
            \inPf{(\hyp{p}{x}{A})}{\Delta}  {$\Gamma = \Delta$}
            \substextendPf{\Delta}{\Omega}  {Given}
            \inPf{(\hyp{p}{x}{[\Omega]A})}{[\Omega]\Gamma}  {By \Lemmaref{lem:ctxapp-uvar-preservation} (ii)}
            \declsynjudgPf{p}{[\Omega]\Gamma}{[\Omega]x}{[\Omega]A} {By \DeclVar}
            \substextendPf{\Delta}{\Omega}   {Given}
            \substextendPf{\Gamma}{\Omega}   {$\Gamma = \Delta$}
            \eqPf{[\Omega]A}{[\Omega][\Gamma]A}   {By \Lemmaref{lem:substitution-monotonicity} (iii)}
\Hand            \declsynjudgPf{p}{[\Omega]\Gamma}{[\Omega]x}{[\Omega][\Gamma]A} {By above equality}
          \end{llproof}

     \DerivationProofCase{\Sub}
          {\synjudg{q}{\Gamma}{e}{A}{\Theta}
            \\
            \subjudg
                [\joinpolarity{\polarity B}{\polarity A}]
                {\Theta}{A}{B}
                {\Delta}
          }
          {\chkjudg{p}{\Gamma}{e}{B}{\Delta}}

          \begin{llproof}
            \synjudgPf{q}{\Gamma}{e}{A}{\Theta}  {Subderivation}
            \subjudgPf{\polvar}{\Theta}{A}{B}{\Delta}  {Subderivation}
            \substextendPf{\Theta}{\Delta}  {By \Lemmaref{lem:typing-extension}}
            \substextendPf{\Delta}{\Omega}  {Given}
            \substextendPf{\Theta}{\Omega}  {By \Lemmaref{lem:extension-transitivity}}
            \declsynjudgPf{q}{[\Omega]\Theta} {[\Omega]e} {[\Omega]A}  {By i.h.}
            \eqPf{[\Omega]\Theta} {[\Omega]\Delta}   {By \Lemmaref{lem:completes-confluence}}
            \declsynjudgPf{q}{[\Omega]\Delta} {[\Omega]e} {[\Omega]A}  {By above equality}
            \proofsep
            \subjudgPf{\joinpolarity{\polarity B}{\polarity A}}{\Theta}{A}{B}{\Delta}  {Subderivation}
            \declsubjudgPf{\joinpolarity{\polarity B}{\polarity A}}{[\Omega]\Delta} {[\Omega]A} {[\Omega]B}  {By \Theoremref{thm:subtyping-soundness}}
\Hand       \declchkjudgPf{p}{[\Omega]\Delta} {[\Omega]e} {[\Omega]B}  {By \DeclSub}
          \end{llproof}

     \DerivationProofCase{\Anno}
          {\judgetp{\Gamma}{A_0 \p}
           \\
           \chkjudg{\p}{\Gamma}{e_0}{[\Gamma]A_0}{\Delta}
          }
          {\synjudg{\p}{\Gamma}{(e_0 : A_0)}{[\Delta]A_0}{\Delta}}
          
          \begin{llproof}
            \chkjudgPf{\p}{\Gamma}{e_0}{[\Gamma]A_0}{\Delta}  {Subderivation}
            \declchkjudgPf{\p}{[\Omega]\Delta} {[\Omega]e_0} {[\Omega][\Gamma]A_0}  {By i.h.}
            \proofsep
            \judgetpPf{\Gamma}{A_0 \p}  {Subderivation}
            \judgetpPf{\Gamma}{A_0}  {By inversion}
            \eqPf{\FEV{A_0}}{\emptyset}   {\ditto}
            \decolumnizePf
            \substextendPf{\Gamma}{\Delta}  {By \Lemmaref{lem:typing-extension}}
            \substextendPf{\Delta}{\Omega}  {Given}
            \substextendPf{\Gamma}{\Omega}  {By \Lemmaref{lem:extension-transitivity}}
            \judgetpPf{\Omega} {A_0}   {By \Lemmaref{lem:extension-weakening}}
            \judgetpPf{[\Omega]\Omega} {[\Omega]A_0}   {By \Lemmaref{lem:completion-wf}}
            \eqPf{[\Omega]\Omega}{[\Omega]\Delta}  {By \Lemmaref{lem:completes-stability}}
            \judgetpPf{[\Omega]\Delta} {[\Omega]A_0}   {By above equality}
            \proofsep
            \eqPf{[\Omega][\Gamma]A_0}{[\Omega]A_0}  {By \Lemmaref{lem:substitution-monotonicity} (iii)}
            \declchkjudgPf{\p}{[\Omega]\Delta} {[\Omega]e_0} {[\Omega]A_0}  {By above equality}
            \proofsep
            \declsynjudgPf{\p}{[\Omega]\Delta} {([\Omega]e_0 : [\Omega]A_0)} {[\Omega]A_0}  {By \DeclAnno}
            \eqPf{[\Omega]A_0} {A_0}  {From definition of substitution}
\Hand     \declsynjudgPf{\p}{[\Omega]\Delta} {[\Omega](e_0 : A_0)} {[\Omega]A_0}  {By above equality}
          \end{llproof}

     \DerivationProofCase{\UnitIntro}
          {}
          {\chkjudg{p}{\Gamma}{\unitexp}{\unitty}{\underbrace{\Gamma}_{\Delta}}}

          \begin{llproof}
            \declchkjudgPf{p}{[\Omega]\Delta}{\unitexp}{\unitty}   {By \DeclUnitIntro}
\Hand    \declchkjudgPf{p}{[\Omega]\Delta} {[\Omega]\unitexp} {[\Omega]\unitty}  {By definition of substitution}
          \end{llproof}

     \DerivationProofCase{\UnitIntroSolve}
            { }
            {\chkjudg{\OK}{\Gamma_0[\ahat:\type]}{\unitexp}{\ahat}{\underbrace{\Gamma_0[\hypeq{\ahat:\type}{\unitty}]}_{\Delta}}}

          \begin{llproof}
            \substextendPf{\Gamma_0[\hypeq{\ahat:\type}{\unitty}]}{\Omega}   {Given}
            \eqPf{[\Omega]\ahat}{[\Omega][\Delta]\ahat}  {By \Lemmaref{lem:substitution-monotonicity} (i)}
            \continueeqPf{[\Omega]\unitty} {By definition of context application}
            \continueeqPf{\unitty} {By definition of context application}
            \declchkjudgPf{\OK}{[\Omega]\Delta} {\unitexp} {\unitty}  {By \DeclUnitIntro}
          \Hand  \declchkjudgPf{\OK}{[\Omega]\Delta} {[\Omega]\unitexp} {[\Omega]\ahat}  {By above equality}
          \end{llproof}

     \DerivationProofCase{\AllIntro}
           {\chkintro{v}
            \\
            \chkjudg{p}{\Gamma, \alpha : \sort}{v}{A_0}{\Delta, \alpha : \sort, \Theta}
           }
           {\chkjudg{p}{\Gamma}{v}{\alltype{\alpha : \sort} A_0}{\Delta}}

          \begin{llproof}
            \substextendPf{\Delta}{\Omega}  {Given}
            \substextendPf{\Delta, \alpha}{\Omega, \alpha}  {By \substextendUU}
            \substextendPf{\Gamma, \alpha}{\Delta, \alpha, \Theta}  {By \Lemmaref{lem:typing-extension}}
            \softPf{\Theta}   {By \Lemmaref{lem:extension-inversion} (i) (with $\Gamma_R = \cdot$, which is soft)}
            \substextendPf{\underbrace{\Delta, \alpha, \Theta}_{\Delta'}}{\underbrace{\Omega, \alpha, \soln{\Theta}}_{\Omega'}}  {By \Lemmaref{lem:soln-completes}}
            \decolumnizePf
            \chkjudgPf{p}{\Gamma, \alpha}{v}{A_0}{\Delta'}  {Subderivation}
            \proofsep
            \declchkjudgPf{p}{[\Omega']\Delta'} {[\Omega]v} {[\Omega']A_0}  {By i.h.}
            \eqPf{[\Omega']A_0}{[\Omega]A_0}   {By \Lemmaref{lem:substitution-stability}}
            \declchkjudgPf{p}{[\Omega']\Delta'} {[\Omega]v} {[\Omega]A_0} {By above equality}
            \decolumnizePf
            \substextendPf{\underbrace{\Delta, \alpha, \Theta}_{\Delta'}}{\underbrace{\Omega, \alpha, \soln{\Theta}}_{\Omega'}}  {Above}
            \Pf{}{}{\Theta\text{~is soft}}   {Above}
            \eqPfParenR{[\Omega']\Delta'} {[\Omega]\Delta, \alpha}   {By \Lemmaref{lem:softness-goes-away}}
            \declchkjudgPf{p}{[\Omega]\Delta, \alpha} {[\Omega]v} {[\Omega]A_0} {By above equality}
            \proofsep
            \declchkjudgPf{p}{[\Omega]\Delta} {[\Omega]v} {\alltype{\alpha} [\Omega]A_0}  {By \DeclAllIntro}
\Hand       \declchkjudgPf{p}{[\Omega]\Delta} {[\Omega]v} {[\Omega](\alltype{\alpha} A_0)}  {By definition of substitution}
          \end{llproof}

     \DerivationProofCase{\AllSpine}
         {\spinejudg{\Gamma, \ahat:\sort}{\appspine{e}{s_0}}{[\ahat/\alpha]A_0}{\OK}{C}{q}{\Delta}}
         {\spinejudg{\Gamma}{\appspine{e}{s_0}}{\alltype{\alpha:\sort} A_0}{p}{C}{q}{\Delta}}

         \begin{llproof}
           \spinejudgPf{\Gamma, \ahat:\sort}{\appspine{e}{s_0}}{[\ahat/\alpha]A_0}{\OK}{C}{q}{\Delta}   {Subderivation}
           \declspinejudgPf{[\Omega]\Delta}{[\Omega](\appspine{e}{s_0})}{[\Omega][\ahat/\alpha]A_0}{\OK}{[\Omega]C}{q}   {By i.h.}
           \declspinejudgPf{[\Omega]\Delta}{[\Omega](\appspine{e}{s_0})}{\big[[\Omega]\ahat/\alpha\big][\Omega]A_0}{\OK}{[\Omega]C}{q}   {By a property of substitution}
           \decolumnizePf
           \judgePf{\Gamma, \ahat:\sort}{\ahat}{\sort}  {By \VarSort}
           \substextendPf{\Gamma, \ahat:\sort}{\Delta}  {By \Lemmaref{lem:typing-extension}}
           \judgePf{\Delta}{\ahat}{\sort}  {By \Lemmaref{lem:extension-weakening}}
           \substextendPf{\Delta}{\Omega}  {Given}
           \judgePf{[\Omega]\Delta}{[\Omega]\ahat}{\sort}  {By \Lemmaref{lem:bundled-subst-sort}}
           \decolumnizePf
           \declspinejudgPf{[\Omega]\Delta}{[\Omega](\appspine{e}{s_0})}{\alltype{\alpha:\sort} [\Omega]A_0}{p}{[\Omega]C}{q}   {By \DeclAllSpine}
         \Hand  \declspinejudgPf{[\Omega]\Delta}{[\Omega](\appspine{e}{s_0})}{[\Omega](\alltype{\alpha:\sort} A_0)}{p}{[\Omega]C}{q}   {By def.\ of subst.}
         \end{llproof}

     \DerivationProofCase{\WithIntro}
           {\chkintro{e}
            \\
             \checkprop{\Gamma}{P}{\Theta}
             \\
             \chkjudg{p}{\Theta}{e}{[\Theta]A_0}{\Delta}
           }
           {\chkjudg{p}{\Gamma}{e}{A_0 \with P}{\Delta}}

           \begin{llproof}
             \checkpropPf{\Gamma}{P}{\Theta}  {Subderivation}
             \substextendPf{\Delta}{\Omega}  {Given}
             \substextendPf{\Theta}{\Delta}  {By \Lemmaref{lem:typing-extension}}
             \substextendPf{\Theta}{\Omega}  {By \Lemmaref{lem:extension-transitivity}}
             \declcheckpropPf{[\Omega]\Theta}{[\Omega]P}  {By \Lemmaref{lem:checkprop-soundness}}
             \declcheckpropPf{[\Omega]\Delta}{[\Omega]P}  {By \Lemmaref{lem:completes-confluence}}
             \proofsep
             \chkjudgPf{p}{\Theta}{e}{[\Theta]A_0}{\Delta}  {Subderivation}
             \declchkjudgPf{p}{[\Omega]\Delta}{[\Omega]e}{([\Omega][\Theta]A_0)}   {By i.h.}
             \declchkjudgPf{p}{[\Omega]\Delta}{[\Omega]e}{([\Omega][\Theta]A_0) \with [\Omega]P}   {By \DeclWithIntro}
             \eqPf{[\Omega][\Theta]A_0}{[\Omega]A_0}   {By \Lemmaref{lem:substitution-monotonicity} (iii)}
             \declchkjudgPf{p}{[\Omega]\Delta}{[\Omega]e}{([\Omega]A_0) \with [\Omega]P}   {By above equality}
           \Hand  \declchkjudgPf{p}{[\Omega]\Delta}{[\Omega]e}{[\Omega](A_0 \with P)}   {By def.\ of substitution}
           \end{llproof}

       \DerivationProofCase{\Nil}
            {\checkprop{\Gamma}{t = \zero}{\Delta}}
            {\chkjudg{p}{\Gamma}{\vecnil}{(\vectype{t}{A})}{\Delta}}

           \begin{llproof}
             \checkpropPf{\Gamma}{t = \zero}{\Delta}  {Subderivation}
             \substextendPf{\Delta}{\Omega}  {Given}
             \declcheckpropPf{[\Omega]\Delta}{[\Omega](t = \zero)}  {By \Lemmaref{lem:checkprop-soundness}}
             \declcheckpropPf{[\Omega]\Delta}{[\Omega]t = \zero}  {By def.\ of substitution}
           \Hand  \declchkjudgPf{p}{[\Omega]\Delta}{[\Omega]\vecnil}{(\vectype{[\Omega]t}{[\Omega]A})}   {By \DeclNil}
           \end{llproof}

      \DerivationProofCase{\Cons}
            {
              \checkprop{\Gamma,
                        \MonnierComma{\ahat},
                        \ahat : \kindnat}
                       {t = \succ{\ahat}}
                       {\Gamma'}
              \\
              \arrayenvbl{
              \chkjudg{p}{\Gamma'}{e_1}{[\Gamma']A_0}{\Theta}
              \\
              \chkjudg{\OK}
                      {\Theta}
                      {e_2}
                      {[\Theta](\vectype{\ahat}{A_0})}
                      {\Delta, \MonnierComma{\ahat}, \Delta'}
              }
            }
            {
              \chkjudg{p}{\Gamma}{\veccons{e_1}{e_2}}{(\vectype{t}{A_0})}{\Delta}
            }

           \begin{llproof}
             \checkpropPf{\Gamma, \MonnierComma{\ahat}, \ahat : \kindnat}{t = \succ{\ahat}}{\Gamma'}  {Subderivation}
             \substextendPf{\Delta}{\Omega}  {Given}
             \substextendPf{\Gamma'}{\Theta}  {By \Lemmaref{lem:typing-extension}}
             \substextendPf{\Theta}{\Delta, \MonnierComma{\ahat}, \Delta'}  {By \Lemmaref{lem:typing-extension}}
             \substextendPf{\Delta, \MonnierComma{\ahat}, \Delta'}{\Omega'}  {By \Lemmaref{lem:soln-completes}}
             \substextendPf{\Gamma'}{\Omega'}  {By \Lemmaref{lem:extension-transitivity}}
             \declcheckpropPf{[\Omega']\Gamma'}{[\Omega'](t = \succ{\ahat})}  {By \Lemmaref{lem:checkprop-soundness}}
             \declcheckpropPf{[\Omega'](\Delta, \MonnierComma{\ahat}, \Delta')}{[\Omega'](t = \succ{\ahat})}  {By \Lemmaref{lem:completes-confluence}}
             \declcheckpropPf{[\Omega'](\Delta, \MonnierComma{\ahat}, \Delta')}{[\Omega](t = \succ{\ahat})}  {By \Lemmaref{lem:substitution-stability}}
             \declcheckpropPf{[\Omega]\Delta}{[\Omega](t = \succ{\ahat})}  {By \Lemmaref{lem:context-partitioning} + thinning}
           1~~~\declcheckpropPf{[\Omega]\Delta}{([\Omega]t) = \succ{[\Omega]\ahat}}  {By def.\ of substitution}
             \decolumnizePf
             \chkjudgPf{p}{\Gamma'}{e_1}{[\Gamma']A_0}{\Theta}  {Subderivation}
             \declchkjudgPf{p}{[\Omega']\Theta}{[\Omega']e_1}{([\Omega'][\Gamma']A_0)}   {By i.h.}
             \eqPf{[\Omega'][\Gamma']A_0}{[\Omega']A_0}   {By \Lemmaref{lem:substitution-monotonicity} (iii)}
             \declchkjudgPf{p}{[\Omega']\Theta}{[\Omega']e_1}{[\Omega']A_0}   {By above equality}
           2~~~\declchkjudgPf{p}{[\Omega]\Delta}{[\Omega]e_1}{[\Omega]A_0}   {Similar to above}
             \decolumnizePf
             \chkjudgPf{\OK}
                      {\Theta}
                      {e_2}
                      {[\Theta](\vectype{\ahat}{A_0})}
                      {\Delta, \MonnierComma{\ahat}, \Delta'}
                      {Subderivation}
             \declchkjudgPf{\OK}
                      {[\Omega'](\Delta, \MonnierComma{\ahat}, \Delta')}
                      {[\Omega']e_2}
                      {[\Omega'][\Theta](\vectype{\ahat}{A_0})}
                      {By i.h.}
             \declchkjudgPf{\OK}
                      {[\Omega]\Delta}
                      {[\Omega]e_2}
                      {[\Omega](\vectype{\ahat}{A_0})}
                      {Similar to above}
           3~~~\declchkjudgPf{p}{[\Omega]\Delta}{[\Omega]e_2}{(\vectype{([\Omega]\ahat)}{[\Omega]A_0})}    {By def.\ of substitution}
             \proofsep
             \declchkjudgPf{p}{[\Omega]\Delta}{\veccons{([\Omega]e_1)}{[\Omega]e_2}}{\vectype{([\Omega]t)}{[\Omega]A_0}}
                          {By \DeclCons (premises: 1, 2, 3)}
           \Hand
             \declchkjudgPf{p}{[\Omega]\Delta}{[\Omega](\veccons{e_1}{e_2})}{[\Omega](\vectype{t}{A_0})}
                          {By def.\ of substitution}
           \end{llproof}

     \clearpage

     \DerivationProofCase{\ExIntro}
           {\chkintro{e} \\     %
            \chkjudg{\ok}{\Gamma, \ahat:\sort}{e}{[\ahat/\alpha]A_0}{\Delta}
           }
           {\chkjudg{p}{\Gamma}{e}{\extype{\alpha:\sort} A_0}{\Delta}}

         \begin{llproof}
           \chkjudgPf{\ok}{\Gamma, \ahat:\sort}{e}{[\ahat/\alpha]A_0}{\Delta}   {Subderivation}
           \declchkjudgPf{\ok}{[\Omega]\Delta}{[\Omega]e}{[\Omega][\ahat/\alpha]A_0}  {By i.h.}
           \proofsep
           \declchkjudgPf{\ok}{[\Omega]\Delta}{[\Omega]e}{\big[[\Omega]\ahat/\alpha\big][\Omega]A_0}
                 {By a property of substitution}
           \proofsep
           \judgePf{\Gamma, \ahat:\sort}{\ahat}{\sort}  {By \VarSort}
           \substextendPf{\Gamma, \ahat:\sort}{\Delta}  {By \Lemmaref{lem:typing-extension}}
           \judgePf{\Delta}{\ahat}{\sort}  {By \Lemmaref{lem:extension-weakening}}
           \substextendPf{\Delta}{\Omega}  {Given}
           \judgePf{[\Omega]\Delta}{[\Omega]\ahat}{\sort}  {By \Lemmaref{lem:bundled-subst-sort}}
           \proofsep
           \declchkjudgPf{p}{[\Omega]\Delta}{[\Omega]e}{\extype{\alpha:\sort} [\Omega]A_0}
                {By \DeclExIntro}
         \Hand
           \declchkjudgPf{p}{[\Omega]\Delta}{[\Omega]e}{[\Omega](\extype{\alpha:\sort} A_0)}
                 {By def.\ of subst.}
         \end{llproof}

     \DerivationProofCase{\ImpliesIntro}
           { \chkintro{v}
             \\
             \elimprop{\Gamma, \MonnierComma{P}}{P}{\Theta^+}
             \\
             \chkjudg{\p}{\Theta^+}{v}{[\Theta^+]A_0}{\Delta, \MonnierComma{P}, \Delta'}
            }
            {\chkjudg{\p}{\Gamma}{v}{P \implies A_0}{\Delta}}

           \begin{llproof}
             \proofsep
             \judgetpPf{\Gamma}{A\;\p}  {Given}
             \eqPf{\FEV{[\Gamma]A}}{\emptyset}  {By inversion on rule \PrincipalWF}
             \eqPf{\FEV{[\Gamma]P}}{\emptyset}  {$A = (P \implies A_0)$}
             \proofsep
             \elimpropPf{\Gamma, \MonnierComma{P}}{P}{\Theta^+}  {Subderivation}
             \elimeqPf{\Gamma, \MonnierComma{P}}{\sigma}{t}{\sort}{\Theta^+}  {By inversion}
             \eqPf{\FEV{[\Gamma]\sigma} \union \FEV{[\Gamma]t}}{\emptyset}  {By $\FEV{[\Gamma]P} = \emptyset$ above}
             \proofsep
             \eqPfParenR{\Theta^+}{\Gamma, \MonnierComma{P}, \Theta}   {By \Lemmaref{lem:elimeq-soundness}}
             \eqPf{[\Omega', \Theta]t'}{[\theta][\Gamma, \MonnierComma{P}]t'}  {\ditto~~~(for all $\Omega'$ extending $(\Gamma, \MonnierComma{P})$ and $t'$ s.t.\ $\judge{\Omega'}{t'}{\sort'}$)}
             \eqPf{\theta}{\mgu{\sigma}{t}}  {\ditto}
             \proofsep
             \substextendPf{\Delta}{\Omega}  {Given}
             \substextendPf{\Theta^+}{\Delta, \MonnierComma{P}, \Delta'}  {By \Lemmaref{lem:typing-extension}}
             \substextendPf{\Gamma, \MonnierComma{P}, \Theta}{\Delta, \MonnierComma{P}, \Delta'}  {By above equalities}
             \LetPf{\Omega^+}{(\Omega, \MonnierComma{P}, \Delta')}{}
             \substextendPf{\Delta, \MonnierComma{P}, \Theta}{\Omega, \MonnierComma{P}, \Delta'}  {By repeated \substextendEqn}
             \substextendPf{\Theta^+}{\Omega^+}  {By \Lemmaref{lem:extension-transitivity}}
             \decolumnizePf
             \eqPf{[\Omega', \Theta]B}{[\theta][\Gamma, \MonnierComma{P}]B}
                     {By \Lemmaref{lem:subst-upgrade} (i)}
                     \trailingjust{~~~(for all $\Omega'$ extending
                       $(\Gamma, \MonnierComma{P}$ and $B$
                       s.t.\ $\judge{\Omega'}{B}{\sort'}$)}
             \decolumnizePf
             \chkjudgPf{\p}{\Theta^+}{v}{[\Theta^+]A_0}{\Delta, \MonnierComma{P}, \Delta'}  {Subderivation}
             \declchkjudgPf{\p}{[\Omega^+](\Delta, \MonnierComma{P}, \Delta')}{[\Omega]v}{[\Omega^+][\Theta^+]A_0}  {By i.h.}
             \proofsep
             \substextendPf{\Gamma, \MonnierComma{P}, \Theta} {\Omega, \MonnierComma{P}, \Delta'}   {By \Lemmaref{lem:extension-transitivity}}
             \substextendPf{\Gamma}{\Omega}   {By \Lemmaref{lem:extension-inversion}}
             \eqPf{[\Omega^+][\Theta^+]A_0}  {[\Omega^+]A_0}  {By \Lemmaref{lem:substitution-monotonicity}}
             \continueeqPf{[\theta][\Omega, \MonnierComma{P}]A_0}
                    {Above, with $(\Omega, \MonnierComma{P})$ as $\Omega'$
                      and $A_0$ as $B$}
             \continueeqPf{[\theta][\Omega]A_0}  {By def.\ of substitution}
             \proofsep
             \eqPf{[\Omega, \MonnierComma{P}, \Theta](\Delta, \MonnierComma{P}, \Delta')}{[\theta][\Omega]\Delta}   {By \Lemmaref{lem:subst-upgrade} (iii)}
             \decolumnizePf
             \declchkjudgPf{\p}{[\theta][\Omega]\Delta}{[\Omega][\theta]v}{[\theta][\Omega]A_0}    {By above equalities}
             \proofsep
             \declchkjudgPf{\p}{[\Omega^+](\Delta, \MonnierComma{P}, \Delta') \ctxsep (\sigma = t)}{[\Omega]v}{[\Omega]A_0}    {By \DeclCheckUnify}
             \eqPf{[\Omega^+](\Delta, \MonnierComma{P}, \Delta')}{[\Omega]\Delta}    {From def.\ of context application}
             \declchkjudgPf{\p}{[\Omega]\Delta \ctxsep (\sigma = t)}{[\Omega]v}{[\Omega]A_0}    {By above equality}
             \declchkjudgPf{\p}{[\Omega]\Delta}{[\Omega]v}{(\sigma = t) \implies [\Omega]A_0}    {By \DeclImpliesIntro}
             \declchkjudgPf{\p}{[\Omega]\Delta}{[\Omega]v}{([\Omega]\sigma = [\Omega]t) \implies [\Omega]A_0}    {By $\xfev$ condition above}
           \end{llproof}

     \DerivationProofCase{\ImpliesIntroBot}
           {\chkintro{v}
            \\
            \elimprop{\Gamma, \MonnierComma{P}}{P}{\bot}
            }
            {\chkjudg{\p}{\Gamma}{v}{P \implies A_0}{\underbrace{\Gamma}_{\Delta}}}

            \begin{llproof}
              \elimpropPf{\Gamma, \MonnierComma{P}}{P}{\bot}        {Subderivation}
              \elimeqPf{\Gamma, \MonnierComma{P}}{\sigma}{t}{\sort}{\bot}        {By inversion}
              \eqPfParenR{P}{\sigma = t}       {\ditto}
              \eqPf{\FEV{[\Gamma]\sigma} \union \FEV{[\Gamma]t}}{\emptyset}    {As in \ImpliesIntro case (above)}
              \eqPf{\mgu{\sigma}{t}} {\bot}   {By \Lemmaref{lem:elimeq-soundness}}
              \decolumnizePf
              \declchkjudgPf{\p}{[\Omega]\Delta \ctxsep (\sigma = t)}{[\Omega]v}{[\Omega]A_0}   {By \DeclCheckBot}
              \declchkjudgPf{\p}{[\Omega]\Delta}{[\Omega]v}{(\sigma = t) \implies [\Omega]A_0}   {By \DeclImpliesIntro}
              \declchkjudgPf{\p}{[\Omega]\Delta}{[\Omega]v}{\big([\Omega](\sigma = t)\big) \implies [\Omega]A_0}   {By above $\xfev$ condition}
            \Hand \declchkjudgPf{\p}{[\Omega]\Delta}{[\Omega]v}{[\Omega]\big(P \implies A_0\big)}   {By def.\ of subst.}
              \LetPf{\Omega'}{\Omega}  {}
            \Hand \substextendPf{\Omega}{\Omega'}   {By \Lemmaref{lem:extension-reflexivity}}
            \Hand \substextendPf{\Delta}{\Omega'}   {Given}
            \end{llproof}

     \DerivationProofCase{\ImpliesSpine}
         {\checkprop{\Gamma}{P}{\Theta}
          \\
          \spinejudg{\Theta}{\appspine{e}{s_0}}{[\Theta]A_0}{p}{C}{q}{\Delta}}
         {\spinejudg{\Gamma}{\appspine{e}{s_0}}{P \implies A_0}{p}{C}{q}{\Delta}}

         \begin{llproof}
           \spinejudgPf{\Theta}{\appspine{e}{s_0}}{[\Theta]A_0}{p}{C}{q}{\Delta}   {Subderivation}
           \proofsep
           \substextendPf{\Theta}{\Delta}  {By \Lemmaref{lem:typing-extension}}
           \substextendPf{\Delta}{\Omega}   {Given}
           \substextendPf{\Theta}{\Omega}  {By \Lemmaref{lem:extension-transitivity}}
           \decolumnizePf
           \declspinejudgPf{[\Omega]\Delta}{[\Omega](\appspine{e}{s_0})}{[\Omega][\Theta]A_0}{p}{[\Omega]C}{q}   {By i.h.}
           \eqPf{[\Omega][\Theta]A_0}{[\Omega]A_0}   {By \Lemmaref{lem:substitution-monotonicity} (iii)}
           \declspinejudgPf{[\Omega]\Delta}{[\Omega](\appspine{e}{s_0})}{[\Omega]A_0}{p}{[\Omega]C}{q}   {By above equality}
           \proofsep
           \checkpropPf{\Gamma}{P}{\Theta}  {Subderivation}
           \declcheckpropPf{[\Omega]\Theta}{[\Omega]P}   {By \Lemmaref{lem:checkprop-completeness}}
           \eqPf{[\Omega]\Theta}{[\Omega]\Delta}  {By \Lemmaref{lem:completes-confluence}}
           \declcheckpropPf{[\Omega]\Delta}{[\Omega]P}   {By above equality}
           \proofsep
           \declspinejudgPf{[\Omega]\Delta}{[\Omega](\appspine{e}{s_0})}{([\Omega]P) \implies [\Omega]A_0}{p}{[\Omega]C}{q}   {By \DeclImpliesSpine}
         \Hand  \declspinejudgPf{[\Omega]\Delta}{[\Omega](\appspine{e}{s_0})}{[\Omega](P \implies A_0)}{p}{[\Omega]C}{q}   {By def.\ of subst.}
         \end{llproof}

\clearpage

         \DerivationProofCase{\ArrIntro}
              {\chkjudg{p}{\Gamma, \hyp{p}{x}{A_1}}{e_0}{A_2}{\Delta, \hyp{p}{x}{A_1}, \Theta}
              }
              {\chkjudg{p}{\Gamma}{\lam{x} e_0}{A_1 \arr A_2}{\Delta}}

              \begin{llproof}
                \substextendPf{\Delta}{\Omega}  {Given}
                \substextendPf{\Delta, \hyp{p}{x}{A_1}}{\Omega, \hyp{p}{x}{[\Omega]A_1}}  {By \substextendVV}
                \substextendPf{\Gamma, \hyp{p}{x}{A_1}}{\Delta, \hyp{p}{x}{A_1}, \Theta}  {By \Lemmaref{lem:typing-extension}}
                \softPf{\Theta}   {By \Lemmaref{lem:extension-inversion} (v)}
                                                \trailingjust{(with $\Gamma_R = \cdot$, which is soft)}
                \substextendPf{\underbrace{\Delta, \hyp{p}{x}{A_1}, \Theta}_{\Delta'}}{\underbrace{\Omega, \hyp{p}{x}{[\Omega]A_1}, \soln{\Theta}}_{\Omega'}}  {By \Lemmaref{lem:soln-completes}}
                \proofsep
                \chkjudgPf{p}{\Gamma, \hyp{p}{x}{A_1}}{e_0}{A_2}{\Delta'}  {Subderivation}
                \proofsep
                \declchkjudgPf{p}{[\Omega']\Delta'} {[\Omega]e_0} {[\Omega']A_2}  {By i.h.}
                \eqPf{[\Omega']A_2}{[\Omega]A_2}   {By \Lemmaref{lem:substitution-stability}}
                \declchkjudgPf{p}{[\Omega']\Delta'} {[\Omega]e_0} {[\Omega]A_2} {By above equality}
                \proofsep
                \substextendPf{\underbrace{\Delta, \hyp{p}{x}{A_1}, \Theta}_{\Delta'}}{\underbrace{\Omega, \hyp{p}{x}{[\Omega]A_1}, \soln{\Theta}}_{\Omega'}}  {Above}
                \softPf{\Theta}   {Above}
                \eqPfParenR{[\Omega']\Delta'} {[\Omega]\Delta, \hyp{p}{x}{[\Omega]A_1}}   {By \Lemmaref{lem:softness-goes-away}}
                \declchkjudgPf{p}{[\Omega]\Delta, \hyp{p}{x}{[\Omega]A_1}} {[\Omega]e_0} {[\Omega]A_2} {By above equality}
                \decolumnizePf
                \declchkjudgPf{p}{[\Omega]\Delta} {\lam{x} [\Omega]e_0} {([\Omega]A_1) \arr ([\Omega]A_2)}  {By \DeclArrIntro}
    \Hand       \declchkjudgPf{p}{[\Omega]\Delta} {[\Omega](\lam{x} e_0)} {[\Omega](A_1 \arr A_2)}  {By definition of substitution}
              \end{llproof}

         \DerivationProofCase{\Rec}
              {
                \chkintro{v}
                \\
                \chkjudg{p}{\Gamma, \hyp{p}{x}{A}}{v}{A}{\Delta, \hyp{p}{x}{A}, \Theta}
              }
              { \chkjudg{p}{\Gamma}{\rec{x} v}{A}{\Delta} }

              Similar to the \ArrIntro case, applying \DeclRec instead of \DeclArrIntro.

\clearpage

         \DerivationProofCase{\ArrIntroSolve}
               {\chkjudg{\OK}
                  {\Gamma[\ahat_1{:}\type, \ahat_2{:}\type, \hypeq{\ahat : \type}{\ahat_1{\arr}\ahat_2}], \hyp{\OK}{x}{\ahat_1}}
                  {e_0}
                  {\ahat_2}
                  {\Delta, \hyp{\OK}{x}{\ahat_1}, \Theta}
               }
               {\chkjudg{\OK}{\Gamma[\ahat : \type]}{\lam{x} e_0}{\ahat}{\Delta}}

                  \begin{llproof}
                    \substextendPf{\Gamma[\ahat_1{:}\type, \ahat_2{:}\type, \hypeq{\ahat : \type}{\ahat_1{\arr}\ahat_2}], \hyp{\OK}{x}{\ahat}}{\Delta, \hyp{\OK}{x}{\ahat}, \Theta}  {By \Lemmaref{lem:typing-extension}}
                    \softPf{\Theta}   {By \Lemmaref{lem:extension-inversion} (v)}
                    \trailingjust{(with $\Gamma_R = \cdot$, which is soft)}
                    \substextendPf{\Gamma[\ahat_1{:}\type, \ahat_2{:}\type, \hypeq{\ahat : \type}{\ahat_1{\arr}\ahat_2}]}{\Delta}   {\ditto}
                    \proofsep
                    \substextendPf{\Delta}{\Omega}  {Given}
                    \substextendPf{\Delta, \hyp{\OK}{x}{\ahat_1}}{\Omega, \hyp{\OK}{x}{[\Omega]\ahat_1}}  {By \substextendVV}
                    \substextendPf{\underbrace{\Delta, \hyp{\OK}{x}{\ahat_1}, \Theta}_{\Delta'}}{\underbrace{\Omega, \hyp{\OK}{x}{[\Omega]\ahat_1}, \soln{\Theta}}_{\Omega'}}  {By \Lemmaref{lem:soln-completes}}
                    \decolumnizePf
                    \chkjudgPf{\OK}{\Gamma[\ahat_1{:}\type, \ahat_2{:}\type, \hypeq{\ahat : \type}{\ahat_1{\arr}\ahat_2}], \hyp{\OK}{x}{\ahat_1}}{e_0}{\ahat_2}{\Delta, \hyp{\OK}{x}{\ahat_1}, \Theta}  {Subderivation}
                    \decolumnizePf
                    \declchkjudgPf{\OK}{[\Omega']\Delta'}
                                 {[\Omega']e_0} {[\Omega']\ahat_2}  {By i.h.}
                    \eqPf{[\Omega']\ahat_2}{\big[\Omega, \hyp{\OK}{x}{[\Omega]\ahat_1}\big]\ahat_2}   {By \Lemmaref{lem:substitution-stability}}
                    \continueeqPf{[\Omega]\ahat_2}   {By definition of substitution}
                    \eqPf{[\Omega']\Delta'}{\big[\Omega, \hyp{\OK}{x}{[\Omega]\ahat_1}\big]\big(\Delta, \hyp{\OK}{x}{\ahat_1}\big)} 
                             {By \Lemmaref{lem:softness-goes-away}}
                    \continueeqPf{[\Omega]\Delta, \hyp{\OK}{x}{[\Omega]\ahat_1}}    {By definition of context substitution}
                    \declchkjudgPf{\OK}{[\Omega]\Delta, \hyp{\OK}{x}{[\Omega]\ahat_1}}
                                 {[\Omega]e_0} {[\Omega]\ahat_2} {By above equalities}
                    \proofsep
                    \declchkjudgPf{\OK}{[\Omega]\Delta}
                                 {\lam{x} [\Omega]e_0} {([\Omega]\ahat_1) \arr [\Omega]\ahat_2} {By \DeclArrIntro}
                    \decolumnizePf
                    \substextendPf{\Gamma[\ahat_1{:}\type, \ahat_2{:}\type, \hypeq{\ahat : \type}{\ahat_1{\arr}\ahat_2}]}{\Omega}   {Above and \Lemmaref{lem:extension-transitivity}}
                    \decolumnizePf
                    \eqPf{[\Omega]\ahat}{[\Omega][\Gamma]\ahat} {By \Lemmaref{lem:substitution-monotonicity} (i)}
                    \continueeqPf{[\Omega]\big(([\Gamma]\ahat_1) \arr [\Gamma]\ahat_2\big)} {By definition of substitution}
                    \continueeqPf{([\Omega][\Gamma]\ahat_1) \arr ([\Omega][\Gamma]\ahat_2)} {By definition of substitution}
                    \continueeqPf{([\Omega]\ahat_1) \arr ([\Omega]\ahat_2)} {By \Lemmaref{lem:substitution-monotonicity} (i)}
        \Hand       \declchkjudgPf{\OK}{[\Omega]\Delta} {[\Omega](\lam{x} e_0)} {[\Omega]\ahat}  {By above equality}
                  \end{llproof}                    

     \DerivationProofCase{\ArrElim}
           {\synjudg{q}{\Gamma}{e_0}{A}{\Theta}
             \\
             \recspinejudg{\Theta}{s_0}{A}{q}{C}{p}{\Delta}
           }
           {\synjudg{p}{\Gamma}{e_0\,s_0}{C}{\Delta}}
       
       \begin{llproof}
         \synjudgPf{q}{\Gamma}{e_0}{A}{\Theta}{Subderivation}
         \recspinejudgPf{\Theta}{s_0}{A}{q}{C}{p}{\Delta}{Subderivation}
         \substextendPf{\Gamma}{\Theta \AND \substextend{\Theta}{\Delta}}{By \Lemmaref{lem:typing-extension}}
         \substextendPf{\Delta}{\Omega}{Given}
         \substextendPf{\Theta}{\Omega}{By \Lemmaref{lem:extension-transitivity}}
         \substextendPf{\Gamma}{\Omega}{By \Lemmaref{lem:extension-transitivity}}
         \eqPf{[\Omega]\Gamma}{[\Omega]\Theta = [\Omega]\Delta}{By \Lemmaref{lem:completes-confluence}}
         \declsynjudgPf{q}{[\Omega]\Gamma}{[\Omega]e_0}{[\Omega]A}{By i.h.}
         \declsynjudgPf{q}{[\Omega]\Delta}{[\Omega]e_0}{[\Omega]A}{By above equality}
         \decolumnizePf
         \declrecspinejudgPf{[\Omega]\Theta}{[\Omega]s_0}{[\Omega]A}{q}{[\Omega]C}{p}{By i.h.}
         \proofsep
      \Hand  \declsynjudgPf{p}{[\Omega]\Delta}{[\Omega](e_0\,s_0)}{[\Omega]C}{By rule \DeclArrElim}
      \end{llproof}

     \DerivationProofCase{\Recover}
          {\spinejudg{\Gamma}{s}{A}{\p}{C}{\OK}{\Delta}
           \\
           \FEV{C} = \emptyset
          }
          {\recspinejudg{\Gamma}{s}{A}{\p}{C}{\p}{\Delta}}

       \begin{llproof}
         \spinejudgPf{\Gamma}{s}{A}{\p}{C}{\OK}{\Delta}{Subderivation}
         \declspinejudgPf{[\Omega]\Gamma}{[\Omega]s}{[\Omega]A}{\p}{[\Omega]C}{q}{By i.h.}
      \end{llproof}
      
      \smallskip

       We show the quantified premise of \DeclRecover, namely,
       \[
           \arrayenvbl{
             \text{for all $C'$.\;}
               \\
               \text{~~if~~}\declspinejudg{[\Omega]\Theta}{s}{[\Omega]A}{\p}{C'}{\OK}
               \text{~~then~}C' = [\Omega]C
           }
       \]
       Suppose we have $C'$ such that $\declspinejudg{[\Omega]\Gamma}{s}{[\Omega]A}{\p}{C'}{\OK}$.
       To apply \DeclRecover, we need to show $C' = [\Omega]C$.

       \begin{llproof}
         \declspinejudgPf{[\Omega]\Gamma}{[\Omega]s}{[\Omega]A}{\p}{C'}{\OK}    {Assumption}
         \substextendPf{\OmegaC}{\Omega}   {By \Lemmaref{lem:canonical-completion}}
         \eqPf{\dom{\OmegaC}}{\dom{\Gamma}}   {\ditto}
         \substextendPf{\Gamma}{\OmegaC}   {\ditto}
         \eqPf{[\Omega]\Gamma}{[\OmegaC]\Gamma}   {By \Lemmaref{lem:multiple-confluence}}
         \eqPf{[\Omega]A}{[\OmegaC]A}   {By \Lemmaref{lem:completing-completeness} (ii)}
         \declspinejudgPf{[\OmegaC]\Gamma}{[\Omega]s}{[\OmegaC]A}{\p}{C'}{\OK}    {By above equalities}
         \proofsep
         \spinejudgPf{\Gamma}{s}{[\Gamma]A}{\p}{C''}{q}{\Delta''}{By \Theoremref{thm:typing-completeness} (iii)}
         \substextendPf{\OmegaC}{\Omega''}{\ditto}
         \substextendPf{\Delta''}{\Omega''}{\ditto}
         \eqPf{C'}{[\Omega'']C''}{\ditto}
         \proofsep
         \spinejudgPf{\Gamma}{s}{[\Gamma]A}{\p}{C''}{q}{\Delta''}{Above}
         \eqPf{[\Gamma]A}{A}  {Given}   %
         \spinejudgPf{\Gamma}{s}{A}{\p}{C''}{q}{\Delta''}{By above equality}
         \spinejudgPf{\Gamma}{s}{A}{\p}{C}{\OK}{\Delta}{Subderivation}
         \proofsep
         \proofsep
         \eqPf{C''}{C    \AND  q = \OK   \AND  \Delta'' = \Delta}{By \Theoremref{thm:typing-det}}
         \proofsep
         \eqPf{C'}{[\Omega'']C''}{Above}
         \continueeqPf{[\Omega'']C}{By above equality}
         \continueeqPf{[\OmegaC]C}{By \Lemmaref{lem:completing-completeness} (ii)}
         \continueeqPf{[\Omega]C}{By \Lemmaref{lem:completing-completeness} (ii)}
      \end{llproof}

      We have thus shown the above ``for all $C'$. \dots'' statement.

      \begin{llproof}
      \Hand  \declrecspinejudgPf{[\Omega]\Gamma}{[\Omega]s}{[\Omega]A}{\p}{[\Omega]C}{\p}{By \DeclRecover}
      \end{llproof}

\clearpage

     \DerivationProofCase{\Pass}
           {
                 \spinejudg{\Gamma}{s}{A}{p}{C}{q}{\Delta}
                 \\
                 \big(
                     (p = \OK)
                     \OR (q = \p)
                     \OR (\FEV{C} \neq \emptyset)
                 \big)
           }
           {\recspinejudg{\Gamma}{s}{A}{p}{C}{q}{\Delta}}
       
       \begin{llproof}
         \spinejudgPf{\Gamma}{s}{A}{p}{C}{q}{\Delta}{Subderivation}
         \declspinejudgPf{[\Omega]\Gamma}{[\Omega]s}{[\Omega]A}{p}{[\Omega]C}{q}{By i.h.}
       \Hand  \declrecspinejudgPf{[\Omega]\Gamma}{[\Omega]s}{[\Omega]A}{p}{[\Omega]C}{q}{By \DeclPass}
       \end{llproof}

     \DerivationProofCase{\EmptySpine}
         {}
         {\spinejudg{\Gamma}{\emptyspine}{A}{p}{A}{p}{\Gamma}}
     
         \begin{llproof}
         \Hand  \declspinejudgPf{[\Omega]\Gamma}{\emptyspine}{[\Omega]A}{p}{[\Omega]A}{p}
               {By \DeclEmptySpine}
         \end{llproof}

     \DerivationProofCase{\ArrSpine}
         {\chkjudg{p}{\Gamma}{e_0}{A_1}{\Theta}
          \\
          \spinejudg{\Theta}{s_0}{[\Theta]A_2}{p}{C}{q}{\Delta}
         }
         {\spinejudg{\Gamma}{\appspine{e_0}{s_0}}{A_1 \arr A_2}{p}{C}{q}{\Delta}}

         \begin{llproof}
           \substextendPf{\Delta}{\Omega}   {Given}
           \substextendPf{\Theta}{\Delta}   {By \Lemmaref{lem:typing-extension}}
           \substextendPf{\Theta}{\Omega}   {By \Lemmaref{lem:extension-transitivity}}
           \proofsep
           \chkjudgPf{p}{\Gamma}{e_0}{A_1}{\Theta}   {Subderivation}
           \declchkjudgPf{p}{[\Omega]\Theta}{[\Omega]e_0}{[\Omega]A_1}   {By i.h.}
           \eqPf{[\Omega]\Theta}{[\Omega]\Delta}   {By \Lemmaref{lem:completes-confluence}}
           \declchkjudgPf{p}{[\Omega]\Delta}{[\Omega]e_0}{[\Omega]A_1}   {By above equality}
           \decolumnizePf
           \spinejudgPf{\Theta}{s_0}{[\Theta]A_2}{p}{C}{q}{\Delta}   {Subderivation}
           \declspinejudgPf{[\Omega]\Delta}{[\Omega]s_0}{[\Omega][\Theta]A_2}{p}{[\Omega]C}{q}   {By i.h.}
           \eqPf{[\Omega][\Theta]A_2}{[\Omega]A_2}  {\hspace{-8ex}By \Lemmaref{lem:substitution-monotonicity}}
           \declspinejudgPf{[\Omega]\Delta}{[\Omega]s_0}{[\Omega]A_2}{p}{[\Omega]C}{q}   {By above equality}
           \proofsep
           \declspinejudgPf{[\Omega]\Delta}{[\Omega](\appspine{e_0}{s_0})}{([\Omega]A_1) \arr [\Omega]A_2}{p}{[\Omega]C}{q}   {By \DeclArrSpine}
         \Hand  \declspinejudgPf{[\Omega]\Delta}{[\Omega](\appspine{e_0}{s_0})}{[\Omega](A_1 \arr A_2)}{p}{[\Omega]C}{q}   {By def.\ of subst.}
         \end{llproof}

     \DerivationProofCase{\SumIntro{k}}
           {\chkjudg{p}{\Gamma}{e_0}{A_k}{\Delta}}
           {\chkjudg{p}{\Gamma}{\inj{k}{e_0}}{A_1 + A_2}{\Delta}}

           \begin{llproof}
             \chkjudgPf{p}{\Gamma}{e_0}{A_k}{\Delta}   {Subderivation}
             \declchkjudgPf{p}{[\Omega]\Delta}{[\Omega]e_0}{[\Omega]A_k}   {By i.h.}
             \declchkjudgPf{p}{[\Omega]\Delta}{\inj{k} [\Omega]e_0}{([\Omega]A_1) + ([\Omega]A_2)}   {By \DeclSumIntro{k}}
           \Hand  \declchkjudgPf{p}{[\Omega]\Delta}{[\Omega](\inj{k} e_0)}{[\Omega](A_1 + A_2)}   {By def.\ of substitution}
           \end{llproof}

     \DerivationProofCase{\SumIntroSolve{k}}
            {\chkjudg{\OK}{\Gamma[\ahat_1 : \type, \ahat_2 : \type, \hypeq{\ahat:\type}{\ahat_1{+}\ahat_2}]}{e_0}{\ahat_k}{\Delta}}
            {\chkjudg{\OK}{\Gamma[\ahat : \type]}{\inj{k} e_0}{\ahat}{\Delta}}

           \begin{llproof}
             \chkjudgPf{\OK}{\Gamma[\dots, \hypeq{\ahat:\type}{\ahat_1{+}\ahat_2}]}{e_0}{\ahat_k}{\Delta}   {Subderivation}
             \declchkjudgPf{\OK}{[\Omega]\Delta}{[\Omega]e_0}{[\Omega]\ahat_k}   {By i.h.}
             \declsynjudgPf{\OK}{[\Omega]\Delta}{\inj{k} [\Omega]e_0}{([\Omega]\ahat_1) + ([\Omega]\ahat_2)}   {By \DeclSumIntro{k}}
             \eqPf{([\Omega]\ahat_1) + ([\Omega]\ahat_2)}{[\Omega]\ahat}   {Similar to the \ArrIntroSolve case (above)}
           \Hand  \declsynjudgPf{\OK}{[\Omega]\Delta}{[\Omega](\inj{k} e_0)}{[\Omega]\ahat}  {By above equality / def.\ of subst.}
           \end{llproof}

     \DerivationProofCase{\PairIntro}
           {\chkjudg{p}{\Gamma}{e_1}{A_1}{\Theta}
             \\
             \chkjudg{p}{\Theta}{e_2}{[\Theta]A_2}{\Delta}
           }
           {\chkjudg{p}{\Gamma}{\pair{e_1}{e_2}}{A_1 \times A_2}{\Delta}}

           \begin{llproof}
             \chkjudgPf{p}{\Theta}{e_2}{[\Theta]A_2}{\Delta}    {Subderivation}
             \substextendPf{\Theta}{\Delta}  {By \Lemmaref{lem:typing-extension}}
             \substextendPf{\Theta}{\Omega}  {By \Lemmaref{lem:extension-transitivity}}
             \proofsep
             \chkjudgPf{p}{\Gamma}{e_1}{A_1}{\Theta}    {Subderivation}
             \declchkjudgPf{p}{[\Omega]\Theta}{[\Omega]e_1}{[\Omega]A_1}   {By i.h.}
             \declchkjudgPf{p}{[\Omega]\Delta}{[\Omega]e_1}{[\Omega]A_1}   {By \Lemmaref{lem:completes-confluence}}
             \proofsep
             \chkjudgPf{p}{\Theta}{e_2}{[\Theta]A_2}{\Delta}    {Subderivation}
             \declchkjudgPf{p}{[\Omega]\Delta}{[\Omega]e_2}{[\Omega][\Theta]A_2}   {By i.h.}
             \judgetpPf{\Gamma}{A_1 \times A_2}  {Given}
             \judgetpPf{\Gamma}{A_2}  {By inversion}
             \substextendPf{\Gamma}{\Theta}  {By \Lemmaref{lem:typing-extension}}
             \judgetpPf{\Theta}{A_2}  {By \Lemmaref{lem:extension-weakening-tp}}
             \declchkjudgPf{p}{[\Omega]\Delta}{[\Omega]e_2}{[\Omega]A_2}   {By \Lemmaref{lem:substitution-monotonicity}}
             \proofsep
             \declchkjudgPf{p}{[\Omega]\Delta}{\pair{[\Omega]e_1}{[\Omega]e_2}}{([\Omega]A_1) \times [\Omega]A_2}   {By \DeclPairIntro}
           \Hand  \declchkjudgPf{p}{[\Omega]\Delta}{[\Omega]\pair{e_1}{e_2}}{[\Omega](A_1 \times A_2)}   {By def.\ of substitution}
           \end{llproof}

     \DerivationProofCase{\PairIntroSolve}
           {\chkjudg{\OK}{\Gamma[\ahat_2 : \type, \ahat_1 : \type, \hypeq{\ahat:\type}{\ahat_1{\times}\ahat_2}]}{e_1}{\ahat_1}{\Theta}
             \\
             \chkjudg{\OK}{\Theta}{e_2}{[\Theta]\ahat_2}{\Delta}
           }
           {\chkjudg{\OK}{\Gamma[\ahat : \type]}{\pair{e_1}{e_2}}{\ahat}{\Delta}}

           \begin{llproof}
             \substextendPf{\Delta}{\Omega}  {Given}
             \substextendPf{\Theta}{\Delta}  {By \Lemmaref{lem:typing-extension}}
             \substextendPf{\Theta}{\Omega}  {By \Lemmaref{lem:extension-transitivity}}
             \chkjudgPf{\OK}{\Gamma[\dots, \hypeq{\ahat:\type}{\ahat_1{\times}\ahat_2}]}{e_1}{\ahat_1}{\Theta}     {Subderivation}
             \declchkjudgPf{\OK}{[\Omega]\Theta}{[\Omega]e_1}{[\Omega]\ahat_1}     {By i.h.}
             \eqPf{[\Omega]\Theta}{[\Omega]\Delta}   {By \Lemmaref{lem:completes-confluence}}
             \declchkjudgPf{\OK}{[\Omega]\Delta}{[\Omega]e_1}{[\Omega]\ahat_1}     {By above equality}
             \proofsep
             \chkjudgPf{\OK}{\Theta}{e_2}{[\Theta]\ahat_2}{\Delta}    {Subderivation}
             \declchkjudgPf{\OK}{[\Omega]\Delta}{[\Omega]e_2}{[\Omega][\Theta]\ahat_2}    {By i.h.}
             \eqPf{[\Omega][\Theta]\ahat_2}{[\Omega]\ahat_2}  {By \Lemmaref{lem:substitution-monotonicity}}
             \declchkjudgPf{\OK}{[\Omega]\Delta}{[\Omega]e_2}{[\Omega]\ahat_2}    {By above equality}
             \decolumnizePf
             \declchkjudgPf{\OK}{[\Omega]\Delta}{\pair{[\Omega]e_1}{[\Omega]e_2}}{([\Omega]\ahat_1) \times [\Omega]\ahat_2}   {By \DeclPairIntro}
             \eqPf{([\Omega]\ahat_1) \times [\Omega]\ahat_2}{[\Omega]\ahat}   {Similar to the \ArrIntroSolve case (above)}
           \Hand  \declchkjudgPf{\OK}{[\Omega]\Delta}{[\Omega]\pair{e_1}{e_2}}{[\Omega]\ahat}   {By above equality}
           \end{llproof}

      \DerivationProofCase{\SolveSpine}
            {\spinejudg{\Gamma[\ahat_2 : \type, \ahat_1 : \type, \hypeq{\ahat:\type}{\ahat_1{\arr}\ahat_2}]}{\appspine{e_0}{s_0}}{(\ahat_1 \arr \ahat_2)}{\OK}{C}{\OK}{\Delta}}
            {\spinejudg{\Gamma[\ahat:\type]}{\appspine{e_0}{s_0}}{\ahat}{\OK}{C}{\OK}{\Delta}}

            \begin{llproof}
              \spinejudgPf{\Gamma[\dots, \hypeq{\ahat:\type}{\ahat_1{\arr}\ahat_2}]}{\appspine{e_0}{s_0}}{(\ahat_1 \arr \ahat_2)}{\OK}{C}{\OK}{\Delta}   {Subderivation}
              \declspinejudgPf{[\Omega]\Delta}{[\Omega](\appspine{e_0}{s_0})}{[\Omega](\ahat_1 \arr \ahat_2)}{\OK}{[\Omega]C}{\OK}   {By i.h.}
              \eqPf{[\Omega](\ahat_1 \arr \ahat_2)}{[\Omega]\ahat}    {Similar to the \ArrIntroSolve case}
            \Hand  \declspinejudgPf{[\Omega]\Delta}{[\Omega](\appspine{e_0}{s_0})}{[\Omega]\ahat}{\OK}{[\Omega]C}{\OK}   {By above equality}
            \end{llproof}

      \DerivationProofCase{\Case}
           {  \synjudg{q}{\Gamma}{e_0}{B}{\Theta}
              \\
              \matchjudg{q}{p}{\Theta}{\Pi}{[\Theta]B}{[\Theta]C}{\Delta}
              \\
              \covers{q}{\Delta}{\Pi}{[\Delta]B}
           }
           {\chkjudg{p}{\Gamma}{\case{e_0}{\Pi}}{C}{\Delta}}

           \begin{llproof}
               \synjudgPf{\p}{\Gamma}{e_0}{B}{\Theta}   {Subderivation}
               \substextendPf{\Theta}{\Delta}  {By \Lemmaref{lem:typing-extension}}
               \substextendPf{\Theta}{\Omega}  {By \Lemmaref{lem:extension-transitivity}}
               \declsynjudgPf{\p}{[\Omega]\Theta}{[\Omega]e_0}{[\Omega]B}   {By i.h.}
               \declsynjudgPf{\p}{[\Omega]\Delta}{[\Omega]e_0}{[\Omega]B}   {By \Lemmaref{lem:completes-confluence}}
               \proofsep
               \matchjudgPf{p}{\Theta}{\Pi}{[\Theta]B}{[\Theta]C}{\Delta}   {Subderivation}
               \proofsep
               \synjudgPf{\p}{\Gamma}{e_0}{B}{\Theta}   {Subderivation}
               \judgetpPf{\Theta}{B\;\p}   {By \Lemmaref{lem:wf-outputs} (Synthesis)}
               \proofsep
               \judgetpPf{\Gamma}{C\;p}   {Given}
               \substextendPf{\Gamma}{\Theta}  {By \Lemmaref{lem:typing-extension}}
               \judgetpPf{\Theta}{C\;p}   {By \Lemmaref{lem:extension-weakening-tpp}}
               \judgetpPf{\Theta}{[\Theta]C\;p}  {By \Lemmaref{lem:substitution-tpp}}
               \proofsep
               \declmatchjudgPf{p}{[\Omega]\Delta}{[\Omega]\Pi}{[\Omega]B}{[\Omega][\Theta]C}   {By i.h.\ (v)}
               \eqPf{[\Omega][\Theta]C}{[\Omega]C}   {By \Lemmaref{lem:substitution-monotonicity}}
               \declmatchjudgPf{p}{[\Omega]\Delta}{[\Omega]\Pi}{[\Omega]B}{[\Omega]C}   {By above equalities}
          \end{llproof}

          Assume $\Omega$ such that $\substextend{\Delta}{\Omega}$. \\
          Assume $D$ such that $\declsynjudg{q}{[\Omega]\Delta}{e}{D}$. \\
          Hence $\declsynjudg{q}{[\Omega]\Gamma}{e}{D}$. \\
          By \Theoremref{thm:typing-completeness}, there exist $B'$ and $\Theta'$ such that $\synjudg{q}{\Gamma}{e_0}{B'}{\Theta'}$ and $\substextend{\Omega}{\Omega'}$ and $D = [\Omega']B'$ and $B' = [\Theta']B'$. \\
          By \Lemmaref{thm:typing-det}, we know $\Theta' = \Theta$ and $B' = B$, which means $D = [\Omega][\Delta]B$. \\ 
          By \Lemmaref{thm:coverage-soundness}, $\covers{q}{[\Omega]\Delta}{[\Omega]\Pi}{[\Omega][\Delta]B}$. \\
          Hence $\covers{q}{[\Omega]\Delta}{[\Omega]\Pi}{D}$. \\
          By rule \DeclCase, \declchkjudg{p}{[\Omega]\Delta}{[\Omega]\case{e_0}{\Pi}}{[\Omega]C}
     \end{itemize}

     Part (v):

     \begin{itemize}
       \ProofCaseRule{\MatchEmpty}
           Apply rule \DeclMatchEmpty.

        \DerivationProofCase{\MatchBase}
              {\chkjudg{p}{\Gamma}{e}{C}{\Delta}
              }
              {\matchjudg{p}{\Gamma}{(\branch{\cdot}{e})}{\cdot}{C}{\Delta}}

              Apply the i.h.\ and \DeclMatchBase.

        \ProofCaseRule{\MatchUnit}
              Apply the i.h.\ and \DeclMatchUnit.

        \DerivationProofCase{\MatchSeq}
              {\matchjudg{q}{p}{\Gamma}{\pi}{\Avec}{C}{\Theta}
               \\
               \matchjudg{q}{p}{\Theta}{\Pi'}{\Avec}{C}{\Delta} }
              {\matchjudg{q}{p}{\Gamma}{\pi \alt \Pi'}{\Avec}{C}{\Delta}}

                Apply the i.h.\ to each premise, using lemmas for well-formedness under
              $\Theta$; then apply \DeclMatchSeq.

        \ProofCasesRules{\MatchExists, \MatchWith, \MatchWild, \MatchNil, \MatchCons}

              Apply the i.h.\ and the corresponding declarative match rule.

        \ProofCasesRules{\MatchPair, \MatchSum{k}}

              We have $\judgetpvec{\Gamma}{\Avec\;\p}$, so the first type in $\Avec$
              has no free existential variables. \\
              Apply the i.h.\ and the corresponding declarative match rule.

        \DerivationProofCase{\MatchNeg}
              { \notWithExists{A}  
                \\
                \matchjudg{q}{p}{\Gamma, \hyp{\p}{z}{A}}{\branch{\patvec}{e}'}{\Avec}{C}{\Delta, \hyp{\p}{z}{A}, \Delta'}
              }
              { \matchjudg{q}{p}{\Gamma}{\branch{z, \patvec}{e}}{A, \Avec}{C}{\Delta} }

              Construct $\Omega'$ and show $\substextend{\Delta, \hyp{\p}{z}{A}, \Delta'}{\Omega'}$ as in the \ArrIntro case. \\
              Use the i.h., then apply rule \DeclMatchNeg.
     \end{itemize}

      Part (vi):

     \begin{itemize}
        \DerivationProofCase{\MatchBot}
              {\elimeq{\Gamma}{\sigma}{\tau}{\sort}{\bot}}
              {\matchelimjudg{p}{\Gamma}{\patvec{p}{e}}{\sigma = \tau}{\Avec}{C}{\Gamma}}

              \begin{llproof}
                \elimeqPf{\Gamma}{\sigma}{\tau}{\sort}{\bot}
                         {Subderivation}
                \eqPf{[\Gamma](\sigma = \tau)}{(\sigma = \tau)}  {Given}
                \eqPf{(\sigma = \tau)}{[\Gamma](\sigma = \tau)}  {Given}
                \continueeqPf{[\Omega](\sigma = \tau)}  {By \Lemmaref{lem:substitution-monotonicity} (i)}
                \eqPf{\mgu{\sigma}{\tau}}{\bot}    {By \Lemmaref{lem:elimeq-soundness}}
                \eqPf{\mgu{[\Omega]\sigma}{[\Omega]\tau}}{\bot}    {By above equality}
                \decolumnizePf
              \Hand
                \declmatchelimjudgPf{p}{[\Omega]\Gamma}{[\Omega](\patvec{p}{e})}{[\Omega](\sigma = \tau)}{[\Omega]\Avec}{[\Omega]C}
                         {By \DeclMatchBot}
              \end{llproof}

        \DerivationProofCase{\MatchUnify}
              { \elimeq{\Gamma, \MonnierComma{P}}{\sigma}{\tau}{\sort}{\Gamma'}
                \\
                \matchjudg{q}{p}{\Gamma'}{\branch{\patvec}{e}}{\Avec}{C}{\Delta, \MonnierComma{P}, \Delta'}
              }
              {\matchelimjudg{p}{\Gamma}{\branch{\patvec}{e}}{\sigma = \tau}{\Avec}{C}{\Delta}}

              \begin{llproof}
                \elimeqPf{\Gamma, \MonnierComma{P}}{\sigma}{\tau}{\sort}{\Gamma'}
                         {Subderivation}
                \eqPf{(\sigma = \tau)}{[\Gamma](\sigma = \tau)}  {Given}
                \continueeqPf{[\Omega](\sigma = \tau)}  {By \Lemmaref{lem:substitution-monotonicity} (i)}
                \eqPfParenR{\Gamma'}{\Gamma, \MonnierComma{P}, \Theta}  {By \Lemmaref{lem:elimeq-soundness}}
                \eqPfParenR{\Theta}{(\hypeq{\alpha_1}{t_1}), \dots, (\hypeq{\alpha_n}{t_n})}
                           {\ditto}
                \eqPf{\theta}{\mgu{[\Omega]\sigma}{[\Omega]\tau}}    {\ditto}
                \eqPf{[\Omega, \MonnierComma{P}, \Theta]t'}
                     {[\theta][\Omega, \MonnierComma{P}]t'}
                     {\ditto ~~for all  $\judge{\Omega, \MonnierComma{P}}{t'}{\sort'}$}
                \decolumnizePf
                \matchjudgPf{p}{\Gamma, \MonnierComma{P}, \Theta}{\branch{\patvec}{e}}{\Avec}{C}{\Delta, \MonnierComma{P}, \Delta'}   {Subderivation}
                \proofsep
                \declmatchjudgPf{p}{[\Omega, \MonnierComma{P}, \Theta](\Delta, \MonnierComma{P}, \Delta')}{[\Omega, \MonnierComma{P}, \Theta](\branch{\patvec}{e})}{[\Omega, \MonnierComma{P}, \Theta]\Avec}{[\Omega, \MonnierComma{P}, \Theta]C}   {By i.h.}
                \decolumnizePf
                \eqPfParenL{\Omega, \MonnierComma{P}, \Theta}
                           {[\theta](\Omega, \MonnierComma{P})}
                     {By \Lemmaref{lem:subst-upgrade} (iii)}
                \eqPf{[\Omega, \MonnierComma{P}, \Theta]\Avec}{[\theta][\Omega, \MonnierComma{P}]\Avec}
                     {By \Lemmaref{lem:subst-upgrade} (i)}
                \eqPf{[\Omega, \MonnierComma{P}, \Theta]C}{[\theta][\Omega, \MonnierComma{P}]C}
                     {By \Lemmaref{lem:subst-upgrade} (i)}
                \eqPf{[\Omega, \MonnierComma{P}, \Theta](\branch{\patvec}{e})}{[\theta][\Omega](\branch{\patvec}{e})}
                     {By \Lemmaref{lem:subst-upgrade} (iv)}
                \decolumnizePf
                \declmatchjudgPf{p}{\theta([\Omega, \MonnierComma{P}]\Gamma)}{[\theta][\Omega](\branch{\patvec}{e})}{\theta([\Omega, \MonnierComma{P}]\Avec)}{\theta([\Omega, \MonnierComma{P}]C)}
                             {By above equalities}
                \declmatchjudgPf{p}{\theta([\Omega]\Gamma)}{[\theta][\Omega](\branch{\patvec}{e})}{\theta([\Omega]\Avec)}{\theta([\Omega]C)}
                             {Subst.\ not affected by $\MonnierComma{P}$}
                \decolumnizePf
              \Hand
                \declmatchelimjudgPf{p}{[\Omega]\Gamma}{[\Omega](\branch{\patvec}{e})}{[\Omega](\sigma = \tau)}{[\Omega]\Avec}{[\Omega]C}
                         {By \DeclMatchUnify}
              \end{llproof}
     \qedhere
     \end{itemize}
\end{proof}

\section{Completeness}

\label{sec:last-section}

\subsection{Completeness of Auxiliary Judgments}

\instantiationcompleteness*
\begin{proof}\proofheading
  By induction on $\tau$.
  
  We have $\declsubjudg[\polvar]{[\Omega]\Gamma}{[\Omega]\ahat}{[\Omega]A}$.
  We now case-analyze the shape of $\tau$. 

    \begin{itemize}
      \ProofCaseRule{$\tau = \bhat$}

          \begin{llproof}
            \notinPf{\ahat}{\FV{\bhat}}   {Given}
            \neqPf{\ahat}{\bhat}   {From definition of $\FV{-}$}
            \inPf{\bhat}{\unsolved{\Gamma}}  {From $[\Gamma]\bhat = \bhat$}
            \LetPf{\Omega'}{\Omega}  {}
          \Hand  \substextendPf{\Omega}{\Omega'}  {By \Lemmaref{lem:extension-reflexivity}}
          \end{llproof}

          Now consider whether $\ahat$ is declared to the left of $\bhat$, or vice versa. 

          \begin{itemize}
            \ProofCaseRule{$\Gamma = \Gamma_0[\ahat:\sort][\bhat:\sort]$}
            
                \begin{llproof}
                  \LetPf{\Delta}{\Gamma_0[\ahat:\sort][\hypeq{\bhat:\sort}{\ahat}]}   {}
                  \instjudgPf{\Gamma}{\ahat}{\bhat}{\sort}{\Delta}   {By \InstReach}
                  \eqPf{[\Omega]\ahat}{[\Omega]\bhat}   {Given}
                  \substextendPf{\Gamma}{\Omega}   {Given}
                \Hand  \substextendPf{\Delta}{\Omega}   {By \Lemmaref{lem:parallel-extension-solve}}
                \Hand  \eqPf{\dom{\Delta}}{\dom{\Omega'}}
                            {$\dom{\Delta} = \dom{\Gamma}$ and $\dom{\Omega'} = \dom{\Omega}$}
                \end{llproof}

            \ProofCaseRule{$(\Gamma = \Gamma_0[\bhat:\sort][\ahat:\sort]$}

                Similar, but using \InstSolve instead of \InstReach.
          \end{itemize}

      \ProofCaseRule{$\tau = \alpha$}

         We have $[\Omega]\ahat = \alpha$, so (since $\Omega$ is well-formed),
           $\alpha$ is declared to the left of $\ahat$ in $\Omega$. \\ 
         We have $\substextend{\Gamma}{\Omega}$. \\
         By \Lemmaref{lem:reverse-declaration-order-preservation},
           we know that $\alpha$ is declared to the left of $\ahat$ in $\Gamma$;
           that is, $\Gamma = \Gamma_L[\alpha:\sort][\ahat:\sort]$. \\
         Let $\Delta = \Gamma_L[\alpha:\sort][\hypeq{\ahat:\sort}{\alpha}]$ and $\Omega' = \Omega$. \\
         By \InstSolve, $\instjudg{\Gamma_L[\alpha:\sort][\ahat:\sort]}{\ahat}{\alpha}{\sort}{\Delta}$. \\
         By \Lemmaref{lem:parallel-extension-solve}, $\substextend{\Gamma_L[\alpha:\sort][\hypeq{\ahat:\sort}{\alpha}]}{\Omega}$. 
\\
         We have $\dom{\Delta} = \dom{\Gamma}$ and $\dom{\Omega'} = \dom{\Omega}$;
         therefore,  $\dom{\Delta} = \dom{\Omega'}$.

      \ProofCaseRule{$\tau = \unitty$}

          Similar to the $\tau = \alpha$ case, but without having to reason about where $\alpha$ is declared.

      \ProofCaseRule{$\tau = \zero$}

          Similar to the $\tau = \unitty$ case.

      \ProofCaseRule{$\tau = \tau_1 \binc \tau_2$}

          \begin{llproof}
            \eqPf{[\Omega]\ahat}{[\Omega](\tau_1 \binc \tau_2)}   {Given}
            \continueeqPf {([\Omega]\tau_1) \binc ([\Omega]\tau_2)}  {By definition of substitution}
            \proofsep
            \eqPf{\tau_1 \binc \tau_2}{[\Gamma](\tau_1 \binc \tau_2)}   {Given}
            \eqPf{\tau_1}{[\Gamma]\tau_1}   {By definition of substitution and congruence}
            \eqPf{\tau_2}{[\Gamma]\tau_2}   {Similarly}
            \proofsep
            \notinPf{\ahat}{\FV{\tau_1 \binc \tau_2}}   {Given}
            \notinPf{\ahat}{\FV{\tau_1}}   {From definition of $\FV{-}$}
            \notinPf{\ahat}{\FV{\tau_2}}   {Similarly}
            \proofsep
            \eqPf{\Gamma}{\Gamma_0[\ahat:\type]}  {By $\ahat \in \unsolved{\Gamma}$}
            \substextendPf{\Gamma}{\Omega}   {Given}
            \substextendPf{\Gamma_0[\ahat:\type]}{\Gamma_0[\ahat_2:\type, \ahat_1:\type, \ahat:\type]}
                  {By \Lemmaref{lem:deep-existential} (i) twice}
            \judgePf{\dots, \ahat_2, \ahat_1} {\ahat_1 \binc \ahat_2}  {\type}   {Straightforward}
            \substextendPf{\Gamma_0[\ahat_2, \ahat_1, \ahat]}
                          {\Gamma_0[\ahat_2, \ahat_1, \hypeq{\ahat}{\ahat_1 \binc \ahat_2}]}
                          {By \Lemmaref{lem:deep-existential} (ii)}
             \substextendPf{\Gamma_0[\ahat]}
                           {\Gamma_0[\ahat_2, \ahat_1, \hypeq{\ahat}{\ahat_1 \binc \ahat_2}]}
                           {By \Lemmaref{lem:extension-transitivity}}
          \end{llproof}

          (In the last few lines above, and the rest of this case, we omit the ``$: \type$'' annotations in contexts.)

          Since $\ahat \in \unsolved{\Gamma}$ and $\substextend{\Gamma}{\Omega}$, we know that
          $\Omega$ has the form $\Omega_0[\hypeq{\ahat}{\tau_0}]$.

          To show that we can extend this context, we apply \Lemmaref{lem:deep-existential} (iii) twice to introduce 
          $\hypeq{\ahat_2}{\tau_2}$ and  $\hypeq{\ahat_1}{\tau_1}$,
          and then \Lemmaref{lem:parallel-extension-update} to overwrite $\tau_0$:
          \[
              \substextend{\underbrace{\Omega_0[\hypeq{\ahat}{\tau_0}]}_{\Omega}~}
                          {~\Omega_0[\hypeq{\ahat_2}{\tau_2},
                                    \hypeq{\ahat_1}{\tau_1},
                                    \hypeq{\ahat}{\ahat_1 \binc \ahat_2}]}
          \]
          We have $\substextend{\Gamma}{\Omega}$, that is,
          \[
                  \substextend{\Gamma_0[\ahat]}{\Omega_0[\hypeq{\ahat}{\tau_0}]}
          \]
          By \Lemmaref{lem:parallel-admissibility} (i) twice, inserting unsolved variables $\ahat_2$ and $\ahat_1$
          on both contexts in the above extension preserves extension:

          {\small
              \begin{llproof}
                      \substextendPf{\Gamma_0[\ahat_2, \ahat_1, \ahat]}
                                  {\Omega_0[\hypeq{\ahat_2}{\tau_2}, \hypeq{\ahat_1}{\tau_1}, \hypeq{\ahat}{\tau_0}]}
                                  {By \Lemmaref{lem:parallel-admissibility} (ii) twice}
                      \substextendPf{\underbrace{\Gamma_0[\ahat_2, \ahat_1, \hypeq{\ahat}{\ahat_1{\binc}\ahat_2}]}_{\Gamma_1}}
                                  {\underbrace{\Omega_0[\hypeq{\ahat_2}{\tau_2}, \hypeq{\ahat_1}{\tau_1}, \hypeq{\ahat}{\ahat_1 {\binc} \ahat_2}]}_{\Omega_1}}
                                  {By \Lemmaref{lem:parallel-extension-update}}
              \end{llproof}
          }

          Since $\ahat \notin \FV{\tau}$, it follows that $[\Gamma_1]\tau = [\Gamma]\tau = \tau$. \\ 
          Therefore $\ahat_1 \notin \FV{\tau_1}$ and $\ahat_1, \ahat_2 \notin \FV{\tau_2}$. \\ 
          By \Lemmaref{lem:completing-completeness} (i) and (iii),
            $[\Omega_1]\Gamma_1 = [\Omega]\Gamma$ and $[\Omega_1]\ahat_1 = \tau_1$.  \\ 
          By i.h., there are $\Delta_2$ and $\Omega_2$ such that 
          $\instjudg{\Gamma_1}{\ahat_1}{\tau_1}{\sort}{\Delta_2}$ and $\substextend{\Delta_2}{\Omega_2}$ and $\substextend{\Omega_1}{\Omega_2}$. 

          Next, note that $[\Delta_2][\Delta_2]\tau_2 = [\Delta_2]\tau_2$. \\ 
          By \Lemmaref{lem:left-unsolvedness-preservation}, we know that $\ahat_2 \in \unsolved{\Delta_2}$. \\ 
          By \Lemmaref{lem:left-free-variable-preservation}, we know that $\ahat_2 \notin \FV{[\Delta_2]\tau_2}$. \\ 
          By \Lemmaref{lem:extension-transitivity}, $\substextend{\Omega}{\Omega_2}$. \\
          We know $[\Omega_2]\Delta_2 = [\Omega]\Gamma$ because: 
          \begin{mathpar}
            \begin{array}{lcll}
              [\Omega_2]\Delta_2 
              & = & [\Omega_2]\Omega_2 & \text{By \Lemmaref{lem:completes-stability}} \\ 
              & = & [\Omega]\Omega & \text{By \Lemmaref{lem:completing-completeness} (iii)} \\
              & = & [\Omega]\Gamma & \text{By \Lemmaref{lem:completes-stability}}
            \end{array}
          \end{mathpar}
          By \Lemmaref{lem:completing-completeness} (i), we know that $[\Omega_2]\ahat_2 = [\Omega_1]\ahat_2 = \tau_2$. \\ 
          By \Lemmaref{lem:completing-completeness} (i), we know that $[\Omega_2]\tau_2 = [\Omega]\tau_2$. \\ 
          Hence we know that
          $\declsubjudg[\polvar]{[\Omega_2]\Delta_2}{[\Omega_2]\ahat_2}{[\Omega_2]\tau_2}$.
          \\ 
          By i.h., we have $\Delta$ and $\Omega'$ such that 
          $\instjudg{\Delta_2}{\ahat_2}{[\Delta_2]\tau_2}{\sort}{\Delta}$ and $\substextend{\Omega_2}{\Omega'}$ and $\substextend{\Delta}{\Omega'}$.  \\
          By rule \InstBin, $\instjudg{\Gamma}{\ahat}{\tau}{\sort}{\Delta}$. \\ 
          By \Lemmaref{lem:extension-transitivity}, $\substextend{\Omega}{\Omega'}$. 

    \ProofCaseRule{$\tau = \succ{\tau_0}$}
        
          Similar to the $\tau = \tau_1 \binc \tau_2$ case, but simpler.
  \qedhere
  \end{itemize}  
\end{proof}

\checkeqcompleteness*
\begin{proof}\proofheading
  By mutual induction on the sizes of $[\Gamma]\sigma$ and $[\Gamma]\tau$.
  
  We distinguish cases of $[\Gamma]\sigma$ and $[\Gamma]\tau$.
  
  \begin{itemize}
      \ProofCaseRule{$[\Gamma]\sigma = [\Gamma]\tau = \unitty$}

          \begin{llproof}
          \Hand  \checkeqPf{\Gamma}{\unitty}{\unitty}{\type}{\underbrace{\Gamma}_{\Delta}}   {By \CheckeqUnit}
                 \LetPf{\Omega'}{\Omega}   {}
                 \substextendPf{\Gamma}{\Omega}   {Given}
          \Hand  \substextendPf{\Delta}{\Omega'}  {$\Delta = \Gamma$ and $\Omega' = \Omega$}
          \Hand  \eqPf{\dom{\Gamma}}{\dom{\Omega}}  {Given}
          \Hand  \substextendPf{\Omega}{\Omega'}  {By \Lemmaref{lem:extension-reflexivity}}
          \end{llproof}

      \ProofCaseRule{$[\Gamma]\sigma = [\Gamma]t = \zero$}

          Similar to the case for $\unitty$, applying \CheckeqZero instead of \CheckeqUnit.

      \ProofCaseRule{$[\Gamma]\sigma = [\Gamma]t = \alpha$}

          Similar to the case for $\unitty$, applying \CheckeqVar instead of \CheckeqUnit.

      \ProofCaseRule{$[\Gamma]\sigma = \ahat$ and $[\Gamma]t = \bhat$}
      
          \begin{itemize}
            \item If $\ahat = \bhat$:  Similar to the case for $\unitty$, applying \CheckeqVar instead of \CheckeqUnit.

            \item If $\ahat \neq \bhat$:

              \begin{llproof}
                \substextendPf{\Gamma}{\Omega}   {Given}
                \notinPf{\ahat}{\FV{\underbrace{\bhat}_{[\Gamma]t}}}   {By definition of $\FV{-}$}
                \eqPf{[\Omega]\sigma}{[\Omega]t}   {Given}
                \eqPf{[\Omega][\Gamma]\sigma}{[\Omega][\Gamma]t}   {By \Lemmaref{lem:substitution-monotonicity} (i) twice}
                \eqPf{[\Omega]\ahat}{[\Omega][\Gamma]t}   {$[\Gamma]\sigma = \ahat$}
                \eqPf{\dom{\Gamma}}{\dom{\Omega}}   {Given}
                \instjudgPf{\Gamma}{\ahat}{[\Gamma]t}{\sort}{\Delta}   {By \Lemmaref{lem:instantiation-completeness}}
              \Hand  \substextendPf{\Omega}{\Omega'}   {\ditto}
              \Hand  \substextendPf{\Delta}{\Omega}   {\ditto}
              \Hand  \eqPf{\dom{\Delta}}{\dom{\Omega'}}   {\ditto}
                \proofsep
              \Hand   \checkeqPf{\Gamma}{\ahat}{[\Gamma]t}{\sort}{\Delta}   {By \CheckeqInstL}
              \end{llproof}
          \end{itemize}

      \ProofCaseRule{$[\Gamma]\sigma = \ahat$ and $[\Gamma]t = \unitty$ or $\zero$ or $\alpha$}

            Similar to the previous case, except:

            \begin{llproof}
                \notinPf{\ahat}{\FV{\underbrace{\unitty}_{[\Gamma]t}}}   {By definition of $\FV{-}$}
            \end{llproof}
            
            and similarly for $\unitty$ and $\alpha$.

      \ProofCaseRule{$[\Gamma]t = \ahat$ and $[\Gamma]\sigma = \unitty$ or $\zero$ or $\alpha$}
          Symmetric to the previous case.

      \ProofCaseRule{$[\Gamma]\sigma = \ahat$ and $[\Gamma]t = \succ{[\Gamma]t_0}$}
          
          If $\ahat \notin \FV{[\Gamma]t_0}$, then $\ahat \notin \FV{[\Gamma]t}$.
          Proceed as in the previous several cases.

          The other case, $\ahat \in \FV{[\Gamma]t_0}$, is impossible:

          $~~~$\begin{tabular}[t]{l}
                      We have $\ahat \subtermof [\Gamma]t_0$.  \\
                      Therefore $\ahat \propersubtermof \succ{[\Gamma]t_0}$, that is, $\ahat \propersubtermof [\Gamma]t$. \\
                      By a property of substitutions, $[\Omega]\ahat \propersubtermof [\Omega][\Gamma]t$. \\
                      Since $\substextend{\Gamma}{\Omega}$,
                        by \Lemmaref{lem:substitution-monotonicity} (i), $[\Omega][\Gamma]t = [\Omega]t$, so
                       $[\Omega]\ahat \propersubtermof [\Omega]t$. \\
                      But it is given that $[\Omega]\ahat = [\Omega]t$, a contradiction.            
          \end{tabular}
      
      \ProofCaseRule{$[\Gamma]t = \ahat$ and $[\Gamma]\sigma = \succ{[\Gamma]\sigma_0}$}
          Symmetric to the previous case.

      \ProofCaseRule{$[\Gamma]\sigma = [\Gamma]\sigma_1 \binc [\Gamma]\sigma_2$
        and $[\Gamma]t = [\Gamma]t_1 \binc [\Gamma]t_2$}
         
          \begin{llproof}
            \substextendPf{\Gamma}{\Omega}   {Given}
            \checkeqPf{\Gamma}{[\Gamma]\sigma_1}{[\Gamma]t_1} {\type} {\Theta}   {By i.h.}
            \substextendPf{\Theta}{\Omega_0}  {\ditto}
            \substextendPf{\Omega}{\Omega_0}  {\ditto}
            \eqPf{\dom{\Theta}}{\dom{\Omega_0}}  {\ditto}
            \proofsep
            \checkeqPf{\Theta}{[\Theta][\Gamma]\sigma_2}{[\Theta][\Gamma]t_2} {\type} {\Delta}   {By i.h.}
          \Hand  \substextendPf{\Delta}{\Omega'}  {\ditto}
            \substextendPf{\Omega_0}{\Omega'}  {\ditto}
          \Hand  \eqPf{\dom{\Delta}}{\dom{\Omega'}}  {\ditto}
          \Hand  \substextendPf{\Omega}{\Omega'}  {By \Lemmaref{lem:extension-transitivity}}
            \proofsep
          \Hand  \checkeqPf{\Gamma}{[\Gamma]\sigma_1 \binc [\Gamma]\sigma_2}{[\Gamma]t_1 \binc [\Gamma]t_2)} {\type} {\Delta}   {By \CheckeqBin}
          \end{llproof}

      \ProofCaseRule{$[\Gamma]\sigma = \ahat$ and $[\Gamma]t = t_1 \binc t_2$}
         Similar to the $\ahat$/$\succ{-}$ case, showing the impossibility of $\ahat \in \FV{[\Gamma]t_k}$
         for $k = 1$ and $k = 2$.

      \ProofCaseRule{$[\Gamma]t = \ahat$ and $[\Gamma]\sigma = \sigma_1 \binc \sigma_2$}
         Symmetric to the previous case.
  \qedhere
  \end{itemize}
\end{proof}

\elimeqcompleteness*
\begin{proof}\proofheading
  By induction on the structure of $[\Gamma]\sigma$ and $[\Gamma]t$.

  \begin{itemize}
  \item Case $[\Omega]\sigma = t = \zero$:

    \begin{llproof}
      \eqPf{\mgu{\zero}{\zero}}{\cdot}{By properties of unification}
      \elimeqPf{\Gamma}{\zero}{\zero}{\ind}{\Gamma}{By rule $\ElimeqZero$}
    \Hand   \elimeqPf{\Gamma}{\zero}{\zero}{\ind}{\Gamma, \Delta}{where $\Delta = \cdot$}
        \Hand \supposePf{\judge{\Gamma}{u}{\sort'}}
      \eqPf{[\Gamma, \Delta]u}{[\Gamma]u}{where $\Delta = \cdot$}
      \continueeqPf{\theta([\Gamma]u)}{where $\theta$ is the identity}
    \end{llproof}

  \item Case $\sigma = \succ{\sigma'}$ and $t = \succ{t'}$:

    \begin{itemize}
    \item Case $\mgu{\succ{\sigma'}}{\succ{t'}} = \theta$:

    \begin{llproof}
      \eqPf{\mgu{\sigma'}{t'}}{\mgu{\succ{\sigma'}}{\succ{t'}} = \theta}{By properties of unification}
      \eqPf{\succ{\sigma'}}{[\Gamma]\succ{\sigma'}}{Given}
      \continueeqPf{\succ{[\Gamma]\sigma'}}{By definition of substitution}
      \eqPf{\sigma'}{[\Gamma]\sigma'}{By injectivity of successor}
      \eqPf{\succ{t'}}{[\Gamma]\succ{t'}}{Given}
      \continueeqPf{\succ{[\Gamma]t'}}{By definition of substitution}
      \eqPf{t'}{[\Gamma]t'}{By injectivity of successor}
      \elimeqPf{\Gamma}{\sigma'}{t'}{\ind}{\Gamma, \Delta}{By i.h.}
      \Hand \eqPf{[\Gamma, \Delta]u}{\theta([\Gamma]u)~\text{for all $u$ such that \dots}}{\ditto}
      \Hand \elimeqPf{\Gamma}{\succ{\sigma'}}{\succ{t'}}{\ind}{\Gamma, \Delta}{By rule \ElimeqSucc}
    \end{llproof}

    \item Case $\mgu{\succ{\sigma'}}{\succ{t'}} = \bot$:

    \begin{llproof}
      \eqPf{\mgu{\sigma'}{t'}}{\mgu{\succ{\sigma'}}{\succ{t'}} = \bot}{By properties of unification}
      \eqPf{\succ{\sigma'}}{[\Gamma]\succ{\sigma'}}{Given}
      \continueeqPf{\succ{[\Gamma]\sigma'}}{By definition of substitution}
      \eqPf{\sigma'}{[\Gamma]\sigma'}{By injectivity of successor}
      \eqPf{\succ{t'}}{[\Gamma]\succ{t'}}{Given}
      \continueeqPf{\succ{[\Gamma]t'}}{By definition of substitution}
      \eqPf{t'}{[\Gamma]t'}{By injectivity of successor}
      \elimeqPf{\Gamma}{\sigma'}{t'}{\ind}{\bot}{By i.h.}
      \Hand \elimeqPf{\Gamma}{\succ{\sigma'}}{\succ{t'}}{\ind}{\bot}{By rule \ElimeqSucc}
    \end{llproof}
  \end{itemize}

  \item Case $\sigma = \sigma_1 \binc \sigma_2$ and $t = t_1 \binc t_2$:
    
    First we establish some properties of the subterms:

    \begin{llproof}
        \eqPf{\sigma_1 \binc \sigma_2}{[\Gamma](\sigma_1 \binc \sigma_2)}{Given}
        \continueeqPf{[\Gamma]\sigma_1 \binc [\Gamma]\sigma_2}{By definition of substitution}
        \Hand \eqPf{[\Gamma]\sigma_1}{\sigma_1}{By injectivity of $\binc$}
        \Hand \eqPf{[\Gamma]\sigma_2}{\sigma_2}{By injectivity of $\binc$}
        \eqPf{t_1 \binc t_2}{[\Gamma](t_1 \binc t_2)}{Given}
        \continueeqPf{[\Gamma]t_1 \binc [\Gamma]t_2}{By definition of substitution}
        \Hand \eqPf{[\Gamma]t_1}{t_1}{By injectivity of $\binc$}
        \Hand \eqPf{[\Gamma]t_2}{t_2}{By injectivity of $\binc$}
    \end{llproof}

    \begin{itemize}
    \item Subcase $\mgu{\sigma}{t} = \bot$:
      \begin{itemize}
      \item Subcase $\mgu{\sigma_1}{t_1} = \bot$:
        
        \begin{llproof}
          \elimeqPf{\Gamma}{\sigma_1}{t_1}{\sort}{\bot}{By i.h.}
          \elimeqPf{\Gamma}{\sigma_1 \binc \sigma_2}{t_1 \binc t_2}{\sort}{\bot}{By rule \ElimeqBinBot}
        \end{llproof}

      \item Subcase $\mgu{\sigma_1}{t_1} = \theta_1$ and $\mgu{\theta_1(\sigma_2)}{\theta_1(t_2)} = \bot$:

        \begin{llproof}
          \elimeqPf{\Gamma}{\sigma_1}{t_1}{\sort}{\Gamma, \Delta_1}{By i.h.}
          \eqPf{[\Gamma, \Delta_1]u}{\theta_1([\Gamma]u)~\text{for all $u$ such that \dots}}{\ditto}
          \decolumnizePf
          \eqPf{[\Gamma, \Delta_1]\sigma_2}{\theta_1([\Gamma]\sigma_2)}{Above line with $\sigma_2$ as $u$}
          \continueeqPf{\theta_1(\sigma_2)}{$[\Gamma]\sigma_2 = \sigma_2$}
          \eqPf{[\Gamma, \Delta_1]t_2}{\theta_1([\Gamma]t_2)}{Above line with $t_2$ as $u$}
          \continueeqPf{\theta_1(t_2)}{Since $[\Gamma]\sigma_2 = \sigma_2$}
          \eqPf{\mgu{[\Gamma, \Delta_1]\sigma_2}{[\Gamma, \Delta_1]t_2}}{\theta_2}{By transitivity of equality}
          \decolumnizePf
          \eqPf{[\Gamma, \Delta_1][\Gamma, \Delta_1]\sigma_2}{[\Gamma, \Delta_1]\sigma_2}{By \Lemmaref{lem:substitution-monotonicity}}
          \eqPf{[\Gamma, \Delta_1][\Gamma, \Delta_1]t_2}{[\Gamma, \Delta_1]t_2}{By \Lemmaref{lem:substitution-monotonicity}}
          \decolumnizePf
          \elimeqPf{\Gamma, \Delta_1}{[\Gamma, \Delta_1]\sigma_2}{[\Gamma, \Delta_1]t_2}{\sort}{\bot}{By i.h.}
          \Hand \elimeqPf{\Gamma}{\sigma_1 \binc \sigma_2}{t_1 \binc t_2}{\sort}{\bot}{By rule \ElimeqBin}
          
        \end{llproof}
      \end{itemize}

    \item Subcase $\mgu{\sigma}{t} = \theta$:

      \begin{llproof}
        \eqPf{\mgu{\sigma_1 \binc \sigma_2}{t_1 \binc t_2}}{\theta = \theta_2 \circ \theta_1}{By properties of unifiers}
        \eqPf{\mgu{\sigma_1}{t_1}}{\theta_1}{\ditto}
        \eqPf{\mgu{\theta_1(\sigma_2)}{\theta_1(t_2)}}{\theta_2}{\ditto}
        \elimeqPf{\Gamma}{\sigma_1}{t_1}{\sort}{\Gamma, \Delta_1}{By i.h.}
*~~        \eqPf{[\Gamma, \Delta_1]u}
                  {\theta_1([\Gamma]u)
                    \text{~for all $u$ such that \dots}
                   }
                   {\ditto}
        \PfLEFT{4}
        \eqPf{[\Gamma, \Delta_1]\sigma_2}{\theta_1([\Gamma]\sigma_2)}{Above line with $\sigma_2$ as $u$}
        \continueeqPf{\theta_1(\sigma_2)}{$[\Gamma]\sigma_2 = \sigma_2$}
        \eqPf{[\Gamma, \Delta_1]t_2}{\theta_1([\Gamma]t_2)}{Above line with $t_2$ as $u$}
        \continueeqPf{\theta_1(t_2)}{$[\Gamma]\sigma_2 = \sigma_2$}
        \eqPf{\mgu{[\Gamma, \Delta_1]\sigma_2}{[\Gamma, \Delta_1]t_2}}{\theta_2}{By transitivity of equality}
        \decolumnizePf
        \eqPf{[\Gamma, \Delta_1][\Gamma, \Delta_1]\sigma_2}{[\Gamma, \Delta_1]\sigma_2}{By \Lemmaref{lem:substitution-monotonicity}}
        \eqPf{[\Gamma, \Delta_1][\Gamma, \Delta_1]t_2}{[\Gamma, \Delta_1]t_2}{By \Lemmaref{lem:substitution-monotonicity}}
        \decolumnizePf
        \elimeqPf{\Gamma, \Delta_1}{[\Gamma, \Delta_1]\sigma_2}{[\Gamma, \Delta_1]t_2}{\sort}{\Gamma, \Delta_1, \Delta_2}{By i.h.}
**~~~~~        \eqPf{[\Gamma, \Delta_1, \Delta_2]u'}
                   {\theta_2([\Gamma, \Delta_1]u')
                     \text{~for all $u'$ such that \dots}
                   }
                   {\ditto}
        \Hand \elimeqPf{\Gamma}{\sigma_1 \binc \sigma_2}{t_1 \binc t_2}{\sort}{\Gamma, \Delta_1, \Delta_2}{By rule \ElimeqBin}
        \decolumnizePf
        \Hand \supposePf{\judge{\Gamma}{u}{\sort'}}
        \eqPf{[\Gamma, \Delta_1, \Delta_2]u}{\theta_2([\Gamma, \Delta_1]u)}{By **}
        \continueeqPf{\theta_2(\theta_1([\Gamma]u))}{By *}
        \continueeqPf{\theta([\Gamma]u)}{$\theta = \theta_2 \circ \theta_1$}
      \end{llproof}
    \end{itemize}

  \item Case $\sigma = \alpha$:

    \begin{itemize}
      \item Subcase $\alpha \in \FV{t}$:

        \begin{llproof}
          \eqPf{\mgu{\alpha}{t}}{\bot}{By properties of unification}
          \Hand \elimeqPf{\Gamma}{\alpha}{t}{\sort}{\bot}{By rule \ElimeqUvarLBot}
        \end{llproof}

      \item Subcase $\alpha \notin \FV{t}$:

        \begin{llproof}
          \eqPf{\mgu{\alpha}{t}}{[t/\alpha]}{By properties of unification}
          \notinPf{(\alpha = t')}{\Gamma}{$[\Gamma]\alpha = \alpha$}
          \Hand \elimeqPf{\Gamma}{\alpha}{t}{\sort}{\Gamma, \hypeq{\alpha}{t}}{By rule \ElimeqUvarL}
          \decolumnizePf
          \Hand \supposePf{\judge{\Gamma}{u}{\sort'}}
          \eqPf{[\Gamma, \hypeq{\alpha}{t}]u}{[\Gamma]([t/\alpha]u)}{By definition of substitution}
          \continueeqPf{\big[[\Gamma]t/\alpha\big][\Gamma]u}{By properties of substitution}
          \continueeqPf{[t/\alpha][\Gamma]u}{$[\Gamma]t = t$}
        \end{llproof}
    \end{itemize}

  \item Case $t = \alpha$:
    Similar to previous case. 
  \qedhere
  \end{itemize}
\end{proof}

\substupgrade*
\begin{proof}
  Part (i): By induction on the given derivation, using the given ``for all'' at the leaves.
  
  Part (ii): By induction on the given derivation, using
  part (i) in the \substextendVV case.

  Part (iii): By induction on $\Delta$.  In the base case ($\Delta = \cdot$), use part (ii).
  Otherwise, use the i.h.\ and the definition of context substitution.

  Part (iv): By induction on $e$, using part (i) in the $e = (e_0 : A)$ case.
\end{proof}

\propequivcompleteness*
\begin{proof}\proofheading
  By induction on the given derivations.  There is only one possible case:

  \begin{itemize}
  \DoubleDerivationProofCase{\EqProp}
          {\judge{\Gamma}{\sigma_1}{\ind} \\
           \judge{\Gamma}{\sigma_2}{\ind} 
          }
          {\judgeprop{\Gamma}{\sigma_1 = \sigma_2}}
  {\EqProp}
          {\judge{\Gamma}{\tau_1}{\ind} \\
           \judge{\Gamma}{\tau_2}{\ind} 
          }
          {\judgeprop{\Gamma}{\tau_1 = \tau_2}}

          \begin{llproof}
            \eqPf{[\Omega](\sigma_1 = \sigma_2)}{[\Omega](\tau_1 = \tau_2)}   {Given}
            \eqPf{[\Omega]\sigma_1}{[\Omega]\tau_1}   {Definition of substitution}
            \eqPf{[\Omega]\sigma_2}{[\Omega]\tau_2}   {\ditto}
            \proofsep
            \judgePf{\Gamma}{\sigma_1}{\ind}  {Subderivation}
            \judgePf{\Gamma}{\tau_1}{\ind}  {Subderivation}
            \checkeqPf{\Gamma}{[\Gamma]\sigma_1}{[\Gamma]\sigma_2}{\ind}{\Theta}  {By \Lemmaref{lem:checkeq-completeness}}
            \substextendPf{\Theta}{\Omega_0}   {\ditto}
            \substextendPf{\Omega}{\Omega_0}   {\ditto}
            \proofsep
            \judgePf{\Gamma}{\sigma_2}{\ind}  {Subderivation}
            \judgePf{\Theta}{\sigma_2}{\ind}  {By \Lemmaref{lem:extension-weakening}}
            \judgePf{\Theta}{\tau_2}{\ind}  {Similarly}
            \checkeqPf{\Theta}{[\Theta]\tau_1}{[\Theta]\tau_2}{\ind}{\Delta}  {By \Lemmaref{lem:checkeq-completeness}}
          \Hand  \substextendPf{\Delta}{\Omega_0}   {\ditto}
            \substextendPf{\Omega_0}{\Omega'}   {\ditto}
            \eqPf{[\Theta]\tau_1}{[\Theta][\Gamma]\tau_1}  {By \Lemmaref{lem:substitution-monotonicity} (i)}
            \eqPf{[\Theta]\tau_2}{[\Theta][\Gamma]\tau_2}  {\ditto}
            \checkeqPf{\Theta}{[\Theta][\Gamma]\tau_1}{[\Theta][\Gamma]\tau_2}{\ind}{\Delta}  {By above equalities}
          \Hand  \substextendPf{\Omega}{\Omega'}   {By \Lemmaref{lem:extension-transitivity}}
            \decolumnizePf
            \propequivjudgPf{\Gamma}{([\Gamma]\sigma_1 = [\Theta]\sigma_2)}{([\Gamma]\tau_1 = [\Theta]\tau_2)}{\Gamma}   {By \PropequivEq}
          \Hand  \propequivjudgPf{\Gamma}{([\Gamma]\sigma_1 = [\Gamma]\sigma_2)}{([\Gamma]\tau_1 = [\Gamma]\tau_2)}{\Gamma}   {By above equalities}
          \end{llproof}
  \qedhere
  \end{itemize}
\end{proof}

\checkpropcompleteness*
\begin{proof}\proofheading
  Only one rule, \DeclCheckpropEq, can derive $\declcheckprop{[\Omega]\Gamma}{[\Omega]P}$,
  so by inversion, $P$ has the form $(t_1 = t_2)$ where $[\Omega]t_1 = [\Omega]t_2$. \\
  By inversion on $\judgeprop{\Gamma}{(t_1 = t_2)}$, we have $\judge{\Gamma}{t_1}{\ind}$ and $\judge{\Gamma}{t_2}{\ind}$. \\
  Then by \Lemmaref{lem:checkeq-completeness}, $\checkeq{\Gamma}{[\Gamma]t_1}{[\Gamma]t_2}{\ind}{\Delta}$
  where $\substextend{\Delta}{\Omega'}$ and $\substextend{\Omega}{\Omega'}$. \\
  By \CheckpropEq, $\checkprop{\Gamma}{(t_1 = t_2)}{\Delta}$.
\end{proof}

\subsection{Completeness of Equivalence and Subtyping}

\equivcompleteness*
\begin{proof}\proofheading
  By induction on the derivations of
  $\judgetp{\Gamma}{A}$ and $\judgetp{\Gamma}{B}$.

  We distinguish cases of the rule concluding the first derivation.
  In the first four cases (\ImpliesWF, \WithWF, \ForallWF, \ExistsWF),
  it follows from $[\Omega]A = [\Omega]B$ and the syntactic invariant that $\Omega$
  substitutes terms $t$ (rather than types $A$)
  that the second derivation is concluded by the \emph{same} rule.
  Moreover, if none of these three rules concluded the first
  derivation, the rule concluding the second derivation must
  \emph{not} be \ImpliesWF, \WithWF, \ForallWF or \ExistsWF either.
  
  Because $\Omega$ is predicative, the head connective of $[\Gamma]A$
  must be the same as the head connective of $[\Omega]A$.  
  
  \newcommand{\CASEhandled}[1]{\textcolor{dDkGreen}{\textbf{#1}}}
  \newcommand{\CASEhandledsymm}[1]{\textcolor{dBlue}{\textrm{#1}}}
  \newcommand{\CASEtobeproved}[1]{\textcolor{red}{\textbf{#1}}}
  \newcommand{\CASEimpossible}{\textcolor{dRed}{\textsl{imposs.}}}
  \newcommand{\CaseBpoly}{\CASEhandled{1 (B poly)}}

  We distinguish cases that are
  \CASEimpossible\;(impossible), \CASEhandled{fully written out}, and
  \CASEhandledsymm{similar to fully-written-out cases}.
  For the lower-right case, where both $[\Gamma]A$ and $[\Gamma]B$ have
  a binary connective $\binc$, it must be the same connective.

  The $\vecop$ type is omitted from the table, but can be treated similarly
  to $\implies$ and $\with$.

  \begin{displ}
  \runonfontsz{9pt}
    \begin{array}[t]{rccccccccccc}
      &\multicolumn{10}{c}{\MAINLABEL{[\Gamma]B}}
      \\[5pt]
      &
        &  \implies
        &  \with
        &  {\alltype{\beta} B'}
        &  {\extype{\beta} B'}
        &  \unitty
        &  \alpha
        &  \bhat
        &    B_1 \binc B_2
        &
      \\\cmidrule[2pt]{3-10}
      &    \implies
      &    \CASEhandled{Implies}
      &   \CASEimpossible
      &    \CASEimpossible & \CASEimpossible
      &    \CASEimpossible & \CASEimpossible
      &    \CASEimpossible & \CASEimpossible
      \\\cmidrule{3-10}
      &    \with
      &    \CASEimpossible
      &    \CASEhandledsymm{With} & \CASEimpossible
      &    \CASEimpossible
      &    \CASEimpossible & \CASEimpossible
      &    \CASEimpossible & \CASEimpossible
      \\\cmidrule{3-10}
      &\alltype{\alpha} A'
      &    \CASEimpossible & \CASEimpossible
                &   \CASEhandled{Forall}
                &   \CASEimpossible
                &   \CASEimpossible
                &   \CASEimpossible
                &   \CASEimpossible
                &   \CASEimpossible
      \\\cmidrule{3-10}
      &\extype{\alpha} A'
      &    \CASEimpossible & \CASEimpossible
                &   \CASEimpossible
                &   \CASEhandledsymm{Exists}
                &   \CASEimpossible
                &   \CASEimpossible
                &   \CASEimpossible
                &   \CASEimpossible
      \\\cmidrule[1pt]{3-10}
      &\unitty
      &    \CASEimpossible & \CASEimpossible
                &   \CASEimpossible
                &   \CASEimpossible
                &   \CASEhandled{2.Units}
                &   \CASEimpossible
                &   \CASEhandledsymm{2.BEx.Unit}
                &   \CASEimpossible
      \\\cmidrule{3-10}
      \MAINLABEL{[\Gamma]A}\hspace{-2ex}
      & \alpha
      &    \CASEimpossible & \CASEimpossible
                &   \CASEimpossible
                &   \CASEimpossible
                &   \CASEimpossible
                &   \CASEhandled{2.Uvars}
                &   \CASEhandledsymm{2.BEx.Uvar}
                &   \CASEimpossible
      \\\cmidrule{3-10}
      & \ahat
      &    \CASEimpossible & \CASEimpossible
                &   \CASEimpossible
                &   \CASEimpossible
                &   \CASEhandled{2.AEx.Unit}
                &   \CASEhandledsymm{2.AEx.Uvar}
                &
                    \begin{tabular}[c]{c}
                      \CASEhandled{2.AEx.SameEx}
                      \\
                      \CASEhandled{2.AEx.OtherEx}
                    \end{tabular}
                &   \CASEhandled{2.AEx.Bin}
      \\\cmidrule{3-10}
      & A_1 \binc A_2
      &    \CASEimpossible & \CASEimpossible
                &   \CASEimpossible
                &   \CASEimpossible
                &   \CASEimpossible
                &   \CASEimpossible
                &   \CASEhandledsymm{2.BEx.Bin}
                &   \CASEhandledsymm{2.Bins}
    \end{array}
  \end{displ}
  
  \begin{itemize}
      \DerivationProofCase{\ImpliesWF}
            {\judgeprop{\Gamma}{P} \\ 
             \judgetp{\Gamma}{A_0}}
            {\judgetp{\Gamma}{P \implies A_0}}

            This case of the rule concluding the first derivation coincides with the \CASEhandled{Implies} entry in the table.

            We have $[\Omega]A = [\Omega]B$, that is, $[\Omega](P \implies A_0) = [\Omega]B$. \\
            Because $\Omega$ is predicative, $B$ must have the form $Q \implies B_0$, where
            $[\Omega]P = [\Omega]Q$ and $[\Omega]A_0 = [\Omega]B_0$.
            
            \begin{llproof}
              \judgepropPf{\Gamma}{P}   {Subderivation}
              \judgetpPf{\Gamma}{A_0}   {Subderivation}
              \judgetpPf{\Gamma}{Q \implies B_0}   {Given}
              \judgepropPf{\Gamma}{Q}   {By inversion on rule \ImpliesWF}
              \judgetpPf{\Gamma}{B_0}   {\ditto}
              \propequivjudgPf{\Gamma}{[\Gamma]P}{[\Gamma]Q}{\Theta}   {By \Lemmaref{lem:propequiv-completeness}}
              \substextendPf{\Theta}{\Omega_0}  {\ditto}
              \substextendPf{\Omega}{\Omega_0}  {\ditto}
              \proofsep
              \substextendPf{\Gamma}{\Theta}   {By \Lemmaref{lem:propequiv-extension}}
              \judgetpPf{\Gamma}{A_0}   {Above}
              \judgetpPf{\Gamma}{B_0}   {Above}
              \eqPf{[\Omega]A_0}{[\Omega]B_0}   {Above}
              \eqPf{[\Omega_0]A_0}{[\Omega_0]B_0}   {By \Lemmaref{lem:completing-completeness} (ii) twice}
              \equivjudgPf{\Gamma}{[\Gamma]A_0}{[\Gamma]B_0}{\Delta}   {By i.h.}
            \Hand  \substextendPf{\Delta}{\Omega'}  {\ditto}
               \substextendPf{\Omega_0}{\Omega'}  {\ditto}
          \end{llproof}

          \begin{llproof}
            \Hand  \substextendPf{\Omega}{\Omega'}  {By \Lemmaref{lem:extension-transitivity}}
              \equivjudgPf{\Gamma}{([\Gamma]P \implies [\Gamma]A_0)}{([\Gamma]Q \implies [\Gamma]B_0)}{\Delta}   {By \EquivImplies}
            \Hand  \equivjudgPf{\Gamma}{[\Gamma](P \implies A_0)}{[\Gamma](Q \implies B_0)}{\Delta}   {By definition of substitution}
            \end{llproof}

      \ProofCaseRule{\WithWF}
            Similar to the \ImpliesWF case, coinciding with the \CASEhandledsymm{With} entry in the table.

      \DerivationProofCase{\ForallWF}
            { \judgetp{\Gamma, \alpha:\sort}{A_0} }
            { \judgetp{\Gamma}{\alltype{\alpha:\sort}{A_0}} }

            This case coincides with the \CASEhandled{Forall} entry in the table.

            \begin{llproof}
              \substextendPf{\Gamma}{\Omega}   {Given}
              \substextendPf{\Gamma, \alpha:\sort}{\Omega, \alpha:\sort}   {By \substextendUU}
              \judgetpPf{\Gamma, \alpha:\sort}{A_0}  {Subderivation}
              \proofsep
              \eqPf{B}{\alltype{\alpha:\sort} B_0}   {$\Omega$ predicative}
              \eqPf{[\Omega]A_0}{[\Omega]B_0}   {From definition of substitution}
              \proofsep
              \equivjudgPf{\Gamma, \alpha:\sort}{[\Gamma]A_0}{[\Gamma]B_0}{\Delta_0}  {By i.h.}
              \substextendPf{\Delta_0}{\Omega_0}  {\ditto}
              \substextendPf{\Omega, \alpha:\sort}{\Omega_0}  {\ditto}
            \Hand  \substextendPf{\Omega}{\Omega'  \AND  \Omega_0 = (\Omega', \alpha:\sort, \dots)}
                                 {By \Lemmaref{lem:extension-inversion} (i)}
              \eqPf{\Delta_0}{(\Delta, \alpha:\sort, \Delta')}  {By \Lemmaref{lem:extension-inversion} (i)}
            \Hand  \substextendPf{\Delta}{\Omega'}  {\ditto}
              \equivjudgPf{\Gamma}{\alltype{\alpha:\sort} [\Gamma]A_0}{\alltype{\alpha:\sort} [\Gamma]B_0}{\Delta}  {By \EquivAll}
            \Hand  \equivjudgPf{\Gamma}{[\Gamma](\alltype{\alpha:\sort} A_0)}{[\Gamma](\alltype{\alpha:\sort} B_0)}{\Delta}  {By definition of substitution}
            \end{llproof}

      \ProofCaseRule{\ExistsWF}
           Similar to the \ForallWF case.  (This is the \CASEhandledsymm{Exists} entry in the table.)
 
      \ProofCaseRule{\BinWF}
           If \BinWF also concluded the second derivation,
           then the proof is similar to the \ImpliesWF case,
           but on the first premise, using the i.h.\ instead of  \Lemmaref{lem:propequiv-completeness}.
           This is the \CASEhandledsymm{2.Bins} entry in the lower right corner of the table.

           If \BinWF did not conclude the second derivation, we are in the \CASEhandled{2.AEx.Bin}
           or \CASEhandledsymm{2.BEx.Bin} entries; see below.
    \end{itemize}

    In the remainder, we cover the $4 \times 4$ region in the lower right corner, starting from
    \CASEhandled{2.Units}.  We already handled the \CASEhandledsymm{2.Bins} entry in the
    extreme lower right corner.  At this point, we split on the forms of $[\Gamma]A$ and $[\Gamma]B$ instead;
    in the remaining cases, one or both types is atomic (\eg \CASEhandled{2.Uvars}, \CASEhandled{2.AEx.Bin})
    and we will not need to use the induction hypothesis.
    
    \begin{itemize}
      \item \textbf{Case \CASEhandled{2.Units}:} $[\Gamma]A = [\Gamma]B = \unitty$

              \begin{llproof}
              \Hand  \equivjudgPf{\Gamma}{\unitty}{\unitty}{\Gamma}   {By \EquivUnit}
                  \substextendPf{\Gamma} {\Omega}   {Given}
                  \LetPf{\Omega'}{\Omega'} {}
              \Hand  \substextendPf{\Delta} {\Omega}   {$\Delta = \Gamma$}
              \Hand  \substextendPf{\Omega} {\Omega'}   {By \Lemmaref{lem:extension-reflexivity} and $\Omega' = \Omega$}
            \end{llproof}

      \item \textbf{Case \CASEhandled{2.Uvars}:} $[\Gamma]A = [\Gamma]B = \alpha$

              \begin{llproof}
                 \substextendPf{\Gamma}{\Omega}   {Given}
                 \LetPf{\Omega'}{\Omega'} {}
        \Hand    \equivjudgPf{\Gamma}{\alpha}{\alpha}{\Gamma}  {By \EquivVar}
        \Hand    \substextendPf{\Delta} {\Omega}   {$\Delta = \Gamma$}
        \Hand    \substextendPf{\Omega} {\Omega'}   {By \Lemmaref{lem:extension-reflexivity} and $\Omega' = \Omega$}        
              \end{llproof}

      \item \textbf{Case \CASEhandled{2.AExUnit}:} $[\Gamma]A = \ahat$ and $[\Gamma]B = \unitty$

              \begin{llproof}
                \substextendPf{\Gamma}{\Omega}   {Given}
                \eqPf{\unitty}{[\Omega]\unitty}   {By definition of substitution}
                \notinPf{\ahat} {\FV{\unitty}}   {By definition of $\FV{-}$}
                \eqPf  {[\Omega]\ahat}  {[\Omega]\unitty}  {Given}
                \proofsep
                \instjudgPf{\Gamma}  {\ahat}  {\unitty} {\type}  {\Delta}  {By \Lemmaref{lem:instantiation-completeness} (1)}
      \Hand          \substextendPf{\Omega}{\Omega'}    {\ditto}
      \Hand          \substextendPf{\Delta}{\Omega'}    {\ditto}
                \proofsep
                \eqPf{\unitty}{[\Gamma]\unitty}   {By definition of substitution}
                \notinPf{\ahat} {\FV{\unitty}}   {By definition of $\FV{-}$}
                \proofsep
      \Hand          \equivjudgPf{\Gamma}  {\ahat}  {\unitty}  {\Delta}   {By \EquivInstL}
              \end{llproof}

      \item \textbf{Case \CASEhandledsymm{2.BExUnit}:} $[\Gamma]A = \unitty$ and $[\Gamma]B = \ahat$
        
        Symmetric to the \CASEhandled{2.AExUnit} case.

      \item \textbf{Case \CASEhandledsymm{2.AEx.Uvar}:} $[\Gamma]A = \ahat$ and $[\Gamma]B = \alpha$
        
          Similar to the \CASEhandled{2.AEx.Unit} case,
          using $\beta = [\Omega]\beta = [\Gamma]\beta$ and $\ahat \notin \FV{\beta}$.
        
      \item \textbf{Case \CASEhandledsymm{2.BExUvar}:} $[\Gamma]A = \unitty$ and $[\Gamma]B = \ahat$
        
        Symmetric to the \CASEhandled{2.AExUvar} case.

      \item \textbf{Case \CASEhandled{2.AEx.SameEx}:} $[\Gamma]A = \ahat = \bhat = [\Gamma]B$
        
              \begin{llproof}
                  \equivjudgPf{\Gamma} {\ahat} {\ahat} {\Gamma}    {By \EquivExvar ($\ahat = \bhat$)}
                  \eqPf{[\Gamma]\ahat}  {\ahat}     {$\ahat$ unsolved in $\Gamma$}
      \Hand         \equivjudgPf{\Gamma} {[\Gamma]\ahat} {[\Gamma]\bhat} {\Gamma}   {By above equality + $\ahat = \bhat$}
                  \substextendPf{\Gamma} {\Omega}    {Given}
      \Hand          \substextendPf{\Delta} {\Omega}   {$\Delta = \Gamma$}
                  \LetPf{\Omega'}{\Omega}  {}
      \Hand          \substextendPf{\Omega} {\Omega'}   {By \Lemmaref{lem:extension-reflexivity} and $\Omega' = \Omega$}
              \end{llproof}

      \item \textbf{Case \CASEhandled{2.AEx.OtherEx}:} $[\Gamma]A = \ahat$ and $[\Gamma]B = \bhat$ and $\ahat \neq \bhat$

        Either $\ahat \in \FV{[\Gamma]\bhat}$, or $\ahat \notin \FV{[\Gamma]\bhat}$.

          \begin{itemize}
                \item $\ahat \in \FV{[\Gamma]\bhat}$: \\
                      We have $\ahat \subtermof [\Gamma]\bhat$.  \\
                      Therefore $\ahat = [\Gamma]\bhat$, or $\ahat \propersubtermof [\Gamma]\bhat$. \\
                      But we are in Case 2.AEx.\emph{Other}Ex, so the former is impossible. \\
                      Therefore, $\ahat \propersubtermof [\Gamma]\bhat$. \\
                      By a property of substitutions, $[\Omega]\ahat \propersubtermof [\Omega][\Gamma]\bhat$. \\
                      Since $\substextend{\Gamma}{\Omega}$,
                        by \Lemmaref{lem:substitution-monotonicity} (iii), $[\Omega][\Gamma]\bhat = [\Omega]\bhat$, so
                       $[\Omega]\ahat \propersubtermof [\Omega]\bhat$. \\
                      But it is given that $[\Omega]\ahat = [\Omega]\bhat$, a contradiction.
                      
                \item $\ahat \notin \FV{[\Gamma]\bhat}$:
                  
                      \begin{llproof}
                        \instjudgPf{\Gamma} {\ahat} {[\Gamma]\bhat} {\type} {\Delta}   {By \Lemmaref{lem:instantiation-completeness}}
              \Hand     \equivjudgPf {\Gamma} {\ahat} {[\Gamma]\bhat} {\Delta}   {By \EquivInstL}
              \Hand     \substextendPf{\Delta}{\Omega'} {\ditto}
              \Hand     \substextendPf{\Omega}{\Omega'} {\ditto}
                      \end{llproof}
          \end{itemize}

      \item \textbf{Case \CASEhandled{2.AEx.Bin}:} $[\Gamma]A = \ahat$ and $[\Gamma]B = B_1 \binc B_2$
        
                   Since $[\Gamma]B$ is an arrow, it cannot be exactly $\ahat$.
                   By the same reasoning as in the previous case (2.AEx.OtherEx), $\ahat \notin \FV{[\Gamma]\bhat}$.
                   
                   \begin{llproof}
                     \instjudgPf{\Gamma} {\ahat} {[\Gamma]B} {\type} {\Delta}   {By \Lemmaref{lem:instantiation-completeness}}
\Hand             \substextendPf{\Delta} {\Omega'}   {\ditto}
\Hand             \substextendPf{\Omega} {\Omega'}   {\ditto}
\Hand             \equivjudgPf{\Gamma} {\underbrace{[\Gamma]A}_{\ahat}} {\underbrace{[\Gamma]B}_{B_1 \binc B_2}} {\Delta}   {By \EquivInstL}
                   \end{llproof}

      \item \textbf{Case \CASEhandledsymm{2.BEx.Bin}:} $[\Gamma]A = A_1 \binc A_2$ and $[\Gamma]B = \bhat$
        
        Symmetric to the \CASEhandled{2.AEx.Bin} case, applying \EquivInstR instead of \EquivInstL.
  \qedhere
  \end{itemize}
\end{proof}

\subtypingcompleteness*
\begin{proof}\proofheading
   By induction on the number of $\AllSym$/$\ExistsSym$ quantifiers in $[\Omega]A$ and $[\Omega]B$.

   It is straightforward to show $\dom{\Delta} = \dom{\Omega'}$; for examples
   of the necessary reasoning, see the proof of \Theoremref{thm:typing-completeness}.
   
   We have  $\declsubjudg[\joinpolarity{\polarity B}{\polarity A}]{[\Omega]\Gamma}{[\Omega]A}{[\Omega]B}$.

    \begin{itemize} 
    \DerivationProofCase{\DsubReflNeg}
          {\judgetp{[\Omega]\Gamma}{[\Omega]A}
            \\
            \nonPos{[\Omega]A}
          }
          {\declsubjudg[-]{[\Omega]\Gamma}{[\Omega]A}{\underbrace{[\Omega]A}_{[\Omega]B}}}

           First, we observe that, since applying $\Omega$ as a substitution leaves quantifiers alone,
           the quantifiers that head $A$ must also head $B$.  For convenience, we alpha-vary $B$ to
           quantify over the same variables as $A$.
          
           \begin{itemize}
               \item 
                 If $A$ is headed by $\AllSym$,
                 then $[\Omega]A
                 = (\alltype{\alpha:\sort} [\Omega]A_0)
                 = (\alltype{\alpha:\sort} [\Omega]B_0)
                 = [\Omega]B$.
                
                 Let $\Gamma_0 = (\Gamma, \alpha:\sort, \MonnierComma{\ahat}, \ahat:\sort)$.

                 Let $\Omega_0 = (\Omega, \alpha:\sort, \MonnierComma{\ahat}, \hypeq{\ahat:\sort}{\alpha})$.

                 \begin{itemize}
                 \item If $\polarity{A_0} \in \{-, 0\}$, then:

                   (We elide the straightforward use of lemmas about context extension.)

                   \begin{llproof}
                     \declsubjudgPf{-}{[\Omega_0]\Gamma_0}{[\Omega]A_0}{[\Omega]A_0}  {By \DsubReflNeg}
                     \declsubjudgPf{-}{[\Omega_0]\Gamma_0}{[\Omega_0][\ahat/\alpha]A_0}{A_0}  {By def.\ of subst.}
                     \substextendPf{\Delta_0}{\Omega_0'}   {By i.h.\ (fewer quantifiers)}
                     \substextendPf{\Omega_0}{\Omega_0'}  {\ditto}
                     \subjudgPf{-}{\Gamma_0}{[\Gamma_0][\ahat/\alpha]A_0}{[\Gamma]B_0}{\Delta_0}  {\ditto}
                     \proofsep
                     \subjudgPf{-}{\Gamma_0}{[\ahat/\alpha][\Gamma_0]A_0}{[\Gamma]B_0}{\Delta_0}  {$\ahat$ unsolved in $\Gamma_0$}
                     \subjudgPf{-}{\Gamma_0}{[\ahat/\alpha][\Gamma]A_0}{[\Gamma]B_0}{\Delta_0}  {$\Gamma_0$ substitutes as $\Gamma$}
                     \proofsep
                     \subjudgPf{-}{\Gamma, \alpha:\sort}{\alltype{\alpha:\sort} [\Gamma]A_0}{[\Gamma]B_0}{\Delta, \alpha:\sort, \Theta}  {By \SubAllL}
                     \subjudgPf{-}{\Gamma}{\alltype{\alpha:\sort} [\Gamma]A_0}{\alltype{\alpha:\sort} [\Gamma]B_0}{\Delta}  {By \SubAllR}
                   \Hand  \subjudgPf{-}{\Gamma}{[\Gamma](\alltype{\alpha:\sort} A_0)}{[\Gamma](\alltype{\alpha:\sort} B_0)}{\Delta}  {By def.\ of subst.}
                   \Hand  \substextendPf{\Delta}{\Omega}   {By lemma}
                   \Hand  \substextendPf{\Omega}{\Omega_0'}  {By lemma}
                   \end{llproof}

                 \item If $\polarity{A_0} = {+}$, then proceed as above, but apply \DsubReflPos
                   instead of \DsubReflNeg, and apply \SubNegPosL after applying the i.h.
                   (Rule \SubNegPosR also works.)
                 \end{itemize}

               \item
                 If $A$ is not headed by $\AllSym$:

                 We have $\nonNeg{[\Omega]A}$.  Therefore $\nonNeg{A}$, and thus $A$ is
                 not headed by $\exists$.  Since the same quantifiers must also head $B$,
                 the conditions in rule \SubEquiv are satisfied.

                \begin{llproof}
                   \substextendPf{\Gamma}{\Omega}  {Given}
                   \equivjudgPf{\Gamma}{[\Gamma]A}{[\Gamma]B}{\Delta}   {By \Lemmaref{lem:equiv-completeness}}
                \Hand  \substextendPf{\Delta}{\Omega'}   {\ditto}
                \Hand  \substextendPf{\Omega}{\Omega'}   {\ditto}
                \Hand  \subjudgPf{-}{\Gamma}{[\Gamma]A}{[\Gamma]B}{\Delta}   {By \SubEquiv}
                \end{llproof}
           \end{itemize}

      \ProofCaseRule{\DsubReflPos}
          Symmetric to the \DsubReflNeg case, using
          \SubPosNegL (or \SubPosNegR), and \SubExistsR/\SubExistsL instead of \SubAllL/\SubAllR.

    \DerivationProofCase{\DsubAllL}
              {\judge{[\Omega]\Gamma}{\tau}{\sort}
                \\
                \declsubjudg[-]{[\Omega]\Gamma}{[\tau/\alpha][\Omega]A_0}{[\Omega]B}
              }
              {\declsubjudg[-]{[\Omega]\Gamma}{\underbrace{\alltype{\alpha : \sort}{[\Omega]A_0}}_{[\Omega]A}}{[\Omega]B}}

        We begin by considering whether or not $[\Omega]B$ is headed by a universal quantifier.

        \begin{itemize}

        \item $[\Omega]B = (\alltype{\beta:\sort'} B')$:

          \begin{llproof}
            \declsubjudgPf{-}{[\Omega]\Gamma, \beta:\sort'}{[\Omega]A}{B'}
                          {By \Lemmaref{lem:declarative-subtyping-inversion}}
          \end{llproof}

          The remaining steps are similar to the $\DsubAllR$ case.
          
        \item $\notAll{[\Omega]B}$:

          \begin{llproof}
            \judgePf{[\Omega]\Gamma} {\tau} {\sort}   {Subderivation}
            \substextendPf{\Gamma}{\Omega}   {Given}
            \substextendPf{\Gamma, \MonnierComma{\ahat}}{\Omega, \MonnierComma{\ahat}}
                       {By \substextendMonMon}
            \substextendPf{\Gamma, \MonnierComma{\ahat}, \ahat:\sort}{\underbrace{\Omega, \MonnierComma{\ahat}, \hypeq{\ahat:\sort}{\tau}}_{\Omega_0}}
                       {By \substextendSolve}
            \eqPf{[\Omega]\Gamma} {[\Omega_0](\Gamma, \MonnierComma{\ahat}, \ahat:\sort)}  {By definition of context application (lines 16, 13)}
            \decolumnizePf
            \proofsep
            \declsubjudgPf{-}{[\Omega]\Gamma} {[\tau/\alpha][\Omega]A_0} {[\Omega]B}   {Subderivation}
            \declsubjudgPf{-}{[\Omega_0](\Gamma, \MonnierComma{\ahat}, \ahat:\sort)}
                          {[\tau/\alpha][\Omega]A_0}
                          {[\Omega]B}
                          {By above equality}
            \declsubjudgPf{-}{[\Omega_0](\Gamma, \MonnierComma{\ahat}, \ahat:\sort)}
                          {\big[[\Omega_0]\ahat / \alpha\big][\Omega]A_0}
                          {[\Omega]B}
                          {By definition of substitution}
            \declsubjudgPf{-}{[\Omega_0](\Gamma, \MonnierComma{\ahat}, \ahat:\sort)}
                          {\big[[\Omega_0]\ahat / \alpha\big][\Omega_0]A_0}
                          {[\Omega_0]B}
                          {By definition of substitution}
            \declsubjudgPf{-}{[\Omega_0](\Gamma, \MonnierComma{\ahat}, \ahat:\sort)}
                          {[\Omega_0][\ahat / \alpha]A_0}
                          {[\Omega_0]B}
                          {By distributivity of substitution}
            \decolumnizePf
            \proofsep
            \subjudgPf{-}{\Gamma, \MonnierComma{\ahat}, \ahat} {[\Gamma, \MonnierComma{\ahat}, \ahat:\sort][\ahat / \alpha]A_0} {[\Gamma, \MonnierComma{\ahat}, \ahat:\sort]B} {\Delta_0}   {By i.h.\ ($A$ lost a quantifier)}
            \substextendPf{\Delta_0} {\Omega''}    {\ditto}
            \substextendPf{\Omega_0} {\Omega''}    {\ditto}
            \decolumnizePf
            \subjudgPf{-}{\Gamma, \MonnierComma{\ahat}, \ahat:\sort} {[\Gamma][\ahat / \alpha]A_0} {[\Gamma]B} {\Delta_0}   {By definition of substitution}
            \decolumnizePf
            \substextendPf{\Gamma, \MonnierComma{\ahat}, \ahat:\sort} {\Delta_0}   {By \Lemmaref{lem:subtyping-extension}}
            \eqPf{\Delta_0}{(\Delta, \MonnierComma{\ahat}, \Theta)}  {By \Lemmaref{lem:extension-inversion} (ii)}
            \substextendPf{\Gamma}{\Delta}   {\ditto}
            \eqPf{\Omega''}{(\Omega', \MonnierComma{\ahat}, \Omega_Z)}  {By \Lemmaref{lem:extension-inversion} (ii)}
        \Hand        \substextendPf{\Delta}{\Omega'}   {\ditto}
            \substextendPf{\Omega_0} {\Omega''}    {Above}
            \substextendPf{\Omega, \MonnierComma{\ahat}, \hypeq{\ahat:\sort}{\tau}} {\Omega', \MonnierComma{\ahat}, \Omega_Z}    {By above equalities}
        \Hand        \substextendPf{\Omega} {\Omega'}    {By \Lemmaref{lem:extension-inversion} (ii)}
            \decolumnizePf
            \subjudgPf{-}{\Gamma, \MonnierComma{\ahat}, \ahat:\sort}
                      {[\Gamma][\ahat/\alpha]A_0}
                      {[\Gamma]B}
                      {\Delta, \MonnierComma{\ahat}, \Theta}
                      {By above equality $\Delta_0 = (\Delta, \MonnierComma{\ahat}, \Theta)$}
            \subjudgPf{-}{\Gamma, \MonnierComma{\ahat}, \ahat:\sort}
                      {[\ahat/\alpha][\Gamma]A_0}
                      {[\Gamma]B}
                      {\Delta, \MonnierComma{\ahat}, \Theta}
                      {By def.\ of subst.\ ($[\Gamma]\ahat = \ahat$ and $[\Gamma]\alpha = \alpha$)}
             \Pf{}{}{\notAll{[\Gamma]B}} {From the case assumption}
             \subjudgPf{-}{\Gamma} {\alltype{\alpha:\sort}{[\Gamma]A_0}} {[\Gamma]B} {\Delta}   {By \SubAllL}
          \Hand      \subjudgPf{-}{\Gamma} {[\Gamma](\alltype{\alpha:\sort}{A_0})} {[\Gamma]B} {\Delta}   {By definition of substitution}
          \end{llproof}

        \end{itemize}

    \DerivationProofCase{\DsubAllR}
          {\declsubjudg[-]{[\Omega]\Gamma, \beta : \sort}{[\Omega]A}{[\Omega]B_0}}
          {\declsubjudg[-]{[\Omega]\Gamma}{[\Omega]A}{\underbrace{\alltype{\beta : \sort} [\Omega]B_0}_{[\Omega]B}}}

        \begin{llproof}
          \eqPf{B} {\alltype{\beta : \sort} B_0}   {$\Omega$ predicative}
          \declsubjudgPf{-}{[\Omega]\Gamma} {[\Omega]A} {[\Omega]B}   {Given}
          \declsubjudgPf{-}{[\Omega]\Gamma} {[\Omega]A} {\alltype{\beta} [\Omega]B_0}   {By above equality}
          \declsubjudgPf{-}{[\Omega]\Gamma, \beta : \sort}
                        {[\Omega]A}
                        {[\Omega]B_0}
                {Subderivation}
          \declsubjudgPf{-}{[\Omega, \beta : \sort](\Gamma, \beta : \sort)}
                        {[\Omega, \beta : \sort]A}
                        {[\Omega, \beta : \sort]B_0}
                {By definitions of substitution}
          \subjudgPf{-}{\Gamma, \beta : \sort}{[\Gamma, \beta{:}\sort]A}{[\Gamma, \beta{:}\sort]B_0} {\Delta'}  {By i.h.\ ($B$ lost a quantifier)}
          \substextendPf{\Delta'}{\Omega_0'}   {\ditto}
          \substextendPf{\Omega, \beta : \sort}{\Omega_0'}   {\ditto}
          \subjudgPf{-}{\Gamma, \beta : \sort}{[\Gamma]A}{[\Gamma]B_0} {\Delta'}  {By definition of substitution}
          \decolumnizePf
          \substextendPf{\Gamma, \beta : \sort} {\Delta'}  {By \Lemmaref{lem:instantiation-extension}}
          \eqPf{\Delta'}{(\Delta, \beta : \sort, \Theta)}  {By \Lemmaref{lem:extension-inversion} (i)}
          \substextendPf{\Gamma}{\Delta}  {\ditto}
          \substextendPf{\Delta, \beta : \sort, \Theta}{\Omega_0'}   {By $\substextend{\Delta'}{\Omega_0'}$ and above equality}
          \eqPf{\Omega_0'}{(\Omega', \beta : \sort, \Omega_R)}  {By \Lemmaref{lem:extension-inversion} (i)}
    \Hand      \substextendPf{\Delta}{\Omega'}  {\ditto}
          \decolumnizePf
          \subjudgPf{-}{\Gamma, \beta : \sort}{[\Gamma]A}{[\Gamma]B_0} {\Delta, \beta : \sort, \Theta}  {By above equality}
          \substextendPf{\Omega, \beta : \sort}{\Omega', \beta : \sort, \Omega_R}   {By above equality}
    \Hand      \substextendPf{\Omega}{\Omega'}   {By \Lemmaref{lem:extension-transitivity}}
          \proofsep
          \subjudgPf{-}{\Gamma}{[\Gamma]A}{\alltype{\beta : \sort}{[\Gamma]B_0}}{\Delta}  {By \SubAllR}
    \Hand      \subjudgPf{-}{\Gamma}{[\Gamma]A}{[\Gamma](\alltype{\beta : \sort} B_0)}{\Delta}  {By definition of substitution}
    \end{llproof}

  \DerivationProofCase{\DsubExistsL}
          {\declsubjudg[+]{[\Omega]\Gamma, \alpha : \sort}{[\Omega]A_0}{[\Omega]B}}
          {\declsubjudg[+]{[\Omega]\Gamma}{\underbrace{\extype{\alpha : \sort}{[\Omega]A_0}}_{[\Omega]A}}{[\Omega]B}}

        \begin{llproof}
          \eqPf{A} {\extype{\alpha : \sort}{A_0}}   {$\Omega$ predicative}
          \declsubjudgPf{+}{[\Omega]\Gamma} {[\Omega]A)} {[\Omega]B}   {Given}
          \declsubjudgPf{+}{[\Omega]\Gamma} {[\Omega]\extype{\alpha : \sort}{A_0}} {[\Omega]B}   {By above equality}
          \declsubjudgPf{+}{[\Omega]\Gamma, \alpha : \sort}
                        {[\Omega]A_0}
                        {[\Omega]B}
                {Subderivation}
          \declsubjudgPf{+}{[\Omega, \alpha : \sort](\Gamma, \alpha : \sort)}
                        {[\Omega, \alpha : \sort]A_0}
                        {[\Omega, \alpha : \sort]B}
                {By definitions of substitution}
          \subjudgPf{+}{\Gamma, \alpha : \sort}{[\Gamma, \beta{:}\sort]A_0}{[\Gamma, \beta{:}\sort]B} {\Delta'}  {By i.h.\ ($A$ lost a quantifier)}
          \substextendPf{\Delta'}{\Omega_0'}   {\ditto}
          \substextendPf{\Omega, \alpha : \sort}{\Omega_0'}   {\ditto}
          \subjudgPf{+}{\Gamma, \alpha : \sort}{[\Gamma]A}{[\Gamma]B_0} {\Delta'}  {By definition of substitution}
          \decolumnizePf
          \substextendPf{\Gamma, \alpha : \sort} {\Delta'}  {By \Lemmaref{lem:instantiation-extension}}
          \eqPf{\Delta'}{(\Delta, \alpha : \sort, \Theta)}  {By \Lemmaref{lem:extension-inversion} (i)}
          \substextendPf{\Gamma}{\Delta}  {\ditto}
          \substextendPf{\Delta, \alpha : \sort, \Theta}{\Omega_0'}   {By $\substextend{\Delta'}{\Omega_0'}$ and above equality}
          \eqPf{\Omega_0'}{(\Omega', \alpha : \sort, \Omega_R)}  {By \Lemmaref{lem:extension-inversion} (i)}
    \Hand      \substextendPf{\Delta}{\Omega'}  {\ditto}
          \decolumnizePf
          \subjudgPf{+}{\Gamma, \alpha : \sort}{[\Gamma]A_0}{[\Gamma]B} {\Delta, \alpha : \sort, \Theta}  {By above equality}
          \substextendPf{\Omega, \alpha : \sort}{\Omega', \alpha : \sort, \Omega_R}   {By above equality}
    \Hand      \substextendPf{\Omega}{\Omega'}   {By \Lemmaref{lem:extension-transitivity}}
          \proofsep
          \subjudgPf{+}{\Gamma}{\extype{\alpha:\sort} [\Gamma]A_0}{[\Gamma]B}{\Delta}  {By \SubAllR}
    \Hand      \subjudgPf{+}{\Gamma}{[\Gamma](\extype{\alpha:\sort} A_0)}{[\Gamma]B}{\Delta}  {By definition of substitution}
    \end{llproof}

  \DerivationProofCase{\DsubExistsR}
          {\judge{\Psi}{\tau}{\sort} \\
           \declsubjudg[+]{\Psi}{[\Omega]A}{[\tau/\beta]B_0}
          }
          {\declsubjudg[+]{\Psi}{[\Omega]A}{\underbrace{\extype{\beta : \sort}{B_0}}_{[\Omega]B} }}

          We consider whether $[\Omega]A$ is headed by an existential.

          If $[\Omega]A = \extype{\alpha:\sort'}{A'}$:

            \begin{llproof}
              \declsubjudgPf{+}{[\Omega]\Gamma, \alpha:\sort'}{A'}{[\Omega]B}{By \Lemmaref{lem:declarative-subtyping-inversion}}
            \end{llproof}

            The remaining steps are similar to the \DsubExistsL case. 

          If $\notExists{[\Omega]A}$:

          \begin{llproof}
            \judgePf{[\Omega]\Gamma} {\tau} {\sort}   {Subderivation}
            \substextendPf{\Gamma}{\Omega}   {Given}
            \substextendPf{\Gamma, \MonnierComma{\ahat}}{\Omega, \MonnierComma{\ahat}}
                       {By \substextendMonMon}
            \substextendPf{\Gamma, \MonnierComma{\ahat}, \ahat:\sort}{\underbrace{\Omega, \MonnierComma{\ahat}, \hypeq{\ahat:\sort}{\tau}}_{\Omega_0}}
                       {By \substextendSolve}
            \eqPf{[\Omega]\Gamma} {[\Omega_0](\Gamma, \MonnierComma{\ahat}, \ahat:\sort)}  {By definition of context application (lines 16, 13)}
            \decolumnizePf
            \declsubjudgPf{+}{[\Omega]\Gamma} {[\Omega]A} {[\tau/\beta][\Omega]B_0}    {Subderivation}
            \declsubjudgPf{+}{[\Omega_0](\Gamma, \MonnierComma{\ahat}, \ahat:\sort)}
                          {[\Omega]A}
                          {[\tau/\beta][\Omega]B_0}
                          {By above equality}
            \declsubjudgPf{+}{[\Omega_0](\Gamma, \MonnierComma{\ahat}, \ahat:\sort)}
                          {[\Omega]A}
                          {\big[[\Omega_0]\ahat / \beta\big][\Omega]B_0}
                          {By definition of substitution}
            \declsubjudgPf{+}{[\Omega_0](\Gamma, \MonnierComma{\ahat}, \ahat:\sort)}
                          {[\Omega_0]A}
                          {\big[[\Omega_0]\ahat / \beta\big][\Omega_0]B_0}
                          {By definition of substitution}
            \declsubjudgPf{+}{[\Omega_0](\Gamma, \MonnierComma{\ahat}, \ahat:\sort)}
                          {[\Omega_0]A}
                          {[\Omega_0][\ahat / \beta]B_0}
                          {By distributivity of substitution}
            \decolumnizePf
            \subjudgPf{+}{\Gamma, \MonnierComma{\ahat}, \ahat} {[\Gamma, \MonnierComma{\ahat}, \ahat:\sort]A} {[\Gamma, \MonnierComma{\ahat}, \ahat:\sort][\ahat / \beta]B_0}  {\Delta_0}   {By i.h.\ ($B$ lost a quantifier)}
            \substextendPf{\Delta_0} {\Omega''}    {\ditto}
            \substextendPf{\Omega_0} {\Omega''}    {\ditto}
            \decolumnizePf
            \subjudgPf{+}{\Gamma, \MonnierComma{\ahat}, \ahat:\sort} {[\Gamma][\ahat / \beta]B_0} {[\Gamma]B} {\Delta_0}   {By definition of substitution}
            \substextendPf{\Gamma, \MonnierComma{\ahat}, \ahat:\sort} {\Delta_0}   {By \Lemmaref{lem:subtyping-extension}}
            \eqPf{\Delta_0}{(\Delta, \MonnierComma{\ahat}, \Theta)}  {By \Lemmaref{lem:extension-inversion} (ii)}
            \substextendPf{\Gamma}{\Delta}   {\ditto}
            \eqPf{\Omega''}{(\Omega', \MonnierComma{\ahat}, \Omega_Z)}  {By \Lemmaref{lem:extension-inversion} (ii)}
        \Hand        \substextendPf{\Delta}{\Omega'}   {\ditto}
            \substextendPf{\Omega_0} {\Omega''}    {Above}
            \substextendPf{\Omega, \MonnierComma{\ahat}, \hypeq{\ahat:\sort}{\tau}} {\Omega', \MonnierComma{\ahat}, \Omega_Z}    {By above equalities}
        \Hand        \substextendPf{\Omega} {\Omega'}    {By \Lemmaref{lem:extension-inversion} (ii)}
            \decolumnizePf
            \subjudgPf{+}{\Gamma, \MonnierComma{\ahat}, \ahat:\sort}
                      {[\Gamma]A}
                      {[\Gamma][\ahat/\beta]B_0}
                      {\Delta, \MonnierComma{\ahat}, \Theta}
                      {By above equality $\Delta_0 = (\Delta, \MonnierComma{\ahat}, \Theta)$}
            \subjudgPf{+}{\Gamma, \MonnierComma{\ahat}, \ahat:\sort}
                      {[\Gamma]A}
                      {[\ahat/\beta][\Gamma]B_0}
                      {\Delta, \MonnierComma{\ahat}, \Theta}
                      {By def.\ of subst.\ ($[\Gamma]\ahat = \ahat$ and $[\Gamma]\beta = \beta$)}
             \Pf{}{}{\notExists{[\Gamma]A}} {From the case hypothesis}
             \subjudgPf{+}{\Gamma} {[\Gamma]A} {\extype{\beta:\sort}{[\Gamma]B_0}} {\Delta}   {By \SubExistsR}
          \Hand      \subjudgPf{+}{\Gamma} {[\Gamma]A} {[\Gamma](\extype{\beta:\sort}{B_0})} {\Delta}   {By definition of substitution}
          \end{llproof}
\qedhere
      \end{itemize}
\end{proof}

\subsection{Completeness of Typing}

\variabledecomposition*
\begin{proof}
  Each of these follows by induction on $\Pi$ and decomposition of the two input derivations. 
\end{proof}

\patterndecompositionsubstitution*
\begin{proof}
  Each case is proved by induction on the relevant derivation. 
\end{proof}

\patterndecompositionfunctionality*
\begin{proof}
  By induction on the derivation of $\expandvar{\Pi}{\Pi'}$. 
\end{proof}

\decidabilityvariableremoval*
\begin{proof}
  This follows from an induction on the structure of $\Pi$. 
\end{proof}

\variableinversion*
\begin{proof}
  This follows by induction on the relevant derivations. 
\end{proof}

\coveragecompleteness*
\begin{proof}\proofheading
  By mutual induction, with the induction metric lexicographically ordered on the number of pattern constructor symbols in the branches of $\Pi$, the number of connectives in $\Avec$, and 1 if $P$ is present/0 if it is absent. 
  \begin{enumerate}
  \item Assume $\judgetpvec{\Gamma}{\Avec\;q}$
    and $[\Gamma]\Avec = \Avec$
    and (for all $\Omega$ such that $\substextend{\Gamma}{\Omega}$,  we have $\covers{q}{[\Omega]\Gamma}{[\Omega]\Pi}{[\Omega]\Avec}$)
    \begin{itemize}
  \item Case $\Avec = \cdot$:

    Choose a completing substitution $\Omega$. \\
    Then we have $\covers{q}{[\Omega]\Gamma}{[\Omega]\Pi}{\cdot}$. \\
      By inversion, we see that \DeclCoversEmpty was the last rule, and that we have
      $\declcovers[q]{[\Omega]\Gamma}{[\Omega]\branch{\cdot}{e_1} \alt \ldots}{\cdot}$. \\
      Hence by \CoversEmpty, we have $\covers{q}{\Gamma}{\branch{\cdot}{e_1} \alt \ldots}{\cdot}$.

  \item Case $\Avec = A, \Bvec$:

    By \Lemmaref{thm:variable-removal} either 
    \begin{itemize}
    \item Case $\expandvar{\Pi}{\Pi'}$:
      \begin{tabbing}
      Assume $\Omega$ such that $\substextend{\Gamma}{\Omega}$. \\
      \qquad \=
         By assumption, $\covers{q}{[\Omega]\Gamma}{[\Omega]\Pi}{[\Omega](A, \Bvec)}$. \+ \\ 
         By \Lemmaref{thm:pattern-decomposition-substitution}, $\expandvar{[\Omega]\Pi}{[\Omega]\Pi'}$. \\
         By \Lemmaref{thm:variable-inversion}, $\covers{q}{[\Omega]\Gamma}{[\Omega]\Pi'}{[\Omega]\Bvec}$. \- \\
      So for all $\Omega$ such that $\substextend{\Gamma}{\Omega}$, $\covers{q}{[\Omega]\Gamma}{[\Omega]\Pi'}{[\Omega]\Bvec}$. \\
      By induction,  $\covers{q}{\Gamma}{\Pi'}{\Bvec}$. \\
      \Hand By rule \CoversVar, $\covers{q}{\Gamma}{\Pi}{A, \Bvec}$. \\
    \end{tabbing}

  \item Case $\forall \Pi'. \lnot(\expandvar{\Pi}{\Pi'})$:

      \begin{itemize}
      \item Case $\ahat, \Bvec$:

        This case is impossible. Choose a completing substitution $\Omega$ such that $[\Omega]\ahat = 1 \to 1$, and
        then by assumption we have $\covers{q}{[\Omega]\Gamma}{[\Omega]\Pi}{1 \to 1, [\Omega]\Bvec}$. By inversion we
        have that $\expandvar{[\Omega]\Pi}{\Pi'}$. By \Lemmaref{thm:pattern-decomposition-substitution}, we have a $\Pi''$ such that $[\Omega]\Pi'' = \Pi'$, and $\expandvar{\Pi}{\Pi''}$. This yields the contradiction. 
        
      \item Case $C \arr D, \Bvec$:
      \item Case $\alltype{\alpha:\sort}{A}, \Bvec$:
      \item Case $\alpha, \Bvec$:

        Similar to the $\ahat$ case. 

      \item Case $\Avec = \unitty, \Bvec$:

        \begin{tabbing}
          Choose an arbitrary $\Omega$ such that $\substextend{\Gamma}{\Omega}$. \\
          By assumption, $\covers{q}{[\Omega]\Gamma}{[\Omega]\Pi}{[\Omega](\unitty, \Bvec)}$.  \\ 
          By inversion, we know the rule \DeclCoversUnit applies (since the variable case has been ruled out). \\
          Hence $\expandunit{[\Omega]\Pi}{\Pi''}$ and $\covers{q}{[\Omega]\Gamma}{\Pi''}{[\Omega]\Bvec}$. \\
          By \Lemmaref{thm:pattern-decomposition-substitution}, there is a $\Pi'$ such that \\\;\;$[\Omega]\Pi' = \Pi''$ and
          $\expandunit{\Pi}{\Pi'}$. \\ \\ 
          
          Assume $\Omega$ such that $\substextend{\Gamma}{\Omega}$. \\
          \qquad \=
          By assumption, $\covers{q}{[\Omega]\Gamma}{[\Omega]\Pi}{[\Omega](\unitty, \Bvec)}$. \+ \\ 
          By inversion, we know the rule \DeclCoversUnit applies (since the variable case has been ruled out). \\
          Hence $\expandunit{[\Omega]\Pi}{\Pi''}$ and $\covers{q}{[\Omega]\Gamma}{\Pi''}{[\Omega]\Bvec}$.  \\
          By \Lemmaref{thm:pattern-decomposition-substitution}, \\ \;\; there is a $\hat{\Pi}'$' such that $\Pi'' = [\Omega]\hat{\Pi}'$' and $\expandunit{\Pi}{\hat{Pi}'}$. \\
          By \Lemmaref{thm:pattern-decomposition-functionality}, we know $\hat{\Pi}' = \Pi'$. \- \\
          So for all $\Omega$ such that $\substextend{\Gamma}{\Omega}$, $\covers{q}{[\Omega]\Gamma}{[\Omega]\Pi'}{[\Omega]\Bvec}$. \\
          By induction,  $\covers{q}{\Gamma}{\Pi'}{\Bvec}$. \\
          By rule \CoversUnit, $\covers{q}{\Gamma}{\Pi}{A, \Bvec}$. \\
        \end{tabbing}

      \item Case $C \times D, \Bvec$:
        \begin{tabbing}
          Choose an arbitrary $\Omega$ such that $\substextend{\Gamma}{\Omega}$. \\
          By assumption, $\covers{q}{[\Omega]\Gamma}{[\Omega]\Pi}{[\Omega](C \times D, \Bvec)}$.  \\ 
          By inversion, we know the rule \DeclCoversTimes applies (since the variable case has been ruled out). \\
          Hence $\expandpair{[\Omega]\Pi}{\Pi''}$ and $\covers{q}{[\Omega]\Gamma}{\Pi''}{[\Omega](C, D, \Bvec)}$. \\
          By \Lemmaref{thm:pattern-decomposition-substitution}, there is a $\Pi'$ such that \\\;\;$[\Omega]\Pi' = \Pi''$ and
          $\expandpair{\Pi}{\Pi'}$. \\ \\ 
          
          Assume $\Omega$ such that $\substextend{\Gamma}{\Omega}$. \\
          \qquad \=
          By assumption, $\covers{q}{[\Omega]\Gamma}{[\Omega]\Pi}{[\Omega](C \times D, \Bvec)}$. \+ \\ 
          By inversion, we know the rule \DeclCoversTimes applies (since the variable case has been ruled out). \\
          Hence $\expandpair{[\Omega]\Pi}{\Pi''}$ and $\covers{q}{[\Omega]\Gamma}{\Pi''}{[\Omega](C, D, \Bvec)}$.  \\
          By \Lemmaref{thm:pattern-decomposition-substitution}, \\ \;\; there is a $\hat{\Pi}'$' such that $\Pi'' = [\Omega]\hat{\Pi}'$' and $\expandpair{\Pi}{\hat{Pi}'}$. \\
          By \Lemmaref{thm:pattern-decomposition-functionality}, we know $\hat{\Pi}' = \Pi'$. \- \\
          So for all $\Omega$ such that $\substextend{\Gamma}{\Omega}$, $\covers{q}{[\Omega]\Gamma}{[\Omega]\Pi'}{[\Omega](C, D, \Bvec)}$. \\
          By induction,  $\covers{q}{\Gamma}{\Pi'}{C, D, \Bvec}$. \\
          By rule \CoversTimes, $\covers{q}{\Gamma}{\Pi}{C \times D, \Bvec}$. \\
        \end{tabbing}
        
      \item Case $C + D, \Bvec$:
        \begin{tabbing}
          Choose an arbitrary $\Omega$ such that $\substextend{\Gamma}{\Omega}$. \\
          By assumption, $\covers{q}{[\Omega]\Gamma}{[\Omega]\Pi}{[\Omega](C \times D, \Bvec)}$.  \\ 
          By inversion, we know the rule \DeclCoversSum applies (since the variable case has been ruled out). \\
          Hence $\expandsum{[\Omega]\Pi}{\Pi'_1}{\Pi'_2}$ and $\covers{q}{[\Omega]\Gamma}{\Pi'_1}{[\Omega](C, \Bvec)}$ and $\covers{q}{[\Omega]\Gamma}{\Pi'_2}{[\Omega](D, \Bvec)}$. \\
          By \Lemmaref{thm:pattern-decomposition-substitution}, there is a $\Pi_1$ and $\Pi_2$ such that \\\;\;$[\Omega]\Pi_1 = \Pi'_1$ and $[\Omega]\Pi_2 = \Pi'_2$ and
          $\expandsum{\Pi}{\Pi_1}{\Pi_2}$. \\ \\ 
          
          Assume $\Omega$ such that $\substextend{\Gamma}{\Omega}$. \\
          \qquad \=
          By assumption, $\covers{q}{[\Omega]\Gamma}{[\Omega]\Pi}{[\Omega](C \times D, \Bvec)}$. \+ \\ 
          By inversion, we know the rule \DeclCoversSum applies (since the variable case has been ruled out). \\
          Hence $\expandsum{[\Omega]\Pi}{\hat{\Pi}'_1}{\hat{\Pi}'_2}$ and $\covers{q}{[\Omega]\Gamma}{\hat{\Pi}'_1}{[\Omega](C, \Bvec)}$ and $\covers{q}{[\Omega]\Gamma}{\hat{\Pi}'_2}{[\Omega](D, \Bvec)}$.  \\
          By \Lemmaref{thm:pattern-decomposition-substitution}, \\ \;\; there is a $\hat{\Pi}_1$' such that $\hat{\Pi}'_1 = [\Omega]\hat{\Pi}_1$ and $\hat{\Pi}'_2 = [\Omega]\hat{\Pi}_2$ and $\expandsum{\Pi}{\hat{Pi}_1}{\hat{\Pi}_2}$. \\
          By \Lemmaref{thm:pattern-decomposition-functionality}, we know $\hat{\Pi}_i = \Pi_i$. \- \\
          So for all $\Omega$ such that $\substextend{\Gamma}{\Omega}$, $\covers{q}{[\Omega]\Gamma}{[\Omega]\Pi_1}{[\Omega](C,\Bvec)}$. \\
          So for all $\Omega$ such that $\substextend{\Gamma}{\Omega}$, $\covers{q}{[\Omega]\Gamma}{[\Omega]\Pi_2}{[\Omega](D, \Bvec)}$. \\
          By induction,  $\covers{q}{\Gamma}{\Pi_1}{C, \Bvec}$. \\
          By induction,  $\covers{q}{\Gamma}{\Pi_2}{D, \Bvec}$. \\
          By rule \CoversSum, $\covers{q}{\Gamma}{\Pi}{C + D, \Bvec}$. \\
        \end{tabbing}

      \item Case $\vectype{n}{A}, \Bvec$:

        Similar to the previous case.
        
      \item Case $\extype{\alpha:\sort}{C}, \Bvec$:

        \begin{tabbing}
          Assume $\Omega$ such that $\substextend{\Gamma}{\Omega}$. \\
          \qquad \=
          By assumption, $\covers{q}{[\Omega]\Gamma}{[\Omega]\Pi}{[\Omega](\extype{\alpha:\sort}{C}, \Bvec)}$. \+ \\ 
          By inversion, we know the rule \DeclCoversEx applies. \\
          Hence $\covers{q}{[\Omega]\Gamma, \alpha:\sort}{[\Omega]\Pi}{[\Omega](C, \Bvec)}$.  \- \\
          So for all $\Omega$ such that $\substextend{\Gamma}{\Omega}$, $\covers{q}{[\Omega](\Gamma, \alpha:\sort)}{[\Omega]\Pi}{[\Omega](C, \Bvec)}$. \\
          By induction,  $\covers{q}{\Gamma, \alpha:\sort}{\Pi}{C, \Bvec}$. \\
          By rule \CoversEx, $\covers{q}{\Gamma}{\Pi}{\extype{\alpha:\sort}{C}, \Bvec}$. \\
        \end{tabbing}
        
      \item Case $C \with P, \Bvec$:
        \begin{itemize}
        \item Case $q = \OK$:
          Similar to the previous case.
        \item Case $q = \p$:
          \begin{tabbing}
            Assume $\Omega$ such that $\substextend{\Gamma}{\Omega}$. \\
            \qquad \=
            By assumption, $\covers{q}{[\Omega]\Gamma}{[\Omega]\Pi}{[\Omega](C \with P, \Bvec)}$. \+ \\ 
            By inversion, we know the rule \DeclCoversWith applies. \\
            Hence $\coverseq{[\Omega]\Gamma}{[\Omega]P}{[\Omega]\Pi}{[\Omega](C, \Bvec)}$.  \- \\
            So for all $\Omega$ such that $\substextend{\Gamma}{\Omega}$, $\coverseq{[\Omega](\Gamma, \alpha:\sort)}{[\Omega]P}{[\Omega]\Pi}{[\Omega](C, \Bvec)}$. \\
            By mutual induction,  $\coverseq{\Gamma}{P}{\Pi}{C, \Bvec}$. \\
            By rule \CoversWith, $\covers{\p}{\Gamma}{\Pi}{C \with P, \Bvec}$. \\
          \end{tabbing}
          
        \end{itemize}
      \end{itemize}
    \end{itemize}
    
  \end{itemize}

\item Assume $[\Gamma]\Avec = \Avec$
    and $[\Gamma]P = P$
    and $\judgetpvec{\Gamma}{\Avec\;\p}$
    and (for all $\Omega$ such that $\substextend{\Gamma}{\Omega}$, we have $\coverseq{[\Omega]\Gamma}{[\Omega]P}{[\Omega]\Pi}{[\Omega]\Avec}$). 

  Let $(t_1 = t_2)$ be $P$. \\ 
  Consider whether the $\mgu{t_1}{t_2}$ exists 
  \begin{itemize}
    \item Case $\theta = \mgu{t_1}{t_2}$:

          \begin{llproof}
            \Pf{}{}{\mgu{t_1}{t_2} = \theta}  {Premise}
            \elimeqPf{\Gamma}{t_1}{t_2}{\sort}{\Gamma, \Theta}
                  {By \Lemmaref{lem:elimeq-completeness} (1)}
            \elimeqPf{\Gamma}{[\Gamma]t_1}{[\Gamma]t_2}{\sort}{\Gamma, \Theta}
                  {Follows from given assumption}
            \end{llproof}\\
      Assume $\Omega$ such that $\substextend{\Gamma, \Theta}{\Omega}$. \\
      By \Lemmaref{lem:canonical-completion}, there is a $\Omega'$ such that $[\Omega]\Gamma = [\Omega']\Gamma$ and $\dom{\Gamma} = \dom{\Gamma'}$. \\
      Moreover, by \Lemmaref{lem:extension-inversion}, we can construct a $\Omega''$ such that $\Omega' = \Omega'', \Theta$ and $\substextend{\Gamma}{\Omega'}$. \\ 
      By assumption, $\declcoverseq{[\Omega'']\Gamma}{[\Omega''](t_1 = t_2)}{[\Omega'']\Pi}{\Avec}$. \\
      There is only one way this derivation could be constructed:
      \begin{itemize}
    \DerivationProofCase{\DeclCoversEq}
          {\theta = \mgu{t_1}{t_2} \\
           \declcovers{[\theta][\Omega'']\Gamma}{[\theta][\Omega'']\Pi}{[\theta][\Omega'']\Avec}
          }
          { \declcoverseq{[\Omega'']\Gamma}{[\Omega''](t_1 = t_2)}{[\Omega'']\Pi}{[\Omega'']\Avec}
          }
          
          \begin{llproof}
            \decolumnizePf
            \declcoversPf{[\theta][\Omega'']\Gamma}{[\theta][\Omega'']\Pi}{([\theta][\Omega'']\Avec)}   {Subderivation}
            \eqPf{[\theta][\Omega'']\Gamma}{[\Omega'', \Theta](\Gamma, \Theta)}
                 {By \Lemmaref{lem:subst-upgrade} (iii)}
             \eqPf{[\theta][\Omega'']\Pi}{[\Omega'', \Theta]\Pi}
                 {By \Lemmaref{lem:subst-upgrade} (iv)}
            \eqPfParen{[\theta][\Omega'']\Avec}{[\Omega, \Theta][\Gamma, \Theta]\Avec}
                 {By \Lemmaref{lem:subst-upgrade} (i)}
            \declcoversPf{[\Omega'', \Theta](\Gamma, \Theta)}{[\Omega'', \Theta]\Pi}{[\Omega'', \Theta][\Gamma, \Theta]\Avec}   {By above equalities}
            \declcoversPf{[\Omega'](\Gamma, \Theta)}{[\Omega']\Pi}{[\Omega'][\Gamma, \Theta]\Avec}   {By above equalities}
            \declcoversPf{[\Omega](\Gamma, \Theta)}{[\Omega]\Pi}{[\Omega][\Gamma, \Theta]\Avec}   {By above equalities}
          \end{llproof}
          So we know by induction that 
            $\covers{\p}{\Gamma, \Theta}{[\Gamma, \Theta]\Pi}{[\Gamma, \Theta]\Avec}$. \\

          Hence by \CoversEq we have $\coverseq{\Gamma}{t_1 = t_2}{\Pi}{\Avec}$. \\
        
        \end{itemize}

      \item Case $\mgu{t_1}{t_2} = \bot$: 
        
          \begin{llproof}
            \Pf{}{}{\mgu{t_1}{t_2} = \bot}  {Premise}
            \elimeqPf{\Gamma}{t_1}{t_2}{\sort}{\bot}
                  {By \Lemmaref{lem:elimeq-completeness} (2)}
            \elimeqPf{\Gamma}{[\Gamma]t_1}{[\Gamma]t_2}{\sort}{\bot}
                  {Follows from given assumption}
          \Hand
            \coverseqPf{\Gamma}{t_1 = t_2}{\Pi}{\Avec}
                     {By \CoversEqBot}
          \end{llproof}
  \end{itemize}

\end{enumerate}
\end{proof}

\typingcompleteness*
\begin{proof}\proofheading
  By induction, using the measure in \Definitionref{def:soundness-completeness-measure}.

  \begin{itemize}
     \DerivationProofCase{\DeclVar}
          {(\hyp{p}{x}{A}) \in [\Omega]\Gamma}
          {\declsynjudg{p}{[\Omega]\Gamma}{x}{A}}
          
          \begin{llproof}
            \inPf{(\hyp{p}{x}{A})} {[\Omega]\Gamma}  {Premise}
            \substextendPf{\Gamma}{\Omega}   {Given}
            \inPf{(\hyp{p}{x}{A'})}{\Gamma\text{~where $[\Omega]A' = A$}}   {From definition of context application}
            \LetPf{\Delta}{\Gamma} {}
            \LetPf{\Omega'}{\Omega} {}
\Hand       \substextendPf{\Gamma}{\Omega}   {Given}
\Hand       \substextendPf{\Omega}{\Omega}   {By \Lemmaref{lem:extension-reflexivity}}
\Hand       \synjudgPf{p}{\Gamma}{x}{[\Gamma]A'}{\Gamma}   {By \Var}
\Hand       \eqPf{[\Gamma]A'}{[\Gamma][\Gamma]A'}  {By idempotence of substitution}
\Hand       \eqPf{\dom{\Gamma}}{\dom{\Omega}}  {Given}
            \substextendPf{\Gamma}{\Omega}   {Given}
            \eqPf{[\Omega][\Gamma]A'}{[\Omega]A'}   {By \Lemmaref{lem:substitution-monotonicity} (iii)}
\Hand                 \continueeqPf{A}   {By above equality}
          \end{llproof}

     \DerivationProofCase{\DeclSub}
          {\declsynjudg{q}{[\Omega]\Gamma}{[\Omega]e}{B}
            \\
            \declsubjudg[\joinpolarity{\polarity A}{\polarity B}]{[\Omega]\Gamma}{B}{[\Omega]A}
          }
          {\declchkjudg{p}{[\Omega]\Gamma}{[\Omega]e}{[\Omega]A}}

          \begin{llproof}
            \declsynjudgPf{q}{[\Omega]\Gamma}{[\Omega]e}{B}   {Subderivation}
            \synjudgPf{q}{\Gamma}{e}{B'}{\Theta}      {By i.h.}
            \eqPf{B}{[\Omega]B'}   {\ditto}
            \substextendPf{\Theta}{\Omega_0}   {\ditto}
            \substextendPf{\Omega}{\Omega_0}   {\ditto}
            \eqPf{\dom{\Theta}}{\dom{\Omega_0}}  {\ditto}
            \proofsep
            \substextendPf{\Gamma}{\Omega}   {Given}
            \substextendPf{\Gamma}{\Omega_0}   {By \Lemmaref{lem:extension-transitivity}}
            \declsubjudgPf{\joinpolarity{\polarity A}{\polarity B}}{[\Omega]\Gamma}{B}{[\Omega]A}  {Subderivation}
            \eqPf{[\Omega]\Gamma} {[\Omega]\Theta}   {By \Lemmaref{lem:completes-confluence}}
            \declsubjudgPf{\joinpolarity{\polarity A}{\polarity B}}{[\Omega]\Theta}{B}{[\Omega]A}  {By above equalities}
            \substextendPf{\Theta}{\Omega_0}  {Above}
            \subjudgPf{\joinpolarity{\polarity A}{\polarity B}}{\Theta}{B'}{A}{\Delta}  {By \Theoremref{thm:subtyping-completeness}}
            \substextendPf{\Omega_0}{\Omega'}   {\ditto}
\Hand       \eqPf{\dom{\Delta}}{\dom{\Omega'}}  {\ditto}
\Hand   \substextendPf{\Delta}{\Omega'}   {By \Lemmaref{lem:extension-transitivity}}
\Hand   \substextendPf{\Omega}{\Omega'}   {By \Lemmaref{lem:extension-transitivity}}
            \proofsep
\Hand   \chkjudgPf{p}{\Gamma}{e}{A}{\Delta}  {By \Sub}
          \end{llproof}

     \DerivationProofCase{\DeclAnno}
          {\judgetp{[\Omega]\Gamma}{[\Omega]A}
            \\
            \declchkjudg{\p}{[\Omega]\Gamma}{[\Omega]e_0}{[\Omega]A}
          }
          {\declsynjudg{\p}{[\Omega]\Gamma}{[\Omega](e_0 : A)}{A}}

          \begin{llproof}
            \declchkjudgPf{\p}{[\Omega]\Gamma}{[\Omega]e_0}{[\Omega]A}   {Subderivation}
            \eqPf{[\Omega]A}{[\Omega][\Gamma]A}   {By \Lemmaref{lem:substitution-monotonicity}}
            \declchkjudgPf{\p}{[\Omega]\Gamma}{[\Omega]e_0}{[\Omega][\Gamma]A}   {By above equality}
            \chkjudgPf{\p}{\Gamma}{e_0}{[\Gamma]A}{\Delta}   {By i.h.}
\Hand       \substextendPf{\Delta}{\Omega}   {\ditto}
\Hand       \substextendPf{\Omega}{\Omega'}   {\ditto}
\Hand       \eqPf{\dom{\Delta}}{\dom{\Omega'}}   {\ditto}
            \substextendPf{\Delta}{\Omega'}   {By \Lemmaref{lem:extension-transitivity}}
            \proofsep
            \judgetpPf{\Gamma}{A\,\p}   {Given}
            \proofsep
\Hand       \synjudgPf{\p}{\Gamma}{(e_0 : A)}{[\Delta]A}{\Delta}   {By \Anno}
\Hand       \eqPf{[\Delta]A}{[\Delta][\Delta]A}   {By idempotence of substitution}
            \eqPf{A}{[\Omega]A}   {Above}
            \continueeqPf{[\Omega']A}  {By \Lemmaref{lem:completing-completeness} (ii)}
\Hand    \continueeqPf{[\Omega'][\Delta]A}   {By \Lemmaref{lem:substitution-monotonicity}}
          \end{llproof}

     \DerivationProofCase{\DeclUnitIntro}
          {}
          {\declchkjudg{p}{[\Omega]\Gamma}{\unitexp}{\unitty}}

          We have $[\Omega]A = \unitty$.  Either $[\Gamma]A = \unitty$, or $[\Gamma]A = \ahat$
          where $\ahat \in \unsolved{\Gamma}$.

          In the former case:

          \begin{llproof}
            \LetPf{\Delta}{\Gamma} {}
            \LetPf{\Omega'}{\Omega} {}
\Hand       \substextendPf{\Gamma}{\Omega}   {Given}
\Hand       \substextendPf{\Omega}{\Omega'}   {By \Lemmaref{lem:extension-reflexivity}}
\Hand       \eqPf{\dom{\Gamma}}{\dom{\Omega}}   {Given}
             \chkjudgPf{p}{\Gamma}{\unitexp}{\unitty}{\Gamma}   {By \UnitIntro}
\Hand       \chkjudgPf{p}{\Gamma}{\unitexp}{[\Gamma]\unitty}{\Gamma}   {$\unitty = [\Gamma]\unitty$}
          \end{llproof}

          In the latter case, since $A = \ahat$ and $\judgetp{\Gamma}{\ahat\;p}$ is given,
          it must be the case that $p = \OK$.
          
          \begin{llproof}
             \chkjudgPf{\OK}{\Gamma_0[\ahat : \type]}{\unitexp}{\ahat}{\Gamma_0[\hypeq{\ahat : \type}{\unitty}]}   {By \UnitIntroSolve}
\Hand  \chkjudgPf{\OK}{\Gamma_0[\ahat : \type]}{\unitexp}{\big[\Gamma_0[\ahat : \type]\big]\ahat}{\Gamma_0[\hypeq{\ahat : \type}{\unitty}]}   {By def.\ of subst.}
             \decolumnizePf
             \substextendPf{\Gamma_0[\ahat : \type]}{\Omega}   {Given}
\Hand  \substextendPf{\Gamma_0[\hypeq{\ahat : \type}{\unitty}]} {\Omega}  {By \Lemmaref{lem:parallel-extension-solve}}
\Hand  \substextendPf{\Omega}{\Omega} {By \Lemmaref{lem:extension-reflexivity}}
          \end{llproof}

     \DerivationProofCase{\DeclAllIntro}
           {\chkintro{v}
             \\
             \declchkjudg{p}{[\Omega]\Gamma, \alpha:\sort}{[\Omega]v}{A_0}
           }
           {\declchkjudg{p}{[\Omega]\Gamma}{[\Omega]v}{\alltype{\alpha:\sort} A_0}}

           \begin{llproof}
             \eqPf{[\Omega]A}{\alltype{\alpha:\sort} A_0}   {Given}
             \continueeqPf{\alltype{\alpha:\sort}{[\Omega]A'}}  {By def.\ of subst.\ and predicativity of $\Omega$}
             \eqPf{A_0}{[\Omega]A'}   {Follows from above equality}
             \declchkjudgPf{p}{[\Omega]\Gamma, \alpha:\sort}{[\Omega]v}{[\Omega]A'}  {Subderivation and above equality}
             \proofsep
             \substextendPf{\Gamma}{\Omega}  {Given}
             \substextendPf{\Gamma, \alpha:\sort}{\Omega, \alpha:\sort}  {By \substextendUU}
             \proofsep
             \eqPf{[\Omega]\Gamma, \alpha:\sort}{[\Omega, \alpha:\sort](\Gamma, \alpha:\sort)}  {By definition of context substitution}
             \declchkjudgPf{p}{[\Omega, \alpha:\sort](\Gamma, \alpha:\sort)}{[\Omega]v}{[\Omega]A'}  {By above equality}
             \declchkjudgPf{p}{[\Omega, \alpha:\sort](\Gamma, \alpha:\sort)}{[\Omega]v}{[\Omega, \alpha:\sort]A'}  {By definition of substitution}
           \end{llproof}

           \begin{llproof}
             \chkjudgPf{p}{\Gamma, \alpha:\sort}{v}{[\Gamma, \alpha:\sort]A'}{\Delta'}  {By i.h.}
             \substextendPf{\Delta'}{\Omega_0'}   {\ditto}
             \substextendPf{\Omega, \alpha:\sort}{\Omega_0'}   {\ditto}
             \eqPf{\dom{\Delta'}}{\dom{\Omega_0'}}   {\ditto}
             \substextendPf{\Gamma, \alpha:\sort}{\Delta'}  {By \Lemmaref{lem:typing-extension}}
             \eqPfParenR{\Delta'} {\Delta, \alpha:\sort, \Theta}   {By \Lemmaref{lem:extension-inversion} (i)}
             \substextendPf{\Delta, \alpha:\sort, \Theta}{\Omega_0'}   {By above equality}
             \eqPfParenR{\Omega_0'} {\Omega', \alpha:\sort, \Omega_Z}   {By \Lemmaref{lem:extension-inversion} (i)}
\Hand        \substextendPf{\Delta} {\Omega'}   {\ditto}
\Hand        \eqPf{\dom{\Delta}}{\dom{\Omega'}}   {\ditto}
\Hand        \substextendPf{\Omega}{\Omega'}  {\hspace{-14ex}By \Lemmaref{lem:extension-inversion} on $\substextend{\Omega, \alpha:\sort}{\Omega_0'}$}
             \proofsep
             \chkjudgPf{p}{\Gamma, \alpha:\sort}{v}{[\Gamma, \alpha:\sort]A'}{\Delta, \alpha:\sort, \Theta}   {By above equality}
             \chkjudgPf{p}{\Gamma, \alpha:\sort}{v}{[\Gamma]A'}{\Delta, \alpha:\sort, \Theta}   {By definition of substitution}
             \chkjudgPf{p}{\Gamma}{v}{\alltype{\alpha:\sort}{[\Gamma]A'}}{\Delta}   {By \AllIntro}
\Hand        \chkjudgPf{p}{\Gamma}{v}{[\Gamma](\alltype{\alpha:\sort}{A'})}{\Delta}   {By definition of substitution}
           \end{llproof}

     \DerivationProofCase{\DeclAllSpine}
         {\judge{[\Omega]\Gamma}{\tau}{\sort}
          \\
          \declspinejudg{[\Omega]\Gamma}{[\Omega](\appspine{e}{s_0})}{[\tau/\alpha][\Omega]A_0}{\OK}{B}{q}
         }
         {\declspinejudg{[\Omega]\Gamma}{[\Omega](\appspine{e}{s_0})}{\alltype{\alpha:\sort} [\Omega]A_0}{p}{B}{q}}

         \begin{llproof}
            \judgePf{[\Omega]\Gamma}{\tau}{\sort}     {Subderivation}
            \proofsep
            \substextendPf{\Gamma}{\Omega}   {Given}
            \substextendPf{\Gamma, \ahat:\sort}{\Omega, \hypeq{\ahat:\sort}{\tau}}  {By \substextendSolve}
            \proofsep
            \declspinejudgPf{[\Omega]\Gamma}{[\Omega](\appspine{e}{s_0})}{[\tau/\alpha][\Omega]A_0}{\OK}{B}{q}  {Subderivation}
            \eqPf{\tau}{[\Omega]\tau}   {$\FEV{\tau} = \emptyset$}
            \eqPf{[\tau/\alpha][\Omega]A_0}   {[\tau/\alpha][\Omega, \hypeq{\ahat:\sort}{\tau}]A_0}   {By def.\ of subst.}
            \continueeqPf{\big[[\Omega]\tau/\alpha\big]\big[\Omega, \hypeq{\ahat:\sort}{\tau}\big]A_0}    {By above equality}
            \continueeqPf {\big[\Omega, \hypeq{\ahat:\sort}{\tau}\big][\ahat/\alpha]A_0}    {By distributivity of substitution}
            \eqPf{[\Omega]\Gamma} {[\Omega, \hypeq{\ahat:\sort}{\tau}](\Gamma, \ahat:\sort)}  {By definition of context application}
            \decolumnizePf
            \declspinejudgPf{[\Omega, \hypeq{\ahat:\sort}{\tau}](\Gamma, \ahat:\sort)}{[\Omega](\appspine{e}{s_0})}{\big[\Omega, \hypeq{\ahat:\sort}{\tau}\big][\ahat/\alpha]A_0}{\OK}{B}{q}  {By above equalities}
            \spinejudgPf{\Gamma, \ahat:\sort}{\appspine{e}{s_0}}{[\Gamma, \ahat:\sort][\ahat/\alpha]A_0}{\OK}{B'}{q}{\Delta}   {By i.h.}
            \eqPf{B}{[\Omega, \hypeq{\ahat:\sort}{\tau}]B'}   {\ditto}
\Hand         \substextendPf{\Delta}{\Omega'}   {\ditto}
\Hand         \eqPf{\dom{\Delta}}{\dom{\Omega'}}   {\ditto}
\Hand         \substextendPf{\Omega}{\Omega'}   {\ditto}
\Hand         \substextendPf{B'}{[\Delta]B'}   {\ditto}
\Hand         \substextendPf{B}{[\Omega']B'}   {\ditto}
            \decolumnizePf
            \eqPf{[\Gamma, \ahat:\sort][\ahat/\alpha]A_0}{[\Gamma][\ahat/\alpha]A_0}   {By def.\ of context application}
            \continueeqPf{[\ahat/\alpha][\Gamma]A_0}   {$\Gamma$ does not subst.\ for $\alpha$}
            \spinejudgPf{\Gamma, \ahat:\sort}{\appspine{e}{s_0}}{[\ahat/\alpha][\Gamma]A_0}{\OK}{B'}{q}{\Delta}   {By above equality}
            \spinejudgPf{\Gamma}{\appspine{e}{s_0}} {\alltype{\alpha:\sort} [\Gamma]A_0}{p}{B'}{q}{\Delta}   {By \AllSpine}
         \Hand  \spinejudgPf{\Gamma}{\appspine{e}{s_0}} {[\Gamma](\alltype{\alpha:\sort} A_0)}{p}{B'}{q}{\Delta}   {By def.\ of subst.}
         \end{llproof}

     \DerivationProofCase{\DeclImpliesIntro}
           {\chkintro{v}
             \\
             \declchkjudg{\p}{[\Omega]\Gamma \ctxsep [\Omega]P}{[\Omega]v}{[\Omega]A_0}
           }
           {\declchkjudg{\p}{[\Omega]\Gamma}{[\Omega]v}{([\Omega]P) \implies [\Omega]A_0}}

           \begin{llproof}
             \declchkjudgPf{\p}{[\Omega]\Gamma \ctxsep [\Omega]P}{[\Omega]v}{[\Omega]A_0}  {Subderivation}
           \end{llproof}

           The concluding rule in this subderivation must be \DeclCheckBot or \DeclCheckUnify.
           In either case, $[\Omega]P$ has the form $(\sigma' = \tau')$ where $\sigma' = [\Omega]\sigma$ and $\tau' = [\Omega]\tau$.

           \smallskip

           \begin{itemize}
                 \DerivationProofCase{\DeclCheckBot}
                       {\mgu{[\Omega]\sigma}{[\Omega]\tau} = \bot}
                       {\declchkjudg{\p}{[\Omega]\Gamma \ctxsep [\Omega](\sigma = \tau)}{[\Omega]v}{[\Omega]A_0}}

                       We have $\mgu{[\Omega]\sigma}{[\Omega]\tau} = \bot$.
                       To apply \Lemmaref{lem:elimeq-completeness} (2), we need to show conditions 1--5.

                       \smallskip
           
                       \begin{llproof}
                     ***~~~~    \judgetpPf{\Gamma}{(\sigma = \tau) \implies A_0\;\p}   {Given}
                         \eqPf{[\Omega]\big((\sigma{=}\tau) \implies A_0\big)}{[\Gamma]\big((\sigma{=}\tau) \implies A_0\big)}   {By \Lemmaref{lem:substitution-tpp-stable} (i)}
                         \eqPf{[\Omega]\sigma}{[\Gamma]\sigma}  {By a property of subst.}
                         \eqPf{[\Omega]\tau}{[\Gamma]\tau}  {Similar}
                         \proofsep \PfLEFT{4}
                         \judgePf{\Gamma}{\sigma}{\sort}   {By inversion}
                      3~   \judgePf{\Gamma}{[\Gamma]\sigma}{\sort}   {By \Lemmaref{lem:substitution-sort}}
                      4~   \judgePf{\Gamma}{[\Gamma]\tau}{\sort}   {Similar}
                         \decolumnizePf
                         \eqPf{\mgu{[\Omega]\sigma}{[\Omega]\tau}}{\bot}   {Given}
                         \eqPf{\mgu{[\Gamma]\sigma}{[\Gamma]\tau}}{\bot}   {By above equalities}
                         \proofsep
                         \eqPf{\FEV{\sigma} \union \FEV{\tau}}{\emptyset}   {By inversion on ***}
                         \eqPf{\FEV{[\Omega]\sigma} \union \FEV{[\Omega]\tau}}{\emptyset}   {By a property of complete contexts}
                     5~   \eqPf{\FEV{[\Gamma]\sigma} \union \FEV{[\Gamma]\tau}}{\emptyset}   {By above equalities}
                      1~   \eqPf{[\Gamma][\Gamma]\sigma}{[\Gamma]\sigma}   {By idempotence of subst.}
                      2~   \eqPf{[\Gamma][\Gamma]\tau}{[\Gamma]\tau}  {By idempotence of subst.}
                         \decolumnizePf
                         \elimeqPf{\Gamma}{[\Gamma]\sigma}{[\Gamma]\tau}{\sort}{\bot}    {By \Lemmaref{lem:elimeq-completeness} (2)}
                         \elimpropPf{\Gamma, \MonnierComma{P}}{[\Gamma]\sigma = [\Gamma]\tau}{\bot}    {By \ElimpropEq}
                         \decolumnizePf
                         \chkjudgPf{\p}{\Gamma}{v}{([\Gamma]\sigma = [\Gamma]\tau) \implies [\Gamma]A_0}{\Gamma}   {By \ImpliesIntroBot}
                       \Hand  \chkjudgPf{\p}{\Gamma}{v}{[\Gamma]\big((\sigma = \tau) \implies A_0\big)}{\Gamma}   {By def.\ of subst.}
                       \Hand  \substextendPf{\Gamma}{\Omega}   {Given}
                       \Hand  \substextendPf{\Omega}{\Omega}   {By \Lemmaref{lem:extension-reflexivity}}
                       \Hand  \eqPf{\dom{\Gamma}}{\dom{\Omega}}   {Given}
                       \end{llproof}

                 \DerivationProofCase{\DeclCheckUnify}
                       {\mgu{[\Omega]\sigma}{[\Omega]\tau} = \theta \\
                        \declchkjudg{\p}{\theta([\Omega]\Gamma)}{\theta([\Omega]e)}{\theta([\Omega]A_0)}
                       }
                       {\declchkjudg{\p}{[\Omega]\Gamma \ctxsep (([\Omega]\sigma) = [\Omega]\tau)}{[\Omega]e}{[\Omega]A_0}}

                       We have $\mgu{[\Omega]\sigma}{[\Omega]\tau} = \theta$, and
                       will need to apply \Lemmaref{lem:elimeq-completeness} (1).  That lemma has
                       five side conditions, which can be shown exactly as in the \DeclCheckBot case above.

                       \begin{llproof}
                         \eqPf{\mgu{\sigma}{\tau}}{\theta}   {Premise}
                         \LetPf{\Omega_0}{(\Omega, \MonnierComma{P})}  {}
                         \substextendPf{\Gamma}{\Omega}  {Given}
                         \substextendPf{\Gamma, \MonnierComma{P}}{\Omega_0}  {By \substextendMonMon}
                         \proofsep
                         \eqPf{\dom{\Gamma}}{\dom{\Omega}}   {Given}
                         \eqPf{\dom{\Gamma, \MonnierComma{P}}}{\dom{\Omega_0}}   {By def.\ of $\dom{-}$}
                         \decolumnizePf
                         \elimeqPf{\Gamma, \MonnierComma{P}}{[\Gamma]\sigma}{[\Gamma]\tau}{\sort}{\Gamma, \MonnierComma{P}, \Theta}   {By \Lemmaref{lem:elimeq-completeness} (1)}
                         \decolumnizePf
                         \elimpropPf{\Gamma, \MonnierComma{P}}{[\Gamma]\sigma = [\Gamma]\tau}{\Gamma, \MonnierComma{P}, \Theta}    {By \ElimpropEq}
                   EQ0   ~~   \Pf{\text{for all}}{\judge{\Gamma, \MonnierComma{P}}{u}{\sort}.~}{%
                                             [\Gamma, \MonnierComma{P}, \Theta]u = \theta([\Gamma, \MonnierComma{P}]u)}                          {\ditto}
                         \decolumnizePf
                         \judgetpPf{\Gamma}{P \implies A_0\;\p}   {Given}
                         \judgetpPf{\Gamma}{A_0\;\p}   {By inversion}
                         \substextendPf{\Gamma}{\Omega}   {Given}
                         EQa~~~ \eqPf{[\Gamma]A_0}{[\Omega]A_0}   {By \Lemmaref{lem:substitution-tpp-stable} (i)}
                         \decolumnizePf
                         \LetPf{\Omega_1}{(\Omega, \MonnierComma{P}, \Theta)}  {}
                         \declchkjudgPf{\p}{\theta([\Omega]\Gamma)}{\theta(e)}{\theta([\Omega]A_0)}   {Subderivation}
                         \proofsep
                         \substextendPf{\Gamma, \MonnierComma{P}, \Theta}{\Omega_1}  {By induction on $\Theta$}
                         \PfLEFT{1} \proofsep
                         \eqPf{\theta([\Omega]A_0)}{\theta([\Gamma]A_0)}   {By above equality EQa}
                         \continueeqPf{[\Gamma, \MonnierComma{P}, \Theta]A_0}   {By \Lemmaref{lem:subst-upgrade} (i) (with EQ0)}
                         \continueeqPf{[\Omega_1]A_0}   {By \Lemmaref{lem:substitution-tpp-stable} (i)}
                         \continueeqPf{[\Omega_1][\Gamma, \MonnierComma{P}, \Theta]A_0}
                                 {By \Lemmaref{lem:substitution-monotonicity} (iii)}
                         \proofsep
                         \eqPf{\theta([\Omega]\Gamma)}{[\Omega_1](\Gamma, \MonnierComma{P}, \Theta)}   {By \Lemmaref{lem:subst-upgrade} (iii)}
                         \eqPf{\theta([\Omega]e)}{[\Omega_1]e}   {By \Lemmaref{lem:subst-upgrade} (iv)}
                         \PfRESTORE \decolumnizePf
                         \declchkjudgPf{\p}{[\Omega_1](\Gamma, \MonnierComma{P}, \Theta)}{[\Omega_1]e}{[\Omega_1][\Gamma, \MonnierComma{P}, \Theta]A_0}   {By above equalities}
                         \proofsep
                         \eqPf{\dom{\Gamma, \MonnierComma{P}, \Theta}}{\dom{\Omega_1}} {$\dom{\Gamma} = \dom{\Omega}$}
                         \decolumnizePf
                         \chkjudgPf{\p}{\Gamma, \MonnierComma{P}, \Theta}{e}{[\Gamma, \MonnierComma{P}, \Theta]A_0}{\Delta'}   {By i.h.}
                         \substextendPf{\Delta'}{\Omega_2'}   {\ditto}
                         \substextendPf{\Omega_1}{\Omega_2'}   {\ditto}
                         \eqPf{\dom{\Delta'}}{\dom{\Omega_2'}} {\ditto}
                         \eqPfParenR{\Delta'}{\Delta, \MonnierComma{P}, \Delta''}  {By \Lemmaref{lem:extension-inversion} (ii)}
                         \eqPfParenR{\Omega_2'}{\Omega', \MonnierComma{P}, \Omega_Z}   {By \Lemmaref{lem:extension-inversion} (ii)}
                       \Hand  \substextendPf{\Delta}{\Omega'} {\ditto}
                         \substextendPf{\Omega_0}{\Omega_2'}   {By \Lemmaref{lem:extension-transitivity}}
                         \substextendPf{\Omega, \MonnierComma{P}}{\Omega', \MonnierComma{P}, \Omega_Z}   {By above equalities}
                       \Hand  \substextendPf{\Omega}{\Omega'}   {By \Lemmaref{lem:extension-inversion} (ii)}
                       \Hand  \eqPf{\dom{\Delta}}{\dom{\Omega'}} {\ditto}
                         \decolumnizePf
                         \chkjudgPf{\p}{\Gamma, \MonnierComma{P}, \Theta}{e}{[\Gamma, \MonnierComma{P}, \Theta]A_0}{\Delta, \MonnierComma{P}, \Delta''}   {By above equality}
                         \chkjudgPf{\p}{\Gamma}{e}{([\Gamma]\sigma = [\Gamma]\tau) \implies [\Gamma]A_0}{\Delta}   {By \ImpliesIntro}
                       \Hand  \chkjudgPf{\p}{\Gamma}{e}{[\Gamma](P \implies A_0)}{\Delta}   {By def.\ of subst.}
                       \end{llproof}
           \end{itemize}

     \DerivationProofCase{\DeclImpliesSpine}
         {\declcheckprop{[\Omega]\Gamma}{[\Omega]P}
          \\
          \declspinejudg{[\Omega]\Gamma}{[\Omega](\appspine{e}{s_0})}{[\Omega]A_0}{p}{B}{q}}
         {\declspinejudg{[\Omega]\Gamma}{[\Omega](\appspine{e}{s_0})}{([\Omega]P) \implies [\Omega]A_0}{p}{B}{q}}

         \begin{llproof}
           \declcheckpropPf{[\Omega]\Gamma}{[\Omega]P}   {Subderivation}
           \declcheckpropPf{[\Omega]\Gamma}{[\Omega][\Gamma]P}   {By \Lemmaref{lem:substitution-monotonicity} (ii)}
           \checkpropPf{\Gamma}{[\Gamma]P}{\Theta}   {By \Lemmaref{lem:checkprop-completeness}}
           \substextendPf{\Theta}{\Omega_1}   {\ditto}
           \substextendPf{\Omega}{\Omega_1}   {\ditto}
           \eqPf{\dom{\Theta}}{\dom{\Omega_1}}   {\ditto}
           \proofsep
           \substextendPf{\Gamma}{\Omega}   {Given}
           \eqPf{[\Omega]\Gamma}{[\Omega_1]\Theta}   {By \Lemmaref{lem:multiple-confluence}}
           \eqPf{[\Omega]A_0}{[\Omega_1]A_0}   {By \Lemmaref{lem:completing-completeness} (ii)}
           \decolumnizePf
           \declspinejudgPf{[\Omega]\Gamma}{[\Omega](\appspine{e}{s_0})}{[\Omega]A_0}{p}{B}{q}   {Subderivation}
           \declspinejudgPf{[\Omega_1]\Theta}{[\Omega](\appspine{e}{s_0})}{[\Omega_1]A_0}{p}{B}{q}   {By above equalities}
           \spinejudgPf{\Theta}{\appspine{e}{s_0}}{[\Theta]A_0}{p}{B'}{q}{\Delta}   {By i.h.}
         \Hand  \eqPf{B'}{[\Delta]B'}   {\ditto}
         \Hand  \eqPf{\dom{\Delta}}{\dom{\Omega'}}  {\ditto}
         \Hand  \eqPf{B}{[\Omega']B'}   {\ditto}
         \Hand     \substextendPf{\Delta}{\Omega'}   {\ditto}
           \substextendPf{\Omega_1}{\Omega'}   {\ditto}
         \Hand  \substextendPf{\Omega}{\Omega'}   {By \Lemmaref{lem:extension-transitivity}}
           \eqPf{[\Theta]A_0}{[\Theta][\Gamma]A_0}   {By \Lemmaref{lem:substitution-monotonicity} (iii)}
           \decolumnizePf
           \spinejudgPf{\Theta}{\appspine{e}{s_0}}{[\Theta][\Gamma]A_0}{p}{B'}{q}{\Delta}   {By above equality}
           \spinejudgPf{\Gamma}{\appspine{e}{s_0}}{([\Gamma]P) \implies [\Gamma]A_0}{p}{B'}{q}{\Delta}   {By \ImpliesSpine}
         \Hand  \spinejudgPf{\Gamma}{\appspine{e}{s_0}}{[\Gamma](P \implies A_0)}{p}{B'}{q}{\Delta}   {By def.\ of subst.}
         \end{llproof}

     \DerivationProofCase{\DeclSumIntro{k}}
           {\declchkjudg{p}{[\Omega]\Gamma}{[\Omega]e_0}{A_k'}}
           {\declchkjudg{p}{[\Omega]\Gamma}{\inj{k} [\Omega]e_0}{\underbrace{A_1' + A_2'}_{[\Omega]A}}}

           Either $[\Gamma]A = A_1 + A_2$ (where $[\Omega]A_k = A_k'$) or $[\Gamma]A = \ahat \in \unsolved{\Gamma}$.
           
           In the former case:

           \begin{llproof}
             \declchkjudgPf{p}{[\Omega]\Gamma}{[\Omega]e_0}{A_k'}  {Subderivation}
             \declchkjudgPf{p}{[\Omega]\Gamma}{[\Omega]e_0}{[\Omega]A_k}  {$[\Omega]A_k = A_k'$}
             \chkjudgPf{p}{\Gamma}{e_0}{[\Gamma]A_k}{\Delta}  {By i.h.}
           \Hand  \substextendPf{\Delta}{\Omega}  {\ditto}
           \Hand  \eqPf{\dom{\Delta}}{\dom{\Omega'}}  {\ditto}
           \Hand  \substextendPf{\Omega}{\Omega'}  {\ditto}
             \chkjudgPf{p}{\Gamma}{\inj{k} e_0}{([\Gamma]A_1) + ([\Gamma]A_2)}{\Delta}  {By \SumIntro{k}}
             \Hand  \chkjudgPf{p}{\Gamma}{\inj{k} e_0}{[\Gamma](A_1 + A_2)}{\Delta}  {By def.\ of subst.}
           \end{llproof}

           In the latter case, $A = \ahat$ and $[\Omega]A = [\Omega]\ahat = A_1' + A_2' = \tau_1' + \tau_2'$. \\
           By inversion on $\judgetp{\Gamma}{\ahat\;p}$, it must be the case that $p = \OK$.

           \begin{llproof}
             \substextendPf{\Gamma}{\Omega}  {Given}
             \eqPf{\Gamma}{\Gamma_0[\ahat : \type]}    {$\ahat \in \unsolved{\Gamma}$}
             \eqPf{\Omega}{\Omega_0[\hypeq{\ahat : \type}{\tau_0}]}    {By \Lemmaref{lem:extension-inversion} (vi)}
             \decolumnizePf
             \LetPf{\Omega_2}{\Omega_0[\hypeq{\ahat_1 : \type}{\tau_1'}, \hypeq{\ahat_1 : \type}{\tau_2'}, \hypeq{\ahat : \type}{\ahat_1{+}\ahat_2}]}   {}
             \LetPf{\Gamma_2}{\Gamma_0[\ahat_1 : \type, \ahat_2 : \type, \hypeq{\ahat : \type}{\ahat_1{+}\ahat_2}]}   {}
             \decolumnizePf
             \substextendPf{\Gamma}{\Gamma_2}   {By \Lemmaref{lem:deep-existential} (iii) twice}
             \trailingjust{and \Lemmaref{lem:parallel-admissibility} (ii)}
             \substextendPf{\Omega}{\Omega_2}   {By \Lemmaref{lem:deep-existential} (iii) twice}
             \trailingjust{and \Lemmaref{lem:parallel-admissibility} (iii)}
             \substextendPf{\Gamma_2}{\Omega_2}   
                  {By \Lemmaref{lem:parallel-admissibility} (ii), (ii), (iii)}
             \decolumnizePf
             \declchkjudgPf{\OK}{[\Omega]\Gamma}{[\Omega]e_0}{[\Omega_2]\ahat_k}  {Subd.\ and $A_k' = \tau_k' = [\Omega_2]\ahat_k$}
             \eqPf{[\Omega]\Gamma}{[\Omega_2]\Gamma_2}   {By \Lemmaref{lem:multiple-confluence}}
             \declchkjudgPf{\OK}{[\Omega_2]\Gamma_2}{e_0}{[\Omega_2]\ahat_k}  {By above equality}
             \chkjudgPf{\OK}{\Gamma_2}{e_0}{[\Gamma_2]\ahat_k}{\Delta}  {By i.h.}
           \Hand  \substextendPf{\Delta}{\Omega'}   {\ditto}
           \Hand  \eqPf{\dom{\Delta}}{\dom{\Omega'}}  {\ditto}
             \substextendPf{\Omega_2}{\Omega'}   {\ditto}
           \Hand  \substextendPf{\Omega}{\Omega'}   {By \Lemmaref{lem:extension-transitivity}}
             \synjudgPf{\OK}{\Gamma}{\inj{k} e_0}{\ahat}{\Delta}  {By \SumIntroSolve{k}}
           \Hand  \synjudgPf{\OK}{\Gamma}{\inj{k} e_0}{[\Gamma]\ahat}{\Delta}  {$\ahat \in \unsolved{\Gamma}$}
           \end{llproof}

     \DerivationProofCase{\DeclArrIntro}
          {\declchkjudg{p}{[\Omega]\Gamma, \hyp{p}{x}{A_1'}}{[\Omega]e_0}{A_2'}
          }
          {\declchkjudg{p}{[\Omega]\Gamma}{\lam{x} [\Omega]e_0}{A_1' \arr A_2'}}

          We have $[\Omega]A = A_1' \arr A_2'$.  Either $[\Gamma]A = A_1 \arr A_2$
          where $A_1' = [\Omega]A_1$ and $A_2' = [\Omega]A_2$---or
          $[\Gamma]A = \ahat$ and $[\Omega]\ahat = A_1' \arr A_2'$.

          In the former case:

          \begin{llproof}
            \declchkjudgPf{p}{[\Omega]\Gamma, \hyp{p}{x}{A_1'}}{[\Omega]e_0}{A_2'}   {Subderivation}
            \proofsep
            \eqPf{A_1'}  {[\Omega]A_1}   {Known in this subcase}
            \continueeqPf {[\Omega][\Gamma]A_1}   {By \Lemmaref{lem:substitution-invariance}}
            \eqPf{[\Omega]A_1'} {[\Omega][\Omega][\Gamma]A_1}   {Applying $\Omega$ on both sides}
            \continueeqPf {[\Omega][\Gamma]A_1}   {By idempotence of substitution}
            \proofsep
            \eqPf{[\Omega]\Gamma, \hyp{p}{x}{A_1'}}  {[\Omega, \hyp{p}{x}{A_1'}](\Gamma, \hyp{p}{x}{[\Gamma]A_1})}   {By definition of context application}
            \decolumnizePf
            \declchkjudgPf{p}{[\Omega, \hyp{p}{x}{A_1'}](\Gamma, \hyp{p}{x}{[\Gamma]A_1})}{[\Omega]e_0}{A_2'}   {By above equality}
            \decolumnizePf
            \substextendPf{\Gamma}{\Omega}   {Given}
            \substextendPf{\Gamma, \hyp{p}{x}{[\Gamma]A_1}}{\Omega, \hyp{p}{x}{A_1'}}   {By \substextendVV}
            \proofsep
            \eqPf{\dom{\Gamma}}{\dom{\Omega}}  {Given}
            \eqPf{\dom{\Gamma, \hyp{p}{x}{[\Gamma]A_1}}}{\Omega, \hyp{p}{x}{A_1'}}  {By def.\ of $\dom{-}$}
            \chkjudgPf{p}{\Gamma, \hyp{p}{x}{[\Gamma]A_1}}{e_0}{A_2}{\Delta'}   {By i.h.}
            \substextendPf{\Delta'}{\Omega_0'}   {\ditto}
            \eqPf{\dom{\Delta'}}{\dom{\Omega_0'}}   {\ditto}
            \substextendPf{\Omega, \hyp{p}{x}{A_1'}}{\Omega_0'}   {\ditto}
            \eqPfParenR{\Omega_0'} {\Omega', \hyp{p}{x}{A_1'}, \Omega_Z}   {By \Lemmaref{lem:extension-inversion} (v)}
\Hand       \substextendPf{\Omega} {\Omega'}  {\ditto}
            \proofsep
            \substextendPf{\Gamma, \hyp{p}{x}{[\Gamma]A_1}}{\Delta'}   {By \Lemmaref{lem:typing-extension}}
            \eqPfParenR{\Delta'}{\Delta, x : \cdots, \Theta}   {By \Lemmaref{lem:extension-inversion} (v)}
            \substextendPf{\Delta, x : \cdots, \Theta}{\Omega', \hyp{p}{x}{A_1'}, \Omega_Z}  {By above equalities}
\Hand       \substextendPf{\Delta}{\Omega'}  {By \Lemmaref{lem:extension-inversion} (v)}
\Hand       \eqPf{\dom{\Delta}}{\dom{\Omega'}}  {\ditto}
            \decolumnizePf
            \chkjudgPf{p}{\Gamma, \hyp{p}{x}{[\Gamma]A_1}}{e_0}{[\Gamma]A_2}{\Delta, \hyp{p}{x}{\cdots}, \Theta}   {By above equality}
            \chkjudgPf{p}{\Gamma}{\lam{x} e_0}{([\Gamma]A_1) \arr ([\Gamma]A_2)}{\Delta}  {By \ArrIntro}
\Hand      \chkjudgPf{p}{\Gamma}{\lam{x} e_0}{[\Gamma](A_1 \arr A_2)}{\Delta}  {By definition of substitution}
          \end{llproof}
          
          In the latter case ($[\Gamma]A = \ahat \in \unsolved{\Gamma}$ and $[\Omega]\ahat = A_1' \arr A_2' = \tau_1' \arr \tau_2'$):

          By inversion on $\judgetp{\Gamma}{\ahat\;p}$, it must be the case that $p = \OK$. \\
          Since $\ahat \in \unsolved{\Gamma}$, the context $\Gamma$ must have the form $\Gamma_0[\ahat:\type]$. \\
          Let $\Gamma_2 = \Gamma_0[\ahat_1:\type, \ahat_2:\type, \hypeq{\ahat:\type}{\ahat_1{\arr}\ahat_2}]$.
          
          \begin{llproof}
            \substextendPf{\Gamma}{\Gamma_2}    {By \Lemmaref{lem:deep-existential} (iii) twice}
                \trailingjust{and \Lemmaref{lem:parallel-admissibility} (ii)}
            \eqPf{[\Omega]\ahat}{\tau_1' \arr \tau_2'}  {Known in this subcase}
            \proofsep
            \substextendPf{\Gamma}{\Omega}   {Given}
            \eqPf{\Omega}{\Omega_0[\hypeq{\ahat : \type}{\tau_0}]}    {By \Lemmaref{lem:extension-inversion} (vi)}
            \decolumnizePf
            \LetPf{\Omega_2}{\Omega_0[\hypeq{\ahat_1 : \type}{\tau_1'}, \hypeq{\ahat_1 : \type}{\tau_2'}, \hypeq{\ahat : \type}{\ahat_1{\arr}\ahat_2}]}   {}
            \decolumnizePf
            \substextendPf{\Gamma}{\Gamma_2}   {By \Lemmaref{lem:deep-existential} (iii) twice}
            \trailingjust{and \Lemmaref{lem:parallel-admissibility} (ii)}
            \substextendPf{\Omega}{\Omega_2}   {By \Lemmaref{lem:deep-existential} (iii) twice}
            \trailingjust{and \Lemmaref{lem:parallel-admissibility} (iii)}
            \substextendPf{\Gamma_2}{\Omega_2}   
                 {By \Lemmaref{lem:parallel-admissibility} (ii), (ii), (iii)}
            \decolumnizePf
            \declchkjudgPf{\OK}{[\Omega]\Gamma, \hyp{\OK}{x}{\tau_1'}}{[\Omega]e_0}{\tau_2'}   {Subderivation}
            \eqPf{[\Omega]\Gamma}{[\Omega_2]\Gamma_2}   {By \Lemmaref{lem:multiple-confluence}}
            \eqPf{\tau_2'}{[\Omega]\ahat_2}   {From above equality}
            \continueeqPf{[\Omega_2]\ahat_2}   {By \Lemmaref{lem:completing-completeness} (i)}
            \eqPf{\tau_1'}{[\Omega_2]\ahat_1}  {Similar}
            \decolumnizePf
            \eqPf{[\Omega_2]\Gamma_2, \hyp{\OK}{x}{\tau_1'}}{[\Omega_2, \hyp{\OK}{x}{\tau_1'}](\Gamma_2, \hyp{\OK}{x}{\ahat_1})}
                    {By def.\ of context application}
            \declchkjudgPf{\OK}{[\Omega_2, \hyp{\OK}{x}{\tau_1'}](\Gamma_2, \hyp{\OK}{x}{\ahat_1})}{[\Omega]e_0}{[\Omega_2]\ahat_2}
                    {By above equalities}
            \proofsep
            \eqPf{\dom{\Gamma}}{\dom{\Omega}}  {Given}
            \eqPf{\dom{\Gamma_2, \hyp{\OK}{x}{\ahat_1}}}{\dom{\Omega_2, \hyp{\OK}{x}{\tau_1'}}}  {By def.\ of $\Gamma_2$ and $\Omega_2$}
            \chkjudgPf{\OK}{\Gamma_2, \hyp{\OK}{x}{\ahat_1}} {e_0} {[\Gamma_2, \hyp{\OK}{x}{\ahat_1}]\ahat_2} {\Delta^+}     {By i.h.}
            \substextendPf{\Delta^+}{\Omega^+}  {\ditto}
            \eqPf{\dom{\Delta^+}}{\dom{\Omega^+}}  {\ditto}
            \substextendPf{\Omega_2}{\Omega^+}  {\ditto}
            \decolumnizePf
            \substextendPf{\Gamma_2, \hyp{\OK}{x}{\ahat_1}}{\Delta^+}   {By \Lemmaref{lem:typing-extension}}
            \eqPfParenR{\Delta^+}{\Delta, \hyp{\OK}{x}{\ahat_1}, \Delta_Z}     {By \Lemmaref{lem:extension-inversion} (v)}
            \eqPfParenR{\Omega^+}{\Omega', \hyp{\OK}{x}{\dots}, \Omega_Z}     {By \Lemmaref{lem:extension-inversion} (v)}
          \Hand  \substextendPf{\Delta}{\Omega'}   {\ditto}
          \Hand  \eqPf{\dom{\Delta}}{\dom{\Omega'}}  {\ditto}
            \substextendPf{\Omega}{\Omega_2}  {Above}
            \substextendPf{\Omega}{\Omega^+}  {By \Lemmaref{lem:extension-transitivity}}
          \Hand  \substextendPf{\Omega}{\Omega'}  {By \Lemmaref{lem:extension-inversion} (v)}
            \chkjudgPf{\OK}{\Gamma} {\lam{x} e_0} {\ahat} {\Delta}     {By \ArrIntroSolve}
            \eqPf{\ahat}{[\Gamma]\ahat}   {$\ahat \in \unsolved{\Gamma}$}
          \Hand  \chkjudgPf{\OK}{\Gamma} {\lam{x} e_0} {[\Gamma]\ahat} {\Delta}     {By above equality}
          \end{llproof}

\clearpage

     \DerivationProofCase{\DeclRec}
          {\declchkjudg{p}{[\Omega]\Gamma, \hyp{p}{x}{[\Omega]A}}{[\Omega]v}{[\Omega]A}
          }
          {\declchkjudg{p}{[\Omega]\Gamma}{\rec{x} [\Omega]v}{[\Omega]A}}

          \begin{llproof}
            \declchkjudgPf{p}{[\Omega]\Gamma, \hyp{p}{x}{[\Omega]A}}{[\Omega]v}{[\Omega]A}   {Subderivation}
            \proofsep
            \eqPf{[\Omega]\Gamma, \hyp{p}{x}{[\Omega]A}}  {[\Omega, \hyp{p}{x}{[\Omega]A}](\Gamma, \hyp{p}{x}{[\Gamma]A})}   {By definition of context application}
            \decolumnizePf
            \declchkjudgPf{p}{[\Omega, \hyp{p}{x}{[\Omega]A}](\Gamma, \hyp{p}{x}{[\Gamma]A})}{[\Omega]v}{[\Omega]A}   {By above equality}
            \decolumnizePf
            \substextendPf{\Gamma}{\Omega}   {Given}
            \substextendPf{\Gamma, \hyp{p}{x}{[\Gamma]A}}{\Omega, \hyp{p}{x}{[\Omega]A}}   {By \substextendVV}
            \proofsep
            \eqPf{\dom{\Gamma}}{\dom{\Omega}}  {Given}
            \eqPf{\dom{\Gamma, \hyp{p}{x}{[\Gamma]A}}}{\Omega, \hyp{p}{x}{[\Omega]A}}  {By def.\ of $\dom{-}$}
            \chkjudgPf{p}{\Gamma, \hyp{p}{x}{[\Gamma]A}}{v}{[\Gamma]A}{\Delta'}   {By i.h.}
            \substextendPf{\Delta'}{\Omega_0'}   {\ditto}
            \eqPf{\dom{\Delta'}}{\dom{\Omega_0'}}   {\ditto}
            \substextendPf{\Omega, \hyp{p}{x}{[\Omega]A}}{\Omega_0'}   {\ditto}
            \eqPfParenR{\Omega_0'} {\Omega', \hyp{p}{x}{[\Omega]A}, \Theta}   {By \Lemmaref{lem:extension-inversion} (v)}
\Hand       \substextendPf{\Omega} {\Omega'}  {\ditto}
            \proofsep
            \substextendPf{\Gamma, \hyp{p}{x}{[\Gamma]A}}{\Delta'}   {By \Lemmaref{lem:typing-extension}}
            \eqPfParenR{\Delta'}{\Delta, x : \cdots, \Theta}   {By \Lemmaref{lem:extension-inversion} (v)}
            \substextendPf{\Delta, x : \cdots, \Theta}{\Omega', \hyp{p}{x}{[\Omega]A}, \Theta}  {By above equalities}
\Hand       \substextendPf{\Delta}{\Omega'}  {By \Lemmaref{lem:extension-inversion} (v)}
\Hand       \eqPf{\dom{\Delta}}{\dom{\Omega'}}  {\ditto}
            \decolumnizePf
            \chkjudgPf{p}{\Gamma, \hyp{p}{x}{[\Gamma]A}}{v}{[\Gamma]A}{\Delta, \hyp{p}{x}{[\Gamma]A}, \Theta}   {By above equality}
\Hand            \chkjudgPf{p}{\Gamma}{\rec{x} v}{[\Gamma]A}{\Delta}  {By \Rec}
          \end{llproof}

     \DerivationProofCase{\DeclArrElim}
          {\declsynjudg{q}{[\Omega]\Gamma}{[\Omega]e_0}{A}
           \\
           \declrecspinejudg{[\Omega]\Gamma}{[\Omega]s_0}{A}{q}{C}{p}
          }
          {\declsynjudg{p}{[\Omega]\Gamma}{[\Omega](\appspine{e_0}{s_0})}{C}}
     
           \begin{llproof}
                \declsynjudgPf{q}{[\Omega]\Gamma}{[\Omega]e_0}{A}    {Subderivation}
                \synjudgPf{q}{\Gamma}{e_0}{A'}{\Theta}   {By i.h.}
                \substextendPf{\Theta}{\OmegaTheta}   {\ditto}
                \eqPf{\dom{\Theta}}{\dom{\OmegaTheta}}   {\ditto}
                \substextendPf{\Omega}{\OmegaTheta}   {\ditto}
                \eqPf{A}{[\OmegaTheta]A'}  {\ditto}
                \eqPf{A'}{[\Theta]A'}  {\ditto}
                \decolumnizePf
                \substextendPf{\Gamma}{\Omega}   {Given}
                \eqPf{[\Omega]\Gamma}{[\OmegaTheta]\Theta}   {By \Lemmaref{lem:multiple-confluence}}
                \declrecspinejudgPf{[\Omega]\Gamma}{[\Omega]s_0}{A}{q}{C}{p}  {Subderivation}
                \declrecspinejudgPf{[\OmegaTheta]\Theta}{[\Omega]s_0}{[\OmegaTheta]A'}{q}{C}{p}  {By above equalities}
                \recspinejudgPf{\Theta}{s_0}{[\Theta]A'}{q}{C'}{p}{\Delta}    {By i.h.}
             \Hand  \eqPf{C'}{[\Delta]C'}   {\ditto}
             \Hand  \substextendPf{\Delta}{\Omega'}   {\ditto}
             \Hand  \eqPf{\dom{\Delta}}{\dom{\Omega'}}  {\ditto}
                \substextendPf{\OmegaTheta}{\Omega'}   {\ditto}
             \Hand  \eqPf{C}{[\Omega']C'}  {\ditto}
                \recspinejudgPf{\Theta}{s_0}{A'}{q}{C'}{p}{\Delta}    {By above equality}
             \Hand  \substextendPf{\Omega}{\Omega'}   {By \Lemmaref{lem:extension-transitivity}}
               \proofsep
             \Hand  \synjudgPf{p}{\Gamma}{\appspine{e_0}{s_0}}{C'}{\Delta}   {By \ArrElim}
           \end{llproof}

\clearpage

     \DerivationProofCase{\DeclRecover}
          {\declspinejudg{[\Omega]\Gamma}{[\Omega]s}{[\Omega]A}{\p}{C}{\OK}
           \\
           \arrayenvbl{
                 \text{for all $C_2$.\;}
                   \\
                   \text{~~if~~}\declspinejudg{[\Omega]\Gamma}{[\Omega]s}{[\Omega]A}{\p}{C_2}{\OK}
                   \text{\;~then~}C_2 = C
           }
          }
          {\declrecspinejudg{[\Omega]\Gamma}{[\Omega]s}{[\Omega]A}{\p}{C}{\p}}

          \begin{llproof}
              \substextendPf{\Gamma}{\Omega}   {Given}
              \declspinejudgPf{[\Omega]\Gamma}{[\Omega]s}{[\Omega]A}{\p}{C}{\OK}  {Subderivation}
              \spinejudgPf{\Gamma}{s}{[\Gamma]A}{\p}{C'}{\OK}{\Delta}    {By i.h.}
          \Hand  \substextendPf{\Delta}{\Omega'}   {\ditto}
          \Hand  \substextendPf{\Omega}{\Omega'}   {\ditto}
          \Hand  \eqPf{\dom{\Delta}}{\dom{\Omega'}}   {\ditto}
          \Hand  \eqPf{C}{[\Omega']C'}  {\ditto}
          \Hand  \eqPf{C'}{[\Delta]C'}  {\ditto}
          \end{llproof}

          Suppose, for a contradiction, that $\FEV{[\Delta]C'} \neq \emptyset$.  \\
          That is, there exists some $\ahat \in \FEV{[\Delta]C'}$.
          
          \begin{llproof}
              \substextendPf{\Delta}{\Omega_2}   {By \Lemmaref{lem:split-solutions}}
              \substextendPf{\underbrace{\Omega_1'[\hypeq{\ahat:\sort}{t_1}]}_{\Omega_1}}{\Omega'}   {\ditto}
              \eqPf{\Omega_2}{\Omega_1'[\hypeq{\ahat:\sort}{t_2}]}   {\ditto}
              \neqPf{t_2}{t_1}  {\ditto}
              (NEQ) \neqPf{[\Omega_2]\ahat}{[\Omega_1']\ahat}  {By def.\ of subst.\ ($t_2 \neq t_1$)}
              (EQ) \eqPf{[\Omega_2]\bhat}{[\Omega_1']\bhat\text{~for all $\bhat \neq \ahat$}}  {By construction of $\Omega_2$}
                 \trailingjust{and $\Omega_2$ canonical}
          \end{llproof}

          Choose $\ahat_R$ such that $\ahat_R \in \FEV{C'}$ and either $\ahat_R = \ahat$ or $\ahat \in \FEV{[\Delta]\ahat_R}$. \\
          Then either $\ahat_R = \ahat$, or $\ahat_R$ is declared to the right of $\ahat$ in $\Delta$.

          \begin{llproof}
              \neqPf{[\Omega_2]C'}{[\Omega']C'}   {From (NEQ) and (EQ)}
              \spinejudgPf{\Gamma}{s}{[\Gamma]A}{\p}{C'}{\OK}{\Delta}    {Above}
              \declspinejudgPf{[\Omega_2]\Gamma}{[\Omega_2]s}{[\Omega_2][\Gamma]A}{\p}{[\Omega_2]C'}{\OK}    {By \Theoremref{thm:typing-soundness}}
              \proofsep
              \spinejudgPf{\Gamma}{s}{[\Gamma]A}{\p}{C'}{\OK}{\Delta}    {Above}
              \judgetpPf{\Gamma}{A\;\p}   {Given}
              \judgetpPf{\Gamma}{[\Gamma]A\;\p}   {By \Lemmaref{lem:substitution-tp}}
              \eqPf{\FEV{[\Gamma]A}}{\emptyset}   {By inversion}
              \exbasisPf{\cdot}{[\Gamma]A}   {Property of $\subseteq$}
              \eqPfParenR{\Delta}{\Delta_L \sep \Delta_R}    {By \Lemmaref{lem:separation-main} (Spines)}
              \sepextendPf{\Gamma}{\cdot}{\Delta_L}{\Delta_R}   {\ditto}
              \exbasisPf{\Delta_R}{C'}   {\ditto}
              \inPf{\ahat_R}{\FEV{C'}}    {Above}
              \inPf{\ahat_R}{\dom{\Delta_R}}    {Property of $\subseteq$}
              \sectPf{\dom{\Delta_L}}{\dom{\Delta_R} = \emptyset}   {$\Delta$ well-formed}
              \notinPf{\ahat_R}{\dom{\Delta_L}}   {~}
              \subseteqPf{\dom{\Gamma}}{\dom{\Delta_L}}   {By \Definitionref{def:separation-extension}}
              \notinPf{\ahat_R}{\dom{\Gamma}}   {}
          \end{llproof}

          \begin{llproof}
              \declspinejudgPf{[\Omega_2]\Gamma}{[\Omega_2]s}{[\Omega_2][\Gamma]A}{\p}{[\Omega_2]C'}{\OK}    {Above}
              \proofsep
              \Pf{}{}{\text{$\Omega_2$ and $\Omega_1$ differ only at $\ahat$}}   {Above}
              \eqPf{\FEV{[\Gamma]A}}{\emptyset}   {Above}
              \eqPf{[\Omega_2][\Gamma]A} {[\Omega_1][\Gamma]A}  {By preceding two lines}
              \proofsep
              \judgetpPf{\Gamma}{[\Gamma]A}  {Above}
              \substextendPf{\Gamma}{\Omega_2}  {By \Lemmaref{lem:extension-transitivity}}
              \judgetpPf{\Omega_2}{[\Gamma]A}  {By \Lemmaref{lem:extension-weakening-tp}}
              \eqPf{\dom{\Omega_2}}{\dom{\Omega_1}}  {$\Omega_1$ and $\Omega_2$ differ only at $\ahat$}
              \judgetpPf{\Omega_1}{[\Gamma]A}  {By \Lemmaref{lem:equal-domains-wf-tp}}
              \decolumnizePf
              \judgetpPf{\Gamma}{[\Gamma]A}  {Above}
              \judgetpPf{\Omega}{[\Gamma]A}  {By \Lemmaref{lem:extension-weakening-tp}}
              \eqPf{[\Omega_1][\Gamma]A}{[\Omega'][\Gamma]A = [\Omega][\Gamma]A}  {\hspace{-1ex}By \Lemmaref{lem:completing-completeness} (ii) twice}
              \continueeqPf{[\Omega]A}  {By \Lemmaref{lem:substitution-monotonicity} (iii)}
              \proofsep
              \eqPf{[\Omega]\Gamma}{[\Omega']\Gamma}   {By \Lemmaref{lem:multiple-confluence}}
              \continueeqPf{[\Omega_1]\Gamma}   {By \Lemmaref{lem:multiple-confluence}}
              \continueeqPf{[\Omega_2]\Gamma}   {Follows from $\ahat_R \notin \dom{\Gamma}$}
              \decolumnizePf
              \eqPf{[\Omega_2]s}{[\Omega]s}   {$\Omega_2$ and $\Omega$ differ only in $\ahat$}
              \decolumnizePf
              \declspinejudgPf{[\Omega]\Gamma}{[\Omega]s}{[\Omega]A}{\p}{[\Omega_2]C'}{\OK}    {By above equalities}
              \proofsep
              \eqPf{C}{[\Omega']C'}   {Above}
              \neqPf{[\Omega']C'}{[\Omega_2]C'}     {By def.\ of subst.}
              \neqPf{C}{[\Omega_2]C'}   {By above equality}
              \eqPf{C}{[\Omega_2]C'}   {Instantiating ``for all $C_2$'' with $C_2 = [\Omega_2]C'$}
              \contraPf{\FEV{[\Delta]C'} = \emptyset}
              \proofsep
          \Hand  \recspinejudgPf{\Gamma}{s}{[\Gamma]A}{\p}{C'}{\p}{\Delta}   {By \Recover}
          \end{llproof}

     \DerivationProofCase{\DeclPass}
          { 
              \declspinejudg{[\Omega]\Gamma}{[\Omega]s}{[\Omega]A}{p}{C}{q}
          }
          {\declrecspinejudg{[\Omega]\Gamma}{[\Omega]s}{[\Omega]A}{p}{C}{q}}
     
          \begin{llproof}
                \declspinejudgPf{[\Omega]\Gamma}{[\Omega]s}{[\Omega]A}{p}{C}{q}  {Subderivation}
                \spinejudgPf{\Gamma}{s}{[\Gamma]A}{p}{C'}{q}{\Delta}  {By i.h.}
             \Hand  \substextendPf{\Delta}{\Omega'}   {\ditto}
             \Hand  \eqPf{\dom{\Delta}}{\dom{\Omega'}}   {\ditto}
             \Hand  \substextendPf{\Omega}{\Omega'}   {\ditto}
             \Hand  \eqPf{C'}{[\Delta]C'}  {\ditto}
             \Hand  \eqPf{C}{[\Omega']C'}  {\ditto}
          \end{llproof}
          
       We distinguish cases as follows:

       \begin{itemize}
       \item If $p = \OK$ or $q = \p$, then we can just apply \Pass:

           \begin{llproof}
           \Hand  \recspinejudgPf{\Gamma}{s}{[\Gamma]A}{p}{C'}{q}{\Delta}  {By \Pass}
           \end{llproof}

       \item
           Otherwise, $p = \p$ and $q = \OK$.
           If $\FEV{C} \neq \emptyset$, we can apply \Pass, as above.
           If $\FEV{C} = \emptyset$, then we instead apply \Recover:

           \begin{llproof}
           \Hand  \recspinejudgPf{\Gamma}{s}{[\Gamma]A}{p}{C'}{\p}{\Delta}  {By \Recover}
           \end{llproof}

           Here, $q' = \p$ and $q = \OK$, so $q' \moreprincipal q$.
       \end{itemize}

     \DerivationProofCase{\DeclEmptySpine}
         {}
         {\declspinejudg{[\Omega]\Gamma}{\emptyspine}{[\Omega]A}{p}{[\Omega]A}{p}}
         
         \begin{llproof}
         \Hand  \spinejudgPf{\Gamma}{\emptyspine}{[\Gamma]A}{p}{[\Gamma]A}{p}{\Gamma}    {By \EmptySpine}
         \Hand  \eqPf{[\Gamma]A}{[\Gamma][\Gamma]A}   {By idempotence of substitution}
         \Hand  \substextendPf{\Gamma}{\Omega}   {Given}
         \Hand  \eqPf{\dom{\Gamma}}{\dom{\Omega}}   {Given}
         \Hand  \eqPf{[\Omega][\Gamma]A} {[\Omega]A}   {By \Lemmaref{lem:substitution-monotonicity} (iii)}
         \Hand  \substextendPf{\Omega}{\Omega}   {By \Lemmaref{lem:extension-reflexivity}}
         \end{llproof}

     \DerivationProofCase{\DeclArrSpine}
         {\declchkjudg{q}{[\Omega]\Gamma}{[\Omega]e_0}{[\Omega]A_1}
          \\
          \declspinejudg{[\Omega]\Gamma}{[\Omega]s_0}{[\Omega]A_2}{q}{B}{p}}
         {\declspinejudg{[\Omega]\Gamma}{[\Omega](\appspine{e_0}{s_0})}{([\Omega]A_1) \arr ([\Omega]A_2)}{q}{B}{p}}

         \begin{llproof}
           \declchkjudgPf{q}{[\Omega]\Gamma}{[\Omega]e_0}{[\Omega]A_1}   {Subderivation}
           \chkjudgPf{q}{\Gamma}{e_0}{A'}{\Theta}   {By i.h.}
           \substextendPf{\Theta}{\OmegaTheta}   {\ditto}
           \substextendPf{\Omega}{\OmegaTheta}   {\ditto}
           \eqPf{A}{[\OmegaTheta]A'}  {\ditto}
           \eqPf{A'}{[\Theta]A'}  {\ditto}
           \proofsep
           \declspinejudgPf{[\Omega]\Gamma}{[\Omega]s_0}{[\Omega]A_2}{q}{B}{p}   {Subderivation}
           \proofsep
           \spinejudgPf{\Gamma}{s_0}{A_2}{q}{B}{p}{\Delta}  {By i.h.}
         \Hand  \substextendPf{\Delta}{\Omega'}    {\ditto}
         \Hand  \eqPf{\dom{\Delta}}{\dom{\Omega'}}   {\ditto}
         \Hand  \substextendPf{\Omega}{\Omega'}    {\ditto}
         \Hand  \eqPf{B'}{[\Delta]B'}   {\ditto}
         \Hand  \eqPf{B}{[\Omega']B'}   {\ditto}
           \proofsep
         \Hand  \spinejudgPf{\Gamma}{\appspine{e_0}{s_0}}{A_1 \arr A_2}{q}{B}{p}{\Delta}  {By \ArrSpine}
         \end{llproof}
         
\clearpage
     \DerivationProofCase{\DeclWithIntro}
           {\declcheckprop{[\Omega]\Gamma}{[\Omega]P}
             \\
             \declchkjudg{p}{[\Omega]\Gamma}{[\Omega]e}{[\Omega]A_0}
           }
           {\declchkjudg{p}{[\Omega]\Gamma}{[\Omega]e}{([\Omega]A_0) \with [\Omega]P}}

           If $\notcase{e}$, then:

           \begin{llproof}
             \declcheckpropPf{[\Omega]\Gamma}{[\Omega]P}   {Subderivation}
             \checkpropPf{\Gamma}{P}{\Theta}   {By \Lemmaref{lem:checkprop-completeness}}
             \substextendPf{\Theta}{\Omega_0'}   {\ditto}
             \substextendPf{\Omega}{\Omega_0'}   {\ditto}
             \proofsep
             \substextendPf{\Gamma}{\Omega}   {Given}
             \substextendPf{\Gamma}{\Omega_0'}   {By \Lemmaref{lem:extension-transitivity}}
             \eqPf{[\Omega]\Gamma}{[\Omega]\Omega}   {By \Lemmaref{lem:completes-stability}}
             \continueeqPf{[\Omega_0']\Omega_0'}   {By \Lemmaref{lem:completing-completeness} (iii)}
             \continueeqPf{[\Omega_0']\Theta}   {By \Lemmaref{lem:completes-confluence}}
             \proofsep
             \judgetpPf{\Gamma}{A_0 \with P\;p}   {Given}
             \judgetpPf{\Gamma}{A_0\;p}   {By inversion}
             \eqPf{[\Omega]A_0}{[\Omega_0']A_0}   {By \Lemmaref{lem:completing-completeness} (ii)}
             \proofsep
             \declchkjudgPf{p}{[\Omega]\Gamma}{[\Omega]e}{[\Omega]A_0}   {Subderivation}
             \declchkjudgPf{p}{[\Omega_0']\Theta}{[\Omega]e}{[\Omega_0']A_0}   {By above equalities}
             \chkjudgPf{p}{\Theta}{e}{[\Theta]A_0}{\Delta}   {By i.h.}
           \Hand  \substextendPf{\Delta}{\Omega'}   {\ditto}
           \Hand  \eqPf{\dom{\Delta}}{\dom{\Omega'}}   {\ditto}
              \substextendPf{\Omega_0'}{\Omega'}   {\ditto}
           \Hand \substextendPf{\Omega}{\Omega'}  {By \Lemmaref{lem:extension-transitivity}}
              \proofsep
           \Hand  \chkjudgPf{p}{\Gamma}{e}{A_0 \with P}{\Delta}   {By \WithIntro}
           \end{llproof}

           Otherwise, we have $e = \case{e_0}{\Pi}$.  Let $n$ be the height of the given derivation.
           
           \begin{llproof}
             $n - 1$~~
             \declchkjudgPf{p}{[\Omega]\Gamma}{[\Omega](\case{e_0}{\Pi})}{[\Omega]A_0}   {Subderivation}
             $n - 2$~~
             \declsynjudgPf{\p}{[\Omega]\Gamma}{[\Omega]e_0}{B}   {By \Lemmaref{lem:case-invertibility}}
             $n - 2$~~
             \declmatchjudgPf{p}{[\Omega]\Gamma}{[\Omega]\Pi}{B}{[\Omega]A_0}  {\ditto}
             $n - 2$~~
             \declcoversPf{[\Omega]\Gamma}{[\Omega]\Pi}{B}   {\ditto}
             $n - 1$~~
             \declcheckpropPf{[\Omega]\Gamma}{[\Omega]P}   {Subderivation}
             $n - 1$~~
             \declmatchjudgPf{p}{[\Omega]\Gamma}{[\Omega]\Pi}{B}{([\Omega]A_0) \with ([\Omega]P)}   {By \Lemmaref{lem:match-with-exists} (1)}
             $n - 1$~~
             \declmatchjudgPf{p}{[\Omega]\Gamma}{[\Omega]\Pi}{B}{[\Omega](A_0 \with P)}   {By def.\ of subst.}
             \decolumnizePf
             \synjudgPf{\p}{\Gamma}{e_0}{B'}{\Theta}   {By i.h.}
             \substextendPf{\Theta}{\Omega_0'}   {\ditto}
             \substextendPf{\Omega}{\Omega_0'}   {\ditto}
             \eqPf{B}{[\Omega_0']B'}   {\ditto}
             \continueeqPf{[\Omega_0'][\Theta]B'}   {By \Lemmaref{lem:substitution-invariance}}
             \proofsep
             \eqPf{[\Omega]\Gamma}{[\Omega_0']\Theta}   {By \Lemmaref{lem:multiple-confluence}}
             \eqPf{[\Omega](A_0 \with P)}{[\Omega_0'](A_0 \with P)}   {By \Lemmaref{lem:completing-completeness} (ii)}
             \decolumnizePf
             $n - 1$~~
             \declmatchjudgPf{p}{[\Omega_0']\Theta}{[\Omega]\Pi}{[\Omega_0'][\Theta]B'}{[\Omega_0'](A_0 \with P)}   {By above equalities}
             \matchjudgPf{p}{\Theta}{\Pi}{[\Theta]B'}{A_0 \with P}{\Delta}  {By i.h.}
           \Hand  \substextendPf{\Delta}{\Omega'}   {\ditto}
           \Hand  \eqPf{\dom{\Delta}}{\dom{\Omega'}}   {\ditto}
             \substextendPf{\Omega_0'}{\Omega'}   {\ditto}
             \decolumnizePf
             \coversPf{\Theta}{\Pi}{[\Theta]B'}  {By \Theoremref{thm:coverage-completeness}}
           \Hand \substextendPf{\Omega}{\Omega'}  {By \Lemmaref{lem:extension-transitivity}}
           \Hand  \chkjudgPf{p}{\Gamma}{\case{e_0}{\Pi}}{A_0 \with P}{\Delta}   {By \Case}
           \end{llproof}

      \DerivationProofCase{\DeclExIntro}
            { \judge{[\Omega]\Gamma}{\tau}{\sort}
              \\
              \declchkjudg{\OK}{[\Omega]\Gamma}{e}{[\tau/\alpha][\Omega]A_0}}
            {\declchkjudg{p}{[\Omega]\Gamma}{e}{\extype{\alpha:\sort} [\Omega]A_0}
            }

            \begin{llproof}
              \declchkjudgPf{\OK}{[\Omega]\Gamma}{e}{[\tau/\alpha][\Omega]A_0}   {Subderivation}
              \LetPf{\Omega_0}{(\Omega, \hypeq{\ahat:\type}{\tau})}  {}
              \eqPf{[\Omega]\Gamma}{[\Omega_0](\Gamma, \ahat:\type)}   {By def.\ of context substitution}
              \declchkjudgPf{\OK}{[\Omega_0](\Gamma, \ahat:\type)}{e}{[\tau/\alpha][\Omega]A_0}   {By above equality}
              \eqPf{[\tau/\alpha][\Omega]A_0}{[\Omega, \hypeq{\ahat:\type}{\tau}][\ahat/\alpha]A_0}   {By a property of substitution}
              \declchkjudgPf{\OK}{[\Omega_0](\Gamma, \ahat:\type)}{e}{[\Omega_0][\ahat/\alpha]A_0}   {By above equality}
              \chkjudgPf{\OK}{\Gamma, \ahat:\type}{e}{[\ahat/\alpha]A_0}{\Delta}   {By i.h.}
            \Hand  \substextendPf{\Delta}{\Omega'}  {\ditto}
            \Hand  \eqPf{\dom{\Delta}}{\dom{\Omega'}}   {\ditto}
              \substextendPf{\Omega_0}{\Omega'}  {\ditto}
              \substextendPf{\Omega}{\Omega_0}  {By \substextendAddSolved}
            \Hand  \substextendPf{\Omega}{\Omega'}  {By \Lemmaref{lem:extension-transitivity}}
              \proofsep
            \Hand  \chkjudgPf{p}{\Gamma}{e}{\extype{\alpha:\sort} A_0} {\Delta}   {By \ExIntro}
            \end{llproof}

     \ProofCaseRule{\DeclNil}
           Similar to the first part of the \DeclWithIntro case.

     \DerivationProofCase{\DeclCons}
            {
              \declcheckprop{[\Omega]\Gamma}{([\Omega]t) = \succ{t_2}}
              \\
              \arrayenvbl{
              \declchkjudg{p}{[\Omega]\Gamma}{[\Omega]e_1}{[\Omega]A_0}
              \\
              \declchkjudg{\OK}{[\Omega]\Gamma}{[\Omega]e_2}{\big(\vectype{t_2}{[\Omega]A_0}\big)}
              }
            }
            {
              \declchkjudg{p}{[\Omega]\Gamma}{\veccons{([\Omega]e_1)}{([\Omega]e_2)}}{\big(\vectype{([\Omega]t)}{[\Omega]A_0}\big)}
            }
         
            Let $\Omega^+ = (\Omega, \MonnierComma{\ahat}, \hypeq{\ahat : \kindnat}{t_2})$.
            
           \begin{llproof}
             \declcheckpropPf{[\Omega]\Gamma}{([\Omega]t) = \succ{t_2}}   {Subderivation}
             \declcheckpropPf{[\Omega^+](\Gamma, \MonnierComma{\ahat}, \ahat : \kindnat)}{([\Omega]t) = [\Omega^+]\succ{\ahat}}   {Defs.\ of extension and subst.}
1~~~             \checkpropPf{\Gamma, \MonnierComma{\ahat}, \ahat : \kindnat}
                         {t = \succ{\ahat}}
                         {\Gamma'}   {By \Lemmaref{lem:checkprop-completeness}}
             \substextendPf{\Gamma'}{\Omega_0'}   {\ditto}
             \substextendPf{\Omega^+}{\Omega_0'}   {\ditto}
             \decolumnizePf
             \substextendPf{\Gamma, \MonnierComma{\ahat}, \ahat : \kindnat}{\Gamma'}   {By \Lemmaref{lem:checkprop-extension}}
             \substextendPf{\Gamma, \MonnierComma{\ahat}, \ahat : \kindnat}{\Omega_0'}   {By \Lemmaref{lem:extension-transitivity}}
             \eqPf{[\Omega]\Gamma}{[\Omega]\Omega}   {By \Lemmaref{lem:completes-stability}}
             \continueeqPf{[\Omega^+]\Omega^+}   {By def.\ of context application}
             \continueeqPf{[\Omega_0']\Omega_0'}   {By \Lemmaref{lem:completing-completeness} (iii)}
             \continueeqPf{[\Omega_0']\Gamma'}   {By \Lemmaref{lem:completes-confluence}}
             \proofsep
             \eqPf{[\Omega]A_0}{[\Omega^+]A_0}   {By def.\ of context application}
             \continueeqPf{[\Omega_0']A_0}   {By \Lemmaref{lem:completing-completeness} (ii)}
             \proofsep
             \declchkjudgPf{p}{[\Omega]\Gamma}{[\Omega]e_1}{[\Omega]A_0}   {Subderivation}
             \declchkjudgPf{p}{[\Omega_0']\Gamma'}{[\Omega]e_1}{[\Omega_0']A_0}   {By above equalities}
2~~~          \chkjudgPf{p}{\Gamma'}{e_1}{[\Gamma']A_0}{\Theta}   {By i.h.}
             \substextendPf{\Theta}{\Omega_0''}  {\ditto}
             \substextendPf{\Omega_0'}{\Omega_0''}  {\ditto}
             \decolumnizePf
             \declchkjudgPf{\OK}{[\Omega]\Gamma}{[\Omega]e_2}{\big(\vectype{t_2}{[\Omega]A_0}\big)}   {Subderivation}
             \declchkjudgPf{\OK}{[\Omega]\Gamma}{[\Omega]e_2}{\big(\vectype{([\Omega^+]\ahat)}{[\Omega]A_0}\big)}   {By def.\ of substitution}
             \declchkjudgPf{\OK}{[\Omega_0'']\Theta}{[\Omega]e_2}{\big(\vectype{([\Omega_0'']\ahat)}{[\Omega_0'']A_0}\big)}   {By lemmas}
             \declchkjudgPf{\OK}{[\Omega_0'']\Theta}{[\Omega]e_2}{[\Omega_0'']\big(\vectype{\ahat}{A_0}\big)}    {By def.\ of subst.}
3~~~          \chkjudgPf{p}{\Theta}{e_2}{[\Theta]A_0}{\Delta, \MonnierComma{\ahat}, \Delta'}   {By i.h.}
             \substextendPf{\Delta, \MonnierComma{\ahat}, \Delta'}{\Omega''}   {\ditto}
             \eqPf{\dom{\Delta, \MonnierComma{\ahat}, \Delta'}}{\dom{\Omega''}}   {\ditto}
             \substextendPf{\Omega_0''}{\Omega''}   {\ditto}
             \decolumnizePf
             \eqPfParenR{\Omega''}{\Omega, \MonnierComma{\ahat}, \dots}
                        {By \Lemmaref{lem:extension-inversion} (ii)}
           \Hand  \substextendPf{\Delta}{\Omega'}   {\ditto}
           \Hand  \eqPf{\dom{\Delta}}{\dom{\Omega'}}   {\ditto}
              \substextendPf{(\Gamma', \MonnierComma{\ahat}, \dots)}{\Omega'}   {By \Lemmaref{lem:extension-transitivity}}
           \Hand \substextendPf{\Omega}{\Omega'}  {By \Lemmaref{lem:extension-inversion} (ii)}
              \proofsep
           \Hand  \chkjudgPf{p}{\Gamma}{\veccons{e_1}{e_2}}{(\vectype{t}{A_0})}{\Delta}   {By \Cons}
           \end{llproof}

     \DerivationProofCase{\DeclPairIntro}
           {\declchkjudg{p}{[\Omega]\Gamma}{[\Omega]e_1}{A_1'}
             \\
             \declchkjudg{p}{[\Omega]\Gamma}{[\Omega]e_2}{A_2'}}
           {\declchkjudg{p}{[\Omega]\Gamma}{\pair{[\Omega]e_1}{[\Omega]e_2}}{A_1' \times A_2'}}

           Either $[\Gamma]A = A_1 \times A_2$ or $[\Gamma]A = \ahat \in \unsolved{\Gamma}$.

           \begin{itemize}
           \item In the first case ($[\Gamma]A = A_1 \times A_2$),
             we have $A_1' = [\Omega]A_1$ and $A_2' = [\Omega]A_2$.

               \begin{llproof}
                 \declchkjudgPf{p}{[\Omega]\Gamma}{[\Omega]e_1}{A_1'}  {Subderivation}
                 \declchkjudgPf{p}{[\Omega]\Gamma}{[\Omega]e_1}{[\Omega]A_1}  {$[\Omega]A_1 = A_1'$}
                 \chkjudgPf{p}{\Gamma}{e_1}{[\Gamma]A_1}{\Theta}  {By i.h.}
                 \substextendPf{\Theta}{\OmegaTheta}  {\ditto}
                 \eqPf{\dom{\Theta}}{\dom{\OmegaTheta}}   {\ditto}
                 \substextendPf{\Omega}{\OmegaTheta}  {\ditto}
                 \proofsep \PfLEFT{4}
                 \declchkjudgPf{p}{[\Omega]\Gamma}{[\Omega]e_2}{A_2'}  {Subderivation}
                 \declchkjudgPf{p}{[\Omega]\Gamma}{[\Omega]e_2}{[\Omega]A_2}  {$[\Omega]A_2 = A_2'$}
                 \proofsep
                 \substextendPf{\Gamma}{\Theta}  {By \Lemmaref{lem:typing-extension}}
                 \eqPf{[\Omega]\Gamma}{[\OmegaTheta]\Theta}   {By \Lemmaref{lem:multiple-confluence}}
                 \eqPf{[\Omega]A_2}{[\OmegaTheta]A_2}   {By \Lemmaref{lem:completing-completeness} (ii)}
                 \proofsep
                 \declchkjudgPf{p}{[\OmegaTheta]\Theta}{[\Omega]e_2}{[\OmegaTheta]A_2}  {By above equalities}
                 \chkjudgPf{p}{\Theta}{e_2}{[\Gamma]A_2}{\Delta}  {By i.h.}
               \Hand  \substextendPf{\Delta}{\Omega'}  {\ditto}
               \Hand  \eqPf{\dom{\Delta}}{\dom{\Omega'}}   {\ditto}
                 \substextendPf{\OmegaTheta}{\Omega'}  {\ditto}
               \Hand  \substextendPf{\Omega}{\Omega'}  {By \Lemmaref{lem:extension-transitivity}} \PfRESTORE
                 \chkjudgPf{p}{\Gamma}{\pair{e_1}{e_2}}{([\Gamma]A_1) \times ([\Gamma]A_2)}{\Delta}  {By \PairIntro}
               \Hand  \chkjudgPf{p}{\Gamma}{\pair{e_1}{e_2}}{[\Gamma](A_1 \times A_2)}{\Delta}  {By def.\ of subst.}
               \end{llproof}

           \item In the second case, where $[\Gamma]A = \ahat$, combine the corresponding subcase for
             \DeclSumIntro{k} with some straightforward additional reasoning about contexts (because here we
             have two subderivations, rather than one).

           \end{itemize}

     \DerivationProofCase{\DeclCase}
           {  \declsynjudg{q}{[\Omega]\Gamma}{[\Omega]e_0}{C}
              \\
              \declmatchjudg{p}{[\Omega]\Gamma}{[\Omega]\Pi}{C}{[\Omega]A}
              \\
              \forall D. \declsynjudg{q}{[\Omega]\Gamma}{[\Omega]e_0}{D} \implies \declcovers{[\Omega]\Gamma}{[\Omega]\Pi}{D}
           }
           {\declchkjudg{p}{[\Omega]\Gamma}{\case{[\Omega]e_0}{[\Omega]\Pi}}{[\Omega]A}}

           \begin{llproof}
             \declsynjudgPf{q}{[\Omega]\Gamma}{[\Omega]e_0}{C}   {Subderivation}
             \synjudgPf{q}{\Gamma}{e_0}{C'}{\Theta}   {By i.h.}
             \substextendPf{\Theta}{\OmegaTheta}   {\ditto}
             \eqPf{\dom{\Theta}}{\dom{\OmegaTheta}}   {\ditto}
             \substextendPf{\Omega}{\OmegaTheta}   {\ditto}
             \eqPf{C}{[\OmegaTheta]C'}   {\ditto}
             \proofsep
             \judgetpPf{\Theta}{C'\;q}  {By \Lemmaref{lem:wf-outputs}}
             \eqPf{\FEV{C'}}{\emptyset}  {By inversion}
             \eqPf{[\OmegaTheta]C'}{C'}   {By a property of substitution}
             \decolumnizePf
             \substextendPf{\Gamma}{\Omega}{Given}
             \substextendPf{\Delta}{\Omega}{Given}
             \substextendPf{\Theta}{\Omega}{By \Lemmaref{lem:extension-transitivity}}
             \eqPf{[\Omega]\Gamma}{[\Omega]\Theta = [\Omega]\Delta}{By \Lemmaref{lem:completes-confluence}}
             \substextendPf{\Gamma}{\Theta}  {By \Lemmaref{lem:typing-extension}}
             \substextendPf{\Gamma}{\OmegaTheta}  {By \Lemmaref{lem:extension-transitivity}}
             \eqPf{[\Omega]\Gamma}{[\OmegaTheta]\Theta}   {By \Lemmaref{lem:multiple-confluence}}
             \judgetpPf{\Gamma}{A}  {Given + inversion}
             \judgetpPf{\Omega}{A}  {By \Lemmaref{lem:extension-weakening-tp}}
             \eqPf{[\Omega]A}{[\OmegaTheta]A}   {By \Lemmaref{lem:completing-completeness} (ii)}
             \declmatchjudgPf{p}{[\Omega]\Gamma}{[\Omega]\Pi}{C}{[\Omega]A}   {Subderivation}
             \declmatchjudgPf{p}{[\OmegaTheta]\Theta}{[\Omega]\Pi}{[\OmegaTheta]C'}{[\OmegaTheta]A}   {By above equalities}
             \matchjudgPf{p}{\Theta}{\Pi}{C'}{[\Theta]A}{\Delta}  {By i.h.\ (v)}
           \Hand  \substextendPf{\Delta}{\Omega'}   {\ditto}
             \eqPf{\dom{\Delta}}{\dom{\Omega'}}   {\ditto}
             \substextendPf{\OmegaTheta}{\Omega}   {\ditto}
           \Hand  \substextendPf{\Omega}{\Omega'}  {By \Lemmaref{lem:extension-transitivity}}
         \end{llproof}

         \begin{llproof}    
             \declcoversPf{[\Omega]\Gamma}{[\Omega]\Pi}{C}   {Instantiation of quantifier}
             \eqPf{[\Omega]\Gamma}{[\Omega]\Delta}  {Above}
             \continueeqPf{[\Omega']\Delta}  {By \Lemmaref{lem:multiple-confluence}}
             \declcoversPf{[\Omega']\Delta}{[\Omega]\Pi}{C'}   {By above equalities}
             \substextendPf{\Delta}{\Omega'}  {By \Lemmaref{lem:extension-transitivity}}
             \proofsep
             \judgetpPf{\Gamma}{C'\;\p}   {Given}
             \substextendPf{\Gamma}{\Delta}  {By \Lemmaref{lem:typing-extension} \& \ref{lem:extension-transitivity}}
             \judgetpPf{\Delta}{C'\;\p}   {\hspace{-6ex}By \Lemmaref{lem:extension-weakening-tpp}}
             \eqPf{[\Delta]C'}{C'}   {By $\FEV{C'} = \emptyset$ and a property of subst.}
             \coversPf{\Delta}{\Pi}{C'}   {By \Theoremref{thm:coverage-completeness}}
             \proofsep
           \Hand  \chkjudgPf{p}{\Gamma}{\case{e_0}{\Pi}}{[\Gamma]A}{\Delta}   {By \Case}
           \end{llproof}

     \DerivationProofCase{\DeclPairIntro}
           {\declchkjudg{p}{[\Omega]\Gamma}{[\Omega]e_1}{A_1}
             \\
             \declchkjudg{p}{[\Omega]\Gamma}{[\Omega]e_2}{A_2}
           }
           {\declchkjudg{p}{[\Omega]\Gamma}{\pair{[\Omega]e_1}{[\Omega]e_2}}{\underbrace{A_1 \times A_2}_{[\Omega]A}}}

           Either $A = \ahat$ where $[\Omega]\ahat = A_1 \times A_2$, or $A = A_1' \times A_2'$ where
           $A_1 = [\Omega]A_1'$ and $A_2 = [\Omega]A_2'$.

           In the former case ($A = \ahat$):
           
           We have $[\Omega]\ahat = A_1 \times A_2$.  Therefore
           $A_1 = [\Omega]A_1'$ and $A_2 = [\Omega]A_2'$.
           Moreover, $\Gamma = \Gamma_0[\ahat : \sort]$.
           
           \begin{llproof}
             \declchkjudgPf{p}{[\Omega]\Gamma}{[\Omega]e_1}{[\Omega]A_1'}   {Subderivation}
             \LetPf{\Gamma'}{\Gamma_0[\ahat_1:\sort, \ahat_2:\sort, \hypeq{\ahat:\sort}{\ahat_1{+}\ahat_2}]}   {}
             \decolumnizePf
             \eqPf{[\Omega]\Gamma}{[\Omega]\Gamma'}   {By def.\ of context substitution}
             \declchkjudgPf{p}{[\Omega]\Gamma'}{[\Omega]e_1}{[\Omega]A_1'}   {By above equality}
             \chkjudgPf{p'}{\Gamma'}{e_1}{[\Gamma']A_1'}{\Theta}   {By i.h.}
             \substextendPf{\Theta}{\Omega_1}  {\ditto}
             \substextendPf{\Omega}{\Omega_1}  {\ditto}
             \eqPf{\dom{\Theta}}{\dom{\Omega_1}}  {\ditto}
             \proofsep
             \declchkjudgPf{p}{[\Omega]\Gamma}{[\Omega]e_2}{[\Omega]A_2'}  {Subderivation}
             \decolumnizePf
             \eqPf{[\Omega]\Gamma}{[\Omega_1]\Theta}   {By \Lemmaref{lem:multiple-confluence}}
             \eqPf{[\Omega]A_2'}{[\Omega_1]A_2'}   {By \Lemmaref{lem:completing-completeness} (ii)}
             \declchkjudgPf{p}{[\Omega_1]\Theta}{[\Omega]e_2}{[\Omega_1]A_2'}  {By above equalities}
             \chkjudgPf{p'}{\Theta}{e_2}{[\Theta]A_2'}{\Delta}  {By i.h.}
           \Hand
             \eqPf{\dom{\Delta}}{\dom{\Omega'}}   {\ditto}
           \Hand
             \substextendPf{\Delta}{\Omega'}   {\ditto}
             \substextendPf{\Omega_1}{\Omega'}   {\ditto}
           \Hand
             \substextendPf{\Omega}{\Omega'}   {By \Lemmaref{lem:extension-transitivity}}
           \Hand
             \chkjudgPf{p'}{\Gamma}{\pair{e_1}{e_2}}{\ahat}{\Delta}
                   {By \PairIntroSolve}
           \end{llproof}

           In the latter case ($A = A_1' \times A_2'$):
           
           \begin{llproof}
             \declchkjudgPf{p}{[\Omega]\Gamma}{[\Omega]e_1}{A_1}        {Subderivation}
             \declchkjudgPf{p}{[\Omega]\Gamma}{[\Omega]e_1}{[\Omega]A_1'}        {$A_1 = [\Omega]A_1'$}
             \chkjudgPf{p}{\Gamma}{e_1}{[\Gamma]A_1'}{\Theta}        {By i.h.}
             \substextendPf{\Theta}{\Omega_0}  {\ditto}
             \eqPf{\dom{\Theta}}{\dom{\Omega_0}}  {\ditto}
             \substextendPf{\Omega}{\Omega_0}  {\ditto}
             \proofsep
             \declchkjudgPf{p}{[\Omega]\Gamma}{[\Omega]e_2}{A_2}        {Subderivation}
             \declchkjudgPf{p}{[\Omega]\Gamma}{[\Omega]e_2}{[\Omega]A_2'}        {$A_2 = [\Omega]A_2'$}
             \judgetpPf{\Gamma}{A_1' \times A_2'\;p}   {Given ($A = A_1' \times A_2'$)}
             \judgetpPf{\Gamma}{A_2'}   {By inversion}
             \substextendPf{\Gamma}{\Omega}  {Given}
             \substextendPf{\Gamma}{\Omega_0}  {By \Lemmaref{lem:extension-transitivity}}
             \judgetpPf{\Omega_0}{A_2'}     {By \Lemmaref{lem:extension-weakening-tp}}
             \declchkjudgPf{p}{[\Omega]\Gamma}{[\Omega]e_2}{[\Omega_0]A_2'}        {By \Lemmaref{lem:completing-completeness}}
             \declchkjudgPf{p}{[\Omega]\Gamma}{[\Omega]e_2}{[\Omega_0][\Theta]A_2'}        {By \Lemmaref{lem:substitution-monotonicity} (iii)}
             \declchkjudgPf{p}{[\Omega]\Theta}{[\Omega]e_2}{[\Omega_0][\Theta]A_2'}        {By \Lemmaref{lem:multiple-confluence}}
             \chkjudgPf{p}{\Theta}{e_2}{[\Theta]A_2'}{\Delta}        {By i.h.}
           \Hand  \substextendPf{\Delta}{\Omega'}    {\ditto}
           \Hand  \eqPf{\dom{\Delta}}{\dom{\Omega'}}    {\ditto}
             \substextendPf{\Omega_0}{\Omega'}    {\ditto}
           \Hand  \substextendPf{\Omega}{\Omega'}    {By \Lemmaref{lem:extension-transitivity}}
             \decolumnizePf
             \chkjudgPf{p}{\Gamma}{\pair{e_1}{e_2}}{([\Omega]A_1) \times ([\Omega]A_2)}{\Delta}  {By \PairIntro}
           \Hand \chkjudgPf{p}{\Gamma}{\pair{e_1}{e_2}}{[\Omega](A_1 \times A_2)}{\Delta}  {By def.\ of substitution}
           \end{llproof}
\end{itemize}

Now we turn to parts (v) and (vi), completeness of matching.

\begin{itemize}
  \ProofCaseRule{\DeclMatchEmpty}
      Apply rule \MatchEmpty.

  \ProofCaseRule{\DeclMatchSeq}
      Apply the i.h.\ twice, along with standard lemmas.

  \ProofCaseRule{\DeclMatchBase}
      Apply the i.h.\ (i) and rule \MatchBase.

  \ProofCaseRule{\DeclMatchUnit}
      Apply the i.h.\ and rule \MatchUnit.

  \ProofCaseRule{\DeclMatchExists}
      By i.h.\ and rule \MatchExists.

  \ProofCaseRule{\DeclMatchPair}
      By i.h.\ and rule \MatchPair.

  \ProofCaseRule{\DeclMatchSum{k}}
      By i.h.\ and rule \MatchSum{k}.

  \DerivationProofCase{\DeclMatchWith}
       {\declmatchelimjudg{p}{[\Omega]\Gamma}{\branch{\patvec}{e}}{P}{[\Omega]A, [\Omega]\Avec}{[\Omega]C} }
       {\declmatchjudg{p}{[\Omega]\Gamma}{\branch{\patvec}{e}}{([\Omega]A \with [\Omega]P), [\Omega]\Avec}{[\Omega]C} }

       To apply the i.h.\ (vi), we will show 
       (1) $\judgetpvec{\Gamma}{(A, \Avec)\;\p}$,
       (2) $\judgeprop{\Gamma}{P}$,
       (3) $\FEV{P} = \emptyset$,
       (4) $\judgetp{\Gamma}{C\;p}$,
       (5) $\declmatchelimjudg{p}{[\Omega]\Gamma}{\branch{\patvec}{[\Omega]e}}{[\Omega]P}{[\Omega]\Avec}{[\Omega]C}$,
       and
       (6) $p' \moreprincipal p$.

      \begin{llproof}
        \judgetpvecPf{\Gamma}{(A \with P, \Avec)\;\p}   {Given}
        \judgetpPf{\Gamma}{(A \with P)\;\p}   {By inversion on \PrincipalTypevec}
        \judgetpPf{\Gamma}{A\;\p}   {By \Lemmaref{lem:principal-inversion} (3)}
     (2)~~   \judgepropPf{\Gamma}{P}   {\ditto}
     (3)~~   \eqPf{\FEV{P}}{\emptyset}   {By inversion}
     (1)~~   \judgetpvecPf{\Gamma}{(A, \Avec)\;\p}   {By inversion and \PrincipalTypevec}
        \decolumnizePf
     (4)~~   \judgetpPf{\Gamma}{C\;p}   {Given}
     (5)~~        \declmatchelimjudgPf{p}{[\Omega]\Gamma}{\branch{\patvec}{[\Omega]e}}{P}{[\Omega]A, [\Omega]\Avec}{[\Omega]C}   {Subderivation}
     (6)~~   \moreprincipalPf{p'}{p}   {Given}
        \decolumnizePf
        \matchelimjudgPf{p'}{\Gamma}{\branch{\patvec}{e}}{[\Gamma]P}{[\Gamma](A, \Avec)}{[\Gamma]C}{\Delta}  {By i.h.\ (vi)}
     \Hand   \substextendPf{\Delta}{\Omega'}   {\ditto}
     \Hand   \eqPf{\dom{\Delta}}{\dom{\Omega'}}   {\ditto}
     \Hand   \substextendPf{\Omega}{\Omega'}   {\ditto}
        \proofsep
        \matchelimjudgPf{p'}{\Gamma}{\branch{\patvec}{e}}{[\Gamma]P}{[\Gamma]A, [\Gamma]\Avec)}{[\Gamma]C}{\Delta}  {By def.\ of subst.}
        \matchjudgPf{p'}{\Gamma}{\branch{\patvec}{e}}{([\Gamma]A \with [\Gamma]P), [\Gamma]\Avec)}{[\Gamma]C}{\Delta}  {By \MatchWith}
      \Hand  \matchjudgPf{p'}{\Gamma}{\branch{\patvec}{e}}{[\Gamma]\big((A \with P), \Avec\big)}{[\Gamma]C}{\Delta}  {By def.\ of subst.}
      \end{llproof}

  \ProofCaseRule{\DeclMatchNeg}
      By i.h.\ and rule \MatchNeg.

  \ProofCaseRule{\DeclMatchWild}
      By i.h.\ and rule \MatchWild.

  \ProofCaseRule{\DeclMatchNil}
      Similar to the \DeclMatchWith case.

  \ProofCaseRule{\DeclMatchCons}
      Similar to the \DeclMatchExists and \DeclMatchWith cases.

  \DerivationProofCase{\DeclMatchBot}
           {\mgu{[\Omega]\sigma}{[\Omega]\tau} = \bot}
           {\declmatchelimjudg{p}{[\Omega]\Gamma}{[\Omega](\branch{\patvec}{e})}{[\Omega]\sigma = [\Omega]\tau}{[\Omega]\Avec}{[\Omega]C}}
       
           \begin{llproof}
           \Hand  \substextendPf{\Gamma}{\Omega}   {Given}
             \eqPf{\FEV{\sigma = \tau}}{\emptyset}  {Given}
             \eqPf{[\Omega]\sigma}{[\Gamma]\sigma}  {By \Lemmaref{lem:substitution-tpp-stable} (i)}
             \eqPf{[\Omega]\tau}{[\Gamma]\tau}  {Similar}
             \decolumnizePf
             \eqPf{\mgu{[\Omega]\sigma}{[\Omega]\tau}}{\bot}  {Given}
             \eqPf{\mgu{[\Gamma]\sigma}{[\Gamma]\tau}}{\bot}  {By above equalities}
             \elimeqPf{\Gamma}{\sigma}{\tau}{\sort}{\bot}  {By \Lemmaref{lem:elimeq-completeness} (2)}
          \Hand   \matchelimjudgPf{p}{\Gamma}{\branch{\patvec}{e}}{[\Gamma]\sigma = [\Gamma]\tau}{[\Gamma]\Avec}{[\Gamma]C}{\Gamma}   {By \MatchBot}
             \proofsep
           \Hand  \substextendPf{\Omega}{\Omega}  {By \Lemmaref{lem:extension-reflexivity}}
           \Hand  \eqPf{\dom{\Gamma}}{\dom{\Omega}}  {Given}
           \end{llproof}
      
  \DerivationProofCase{\DeclMatchUnify}
           { \mgu{[\Omega]\sigma}{[\Omega]\tau} = \theta
             \\
             \declmatchjudg{p}{\theta([\Omega]\Gamma)}{\theta(\branch{\patvec}{[\Omega]e})}{\theta([\Omega]\Avec)}{\theta([\Omega]C)}
           }
           {\declmatchelimjudg{p}{[\Omega]\Gamma}{\branch{\patvec}{[\Omega]e}}{[\Omega]\sigma = [\Omega]\tau}{[\Omega]\Avec}{[\Omega]C}}

           \begin{llproof}
             \eqPf{([\Omega]\sigma}{[\Gamma]\sigma) \AND ([\Omega]\tau = [\Gamma]\tau)}  {As in \DeclMatchBot case}
             \eqPf{\mgu{[\Omega]\sigma}{[\Omega]\tau}}{\theta}  {Given}
             \eqPf{\mgu{[\Gamma]\sigma}{[\Gamma]\tau}}{\theta}  {By above equalities}
             \proofsep
             \elimeqPf{\Gamma}{\sigma}{\tau}{\sort}{(\Gamma, \Theta)}  {By \Lemmaref{lem:elimeq-completeness} (1)}
             \eqPf{\Theta}{(\hypeq{\alpha_1}{t_1}, \dots, \hypeq{\alpha_n}{t_n})}  {\ditto}
             \eqPf{[\Gamma, \Theta]u}{\theta([\Gamma]u)}   {\ditto ~for all $\judge{\Gamma}{u}{\sort}$}
             \decolumnizePf
             \declmatchjudgPf{p}{\theta([\Omega]\Gamma)}{\theta(\branch{\patvec}{[\Omega]e})}{\theta([\Omega]\Avec)}{\theta([\Omega]C)}     {Subderivation}
             \PfLEFT{6} \proofsep
             \eqPf{\theta([\Omega]\Gamma)}{[\Omega, \MonnierComma{P}, \Theta](\Gamma, \MonnierComma{P}, \Theta)}    {By \Lemmaref{lem:subst-upgrade} (iii)}
             \eqPf{\theta([\Omega]\Avec)}{[\Omega, \MonnierComma{P}, \Theta]\Avec}    {By \Lemmaref{lem:subst-upgrade} (i) (over $\Avec$)}
             \eqPf{\theta([\Omega]C)}{[\Omega, \MonnierComma{P}, \Theta]C}    {By \Lemmaref{lem:subst-upgrade} (i)}
             \eqPf{\theta(\branch{\patvec}{[\Omega]e})}{[\Omega, \MonnierComma{P}, \Theta](\branch{\patvec}{e})}   {By \Lemmaref{lem:subst-upgrade} (iv)}
             \PfRESTORE \decolumnizePf
             \declmatchjudgPf{p}{[\Omega, \MonnierComma{P}, \Theta](\Gamma, \MonnierComma{P}, \Theta)}{[\Omega, \MonnierComma{P}, \Theta](\branch{\patvec}{e})}{[\Omega, \MonnierComma{P}, \Theta]\Avec}{[\Omega, \MonnierComma{P}, \Theta]C}     {By above equalities}
             \decolumnizePf
             \matchjudgPf{p}{\Gamma, \MonnierComma{P}, \Theta}{(\branch{\patvec}{e})}{[\Gamma, \MonnierComma{P}, \Theta]\Avec}{[\Gamma, \MonnierComma{P}, \Theta]C}{\Delta, \MonnierComma{P}, \Delta'}  {By i.h.}
             \substextendPf{\Delta, \MonnierComma{P}, \Delta'}{\Omega', \MonnierComma{P}, \Omega''}  {\ditto}
             \substextendPf{\Omega, \MonnierComma{P}, \Theta}{\Omega', \MonnierComma{P}, \Omega''}  {\ditto}
             \eqPf{\dom{\Delta, \MonnierComma{P}, \Delta'}}{\dom{\Omega', \MonnierComma{P}, \Omega''}}  {\ditto}
             \decolumnizePf
           \Hand \substextendPf{\Delta}{\Omega'}   {By \Lemmaref{lem:extension-inversion} (ii)}
           \Hand  \eqPf{\dom{\Delta}}{\dom{\Omega'}}  {\ditto}
           \Hand \substextendPf{\Omega}{\Omega'}   {By \Lemmaref{lem:extension-inversion} (ii)}
             \proofsep
          \Hand   \matchelimjudgPf{p}{\Gamma}{\branch{\patvec}{e}}{[\Gamma]\sigma = [\Gamma]\tau}{[\Gamma]\Avec}{[\Gamma]C}{\Delta}   {By \MatchUnify}
           \end{llproof}
           \qedhere
\end{itemize}

\end{proof}

%% file: flags.tex
\def\OPTIONLoudLabels{0}
\gdef\OPTIONArxiv{1}

%% file: arxivmagic.tex
\setcopyright{none}

%% file: local.tex
\usepackage[T1]{fontenc}
\usepackage{pifont}

\gdef\InsideProofFlag{0}
\gdef\CurrentProofName{}
\gdef\ProofNameForHeader{}
\gdef\HeaderProofFlag{0}

\def\OPTIONInAppendix{0}

\newcommand{\stopproofheading}{{%
       \gdef\InsideProofFlag{0}}}

\addtotheorempostfoothook[proof]{\stopproofheading}

\ifnum\OPTIONConf=0
\makeatletter
  \def\ps@xheadings{%
      \let\@mkboth\markboth%
      \def\sectionmark##1{%
        \markright {%
          {%
          \ifnum \c@secnumdepth >\z@
            \thesection\quad
          \fi
          ##1}}{}}%
      \def\subsectionmark##1{%
        \markright {%
          \ifnum \c@secnumdepth >\@ne
            \thesubsection\quad
          \fi
          ##1}}
      \let\@oddfoot\@empty\let\@evenfoot\@empty%
      \def\@oddhead{\rule[-3pt]{\textwidth}{0.01in}\hspace{-\textwidth}%
                    {%
                      \slshape{%
                        \ifnum\HeaderProofFlag=1
                             Proof of \autoref{PROOF\ProofNameForHeader}%
                        ~\bf(\nameref{PROOF\ProofNameForHeader})
                          ~\normalfont\small \textvtt{\ProofNameForHeader}%
                        \else%
                             {\rightmark}%
                        \fi%
                     }%
                   }\hfil \thepage}%
      \def\@evenhead{\@oddhead}%
      \def\@evenfoot{\!\!\!\!\begin{minipage}{\textwidth}\normalfont%
          \begin{tabular*}{\textwidth}{c}%
            \rule[-23pt]{\textwidth}{0.01in}\\
            {%
              \gdef\HeaderProofFlag{0}  %
              \ifnum\InsideProofFlag=1
                   {\normalfont
                             Proof of \autoref{PROOF\CurrentProofName}%
                        ~\bf(\nameref{PROOF\CurrentProofName})
                          ~\normalfont\small \textvtt{\CurrentProofName}%
                   }%
                   \gdef\HeaderProofFlag{1}%
                   \global\edef\ProofNameForHeader{\CurrentProofName}%
              \else
                        {\tiny\slshape\today}%
              \fi%
            }%
          \end{tabular*}%
        \end{minipage}}%
      \def\@oddfoot{\@evenfoot}%
  }
\makeatother
\fi

\ifnum\OPTIONConf=0
    \declaretheoremstyle[
      bodyfont=\sl
    ]{mytheoremstyle}

     \declaretheorem[style=mytheoremstyle]{theorem}

    
    
    \declaretheorem[style=mytheoremstyle]{definition}
\fi

\newcommand{\addmytocentry}[2]{%
   \addtocontents{toc}{\protect\contentsline {subsubsection}{\protect\numberline {#1}#2}{\thepage}{xx}}%
}
\newcommand{\MyThmRestateHook}[3]{%
  \relax
}

\newcommand{\lxbel}[1]{\label{#1}}

\makeatletter
\newcommand{\XLabel}[1]{%
  \ifcsname IDEMPOTFLAG#1\endcsname%
      \lxbel{PROOF#1}%
      \LoudLabel{#1}%
      \global\edef\CurrentProofName{#1}%
      \addmytocentry{\ref{PROOF#1}}{\hyperref[PROOF#1]{Proof of \thmt@thmname~(\thmt@optarg)}}%
  \else%
       \@ifundefined{r@PROOF#1}{%
          \textcolor{red}{\bf NO PROOF}
       }{%
          \hyperref[PROOF#1]{\textcolor{blue}{Go to proof}}%
       }%
       \Label{#1}%
       \addmytocentry{\ref{#1}}{\hyperref[#1]{\thmt@thmname~(\thmt@optarg)}%
      }%
      \expandafter\gdef\csname IDEMPOTFLAG#1\endcsname{d}%
   \fi}
\makeatother

\newcommand{\LoudLabel}[1]{\idempotentlabel{#1}%
\ifnum\OPTIONLoudLabels=1%
  \ifnum\OPTIONConf=1%
  \marginnote{\tiny\textvtt{#1}}%
  \else%
  \marginnote{\textvtt{#1}}%
  \fi%
\fi%
}

\newcommand{\idempotentlabel}[1]{%
    \ifcsname IDEMPFLAG#1\endcsname%
      \message{YYY ALREADY DEFINED: #1}
    \else%
      \message{YYZ NOT ALREADY DEFINED: #1}
      \expandafter\gdef\csname IDEMPFLAG#1\endcsname{d}%
      \label{#1}%
    \fi}

\ifnum\OPTIONLoudLabels=1%
\newcommand{\Label}[1]{\LoudLabel{#1}}%
\newcommand{\FLabel}[1]{\idempotentlabel{#1}%
{\tt\scriptsize{#1}}}%
\else%
\newcommand{\Label}[1]{\idempotentlabel{#1}}%
\newcommand{\FLabel}[1]{\idempotentlabel{#1}}%
\fi

\newcommand{\mypara}[1]{\paragraph*{#1.}}  %

\newcommand{\setof}[1]{\left\{#1\right\}}

\newcommand{\extype}[1]{\ExistsSym {#1}.\:}
\renewcommand{\implies}{\mathrel{\supset}}
\newcommand{\with}{\mathrel{\land}}
\newcommand{\binc}{\mathrel{\oplus}}

\newcommand{\ctxoutsym}{\ifnum\CompactJudgments=1%
     \dashv%
  \else%
     \dashv%
  \fi}
\newcommand{\ctxout}[1]{\mathrel{\ctxoutsym}{#1}}

\newcommand{\smallblacktriangle}{\text{\textscale{0.7}{$\blacktriangleright$}}}
\newcommand{\MonnierCommaSym}{{\smallblacktriangle}}
\newcommand{\MonnierComma}[1]{{\MonnierCommaSym}_{#1}}

\newcommand{\xfev}{\mathsf{FEV}}
\newcommand{\fev}[1]{\xfev(#1)}
\newcommand{\FEV}[1]{\fev{#1}}

\newcommand{\kindnat}{\mathbb{N}}

\newcommand{\emptyspine}{\cdot}

\newcommand{\appspine}[2]{{#1}\;{#2}}

\newcommand{\emptyctx}{\cdot}

\newcommand{\sort}{\kappa}
\newcommand{\type}{{\star}}   %

\newcommand{\ind}{\mathbb N}

\newcommand{\zero}{\mathsf{zero}}
\newcommand{\xsucc}{\mathsf{succ}}
\renewcommand{\succ}[1]{\xsucc\texttt{(}{#1}\texttt{)}}

\newcommand{\unitexp}{\text{\normalfont \tt()}}

\newcommand{\unitty}{\tyname{1}}

\newcommand{\pair}[2]{\langle{#1}, {#2}\rangle}
\newcommand{\wild}{\_}

\newcommand{\injop}[1]{\mathsf{inj}_{#1}}

\newcommand{\inj}[1]{\injop{#1}\,}
\newcommand{\inl}{\inj{1}}
\newcommand{\inr}{\inj{2}}

\newcommand{\pat}{\rho}

\newcommand{\patvec}{\vec{\pat}}

\newcommand{\branch}[2]{{#1} \Rightarrow {#2}}

\newcommand{\pvec}{\patvec}

\newcommand{\Avec}{\vec{A}}

\newcommand{\caseop}{\mathsf{case}}
\newcommand{\case}[2]{\caseop({#1}, {#2})}

\newcommand{\subtypingycolor}[1]{\textcolor{dDigPurple}{#1}}
\let\subtype\undefined
\newcommand{\subtype}{\mathrel{\normalfont\texttt{\subtypingycolor{<:}}}}  %
\newcommand{\declsubtype}{\mathrel{\leq}}

\newcommand{\Ocaml}{{OCaml}\xspace}

\newcommand{\xdom}{\mathsf{dom}}
\newcommand{\dom}[1]{\xdom(#1)}

\newcommand{\xunsolved}{\mathsf{unsolved}}
\newcommand{\unsolved}[1]{\xunsolved(#1)}

\ifnum\OPTIONConf=1%
  \newenvironment{bnfarray}%
             {\begin{center} ~\!\!\begin{array}[t]{lr@{~~}c@{~~}ll}}%
             {\end{array}\end{center}}
\else
  \newenvironment{bnfarray}%
             {\begin{center} ~\!\!\begin{array}[t]{lr@{~~}c@{~~}ll}}%
             {\end{array}\end{center}}
\fi

\newcommand{\derives}{\mathrel{::}}

\newcommand{\AllSym}{\forall}
\newcommand{\ExistsSym}{\exists}

\newcommand{\xAll}[1]{\AllSym#1}
\newcommand{\All}[1]{\xAll{#1}.\:}

\newcommand{\judgword}[1]{\mathit{#1}}
\newcommand{\fun}[1]{\lam{#1}}

\newcommand{\polaritysym}{\mathrm{pol}}
\newcommand{\polarity}[1]{\polaritysym({#1})}
\newcommand{\joinpolarity}[2]{\mathrm{join}({#1}, {#2}) }

\newcommand{\xnotWithExists}{\text{~not headed by $\with$ or $\ExistsSym$}}
\newcommand{\notWithExists}[1]{{#1}\xnotWithExists}

\newcommand{\xPos}{\textit{pos}}
\newcommand{\Pos}[1]{\xPos(#1)}
\newcommand{\xNegative}{\textit{neg}}
\newcommand{\Negative}[1]{\xNegative(#1)}
\newcommand{\Neg}[1]{\Negative{#1}}
\newcommand{\Not}[1]{\textit{non}{#1}}
\newcommand{\nonPos}[1]{\Not{\Pos{#1}}}
\newcommand{\nonNeg}[1]{\Not{\Neg{#1}}}

\newcommand{\xcheckingnotcase}{\textit{chk-I}}
\newcommand{\checkingnotcase}[1]{{#1}~\xcheckingnotcase}

\newcommand{\chkintro}[1]{\checkingnotcase{#1}}
\newcommand{\xnotcase}{\textrm{not a \textkw{case}}}
\newcommand{\notcase}[1]{{#1}~\xnotcase}

\newcommand{\declsubjudg}[4][\mathcolor{red}{\ast}]{\ensuremath{{#2} \entails {#3} \declsubtype^{#1} {#4}}}

\newcommand{\subjudg}[5][\mathcolor{red}{\ast}]{\ensuremath{{#2} \entails {#3} \subtype^{#1} {#4} \ctxout{#5}}}

\newcommand{\checkprop}[3]{{#1} \entails {#2}~\judgword{true} \ctxout{#3}}

\newcommand{\declcheckprop}[2]{{#1} \entails {#2}~\judgword{true}}

\newcommand{\eqsym}{\circeq}
\newcommand{\checkeq}[5]{{#1} \entails {#2} \eqsym {#3} : {#4} \ctxout{#5}}

\newcommand{\elimprop}[3]{{#1} \ctxsep {#2}  \ctxout{#3}}

\newcommand{\elimeq}[5]{{#1} \ctxsep {#2} \eqsym {#3} : {#4} \ctxout{#5}}

\newcommand{\instsym}{\subtypingycolor{\raisebox{-0.09ex}{\text{\normalfont{:}\hspace{-0.2ex}{=}}}}}

\newcommand{\instsymop}{\mathrel{\instsym}}

\newcommand{\equivsym}{\equiv}
\newcommand{\propequivsym}{\equivsym}

\newcommand{\equivjudg}[4]{\ensuremath{{#1} \entails {#2} \equivsym {#3} \ctxout{#4}}}

\newcommand{\propequivjudg}[4]{\ensuremath{{#1} \entails {#2} \propequivsym {#3} \ctxout{#4}}}

\newcommand{\instjudg}[5]{\ensuremath{{#1} \entails {#2} \instsymop {#3} : {#4} \ctxout{#5}}}

\newcommand{\chkcolor}{dBlue}
\newcommand{\syncolor}{dRed}

\newcommand{\chk}{\mathrel{\mathcolor{\chkcolor}{\Leftarrow}}}
\newcommand{\uncoloredsyn}{{\Rightarrow}}
\newcommand{\syn}{\mathrel{\mathcolor{\syncolor}{\uncoloredsyn}}}

\newcommand{\coverssym}{\judgword{\;covers\;}}
\newcommand{\spinejudgsym}{\mathrel{\mathcolor{dDigPurple}{\gg}}}

\newcommand{\chkjudg}[5]{\ensuremath{{#2} \entails {#3} \chk {#4}\;{#1} \ctxout{#5}}}
\newcommand{\synjudg}[5]{\ensuremath{{#2} \entails {#3} \syn {#4}\;{#1} \ctxout{#5}}}

\newcommand{\spinejudg}[7]{%
    \ensuremath{{#1} \entails {#2} : {#3}\;{#4} \spinejudgsym {#5}\;{#6} \ctxout{#7}}%
}

\newcommand{\symrecspine}[1]{\lceil{#1}\rceil}
\newcommand{\recspinejudg}[7]{%
    \ensuremath{{#1} \entails {#2} : {#3}\;{#4} \spinejudgsym {#5}\;\symrecspine{#6} \ctxout{#7}}%
}

\newcommand{\matchjudg}[7]{\ensuremath{{#3} \entails {#4} :: {#5}\;{#1} \chk {#6} \;{#2} \ctxout{#7}}}
\newcommand{\matchelimjudg}[7]{\ensuremath{{#2} \ctxsep {#4} \entails {#3} :: {#5}\;! \chk {#6}\;{#1} \ctxout{#7}}}

\newcommand{\covers}[4]{\ensuremath{{#2} \entails {#3} \coverssym {#4}\;{#1}}}
\newcommand{\coverseq}[4]{\ensuremath{{#1} \;/\; {#2} \entails {#3} \coverssym {#4} \;\p}}

\newcommand{\declchkjudg}[4]{\ensuremath{{#2} \entails {#3} \chk {#4}\;{#1}}}

\newcommand{\declsynjudg}[4]{\ensuremath{{#2} \entails {#3} \syn {#4}\;{#1}}}

\newcommand{\declspinejudg}[6]{%
    \ensuremath{{#1} \entails {#2} : {#3}\;{#4} \spinejudgsym {#5}\;{#6}}%
}
\newcommand{\SPLdeclspinejudg}[6]{%
    \ensuremath{{#1} \entails \arrayenvl{{#2} \\ \!\!\!: {#3}\;{#4} \spinejudgsym {#5}\;{#6}}}%
}

\newcommand{\declrecspinejudg}[6]{%
    \ensuremath{{#1} \entails {#2} : {#3}\;{#4} \spinejudgsym {#5}\;\lceil{#6}\rceil}%
}

\newcommand{\declmatchjudg}[6][!]{\ensuremath{{#3} \entails {#4} :: {#5}\;{#1} \chk {#6}\;{#2}}}
\newcommand{\declmatchelimjudg}[6]{\ensuremath{{#2} \ctxsep {#4} \entails {#3} :: {#5}\;! \chk {#6}\;{#1}}}

\newcommand{\expandmatch}[3]{{#1} \,\stackrel{#2}{\leadsto}\, {#3}}
\newcommand{\expandpair}[2]{\expandmatch{#1}{\times}{#2}}
\newcommand{\expandsum}[3]{\expandmatch{#1}{+}{{#2} \parallel {#3}}}
\newcommand{\expandvec}[3]{\expandmatch{#1}{\vecop}{{#2} \parallel {#3}}}
\newcommand{\expandvar}[2]{\expandmatch{#1}{\mathrm{var}}{#2}}
\newcommand{\expandunit}[2]{\expandmatch{#1}{\unitty}{#2}}

\newcommand{\declcovers}[4][!]{{#2} \entails {#3} \coverssym {#4} \,{#1}}
\newcommand{\declcoverseq}[4]{{#1} \;/\; {#2} \entails {#3} \coverssym {#4} !}

\newcommand{\guarded}[1]{{#1}\;\mathsf{guarded}}

\newcommand{\p}{{\normalfont\textbf{!}}}
\newcommand{\ok}{}
\newcommand{\OK}{\hspace*{-0.6pt}{\not{~}\text{\hspace*{-3.5pt}\normalfont{!}}}\hspace{1.3pt}}

\newcommand{\moreprincipal}{\sqsubseteq}

\newcommand{\hyp}[3]{{#2}\,{:}\,{#3}\,{#1}}
\newcommand{\hypeq}[2]{{#1}\,{=}\,{#2}}

\newcommand{\BotCtx}[1]{{#1}^{\bot}}
\newcommand{\Deltabot}{\BotCtx{\Delta}}

\newcommand{\alltype}[1]{\All{#1}}

\newcommand{\vecop}{\mathsf{Vec}}
\newcommand{\vectype}[2]{\vecop\;{#1}\;{#2}}
\newcommand{\vecnil}{{}\texttt{[]}}
\newcommand{\vecconssym}{::}
\newcommand{\veccons}[2]{{#1} \vecconssym {#2}}

\newcommand{\extendssym}{\longrightarrow}
\newcommand{\extends}[2]{{#1} \extendssym {#2}}
\newcommand{\substextend}[2]{\extends{#1}{#2}}

\newcommand{\sepsym}{\hspace{-0.1ex}{*}}
\newcommand{\sep}{\mathrel{\sepsym}}

\newlength{\zzsepextendsymsz}
\newcommand{\sepextendsym}{\mathrel{\text{%
     \settowidth{\zzsepextendsymsz}{$\longrightarrow$}%
     \raisebox{0.3ex}{$\longrightarrow$}%
     \hspace{-0.55\zzsepextendsymsz}%
     \hspace{-0.25ex}%
     \raisebox{-0.45ex}{$\sepsym$}%
     \hspace{0.55\zzsepextendsymsz}%
     \hspace{-0.70ex}%
}}}

\newcommand{\sepextend}[4]{({#1} \sep {#2}) \mathrel{\sepextendsym} ({#3} \sep {#4})}

\newcommand{\judgetp}[2]{{#1} \entails {#2}~\judgword{type}}

\newcommand{\judgetpvec}[2]{{#1} \entails {#2}~\judgword{types}}

\newcommand{\judgeprop}[2]{{#1} \entails {#2}~\judgword{prop}}

\newcommand{\exbasis}[2]{\FEV{#2} \subseteq \dom{#1}}

\newcommand{\ahat}{\hat{\alpha}}
\newcommand{\bhat}{\hat{\beta}}
\newcommand{\chat}{\hat{\gamma}}

\newcommand{\rulename}[1]{\text{\normalfont\textsf{#1}}}

\newcommand{\declcoversrule}[1]{\rulename{DeclCovers{#1}}}

\newrulecommand{DeclCoversEmpty}{\declcoversrule{Empty}}
\newrulecommand{DeclCoversVar}{\declcoversrule{Var}}
\newrulecommand{DeclCoversUnit}{\declcoversrule{$\unitty$}}
\newrulecommand{DeclCoversTimes}{\declcoversrule{$\times$}}
\newrulecommand{DeclCoversSum}{\declcoversrule{$+$}}
\newrulecommand{DeclCoversEx}{\declcoversrule{$\ExistsSym$}}
\newrulecommand{DeclCoversWith}{\declcoversrule{$\land$}}
\newrulecommand{DeclCoversWithOK}{\declcoversrule{$\land\OK$}}
\newrulecommand{DeclCoversEq}{\declcoversrule{Eq}}
\newrulecommand{DeclCoversVec}{\declcoversrule{Vec}}
\newrulecommand{DeclCoversVecOK}{\declcoversrule{Vec$\OK$}}
\newrulecommand{DeclCoversEqBot}{\declcoversrule{EqBot}}

\newcommand{\Unifyrule}[1]{\rulename{Unify#1}}
\newrulecommand{UnifyVarL}{\Unifyrule{VarL}}
\newrulecommand{UnifyVarR}{\Unifyrule{VarR}}
\newrulecommand{UnifyCon}{\Unifyrule{Con}}
\newrulecommand{UnifySeqNil}{\Unifyrule{SeqNil}}
\newrulecommand{UnifySeqCons}{\Unifyrule{SeqCons}}

\newcommand{\substextendrulename}[1]{\ensuremath{{\extendssym}{\rulename{#1}}}}
\newrulecommand{substextendId}{\substextendrulename{Id}}
\newrulecommand{substextendBot}{\substextendrulename{Bot}}

\newrulecommand{substextendUU}{\substextendrulename{Uvar}}
\newrulecommand{substextendVV}{\substextendrulename{Var}}
\newrulecommand{substextendEE}{\substextendrulename{Unsolved}}
\newrulecommand{substextendEqn}{\substextendrulename{Eqn}}
\newrulecommand{substextendSolved}{\substextendrulename{Solved}}
\newrulecommand{substextendMonMon}{\substextendrulename{Marker}}

\newrulecommand{substextendSolve}{\substextendrulename{Solve}}
\newrulecommand{substextendAdd}{\substextendrulename{Add}}
\newrulecommand{substextendAddSolved}{\substextendrulename{AddSolved}}

\newcommand{\Dsubrulename}[1]{\ensuremath{{\declsubtype}\rulename{#1}}}
\newrulecommandONE{DsubRefl}{\Dsubrulename{Refl$#1$}}
    \newrulecommand{DsubReflPos}{\DsubRefl{+}}
    \newrulecommand{DsubReflNeg}{\DsubRefl{-}}
    \newrulecommand{DsubReflPolvar}{\DsubRefl{\polvar}}
    \newrulecommand{DsubReflPm}{\DsubReflPolvar}

\newcommand{\Dsubrule}[3]{\ensuremath{{\declsubtype}_{#1}^{{#2}\rulename{#3}}}}

\newrulecommand{DsubPosNeg}{\Dsubrule{+}{-}{}}
\newrulecommand{DsubNegPos}{\Dsubrule{-}{+}{}}

\newrulecommand{DsubUnit}{\Dsubrulename{Unit}}
\newrulecommand{DsubArr}{\Dsubrulename{\ensuremath{\arr}}}
\newrulecommand{DsubAllL}{\Dsubrulename{\ensuremath{\forall}{L}}}
\newrulecommand{DsubAllR}{\Dsubrulename{\ensuremath{\forall}{R}}}

\newrulecommand{DsubExistsL}{\Dsubrulename{\ensuremath{\exists}{L}}}
\newrulecommand{DsubExistsR}{\Dsubrulename{\ensuremath{\exists}{R}}}

\newrulecommand{DsubEqE}{\Dsubrulename{\ensuremath{\rulename{Eq}}{E}}}

\newrulecommand{DsubImpL}{\Dsubrulename{\ensuremath{\implies}{L}}}
\newrulecommand{DsubImpR}{\Dsubrulename{\ensuremath{\implies}{R}}}
\newrulecommand{DsubWithL}{\Dsubrulename{\ensuremath{\with}{L}}}
\newrulecommand{DsubWithR}{\Dsubrulename{\ensuremath{\with}{R}}}

\newcommand{\Addproprulename}[1]{\ensuremath{\rulename{Addprop{#1}}}}
\newrulecommand{AddpropBot}{\Addproprulename{Bot}}
\newrulecommand{AddpropConsistent}{\Addproprulename{Consistent}}

\newcommand{\Propequivrulename}[1]{\ensuremath{{\equiv}\rulename{Prop{#1}}}}
\newrulecommand{PropequivEq}{\Propequivrulename{Eq}}

\newcommand{\Equivrulename}[1]{\ensuremath{{\equiv}\rulename{#1}}}
\newrulecommand{EquivVar}{\Equivrulename{Var}}
\newrulecommand{EquivExvar}{\Equivrulename{Exvar}}
\newrulecommand{EquivUnit}{\Equivrulename{Unit}}
\newrulecommand{EquivArr}{\Equivrulename{\ensuremath{\arr}}}
\newrulecommand{EquivBin}{\Equivrulename{\ensuremath{\binc}}}
\newrulecommand{EquivVec}{\Equivrulename{\ensuremath{\vecop}}}
\newrulecommand{EquivAll}{\Equivrulename{\ensuremath{\forall}}}
\newrulecommand{EquivExists}{\Equivrulename{\ensuremath{\exists}}}
\newrulecommand{EquivImplies}{\Equivrulename{\ensuremath{\implies}}}
\newrulecommand{EquivWith}{\Equivrulename{\ensuremath{\with}}}
\newrulecommandONE{EquivInst}{\Equivrulename{Instantiate{#1}}}
\newrulecommand{EquivInstL}{\EquivInst{L}}
\newrulecommand{EquivInstR}{\EquivInst{R}}
\newrulecommandONE{EquivSubstUniv}{\Equivrulename{SubstUniv}{#1}}
\newrulecommand{EquivSubstUnivL}{\EquivSubstUniv{L}}
\newrulecommand{EquivSubstUnivR}{\EquivSubstUniv{R}}
\newrulecommandONE{EquivSubstEx}{\Equivrulename{SubstEx}{#1}}
\newrulecommand{EquivSubstExL}{\EquivSubstEx{L}}
\newrulecommand{EquivSubstExR}{\EquivSubstEx{R}}

\newcommand{\Subrulename}[1]{\ensuremath{{\subtype}\rulename{#1}}}
\newrulecommand{SubVar}{\Subrulename{Var}}
\newrulecommand{SubUnit}{\Subrulename{Unit}}
\newrulecommand{SubExvar}{\Subrulename{Exvar}}
\newrulecommand{SubArr}{\Subrulename{\ensuremath{\arr}}}
\newrulecommand{SubImp}{\SubArr} %
\newrulecommand{SubExistsL}{\Subrulename{\ensuremath{\exists}{L}}}
\newrulecommand{SubExistsR}{\Subrulename{\ensuremath{\exists}{R}}}

\newcommand{\Subrule}[3]{\ensuremath{{\subtype}_{\text{$#1$}}^{\text{$#2$}}\rulename{#3}}}
\newrulecommand{SubPosNegL}{\Subrule{+}{-}{L}}
\newrulecommand{SubPosNegR}{\Subrule{+}{-}{R}}
\newrulecommand{SubNegPosL}{\Subrule{-}{+}{L}}
\newrulecommand{SubNegPosR}{\Subrule{-}{+}{R}}

\newrulecommand{SubAllL}{\Subrulename{\ensuremath{\forall}{L}}}
\newrulecommand{SubAllR}{\Subrulename{\ensuremath{\forall}{R}}}
\newrulecommand{SubImpliesL}{\Subrulename{\ensuremath{\implies}{L}}}
\newrulecommand{SubImpliesR}{\Subrulename{\ensuremath{\implies}{R}}}
\newrulecommand{SubWithL}{\Subrulename{\ensuremath{\with}{L}}}
\newrulecommand{SubWithR}{\Subrulename{\ensuremath{\with}{R}}}
\newrulecommand{SubBot}{\Subrulename{\ensuremath{\bot}}}

\newcommand{\SubSubst}[1]{\Subrulename{\flaming{Subst{#1}}}}
\newrulecommand{SubSubstL}{\SubSubst{L}}
\newrulecommand{SubSubstR}{\SubSubst{R}}

\newrulecommand{SubEquiv}{\Subrulename{Equiv}}

\newcommand{\Checkproprulename}[1]{\ensuremath{\rulename{Checkprop{#1}}}}
\newrulecommand{CheckpropEq}{\Checkproprulename{Eq}}

\newcommand{\Elimproprulename}[1]{\ensuremath{\rulename{Elimprop{#1}}}}
\newrulecommand{ElimpropBot}{\Elimproprulename{Bot}}
\newrulecommand{ElimpropEq}{\Elimproprulename{Eq}}

\newcommand{\Checkeqrulename}[1]{\ensuremath{\rulename{Checkeq{#1}}}}
\newrulecommand{CheckeqZero}{\Checkeqrulename{Zero}}
\newrulecommand{CheckeqSucc}{\Checkeqrulename{Succ}}
\newrulecommandONE{CheckeqInst}{\Checkeqrulename{Inst{#1}}}
\newrulecommand{CheckeqInstL}{\CheckeqInst{L}}
\newrulecommand{CheckeqInstR}{\CheckeqInst{R}}
\newrulecommand{CheckeqVar}{\Checkeqrulename{Var}}
\newrulecommand{CheckeqArr}{\Checkeqrulename{Arr}}
\newrulecommand{CheckeqBin}{\Checkeqrulename{Bin}}
\newrulecommand{CheckeqUnit}{\Checkeqrulename{Unit}}
\newrulecommandONE{CheckeqSubstUniv}{\Checkeqrulename{SubstUniv}{#1}}
\newrulecommand{CheckeqSubstUnivL}{\CheckeqSubstUniv{L}}
\newrulecommand{CheckeqSubstUnivR}{\CheckeqSubstUniv{R}}
\newrulecommandONE{CheckeqSubstEx}{\Checkeqrulename{SubstEx}{#1}}
\newrulecommand{CheckeqSubstExL}{\CheckeqSubstEx{L}}
\newrulecommand{CheckeqSubstExR}{\CheckeqSubstEx{R}}

\newcommand{\Elimeqrulename}[1]{\ensuremath{\rulename{Elimeq{#1}}}}
\newrulecommand{ElimeqZero}{\Elimeqrulename{Zero}}
\newrulecommand{ElimeqSucc}{\Elimeqrulename{Succ}}
\newrulecommandONE{ElimeqInst}{\Elimeqrulename{Inst{#1}}}
\newrulecommand{ElimeqInstL}{\ElimeqInst{L}}
\newrulecommand{ElimeqInstR}{\ElimeqInst{R}}
\newrulecommand{ElimeqUvarRefl}{\Elimeqrulename{UvarRefl}}
\newrulecommandONE{ElimeqUvar}{\Elimeqrulename{Uvar{#1}}}
\newrulecommand{ElimeqUvarL}{\ElimeqUvar{L}}
\newrulecommand{ElimeqUvarR}{\ElimeqUvar{R}}
\newrulecommand{ElimeqUvarLBot}{\ElimeqUvar{L\ensuremath{\bot}}}
\newrulecommand{ElimeqUvarRBot}{\ElimeqUvar{R\ensuremath{\bot}}}
\newrulecommand{ElimeqUnit}{\Elimeqrulename{Unit}}
\newrulecommand{ElimeqArr}{\Elimeqrulename{Arr}}
\newrulecommand{ElimeqBin}{\Elimeqrulename{Bin}}
\newrulecommand{ElimeqArrBot}{\Elimeqrulename{ArrBot}}
\newrulecommand{ElimeqBinBot}{\Elimeqrulename{BinBot}}
\newrulecommand{ElimeqClash}{\Elimeqrulename{Clash}}

\newcommand{\ElimeqSubstx}[2]{\Elimeqrulename{Subst{#1}{#2}}}
\newrulecommandONE{ElimeqSubstUniv}{\ElimeqSubstx{Univ}{#1}}
\newrulecommand{ElimeqSubstUnivL}{\ElimeqSubstUniv{L}}
\newrulecommand{ElimeqSubstUnivR}{\ElimeqSubstUniv{R}}
\newrulecommandONE{ElimeqSubstEx}{\ElimeqSubstx{Ex}{#1}}
\newrulecommand{ElimeqSubstExL}{\ElimeqSubstEx{L}}
\newrulecommand{ElimeqSubstExR}{\ElimeqSubstEx{R}}

\newcommand{\DeclSortrulename}[1]{\ensuremath{\rulename{{#1}Sort}}}
\newrulecommand{DeclUvarSort}{\DeclSortrulename{Uvar}}
\newrulecommand{DeclUnitSort}{\DeclSortrulename{Unit}}
\newrulecommand{DeclBinSort}{\DeclSortrulename{Bin}}
  \newrulecommand{DeclArrowSort}{\DeclSortrulename{Arrow}}
  \newrulecommand{DeclProductSort}{\DeclSortrulename{Product}}
  \newrulecommand{DeclSumSort}{\DeclSortrulename{Sum}}
\newrulecommand{DeclZeroSort}{\DeclSortrulename{Zero}}
\newrulecommand{DeclSuccSort}{\DeclSortrulename{Succ}}

\newcommand{\DeclWFvecrulename}[1]{\ensuremath{\rulename{Decl{#1}WF}}}
\newrulecommand{DeclTypevec}{\DeclWFvecrulename{Typevec}}

\newcommand{\DeclWFrulename}[1]{\ensuremath{\rulename{Decl{#1}WF}}}
\newrulecommand{DeclUvarWF}{\DeclWFrulename{Uvar}}
\newrulecommand{DeclUnitWF}{\DeclWFrulename{Unit}}
\newrulecommand{DeclBinWF}{\DeclWFrulename{Bin}}
  \newrulecommand{DeclArrowWF}{\DeclWFrulename{Arrow}}
  \newrulecommand{DeclSumWF}{\DeclWFrulename{Sum}}
  \newrulecommand{DeclProductWF}{\DeclWFrulename{Product}}
\newrulecommand{DeclForallWF}{\DeclWFrulename{All}}
\newrulecommand{DeclExistsWF}{\DeclWFrulename{Exists}}
\newrulecommand{DeclZeroWF}{\DeclWFrulename{Zero}}
\newrulecommand{DeclSuccWF}{\DeclWFrulename{Succ}}
\newrulecommand{DeclImpliesWF}{\DeclWFrulename{Implies}}
\newrulecommand{DeclWithWF}{\DeclWFrulename{With}}
\newrulecommand{DeclVecWF}{\DeclWFrulename{Vec}}

\newcommand{\Solverulename}[1]{\ensuremath{\rulename{PropSolve{#1}}}}
\newrulecommand{SolveEq}{\Solverulename{=}}
\newrulecommand{SolveBot}{\Solverulename{\bot}}

\newrulecommand{Typevec}{\WFrulename{Typevec}}
\newrulecommand{PrincipalTypevec}{\WFrulename{PrincipalTypevec}}

\newcommand{\WFrulename}[1]{\ensuremath{\rulename{{#1}WF}}}
\newrulecommand{VarWF}{\WFrulename{Var}}
\newrulecommand{UnitWF}{\WFrulename{Unit}}
\newrulecommand{SolvedVarWF}{\WFrulename{SolvedVar}}
\newrulecommand{BinWF}{\WFrulename{Bin}}
  \newrulecommand{ArrowWF}{\WFrulename{Arrow}}
  \newrulecommand{SumWF}{\WFrulename{Sum}}
  \newrulecommand{ProductWF}{\WFrulename{Product}}
\newrulecommand{ForallWF}{\WFrulename{Forall}}
\newrulecommand{ExistsWF}{\WFrulename{Exists}}
\newrulecommand{ImpliesWF}{\WFrulename{Implies}}
\newrulecommand{WithWF}{\WFrulename{With}}
\newrulecommand{VecWF}{\WFrulename{Vec}}

\newrulecommand{PrincipalWF}{\WFrulename{Principal}}
\newrulecommand{NonPrincipalWF}{\WFrulename{NonPrincipal}}

\newcommand{\Proprulename}[1]{\ensuremath{\rulename{{#1}Prop}}}
\newrulecommand{EqProp}{\Proprulename{Eq}}

\newcommand{\DeclProprulename}[1]{\ensuremath{\rulename{{#1}DeclProp}}}
\newrulecommand{DeclEqProp}{\DeclProprulename{Eq}}

\newcommand{\Sortrulename}[1]{\ensuremath{\rulename{{#1}Sort}}}
\newrulecommand{VarSort}{\Sortrulename{Var}}
\newrulecommand{UnitSort}{\Sortrulename{Unit}}
\newrulecommand{SolvedVarSort}{\Sortrulename{SolvedVar}}
\newrulecommand{BinSort}{\Sortrulename{Bin}}
  \newrulecommand{ArrowSort}{\Sortrulename{Arrow}}
  \newrulecommand{SumSort}{\Sortrulename{Sum}}
  \newrulecommand{ProductSort}{\Sortrulename{Product}}
\newrulecommand{ZeroSort}{\Sortrulename{Zero}}
\newrulecommand{SuccSort}{\Sortrulename{Succ}}

\newcommand{\DeclCWFrulename}[1]{\ensuremath{\rulename{{#1}DeclCtx}}}
\newrulecommand{DeclEmptyCWF}{\DeclCWFrulename{Empty}}
\newrulecommand{DeclVarCWF}{\DeclCWFrulename{Var}}
\newrulecommand{DeclHypCWF}{\DeclCWFrulename{Hyp}}

\newcommand{\CWFrulename}[1]{\ensuremath{\rulename{{#1}Ctx}}}
\newrulecommand{EmptyCWF}{\CWFrulename{Empty}}
\newrulecommand{VarCWF}{\CWFrulename{Var}}
\newrulecommand{SolvedCWF}{\CWFrulename{Solved}}
\newrulecommand{EqnCWF}{\CWFrulename{EqnVar}}
\newrulecommand{HypCWF}{\CWFrulename{Hyp}}
\newrulecommand{HypCWFp}{\CWFrulename{Hyp\p}}
\newrulecommand{MarkerCWF}{\CWFrulename{Marker}}
\newrulecommand{BotCWF}{\CWFrulename{Bot}}

\newcommand{\Instrulename}[1]{\ensuremath{\rulename{Inst{#1}}}}

\newrulecommand{InstSolve}{\Instrulename{Solve}}
\newrulecommand{InstReach}{\Instrulename{Reach}}
\newrulecommand{InstUnfold}{\Instrulename{Unfold}}
\newrulecommand{InstBin}{\Instrulename{Bin}}
  \newrulecommand{InstArr}{\Instrulename{Arr}}
\newrulecommand{InstAll}{\Instrulename{All}}
\newrulecommand{InstZero}{\Instrulename{Zero}}
\newrulecommand{InstSucc}{\Instrulename{Succ}}
\newrulecommand{InstImplies}{\Instrulename{Implies}}
\newrulecommand{InstWith}{\Instrulename{With}}

\newcommand{\Decltyrulename}[1]{\rulename{Decl{#1}}}

\newcommand{\DeclIntrorulename}[1]{\Decltyrulename{\ensuremath{#1}I}}
\newcommand{\DeclIntroSynrulename}[1]{\Decltyrulename{\ensuremath{#1}I$\syn$}}
\newcommand{\DeclElimrulenameextra}[2]{\Decltyrulename{\ensuremath{#1}E{#2}}}
\newcommand{\DeclElimrulename}[1]{\DeclElimrulenameextra{#1}{}}
\newcommand{\DeclApprulename}[1]{\Decltyrulename{\ensuremath{#1}App}}
\newcommand{\DeclSpinerulename}[1]{\Decltyrulename{\ensuremath{#1}Spine}}

\newcommand{\DeclMatchrulename}[1]{\Decltyrulename{Match\ensuremath{#1}}}

\newrulecommand{DeclEmptySpine}{\DeclSpinerulename{\text{Empty}}}
\newrulecommand{DeclArrSpine}{\DeclSpinerulename{\arr}}
\newrulecommand{DeclAllSpine}{\DeclSpinerulename{\forall}}
\newrulecommand{DeclImpliesSpine}{\DeclSpinerulename{\implies}}

\newcommand{\DeclRecspinerulename}[1]{\Decltyrulename{Spine{#1}}}
\newrulecommand{DeclRecover}{\DeclRecspinerulename{Recover}}
\newrulecommand{DeclPass}{\DeclRecspinerulename{Pass}}

\newrulecommand{DeclVar}{\Decltyrulename{Var}}
\newrulecommand{DeclSub}{\Decltyrulename{Sub}}
\newrulecommand{DeclAnno}{\Decltyrulename{Anno}}

\newrulecommand{DeclRec}{\Decltyrulename{Rec}}

\newrulecommand{DeclUnitIntro}{\DeclIntrorulename{\unitty}}
\newrulecommand{DeclUnitIntroSyn}{\DeclIntroSynrulename{\unitty}}

\newrulecommand{DeclArrIntro}{\DeclIntrorulename{\arr}}
\newrulecommand{DeclArrIntroSyn}{\DeclIntroSynrulename{\arr}}
\newrulecommand{DeclArrElim}{\DeclElimrulename{\arr}}
\newrulecommand{DeclArrElimSplit}{\DeclElimrulenameextra{\arr}{-Split}}
\newrulecommand{DeclAllIntro}{\DeclIntrorulename{\AllSym}}
\newrulecommand{DeclImpliesIntro}{\DeclIntrorulename{\implies}}

\newrulecommand{DeclWithIntro}{\DeclIntrorulename{\with}}
\newrulecommand{DeclExIntro}{\DeclIntrorulename{\exists}}

\newrulecommand{DeclArrApp}{\DeclApprulename{\arr}}
\newrulecommand{DeclAllApp}{\DeclApprulename{\forall}}

\newrulecommand{DeclImpliesApp}{\DeclApprulename{\implies}}

\newrulecommand{DeclCase}{\Decltyrulename{\rulename{Case}}}
\newrulecommand{DeclCaseOK}{\Decltyrulename{\rulename{Case$\OK$}}}

\newrulecommand{DeclPairIntro}{\DeclIntrorulename{\times}}
\newrulecommandONE{DeclSumIntro}{\DeclIntrorulename{+}\ensuremath{{}_{#1}}}
\newrulecommand{DeclSumIntroL}{\DeclSumIntro{1}}
\newrulecommand{DeclSumIntroR}{\DeclSumIntro{2}}

\newrulecommand{DeclNil}{\Decltyrulename{\rulename{Nil}}}
\newrulecommand{DeclCons}{\Decltyrulename{\rulename{Cons}}}

\newrulecommand{DeclPairIntroSyn}{\DeclIntroSynrulename{\times}}

\newrulecommandONE{DeclSumIntroSyn}{\DeclIntroSynrulename{+}\ensuremath{{}_{#1}}}
\newrulecommand{DeclSumIntroLSyn}{\DeclSumIntroSyn{1}}
\newrulecommand{DeclSumIntroRSyn}{\DeclSumIntroSyn{2}}

\newrulecommand{DeclCheckBot}{\Decltyrulename{Check\ensuremath{\bot}}}
\newrulecommand{DeclCheckUnify}{\Decltyrulename{CheckUnify}}

\newrulecommand{DeclCheckpropEq}{\rulename{DeclCheckpropEq}}

\newcommand{\coversrule}[1]{\rulename{Covers#1}}

\newrulecommand{CoversEmpty}{\coversrule{Empty}}
\newrulecommand{CoversVar}{\coversrule{Var}}
\newrulecommand{CoversUnit}{\coversrule{$\unitty$}}
\newrulecommand{CoversTimes}{\coversrule{$\times$}}
\newrulecommand{CoversSum}{\coversrule{$+$}}
\newrulecommand{CoversEx}{\coversrule{$\ExistsSym$}}
\newrulecommand{CoversWith}{\coversrule{$\land$}}
\newrulecommand{CoversWithOK}{\coversrule{$\land\OK$}}
\newrulecommand{CoversVec}{\coversrule{Vec}}
\newrulecommand{CoversVecOK}{\coversrule{Vec$\OK$}}
\newrulecommand{CoversEq}{\coversrule{Eq}}
\newrulecommand{CoversEqBot}{\coversrule{EqBot}}

\newrulecommand{MatchEmpty}{\Matchrulename{\mbox{Empty}}}
\newrulecommand{MatchSeq}{\Matchrulename{\mbox{Seq}}}
\newrulecommand{MatchBase}{\Matchrulename{\mbox{Base}}}
\newrulecommand{MatchUnit}{\Matchrulename{\mbox{Unit}}}
\newrulecommand{MatchPair}{\Matchrulename{\times}}
\newrulecommandONE{MatchSum}{\Matchrulename{+\ensuremath{{}_{#1}}}}
\newrulecommand{MatchSumL}{\MatchSum{\mathsf{1}}}
\newrulecommand{MatchSumR}{\MatchSum{\mathsf{2}}}

\newrulecommand{MatchWith}{\Matchrulename{\with}}
\newrulecommand{MatchWithOK}{\Matchrulename{\with\OK}}
\newrulecommand{MatchExists}{\Matchrulename{\exists}}
\newrulecommand{MatchUnify}{\Matchrulename{\mbox{Unify}}}
\newrulecommand{MatchBot}{\Matchrulename{\bot}}
\newrulecommand{MatchNeg}{\Matchrulename{\mbox{Neg}}}
\newrulecommand{MatchWild}{\Matchrulename{\mbox{Wild}}}

\newrulecommand{MatchNil}{\Matchrulename{\mbox{Nil}}}
\newrulecommand{MatchCons}{\Matchrulename{\mbox{Cons}}}
\newrulecommand{MatchNilOK}{\Matchrulename{\mbox{Nil}\OK}}
\newrulecommand{MatchConsOK}{\Matchrulename{\mbox{Cons}\OK}}

\newrulecommand{DeclMatchEmpty}{\DeclMatchrulename{\mbox{Empty}}}
\newrulecommand{DeclMatchSeq}{\DeclMatchrulename{\mbox{Seq}}}
\newrulecommand{DeclMatchBase}{\DeclMatchrulename{\mbox{Base}}}
\newrulecommand{DeclMatchUnit}{\DeclMatchrulename{\mbox{Unit}}}

\newrulecommand{DeclMatchPair}{\DeclMatchrulename{\times}}
\newrulecommandONE{DeclMatchSum}{\DeclMatchrulename{+\ensuremath{{}_{#1}}}}
\newrulecommand{DeclMatchSumL}{\DeclMatchSum{1}}
\newrulecommand{DeclMatchSumR}{\DeclMatchSum{2}}
\newrulecommand{DeclMatchWith}{\DeclMatchrulename{\with}}
\newrulecommand{DeclMatchWithOK}{\DeclMatchrulename{{\with}\OK}}
\newrulecommand{DeclMatchExists}{\DeclMatchrulename{\exists}}
\newrulecommand{DeclMatchUnify}{\DeclMatchrulename{\mbox{Unify}}}
\newrulecommand{DeclMatchBot}{\DeclMatchrulename{\bot}}
\newrulecommand{DeclMatchNeg}{\DeclMatchrulename{\mbox{Neg}}}
\newrulecommand{DeclMatchWild}{\DeclMatchrulename{\mbox{Wild}}}

\newrulecommand{DeclMatchNil}{\DeclMatchrulename{\mbox{Nil}}}
\newrulecommand{DeclMatchCons}{\DeclMatchrulename{\mbox{Cons}}}
\newrulecommand{DeclMatchNilOK}{\DeclMatchrulename{\mbox{Nil$\OK$}}}
\newrulecommand{DeclMatchConsOK}{\DeclMatchrulename{\mbox{Cons$\OK$}}}

\newcommand{\Tyrulename}[1]{\ensuremath{\rulename{#1}}}

\newcommand{\Introrulename}[1]{\Tyrulename{\ensuremath{#1}I}}
\newcommand{\IntroSynrulename}[1]{\Tyrulename{\ensuremath{#1}I$\syn$}}
\newcommand{\IntroSolverulename}[1]{\Tyrulename{\ensuremath{#1}I$\ahat$}}
\newcommand{\Elimrulenameextra}[2]{\Tyrulename{\ensuremath{#1}E{#2}}}
\newcommand{\Elimrulename}[1]{\Elimrulenameextra{#1}{}}
\newcommand{\Apprulename}[1]{\Tyrulename{\ensuremath{#1}App}}
\newcommand{\Spinerulename}[1]{\Tyrulename{\ensuremath{#1}Spine}}

\newrulecommand{Rec}{\Tyrulename{Rec}}

\newrulecommand{EmptySpine}{\Spinerulename{\text{Empty}}}
\newrulecommand{ArrSpine}{\Spinerulename{\arr}}
\newrulecommand{AllSpine}{\Spinerulename{\forall}}
\newrulecommand{ImpliesSpine}{\Spinerulename{\implies}}
\newrulecommand{SolveSpine}{\Spinerulename{\ahat}}

\newcommand{\Recspinerulename}[1]{\Tyrulename{Spine{#1}}}

\newrulecommand{Recover}{\Recspinerulename{Recover}}
\newrulecommand{Pass}{\Recspinerulename{Pass}}

\newrulecommand{Var}{\Tyrulename{Var}}
\newrulecommand{Sub}{\Tyrulename{Sub}}
\newrulecommand{Anno}{\Tyrulename{Anno}}

\newrulecommand{SubstSyn}{\Tyrulename{Subst$\syn$}}
\newrulecommand{SubstChk}{\Tyrulename{Subst$\chk$}}

\newrulecommand{UnitIntro}{\Introrulename{\unitty}}
\newrulecommand{UnitIntroSyn}{\IntroSynrulename{\unitty}}
\newrulecommand{UnitIntroSolve}{\IntroSolverulename{\unitty}}

\newrulecommand{ArrIntro}{\Introrulename{\arr}}
\newrulecommand{ArrIntroSyn}{\IntroSynrulename{\arr}}
\newrulecommand{ArrIntroSolve}{\IntroSolverulename{\arr}}
\newrulecommand{ArrElim}{\Elimrulename{\arr}}
\newrulecommand{ArrElimSplit}{\Elimrulenameextra{\arr}{-Split}}
\newrulecommand{AllIntro}{\Introrulename{\AllSym}}
\newrulecommand{ImpliesIntro}{\Introrulename{\implies}}
\newrulecommand{ImpliesIntroBot}{\Introrulename{\implies}\ensuremath{\bot}}

\newrulecommand{Case}{\rulename{Case}}

\newrulecommand{ArrApp}{\Apprulename{\arr}}
\newrulecommand{AllApp}{\Apprulename{\forall}}
\newrulecommand{SubstApp}{\Apprulename{\rulename{Subst}}}
\newrulecommand{SolveApp}{\Apprulename{\ahat}}
\newrulecommand{ImpliesApp}{\Apprulename{\implies}}

\newrulecommand{PairIntro}{\Introrulename{\times}}
\newrulecommandONE{SumIntro}{\Introrulename{+}\ensuremath{{}_{#1}}}
\newrulecommand{SumIntroL}{\SumIntro{1}}
\newrulecommand{SumIntroR}{\SumIntro{2}}

\newrulecommand{PairIntroSolve}{\IntroSolverulename{\times}}
\newrulecommand{PairIntroSyn}{\IntroSynrulename{\times}}

\newrulecommand{Nil}{\Tyrulename{\rulename{Nil}}}
\newrulecommand{Cons}{\Tyrulename{\rulename{Cons}}}

\newrulecommandONE{SumIntroSolve}{\IntroSolverulename{+}\ensuremath{{}_{#1}}}
\newrulecommand{SumIntroLSolve}{\SumIntroSolve{1}}
\newrulecommand{SumIntroRSolve}{\SumIntroSolve{2}}

\newrulecommandONE{SumIntroSyn}{\IntroSynrulename{+}\ensuremath{{}_{#1}}}
\newrulecommand{SumIntroLSyn}{\SumIntroSyn{1}}
\newrulecommand{SumIntroRSyn}{\SumIntroSyn{2}}

\newrulecommand{WithIntro}{\Introrulename{\with}}
\newrulecommand{ExIntro}{\Introrulename{\exists}}

\newcommand{\Matchrulename}[1]{\Tyrulename{Match\ensuremath{#1}}}

\newcommand{\AssignRuleName}[1]{\ensuremath{\rulename{A#1}}}
\newcommand{\AssignIntroName}[1]{\AssignRuleName{\ensuremath{#1}I}}
\newcommand{\AssignElimName}[1]{\AssignRuleName{\ensuremath{#1}E}}
\newrulecommand{AssignVar}{\AssignRuleName{Var}}
\newrulecommand{AssignUnit}{\AssignRuleName{Unit}}
\newrulecommand{AssignArrIntro}{\AssignIntroName{\arr}}
\newrulecommand{AssignArrElim}{\AssignElimName{\arr}}
\newrulecommand{AssignAllIntro}{\AssignIntroName{\forall}}
\newrulecommand{AssignAllElim}{\AssignElimName{\forall}}
\newcommand{\judge}[3]{{#1} \entails {#2} : {#3}}  %

\newcommand{\LOCALCOPY}[1]{%
          \href{papers/#1}{\bf \textcolor{dGreen}{local copy}}}

\newcommand{\ctxsep}{\ifnum\CompactJudgments=1%
\ensuremath{\,/\,}%
\else%
\ensuremath{\;/\;}%
\fi}

\newcommand{\csu}[2]{\rulename{csu}({#1}, {#2})}  %
\newcommand{\mgu}[2]{\rulename{mgu}({#1}, {#2})}  %

\newcommand{\matchor}{\ensuremath{\normalfont\,\texttt{|}\hspace{-4.85pt}\texttt{|}\,}}
\newcommand{\alt}{\matchor}

\newcommand{\xrec}{\keyword{rec}}
\newcommand{\rec}[1]{\xrec\;#1.\,}

\ifnum\OPTIONSmallFigs=1%
  \newenvironment{myfigure*}[1][]{\begin{figure*}[#1]\bgroup\small}{\egroup\end{figure*}}
\else
  \newenvironment{myfigure*}[1][]{\begin{figure*}[#1]\bgroup}{\egroup\end{figure*}}
\fi

\ifnum\OPTIONSmallFigs=1%
  \newenvironment{myfigure}[1][]{\begin{figure}[#1]\bgroup\small}{\egroup\end{figure}}
\else
  
\fi

\newcommand{\repeatcaption}[3]{%
  \addtocounter{figure}{-1}%
  \renewcommand{\thefigure}{\ref{#1}{#2}}%
  \captionsetup{list=no}%
  \caption{#3}%
}

\newcommand{\polvar}{\mathcal{P}}

\newtheoremstyle{better}%
   {\topsep}%
   {\topsep}%
   {\upshape\slshape}%
   {}%
   {\bfseries}%
   {.}%
   {0.7em}%
   {}%
\theoremstyle{better}
\newtheorem*{lemma*}{Lemma}
\newtheorem*{definition*}{Definition}

%% file: abstract.tex
\begin{abstract}
  Bidirectional typechecking, in which terms either synthesize a type
  or are checked against a known type, has become popular for its
  applicability to a variety of type systems, its error reporting, and
  its ease of implementation.  Following principles from proof theory,
  bidirectional typing can be applied to many type constructs.
  The principles underlying a bidirectional approach to indexed types
  (\emph{generalized algebraic datatypes}) are
  less clear.  Building on proof-theoretic treatments of equality,
  we give a declarative specification of typing based on
  \emph{focalization}.  This approach
  permits declarative rules for coverage of pattern matching,
  as well as support for first-class existential types using a focalized
  subtyping judgment.
  We use refinement types to avoid
  explicitly passing equality proofs in our term syntax, making our
  calculus similar to languages such as
  Haskell and OCaml.
We also extend the declarative specification with
an explicit rules for deducing when a type is principal,
permitting us to give a complete declarative specification for a
rich type system with significant type inference.
  We also give a set of algorithmic typing rules, and prove that it is sound
  and complete with respect to the declarative system.
  The proof requires a number of technical innovations, including
  proving soundness and completeness in a mutually recursive fashion.
\end{abstract}

%% file: introduction.tex
\section{Introduction}
\Label{sec:intro}

Consider a list type $\vecop$ with a numeric index representing
its length, in Agda-like notation:
\begin{mathpar}
 \begin{array}{l}
   \keyword{data} \; \textsf{Vec} : \textsf{Nat} \arr \textsf{Type} \arr \textsf{Type}~\keyword{where} \\
   ~~~~~~
   \begin{array}{l@{~~}l@{~~}l}
     \texttt{[]} &:& A \arr \vectype{0}{A}
     \\
     (::) &:& A \arr \vectype{n}{A} \arr \vectype{(\xsucc\;n)}{A}
   \end{array}
 \end{array}
\end{mathpar}
We can use this definition to write a head function that
always gives us an element of type $A$ when the length is at least
one:
\vspace*{-0.5ex}
\begin{mathpar}
 \begin{array}{l}
   \texttt{head} : \forall n, A.\; \vectype{(\xsucc\;n)}{A} \to A \\
   \texttt{head} \;(x :: xs)  =  x
 \end{array}
\end{mathpar}
This clausal definition omits the clause for $\vecnil$, which has an
index of $0$.  The type annotation tells us that
$\texttt{head}$'s argument has an index of $\succ{n}$ for some
$n$.  Since there is no natural number $n$ such that $0 =
\succ{n}$, the nil case cannot occur and can be omitted.

This is an entirely reasonable explanation for programmers,
but language designers and implementors will have more questions.
First, designers of functional languages are accustomed to the
benefits of the Curry--Howard correspondence, and expect a
\emph{logical} reading of type systems to accompany the operational
reading. So what is the logical reading of GADTs?  Second, how can we
implement such a type system?  Clearly we needed some equality
reasoning to justify leaving off the nil case, which is not trivial in
general. %

Since we relied on equality information to omit the nil case,
it seems reasonable to look to logical accounts of
equality. In proof theory, it is
possible to formulate equality in (at least) two different ways. The
better-known is the \emph{identity type} of Martin-L\"{o}f, but GADTs
actually correspond best to the equality of
\citet{schroeder-heister-equality} and \citet{girard-equality}.
The Girard--Schroeder-Heister (GSH) approach introduces equality via
the reflexivity principle:
\begin{mathpar}
    \inferrule*[]
          { }
          {\Gamma \entails t = t}
\end{mathpar}
The GSH elimination rule was originally formulated in a sequent
calculus style, as follows:
\begin{mathpar}
\inferrule*[]
          {\text{for all }\theta.\; \text{if } \theta \in \csu{s}{t} \text{ then }
           \theta(\Gamma) \entails \theta(C)}
          {\Gamma, (s=t) \entails C}
\end{mathpar}
Here, we write $\csu{s}{t}$ for a \emph{complete
set of unifiers} of $s$ and $t$. So the rule says that we can
eliminate an equality $s = t$ if, for every substitution $\theta$
that makes $s$ and $t$ equal, we can prove
the goal $C$.

This rule has three important features, two good and one bad.

\begin{enumerate}

\item The GSH elimination rule is an invertible left rule. By ``left
  rule'', we mean that the rule decomposes the assumptions to the left
  of the turnstile (in this case, the assumption that $s = t$), and by
  ``invertible'', we mean the conclusion of the rule implies the
  premises.\footnote{The invertibility of equality elimination is
    certainly not obvious; see
    \citet{schroeder-heister-equality} and \citet{girard-equality}. %
}

Invertible left rules are interesting, because they are known to
correspond (via Curry--Howard) to the deconstruction steps carried
out by pattern-matching rules~\citep{Krishnaswami09:pattern-matching}.
This is our first hint that the GSH rule has something to do with GADTs;
programming languages like Haskell and OCaml indeed use pattern
matching to propagate equality information.

\item Observe that if we have an inconsistent equation, we can
  immediately prove the goal. If we specialize the rule above to the
  equality $0 = 1$, we get:
\begin{mathpar}
\inferrule*[]
          { }
          {\Gamma, (0=1) \entails C}
\end{mathpar}
Because $0$ and $1$ have no unifiers, the complete set of unifiers is
the empty set. As a result, the GSH rule for this instance has no
premises, and the elimination rule for an absurd equation ends up
looking exactly like the elimination rule for the empty type:
\begin{mathpar}
\inferrule*[]
          { }
          {\Gamma, \bot \entails C}
\end{mathpar}
Moreover, recall that when we eliminate an empty type, we can view the
eliminator $\mathsf{abort}(e)$ as a pattern match with no clauses.
Together, these two features line up nicely with our definition of
\texttt{head}, where the impossibility of the case for $\vecnil$ was
indicated by the \emph{absence} of a pattern clause.
The use of equality in GADTs corresponds perfectly with the
GSH equality.

\item Alas, we cannot simply give a proof term assignment for first-order
logic and call it a day. The third important feature of the GSH
equality rule is its use of \emph{unification}: it works by treating the free
variables of the two terms as unification variables.  But type
inference algorithms also use unification, introducing unification
variables to stand for unknown types.

These two uses of unification are \emph{entirely different!} Type
inference introduces unification variables to stand for the specific
instantiations of universal quantifiers. In contrast, the
Girard--Schroeder-Heister rule uses unification to constrain the
universal parameters.
As a result, we need to understand how to integrate these two uses of
unification, or at least how to keep them decently separated, in order to
take this logical specification and implement type inference for it.
\end{enumerate}

This problem---formulating indexed types in as logical a style as
feasible, while retaining the ability to implement type inference
algorithms for them---is the subject of this paper.

\mypara{Contributions}
It has long been known that
GADTs are equivalent to the combination of existential types and
equality constraints \citep{Xi03:guarded}.
Our key contribution is to reduce GADTs to standard logical
ingredients, \emph{while retaining the implementability of the type
  system}.  We manage this by formulating a system of indexed
types in a bidirectional style (combining type \emph{synthesis} with
\emph{checking} against a known type), which is both
practically implementable and theoretically tidy. 

\begin{itemize}
\item Our language supports implicit higher-rank polymorphism (in
  which quantifiers can be nested under arrows) including existential
  types. While algorithms for higher-rank universal polymorphism are
  well-known \citep{PeytonJones07,Dunfield13}, our approach to
  supporting existential types is novel.

  We go beyond the standard practice of tying existentials to datatype
  declarations~\citep{Laufer94:existential-type-inference}, in favour
  of a first-class treatment of implicit existential types.  This
  approach has historically been thought difficult, since treating
  existentials in a first-class way opens the door to higher-rank
  polymorphism that mixes universal and existential quantifiers.

  Our approach extends existing bidirectional methods for handling
  higher-rank polymorphism, by adapting the
  proof-theoretic technique of \emph{focusing} to give a novel
  \emph{polarized subtyping judgment}, which lets us treat mixed
  quantifiers in a way that retains decidability while maintaining
  the essential properties of subtyping, such as stability under substitution
  and transitivity.  %

\item  Our language includes equality types in the style of Girard
  and Schroeder-Heister, but without an
  explicit introduction form  %
  for equality. 
  Instead, we treat equalities as property types,
  in the style of intersection or refinement types: we
  do not write explicit equality proofs in our syntax, permitting
  us to more closely model how equalities are used in
  OCaml and Haskell. 
  
\item The use of focusing also lets us equip our calculus with nested
  pattern matching. This fits in neatly with our bidirectional
  framework, and permits us to give a formal specification of coverage
  checking with GADTs, which is easy to understand, easy to implement,
  and theoretically well-motivated. 

\item In contrast to systems which globally possess or lack principal
  types, our declarative system tracks whether or not a derivation has
  a principal type.

  Our system includes an unusual ``higher-order principality'' rule,
  which says that if only a single type can be synthesized for a term,
  then that type is principal.  While this style of hypothetical
  reasoning is natural to explain to programmers, formalizing it
  requires giving an inference rule with universal quantification over
  possible typing derivations in the premise. This is an extremely
  non-algorithmic rule (even its well-foundedness is
  not immediate).

  As a result, the soundness and completeness proofs for our
  implementation have to be done in a new style. It is no longer
  possible to prove soundness and completeness independently, and
  instead we must prove them mutually recursively.
  
\item We formulate an algorithmic type system
  (\Sectionref{sec:alg-typing}) for our declarative calculus, and
  prove that typechecking is decidable, deterministic
  (\ref{sec:alg-determinacy}), and sound and complete (Sections
  \ref{sec:soundness}--\ref{sec:completeness}) with respect to the
  declarative system.

  The resulting type system is relatively easy to implement (an
  undergraduate implemented most of it on his own from a draft of the
  paper, with minimal contact with the authors), and is close in style
  to languages such as Haskell or OCaml. As a result, it seems like a
  reasonable basis for implementing new languages with expressive type
  systems.
\end{itemize}

Our algorithmic system (and, to a lesser extent, our declarative system)
uses some techniques developed by \citet{Dunfield13}, but we extend these
to a far richer type language (existentials, indexed types, 
sums, products, equations over type variables), and we differ by supporting
pattern matching, polarized subtyping, and principality tracking.

\mypara{Supplementary material}
The appendix
contains rules omitted for space reasons, and full proofs.

%% file: overview.tex
\section{Overview}
\Label{sec:overview}

To orient the reader, we give an overview and rationale of the
novelties in our type system, before getting into the details of the
typing rules and algorithm. 
As is well-known~\citep{Cheney03:FirstClassPhantom,Xi03:guarded},
GADTs can be desugared into type expressions that use equality and
existential types to express the return type constraints.  These two
features lead to difficulties in typechecking for GADTs.

\mypara{Universal, existentials, and type inference} Practical
typed functional languages must support some degree of type inference,
most critically the inference of type arguments. That is, if we have a
function $f$ of type $\alltype{a}{a \to a}$, and we want to apply it
to the argument $3$, then we want to write $f\;3$, and not
$f\;[\mathsf{Nat}]\;3$ (as we would in pure System F). Even
with a single type argument, the latter style is noisy, and
programs using even moderate amounts of polymorphism rapidly
become unreadable.
However, omitting type arguments has significant metatheoretical
implications.  In particular, it forces us to include subtyping in our
typing rules, so that (for instance) the polymorphic type
$\alltype{a} a \arr a$ is a subtype of its instantiations (like
$\mathsf{Nat} \arr \mathsf{Nat}$).

The subtype relation induced by System F-style polymorphism and
function contravariance is already undecidable \citep{Tiuryn96,Chrzaszcz98},
so even at the first step we must introduce restrictions on type inference to
ensure decidability. In our case, matters are further complicated by
the fact that we need to support \emph{existential types} in addition
to universal types.

Existentials are required to encode GADTs \citep{Xi99popl}, but
programming languages have traditionally stringently restricted the
use of existential types. Following the approach of
\citet{Laufer94:existential-type-inference}, languages such as \Ocaml
and Haskell tie existential introduction and elimination to datatype
declarations, so that there is always a syntactic marker for when to
introduce or eliminate existential types. This choice permits leaving
existentials out of subtyping altogether, at the price of no longer
permitting implicit subtyping (such as using $\fun{x}{x+1}$ at type
$\extype{a} a \arr a$).

While this is a practical solution, it increases the distance between
a surface language and its type-theoretic core. Our goal is to give a
\emph{direct} type-theoretic account of as many features of our
surface languages as possible. In addition to the theoretical
tidiness, this also has practical language design benefits. By
avoiding a complex elaboration step, we are forced to specify the type
inference algorithm in terms of a language close to a
programmer-visible surface language. This does increase the complexity
of the approach, but in a productive way: we are forced to analyze and
understand how type inference will look to the end programmer.

The key problem is that when both universal and existential
quantifiers are permitted, the order in which to instantiate
quantifiers when computing subtype entailments becomes unclear. For
example, suppose we need to decide
$\declsubjudg[]{\Gamma}{\alltype{a_1}{\extype{a_2}{A(a_1,
      a_2)}}}{\extype{b_1}{\alltype{b_2}{B(b_1,b_2)}}}$.  An algorithm
to solve this must either first introduce a unification variable for
$a_1$ and a parameter for $a_2$ first, and only then introduce a
unification variable for $b_1$ and a parameter for $b_2$, or the other
way around---and the order in which we make these choices matters!
With the first order, the instantiation for $b_1$ may refer to $a_2$,
but the instantiation for $a_1$ cannot have $b_2$ as a free variable.
With the second order, the instantiation for $a_1$ may have $b_2$ as a
free variable, but $b_1$ may not refer to $a_2$. 

In some cases, depending on what $A(a_1, a_2)$ and $B(b_1, b_2)$ are,
only one choice of order works. For example, if we are trying to decide
$\declsubjudg[]{\Gamma}{\alltype{a_1}{\extype{a_2}{a_1 \to
      a_2}}}{\extype{b_1}{\alltype{b_2}{b_2 \to b_1}}}$,
we must choose the first order:
we must pick an instantiation for $a_1$,
and then make $a_2$ into a parameter before we can instantiate $b_1$ as $a_2$.
The second order will not work, because $b_1$ must depend on $a_2$.
Conversely, if we are trying to solve
$\declsubjudg[]{\Gamma}{\alltype{a_1}{\extype{a_2}{a_1 \to
      a_2}}}{\extype{b_1}{\alltype{b_2}{\extype{b_3}{b_1 \times b_2 \to b_3}}}}$, the
first order will not work; we must instantiate $b_1$ (say, to $\mathsf{int}$) and quantify
over $b_2$ before instantiating $a_1$ as $\mathsf{int} \times b_2$.

As a result, the outermost connectives do not reliably determine
which side of a subtype judgement
$\declsubjudg[]{\Gamma}{\alltype{a}{A}}{\extype{b}{B}}$ to specialize
first. %

One implementation strategy is simply to give up determinism:
an algorithm could backtrack when faced with deciding subtype
entailments of this form. Unfortunately, backtracking is dangerous
for a practical implementation, since it potentially causes type-checking
to take exponential time.  %
This tends to defeat the benefit of a complete
declarative specification, since different implementations
with different backtracking strategies could have
radically different running times when checking the same program.
So we may end up with an implementation that is theoretically
complete, but incomplete in practice.

Instead, we turn to ideas from proof theory---specifically, polarized type theory.
In the language of polarization, universals are a \emph{negative} type,
and existentials
are a \emph{positive} type. So we introduce two
mutually recursive subtype relations:
$\declsubjudg[+]{\Gamma}{A}{B}$
for positive types
and
$\declsubjudg[-]{\Gamma}{A}{B}$
for negative types.
The positive subtype relation only deconstructs existentials, and the negative
subtype relation only deconstructs universals. This fixes the order in which
quantifiers are instantiated, making the problem decidable (in fact, rather
straightforward).

The price we pay is that fewer subtype entailments are derivable.
Fortunately, any program typeable under a more liberal subtyping
judgement can be made typable in our discipline by $\eta$-expanding
it.  Moreover, the lost subtype entailments seem to be rare in
practice: most of the types we see in practice are of the form
$\alltype{\vec{a}}{\vec{A} \to \extype{\vec{b}}{B}}$, and this fits
perfectly with the kinds of types our polarized subtyping
judgement works best on. As a result,  we
keep fundamental expressivity, and also efficient decidability.

\paragraph{Equality as a property.} The usual convention in Haskell
and OCaml is to make equality proofs in GADT definitions implicit. We
would like to model this feature directly, so that our calculus stays
close to surface languages, without sacrificing the logical reading of
the system.
In this case, the appropriate logical concepts come from the theory of
intersection types.  A typing judgment such as $e : A \times
B$ can be viewed as giving instructions on how to construct a value
(pair an $A$ with a $B$).  But types can also be viewed as
\emph{properties}, where $e : X$ is read ``$e$ has property $X$''.

To model GADTs, we need both of these readings! For example, a term of
vector type is constructed from nil and cons constructors, but also
has a property governing its index. To support this combination, we
introduce a type constructor $A \with P$, read ``$A$ with $P$'', to
model elements of type $A$ satisfying the property (equation) $P$. (We
also introduce $P \implies A$, read ``$P$ implies $A$'', for its
adjoint dual, consisting of terms which have the type $A$
conditionally under the assumption that $P$ holds.) Then we make
equality $t = t'$ into a property, and make use of standard rules
for property types (which omit explicit proof terms) to type equality
constraints~\citep[Section 2.4]{DunfieldThesis}. %

This gives us a logical account of how OCaml and Haskell avoid
requiring explicit equality proofs in their syntax. The benefit of
handling equality constraints through intersection types is that
certain restrictions on typing that are useful for decidability, such
as restricting property introduction to values, arise naturally from
the semantic point of view---via the value restriction needed for
soundly modeling intersection and union
types~\citep{Davies00icfpIntersectionEffects,Dunfield03:IntersectionsUnionsCBV}.
In addition, the appropriate approach to take when combining GADTs and
effects is clear.\footnote{The traditional $\mathtt{eq}$ GADT and its
  constructor $\mathtt{refl}$ can be encoded into our system as the
  type $1 \land (s=t)$, which which can be constructed as a unit value
  only under the constraint that $s$ equals $t$.}

\paragraph{Bidirectionality, pattern matching, and principality.}
Something that is not by itself novel in our approach is our
decision to formulate both the declarative and algorithmic systems in
a bidirectional style. Bidirectional checking \citep{Pierce00} is a popular
implementation choice for systems ranging from dependent types
\citep{Coquand96:typechecking-dependent-types,Abel08:MPC}
and contextual types \citep{Pientka08:POPL} to object-oriented
languages~\citep{Odersky01:ColoredLocal},
but also has good proof-theoretic foundations~\citep{Watkins04},
making it useful both for specifying and implementing type
systems. Bidirectional approaches make it clear to programmers where
annotations are needed (which is good for specification), and
can also remove unneeded nondeterminism from typing (which is
good for both implementation and proving its correctness).

However, it is worth highlighting that because both
bidirectionality and pattern matching arise from focalization,
these two features fit together extremely well. In fact, by
following the blueprint of focalization-based pattern matching, we
can give a coverage-checking algorithm
that explains when it is permissible to omit clauses in GADT pattern
matching.

In the propositional case, the type synthesis judgment of a
bidirectional type system generates principal types: if a type can be
inferred for a term, that type is the most specific possible type for
that term. (Indeed, in many cases, including the current system, the
inferred type will even be unique.) This property is lost once
quantifiers are introduced into the system, which is why it is not
much remarked upon.  However, prior work on GADTs, starting with
\citet{Simonet07:constraint}, has emphasized the importance of the
fact that handling equality constraints is much easier when the type
of a scrutinee is principal. Essentially, this ensures that no
existential variables can appear in equations, which prevents equation
solving from interfering with unification-based type inference.  The
OutsideIn algorithm takes this consequence as a definition, permitting
non-principal types just so long as they do not change the values of
equations. However, \citet{Vytiniotis11} note that while their system
is sound, they no longer have a completeness result for their type
system.

We use this insight to extend our bidirectional typechecking algorithm
to track principality:  The judgments we give track whether
types are principal, and we use this to give a
relatively simple specification for whether or not type annotations
are needed.  We are able to give a very natural spec to programmers---cases
on GADTs must scrutinize terms with
principal types, and an inferred type is principal just when it is the
only type that can be inferred for that term---which soundly and
completely corresponds to the implementation-side constraints: a type
is principal when it contains no existential unification variables.

%% file: examples.tex
\section{Examples}
\Label{sec:examples}

In this section, we give some examples of terms from our language,
which illustrate the key features of our system and give a sense
of how many type annotations are needed in practice. To help make this
point clearly, all of the examples which follow are unsugared: they
are the \emph{actual} terms from our core calculus.

\paragraph{Mapping over lists.}
First, we begin with the traditional \emph{map} function, which takes
a function and applies it to every element of a list.
\begin{center}
\parbox{0pt}{\begin{tabbing}
 $\rec{\textit{map}} \lam{f}\lam{xs}$\=
  $\caseop\big(xs,$ \= $\;\;\, \branch{ \vecnil }{ \vecnil }$ \+ \+\\
$\alt \branch{ \veccons{y}{ys} }{ \veccons{(f\;y)}{\textit{map}\;f\;ys }}\big)$ \-\-\\
$ : \alltype{n:\ind} 
    \alltype{\alpha:\type}
    \alltype{\beta:\type}
      (\alpha \to \beta) \to \vectype{n}{\alpha} \to \vectype{n}{\beta}$    
\end{tabbing}}
\end{center}
This
function
case-analyzes its
second argument $xs$. Given an empty $xs$, it returns the empty list;
given a cons cell $\veccons{y}{ys}$,
it applies the argument function $f$ to the head $y$
and makes a recursive call on the tail $ys$.

In addition, we annotate the definition with a type. We have two type
parameters $\alpha$ and $\beta$ for the input and output element types.
Since we are working with length-indexed lists, we also have
a length index parameter $n$, which lets us show by typing
that the input and output of $\textit{map}$ have the same length.

In our system, this type annotation is mandatory.  Full type inference
for definitions using GADTs requires polymorphic recursion, which is
undecidable. As a result, this example also requires annotation in
OCaml and GHC Haskell. However, Haskell and OCaml infer polymorphic
types when no polymorphic recursion is needed.
We adopt the simpler rule that \emph{all} polymorphic definitions are
annotated. This choice is motivated by \citet{Vytiniotis10}, who
analyzed a large corpus of Haskell code and showed that implicit
let-generalization was used primarily only for top-level definitions,
and even then it is typically considered good practice to annotate
top-level definitions for documentation purposes. Furthermore,
experience with languages such as Agda and Idris (which do not
implicitly generalize) show this is a modest burden in practice.

\paragraph{Nested patterns and GADTs.}
Now, we consider the \textit{zip}
function, which converts a pair of lists into a list of
pairs. In ordinary ML or Haskell, we must consider what to do when the
two lists are not the same length. However, with length-indexed lists, 
we can statically reject passing two lists of differing length:
\begin{center}
\parbox{0pt}{\begin{tabbing}
 $\rec{\textit{zip}} \lam{p}
\caseop\big(p,$\= $\;\;\; \branch{ (\vecnil,\vecnil) }{ \vecnil }$  \+\\
$\alt \branch{ (\veccons{x}{xs}, \veccons{y}{ys}) }{ \veccons{(x,y)}{\textit{zip}\;(xs, ys) }}\big)$ \-\\
$ : \alltype{n:\ind} 
    \alltype{\alpha:\type}
    \alltype{\beta:\type}
      (\vectype{n}{\alpha} \times \vectype{n}{\beta}) \to \vectype{n}{(\alpha \times \beta)}$    
\end{tabbing}}
\end{center}
This case expression has only two patterns, one for when both
lists are empty and one for when both lists have elements, with the
type annotation indicating that both lists must be of length $n$. Typing
shows that the cases where one list is empty and the other is non-empty are impossible,
so our coverage checking rules accept this as
a complete set of patterns.  This example also illustrates
that we support nested pattern matching.

\paragraph{Existential Types}
Now, we consider the \textit{filter} function, which takes a predicate
and a list, and returns a list containing the elements satisfying that
predicate. This example makes a nice showcase for supporting existential
types, since the size of the return value is not predictable statically.
\begin{center}
\parbox{0pt}{\begin{tabbing}
 $\rec{\textit{filter}} \lam{p} \lam{xs}$
  $\caseop\big(xs,$\= $\;\;\; \branch{ \vecnil }{ \vecnil }$  \+\\
 $\alt \branch{ \veccons{x}{xs} }{}\mathsf{let}$\=$\; tl = \textit{filter}\;p\;xs \;\mathsf{in}$ \+ \\
                                       $\caseop\big(p\;x,$\= \+
                                                    $\;\;\;\branch{\inj{1}{\wild}}{ tl }$ \\
                                                    $\alt \branch{\inj{2}{\wild}}{ \veccons{x}{tl} }\big)\big)$\-\-\- \\
  
$ : \alltype{n:\ind} 
    \alltype{\alpha:\type}
      (\alpha \to 1+1) \to \vectype{n}{\alpha} \to \extype{k:\ind} \vectype{k}{\alpha} $ 
\end{tabbing}}
\end{center}
So, this function takes predicate and a vector of arbitrary size, and
then returns a list of unknown size (represented by the existential
type $\extype{k:\ind}{\vectype{k}{\alpha}}$). Note that we did not
need to package the existential in another datatype, as one
would have to in OCaml or GHC Haskell---we are free to use
existential types as ``just another type constructor''.

%% file: declarative.tex
\input{fig-source-syntax.tex}

\input{fig-decl-syntax.tex}

\input{confdeclgrue.tex}

\input{fig-decl-subtyping.tex}

\section{Declarative Typing}
\Label{sec:decl-typing}

\mypara{Expressions}
Expressions (\Figureref{fig:source-syntax})
are variables $x$;
the unit value $\unitexp$;
functions $\lam{x} e$;
applications to a spine $e\,s^+$;
fixed points $\rec{x} v$;
annotations $(e : A)$;
pairs $\pair{e_1}{e_2}$;
injections into a sum type $\inj{k}{e}$;
case expressions $\case{e}{\Pi}$ where $\Pi$ is a list of branches
$\pi$, which can eliminate pairs and injections (see below);
the empty vector $\vecnil$; and consing a head $e_1$ to a tail
vector $e_2$.

Values $v$ are standard for a call-by-value semantics; the variables
introduced by fixed points are considered values, because we only
allow fixed points of values.
A spine $s$ is a list of expressions---arguments to a function.
Allowing empty spines (written $\emptyspine$) is convenient in the typing rules, but
would be strange in the source syntax, so (in the grammar of expressions $e$)
we require a nonempty spine $s^+$.  We usually omit the empty spine $\emptyspine$,
writing $e_1\,e_2$ instead of $e_1\,e_2\,\cdot$.  Since we use juxtaposition for both
application $e\;s^+$ and spines, some strings are ambiguous;
we resolve this ambiguity in favour of the spine, so $e_1\,e_2\,e_3$ is
parsed as the application of $e_1$ to the spine $e_2\,e_3$, which is technically
$\appspine{e_2}{(\appspine{e_3}{\emptyspine})}$.
Patterns $\pat$ consist of pattern variables, pairs, and injections.
A branch $\pi$ is a \emph{sequence} of patterns $\patvec$ with a
branch body $e$.  We represent patterns as sequences, which
enables us to deconstruct tuple patterns. %

\mypara{Types}
We write types as $A$, $B$ and $C$.  We have the unit type
$\unitty$, functions $A \arr B$, sums $A + B$, and products
$A \times B$.
We have universal and existential types
$\alltype{\alpha:\sort} A$ and
$\extype{\alpha:\sort} A$; these are predicative quantifiers
over monotypes (see below).  We write $\alpha$, $\beta$, etc.\ for
type variables; these are
universal, except when bound within an existential type.
We also have a \emph{guarded type} $P \implies A$, read
``$P$ implies $A$''.  This implication corresponds
to type $A$, provided $P$ holds.  Its dual is the \emph{asserting type}
$A \with P$, read ``$A$ with $P$'', which witnesses the proposition $P$.
In both, $P$ has no runtime content.

\mypara{Sorts, terms, monotypes, and propositions}
Terms and monotypes $t$, $\tau$, $\sigma$ share a grammar
but are distinguished by their \emph{sorts} $\sort$.
Natural numbers $\zero$ and $\succ{t}$ are \emph{terms} and have sort $\ind$.
Unit $\unitty$ has the sort $\type$ of \emph{monotypes}.
A variable $\alpha$ stands for a term or a monotype, depending on the sort $\sort$
annotating its binder.  Functions, sums, and products of monotypes are monotypes
and have sort $\type$.
We tend to prefer $t$ for terms and $\sigma$, $\tau$ for monotypes.

A proposition $P$ or $Q$ is simply an equation $t = t'$.
Note that terms, which represent runtime-irrelevant information, are
distinct from expressions; however, an expression may include type
annotations of the form $P \implies A$ and $A \with P$, where $P$
contains terms.

\mypara{Contexts}
A declarative context $\Psi$ is an \emph{ordered} sequence of universal variable declarations
$\alpha : \sort$ and expression variable typings $\hyp p x A$, where
$p$ denotes whether the type $A$ is principal (\Sectionref{sec:decl-typing-judgments}).
A variable $\alpha$ can be free in a type $A$ only if $\alpha$ was declared to
the left: $\alpha\,{:}\,\type,\, \hyp p x \alpha$ is well-formed,
but  $\hyp p x \alpha,\, \alpha\,{:}\,\type$ is not.

\subsection{Subtyping}

We give our two subtyping relations,
$\declsubtype^{+}$ and $\declsubtype^{-}$,
in Figure~\ref{fig:decl-subtyping}.
We treat the universal quantifier as a negative type (since it is a
function in System F), and the existential as a positive type (since
it is a pair in System F).  %
We have two typing rules 
for each of these connectives, corresponding to the left and right
rules for universals and existentials in the sequent calculus.
We treat all other types as having no polarity. The positive
and negative subtype judgments are mutually recursive, and the
\DsubPosNeg rule permits switching the polarity of subtyping from 
positive to negative  when both of the types are non-positive,
and conversely for \DsubNegPos. When both types are neither positive
nor negative, we require them to be equal (\DsubReflPm). 

In logical terms, functions and guarded types are negative;
sums, products and assertion types are positive.  We could
potentially operate on these types in the negative and positive subtype
relations, respectively.  Leaving out (for example) function
subtyping means that we will have to do some $\eta$-expansions
to get programs to typecheck; we omit these rules to keep the implementation
complexity low.  (The idea that $\eta$-expansion can substitute for subsumption
dates to \citet{Barendregt83}.)

This also illustrates a nice feature of bidirectional typing:
we are relatively free to adjust the subtype relation to taste. Moreover,
the structure of polarization makes it easy to work out just what the
rules should be. E.g., to add function subtyping
to our system, we would use the rule:
\begin{displaymath}
  \inferrule*[]
             { \declsubjudg[+]{\Psi}{A'}{A} \\
               \declsubjudg[-]{\Psi}{B}{B'} }
             { \declsubjudg[-]{\Psi}{A \to B}{A' \to B'} }
\end{displaymath}
As polarized function types are a negative type of the form
$X^+ \to Y^-$, we see (1) the rule as a whole lives in the
negative subtyping judgement, (2) argument types compare in the
positive judgement (with the usual contravariant twist), and (3) result
types compare in the negative judgement.

\input{fig-chkintro.tex}

\input{fig-decl-typing.tex}

\subsection{Typing Judgments}

\Label{sec:decl-typing-judgments}

\mypara{Principality}
Our typing judgments carry \emph{principalities}:
$A\,\p$ means that $A$ is principal, and $A\,\OK$ means $A$ is
not principal.
Note that a principality is part of a judgment, not part of a type.
In the checking judgment $\declchkjudg{p}{\Psi}{e}{A}$
the type $A$ is input; if $p = \p$, we know that $A$ is not the result
of guessing.  For example, the $e$ in $(e : A)$ is checked against $A\;\p$.
In the synthesis judgment $\declsynjudg{p}{\Psi}{e}{A}$, the type $A$ is output,
and $p = \p$ means
it is impossible to synthesize any other type, as in $\declsynjudg{\p}{\Psi}{(e : A)}{A}$.

We sometimes omit a principality when it is $\OK$ (``not principal'').
We write $p \moreprincipal q$, read ``$p$ at least as principal as $q$'', for the
reflexive closure of $\p \moreprincipal \OK$.

\mypara{Spine judgments}
The ordinary form of spine judgment, $\declspinejudg{\Psi}{s}{A}{p}{C}{q}$,
says that if arguments $s$ are passed to a function of type $A$,
the function returns type $C$.  For a function $e$ applied to one argument
$e_1$, we write $e\;e_1$ as syntactic sugar for
$e\;(\appspine{e_1}{\emptyspine})$.  Supposing $e$ synthesizes
$A_1 \arr A_2$, we apply \DeclArrSpine, checking $e_1$ against $A_1$
and using \DeclEmptySpine to derive
$\declspinejudg{\Psi}{\emptyspine}{A_2}{p}{A_2}{p}$.

Rule \DeclAllSpine does not decompose $\appspine{e}{s}$ but
instantiates a $\AllSym$.  Note that, even if the given type
$\alltype{\alpha:\sort}{A}$ is principal ($p = \p$), the type $[\tau/\alpha]A$
in the premise is not principal---we could choose a different $\tau$.
In fact, the $q$ in \DeclAllSpine is also always $\OK$, because no rule
deriving the ordinary spine judgment can recover principality.

The \emph{recovery spine judgment}
$\declrecspinejudg{\Psi}{s}{A}{p}{C}{q}$, however, can restore
principality in situations where the choice of $\tau$ in \DeclAllSpine
cannot affect the result type $C$.
If $A$ is principal ($p = \p$) but
the ordinary spine judgment produces a non-principal $C$,
we can try to recover principality with \DeclRecover.
Its first premise is $\declspinejudg{\Psi}{s}{A}{\p}{C}{\OK}$;
its second premise (really, an infinite set of premises) quantifies
over all derivations of $\declspinejudg{\Psi}{s}{A}{\p}{C'}{\OK}$.
If $C' = C$ in all such derivations, then the ordinary spine rules
erred on the side of caution: $C$ is actually principal, so we can
set $q = \p$ in the conclusion of \DeclRecover.

If some $C' \neq C$, then $C$ is certainly not principal, and we must
apply \DeclPass, which simply transitions from the ordinary judgment
to the recovery judgment.

\Figureref{fig:declgrue} shows the dependencies between the
declarative judgments.  Given the cycle containing the spine typing judgments,
we need to stop and ask:  Is \DeclRecover well-founded?
For well-foundedness of type systems, we can often make a straightforward
argument that, as we move from the conclusion of a rule to its premises,
either the expression gets smaller, or the expression stays the same but the
type gets smaller.  In \DeclRecover, neither the expression nor the type get smaller.
Fortunately, the rule that gives rise to the arrow from ``spine typing''
to ``type checking'' in \Figureref{fig:declgrue}---\DeclArrSpine---\emph{does}
decompose its subject, and any derivations of a recovery judgment lurking
within the second premise of \DeclRecover must be for a smaller spine.
In the appendix (Lemma \ref{fancy:lem:declarative-well-founded},
p.\ \pageref{fancy:PROOFlem:declarative-well-founded}),
we prove that the recovery judgment, and all the
other declarative judgments, are well-founded.

\mypara{Example}  In \Sectionref{sec:typing-examples} we present some
example derivations that illustrate how the spine typing rules work to recover principality.

\mypara{Subtyping} Rule \DeclSub invokes the subtyping judgment, at
the join of the polarities of $B$ (the type being checked against) and
$A$ (the type being synthesized). Using the join ensures that the
polarity of $B$ takes precedence over $A$'s, which means the programmer
control which subtyping mode to begin with via a type annotation.

Furthermore, the subtyping rule allows \DeclSub to play the role of an
existential introduction rule, by applying subtyping rule \DsubExistsR
when $B$ is an existential type.

\mypara{Pattern matching} Rule \DeclCase checks that the scrutinee
has a type and principality, and then invokes the two main judgments for
pattern matching. The $\declmatchjudg[q]{p}{\Psi}{\Pi}{\vec{A}}{C}$ judgement
checks that each branch in the list of branches $\Pi$ is well-typed, taking a vector
$\vec{A}$ of pattern types to simplify the specification of coverage
checking, as a well as a principality annotation covering all of the types
(i.e., if any of the types in $\vec{A}$ is non-principal, the whole vector
is not principal).

The $\covers{q}{\Psi}{\Pi}{\vec{A}}$ judgement does coverage checking
for the list of branches. However, the $\DeclCase$ does not simply
check that the patterns cover for the inferred type of the scrutinee---%
it checks that they cover for \emph{every} possible type that could
be inferred for the scrutinee. In the case that the scrutinee is
principal, this is the same as checking coverage at the scrutinee's
type, but when the scrutinee is not principal, this rule has the
effect of preventing type inference from using the shape of the
patterns to infer a type, which is notoriously problematic with GADTs
(\eg, whether a missing nil in a list match should be taken as evidence
of coverage failure or that the length is non-zero). As with spine recovery,
this rule is only well-founded because the universal quantification ranges
over synthesized types over a subterm.

\input{fig-decl-pattern-matching.tex}

The $\declmatchjudg[q]{p}{\Psi}{\Pi}{\vec{A}}{C}$ judgment
(rules in
Figure~\ref{%
fig:decl-pattern-matching})
systematically checks the
typing of each branch in $\Pi$: rule \DeclMatchEmpty
succeeds on the empty list, and \DeclMatchSeq checks one
branch and recurs on the remaining branches.
Rules for sums, units, and products break down patterns
left to right, one constructor at a time.
Products also extend the
sequences of patterns and types, with \DeclMatchPair breaking down a
pattern vector headed by a pair pattern $\pair{p}{p'}, \vec{p}$ into
$p, p', \vec{p}$, and breaking down the type sequence
$(A \times B), \vec{C}$ into $A, B, \vec{C}$.
Once all the patterns are eliminated,
the \DeclMatchBase rule says that if the body typechecks, then
the branch typechecks. For completeness, the variable and wildcard
rules are restricted so that any top-level existentials and
equations are eliminated before discarding the type.

The existential elimination rule \DeclMatchExists unpacks an
existential type, and \DeclMatchWith breaks apart a conjunction by
eliminating the equality using unification. The \DeclMatchBot rule
says that if the equation is false then typing succeeds,
because this case is impossible.
The \DeclMatchUnify rule unifies the
two terms of an equation and applies the substitution before
continuing to check typing. Together, these two rules implement the
Schroeder-Heister equality elimination rule.  Because our language of
terms has only simple first-order terms, either unification will fail,
or there is a most general unifier. Note, however, that \DeclMatchWith
only applies when the pattern type is principal.  Otherwise, we use
the \DeclMatchWithOK rule, which throws away the equation and does not
refine any types at all. In this way, we can ensure that we will only
try to eliminate equations which are fully known (i.e., principal).
Similar considerations apply to vectors, with length information being
used to refine types only when the type of the scrutinee is principal.

\input{fig-decl-match-coverage.tex}

The $\declcovers[p]{\Psi}{\Pi}{\Avec}$ judgment
(in
Figure~\ref{%
fig:decl-match-coverage})
checks whether a set of patterns
covers all possible cases. As with match typing, we
systematically deconstruct the sequence of types in the
branch, but we also need auxiliary operations to
\emph{expand} the patterns. For example, the $\expandpair{\Pi}{\Pi'}$
operation takes every branch $\branch{\pair{p}{p'}, \pvec}{e}$ and
expands it to $\branch{{p}, {p'}, \pvec}{e}$. To keep the sequence of
patterns aligned with the sequence of types, we also expand
variables and wildcard patterns into two wildcards:
$\branch{x, \pvec}{e}$ becomes $\branch{\wild, \wild, \pvec}{e}$.
After expanding out all the pairs, \DeclCoversTimes checks coverage
by breaking down the pair type.  

For sum types, we expand a list of branches into \emph{two} lists,
one for each injection. So
$\expandsum{\Pi}{\Pi_L}{\Pi_R}$ will send all branches headed by $\inl{p}$
into $\Pi_L$ and all branches headed by $\inr{p}$ into $\Pi_R$, with
variables and wildcards being sent to both sides. Then \DeclCoversSum
checks the left and right branches independently.

As with typing, \DeclCoversEx just unpacks the existential type.
Likewise, \DeclCoversEqBot and \DeclCoversEq handle the two cases
arising from equations.  If an equation is unsatisfiable,
coverage succeeds since there are no possible values of that type. 
If it is satisfiable, we apply the substitution and continue
coverage checking. Just as when typechecking patterns, we only
use property types to refine coverage checking when the equations
come from a principal type --- the \DeclCoversWithOK rule simply
throws away the equation when the type is not principal. (This is a
sound approximation which ends up requiring more patterns when
the type is not principal.)

So far, the coverage rules for pattern matching are almost purely
type-directed. However, once recursive types like $\vectype{n}{A}$
enter the picture, matters become a little more subtle. The issue is
that if we split a wildcard $\wild$ of type $\vectype{n}{A}$, the type
\emph{doesn't tell us when to stop}. That is, we could split a
wildcard into a nil $\vecnil$ and cons $\veccons{\wild}{\wild}$
pattern; or we could turn it into a nil $\vecnil$, singleton
$\veccons{\wild}{\vecnil}$ and two-or-longer
$\veccons{\wild}{\veccons{\wild}{\wild}}$ pattern; and so on.  The key
issue is that the tail of a list has the same set of possible patterns
as the list itself, and so blindly following the type structure will
not ensure termination of coverage checking.

In this paper, we take the view that the patterns the programmer wrote
should guide how much to split types when doing coverage checking for
inductive types. In Figure~\ref{fig:decl-match-coverage}, we introduce
the $\guarded{\Pi}$ judgement, which checks to see if a constructor
pattern is present in the leading column of patterns. If it is, then
our algorithm will unfold the recursive type as part of type checking,
and otherwise it will not. This is by no means a canonical choice: our
choice is similar to the choice Agda makes, but other language
implementations make other choices. In contrast, the \Ocaml coverage
checking algorithm unfolds wildcard patterns at GADT type one step
more than what the programmer wrote~\cite{le-normand}. (They also
observe that precise exhaustiveness checking is undecidable, meaning
that some choice of heuristic is unavoidable.)

\subsection{Design Considerations for Pattern Matching}

\paragraph{Evaluation Order} Our typing and coverage checking rules
are given assuming a \emph{call-by-value} evaluation strategy. These
coverage rules are not sound under a call-by-name evaluation
order. Consider the following program, writing $\bot$ for a looping term:
$$\case{\bot : A \with (s = t)}{\branch{x}{e}}$$

When type-checking this program, the \DeclMatchWith and
\DeclCoversWith rules are permitted to eliminate the equality $s = t$
when checking $e$. However, one can use a looping program to inhabit
$A \with (s=t)$ for any $P$, and so we have introduced a spurious
equality into the context when checking $e$. In contrast, in a
call-by-value language the scrutinee of a case will be reduced before
the match proceeds, so this issue cannot arise. (In a total language
such as Koka, these rules would be sound irrespective of evaluation
order, since all evaluation strategies are indistinguishable.)

\paragraph{Redundant Patterns} These rules do not check for
redundancy: \DeclCoversEmpty applies even when branches are left over.
When \DeclCoversEmpty is applied, we could mark the
$\branch{\cdot}{e_1}$ branch, and issue a warning for unmarked
branches.  This seems better as a warning than an error, since
redundancy is not stable under substitution.  For example, a case over
$(\tyname{Vec}\;n\;A)$ with $\vecnil$ and $\vecconssym$ branches is
not redundant---but if we substitute $0$ for $n$, the $\vecconssym$
branch becomes redundant.

\mypara{Synthesis}
Bidirectional typing is a form of partial type inference, which
\citet{Pierce00} said should ``eliminate especially those type
annotations that are both \emph{common} and \emph{silly}''.
But our rules are rather parsimonious in what they synthesize;
for instance, $\unitexp$ does not synthesize $\unitty$, and so might
need an annotation.
Fortunately, it would be straightforward to add such rules, following
the style of \citet{Dunfield13}.

%% file: fig-source-syntax.tex
\begin{figure}[t]

    \begin{bnfarray}
        \mbox{Expressions} & e & \bnfas &
                      x
              \bnfalt \unitexp
              \bnfalt \lam{x} e
              \bnfalt e\;s^+ %
              \bnfalt \rec{x} v
              \bnfalt (e : A)
    \\ &&&\!\!\!
              \bnfalt \pair{e_1}{e_2}
              \bnfalt \inj{1}{e}
              \bnfalt \inj{2}{e}
              \bnfalt \case{e}{\Pi}
     \\ &&&\!\!\!
              \bnfalt \vecnil
              \bnfalt \veccons{e_1}{e_2}
    \\[2pt]
        \mbox{Values} & v & \bnfas &
                      x
              \bnfalt \unitexp
              \bnfalt \lam{x} e
              \bnfalt \rec{x} v
              \bnfalt (v : A)
    \\ &&&\!\!\!
              \bnfalt \pair{v_1}{v_2}
              \bnfalt \inj{1} v
              \bnfalt \inj{2} v
              \bnfalt \vecnil
              \bnfalt \veccons{v_1}{v_2}
    \\[2pt]
        \mbox{Spines} & s & \bnfas &
                     \emptyspine
             \bnfalt
                     \appspine{e}{s}
    \\[1pt]
        \mbox{Nonempty spines}\!\!\! & s^+ & \bnfas &
                     \appspine{e}{s}
    \\[2pt]
        \mbox{Patterns} & \pat & \bnfas &
                      x
              \bnfalt \pair{\pat_1}{\pat_2}
              \bnfalt \inj{1} \pat
              \bnfalt \inj{2} \pat
              \bnfalt \vecnil
              \bnfalt \veccons{\pat_1}{\pat_2}
    \\[2pt]
        \mbox{Branches} & \pi & \bnfas & \branch{\patvec}{e}
    \\[2pt]
        \mbox{Branch lists}\!\!\! & \Pi & \bnfas &
                      \cdot
             \bnfalt \big(\pi \alt \Pi\big)
    \vspace{0.6ex}
    \end{bnfarray}

  \caption{Source syntax}
  \FLabel{fig:source-syntax}
\end{figure}

%% file: fig-decl-syntax.tex
\begin{figure}[t]

   \begin{bnfarray}
      \text{Universal variables}\hspace{-1ex} & \alpha, \beta, \gamma &
      \\[2pt]
      \text{Sorts} & \sort & \bnfas & 
           \type
           \bnfalt \ind 
    \\[1pt]
    \text{Types} & A, B, C\! & \bnfas &
          \unitty
          \bnfalt  A \arr B
          \bnfalt  A + B
          \bnfalt  A \times B
    \\ &&&\!\!\!
          \bnfalt  \alpha
          \bnfalt  \alltype{\alpha:\sort}{A}
          \bnfalt  \extype{\alpha:\sort}{A}
    \\ &&&\!\!\!
          \bnfalt  P \implies A
          \bnfalt  A \with P
          \bnfalt  \vectype{t}{A}
    \\[1pt]
    \text{Terms/monotypes}\hspace{-1ex} & t,\tau,\sigma & \bnfas &
           \zero
           \bnfalt  \succ{t}
           \bnfalt  \unitty
           \bnfalt  \alpha
    \\ &&&\!\!\!
          \bnfalt  \tau \arr \sigma
          \bnfalt  \tau + \sigma
          \bnfalt  \tau \times \sigma
    \\[1pt]
      \text{Propositions} & P,Q & \bnfas & t = t'
    \\[1pt]
    \text{Contexts} & \Psi
                    & \bnfas & \cdot
                     \bnfalt   \Psi, \hyp{}{\alpha}{\sort}
                     \bnfalt   \Psi, \hyp{p}{x}{A}
    \\[1pt]
      \text{Polarities} & \polvar & \bnfas & + \bnfalt - \bnfalt \circ 
    \\
    \text{Binary connectives}\hspace{-4ex} & \binc & \bnfas &
           {\arr}
           \bnfalt  {+}
           \bnfalt  {\times}
    \\
    \text{Principalities} & p, q & \bnfas &
          \p
          \bnfalt
          \underbrace{\OK}_{\text{\hspace{-4ex}sometimes omitted\hspace{-4ex}}}
    \end{bnfarray}

\caption{Syntax of declarative types and contexts}
\FLabel{fig:decl-syntax}
\end{figure}

%% file: confdeclgrue.tex
\begin{figure}
\newcommand{\Room}[2]{
  \!\!\!\begin{tabular}[t]{l}
    #1
    \\{}%
    \ensuremath{#2}%
  \end{tabular}\!\!\!%
}
\small
    $~$\!\begin{tikzpicture}
    \gdef\CompactJudgments{1}
      [auto, node distance=2cm, >=latex,  %
       descr/.style= {fill=white, inner sep=4pt, anchor=center}
      ]

      \node (sub) {\Room{subtyping}{\declsubjudg[\polvar]{\Psi}{A}{B}}};

      \node [below of=sub, node distance=38pt] (chk) {\Room{type checking}{\declchkjudg{p}{\Psi}{e}{A}}};

      \node [below of=chk, left=-15pt, node distance=48pt] (syn) {\Room{type synthesis}{\declsynjudg{p}{\Psi}{e}{B}}};

      \node [left of=sub, node distance=74pt] (chkelim) {\Room{checking, eq.\ elim.}{\declchkjudg{p}{\Psi \ctxsep P}{e}{C}}};

      \draw [->] (chk) -- (sub);

      \node [left of=chk, node distance=85pt] (spine) {\small \Room{spine typing}{\declspinejudg{\Psi}{s}{A\!}{p\!}{\!B}{\!q}}\!\!};

      \node [below of=spine, node distance=45pt, left=-5pt] (recspine) {\small\Room{\!\!\!\!\begin{tabular}{l} principality-recovering\!\!\!\! \\ spine typing \end{tabular}}{\declrecspinejudg \Psi s A p B q}};

      \path [->, %
      line width=1pt, color=dRed] (syn) edge (recspine);
      \draw [->, line width=1pt, color=dRed] (recspine) -- (spine);
      \draw [<->, line width=1pt, color=dRed] (chk) -- (syn);
      \draw [<->, line width=1pt, color=dRed] (chk) -- (chkelim);
      \draw [->, line width=1pt, color=dRed] (spine) -- (chk);
      \node [below of=chk, right=35pt, node distance=46pt] (match) {\Room{pattern matching}{\declmatchjudg p \Psi \Pi \Avec C}};
      \node [right of=chk, node distance=94pt] (elimeq) {\Room{match, eq.\ elim.}{\declmatchelimjudg p \Psi \Pi P \Avec C}};

      \node [right of=sub, below=1pt, right=24pt, node distance=15pt] (covers) {\Room{coverage}{\covers{q}{\Psi}{\Pi}{\Avec}}};
      \draw [->] (chk) -- (covers);

      \draw [<->, line width=1pt, color=dRed] (match) -- (chk);
      \draw [<->, line width=1pt, color=dRed] (match) -- (elimeq);

    \end{tikzpicture}

\vspace*{-0.6ex}

\caption{Dependency structure of the declarative judgments}
\FLabel{fig:declgrue}
\end{figure}

%% file: fig-decl-subtyping.tex
\begin{figure}[t]
\raggedright
\begin{minipage}[t]{0.45\linewidth}
~\\[-11pt]

  \loudjudgbox{polarity}
               {\arrayenvcl{
                   \polarity{A}
                   \\
                   \nonPos{A}
                   \\
                   \nonNeg{A}
                 }
               } 
               {Determine the polarity of a type
                 \\[0.5ex] 
                Check if $A$ not positive
                \\[0.5ex] 
                Check if $A$ not negative
                \\[0.3ex]
              }
  \begin{mathpar}
    \begin{array}{lcl}
      \polarity{\alltype{\alpha:\sort}{A}} & = & - \\
      \polarity{\extype{\alpha:\sort}{A}} & = & + \\
      \polarity{A} & = & \circ \;\;\mbox{otherwise} \\[0.6ex]

      \nonPos{A} & \text{iff} & \polarity{A} \neq + \\
      \nonNeg{A} & \text{iff} & \polarity{A} \neq -
    \end{array}
  \end{mathpar}
\end{minipage}
~~
\begin{minipage}[t]{0.45\linewidth}
~\\[-11pt]

  \loudjudgbox{polarity}
               {\joinpolarity{\polvar_1}{\polvar_2}
               } 
               {
                Join polarities
              }
  \begin{mathpar}
    \begin{array}{lcl}
      \joinpolarity{+}{\polvar_2} & = & + \\ 
      \joinpolarity{-}{\polvar_2} & = & - \\ 
      \joinpolarity{\circ}{+} & = & + \\ 
      \joinpolarity{\circ}{-} & = & - \\ 
      \joinpolarity{\circ}{\circ} & = & -
    \end{array}
  \end{mathpar}
\end{minipage}

\medskip
  
  \loudjudgbox{declsubjudg}
                {\declsubjudg[\polvar]{\Psi}{A}{B}}
                {Under context $\Psi$, type $A$ is a subtype of $B$,
                 decomposing head connectives of polarity $\polvar$}
  \vspace*{-1.0ex}
  \begin{mathpar} 
    \Infer{\!\DsubReflPm}
          {
            \judgetp{\Psi}{A}
            \\
            \nonPos{A} \\ \nonNeg{A}
          }
          {\declsubjudg[\polvar]{\Psi}{A}{A}}
   \hspace*{-0.3ex}
   \arrayenvl{
       \Infer{\DsubPosNeg}
             {\declsubjudg[-]{\Psi}{A}{B} \\ \nonPos{A} \\ \nonPos{B}}
             {\declsubjudg[+]{\Psi}{A}{B}}
       \\[0.4ex]
       \Infer{\DsubNegPos}
             {\declsubjudg[+]{\Psi}{A}{B} \\ \nonNeg{A} \\ \nonNeg{B}}
             {\declsubjudg[-]{\Psi}{A}{B}}
   }
   \\
    \Infer{\!\DsubAllL}
              {\judge{\Psi}{\tau}{\sort}
                \\
                \declsubjudg[-]{\Psi}{[\tau/\alpha]A}{B}}
              {\declsubjudg[-]{\Psi}{\alltype{\alpha{:}\sort}{A}}{B}}
    ~~~~
    \Infer{\!\DsubAllR}
          {\declsubjudg[-]{\Psi, \beta : \sort}{A}{B}}
          {\declsubjudg[-]{\Psi}{A}{\alltype{\beta{:}\sort}{B}}}
    \\
    \Infer{\!\DsubExistsL}
              {\declsubjudg[+]{\Psi, \alpha:\sort}{A}{B}}
              {\declsubjudg[+]{\Psi}{\extype{\alpha{:}\sort}{A}}{B}}
    ~~~~
    \Infer{\!\DsubExistsR}
          {\judge{\Psi}{\tau}{\sort} \\
           \declsubjudg[+]{\Psi}{A}{[\tau/\beta]B}}
          {\declsubjudg[+]{\Psi}{A}{\extype{\beta{:}\sort}{B}}}
  \vspace{-0.5ex}
  \end{mathpar}
  
  \caption{Subtyping in the declarative system}
  \FLabel{fig:decl-subtyping}
\end{figure}

%% file: fig-chkintro.tex
\begin{figure*}[htbp]
  \raggedright
  \loudjudgbox{chkintro}{\chkintro{e}}{\small \!Expression $e$ is a checked introduction form}
  \vspace*{-3.4ex}
  \[
      \Infer{}{}{\chkintro{\lam{x} e}}
      ~~~~~
      \Infer{}{}{\chkintro{\unitexp}}
      ~~~~
      \Infer{}{}{\chkintro{\pair{e_1}{e_2}}}
      ~~~~~
      \Infer{}{}{\chkintro{\inj{k} e}}
      ~~~~~
      \Infer{}{}{\chkintro{\vecnil}}
      ~~~~
      \Infer{}{}{\chkintro{\veccons{e_1}{e_2}}}
  \vspace*{-1.0ex}
  \]

  \caption{``Checking intro form''}
  \label{fig:chkintro}
\end{figure*}

%% file: fig-decl-typing.tex
\begin{figure}[htbp]
  \centering
  \raggedright
\ifnum\OPTIONInAppendix=1
\runonfontsz{9.0pt}
\fi
  \loudjudgbox{declcheckprop}{\declcheckprop{\Psi}{P}}
          {Under context $\Psi$, check $P$}
   ~\\[-4.0ex]
   \hspace*{0.4\textwidth}$\Infer{\DeclCheckpropEq}
         { }
         {\declcheckprop{\Psi}{(t = t)}}$

  \begin{tabular}[t]{@{}l@{}}
    \vspace*{1.0ex}
  \judgbox{
    \arrayenvcl{
    \declchkjudg{p}{\Psi}{e}{A}
    \\
    \declsynjudg{p}{\Psi}{e}{A}
  }
  }%
     {Under context $\Psi$, expression $e$ checks against input type $A$
       \\[2pt]
     Under context $\Psi$, expression $e$ synthesizes output type $A$}
  \end{tabular}
  \begin{mathpar}
     \Infer{\DeclVar}
          {\hyp{p}{x}{A} \in \Psi}
          {\declsynjudg{p}{\Psi}{x}{A}}
     \and
     \Infer{\DeclSub}
          {\declsynjudg{q}{\Psi}{e}{A}
            \\
            \declsubjudg[\joinpolarity{\polarity{B}}{\polarity{A}}]{\Psi}{A}{B}
          }
          {\declchkjudg{p}{\Psi}{e}{B}}
     \and
     \Infer{\DeclAnno}
          {\judgetp{\Psi}{A}
           \\
           \declchkjudg{\p}{\Psi}{e}{A}
          }
          {\declsynjudg{\p}{\Psi}{(e : A)}{A}}
     ~~~~
     \Infer{\!\DeclRec}
          {
            \declchkjudg{p}{\Psi, \hyp{p}{x}{A}}{v}{A}
          }
          { \declchkjudg{p}{\Psi}{\rec{x} v}{A} }
     ~~~~
     \Infer{\DeclUnitIntro}
          {}
          {\declchkjudg{p}{\Psi}{\unitexp}{\unitty}}
     \and
     \Infer{\DeclAllIntro}
           {
             \chkintro{v}
             \\
             \declchkjudg{p}{\Psi, \alpha:\sort}{v}{A}
           }
           {\declchkjudg{p}{\Psi}{v}{(\alltype{\alpha:\sort} A)}}
     \and
     \Infer{\DeclExIntro}
           {\judge{\Psi}{\tau}{\sort} \\
            \declchkjudg{\ok}{\Psi}{e}{[\tau/\alpha]A}}
           {
             \declchkjudg{p}{\Psi}{e}{(\extype{\alpha:\sort} A)}
           }
     \and
     \Infer{\DeclImpliesIntro}
           { \chkintro{v}
             \\
             \declchkjudg{\p}{\Psi \ctxsep P}{v}{A}
           }
           {\declchkjudg{\p}{\Psi}{v}{(P \implies A)}}
     \and
     \Infer{\DeclWithIntro}
           {
                \declcheckprop{\Psi}{P}
                \\
                \declchkjudg{p}{\Psi}{e}{A}
           }
           {\declchkjudg{p}{\Psi}{e}{(A \with P)}}
     \\
     \Infer{\!\DeclArrIntro}
          {\declchkjudg{p}{\Psi, \hyp{p}{x}{A}}{e}{B}
          }
          {\declchkjudg{p}{\Psi}{\lam{x} e}{A \arr B}}
     ~~
     \Infer{\!\DeclArrElim}
          {\declsynjudg{p}{\Psi}{e}{A}
           \\
           \declrecspinejudg{\Psi}{s}{A}{p}{C}{q}
          }
          {\declsynjudg{q}{\Psi}{e\;s}{C}}
     \ifnum\OPTIONConf=1%
       \vspace{-0.0ex}
     \fi
\vspace{-0.5ex}
\vspace{-0.5ex}
     \\
\ifnum\OPTIONInAppendix=1%
     \Infer{\DeclSumIntro{k}}
           {\declchkjudg{p}{\Psi}{e}{A_k}}
           {\declchkjudg{p}{\Psi}{\inj{k}{e}}{A_1 + A_2}}
     \and
     \Infer{\DeclPairIntro}
            {\declchkjudg{p}{\Psi}{e_1}{A_1}
              \\
              \declchkjudg{p}{\Psi}{e_2}{A_2}}
            {\declchkjudg{p}{\Psi}{\pair{e_1}{e_2}}{A_1 \times A_2}}
      \\
      \Infer{\DeclNil}
            {\declcheckprop{\Psi}{t = \zero}}
            {\declchkjudg{p}{\Psi}{\vecnil}{(\vectype{t}{A})}}
      ~~~
      \Infer{\DeclCons}
            {
              \declcheckprop{\Psi}{t = \succ{t_2}}
              \\
              \arrayenvbl{
              \declchkjudg{p}{\Psi}{e_1}{A}
              \\
              \declchkjudg{\OK}{\Psi}{e_2}{(\vectype{t_2}{A})}
            }
            }
            {
              \declchkjudg{p}{\Psi}{\veccons{e_1}{e_2}}{(\vectype{t}{A})}
            }
      \\
\fi
     \Infer{\DeclCase}
           {
                 \declsynjudg{q}{\Psi}{e}{A}
                  \\
                  \declmatchjudg{p}{\Psi}{\Pi}{A}{C}
                  \\
                  \forall B.\; \mbox{if } \declsynjudg{q}{\Psi}{e}{B} \mbox{ then } 
                  \declcovers[q]{\Psi}{\Pi}{B}
           }
           {\declchkjudg{p}{\Psi}{\case{e}{\Pi}}{C}}
  \ifnum\OPTIONInAppendix=0%
    \vspace{-1.0ex}
  \fi
     \end{mathpar}

  \judgbox{\arrayenvcl{
      \declspinejudg{\Psi}{s}{A}{p}{C}{q}
      \\
      \declrecspinejudg{\Psi}{s}{A}{p}{C}{q}
     }}%
     {Under context $\Psi$, \\
       passing spine $s$ to a function of type $A$ synthesizes type $C$; \\
       in the $\symrecspine{q}$ form, recover principality in $q$ if possible
    }
  \begin{mathpar}
     \Infer{\!\DeclAllSpine}
          {\judge{\Psi}{\tau}{\sort}
            ~~~~~~
            \declspinejudg{\Psi}{\appspine{e}{s}}{[\tau/\alpha]A}{\OK}{C}{q}
          }
          {\declspinejudg{\Psi}{\appspine{e}{s}}{(\alltype{\alpha:\sort} A)}{p}{C}{q}}
     ~~
     \Infer{\!\DeclImpliesSpine}
         {\declcheckprop{\Psi}{P}
          ~~~~~~
          \declspinejudg{\Psi}{\appspine{e}{s}}{A}{p}{C}{q}}
         {\declspinejudg{\Psi}{\appspine{e}{s}}{(P \implies A)}{p}{C}{q}}
     \\ 
     \Infer{\DeclEmptySpine}
         {}
         {\declspinejudg{\Psi}{\emptyspine}{A}{p}{A}{p}}
     ~~~~~
     \Infer{\DeclArrSpine}
         {%
             \declchkjudg{p}{\Psi}{e}{A}
             \\
             \declspinejudg{\Psi}{s}{B}{p}{C}{q}
         }
         {\declspinejudg{\Psi}{\appspine{e}{s}}{A \arr B}{p}{C}{q}}
     \\
     \Infer{\!\DeclRecover}
          {
            \declspinejudg{\Psi}{s}{A}{\p}{C}{\OK}
            ~~~
            \arrayenvbl{
              \text{for all $C'$.\;}
              \\
                  \text{~~if~~}\declspinejudg{\Psi}{s}{A}{\p}{C'}{\OK}
                  \\
                  ~\text{\;~then~}C' = C
             }
          }
          {\declrecspinejudg{\Psi}{s}{A}{\p}{C}{\p}}
     ~~~
     \Infer{\!\DeclPass}
          { \arrayenvbl{
              \declspinejudg{\Psi}{s}{A}{p}{C}{q}
            }
          }
          {\declrecspinejudg{\Psi}{s}{A}{p}{C}{q}}
  \end{mathpar}

  \judgbox{\declchkjudg{p}{\Psi \ctxsep P}{e}{C}}
       {Under context $\Psi$, incorporate proposition $P$
         and check $e$ against $C$}
\ifnum\OPTIONInAppendix=0
  \vspace{-2.7ex}
\else
  \vspace{-2.7ex}
\fi
  \begin{mathpar}
     \Infer{\DeclCheckBot}
           {\mgu{\sigma}{\tau} = \bot}
           {\declchkjudg{p}{\Psi \ctxsep (\sigma = \tau)}{e}{C}}
     ~~~~~~~~~~
     \Infer{\DeclCheckUnify}
           {\arrayenvbl{
               \mgu{\sigma}{\tau} = \theta
               \\
               \declchkjudg{p}{\theta(\Psi)}{\theta(e)}{\theta(C)}
             }
           }
           {\declchkjudg{p}{\Psi \ctxsep (\sigma = \tau)}{e}{C}}
  \vspace*{-1.0ex}
  \end{mathpar}

\ifnum\OPTIONInAppendix=0
  \caption{Declarative typing, omitting rules for $\times$, $+$, and $\vecop$}
  \FLabel{fig:decl-typing}
\else
  \repeatcaption{mainfig:decl-typing}{a}{Declarative typing, including rules omitted from main paper}
\fi

\end{figure}

%% file: fig-decl-pattern-matching.tex
\begin{figure*}
\runonfontsz{9pt}
  \raggedright
  \judgbox{\declmatchjudg[q]{p}{\Psi}{\Pi}{\Avec}{C}}
     {Under context $\Psi$,
       \\
       check branches $\Pi$ with patterns of type $\Avec$ and bodies of type $C$}
  \begin{mathpar}
    \Infer{\DeclMatchEmpty}
          {}
          {\declmatchjudg[q]{p}{\Psi}{\cdot}{\Avec}{C}}
    ~~~~~
    \Infer{\DeclMatchSeq}
          {\declmatchjudg[q]{p}{\Psi}{\pi}{\Avec}{C} \\
           \declmatchjudg[q]{p}{\Psi}{\Pi}{\Avec}{C} }
          {\declmatchjudg[q]{p}{\Psi}{(\pi \alt \Pi)}{\Avec}{C}}
    \and
    \Infer{\DeclMatchBase}
          {\declchkjudg{p}{\Psi}{e}{C}}
          {\declmatchjudg[q]{p}{\Psi}{(\branch{\cdot}{e})}{\cdot}{C}}
    \and
    \Infer{\DeclMatchUnit}
          { 
            \declmatchjudg[q]{p}{\Psi}{\branch{\patvec}{e}}{\Avec}{C} }
          { \declmatchjudg[q]{p}{\Psi}{\branch{\unitexp, \patvec}{e}}{\unitty, \Avec}{C} }

    \\ 
    \Infer{\!\runonfontsz{8pt}\DeclMatchExists}
          {\declmatchjudg[q]{p}{\Psi, \alpha:\sort}{\branch{\patvec}{e}}{A,\Avec}{C} }
          {\declmatchjudg[q]{p}{\Psi}{(\branch{\patvec}{e})}{(\extype{\alpha:\sort} A),\Avec}{C}}
    \Infer{\!\runonfontsz{8pt}\DeclMatchPair}
          {\declmatchjudg[q]{p}{\Psi}{\branch{\pat_1, \pat_2, \patvec}{e}}{A_1, A_2, \Avec}{C}}
          {\declmatchjudg[q]{p}{\Psi}{\branch{\pair{\pat_1}{\pat_2}, \patvec}{e}}{(A_1 \times A_2), \Avec}{C}}
    \and
    \Infer{\!\DeclMatchSum{k}}
          {\declmatchjudg[q]{p}{\Psi}{\branch{\pat, \patvec}{e}}{A_k, \Avec}{C}}
          {\declmatchjudg[q]{p}{\Psi}{\branch{\inj{k}{\pat}, \patvec}{e}}{A_1 + A_2, \Avec}{C} }
    \and
    \Infer{\!\DeclMatchWith}
          {\declmatchelimjudg{p}{\Psi}{\branch{\patvec}{e}}{P}{A, \Avec}{C} }
          {\declmatchjudg[!]{p}{\Psi}{\branch{\patvec}{e}}{(A \with P), \Avec}{C} }
    ~~
    \Infer{\!\DeclMatchWithOK}
          {
            \declmatchjudg[\OK]{p}{\Psi}{\branch{\patvec}{e}}{A, \Avec}{C}
          }
          {
            \declmatchjudg[\OK]{p}{\Psi}{\branch{\patvec}{e}}{(A \with P), \Avec}{C}
          }
    \\
    \Infer{\DeclMatchConsOK}
          {
            \declmatchjudg[\OK]{p}
                              {\Psi, \alpha : \kindnat}
                              {\branch{\pat_1, \pat_2, \patvec}{e}}
                              {A, (\vectype{\alpha}{A}), \Avec}
                              {C}
          }
          {
            \declmatchjudg[\OK]{p}{\Psi}{\branch{(\veccons{\pat_1}{\pat_2}), \patvec}{e}}{(\vectype{t}{A}), \Avec}{C}
          }
    \and
    \Infer{\DeclMatchCons}
          {
            \declmatchelimjudg{p}
                              {\Psi, \alpha : \kindnat}
                              {\branch{\pat_1, \pat_2, \patvec}{e}}
                              {(t = \succ{\alpha})}
                              {A, (\vectype{\alpha}{A}), \Avec}
                              {C}
          }
          {
            \declmatchjudg[!]{p}{\Psi}{\branch{(\veccons{\pat_1}{\pat_2}), \patvec}{e}}{(\vectype{t}{A}), \Avec}{C}
          }
    \\
    \Infer{\!\DeclMatchNilOK}
          {
            \declmatchjudg[\OK]{p}{\Psi}{\branch{\patvec}{e}}{\Avec}{C} 
          }
          {
            \declmatchjudg[\OK]{p}{\Psi}{\branch{\vecnil, \patvec}{e}}{(\vectype{t}{A}), \Avec}{C}
          }
    \Infer{\!\DeclMatchNil}
          {
            \declmatchelimjudg{p}{\Psi}{\branch{\patvec}{e}}{(t = \zero)}{\Avec}{C} 
          }
          {
            \declmatchjudg[!]{p}{\Psi}{\branch{\vecnil, \patvec}{e}}{(\vectype{t}{A}), \Avec}{C}
          }
    \\
    \Infer{\!\DeclMatchNeg}
          {
            \arrayenvbl{
              \notWithExists{A}
              \\
              \declmatchjudg[q]{p}{\Psi, \hyp{\p}{x}{A}}{\branch{\patvec}{e}}{\Avec}{C}
            }
          }
          { \declmatchjudg[q]{p}{\Psi}{\branch{x, \patvec}{e}}{A, \Avec}{C} }
    \and
    \Infer{\!\DeclMatchWild}
          {
            \arrayenvbl{
              \notWithExists{A}
              \\
              \declmatchjudg[q]{p}{\Psi}{\branch{\patvec}{e}}{\Avec}{C}
            }
          }
          { \declmatchjudg[q]{p}{\Psi}{\branch{\wild, \patvec}{e}}{A, \Avec}{C} }
  \end{mathpar}

  \medskip

  \judgbox{\declmatchelimjudg{p}{\Psi}{\Pi}{P}{\Avec}{C}}
       {Under context $\Psi$, incorporate proposition $P$ while checking branches $\Pi$
         \\
         with patterns of type $\Avec$ and bodies of type $C$}
  \begin{mathpar}
     \Infer{\DeclMatchBot}
           {\mgu{\sigma}{\tau} = \bot}
           {\declmatchelimjudg{p}{\Psi}{\branch{\patvec}{e}}{\sigma = \tau}{\Avec}{C}}
     \and 
     \Infer{\DeclMatchUnify}
           {\mgu{\sigma}{\tau} = \theta \\
            \declmatchjudg[q]{p}{\theta(\Psi)}{\theta(\branch{\patvec}{e})}{\theta(\Avec)}{\theta(C)}
           }
           {\declmatchelimjudg{p}{\Psi}{\branch{\patvec}{e}}{\sigma = \tau}{\Avec}{C}}
\vspace*{-1ex}
\end{mathpar}

  \caption{Declarative pattern matching}
  \FLabel{fig:decl-pattern-matching}
\end{figure*}

%% file: fig-decl-match-coverage.tex
\begin{figure}
\raggedright
\runonfontsz{8.7pt}
  \judgbox{
    \arrayenvcl{
    \declcovers[p]{\Psi}{\Pi}{\Avec}
    \\
    \declcoverseq{\Psi}{P}{\Pi}{\Avec}
    \\
    \guarded{\Pi}
  }
  }
     {Patterns $\Pi$ cover the types $\Avec$ in context $\Psi$
       \\[3pt]
      Patterns $\Pi$ cover the types $\Avec$ in context $\Psi$, assuming $P$
      \\[3pt]
      Pattern list $\Pi$ contains a list pattern constructor at the head position
    }
  \vspace{-1.3ex}
  \begin{mathpar}
    \Infer{\DeclCoversEmpty}
          { }
          {\declcovers[p]{\Psi}{(\branch{\cdot}{e_1}) \alt \Pi'}{\cdot}}
    \and
    \Infer{\DeclCoversVar}
          {\expandvar{\Pi}{\Pi'} \\
           \declcovers[p]{\Psi}{\Pi'}{\Avec}
          }
          { \declcovers[p]{\Psi}{\Pi}{A, \Avec}
          }
    \and
    \Infer{\DeclCoversUnit}
          {\expandunit{\Pi}{\Pi'} \\
           \declcovers[p]{\Psi}{\Pi'}{\Avec}
          }
          { \declcovers[p]{\Psi}{\Pi}{\unitty, \Avec}
          }
    ~~~
    \Infer{\DeclCoversTimes}
          { \expandpair{\Pi}{\Pi'} \\
            \declcovers[p]{\Psi}{\Pi'}{A_1, A_2, \Avec} 
          }
          { \declcovers[p]{\Psi}{\Pi}{(A_1 \times A_2), \Avec}
          }
    \and
    \Infer{\DeclCoversSum}
          {
            \expandsum{\Pi}{\Pi_L}{\Pi_R}
            \\
            \arrayenvbl{
              \declcovers[p]{\Psi}{\Pi_L}{A_1, \Avec}
              \\
              \declcovers[p]{\Psi}{\Pi_R}{A_2, \Avec} 
            }
          }
          { \declcovers[p]{\Psi}{\Pi}{(A_1 + A_2), \Avec}
          }
    \and
    \Infer{\DeclCoversEx}
          {
           \declcovers[p]{\Psi, \alpha : \sort}{\Pi}{A, \Avec}
          }
          { \declcovers[p]{\Psi}{\Pi}{(\extype{\alpha:\sort} A), \Avec}
          }
    \\
    \Infer{\DeclCoversWith}
          { \declcoverseq{\Psi}{t_1 = t_2}{\Pi}{A_0, \Avec} }
          { \declcovers{\Psi}{\Pi}{\big(A_0 \with (t_1 = t_2)\big), \Avec} 
          }
    \and
    \Infer{\DeclCoversWithOK}
          { \declcovers[\OK]{\Psi}{\Pi}{A_0, \Avec} }
          { \declcovers[\OK]{\Psi}{\Pi}{\big(A_0 \with (t_1 = t_2)\big), \Avec} 
          }
    \\
    \Infer{\DeclCoversVec}
          {
           \guarded{\Pi} \and 
            \expandvec{\Pi}{\Pi_{[]}}{\Pi_{::}}
            \\
            \arrayenvbl{
              \declcoverseq{\Psi}{t = \zero}{\Pi_{[]}}{\Avec}
              \\
              \declcoverseq{\Psi, n:\ind}{t = \succ{n}}{\Pi_{::}}{(A, \vectype{n}{A}, \Avec)} 
            }
          }
          { \declcovers{\Psi}{\Pi}{\vectype{t}{A}, \Avec} }
    \and
    \Infer{\DeclCoversVecOK}
          {
           \guarded{\Pi} \and 
            \expandvec{\Pi}{\Pi_{[]}}{\Pi_{::}}
            \\
            \arrayenvbl{
              \declcovers[\OK]{\Psi}{\Pi_{[]}}{\Avec}
              \\
              \declcovers[\OK]{\Psi, n:\ind}{\Pi_{::}}{(A, \vectype{n}{A}, \Avec)} 
            }
          }
          { \declcovers[\OK]{\Psi}{\Pi}{\vectype{t}{A}, \Avec} }
   \vspace*{-0.2ex}
    \\
    \Infer{\DeclCoversEq}
          {
              \mgu{t_1}{t_2} = \theta
              \\
              \declcovers{\theta(\Psi)}{\theta(\Pi)}{\theta(\Avec)}
          }
          { \declcoverseq{\Psi}{t_1 = t_2}{\Pi}{\Avec} 
          }
    ~~~~
    \Infer{\DeclCoversEqBot}
          {\mgu{t_1}{t_2} = \bot}
          {\declcoverseq{\Psi}{t_1 = t_2}{\Pi}{\Avec}}
          \\
  \Infer{ }
        { }
        {\guarded{\branch{\vecnil, \vec{p}}{e} \alt \Pi}}
  ~~~
  \Infer{ }
        { }
        {\guarded{\branch{\veccons{p}{p'}, \vec{p}}{e} \alt \Pi}}
  ~~~
  \Infer{ }
        {\guarded{\Pi} }
        {\guarded{\branch{\wild, \vec{p}}{e} \alt \Pi}}
  ~~~
  \Infer{ }
        {\guarded{\Pi} }
        {\guarded{\branch{x, \vec{p}}{e} \alt \Pi}}
  \end{mathpar}

  \caption{Match coverage}
  \FLabel{fig:decl-match-coverage}
\end{figure}

\begin{figure}
\runonfontsz{8.7pt}
\raggedright
\judgbox{\expandvec{\Pi}{\Pi_{[]}}{\Pi_{::}}}
           {Expand vector patterns in $\Pi$}
  \vspace*{-4.3ex}
  \begin{mathpar}
     \hspace*{15ex}\Infer{}
           { }
           {\expandvec{\cdot}{\cdot}{\cdot}}
     \and
     \Infer{}
          { \pat \in \setof{x, \wild}
            \\
            \expandvec{\Pi}{\Pi_{[]}}{\Pi_{::}}
          }
          {\expandvec{\big(\branch{\pat, \pvec}{e}\big) \alt \Pi }
                     {\big(\branch{\pvec}{e}\big) \alt \Pi_{[]}}
                     {\big(\branch{\wild, \wild, \pvec}{e}\big) \alt \Pi_{::}}
          }
     \and
     \Infer{}
           {\expandvec{\Pi}{\Pi_{[]}}{\Pi_{::}} }
           {\expandvec{\big(\branch{[], \pvec}{e}\big) \alt \Pi }
                      {\big(\branch{\pvec}{e}\big) \alt \Pi_{[]}}
                      {\Pi_{::}}
           }
     \and
     \Infer{}
           {\expandvec{\Pi}{\Pi_{[]}}{\Pi_{::}} }
           {\expandvec{\big(\branch{(\veccons{\pat}{\pat'}, \pvec}{e}\big) \alt \Pi }
                      {\Pi_{[]}}
                      {\big(\branch{\pat, \pat', \pvec}{e}\big) \alt \Pi_{::}}
           }
  \end{mathpar}
   
  \judgbox{\expandpair{\Pi}{\Pi'}}
     {~\\[-16pt] Expand head pair patterns in $\Pi$}
  \vspace{-3.0ex}
  \begin{mathpar}
     \Infer{}
           { }
           {\expandpair{\cdot}{\cdot}}
     ~
     \Infer{}
           {\expandpair{\Pi}{\Pi'}  }
           {\expandpair{\big( \branch{\pair{\pat_1}{\pat_2}, \pvec}{e} \big) \alt \Pi}
                       { \big(\branch{\pat_1, \pat_2, \pvec}{e}\big) \alt \Pi'}
            }
    ~
    \Infer{}
          {\pat \in \setof{z, \wild}
           \\
           \expandpair{\Pi}{\Pi'}  }
          {\expandpair{\big(\branch{\pat, \pvec}{e}\big) \alt \Pi }
                      {\big(\branch{\wild, \wild, \pvec}{e}\big) \alt \Pi' }
          }
  \end{mathpar}

  \judgbox{\expandsum{\Pi}{\Pi_L}{\Pi_R}}
     {Expand head sum patterns in $\Pi$ into left $\Pi_L$ and right $\Pi_R$ sets}
  \vspace{-2.7ex}
  \begin{mathpar}
     \Infer{}
           { }
           {\expandsum{\cdot}{\cdot}{\cdot}}
     \and
     \Infer{}
          { \pat \in \setof{x, \wild}
            \\
            \expandsum{\Pi}{\Pi_L}{\Pi_R}
          }
          {\expandsum{\big(\branch{\pat, \pvec}{e}\big) \alt \Pi }
                     {\big(\branch{\wild, \pvec}{e}\big) \alt \Pi_L}
                     {\big(\branch{\wild, \pvec}{e}\big) \alt \Pi_R}
          }
     \and
     \Infer{}
           {\expandsum{\Pi}{\Pi_L}{\Pi_R} }
           {\expandsum{\big(\branch{\inj{1}{\pat}, \pvec}{e}\big) \alt \Pi }
                      {\big(\branch{\pat, \pvec}{e}\big) \alt \Pi_L}
                      {\Pi_R}
           }
     ~~
     \Infer{}
           {\expandsum{\Pi}{\Pi_L}{\Pi_R} }
           {\expandsum{\big(\branch{\inj{2}{\pat}, \pvec}{e}\big) \alt \Pi }
                      {\Pi_L}
                      {\big(\branch{\pat, \pvec}{e}\big) \alt \Pi_R}
           }
  \end{mathpar}

  \begin{minipage}{0.45\linewidth}
  \judgbox{\expandvar{\Pi}{\Pi'}}
     {Remove head variable
        \\ and wildcard patterns from $\Pi$}
  \[
    \Infer{}
          { }
          {\expandvar{\cdot}{\cdot}}
    ~~~~~
    \Infer{}
          { \pat \in \setof{x, \wild} 
           \\
           \expandvar{\Pi}{\Pi'}}
          { \expandvar{\big(\branch{\pat, \pvec}{e}\big) \alt \Pi }
                     {\big(\branch{\pvec}{e}\big) \alt \Pi' }
                     }
  \]
  \end{minipage}
   ~~~~~~~
  \begin{minipage}{0.45\linewidth}
  \judgbox{\expandunit{\Pi}{\Pi'}}
     {Remove head variable, wildcard,
       \\ and unit patterns from $\Pi$}
  \[
    \Infer{}
          { }
          {\expandunit{\cdot}{\cdot}}
    ~~~~~~~~~~
    \Infer{}
          { \pat \in \setof{x, \wild, \unitexp}
           \\
           \expandunit{\Pi}{\Pi'}}
          { \expandunit{\big(\branch{\pat, \pvec}{e}\big) \alt \Pi }
                     {\big(\branch{\pvec}{e}\big) \alt \Pi' }
                     }
  \]
  \end{minipage}

  \vspace*{-1ex}
  
  \caption{Pattern expansion}
  \FLabel{fig:decl-pattern-expansion}
\end{figure}

%% file: algorithmic.tex
\section{Algorithmic Typing}
\Label{sec:alg-typing}

\input{confgrue.tex}

Our algorithmic rules closely mimic our declarative rules,
except that whenever a declarative rule
would make a guess, the algorithmic rule adds to the context an existential
variable (written with a hat $\ahat$). As typechecking proceeds, we
add solutions to the existential variables, reflecting
increasing knowledge. Hence, each declarative typing
judgment has a corresponding algorithmic judgment with an
output context as well as an input context.  The algorithmic type checking
judgment $\chkjudg{p}{\Gamma}{e}{A}{\Delta}$ takes an input context
$\Gamma$ and yields an output context $\Delta$ that includes
increased knowledge about what the types have to be.  The notion
of increasing knowledge is formalized by a judgment
$\substextend{\Gamma}{\Delta}$ (\Sectionref{sec:context-extension}).

Figure~\ref{fig:grue} shows a dependency graph of the algorithmic judgments.
Each declarative judgment has a corresponding algorithmic judgment,
but the algorithmic system adds judgments such as type
equivalence checking
$\equivjudg{\Gamma}{A}{B}{\Delta}$ and variable instantiation
$\instjudg{\Gamma}{\ahat}{t}{\sort}{\Delta}$.
Declaratively, these
judgments correspond to uses of reflexivity axioms;
algorithmically, they correspond to solving existential
variables to equate terms. 

We give the algorithmic typing rules in Figure~\ref{fig:alg-typing};
rules for most other judgments are in the appendix.
Our style of specification broadly follows \mbox{\citet{Dunfield13}}:
we adapt their mechanisms of variable instantiation,
context extension, and context application (to both types and other contexts).
Our versions of these mechanisms, however, support
indices, equations over universal variables, and the $\exists$/$\implies$/$\with$
connectives.  We also differ in our formulation of spine typing,
and by being able to track which types are principal.

\subsection{Examples}
\label{sec:typing-examples}
To show how the spine typing rules %
recover principality,
we present some example derivations.

\newcommand{\Thespine}{(\appspine{\unitexp}{\emptyspine})}
\newcommand{\ahu}{\ahat : \type = \unitty}

Suppose we have an identity function $id$, defined in an algorithmic context
$\Gamma$ by the hypothesis $\hyp{\p}{id}{(\alltype{\alpha:\type} \alpha \arr \alpha)}$.
Since the hypothesis has $\p$, the type of $id$ is known to be principal.
If we apply $id$ to $\unitexp$, we expect to get something of unit type $\unitty$.
Despite the $\forall$ in the type of $id$,
the resulting type should be principal, because no other type is possible.
We can indeed derive that type:
\[
  \Infer{\ArrElim}
       {
        \Infer{\Var}
            {(\hyp{\p}{id}{(\alltype{\alpha:\type} \alpha \arr \alpha)}) \in \Gamma}
            {\synjudg{\p}{\Gamma}{id}{(\alltype{\alpha:\type} \alpha \arr \alpha)}{\Gamma}}
        ~~
              \recspinejudg{\Gamma}{\Thespine}{(\alltype{\alpha:\type} \alpha \arr \alpha)}{\p}{\unitty}{{\p}}{\Gamma, \ahu}
       }
       {\synjudg{\p}{\Gamma}{id\;\Thespine}{\unitty}{\Gamma, \ahu}}
\]
(Here, we write the application $id\;\unitexp$ as $id\;\Thespine$, to show the
structure of the spine as analyzed by the typing rules.)
In the derivation of the second premise of \ArrElim, shown below, we can follow
the evolution of the principality marker.
\vspace*{-1.3ex}
\[
\hspace*{11.5ex}
        \Infer{\!\Recover}
            {
              \hspace*{-1ex}
              \Infer{\!\AllSpine}
                 {
                   \hspace*{-11ex}
                   \Infer{\ArrSpine}
                      {
                        \Infer{\UnitIntroSolve}
                            {}
                            {\chkjudg{\OK}{\Gamma, \ahat{:}\type}{\unitexp}{\ahat}{\Gamma, \ahu}}
                        ~~~
                        \Infer{\EmptySpine}
                             {}
                             {
                               \spinejudg{\Gamma, \ahu}{\emptyspine}{\unitty}{\OK}{\unitty}{\OK}{\Gamma, \ahu}
                             }
                            \hspace*{-12ex}
                      }
                      {
                        \spinejudg{\Gamma, \ahat{:}\type}{\Thespine}{\ahat \arr \ahat}{\fighi{\OK}}{\unitty}{\OK}{\Gamma, \ahu}
                      }
                    \hspace*{-25ex}
                 }
                 {
                   \spinejudg{\Gamma}{\Thespine}{(\alltype{\alpha:\type} \alpha \arr \alpha)}{\p}{\unitty}{\fighi{\OK}}{\Gamma, \ahu}
                 }
              ~
              \FEV{\unitty} = \emptyset
            }
            {
              \recspinejudg{\Gamma}{\Thespine}{(\alltype{\alpha:\type} \alpha \arr \alpha)}{\underbrace{\p}_{\!\!\!\text{input}\!\!\!}}{\unitty}{\fighi{\p}}{\Gamma, \ahu}
            }
\]
\begin{itemize}
\item The input principality (marked ``input'') is $\p$, because the input type
  $(\alltype{\alpha:\type} \alpha \arr \alpha)$ was marked as principal in the
  hypothesis typing $id$.
\item Rule \Recover begins by invoking the ordinary (non-recovering) spine judgment,
  passing all inputs unchanged, including the principality $\p$.
\item Rule \AllSpine adds an existential variable $\ahat$ to represent the instantiation
  of the quantified type variable $\alpha$, and substitutes $\ahat$ for $\alpha$.
  Since this instantiation is, in general, \emph{not} principal, it replaces $\p$ with
  $\OK$ (highlighted) in its premise.  This marks the type $\ahat \arr \ahat$ as
  non-principal.
\item Rule \ArrSpine decomposes $\ahat \arr \ahat$ and checks $\unitexp$
  against $\ahat$, maintaining the principality $\OK$.
  Once principality is lost, it can only be recovered within the \Recover rule itself.
\item Rule \UnitIntroSolve notices that we are checking $\unitexp$ against
  an unknown type $\ahat$; since the expression is $\unitexp$, the type $\ahat$
  must be $\unitty$, so it adds that solution to its output context.
\item Moving to the second premise of \ArrSpine, we analyze the remaining
  part of the spine.  That is just the empty spine $\emptyspine$, and rule \EmptySpine
  passes its inputs along as outputs.  In particular, the principality $\OK$ is unchanged.
\item The principalities are passed down %
  to the conclusion of $\AllSpine$,
  where $\OK$ is highlighted.
\item In \Recover, we notice that the output type $\unitty$ has no existential variables
  ($\FEV{\unitty} = \emptyset$),
  which allows us to recover principality of the output type: $\symrecspine{\fighi{\p}}$.
\end{itemize}
In the corresponding derivation in our declarative system, we have, instead, a
check that no other types are derivable:
\vspace*{-2.0ex}
\[
        \Infer{\!\DeclRecover}
            {
              \Infer{\!\!\DeclAllSpine}
                 {
                   \judge{\Psi}{\unitty}{\type}
                   ~
                   \Infer{\!\!\DeclArrSpine}
                      {
                        \Infer{\!\DeclUnitIntro}
                            {}
                            {\declchkjudg{\OK}{\Psi}{\unitexp}{\unitty}}
                        ~~
                        \Infer{\!\!\DeclEmptySpine}
                             {}
                             {
                               \declspinejudg{\Psi}{\emptyspine}{\unitty}{\OK}{\unitty}{\OK}
                             }
                            \hspace*{-16ex}
                      }
                      {
                        \declspinejudg{\Psi}{\Thespine}{\fighi{\unitty} \arr \fighi{\unitty}}{\fighi{\OK}}{\unitty}{\OK}
                      }
                    \hspace*{-18ex}
                 }
                 {\declspinejudg{\Psi}{\Thespine}{(\alltype{\alpha:\type} \alpha \arr \alpha)}{\p}{\unitty}{\fighi{\OK}}}
              ~
              \hspace*{9.7ex}
              \arrayenvbl{
                 \text{for all $C'$.\;\;if}
                 \\
                     \text{~}
                     \SPLdeclspinejudg
                         {\Psi}
                         {\Thespine}
                         {(\alltype{\alpha{:}\type} \alpha{\arr}\alpha)}
                         {\p}
                         {C'}
                         {\OK}
                     \\
                     \text{~then~}C' = \unitty
               }
               \hspace*{-14ex}
            }
            {
              \declrecspinejudg{\Psi}{\Thespine}{(\alltype{\alpha:\type} \alpha \arr \alpha)}{\underbrace{\p}_{\!\!\!\text{input}\!\!\!}}{\unitty}{\fighi{\p}}
            }
\]
Here, we highlight the replacement in \DeclAllSpine of the quantified type variable
$\alpha$ by the ``guessed'' solution $\unitty$.  The second premise of \DeclRecover
checks that no other output type $C'$ could have been produced, no matter what solution
was chosen by \DeclAllSpine for $\alpha$.

\mypara{Syntax}
Expressions are the same as in the declarative system.

\input{fig-alg-syntax.tex}

\mypara{Existential variables}
The algorithmic system adds existential variables $\ahat$, $\bhat$, $\chat$
to types and terms/monotypes (\Figureref{fig:alg-syntax}).
We use the same meta-variables $A$, $\dots$.   %
We write $u$ for either a universal variable $\alpha$
or an existential variable $\ahat$.

\mypara{Contexts}
An algorithmic context $\Gamma$ is a sequence that,
like a declarative context, may
contain universal variable declarations $\alpha : \sort$ and expression variable
typings $\hyp p x A$.  However, it may also have
(1) \emph{unsolved} existential variable declarations
$\ahat : \sort$ (included in the $\Gamma, u : \sort$ production);
(2) \emph{solved} existential variable declarations
$\hypeq{\ahat:\sort}{\tau}$;
(3) \emph{equations} over universal variables $\hypeq{\alpha}{\tau}$;
and
(4)
\emph{markers} $\MonnierComma{u}$.  
An equation $\hypeq{\alpha}{\tau}$ must appear to the right
of the universal variable's declaration $\alpha : \sort$.
We use markers as delimiters within contexts.  For example, rule 
\ImpliesIntro adds $\MonnierComma{P}$, which tells it how much
of its last premise's output context ($\Delta, \MonnierComma{P}, \Delta'$)
should be dropped.
(We abuse notation by writing $\MonnierComma{P}$ rather than
cluttering the context with a dummy $\alpha$ and writing
$\MonnierComma{\alpha}$.)

A complete algorithmic context, denoted by $\Omega$, is
an algorithmic context with no unsolved existential variable
declarations.

Assuming an equality can yield inconsistency: for example,
$\zero = \succ{\zero}$.  We write $\Deltabot$ for either a valid
algorithmic context $\Delta$ or inconsistency $\bot$.

\subsection{Context Substitution $[\Gamma]A$
  and Hole Notation $\Gamma[\Theta]$}

\input{fig-substitution.tex}

An algorithmic context can be viewed as a substitution for its solved
existential variables.  For example,
$\hypeq{\ahat}{\unitty}, \hypeq{\bhat}{\ahat{\arr}\unitty}$ 
can be applied as if it were the substitution
$\unitty/\ahat, (\ahat{\arr}\unitty) / \bhat$
(applied right to left), or the simultaneous substitution
$\unitty/\ahat, (\unitty{\arr}\unitty) / \bhat$.
We write $[\Gamma]A$ for $\Gamma$ applied as a substitution
(\Figureref{fig:substitution}).

Applying a complete context to a type
$A$ (provided it is well-formed: $\judgetp{\Omega}{\!A}$)
yields a type $[\Omega]A$ with no existentials.  
Such a type is well-formed under the \emph{declarative} context
obtained by dropping all the existential declarations and applying
$\Omega$ to declarations $x : A$ (to yield $x : [\Omega]A$).
We can think of this context as the result of applying $\Omega$ to itself:
$[\Omega]\Omega$.
More generally, we can apply $\Omega$ to any
context $\Gamma$ that it extends:
context application $[\Omega]\Gamma$
is given in \Figureref{fig:context-substitution}.
The application $[\Omega]\Gamma$ is defined if and only if
$\substextend{\Gamma}{\Omega}$
(context extension; see \Sectionref{sec:context-extension}),
and applying $\Omega$ to any such $\Gamma$
yields the same declarative context $[\Omega]\Omega$.

\input{fig-context-substitution.tex}

In addition to appending declarations (as in the declarative system),
we sometimes insert and replace declarations, so
a notation for contexts with a hole is useful:
$\Gamma = \Gamma_0[\Theta]$ means $\Gamma$ has the form $(\Gamma_L, \Theta, \Gamma_R)$.
For example, if $\Gamma = \Gamma_0[\bhat] = (\ahat, \bhat, x : \bhat)$,
then $\Gamma_0[\hypeq{\bhat}{\ahat}] = (\ahat, \hypeq{\bhat}{\ahat}, x : \bhat)$.

We also use contexts with \emph{two} ordered holes:
if $\Gamma = \Gamma_0[\Theta_1][\Theta_2]$
then
$\Gamma = (\Gamma_L, \Theta_1, \Gamma_M, \Theta_2, \Gamma_R)$.

\subsection{The Context Extension Relation $\substextend{\Gamma}{\Delta}$}
\Label{sec:context-extension}

\input{fig-alg-typing.tex}

A context $\Gamma$ \emph{is extended by} a context $\Delta$,
written $\substextend{\Gamma}{\Delta}$, if $\Delta$ has at least
as much information as $\Gamma$, while conforming to the same
declarative context---that is, $[\Omega]\Gamma = [\Omega]\Delta$
for some $\Omega$.
In a sense, $\substextend{\Gamma}{\Delta}$ says
that $\Gamma$ is \emph{entailed by} $\Delta$: all positive information derivable
from $\Gamma$ %
can also be derived from $\Delta$ (which may have more information,
say, that $\ahat$ is equal to a particular type). We give the rules for
extension in \Figureref{fig:substextend}.

\input{fig-substextend.tex}

The rules deriving the context extension judgment (\Figureref{fig:substextend}) say that
the empty context extends the empty context (\substextendId);
a term variable typing with $A'$ extends one with $A$ if applying
the extending context $\Delta$ to $A$ and $A'$ yields the same type (\substextendVV);
universal variable declarations and equations must match (\substextendUU, \substextendEqn);
scope markers must match (\substextendMonMon);
and, existential variables may either match (\substextendEE, \substextendSolved),
get solved by the extending context (\substextendSolve),
or be added by the extending context (\substextendAdd, \substextendAddSolved).

Extension may change solutions, if information is preserved or increased:
$
   \substextend{(\ahat:\type, \hypeq{\bhat:\type}{\ahat})}{(\hypeq{\ahat:\type}{\unitty}, \hypeq{\bhat:\type}{\ahat})}
$
directly increases information about $\ahat$, and indirectly increases information
about $\bhat$.  More interestingly,
if $\Delta = (\hypeq{\ahat{:}\type}{\unitty}, \hypeq{\bhat{:}\type}{\ahat})$
and
$\Omega = (\hypeq{\ahat{:}\type}{\unitty}, \hypeq{\bhat{:}\type}{\unitty})$,
then $\substextend{\Delta}{\Omega}$:
while the solution of $\bhat$ in $\Omega$ is different, in the sense that
$\Omega$ contains $\hypeq{\bhat:\type}{\unitty}$
while $\Delta$ contains $\hypeq{\bhat:\type}{\ahat}$,
applying $\Omega$ to the solutions gives the same result:
$[\Omega]\ahat = [\Omega]\unitty = \unitty$, 
the same as $[\Omega]\unitty = \unitty$.

Extension is quite rigid, however, in two senses.  First, if a declaration
appears in $\Gamma$, it appears in all extensions of $\Gamma$.  Second,
\emph{extension preserves order}.  For example,
if $\bhat$ is declared after $\ahat$ in $\Gamma$, then $\bhat$ will also be declared after
$\ahat$ in every extension of $\Gamma$.  This holds for every variety of declaration,
including equations of universal variables.
This rigidity aids in enforcing type variable scoping and dependencies, which are
nontrivial in a setting with higher-rank polymorphism.

\subsection{Determinacy}
\Label{sec:alg-determinacy}

Given appropriate inputs ($\Gamma$, $e$, $A$, $p$)
to the algorithmic judgments,
only one set of outputs ($C$, $q$, $\Delta$)
is derivable
(\Theoremref{fancy:thm:typing-det} in the supplementary material, p.~\pageref{fancy:thm:typing-det}).
We use this property (for spine judgments) in the proof of soundness.

%% file: confgrue.tex
\begin{figure}
\newcommand{\Room}[2]{
  \!\!\!\begin{tabular}[t]{l}
    #1
    \\{}%
    \ensuremath{#2}%
  \end{tabular}\!\!\!%
}
\small
    $~$\!\!\!\!\!\begin{tikzpicture}
    \renewcommand{\CompactJudgments}{1}
      [auto, node distance=2cm, >=latex,  %
       descr/.style= {fill=white, inner sep=4pt, anchor=center}
      ]

      \node (sub) {\Room{subtyping}{\subjudg[\polvar]{\Gamma}{A}{B}{\Delta}}};

      \node [above of=sub, node distance=72pt] (inst) {\Room{instantiation}{\instjudg{\Gamma}{\ahat}{t}{\sort}{\Delta}}};

      \node [below of=sub, node distance=38pt] (chk) {\Room{type checking}{\chkjudg{p}{\Gamma}{e}{A}{\Delta}}};

      \node [below of=inst, node distance=36pt] (equiv) {\Room{equiv.\ types}{\equivjudg{\Gamma}{A}{B}{\Delta}}};

      \node [below of=chk, left=-12pt, node distance=45pt] (syn) {\Room{type synthesis}{\synjudg{p}{\Gamma}{e}{B}{\Delta}}};

      \node [left of=inst, node distance=100pt] (checkeq) {\Room{check equation}{\checkeq{\Gamma}{t_1}{t_2}{\sort}{\Delta}}};
      \draw [->] (checkeq) -- (inst);
      \node [left of=equiv, node distance=80pt] (propequiv) {\Room{equiv.\ props.}{\propequivjudg{\Gamma}{P}{Q}{\Delta}}};
      \draw [->] (propequiv) -- (checkeq);
      \node [left of=sub, node distance=110pt] (checkprop) {\Room{check prop.}{\checkprop{\Gamma}{P}{\Delta}}};
      \path [->, bend left] (checkprop) edge (checkeq);
      \draw [->] (equiv) -- (propequiv);
      \draw [->] (equiv) -- (inst);
      \draw [->] (sub) -- (equiv);

      \draw [->] (chk) -- (sub);

      \node [left of=chk, node distance=100pt] (spine) {\small \Room{spine typing}{\spinejudg{\Gamma}{s}{A}{p}{B}{q}{\Delta}}};

      \node [below of=spine, left=20pt, node distance=38pt] (recspine) {\small\Room{\begin{tabular}{l} principality-recovering \\ spine typing \end{tabular}}{\recspinejudg{\Gamma}{s}{A}{p}{B}{q}{\Delta}}\!\!\!\!};

      \draw [->] (spine) -- (checkprop);
      \path [->, %
      line width=1pt, color=dRed] (syn) edge (recspine);
      \draw [->, line width=1pt, color=dRed] (recspine) -- (spine);
      \draw [<->, line width=1pt, color=dRed] (chk) -- (syn);
      \draw [->] (chk) -- (checkprop);
      \draw [->, line width=1pt, color=dRed] (spine) -- (chk);
      \node [below of=chk, right=60pt, node distance=40pt] (match) {\Room{pattern matching}{\matchjudg{q}{p}{\Gamma}{\Pi}{\Avec}{C}{\Delta}}};
      \node [right of=chk, right=20pt, node distance=71pt] (elimeq) {\Room{equality elim.}{\elimeq{\Gamma}{s\!}{\!t\!}{\!\sort}{\Deltabot}}};

      \node [right of=sub, below=1pt, right=25pt, node distance=15pt] (covers) {\Room{coverage}{\covers{q}{\Gamma}{\Pi}{\Avec}}};
      \draw [->] (chk) -- (covers);
      \draw [->] (covers) -- (elimeq);
      \draw [<->, line width=1pt, color=dRed] (match) -- (chk);
      \draw [->] (match) -- (elimeq);

    \end{tikzpicture}
\vspace*{-1.0ex}

\caption{Dependency structure of the algorithmic judgments}
\FLabel{fig:grue}
\end{figure}

%% file: fig-alg-syntax.tex
\begin{figure}[t]

   \begin{bnfarray}
    \text{Universal variables}\hspace{-1ex} & \alpha, \beta, \gamma &
    \\[1pt]
    \text{Existential variables}\hspace{-1ex} & \ahat, \bhat, \chat &
    \\[1pt]
    \text{Variables} & u & \bnfas & \alpha \bnfalt \ahat
    \\[1pt]
    \text{Types} & A, B, C\! & \bnfas &
          \unitty
          \bnfalt  A \arr B
          \bnfalt  A + B
          \bnfalt  A \times B
          \bnfalt  \alpha
          \bnfalt  \fighi{\ahat}
          \bnfalt  \alltype{\alpha:\sort}{A}
          \bnfalt  \extype{\alpha:\sort}{A}
    \\ &&&\!\!\!
          \bnfalt  P \implies A
          \bnfalt  A \with P
          \bnfalt  \vectype{t}{A}
    \\[1pt]
    \text{Propositions} & P,Q & \bnfas & t = t'      
    \\[1pt]
    \text{Binary connectives}\hspace{-2ex} & \binc & \bnfas &
           {\arr}
           \bnfalt  {+}
           \bnfalt  {\times}
    \\[1pt]
      \text{Terms/monotypes} & t, \tau, \sigma & \bnfas &
            \zero
            \bnfalt  \succ{t}
            \bnfalt  \unitty
            \bnfalt  \alpha
            \bnfalt  \ahat
            \bnfalt  \tau \arr \sigma
            \bnfalt  \tau + \sigma
            \bnfalt  \tau \times \sigma
\\[1pt]
      \text{Contexts} & \Gamma, \Delta, \Theta & \bnfas &
                  \cdot
                  \bnfalt  \Gamma, u:\sort
              \;\bnfalt  \Gamma, \hyp{p}{x}{A}
              \,\bnfalt  \Gamma, \hypeq{\ahat:\sort}{\tau}
              \,\bnfalt  \Gamma, \hypeq{\alpha}{t}
              \,\bnfalt  \Gamma, \MonnierComma{u}
      \\[1pt]
      \text{Complete contexts}\hspace{-5ex}     & \Omega & \bnfas &
                  \cdot
                  \bnfalt    \Omega, \alpha:\sort
                  \bnfalt    \Omega, \hyp{p}{x}{A}
                  \bnfalt    \Omega, \hypeq{\ahat:\sort}{\tau} 
                  \bnfalt    \Omega, \hypeq{\alpha}{t}
                  \bnfalt    \Omega, \MonnierComma{u}
   \vspace{-5pt}
   \end{bnfarray}

   \smallskip

   \begin{bnfarray}
      \text{Possibly inconsistent contexts} & \Deltabot & \bnfas &
                 \Delta
                 \bnfalt \bot
   \vspace*{-0.3ex}
   \end{bnfarray}
   
   \caption{Syntax of types, contexts, and other objects in the algorithmic system}
   \FLabel{fig:alg-syntax}
\end{figure}

%% file: fig-substitution.tex
\begin{figure}[t]

  \begin{minipage}{0.45\linewidth}
  \begin{array}[t]{@{}l@{~~}c@{~~}ll@{}}
      ~\\[-9pt]
      \multicolumn{3}{l}{
      [\Gamma] \alpha
               ~~ = ~~   \left\{\begin{array}{@{}l@{~}l@{}}
                                [\Gamma]\tau  & \mbox{when }(\hypeq{\alpha}{\tau}) \in \Gamma \\[1pt]
                                \alpha        & \mbox{otherwise}
                              \end{array}
                       \right.
      }
     \\[8pt]
      [\Gamma](P \implies A)   & = &
          ([\Gamma]P) \implies ([\Gamma]A) &
      \\[0pt]
      [\Gamma](A \with P)   & = &
          ([\Gamma]A) \with ([\Gamma]P) &
      \\[0pt]
      [\Gamma](A \binc B)   & = &
          ([\Gamma]A) \binc ([\Gamma]B) &
      \\[0pt]
      [\Gamma](\vectype{t}{A})   & = &
          \vectype{([\Gamma]t)}{([\Gamma]A)} &
  \end{array}
  \end{minipage}
  ~
  \begin{minipage}{0.45\linewidth}
  \begin{array}[t]{@{}l@{~~}c@{~~}ll@{}}
      \!\big[\Gamma[\hypeq{\ahat:\sort}{\tau}]\big] \ahat
               & = &   [\Gamma]\tau 
      \\[2pt]
      \!\big[\Gamma[\ahat:\sort]\big]\ahat   & = &   \ahat &
      \\[8pt]
      [\Gamma](\alltype{\alpha:\sort} A)
         & = & 
         \alltype{\alpha:\sort} [\Gamma]A & 
      \\[1pt]{}
      [\Gamma](\extype{\alpha:\sort} A)
         & = & 
         \extype{\alpha:\sort} [\Gamma]A & 
      \\[3pt]
      [\Gamma](t_1 = t_2) & = & ([\Gamma]t_1) = ([\Gamma]t_2)
  \end{array}
  \end{minipage}

  \vspace*{-0.3ex}

  \caption{Applying a context, as a substitution, to a type}
  \FLabel{fig:substitution}
\end{figure}

%% file: fig-context-substitution.tex
\begin{figure}
\begin{displ}
    \begin{array}[t]{l@{~~}c@{~~}l@{~~}l}
     [\cdot]\cdot & = &   \cdot &              %
    \\[1pt]{}
     [\Omega, \hyp{p}{x}{A}](\Gamma, \hyp{p}{x}{A_\Gamma}) & = &   [\Omega]{\Gamma},\; \hyp{p}{x}{[\Omega]A} & 
       \mbox{if } [\Omega]A = [\Omega]A_\Gamma
    \\[1pt]{}
      [\Omega, \alpha:\sort](\Gamma, \alpha:\sort) & = &   [\Omega]\Gamma,\; \alpha:\sort &
    \\[1pt]{}
    [\Omega, \MonnierComma{u}](\Gamma, \MonnierComma{u}) & = &   [\Omega]\Gamma &
    \\[1pt]{}
    [\Omega, \hypeq{\alpha}{t}](\Gamma, \hypeq{\alpha}{t'}) & = &  \big[[\Omega]t/\alpha\big] [\Omega] \Gamma & 
    \mbox{if } [\Omega]t = [\Omega]t'
    \\[3pt]
    [\Omega, \hypeq{\ahat:\sort}{t}]\Gamma & = &  
    \multicolumn{2}{l}{%
      \!\!\!\left\{
        \begin{array}{@{}l@{~}l}
          [\Omega]\Gamma' & \text{when~} \Gamma = (\Gamma', \hypeq{\ahat:\sort}{t'}) \\ {}
          [\Omega]\Gamma' & \text{when~} \Gamma = (\Gamma', \ahat:\sort) \\ {}
          [\Omega]\Gamma  & \text{otherwise}
        \end{array}
      \right.
    }
    \end{array}
\end{displ}
\vspace{-1ex}
\caption{Applying a complete context $\Omega$ to a context}
\FLabel{fig:context-substitution}
\end{figure}

%% file: fig-alg-typing.tex
\begin{figure*}[htbp]
\raggedright
\ifnum\OPTIONInAppendix=1
\runonfontsz{8.7pt}
\fi
  \judgbox{
    \arrayenvcl{
    \chkjudg{p}{\Gamma}{e}{A}{\Delta}
    \\
    \synjudg{p}{\Gamma}{e}{A}{\Delta}
  }}%
     {Under input context $\Gamma$, expression $e$ checks against input type $A$, \\
       with output context $\Delta$ \\[2pt]
      Under input context $\Gamma$, expression $e$ synthesizes output type $A$, \\
       with output context $\Delta$}%
  \begin{mathpar}
     \Infer{\!\Var}
          {(\hyp{p}{x}{A}) \in \Gamma}
          {\synjudg{p}{\Gamma}{x}{[\Gamma]A}{\Gamma}}
     \and
     \Infer{\Sub}
          {
              \synjudg{q}{\Gamma}{e}{A}{\Theta}
              \\
              \subjudg[\joinpolarity{\polarity{B}}{\polarity{A}}]{\Theta}{A}{B}{\Delta}
          }
          {\chkjudg{p}{\Gamma}{e}{B}{\Delta}}
     \\
     \Infer{\Anno}
          {
              \judgetp{\Gamma}{A\,\p}
              \\
              \chkjudg{\p}{\Gamma}{e}{[\Gamma]A}{\Delta}
          }
          {\synjudg{\p}{\Gamma}{(e : A)}{[\Delta]A}{\Delta}}
    \and
     \Infer{\!\Rec}
          {
            \chkjudg{p}{\Gamma, \hyp{p}{x}{A}}{v}{A}{\Delta, \hyp{p}{x}{A}, \Theta}
          }
          { \chkjudg{p}{\Gamma}{\rec{x} v}{A}{\Delta} }
       \vspace{-0.8em}
     \\
     \Infer{\UnitIntro}
           {}
           {\chkjudg{p}{\Gamma}{\unitexp}{\unitty}{\Gamma}}
     \and
     \Infer{\UnitIntroSolve}
            { }
            {\chkjudg{\ok}{\Gamma[\ahat:\type]}{\unitexp}{\ahat}{\Gamma[\hypeq{\ahat:\type}{\unitty}]}}
     \\
     \Infer{\AllIntro}
           {\chkintro{v}   %
            \\
            \chkjudg{p}{\Gamma, \alpha:\sort}{v}{A}{\Delta, \alpha:\sort, \Theta}
           }
           {\chkjudg{p}{\Gamma}{v}{\alltype{\alpha:\sort} A}{\Delta}}
     ~~~~~~
     \Infer{\ExIntro}
           {\chkintro{e} \\     %
            \chkjudg{\ok}{\Gamma, \ahat:\sort}{e}{[\ahat/\alpha]A}{\Delta}
           }
           {\chkjudg{p}{\Gamma}{e}{\extype{\alpha:\sort} A}{\Delta}}
     \and
     \Infer{\ImpliesIntro}
           {
             \arrayenvbl{
               \chkintro{v}
               \\
               \elimprop{\Gamma, \MonnierComma{P}}{P}{\Theta}
               \\
               \chkjudg{\p}{\Theta}{v}{[\Theta]A}{\Delta, \MonnierComma{P}, \Delta'}
             }
            }
            {\chkjudg{\p}{\Gamma}{v}{P \implies A}{\Delta}}
     ~~~
     \Infer{\ImpliesIntroBot}
           {
              \chkintro{v}  \\     %
              \elimprop{\Gamma, \MonnierComma{P}}{P}{\bot}
            }
            {\chkjudg{\p}{\Gamma}{v}{P \implies A}{\Gamma}}
     ~~~
     \Infer{\WithIntro}
           {
            \arrayenvbl{
              \notcase{e} \\      %
              \checkprop{\Gamma}{P}{\Theta}
              \\
              \chkjudg{p}{\Theta}{e}{[\Theta]A}{\Delta}
            }
           }
           {\chkjudg{p}{\Gamma}{e}{A \with P}{\Delta}}
\\
     \Infer{\!\ArrIntro}
          {\chkjudg{p}{\Gamma, \hyp{p}{x}{A}}{e}{B}{\Delta, \hyp{p}{x}{A}, \Theta}
          }
          {\chkjudg{p}{\Gamma}{\lam{x} e}{A \arr B}{\Delta}}
   ~~
  \mathsz{9pt}{
     \Infer{\!\ArrIntroSolve}
           {\chkjudg{\ok}{%
                 \Gamma[\ahat_1{:}\type, \ahat_2{:}\type, \hypeq{\ahat{:}\type}{\ahat_1{\arr}\ahat_2}], \hyp{\ok}{x}{\ahat_1}}
              {e}
              {\ahat_2}
              {\Delta, \hyp{\ok}{x}{\ahat_1}, \Delta'}
           }
           {\chkjudg{\ok}{\Gamma[\ahat : \type]}{\lam{x} e}{\ahat}{\Delta}}
  }
  \\
   \Infer{\ArrElim}
           {
             \arrayenvbl{
               \synjudg{p}{\Gamma}{e}{A}{\Theta}
               \\
               \recspinejudg{\Theta}{s}{A}{p}{C}{q}{\Delta}
             }
           }
           {\synjudg{q}{\Gamma}{e\,s}{C}{\Delta}}
   \and
      \Infer{\!\Case}
           {
             \synjudg{q}{\Gamma}{e}{A}{\Theta}
             \\
             \arrayenvbl{
                \matchjudg{q}{p}{\Theta}{\Pi}{[\Theta]A}{[\Theta]C}{\Delta}
                \\
                \covers{q}{\Delta}{\Pi}{[\Delta]A}
             }
           }
           {\chkjudg{p}{\Gamma}{\case{e}{\Pi}}{C}{\Delta}}
\\
\ifnum\OPTIONInAppendix=1%
     \Infer{\SumIntro{k}}
           {\chkjudg{p}{\Gamma}{e}{A_k}{\Delta}}
           {\chkjudg{p}{\Gamma}{\inj{k} e}{A_1 + A_2}{\Delta}}
     \and
     \Infer{\SumIntroSolve{k}}
            {\chkjudg{\ok}{\Gamma[\ahat_1 : \type, \ahat_2 : \type, \hypeq{\ahat:\type}{\ahat_1{+}\ahat_2}]}{e}{\ahat_k}{\Delta}}
            {\chkjudg{\ok}{\Gamma[\ahat : \type]}{\inj{k} e}{\ahat}{\Delta}}
      \\
     \Infer{\PairIntro}
           {\chkjudg{p}{\Gamma}{e_1}{A_1}{\Theta}
             \\
             \chkjudg{p}{\Theta}{e_2}{[\Theta]A_2}{\Delta}
           }
           {\chkjudg{p}{\Gamma}{\pair{e_1}{e_2}}{A_1 \times A_2}{\Delta}}
~~~~~
     \Infer{\PairIntroSolve}
           {\arrayenvbl{
               \chkjudg{\!\!\ok}{\Gamma[\ahat_2{:}\type, \ahat_1{:}\type, \hypeq{\ahat{:}\type}{\ahat_1{\times}\ahat_2}]}{e_1}{\ahat_1}{\Theta}
               \\
               \chkjudg{\!\!\ok}{\Theta}{e_2}{[\Theta]\ahat_2}{\Delta}
             }
           }
           {\chkjudg{\ok}{\Gamma[\ahat : \type]}{\pair{e_1}{e_2}}{\ahat}{\Delta}}
      \\
      \Infer{\Nil}
            {\checkprop{\Gamma}{t = \zero}{\Delta}}
            {\chkjudg{p}{\Gamma}{\vecnil}{(\vectype{t}{A})}{\Delta}}
     ~~~~
      \Infer{\Cons}
            {
              \arrayenvbl{
              \checkprop{\Gamma,
                        \MonnierComma{\ahat},
                        \ahat : \kindnat}
                       {t = \succ{\ahat}}
                       {\Gamma'}
              \\
              \chkjudg{p}{\Gamma'}{e_1}{[\Gamma']A}{\Theta}
              \\
              \chkjudg{\OK}
                      {\Theta}
                      {e_2}
                      {[\Theta](\vectype{\ahat}{A})}
                      {\Delta, \MonnierComma{\ahat}, \Delta'}
              }
            }
            {
              \chkjudg{p}{\Gamma}{\veccons{e_1}{e_2}}{(\vectype{t}{A})}{\Delta}
            }
\fi  %
\end{mathpar}

  \judgbox{\arrayenvcl{
      \spinejudg{\Gamma}{s}{A}{p}{C}{q}{\Delta}
      \\
          \recspinejudg{\Gamma}{s}{A}{p}{C}{q}{\Delta}
     }}%
     {Under input context $\Gamma$, \\
       passing spine $s$ to a function of type $A$ synthesizes type $C$; \\
       in the $\symrecspine{q}$ form, recover principality in $q$ if possible}
\begin{mathpar}
     \Infer{\!\AllSpine}
         {\spinejudg{\Gamma, \ahat:\sort}{\appspine{e}{s}}{[\ahat/\alpha]A}{\ok}{C}{q}{\Delta}}
         {\spinejudg{\Gamma}{\appspine{e}{s}}{\alltype{\alpha:\sort} A}{p}{C}{q}{\Delta}}
    ~~
     \Infer{\!\ImpliesSpine}
         {\checkprop{\Gamma}{P}{\Theta}
          \\
          \spinejudg{\Theta}{\appspine{e}{s}}{[\Theta]A}{p}{C}{q}{\Delta}}
         {\spinejudg{\Gamma}{\appspine{e}{s}}{P \implies A}{p}{C}{q}{\Delta}}
  \\
     \Infer{\!\EmptySpine}
         {}
         {\spinejudg{\Gamma}{\emptyspine}{A}{p}{A}{p}{\Gamma}}
     ~~~~~
     \Infer{\!\ArrSpine}
         {
             \chkjudg{p}{\Gamma}{e}{A}{\Theta}
             \\
             \spinejudg{\Theta}{s}{[\Theta]B}{p}{C}{q}{\Delta}
s%
         }
         {\spinejudg{\Gamma}{\appspine{e}{s}}{A \arr B}{p}{C}{q}{\Delta}}
\ifnum\OPTIONInAppendix=0%
\fi
  \\
\def\CompactJudgments{0}
      \Infer{\!\SolveSpine}
            {
              \spinejudg{\Gamma[\ahat_2{:}\type, \ahat_1{:}\type, \hypeq{\ahat{:}\type}{\ahat_1{\arr}\ahat_2}]}{\appspine{e}{s}}{(\ahat_1\,{\arr}\,\ahat_2)}{\ok}{C}{\ok}{\Delta}
            }
            {\spinejudg{\Gamma[\ahat:\type]}{\appspine{e}{s}}{\ahat}{\ok}{C}{\ok}{\Delta}}
     \\
     \Infer{\Recover}
           {
             \arrayenvbl{
                 \spinejudg{\Gamma}{s}{A}{\p}{C}{\OK}{\Delta}
                 \\
                 \FEV{C} = \emptyset
             }
           }
           {\recspinejudg{\Gamma}{s}{A}{\p}{C}{\p}{\Delta}}
    ~~~~~
     \Infer{\Pass}
           {
             \arrayenvbl{
                 \spinejudg{\Gamma}{s}{A}{p}{C}{q}{\Delta}
                 \\
                 \big(
                     (p = \OK)
                     \OR (q = \p)
                     \OR (\FEV{C} \neq \emptyset)
                 \big)
             }
           }
           {\recspinejudg{\Gamma}{s}{A}{p}{C}{q}{\Delta}}
  \vspace*{-1.3ex}
  \end{mathpar}

\ifnum\OPTIONInAppendix=0
  \caption{Algorithmic typing, omitting rules for $\times$, $+$, and $\vecop$}
  \FLabel{fig:alg-typing}
\else
  \repeatcaption{mainfig:alg-typing}{a}{Algorithmic typing, including rules omitted from main paper}
\fi

\end{figure*}

%% file: fig-substextend.tex
\begin{figure}[t]
\raggedright

  \loudjudgbox{substextend}{\substextend{\Gamma}{\Delta}}{%
          $\Gamma$ is extended by $\Delta$%
  }
  \vspace*{-1ex}
  \begin{mathpar}
    \Infer{\substextendId}
        { }
        { \substextend{\emptyctx}{\emptyctx} }
    ~~~~
    \Infer{\substextendVV}
          {\substextend{\Gamma}{\Delta}
           \\
           [\Delta]A = [\Delta]A'
          }
          {\substextend{\Gamma, \hyp{p}{x}{A}}{\Delta, \hyp{p}{x}{A'}}}
    ~~~~
    \Infer{\substextendUU}
        { \substextend{\Gamma}{\Delta} }
        { \substextend{\Gamma, \alpha:\sort}{\Delta, \alpha:\sort} }
\ifnum\OPTIONInAppendix=0
    \and
\else
   \and
\fi
    \Infer{\!\substextendEqn}
        { \substextend{\Gamma}{\Delta}
          \\
          [\Delta]t = [\Delta]t'
        }
        { \substextend{\Gamma, \hypeq{\alpha}{t}}{\Delta, \hypeq{\alpha}{t'}}
        }
    ~~~
    \Infer{\!\substextendEE}
        { \substextend{\Gamma}{\Delta} }
        { \substextend{\Gamma, \ahat:\sort}{\Delta, \ahat:\sort} }
    ~~~
    \Infer{\!\substextendMonMon}
        { \substextend{\Gamma}{\Delta} }
        { \substextend{\Gamma, \MonnierComma{u}}{\Delta, \MonnierComma{u}} }
    \and
    \Infer{\substextendSolved}  %
        { \substextend{\Gamma}{\Delta}
          \\
          [\Delta]t = [\Delta]t'
        }
        { \substextend{\Gamma, \hypeq{\ahat:\sort}{t}}{\Delta, \hypeq{\ahat:\sort}{t'}}
        }
    \and
    \Infer{\substextendSolve}
        { \substextend{\Gamma}{\Delta} }
        { \substextend{\Gamma, \bhat:\sort'}{\Delta, \hypeq{\bhat:\sort'}{t}} }
   \and
    \Infer{\substextendAdd}
        { \substextend{\Gamma}{\Delta} }
        { \substextend{\Gamma}{\Delta, \ahat:\sort} }
    \and
    \Infer{\substextendAddSolved}
        { \substextend{\Gamma}{\Delta} }
        { \substextend{\Gamma}{\Delta, \hypeq{\ahat:\sort}{t}} }
  \vspace*{-1.3ex}
  \end{mathpar}
  
  \caption{Context extension}
  \FLabel{fig:substextend}
\label{figLAST}   %
\end{figure}

%% file: soundness.tex
\section{Soundness}
\Label{sec:soundness}

We show that the algorithmic system is sound with respect to the declarative system.
Soundness for the mutually recursive judgments depends on lemmas for the auxiliary
judgments (instantiation, equality elimination, checkprop, algorithmic subtyping and
match coverage), which are in \mbox{\Appendixref{fancy:apx:soundness}} for space reasons.
The main soundness result has mutually recursive parts for
checking, synthesis, spines and matching---including the
principality-recovering spine judgment.

\setcounter{theorem}{7}
\begin{theorem}[Soundness of Algorithmic Typing]  %
Given $\substextend{\Delta}{\Omega}$: \\[-2.5ex]
    \begin{enumerate}[(i)]
      \item %
          If $\chkjudg{p}{\Gamma}{e}{A}{\Delta}$
          and $\judgetp{\Gamma}{A\;p}$
          then
          $\declchkjudg{p}{[\Omega]\Delta}{[\Omega]e}{[\Omega]A}$. 
      
      \item %
          If $\synjudg{p}{\Gamma}{e}{A}{\Delta}$ 
          then
          $\declsynjudg{p}{[\Omega]\Delta}{[\Omega]e}{[\Omega]A}$. 

      \item %
          If $\spinejudg{\Gamma}{s}{A}{p}{B}{q}{\Delta}$
          and $\judgetp{\Gamma}{A\;p}$
          then 
          $\declspinejudg{[\Omega]\Delta}{[\Omega]s}{[\Omega]A}{p}{[\Omega]B}{q}$.

      \item %
          If $\recspinejudg{\Gamma}{s}{A}{p}{B}{q}{\Delta}$
          and $\judgetp{\Gamma}{A\;p}$
          then 
          $\declrecspinejudg{[\Omega]\Delta}{[\Omega]s}{[\Omega]A}{p}{[\Omega]B}{q}$.

      \item %
          If $\matchjudg{q}{p}{\Gamma}{\Pi}{\Avec}{C}{\Delta}$
          and $\judgetpvec{\Gamma}{\Avec\;q}$
          and $[\Gamma]\Avec = \Avec$
          and $\judgetp{\Gamma}{C\;p}$
          \\
          then
          $\declmatchjudg[q]{p}{[\Omega]\Delta}{[\Omega]\Pi}{[\Omega]\Avec}{[\Omega]C}$. 

      \item %
          If $\matchelimjudg{p}{\Gamma}{\Pi}{P}{\Avec}{C}{\Delta}$
          and $\judgeprop{\Gamma}{P}$
          and $\FEV{P} = \emptyset$
          and $[\Gamma]P = P$
          \\
          and $\judgetpvec{\Gamma}{\Avec\;\p}$
          and $\judgetp{\Gamma}{C\;p}$
          then
          $\declmatchelimjudg{p}{[\Omega]\Delta}{[\Omega]\Pi}{[\Omega]P}{[\Omega]\Avec}{[\Omega]C}$. 
    \end{enumerate}
\end{theorem}

Much of this proof ``turns the crank'': apply the induction hypothesis
to each premise, yielding derivations of corresponding declarative judgments
(with $\Omega$ applied everywhere), then apply the corresponding
declarative rule; for example, in the \Sub case we finish by applying \DeclSub.
However, in the \Recover case we finish by applying \DeclRecover, but
since \DeclRecover contains a premise that quantifies over all declarative derivations
of a certain form, we must appeal to completeness!  Consequently, soundness
and completeness are really one theorem.

These parts are mutually
recursive---later, we'll see that the \DeclRecover case of completeness
must appeal to soundness (to show that the algorithmic type has no free
existential variables).  We cannot induct on the given derivation alone,
because the derivations in the ``for all'' part of \DeclRecover
are not subderivations.  So we need a more involved induction measure
that can make the leaps between soundness and completeness:
lexicographic order with (1) the size of the subject term,
(2) the judgment form, with ordinary spine judgments considered
smaller than recovering spine judgments, and (3) the height of the derivation:
\vspace*{-0.1ex}
\[
      \left\langle
        e / s / \Pi,~~
        {\!\begin{array}[c]{c}
            \text{ordinary spine judgment}
          \\[-0.13ex]
          <
          \\[-0.06ex]
            \text{recovering spine judgment}
        \end{array}%
        \!},
        ~~
        \text{height}(\Dee)
      \right\rangle
\]
\mypara{Proof sketch---\Recover case}
By i.h., $\declspinejudg{[\Omega]\Gamma}{[\Omega]s}{[\Omega]A}{\p}{[\Omega]C}{q}$.
Our goal is to apply \DeclRecover, which requires that we show
that for \emph{all} $C'$ such that $\declspinejudg{[\Omega]\Theta}{s}{[\Omega]A}{\p}{C'}{\OK}$,
we have $C' = [\Omega]C$.  Suppose we have such a $C'$.
By completeness (Theorem \ref{fancy:thm:typing-completeness}),
$\spinejudg{\Gamma}{s}{[\Gamma]A}{\p}{C''}{q}{\Delta''}$
where $\substextend{\Delta''}{\Omega''}$.
We already have (as a subderivation) $\spinejudg{\Gamma}{s}{A}{\p}{C}{\OK}{\Delta}$,
so by determinacy, $C'' = C$ and $q = \OK$ and $\Delta'' = \Delta$.
With the help of lemmas about context application, we can show
$C' = [\Omega'']C'' = [\Omega'']C = [\Omega]C$.
(Using completeness is permitted since our
measure says a non-principality-restoring judgment is smaller.)

\subsection{Auxiliary Soundness}
For several auxiliary judgment forms, soundness is a matter of showing that,
given two algorithmic terms, their declarative versions are equal.  For example,
for the instantiation judgment we have:

\begin{lemma*}[Soundness of Instantiation]
~\\
    If $\instjudg{\Gamma}{\ahat}{\tau}{\sort}{\Delta}$
    and $\ahat \notin \FV{[\Gamma]\tau}$
    and $[\Gamma]\tau = \tau$
    and $\substextend{\Delta}{\Omega}$
    then
    $[\Omega]\ahat = [\Omega]\tau$.
\end{lemma*}

We have similar lemmas for
term equality ($\checkeq{\Gamma}{\sigma}{t}{\sort}{\Delta}$),
propositional equivalence ($\propequivjudg{\Gamma}{P}{Q}{\Delta}$)
and type equivalence ($\equivjudg{\Gamma}{A}{B}{\Delta}$).

Our eliminating judgments incorporate assumptions into the
context $\Gamma$.  We show that the algorithmic rules for
these judgments just append equations over universal variables:

\begin{lemma*}[Soundness of Equality Elimination]
  If
  $[\Gamma]\sigma = \sigma$
  and $[\Gamma]t = t$
  and $\judge{\Gamma}{\sigma}{\sort}$
  and $\judge{\Gamma}{t}{\sort}$
  and $\FEV{\sigma} \union \FEV{t} = \emptyset$,
  then:

  \begin{enumerate}[(1)]
    \item 
        If $\elimeq{\Gamma}{\sigma}{t}{\sort}{\Delta}$
        then
        $\Delta = (\Gamma, \Theta)$
        where $\Theta = (\hypeq{\alpha_1}{t_1}, \dots, \hypeq{\alpha_n}{t_n})$
        and
        for all $\Omega$ such that $\substextend{\Gamma}{\Omega}$
        and all $t'$ s.t.\ 
        $\judge{\Omega}{t'}{\sort'}$
            we have
            $[\Omega, \Theta]t' = [\theta][\Omega]t'$
            where $\theta = \mgu{\sigma}{t}$.

    \item  
        If $\elimeq{\Gamma}{\sigma}{t}{\sort}{\bot}$
        then %
        no most general unifier exists.
  \end{enumerate}
\end{lemma*}

The last lemmas for soundness move directly from an algorithmic judgment
to the corresponding declarative judgment.

\begin{lemma*}[Soundness of Checkprop]
  If $\checkprop{\Gamma}{P}{\Delta}$ and $\substextend{\Delta}{\Omega}$
  then
  $\declcheckprop{\Psi}{[\Omega]P}$.
\end{lemma*}

\begin{lemma*}[Soundness of Match Coverage]
~
\begin{enumerate}
\item If $\covers{q}{\Gamma}{\Pi}{\Avec}$
    and $\substextend{\Gamma}{\Omega}$
    and $\judgetpvec{\Gamma}{\Avec\;\p}$
    and $[\Gamma]\Avec = \Avec$
    then 
    $\covers{q}{[\Omega]\Gamma}{\Pi}{\Avec}$. 
\item If $\coverseq{\Gamma}{P}{\Pi}{\Avec}$
    and $\substextend{\Gamma}{\Omega}$
    and $\judgetpvec{\Gamma}{\Avec\;\p}$
    and $[\Gamma]\Avec = \Avec$
    and $[\Gamma]P = P$
    then 
    $\coverseq{[\Omega]\Gamma}{P}{\Pi}{\Avec}$. 
\end{enumerate}
\end{lemma*}

\begin{restatable}[Soundness of Algorithmic Subtyping]{theorem}{subtypingsoundness} 
\Label{thm:subtyping-soundness}
  If
  $[\Gamma]A = A$
  and $[\Gamma]B = B$
  and $\judgetp{\Gamma}{A}$
  and $\judgetp{\Gamma}{B}$
  and $\substextend{\Delta}{\Omega}$
  and $\subjudg[\polvar]{\Gamma}{A}{B}{\Delta}$ then 
    $\declsubjudg[\polvar]{[\Omega]\Delta}{[\Omega]A}{[\Omega]B}$.
\end{restatable}

%% file: completeness.tex
\section{Completeness}
\Label{sec:completeness}

We show that the algorithmic system is complete with respect to the declarative system.
As with soundness, we need to show completeness of the auxiliary algorithmic judgments.
We omit the full statements of these lemmas;
as an example, if $[\Omega]\ahat = [\Omega]\tau$ and $\ahat \notin \FV{\tau}$
then $\instjudg{\Gamma}{\ahat}{\tau}{\sort}{\Delta}$.

\subsection{Separation}

To show completeness, we will need to show that wherever the declarative
rule \DeclRecover is applied, we can apply the algorithmic rule \Recover.
Thus, we need to show that \emph{semantic} principality---that
no other type can be given---entails that a type has no free existential variables.

The principality-recovering rules are potentially applicable when we start
with a principal type $A\;\p$ but produce $C\;\OK$, with \DeclAllSpine
changing $\p$ to $\OK$.  Completeness (Thm.\ \ref{fancy:thm:typing-completeness})
will use the ``for all'' part of \DeclRecover, which
quantifies over all types produced by the spine rules under a given declarative
context $[\Omega]\Gamma$.
By i.h.\ we get an algorithmic
spine judgment $\spinejudg{\Gamma}{s}{A'}{\p}{C'}{\OK}{\Delta}$.
Since $A'$ is principal, unsolved existentials in $C'$ must have been introduced
within this derivation---they can't be in $\Gamma$ already.  Thus, we might
have $\spinejudg{\ahat:\type}{s}{A'}{\p}{\bhat}{\OK}{\ahat:\type, \bhat:\type}$
where a \DeclAllSpine subderivation introduced $\bhat$, but $\ahat$ can't
appear in $C'$.  We also can't equate $\ahat$ and $\bhat$ in $\Delta$, which
would be tantamount to $C' = \ahat$.  Knowing that unsolved existentials
in $C'$ are ``new'' and independent from those in $\Gamma$
means we can argue that, if there \emph{were} an unsolved existential in $C'$,
it would correspond to an unforced choice in a \DeclAllSpine subderivation,
invalidating the ``for all'' part of \DeclRecover.
Formalizing %
``must have been introduced'' requires several
definitions.

\begin{definition}[Separation]
\Label{def:separation}
  An algorithmic context $\Gamma$ is \emph{separable into} $\Gamma_L \sep \Gamma_R$
  if (1) $\Gamma = (\Gamma_L, \Gamma_R)$
  and (2) for all $(\hypeq{\ahat:\sort}{\tau}) \in \Gamma_R$
               it is the case that $\FEV{\tau} \subseteq \dom{\Gamma_R}$.
\end{definition}

If $\Gamma$ is separable into $\Gamma_L \sep \Gamma_R$, then $\Gamma_R$ is
self-contained in the sense that all existential variables declared in $\Gamma_R$
have solutions whose existential variables are themselves declared in $\Gamma_R$.
Every context $\Gamma$ is separable into $\cdot \sep \Gamma$ and into $\Gamma \sep \cdot$.

\begin{definition}[Separation-Preserving Ext.]
\Label{def:separation-extension}
  Separated context $\Gamma_L \sep \Gamma_R$ extends to $\Delta_L \sep \Gamma_R$,
  written 
  $
      \sepextend{\Gamma_L\!}{\!\Gamma_R}{\Delta_L\!}{\!\Delta_R}
  $,
  if~$\substextend{(\Gamma_L, \Gamma_R)}{(\Delta_L, \Delta_R)}$
  and $\dom{\Gamma_L}  \subseteq  \dom{\Delta_L}$
  and $\dom{\Gamma_R} \subseteq  \dom{\Delta_R}$.
\end{definition}

Separation-preserving extension says that variables from one side of $\sep$ haven't ``jumped''
to the other side.  Thus, $\Delta_L$ may add existential variables to $\Gamma_L$, and $\Delta_R$
may add existential variables to $\Gamma_R$, but no variable from $\Gamma_L$ ends up
in $\Delta_R$ and no variable from $\Gamma_R$ ends up in $\Delta_L$.
It is necessary to write $\sepextend{\Gamma_L}{\Gamma_R}{\Delta_L}{\Delta_R}$
rather than
$\substextend{(\Gamma_L \sep \Gamma_R)}{(\Delta_L \sep \Delta_R)}$,
because only $\sepextendsym$ includes the domain conditions.  For example,
$\substextend{(\ahat \sep \bhat)}{(\ahat, \hypeq{\bhat}{\ahat}) \sep \cdot}$,
but $\bhat$ has jumped to the left of $\sep$ in
the context $(\ahat, \hypeq{\bhat}{\ahat}) \sep \cdot$.

We prove many lemmas about separation, but use only one of them in the subsequent
development (in the \DeclRecover case of typing completeness), and then only
the part for spines.  It says that if we have a spine whose type $A$ mentions only
variables in $\Gamma_R$, then the output context $\Delta$ extends $\Gamma$
and preserves separation, and the output type $C$
mentions only variables in $\Delta_R$:

\begin{lemma*}[Separation---Main]
        If\; $\spinejudg{\Gamma_L{\sep}\Gamma_R}{\!s}{A}{p}{C}{q\!}{\Delta}$
        or  $\recspinejudg{\Gamma_L{\sep}\Gamma_R}{\!s}{A}{p}{C}{q}{\Delta}$
        and
        \\
        $\judgetp{\Gamma_L \sep \Gamma_R}{A\;p}$
        and $\exbasis{\Gamma_R}{A}$
        then
        $\Delta = (\Delta_L \sep \Delta_R)$
        and $\sepextend{\Gamma_L}{\Gamma_R}{\Delta_L}{\Delta_R}$
        and $\exbasis{\Delta_R}{C}$.
\end{lemma*}

\subsection{Completeness of Typing}

Like soundness, completeness has several mutually recursive parts
(see the appendix, p.\ \pageref{fancy:thm:typing-completeness}).

\setcounter{theorem}{10}
\begin{theorem}[Completeness of Algorithmic Typing]
  Given $\substextend{\Gamma}{\Omega}$ s.t.\ $\dom{\Gamma} = \dom{\Omega}$:

    \begin{enumerate}[(i)]
      \item %
           If $\judgetp{\Gamma}{A\;p}$
           and $\declchkjudg{p}{[\Omega]\Gamma}{[\Omega]e}{[\Omega]A}$
           and $p' \moreprincipal p$
           then there exist $\Delta$ and $\Omega'$
           such that
           $\substextend{\Delta}{\Omega'}$
           and $\dom{\Delta} = \dom{\Omega'}$
           and $\substextend{\Omega}{\Omega'}$ 
           and 
           $\chkjudg{p'}{\Gamma}{e}{[\Gamma]A}{\Delta}$.
      
      \item %
           If  $\judgetp{\Gamma}{A\;p}$
           and $\declsynjudg{p}{[\Omega]\Gamma}{[\Omega]e}{A}$
           then there exist $\Delta$, $\Omega'$, $A'$, and $p' \moreprincipal p$
           such that
           $\substextend{\Delta}{\Omega'}$
           and $\dom{\Delta} = \dom{\Omega'}$
           and $\substextend{\Omega}{\Omega'}$
           and $\synjudg{p'}{\Gamma}{e}{A'}{\Delta}$
           and $A' = [\Delta]A'$
           and $A = [\Omega']A'$. 

      \item %
           If $\judgetp{\Gamma}{A\;p}$
           and $\declspinejudg{[\Omega]\Gamma}{[\Omega]s}{[\Omega]A}{p}{B}{q}$
           and $p' \moreprincipal p$
           then there exist $\Delta$, $\Omega'$, $B'$, and $q' \moreprincipal q$
           such that
           $\substextend{\Delta}{\Omega'}$
           and $\dom{\Delta} = \dom{\Omega'}$
           and $\substextend{\Omega}{\Omega'}$
           and $\spinejudg{\Gamma}{s}{[\Gamma]A}{p'}{B'}{q'}{\Delta}$
           and $B' = [\Delta]B'$
           and $B = [\Omega']B'$. 

      \item %
           As part (iii), but with $\spinejudgsym B~\lceil q \rceil \cdots$ and
           $\spinejudgsym B'~\lceil q' \rceil \cdots$.
     \end{enumerate}
\end{theorem}

\mypara{Proof sketch---\DeclRecover case}
By i.h., $\spinejudg{\Gamma}{s}{[\Gamma]A}{\p}{C'}{\OK}{\Delta}$
where $\substextend{\Delta}{\Omega'}$ and $\substextend{\Omega}{\Omega'}$
and $\dom{\Delta} = \dom{\Omega'}$
and $C = [\Omega']C'$.

To apply \Recover, we need to show $\FEV{[\Delta]C'} = \emptyset$.
Suppose, for a contradiction, that $\FEV{[\Delta]C'} \neq \emptyset$.
Construct a variant of $\Omega'$ called $\Omega_2$ that has a different
solution for some $\ahat \in \FEV{[\Delta]C'}$.
By soundness (Thm.~\ref{fancy:thm:typing-completeness}),
$\declspinejudg{[\Omega_2]\Gamma}{[\Omega_2]s}{[\Omega_2]A}{\p}{[\Omega_2]C'}{\OK}$.
Using a separation lemma with the trivial $\Gamma = (\Gamma \sep \cdot)$ we get $\Delta = (\Delta_L \sep \Delta_R)$ and $\sepextend{\Gamma}{\cdot}{\Delta_L}{\Delta_R}$ and $\FEV{C'} \subseteq \dom{\Delta_R}$.
That is, all existentials in $C'$ were introduced within the derivation of the
(algorithmic) spine judgment.  Thus, applying $\Omega_2$ to things gives the
same result as $\Omega$, except for $C'$, giving
$
\declspinejudg{[\Omega]\Gamma}{[\Omega]s}{[\Omega]A}{\p}{[\Omega_2]C'}{\OK}
$.
Now instantiate the ``for all $C_2$'' premise with $C_2 = [\Omega_2]C'$,
giving $C = [\Omega_2]C'$.  But we chose $\Omega_2$ to have a different
solution for $\ahat \in \FEV{C'}$, so we have $C \neq [\Omega_2]C'$:  Contradiction.
Therefore $\FEV{[\Delta]C'} = \emptyset$, so we can apply \Recover.

%% file: related.tex
\section{Discussion and Related Work}
\Label{sec:related}

A staggering amount of work has been done on GADTs and indexed types,
and for space reasons we cannot offer a comprehensive survey of the
literature.  So we compare more deeply to fewer papers, to communicate
our understanding of the design space.

\paragraph{Proof theory and type theory.}
As described in \Sectionref{sec:intro}, there are two logical accounts of
equality---the identity type of Martin-L\"{o}f and the equality type
of \citet{schroeder-heister-equality} and \citet{girard-equality}.
The Girard/Schroeder-Heister equality has a more direct connection to
pattern matching, which is why we make use of it.
\citet{Coquand96:typechecking-dependent-types} pioneered the study of
pattern matching in dependent type theory. One perhaps surprising
feature of Coquand's pattern-matching syntax is that it is strictly
stronger than Martin-L\"{o}f's eliminators. His rules can derive
the uniqueness of identity proofs as well as the
disjointness of constructors. Constructor disjointness is also
derivable from the Girard/Schroeder-Heister equality, because
there is no unifier for two distinct constructors.

In future work, we hope to study the relation between these two
notions of equality in more depth; richer equational theories (such as
the theory of commutative rings or the $\beta\eta$-theory of the
lambda calculus) do not have decidable unification, but it seems
plausible that there are hybrid approaches which might let us retain
some of the convenience of the G/SH equality rule while retaining the
decidability of Martin-L\"{o}f's $\mathsf{J}$ eliminator.

\paragraph{Indexed and refinement types.}
Dependent ML~\citep{Xi99popl} indexed
programs with propositional constraints, extending the
 ML type discipline to 
maintain additional invariants. %
DML collected constraints from the program and
passed them to a constraint solver, a 
technique used by systems like Stardust~\citep{Dunfield07:Stardust}
and liquid types~\citep{Rondon08:LT}. 

\paragraph{From phantom types to GADTs.}
\citet{Leijen99:DSEC} introduced the term \emph{phantom
type} to describe a technique for programming in ML/Haskell where
additional type parameters are used to constrain when values are
well-typed. This idea proved to have many applications, ranging from
foreign function interfaces \citep{Blume01entcs:NLFFI} to 
encoding Java-style subtyping \citep{Fluet06:Phantoms}.
Phantom types allow \emph{constructing} values with constrained
types, but do not easily permit \emph{learning} about type equalities
by \emph{analyzing} them, putting applications such as intensional
type analysis~\citep{Harper95:intensional} out of reach. Both
\citet{Cheney03:FirstClassPhantom} and \citet{Xi03:guarded} proposed
treating equalities as a first-class concept, giving explicitly-typed
calculi for equalities, but without studying
algorithms for type inference. 

\citet{Simonet07:constraint} gave a constraint-based algorithm for
type inference for GADTs. It is this work which first identified the
potential intractibility of type inference arising from the
interaction of hypothetical constraints and unification variables. To
resolve this issue they introduce the notion of \emph{tractable}
constraints (\ie, constraints where hypothetical equations never
contain existentials), and require placing enough annotations that
all constraints are tractable.
In general, this could require annotations on case
expressions, so subsequent work focused on relaxing this
requirement. Though quite different in technical detail,
stratified inference~\citep{Pottier06:stratified}
and \emph{wobbly types}~\citep{PeytonJones06:GADTs} both
work by pushing type information from annotations to case expressions,
with stratified type inference literally moving annotations around,
and wobbly types tracking which parts of a type have no
unification variables.
Modern GHC uses the OutsideIn algorithm~\citep{Vytiniotis11},
which further relaxes the  constraint:
case analysis is permitted as long as it cannot modify what is known about an equation.

In our type system, the checking judgment of the bidirectional
algorithm serves to propagate annotations; our requirement that
the scrutinee of a case expression be principal ensures that no
equations contain unification variables. The result is close in effect to
stratified types, and is less expressive than OutsideIn.
This is a deliberate design choice to keep the meaning of principality---that
only a single type can be inferred for a term---clear and easy to
understand.

To specify the OutsideIn approach, the case rule in our declarative
system should permit scrutinizing an expression if all types that can
be synthesized for it have exactly the same equations, even if they
differ in their monotype parts. To achieve this, we would need to
introduce a relation $C' \sim C$ which checks whether the equational
constraints in $C$ and $C'$ are the same, and then modify the
higher-order premise of the \DeclRecover rule to check that
$C' \sim C$ (rather than $C' = C$, as it is currently). However, we
thought such a spec is harder for programmers to develop an intuition
for than simply saying that a scrutinee must synthesize a unique type.

\citet{Garrigue13} proposed \emph{ambivalent types}, which
are a way of deciding when it is safe to generalize the type of a
function using GADTs. This idea is orthogonal to our calculus, simply
because we do no generalization at all:
\emph{every} polymorphic function takes an annotation. However, \citet{Garrigue13}
also emphasize the importance of \emph{monotonicity}, which 
says that substitution should be stable under subtyping, that is,
giving a more general type should not cause subtyping to fail. This
condition is satisfied by our bidirectional system.

\citet{Karachalias15} developed a coverage algorithm for GADTs that
depends on external constraint solving; we offer a more self-contained
but still logically-motivated approach.

\paragraph{Polarized subtyping.} \citet{Barendregt83} observed that a
program which typechecks under a subtyping discipline can be checked
without subtyping, provided that the program is sufficiently
$\eta$-expanded.  This idea of subtyping as $\eta$-expansion was
investigated in a focused (albeit infinitary) setting by
\citet{Zeilberger09}. Another notion of polarity arises from
considering the (co-, contra-, in-)variance of type constructors. It is
used by \citet{Abel06} to give a version of $F^\omega$ with subtyping,
and \citet{Dolan17} apply this version of polarity to give a complete
type inference algorithm for an ML-style language with subtyping.
Our polarized subtyping judgment is closest in spirit to the work of
\citet{Zeilberger09}. The restriction on our subtyping relation can be
understood in terms of requiring the $\eta$ expansions our subtyping
relation infers to be in a focused normal form.

\paragraph{Extensions.}  To keep our formalization manageable, we left
out some features that would be desirable in practice.  In particular,
we need (1) type constructors which take arguments and (2) recursive
types \citep[chapter 20]{Pierce02:TAPL}.  The issue with both
of these features is that they need to permit instantiating
quantifiers with existentials and other binders, and our system relies
upon monotypes (which do not contain such connectives).
This limitation should create no difficulties in typical
practice if we treat user-defined type constructors like \code{List} as monotypes,
expanding the definition only as needed:
when checking an expression against a user type constructor,
and for pattern matching.
Another extension, which we intend as future work,
is to replace ordinary unification with pattern or nominal unification,
to allow type instantiations containing binders.

Another extension is to increase the amount of type inference
done. For instance, a natural question is whether we can extend the
bidirectional approach to subsume the inference done by the algorithm
of \citet{Damas82}.  On the implementation side, this seems easy---%
to support ML-style type inference,
we can add rules to infer types for values:
\[
\Infer{}
      {\chkjudg{\OK}{\Gamma, \MonnierComma{\ahat}, \ahat, \bhat, x:\ahat}{e}{\bhat}{\Delta, \MonnierComma{\ahat}, \Delta'}
       \\
       \vec{\chat} = \unsolved{\Delta'}
      }
      {\synjudg{\p}{\Gamma}{\fun{x}{e}}{\alltype{\vec{\alpha}}{{[\vec{\alpha}/\vec{\chat}]}{[\Delta']}(\ahat \to \bhat)}}{\Delta}}
\]
This rule adds a marker $\MonnierComma{\ahat}$ to the
context, then checks the body $e$ against the type $\bhat$.
Our output type substitutes away all the solved
existential variables to the right of $\MonnierComma{\ahat}$,
and generalizes over all unsolved variables to the right of the marker.
Using an ordered context gives precise control over the scope of the
existential variables, easily expressing polymorphic generalization.

However, in the presence of generalization, the declarative
specification of type inference no longer strictly specifies the
order of polymorphic quantifiers (i.e.,
$\alltype{\alpha, \beta} \alpha \to \beta \to (\alpha \times \beta)$
and
$\alltype{\beta, \alpha} \alpha \to \beta \to (\alpha \times \beta)$
should be equivalent) and so our principal synthesis would no longer
return types stable up to alpha-equivalence. Fixing this would be
straightforward (by relaxing the definition of type equivalence),
but we have not pursued this because we do not value
let-generalization enough to pay the price of increased complexity in
our proofs.

%% file: acks.tex
\begin{acks}
  We thank the anonymous reviewers of this version, and of several previous versions, for their comments. We also thank Soham Chowdhury for his work on implementing the system presented in this paper. This work was partially funded by EPSRC grant EP/N02706X/2. 
\end{acks}

%% file: appendix.tex
\section{Figures}
\label{sec:figures}

We repeat some figures from the main paper.
In Figures \ref{mainfig:decl-typing}a and \ref{mainfig:alg-typing}a, we include rules omitted
from the main paper for space reasons.

\input{fig-decl-typing.tex}

\input{fig-alg-typing.tex}

\input{fig-decl-wf.tex}

\clearpage

\input{fig-algorithmic.tex}
\clearpage

\input{fig-alg-prop.tex}

\input{fig-alg-checkeq.tex}

\input{fig-clash.tex}

\input{fig-alg-elimeq.tex}

\input{fig-alg-subtyping.tex}

\input{fig-instantiation.tex}

\input{fig-alg-pattern-matching.tex}

\input{fig-alg-match-coverage.tex}

\clearpage

%% file: fig-decl-wf.tex
\begin{figure}[thbp]
\raggedright
    \loudjudgbox{judge}{\judge{\Psi}{t}{\sort}}{Under context $\Psi$, term $t$ has sort $\sort$}
    \vspace*{-1ex}
    \begin{mathpar}
    \small
        \Infer{\DeclUvarSort}
              {(\alpha:\sort) \in \Psi}
              {\judge{\Psi}{\alpha}{\sort}}
        ~~~~~
        \Infer{\DeclUnitSort}
              { }
              { \judge{\Psi}{\unitty}{\type} }
        ~~~~~
        \Infer{\DeclBinSort}
              { \judge{\Psi}{t_1}{\type}
                \\
                \judge{\Psi}{t_2}{\type}
              }
              { \judge{\Psi}{t_1 \binc t_2}{\type} }
        \and
        \Infer{\DeclZeroSort}
              { }
              {\judge{\Psi}{\zero}{\ind}}
        ~~~~~
        \Infer{\DeclSuccSort}
              {\judge{\Psi}{t}{\ind} }
              {\judge{\Psi}{\succ{t}}{\ind}}
    \end{mathpar}

    \loudjudgbox{judgeprop}{\judgeprop{\Psi}{P}}{Under context $\Psi$, proposition $P$ is well-formed}
    \vspace{-1ex}
    \begin{mathpar}
      \Infer{\DeclEqProp}
            {\judge{\Psi}{t}{\ind} \\
             \judge{\Psi}{t'}{\ind} 
            }
            {\judgeprop{\Psi}{t = t'}}
    \end{mathpar}

    \loudjudgbox{judgetp}{\judgetp{\Psi}{A}}{Under context $\Psi$, type $A$ is well-formed}
    \begin{mathpar}
        \Infer{\DeclUvarWF}
              {(\alpha:\type) \in \Psi}
              {\judgetp{\Psi}{\alpha}}
        \and
        \Infer{\DeclUnitWF}
              { }
              { \judgetp{\Psi}{\unitty} }
        \and
        \Infer{\DeclBinWF}
              { \judgetp{\Psi}{A}
                \\
                \judgetp{\Psi}{B}
                \\
                \binc\, \in \setof{\to, \times, +}
              }
              { \judgetp{\Psi}{A \binc B} }
        \and
        \Infer{\DeclVecWF}
            { \judge{\Gamma}{t}{\kindnat}
              \\
              \judgetp{\Gamma}{A}
            }
            { \judgetp{\Gamma}{\vectype{t}{A}} }%
        \and
        \Infer{\DeclForallWF}
              {\judgetp{\Psi, \alpha : \sort} {A}}
              {\judgetp{\Psi}{(\alltype{\alpha : \sort} A)}}
        \and
        \Infer{\DeclExistsWF}
              {\judgetp{\Psi, \alpha : \sort} {A}}
              {\judgetp{\Psi}{(\extype{\alpha : \sort} A)}}
        \\
        \Infer{\DeclImpliesWF}
              {\judgeprop{\Psi}{P} \\ 
               \judgetp{\Psi}{A}}
              {\judgetp{\Psi}{P \implies A}}
        \and
        \Infer{\DeclWithWF}
              {\judgeprop{\Psi}{P} \\ 
               \judgetp{\Psi}{A}}
              {\judgetp{\Psi}{A \with P}}
    \end{mathpar}

    \loudjudgbox{judgetpvec}{\judgetpvec{\Psi}{\Avec}}{Under context $\Psi$, types in $\Avec$ are well-formed}
    \vspace*{-1ex}
    \begin{mathpar}
        \Infer{\DeclTypevec}
              {\arrayenvbl{\text{for all~}A \in \Avec.
                   \\ ~~~ \judgetp{\Psi}{A}}
              }
              {\judgetpvec{\Psi}{\Avec}}
    \end{mathpar}
    
    \loudjudgbox{judgectx}{\judgectx{\Psi}}{Declarative context $\Psi$ is well-formed}
    \vspace*{-1ex}
    \begin{mathpar}
        \Infer{\DeclEmptyCWF}
              { }
              { \judgectx{\cdot} }
        \and
        \Infer{\DeclHypCWF}
              { \judgectx{\Psi}
                \\
                x \notin \dom{\Psi}
                 \\
                \judgetp{\Psi}{A}
              }
              { \judgectx{\Psi, x : A} }
        \and
        \Infer{\DeclVarCWF}
              { \judgectx{\Psi}
                \\
                \alpha \notin \dom{\Psi}}
              { \judgectx{\Psi, \alpha:\sort} }
    \end{mathpar}
\caption{Sorting; well-formedness of propositions, types, and contexts in the declarative system}
\FLabel{fig:decl-wf}
\end{figure}

%% file: fig-algorithmic.tex
\begin{figure}
  \raggedright

  \loudjudgbox{judge}{\judge{\Gamma}{\tau}{\sort}}{Under context $\Gamma$, term $\tau$ has sort $\sort$}
  \vspace*{-1ex}
  \begin{mathpar}
      \Infer{\VarSort}
            {(u:\sort) \in \Gamma}
            {\judge{\Gamma}{u}{\sort}}
      \and
      \Infer{\SolvedVarSort}
            {(\hypeq{\ahat:\sort}{\tau}) \in \Gamma }
            { \judge{\Gamma}{\ahat}{\sort} }
      \and
      \Infer{\UnitSort}
            { }
            { \judge{\Gamma}{\unitty}{\type} }
      \and
      \Infer{\BinSort}
            { \judge{\Gamma}{\tau_1}{\type}
              \\
              \judge{\Gamma}{\tau_2}{\type}
            }
            { \judge{\Gamma}{\tau_1 \binc \tau_2}{\type} }
      \and
      \Infer{\ZeroSort}
            { }
            {\judge{\Gamma}{\zero}{\ind}}
      \and
      \Infer{\SuccSort}
            {\judge{\Gamma}{t}{\ind} }
            {\judge{\Gamma}{\succ{t}}{\ind}}
  \end{mathpar}

  \loudjudgbox{judgeprop}{\judgeprop{\Gamma}{P}}{Under context $\Gamma$, proposition $P$ is well-formed}
  \vspace*{-1ex}
  \begin{mathpar}
    \Infer{\EqProp}
          {\judge{\Gamma}{t}{\ind} \\
           \judge{\Gamma}{t'}{\ind} 
          }
          {\judgeprop{\Gamma}{t = t'}}
  \end{mathpar}

  \loudjudgbox{judgetp}{\judgetp{\Gamma}{A}}{Under context $\Gamma$, type $A$ is well-formed}
  \vspace*{-1ex}
  \begin{mathpar}
      \Infer{\VarWF}
            { (u : \type) \in \Gamma }
            { \judgetp{\Gamma}{u} }
      \and
      \Infer{\SolvedVarWF}
            { (\hypeq{\ahat:\type}{\tau}) \in \Gamma }
            { \judgetp{\Gamma}{\ahat} }
      \and
      \Infer{\UnitWF}
            { }
            { \judgetp{\Gamma}{\unitty} }
      \and
      \Infer{\BinWF}
            { \judgetp{\Gamma}{A}
              \\
              \judgetp{\Gamma}{B} 
              \\
              \binc\, \in \setof{\to, \times, +}
            }
            { \judgetp{\Gamma}{A \binc B} }
      \and
      \Infer{\VecWF}
            { \judge{\Gamma}{t}{\kindnat}
              \\
              \judgetp{\Gamma}{A}
            }
            { \judgetp{\Gamma}{\vectype{t}{A}} }%
      \and
      \Infer{\ForallWF}
            { \judgetp{\Gamma, \alpha:\sort}{A} }
            { \judgetp{\Gamma}{\alltype{\alpha:\sort}{A}} }
      ~~~~
      \Infer{\ExistsWF}
            { \judgetp{\Gamma, \alpha:\sort}{A} }
            { \judgetp{\Gamma}{\extype{\alpha:\sort}{A}} }
      \\
      \Infer{\ImpliesWF}
            {\judgeprop{\Gamma}{P} \\ 
             \judgetp{\Gamma}{A}}
            {\judgetp{\Gamma}{P \implies A}}
      ~~~~
      \Infer{\WithWF}
            {\judgeprop{\Gamma}{P} \\ 
             \judgetp{\Gamma}{A}}
            {\judgetp{\Gamma}{A \with P}}
  \end{mathpar}
 
  \medskip

  \loudjudgbox{judgetp}{\judgetp{\Gamma}{A\;p}}{Under context $\Gamma$, type $A$ is well-formed and respects principality $p$}
  \begin{mathpar}
        \Infer{\PrincipalWF}
              {\judgetp{\Gamma}{A}
               \\
               \FEV{[\Gamma]A} = \emptyset}
              {\judgetp{\Gamma}{A\;\p}}
        \and
        \Infer{\NonPrincipalWF}
              {\judgetp{\Gamma}{A}}
              {\judgetp{\Gamma}{A \OK}}
  \end{mathpar}
  
  \medskip
  
  \loudjudgbox{judgetpvec}{\judgetpvec{\Gamma}{\Avec\;\text{[}p\text{]}}}{Under context $\Gamma$, types in $\Avec$ are well-formed [with principality $p$]}
  \begin{mathpar}
      \Infer{\Typevec}
              {\arrayenvbl{\text{for all~}A \in \Avec.
                  ~~ \judgetp{\Gamma}{A}}
              }
              {\judgetpvec{\Gamma}{\Avec}}
      \and
      \Infer{\PrincipalTypevec}
              {\arrayenvbl{\text{for all~}A \in \Avec.
                  ~~ \judgetp{\Gamma}{A\;p}}
              }
              {\judgetpvec{\Gamma}{\Avec\;p}}
  \end{mathpar}
  
  \medskip

  \loudjudgbox{judgectx}{\judgectx{\Gamma}}{Algorithmic context $\Gamma$ is well-formed}
  \vspace*{-2.3ex}
  \begin{mathpar}
      \Infer{\!\EmptyCWF}
            { }
            { \judgectx{\cdot} }
      ~~
      \Infer{\!\HypCWF}
            { \judgectx{\Gamma}
              ~~~~
              \arrayenvbl{
                 x \notin \dom{\Gamma}
                 \\
                 \judgetp{\Gamma}{A}
              }
            }
            { \judgectx{\Gamma, \hyp{\OK}{x}{A}} }
      ~~
      \Infer{\!\HypCWFp}
            { \judgectx{\Gamma}
              ~~~~
              \arrayenvbl{
                  x \notin \dom{\Gamma}
                  \\
                  \judgetp{\Gamma}{A}
              }
              ~~~~
              \FEV{[\Gamma]A} = \emptyset
            }
            { \judgectx{\Gamma, \hyp{\p}{x}{A}} }
      \\
      \Infer{\VarCWF}
            { \judgectx{\Gamma}
              \\
              u \notin \dom{\Gamma}}
            { \judgectx{\Gamma, u:\sort} }
      \and
      \Infer{\SolvedCWF }
            { \judgectx{\Gamma}
              \\
              \ahat \notin \dom{\Gamma} 
              \\
              \judge{\Gamma}{t}{\sort} 
            }
            { \judgectx{\Gamma, \hypeq{\ahat:\sort}{t}}
            }
      \and
      \Infer{\EqnCWF}
            { \judgectx{\Gamma}
              \\
              \alpha:\sort \in \Gamma
              \\
              (\hypeq{\alpha}{-}) \notin \Gamma
              \\
              \judge{\Gamma}{\tau}{\sort}
            }
            { \judgectx{\Gamma, \hypeq{\alpha}{\tau}} }
      \and
      \Infer{\MarkerCWF}
            { \judgectx{\Gamma}
              \\
              \MonnierComma{u} \notin \Gamma
            }
            { \judgectx{\Gamma, \MonnierComma{u}} }
  \vspace*{-1.3ex}
  \end{mathpar}

  \caption{Well-formedness of types and contexts in the algorithmic system}
  \FLabel{fig:alg-wf}
\end{figure}

%% file: fig-alg-prop.tex
\begin{figure}
\raggedright
  \loudjudgbox{checkprop}{\checkprop{\Gamma}{P}{\Delta}}
          {Under context $\Gamma$, check $P$, with output context $\Delta$}
  \begin{mathpar}
    \Infer{\CheckpropEq}
          { \checkeq{\Gamma}{t_1}{t_2}{\ind}{\Delta} }
          {\checkprop{\Gamma}{t_1 = t_2}{\Delta}}
  \end{mathpar}

  \loudjudgbox{elimprop}{\elimprop{\Gamma}{P}{\Deltabot}}
          {Incorporate hypothesis $P$ into $\Gamma$, producing $\Delta$ or inconsistency $\bot$}
  \begin{mathpar}
    \Infer{\ElimpropEq}
          {\elimeq{\Gamma}{t_1}{t_2}{\ind}{\Deltabot} }
          {\elimprop{\Gamma}{t_1 = t_2}{\Deltabot}}
  \end{mathpar}

\caption{Checking and assuming propositions}
\FLabel{fig:alg-prop}
\end{figure}

%% file: fig-alg-checkeq.tex
\begin{figure}
  \loudjudgbox{checkeq}{\checkeq{\Gamma}{t_1}{t_2}{\sort}{\Delta}}
          {Check that $t_1$ equals $t_2$, taking $\Gamma$ to $\Delta$}
  \begin{mathpar}
    \Infer{\CheckeqVar}
         { }
         {\checkeq{\Gamma}{u}{u}{\sort}{\Gamma}}
    \and
    \Infer{\CheckeqUnit}
         { }
         {\checkeq{\Gamma}{\unitty}{\unitty}{\type}{\Gamma}}
    \and
    \Infer{\CheckeqBin}
          {\checkeq{\Gamma}{\tau_1}{\tau'_1}{\type}{\Theta}
           \\
           \checkeq{\Theta}{[\Theta]\tau_2}{[\Theta]\tau'_2}{\type}{\Delta}
          }
          {
            \checkeq{\Gamma}{(\tau_1 \binc \tau_2)}{(\tau'_1 \binc \tau'_2)}{\type}{\Delta}
          }
    \\
    \Infer{\CheckeqZero}
          { }
          {\checkeq{\Gamma}{\zero}{\zero}{\ind}{\Gamma}}
    \and
    \Infer{\CheckeqSucc}
          {\checkeq{\Gamma}{t_1}{t_2}{\ind}{\Delta}}
          {\checkeq{\Gamma}{\succ{t_1}}{\succ{t_2}}{\ind}{\Delta}}
    \\
    \Infer{\CheckeqInstL}
         {\instjudg{\Gamma[\ahat:\sort]}{\ahat}{t}{\sort}{\Delta}
          \\
          \ahat \notin \FV{t}
         }
         {\checkeq{\Gamma[\ahat:\sort]}{\ahat}{t}{\sort}{\Delta}}
    \\
    \Infer{\CheckeqInstR}
         {\instjudg{\Gamma[\ahat:\sort]}{\ahat}{t}{\sort}{\Delta}
          \\
          \ahat \notin \FV{t}
         }
         {\checkeq{\Gamma[\ahat:\sort]}{t}{\ahat}{\sort}{\Delta}}
  \vspace*{-1.2ex}
  \end{mathpar}
  \caption{Checking equations}
  \FLabel{fig:alg-checkeq}
\end{figure}

%% file: fig-clash.tex
\begin{figure}
  \raggedright
  \loudjudgbox{clash}{\clash{t_1}{t_2}}
  {$t_1$ and $t_2$ have incompatible head constructors}
  \vspace{-2.5ex}
  \begin{mathpar}
    \Infer{}
          { }
          { \clash{\zero}{\succ{t}} }
    ~~~~
    \Infer{}
          { }
          { \clash{\succ{t}}{\zero} }
    ~~~~
    \Infer{}
          { }
          { \clash{\unitty}{(\tau_1 \binc \tau_2)} }
    ~~~~
    \Infer{}
          { }
          { \clash{(\tau_1 \binc \tau_2)}{\unitty} }
    ~~~~
    \Infer{}
          { {\binc_1} \neq {\binc_2} }
          { \clash{(\sigma_1 \binc_1 \tau_1)}{(\sigma_2 \binc_2 \tau_2)} }
  \vspace*{-1.3ex}
  \end{mathpar}
  
  \caption{Head constructor clash}
  \FLabel{fig:clash}
\end{figure}

%% file: fig-alg-elimeq.tex
\begin{figure}
  \raggedright
  \loudjudgbox{elimeq}{\elimeq{\Gamma}{\sigma}{\tau}{\sort}{\Deltabot}}
          {Unify $\sigma$ and $\tau$, taking $\Gamma$ to $\Delta$, or to inconsistency $\bot$}
  \begin{mathpar}
      \Infer{\ElimeqUvarRefl}
            { }
            {\elimeq{\Gamma}{\alpha}{\alpha}{\sort}{\Gamma}}
      \\
      \Infer{\ElimeqZero}
            { }
            {\elimeq{\Gamma}{\zero}{\zero}{\ind}{\Gamma}}
      \and
      \Infer{\ElimeqSucc}
            {\elimeq{\Gamma}{\sigma}{\tau}{\ind}{\Deltabot}}
            {\elimeq{\Gamma}{\succ{\sigma}}{\succ{\tau}}{\ind}{\Deltabot}}
     \\
    \Infer{\ElimeqUvarL}
          { \alpha \notin \FV{\tau}
            \\
            (\hypeq{\alpha}{-}) \notin \Gamma
          }
          {\elimeq{\Gamma}{\alpha}{\tau}{\sort}{\Gamma, \hypeq{\alpha}{\tau}}}
    \and
    \Infer{\ElimeqUvarR}
          { \alpha \notin \FV{\tau}
            \\
            (\hypeq{\alpha}{-}) \notin \Gamma
          }
          {\elimeq{\Gamma}{\tau}{\alpha}{\sort}{\Gamma, \hypeq{\alpha}{\tau}}}
    \and
    \Infer{\ElimeqUvarLBot}
          { t \neq \alpha \\ \alpha \in \FV{\tau} }
          {\elimeq{\Gamma}{\alpha}{\tau}{\sort}{\bot}}
    \and
    \Infer{\ElimeqUvarRBot}
          { t \neq \alpha \\ \alpha \in \FV{\tau} }
          {\elimeq{\Gamma}{\tau}{\alpha}{\sort}{\bot}}
    \\
    \Infer{\ElimeqUnit}
          { }
          {\elimeq{\Gamma}{\unitty}{\unitty}{\type}{\Gamma}}
    \and
    \Infer{\ElimeqBin}  %
          {\elimeq{\Gamma}{\tau_1}{\tau'_1}{\type}{\Theta} \\ 
           \elimeq{\Theta}{[\Theta]\tau_2}{[\Theta]\tau'_2}{\type}{\Deltabot}
          }
          {\elimeq{\Gamma}{(\tau_1 \binc \tau_2)}{(\tau'_1 \binc \tau'_2)}{\type}{\Deltabot}}
    \and
    \Infer{\ElimeqBinBot}  %
          {\elimeq{\Gamma}{\tau_1}{\tau'_1}{\type}{\bot} }
          {\elimeq{\Gamma}{(\tau_1 \binc \tau_2)}{(\tau'_1 \binc \tau'_2)}{\type}{\bot}}
    \\
    \Infer{\ElimeqClash}
          { \clash{\sigma}{\tau} }
          { \elimeq{\Gamma}{\sigma}{\tau}{\sort}{\bot} }
  \end{mathpar}
  \caption{Eliminating equations}
  \FLabel{fig:alg-elimeq}
\end{figure}

%% file: fig-alg-subtyping.tex
\begin{figure*}[htbp]
  \raggedright
  \loudjudgbox{subjudg}{\subjudg[\polvar]{\Gamma}{A}{B}{\Delta}}%
     {Under input context $\Gamma$,
       type $A$ is a subtype of $B$, with output context $\Delta$}
  \vspace*{-1.3ex}
  \begin{mathpar}
      \Infer{\SubEquiv}
                {
                  \arrayenvbl{\notAllExists{A} \\ \notAllExists{B}}
                  \\
                  \equivjudg{\Gamma}{A}{B}{\Delta}
                }
                {\subjudg[\polvar]{\Gamma}{A}{B}{\Delta}}
      \vspace{-2ex}
      \\
      \Infer{\SubAllL}
                {
                  \arrayenvbl{\notAll{B}
                      \\
                      \subjudg[-]{\Gamma, \MonnierComma{\ahat}, \ahat:\sort}
                             {[\ahat/\alpha]A}
                             {B}
                             {\Delta, \MonnierComma{\ahat}, \Theta}
                   }
                }
                {\subjudg[-]{\Gamma}{\alltype{\alpha:\sort}{A}}{B}{\Delta}}
      \and
      \Infer{\SubAllR}
                {\subjudg[-]{\Gamma, \beta :\sort}{A}{B}{\Delta, \beta:\sort, \Theta}}
                {\subjudg[-]{\Gamma}{A}{\alltype{\beta:\sort}{B}}{\Delta}}
      \\
      \Infer{\SubExistsL}
                {
                 \subjudg[+]{\Gamma, \alpha:\sort}
                         {A}
                         {B}
                         {\Delta, \alpha:\sort, \Theta}}
                {\subjudg[+]{\Gamma}
                         {\extype{\alpha:\sort}{A}}
                         {B}
                         {\Delta}}
      \and
      \Infer{\SubExistsR}
                {
                  \arrayenvbl{
                      \notExists{A}
                      \\
                      \subjudg[+]{\Gamma, \MonnierComma{\bhat}, \bhat:\sort}
                             {A}
                             {[\bhat/\beta]B}
                             {\Delta, \MonnierComma{\bhat}, \Theta}
                  }
                }
                {\subjudg[+]{\Gamma}
                         {A}
                         {\extype{\beta:\sort}{B}}
                         {\Delta}}
       \\
       \Infer{\SubPosNegL}
             {\subjudg[-]{\Gamma}{A}{B}{\Delta}
               \\
               \arrayenvbl{\Neg{A}
               \\
               \nonPos{B}}}
             {\subjudg[+]{\Gamma}{A}{B}{\Delta}}
       \and
       \Infer{\SubPosNegR}
             {\subjudg[-]{\Gamma}{A}{B}{\Delta}
              \\
              \arrayenvbl{\nonPos{A}
              \\
              \Neg{B}}}
             {\subjudg[+]{\Gamma}{A}{B}{\Delta}}
       \\
       \Infer{\SubNegPosL}
             {\subjudg[+]{\Gamma}{A}{B}{\Delta}
              \\
              \arrayenvbl{\Pos{A}
              \\
              \nonNeg{B}}}
             {\subjudg[-]{\Gamma}{A}{B}{\Delta}}
       \and
       \Infer{\SubNegPosR}
             {\subjudg[+]{\Gamma}{A}{B}{\Delta}
              \\
              \arrayenvbl{\nonNeg{A}
              \\
              \Pos{B}}}
             {\subjudg[-]{\Gamma}{A}{B}{\Delta}}
  \end{mathpar}

   \loudjudgbox{propequivjudg}{\propequivjudg{\Gamma}{P}{Q}{\Delta}}
      {Under input context $\Gamma$,
        \\
        check that $P$ is equivalent to $Q$
        \\
        with output context $\Delta$}
      \begin{mathpar}
         \Infer{\PropequivEq}
              {\checkeq{\Gamma}{t_1}{t_2}{\ind}{\Theta}
               \\
               \checkeq{\Theta}{[\Theta]t_1'}{[\Theta]t_2'}{\ind}{\Delta}}
              {\propequivjudg{\Gamma}{(t_1 = t_1')}{(t_2 = t_2')}{\Delta}}
       \end{mathpar}

  \loudjudgbox{equivjudg}{\equivjudg{\Gamma}{A}{B}{\Delta}}
     {Under input context $\Gamma$,
       \\
       check that $A$ is equivalent to $B$
       \\
       with output context $\Delta$}
     \vspace{-1.5ex}
     \begin{mathpar}
        \Infer{\EquivVar}
             { }
             {\equivjudg{\Gamma}{\alpha}{\alpha}{\Gamma}}
        \and
        \Infer{\EquivExvar}
             { }
             {\equivjudg{\Gamma}{\ahat}{\ahat}{\Gamma}}
        \and
        \Infer{\EquivUnit}
             { }
             {\equivjudg{\Gamma}{\unitty}{\unitty}{\Gamma}}
        \and
        \Infer{\EquivBin}     %
             {\equivjudg{\Gamma}{A_1}{B_1}{\Theta}
               \\
               \equivjudg{\Theta}{[\Theta]A_2}{[\Theta]B_2}{\Delta}
             }
             {\equivjudg{\Gamma}{(A_1 \binc A_2)}{(B_1 \binc B_2)}{\Delta}}
        \and
        \Infer{\EquivVec}
             {\equivjudg{\Gamma}{t_1}{t_2}{\Theta}
               \\
               \equivjudg{\Theta}{[\Theta]A_1}{[\Theta]A_2}{\Delta}
             }
             {\equivjudg{\Gamma}{(\vectype{t_1}{A_1})}{(\vectype{t_2}{A_2})}{\Delta}}
        \\
        \Infer{\EquivAll}
             {\equivjudg{\Gamma, \alpha:\sort}{A}{B}{\Delta, \alpha:\sort, \Delta'}}
             {\equivjudg{\Gamma}{(\alltype{\alpha:\sort} A)}{(\alltype{\alpha:\sort} B)}{\Delta}}
        \and
        \Infer{\EquivExists}
             {\equivjudg{\Gamma, \alpha:\sort}{A}{B}{\Delta, \alpha:\sort, \Delta'}}
             {\equivjudg{\Gamma}{(\extype{\alpha:\sort} A)}{(\extype{\alpha:\sort} B)}{\Delta}}
        \\
        \Infer{\EquivImplies}
             {\propequivjudg{\Gamma}{P}{Q}{\Theta}
              \\
              \equivjudg{\Theta}{[\Theta]A}{[\Theta]B}{\Delta}}
             {\equivjudg{\Gamma}{(P \implies A)}{(Q \implies B)}{\Delta}}
        \and
        \Infer{\EquivWith}
             {\propequivjudg{\Gamma}{P}{Q}{\Theta}
              \\
              \equivjudg{\Theta}{[\Theta]A}{[\Theta]B}{\Delta}}
             {\equivjudg{\Gamma}{(A \with P)}{(B \with Q)}{\Delta}}
        \\
        \Infer{\EquivInstL}
            {
              \ahat \notin \FV{\tau}
              \\
              \instjudg{\Gamma[\ahat]}{\ahat}{\tau}{\type}{\Delta}
            }
            {\equivjudg{\Gamma[\ahat]}{\ahat}{\tau}{\Delta}}
        \and
        \Infer{\EquivInstR}
            {
              \arrayenvbl{
                   \ahat \notin \FV{\tau}
              }
             \\
             \instjudg{\Gamma[\ahat]}{\ahat}{\tau}{\type}{\Delta}
            }
            {\equivjudg{\Gamma[\ahat]}{\tau}{\ahat}{\Delta}}
     \end{mathpar}

  \caption{Algorithmic subtyping and equivalence}
  \FLabel{fig:alg-subtyping}
\end{figure*}

%% file: fig-instantiation.tex
\begin{figure*}[htbp]
    \raggedright
      \loudjudgbox{instjudg}{\instjudg{\Gamma}{\ahat}{t}{\sort}{\Delta}}%
         {Under input context $\Gamma$,
           \\
           instantiate $\ahat$ such that $\ahat = t$
           with output context $\Delta$}
      \begin{mathpar}
        \Infer{\InstSolve}
                { \judge{\Gamma_0}{\tau}{\sort} }
                { \instjudg{\Gamma_0, \ahat:\sort, \Gamma_1}
                            {\ahat}
                            {\tau}
                            {\sort}
                            {\Gamma_0, \hypeq{\ahat:\sort}{\tau}, \Gamma_1}
                 }
        \and
        \Infer{\InstReach}
                {\bhat \in \unsolved{\Gamma[\ahat:\sort][\bhat:\sort]} }
                {\instjudg{\Gamma[\ahat:\sort][\bhat:\sort]}
                            {\ahat}
                            {\bhat}
                            {\sort}
                            {\Gamma[\ahat:\sort][\hypeq{\bhat:\sort}{\ahat}]}
                }
        \and
        \Infer{\InstBin}
                {\instjudg{\Gamma[\ahat_2:\type, \ahat_1:\type, \hypeq{\ahat:\type}{\ahat_1 \binc \ahat_2}]}
                            {\ahat_1}
                            {\tau_1}
                            {\type}
                            {\Theta} \\
                 \instjudg{\Theta}
                            {\ahat_2}
                            {[\Theta]\tau_2}
                            {\type}
                            {\Delta}}
                {\instjudg{\Gamma[\ahat:\type]}
                            {\ahat}
                            {\tau_1 \binc \tau_2}
                            {\type}
                            {\Delta}}
        \and
        \Infer{\InstZero}
              { }
              {\instjudg{\Gamma[\ahat:\ind]}{\ahat}{\zero}{\ind}{\Gamma[\hypeq{\ahat:\ind}{\zero}]}}
        \and
        \Infer{\InstSucc}
                {\instjudg{\Gamma[\ahat_1:\ind, \hypeq{\ahat:\ind}{\succ{\ahat_1}}]}
                            {\ahat_1}
                            {t_1}
                            {\ind}
                            {\Delta} 
                }
                {\instjudg{\Gamma[\ahat:\ind]}
                            {\ahat}
                            {\succ{t_1}}
                            {\ind}
                            {\Delta}}
      \end{mathpar}    

\caption{Instantiation}
\FLabel{fig:instantiation}
\end{figure*}

%% file: fig-alg-pattern-matching.tex
\begin{figure}[t]
  \raggedright
  \judgbox{\matchjudg{q}{p}{\Gamma}{\Pi}{\Avec}{C}{\Delta}}
     {Under context $\Gamma$,
       \\
       check branches $\Pi$ with patterns of type $\Avec$ and bodies of type $C$}
  \begin{mathpar}
     \Infer{\MatchEmpty}
           {}
           {\matchjudg{q}{p}{\Gamma}{\cdot}{\Avec}{C}{\Gamma}}
     \and
     \Infer{\MatchSeq}
           {\matchjudg{q}{p}{\Gamma}{\pi}{\Avec}{C}{\Theta}
            \\
            \matchjudg{q}{p}{\Theta}{\Pi'}{[\Theta]\Avec}{C}{\Delta} }
           {\matchjudg{q}{p}{\Gamma}{\pi \alt \Pi'}{\Avec}{C}{\Delta}}
     \\
     \Infer{\MatchBase}
           {\chkjudg{p}{\Gamma}{e}{C}{\Delta}
           }
           {\matchjudg{q}{p}{\Gamma}{(\branch{\cdot}{e})}{\cdot}{C}{\Delta}}
     \and
     \Infer{\MatchUnit}
          { 
            \matchjudg{q}{p}{\Gamma}{\branch{\patvec}{e}}{\Avec}{C}{\Delta} }
          { \matchjudg{q}{p}{\Gamma}{\branch{\unitexp, \patvec}{e}}{\unitty, \Avec}{C}{\Delta} }
     \\ 
     \Infer{\MatchExists}
           {\matchjudg{q}{p}{\Gamma, \alpha:\sort}{\branch{\patvec}{e}}{A,\Avec}{C}{\Delta, \alpha:\sort, \Theta} }    %
           {\matchjudg{q}{p}{\Gamma}{\branch{\patvec}{e}}{(\extype{\alpha:\sort} A),\Avec}{C}{\Delta}}
     \and
     \Infer{\MatchWith}
           {\matchelimjudg{p}{\Gamma}{\branch{\patvec}{e}}{P}{A, \Avec}{C}{\Delta} }
           {\matchjudg{\p}{p}{\Gamma}{\branch{\patvec}{e}}{A \with P, \Avec}{C}{\Delta} }
     \and
     \Infer{\MatchWithOK}
           {\matchjudg{\OK}{p}{\Gamma}{\branch{\patvec}{e}}{A, \Avec}{C}{\Delta} }
           {\matchjudg{\OK}{p}{\Gamma}{\branch{\patvec}{e}}{A \with P, \Avec}{C}{\Delta} }
     \\
     \Infer{\MatchPair}
           {\matchjudg{q}{p}{\Gamma}{\branch{\pat_1, \pat_2, \patvec}{e}}{A_1, A_2, \Avec}{C}{\Delta}}
           {\matchjudg{q}{p}{\Gamma}{\branch{\pair{\pat_1}{\pat_2}, \patvec}{e}}{A_1 \times A_2, \Avec}{C}{\Delta}}
     \and
     \Infer{\MatchSum{k}}
           {\matchjudg{q}{p}{\Gamma}{\branch{\pat, \patvec}{e}}{A_k, \Avec}{C}{\Delta}}
           {\matchjudg{q}{p}{\Gamma}{\branch{(\inj{k} \pat), \patvec}{e}}{A_1 + A_2, \Avec}{C}{\Delta} }
     \\ 
     \Infer{\MatchNeg}
           { \notWithExists{A}   %
             \\
             \matchjudg{q}{p}{\Gamma, \hyp{\p}{z}{A}}{\branch{\patvec}{e}'}{\Avec}{C}{\Delta, \hyp{\p}{z}{A}, \Delta'}
           }
           { \matchjudg{q}{p}{\Gamma}{\branch{z, \patvec}{e}}{A, \Avec}{C}{\Delta} }
     \and
     \Infer{\MatchWild}
           {  \notWithExists{A}  %
             \\
             \matchjudg{q}{p}{\Gamma}{\branch{\patvec}{e}}{\Avec}{C}{\Delta} }
           { \matchjudg{q}{p}{\Gamma}{\branch{\wild, \patvec}{e}}{A, \Avec}{C}{\Delta} }
     \\
    \Infer{\MatchNil}
          {
            \matchelimjudg{p}{\Gamma}{\branch{\patvec}{e}}{(t = \zero)}{\Avec}{C}{\Delta}
          }
          {
            \matchjudg{\p}{p}{\Gamma}{\branch{\vecnil, \patvec}{e}}{(\vectype{t}{A}), \Avec}{C}{\Delta}
          }
    \and
    \Infer{\MatchCons}
          {
            \matchelimjudg{p}
                              {\Gamma, \alpha : \kindnat}
                              {\branch{\pat_1, \pat_2, \patvec}{e}}
                              {(t = \succ{\alpha})}
                              {A, (\vectype{\alpha}{A}), \Avec}
                              {C}
                              {\Delta, \alpha : \kindnat, \Theta}
          }
          {
            \matchjudg{\p}{p}{\Gamma}{\branch{(\veccons{\pat_1}{\pat_2}), \patvec}{e}}{(\vectype{t}{A}), \Avec}{C}{\Delta}
          }
     \\
    \Infer{\MatchNilOK}
          {
            \matchjudg{\OK}{p}{\Gamma}{\branch{\patvec}{e}}{\Avec}{C}{\Delta}
          }
          {
            \matchjudg{\OK}{p}{\Gamma}{\branch{\vecnil, \patvec}{e}}{(\vectype{t}{A}), \Avec}{C}{\Delta}
          }
    \and
    \Infer{\MatchConsOK}
          {
            \matchjudg{\OK}{p}
                      {\Gamma, \alpha : \kindnat}
                      {\branch{\pat_1, \pat_2, \patvec}{e}}
                      {A, (\vectype{\alpha}{A}), \Avec}
                      {C}
                      {\Delta, \alpha : \kindnat, \Theta}
          }
          {
            \matchjudg{\OK}{p}{\Gamma}{\branch{(\veccons{\pat_1}{\pat_2}), \patvec}{e}}{(\vectype{t}{A}), \Avec}{C}{\Delta}
          }
  \end{mathpar}

  \medskip

  \judgbox{\matchelimjudg{p}{\Gamma}{\Pi}{P}{\Avec}{C}{\Delta}}
       {Under context $\Gamma$, incorporate proposition $P$ while checking branches $\Pi$
         \\
         with patterns of type $\Avec$ and bodies of type $C$}
  \begin{mathpar}
     \Infer{\MatchBot}
           {\elimeq{\Gamma}{\sigma}{\tau}{\sort}{\bot}}
           {\matchelimjudg{p}{\Gamma}{\branch{\patvec}{e}}{\sigma = \tau}{\Avec}{C}{\Gamma}}
     \and
     \Infer{\MatchUnify}
           {\elimeq{\Gamma, \MonnierComma{P}}{\sigma}{\tau}{\sort}{\Theta} \\
            \matchjudg{!}{p}{\Theta}{\branch{\patvec}{e}}{\Avec}{C}{\Delta, \MonnierComma{P}, \Delta'}
           }
           {\matchelimjudg{p}{\Gamma}{\branch{\patvec}{e}}{\sigma = \tau}{\Avec}{C}{\Delta}}
  \end{mathpar}

  \caption{Algorithmic pattern matching}
  \FLabel{fig:alg-pattern-matching}
\end{figure}

%% file: fig-alg-match-coverage.tex
\begin{figure}[t]
  \raggedright
  \judgbox{\arrayenvcl{
    \covers{q}{\Gamma}{\Pi}{\Avec}
    \\
    \coverseq{\Gamma}{P}{\Pi}{\Avec}
    \\
    \guarded{\Pi}
  }
  }
     {Under context $\Gamma$, patterns $\Pi$ cover the types $\Avec$ \\[2pt]
      Under context $\Gamma$, patterns $\Pi$ cover the types $\Avec$ assuming $P$
      \\[3pt]
      Pattern list $\Pi$ contains a list pattern constructor at the head position}
  \begin{mathpar}
    \Infer{\CoversEmpty}
          { }
          {\covers{q}{\Gamma}{(\branch{\cdot}{e_1}) \alt \Pi}{\cdot}}
    \and
    \Infer{\CoversVar}
          {\expandvar{\Pi}{\Pi'} \\
           \covers{q}{\Gamma}{\Pi'}{\Avec}
          }
          { \covers{q}{\Gamma}{\Pi}{A, \Avec}
          }
    \and
    \Infer{\CoversUnit}
          { \expandunit{\Pi}{\Pi'}
            \\
            \covers{q}{\Gamma}{\Pi'}{\Avec}
          }
          { \covers{q}{\Gamma}{\Pi}{\unitty, \Avec}
          }
    \and
    \Infer{\CoversTimes}
          { \expandpair{\Pi}{\Pi'} \\
            \covers{q}{\Gamma}{\Pi'}{A_1, A_2, \Avec} 
          }
          { \covers{q}{\Gamma}{\Pi}{(A_1 \times A_2), \Avec}
          }
    \and
    \Infer{\CoversSum}
          { \expandsum{\Pi}{\Pi_L}{\Pi_R}
            \\
            \covers{q}{\Gamma}{\Pi_L}{A_1, \Avec}
            \\
            \covers{q}{\Gamma}{\Pi_R}{A_2, \Avec} 
          }
          { \covers{q}{\Gamma}{\Pi}{(A_1 + A_2), \Avec}
          }
    \and
    \Infer{\CoversEx}
          {
           \covers{q}{\Gamma, \alpha : \sort}{\Pi}{\Avec}
          }
          { \covers{q}{\Gamma}{\Pi}{(\extype{\alpha:\sort}{A}), \Avec}
          }
    \and 
    \Infer{\CoversWith}
          { \coverseq{\Gamma}{t_1 = t_2}{\Pi}{A_0, \Avec} }
          { \covers{\p}{\Gamma}{\Pi}{\big(A_0 \with (t_1 = t_2)\big), \Avec} 
          }
    \and 
    \Infer{\CoversWithOK}
          { \covers{\Gamma}{\Pi}{A_0, \Avec} }
          { \covers{\OK}{\Gamma}{\Pi}{\big(A_0 \with (t_1 = t_2)\big), \Avec} 
          }
    \and
    \Infer{\CoversVec}
          {
           \guarded{\Pi} \and 
            \expandvec{\Pi}{\Pi_{[]}}{\Pi_{::}}
            \\
            \arrayenvbl{
              \coverseq{\Gamma}{t = \zero}{\Pi_{[]}}{\Avec}
              \\
              \coverseq{\Gamma, n:\ind}
                    {t = \succ{n}}
                    {\Pi_{::}}
                    {(A, \vectype{n}{A}, \Avec)}
            }
          }
          { \covers{\p}{\Gamma}{\Pi}{\vectype{t}{A}, \Avec} }
    \and
    \Infer{\CoversVecOK}
          {
           \guarded{\Pi} \and 
           \expandvec{\Pi}{\Pi_{[]}}{\Pi_{::}}
            \\
            \arrayenvbl{
              \covers{\OK}{\Gamma}{\Pi_{[]}}{\Avec}
              \\
              \covers{\OK}
                     {\Gamma, n:\ind}
                     {\Pi_{::}}
                     {(A, \vectype{n}{A}, \Avec)}
            }
          }
          { \covers{\OK}{\Gamma}{\Pi}{\vectype{t}{A}, \Avec} }
    \and
    \Infer{\CoversEq}
          { \elimeq{\Gamma}{[\Gamma]t_1}{[\Gamma]t_2}{\sort}{\Delta}
            \\
            \covers{q}{\Delta}{[\Delta]\Pi}{[\Delta]\Avec}
          }
          { \coverseq{\Gamma}{t_1 = t_2}{\Pi}{\Avec} 
          }
    \and
    \Infer{\CoversEqBot}
          {\elimeq{\Gamma}{[\Gamma]t_1}{[\Gamma]t_2}{\sort}{\bot}}
          {\coverseq{\Gamma}{t_1 = t_2}{\Pi}{\Avec}}
          \\\vspace{1em}
  \Infer{ }
        { }
        {\guarded{\branch{\vecnil, \vec{p}}{e} \alt \Pi}}
  \and
  \Infer{ }
        { }
        {\guarded{\branch{\veccons{p}{p'}, \vec{p}}{e} \alt \Pi}}
  \and
  \Infer{ }
        {\guarded{\Pi} }
        {\guarded{\branch{\wild, \vec{p}}{e} \alt \Pi}}
  \and
  \Infer{ }
        {\guarded{\Pi} }
        {\guarded{\branch{x, \vec{p}}{e} \alt \Pi}}
  \end{mathpar}
  \caption{Algorithmic match coverage}
  \FLabel{fig:alg-match-coverage}
\end{figure}

%% file: fig-judgments.tex
\begin{tabular}{lll}
\emph{Judgment} & \emph{Description} & \emph{Location}
\\[2ex]
  $\judge{\Psi}{t}{\sort}$
  & Index term/monotype is well-formed
  & Figure \ref{fig:decl-wf}%
\\
  $\judgeprop{\Psi}{P}$
  & Proposition is well-formed
  & Figure \ref{fig:decl-wf}%
\\
  $\judgetp{\Psi}{A}$
  & Type is well-formed
  & Figure \ref{fig:decl-wf}%
\\
  $\judgetpvec{\Psi}{\Avec}$
  & Type vector is well-formed
  & Figure \ref{fig:decl-wf}%
\\
  $\judgectx{\Psi}$
  & Declarative context is well-formed
  & Figure \ref{fig:decl-wf}%
\\[1ex]
  $\declsubjudg[\polvar]{\Psi}{A}{B}$
  & Declarative subtyping
  & Figure \ref{mainfig:decl-subtyping}
\\[1ex]
  $\declcheckprop{\Psi}{P}$
  & Declarative truth
  & Figure \ref{mainfig:decl-typing}
\\[1ex]
  $\declchkjudg{p}{\Psi}{e}{A}$ &  Declarative checking
&  Figure \ref{mainfig:decl-typing} \\
  $\declsynjudg{p}{\Psi}{e}{A}$ & Declarative synthesis
&  Figure \ref{mainfig:decl-typing} \\
  $\declspinejudg{\Psi}{s}{A}{p}{C}{q}$ & Declarative spine typing
&  Figure \ref{mainfig:decl-typing} \\
  $\declrecspinejudg{\Psi}{s}{A}{p}{C}{q}$ & Declarative spine typing, recovering principality  &  Figure \ref{mainfig:decl-typing}
\\[1ex]
  $\declmatchjudg{p}{\Psi}{\Pi}{\Avec}{C}$
  & Declarative pattern matching
  & Figure \ref{mainfig:decl-pattern-matching}%
\\
  $\declmatchelimjudg{p}{\Psi}{\Pi}{P}{\Avec}{C}$
  & Declarative proposition assumption
  & Figure \ref{mainfig:decl-pattern-matching}%
\\[1ex]
  $\declcovers{\Psi}{\Pi}{\Avec}$
  &
  Declarative match coverage
  &
  Figure \ref{mainfig:decl-match-coverage}%
\\[1ex]
$\judge{\Gamma}{\tau}{\sort}$    & Index term/monotype is well-formed
& Figure \ref{fig:alg-wf}%
\\ 
$\judgeprop{\Gamma}{P}$          & Proposition is well-formed
& Figure \ref{fig:alg-wf}%
\\ 
$\judgetp{\Gamma}{A}$            & Polytype is well-formed
& Figure \ref{fig:alg-wf}%
\\ 
$\judgectx{\Gamma}$              & Algorithmic context is well-formed
& Figure \ref{fig:alg-wf}%
\\[1ex]
  $[\Gamma]A$
  &
  Applying a context, as a substitution, to a type
  &
  Figure \ref{mainfig:substitution}
\\[1ex]
$\checkprop{\Gamma}{P}{\Delta}$ & Check proposition
& Figure  \ref{fig:alg-prop}%
\\ 
$\elimprop{\Gamma}{P}{\Deltabot}$ & Assume proposition
& Figure  \ref{fig:alg-prop}%
\\ 
$\checkeq{\Gamma}{s}{t}{\sort}{\Delta}$ & Check equation
& Figure  \ref{fig:alg-checkeq}%
\\ 
$\clash{s}{t}$ & Head constructors clash
& Figure  \ref{fig:clash}%
\\ 
$\elimeq{\Gamma}{s}{t}{\sort}{\Deltabot}$ & Assume/eliminate equation
& Figure  \ref{fig:alg-elimeq}%
\\[1ex]
$\subjudg[\polvar]{\Gamma}{A}{B}{\Delta}$ & Algorithmic subtyping
& Figure \ref{fig:alg-subtyping}%
\\ 
$\addpropsubjudg{\Gamma}{P}{A}{B}{\Delta}$ & Assume/eliminate proposition
& Figure \ref{fig:alg-subtyping}%
\\
$\propequivjudg{\Gamma}{P}{Q}{\Delta}$ & Equivalence of propositions
& Figure \ref{fig:alg-subtyping}%
\\
$\equivjudg{\Gamma}{A}{B}{\Delta}$ & Equivalence of types
& Figure \ref{fig:alg-subtyping}%
\\
$\instjudg{\Gamma}{\ahat}{t}{\sort}{\Delta}$ & Instantiate
& Figure \ref{fig:instantiation}%
\\[1ex]
  $\chkintro{e}$
  &
  Checking intro form
  &
  Figure \ref{mainfig:chkintro}%
\\[1ex]
  $\chkjudg{p}{\Gamma}{e}{A}{\Delta}$ & Algorithmic checking
  & Figure \ref{mainfig:alg-typing}%
  \\ $\synjudg{p}{\Gamma}{e}{A}{\Delta}$ & Algorithmic synthesis
  & Figure \ref{mainfig:alg-typing}%
  \\ $\spinejudg{\Gamma}{s}{A}{p}{C}{q}{\Delta}$ & Algorithmic spine typing
  & Figure \ref{mainfig:alg-typing}%
  \\ $\recspinejudg{\Gamma}{s}{A}{p}{C}{q}{\Delta}$ & Algorithmic spine typing, recovering principality
  & Figure \ref{mainfig:alg-typing}%
\\[1ex]
  $\matchjudg{q}{p}{\Gamma}{\Pi}{\Avec}{C}{\Delta}$
  &
  Algorithmic pattern matching
  &
  Figure \ref{fig:alg-pattern-matching}%
\\
  $\matchelimjudg{p}{\Gamma}{\Pi}{P}{\Avec}{C}{\Delta}$
  &
  Algorithmic pattern matching (assumption)
  &
  Figure \ref{fig:alg-pattern-matching}%
\\[1ex]
  $\covers{q}{\Gamma}{\Pi}{\Avec}$
  &
  Algorithmic match coverage
  &
 Figure \ref{fig:alg-match-coverage}%
\\[1ex]
  $\substextend{\Gamma}{\Delta}$ & Context extension
& Figure \ref{mainfig:substextend}%
\\[1ex]
  $[\Omega]\Gamma$ & Apply complete context & Figure \ref{mainfig:context-substitution}
\end{tabular}

%% file: aaagadt.bbl

\begin{thebibliography}{40}

%%% the Web address.

\ifx \showCODEN    \undefined \def \showCODEN     #1{\unskip}     \fi
\ifx \showDOI      \undefined \def \showDOI       #1{#1}\fi
\ifx \showISBNx    \undefined \def \showISBNx     #1{\unskip}     \fi
\ifx \showISBNxiii \undefined \def \showISBNxiii  #1{\unskip}     \fi
\ifx \showISSN     \undefined \def \showISSN      #1{\unskip}     \fi
\ifx \showLCCN     \undefined \def \showLCCN      #1{\unskip}     \fi
\ifx \shownote     \undefined \def \shownote      #1{#1}          \fi
\ifx \showarticletitle \undefined \def \showarticletitle #1{#1}   \fi
\ifx \showURL      \undefined \def \showURL       {\relax}        \fi
% invisible to TeX
\providecommand\bibfield[2]{#2}
\providecommand\bibinfo[2]{#2}
\providecommand\natexlab[1]{#1}
\providecommand\showeprint[2][]{arXiv:#2}

\bibitem[\protect\citeauthoryear{Abel}{Abel}{2006}]%
        {Abel06}
\bibfield{author}{\bibinfo{person}{Andreas Abel}.}
  \bibinfo{year}{2006}\natexlab{}.
\newblock \showarticletitle{Towards Generic Programming with Sized Types}. In
  \bibinfo{booktitle}{\emph{Mathematics of Program Construction}}
  \emph{(\bibinfo{series}{LNCS})}, \bibfield{editor}{\bibinfo{person}{Tarmo
  Uustalu}} (Ed.), Vol.~\bibinfo{volume}{4014}. \bibinfo{publisher}{Springer},
  \bibinfo{pages}{10--28}.
\newblock


\bibitem[\protect\citeauthoryear{Abel, Coquand, and Dybjer}{Abel
  et~al\mbox{.}}{2008}]%
        {Abel08:MPC}
\bibfield{author}{\bibinfo{person}{Andreas Abel}, \bibinfo{person}{Thierry
  Coquand}, {and} \bibinfo{person}{Peter Dybjer}.}
  \bibinfo{year}{2008}\natexlab{}.
\newblock \showarticletitle{Verifying a Semantic $\beta\eta$-Conversion Test
  for {Martin-L\"of} Type Theory}. In \bibinfo{booktitle}{\emph{Mathematics of
  Program Construction (MPC'08)}} \emph{(\bibinfo{series}{LNCS})},
  Vol.~\bibinfo{volume}{5133}. \bibinfo{publisher}{Springer},
  \bibinfo{pages}{29--56}.
\newblock


\bibitem[\protect\citeauthoryear{Barendregt, Coppo, and
  Dezani-Ciancaglini}{Barendregt et~al\mbox{.}}{1983}]%
        {Barendregt83}
\bibfield{author}{\bibinfo{person}{Henk Barendregt}, \bibinfo{person}{Mario
  Coppo}, {and} \bibinfo{person}{Mariangiola Dezani-Ciancaglini}.}
  \bibinfo{year}{1983}\natexlab{}.
\newblock \showarticletitle{A Filter Lambda Model and the Completeness of Type
  Assignment}.
\newblock \bibinfo{journal}{\emph{J. Symbolic Logic}} \bibinfo{volume}{48},
  \bibinfo{number}{4} (\bibinfo{year}{1983}), \bibinfo{pages}{931--940}.
\newblock


\bibitem[\protect\citeauthoryear{Blume}{Blume}{2001}]%
        {Blume01entcs:NLFFI}
\bibfield{author}{\bibinfo{person}{Matthias Blume}.}
  \bibinfo{year}{2001}\natexlab{}.
\newblock \showarticletitle{No-Longer-Foreign: Teaching an {ML} compiler to
  speak {C} ``natively''}.
\newblock \bibinfo{journal}{\emph{Electronic Notes in Theoretical Computer
  Science}} \bibinfo{volume}{59}, \bibinfo{number}{1} (\bibinfo{year}{2001}).
\newblock


\bibitem[\protect\citeauthoryear{Cheney and Hinze}{Cheney and Hinze}{2003}]%
        {Cheney03:FirstClassPhantom}
\bibfield{author}{\bibinfo{person}{James Cheney} {and} \bibinfo{person}{Ralf
  Hinze}.} \bibinfo{year}{2003}\natexlab{}.
\newblock \bibinfo{booktitle}{\emph{First-Class Phantom Types}}.
\newblock \bibinfo{type}{{T}echnical {R}eport} CUCIS TR2003-1901.
  \bibinfo{institution}{Cornell University}.
\newblock


\bibitem[\protect\citeauthoryear{Chrz{\k a}szcz}{Chrz{\k a}szcz}{1998}]%
        {Chrzaszcz98}
\bibfield{author}{\bibinfo{person}{Jacek Chrz{\k a}szcz}.}
  \bibinfo{year}{1998}\natexlab{}.
\newblock \showarticletitle{Polymorphic Subtyping Without Distributivity}. In
  \bibinfo{booktitle}{\emph{Mathematical Foundations of Computer Science}}
  \emph{(\bibinfo{series}{LNCS})}, Vol.~\bibinfo{volume}{1450}.
  \bibinfo{publisher}{Springer}, \bibinfo{pages}{346--355}.
\newblock


\bibitem[\protect\citeauthoryear{Coquand}{Coquand}{1996}]%
        {Coquand96:typechecking-dependent-types}
\bibfield{author}{\bibinfo{person}{Thierry Coquand}.}
  \bibinfo{year}{1996}\natexlab{}.
\newblock \showarticletitle{An Algorithm for Type-Checking Dependent Types}.
\newblock \bibinfo{journal}{\emph{Science of Computer Programming}}
  \bibinfo{volume}{26}, \bibinfo{number}{1--3} (\bibinfo{year}{1996}),
  \bibinfo{pages}{167--177}.
\newblock


\bibitem[\protect\citeauthoryear{Damas and Milner}{Damas and Milner}{1982}]%
        {Damas82}
\bibfield{author}{\bibinfo{person}{Luis Damas} {and} \bibinfo{person}{Robin
  Milner}.} \bibinfo{year}{1982}\natexlab{}.
\newblock \showarticletitle{Principal type-schemes for functional programs}. In
  \bibinfo{booktitle}{\emph{POPL}}. \bibinfo{publisher}{ACM Press},
  \bibinfo{pages}{207--212}.
\newblock


\bibitem[\protect\citeauthoryear{Davies and Pfenning}{Davies and
  Pfenning}{2000}]%
        {Davies00icfpIntersectionEffects}
\bibfield{author}{\bibinfo{person}{Rowan Davies} {and} \bibinfo{person}{Frank
  Pfenning}.} \bibinfo{year}{2000}\natexlab{}.
\newblock \showarticletitle{Intersection Types and Computational Effects}. In
  \bibinfo{booktitle}{\emph{ICFP}}. \bibinfo{publisher}{ACM Press},
  \bibinfo{pages}{198--208}.
\newblock


\bibitem[\protect\citeauthoryear{Dolan and Mycroft}{Dolan and Mycroft}{2017}]%
        {Dolan17}
\bibfield{author}{\bibinfo{person}{Stephen Dolan} {and} \bibinfo{person}{Alan
  Mycroft}.} \bibinfo{year}{2017}\natexlab{}.
\newblock \showarticletitle{Polymorphism, Subtyping, and Type Inference in
  {MLsub}}. In \bibinfo{booktitle}{\emph{POPL}}. \bibinfo{publisher}{ACM
  Press}, \bibinfo{pages}{60--72}.
\newblock


\bibitem[\protect\citeauthoryear{Dunfield}{Dunfield}{2007a}]%
        {Dunfield07:Stardust}
\bibfield{author}{\bibinfo{person}{Jana Dunfield}.}
  \bibinfo{year}{2007}\natexlab{a}.
\newblock \showarticletitle{Refined typechecking with {Stardust}}. In
  \bibinfo{booktitle}{\emph{Programming Languages meets Programming
  Verification (PLPV '07)}}. \bibinfo{publisher}{ACM Press},
  \bibinfo{pages}{21--32}.
\newblock


\bibitem[\protect\citeauthoryear{Dunfield}{Dunfield}{2007b}]%
        {DunfieldThesis}
\bibfield{author}{\bibinfo{person}{Jana Dunfield}.}
  \bibinfo{year}{2007}\natexlab{b}.
\newblock \emph{\bibinfo{title}{A Unified System of Type Refinements}}.
\newblock \bibinfo{thesistype}{Ph.D. Dissertation}. \bibinfo{school}{Carnegie
  Mellon University}.
\newblock
\newblock
\shownote{CMU-CS-07-129.}


\bibitem[\protect\citeauthoryear{Dunfield and Krishnaswami}{Dunfield and
  Krishnaswami}{2013}]%
        {Dunfield13}
\bibfield{author}{\bibinfo{person}{Jana Dunfield} {and}
  \bibinfo{person}{Neelakantan~R. Krishnaswami}.}
  \bibinfo{year}{2013}\natexlab{}.
\newblock \showarticletitle{Complete and Easy Bidirectional Typechecking for
  Higher-Rank Polymorphism}. In \bibinfo{booktitle}{\emph{ICFP}}.
  \bibinfo{publisher}{ACM Press}, \bibinfo{pages}{429--442}.
\newblock
\newblock
\shownote{\href{http://arxiv.org/abs/1306.6032}{arXiv:{\tt 1306.6032
  [cs.PL]}}.}


\bibitem[\protect\citeauthoryear{Dunfield and Pfenning}{Dunfield and
  Pfenning}{2003}]%
        {Dunfield03:IntersectionsUnionsCBV}
\bibfield{author}{\bibinfo{person}{Jana Dunfield} {and} \bibinfo{person}{Frank
  Pfenning}.} \bibinfo{year}{2003}\natexlab{}.
\newblock \showarticletitle{Type Assignment for Intersections and Unions in
  Call-by-Value Languages}. In \bibinfo{booktitle}{\emph{FoSSaCS}}.
  \bibinfo{publisher}{Springer}, \bibinfo{pages}{250--266}.
\newblock


\bibitem[\protect\citeauthoryear{Fluet and Pucella}{Fluet and Pucella}{2006}]%
        {Fluet06:Phantoms}
\bibfield{author}{\bibinfo{person}{Matthew Fluet} {and}
  \bibinfo{person}{Riccardo Pucella}.} \bibinfo{year}{2006}\natexlab{}.
\newblock \bibinfo{title}{Phantom types and subtyping}.
  (\bibinfo{year}{2006}).
\newblock
\newblock
\shownote{\href{http://arxiv.org/abs/cs.PL/0403034}{arXiv:{\tt cs/0403034
  [cs.PL]}}.}


\bibitem[\protect\citeauthoryear{Garrigue and {Le Normand}}{Garrigue and {Le
  Normand}}{2015}]%
        {le-normand}
\bibfield{author}{\bibinfo{person}{Jacques Garrigue} {and}
  \bibinfo{person}{Jacques {Le Normand}}.} \bibinfo{year}{2015}\natexlab{}.
\newblock \showarticletitle{GADTs and Exhaustiveness: Looking for the
  Impossible}. In \bibinfo{booktitle}{\emph{Proceedings {ML} Family / OCaml
  Users and Developers workshops, {ML} Family/OCaml 2015}}
  \emph{(\bibinfo{series}{EPTCS})}. \bibinfo{pages}{23--35}.
\newblock


\bibitem[\protect\citeauthoryear{Garrigue and R{\'e}my}{Garrigue and
  R{\'e}my}{2013}]%
        {Garrigue13}
\bibfield{author}{\bibinfo{person}{Jacques Garrigue} {and}
  \bibinfo{person}{Didier R{\'e}my}.} \bibinfo{year}{2013}\natexlab{}.
\newblock \showarticletitle{Ambivalent Types for Principal Type Inference with
  {GADTs}}. In \bibinfo{booktitle}{\emph{APLAS}}.
  \bibinfo{publisher}{Springer}, \bibinfo{pages}{257--272}.
\newblock


\bibitem[\protect\citeauthoryear{Girard}{Girard}{1992}]%
        {girard-equality}
\bibfield{author}{\bibinfo{person}{Jean-Yves Girard}.}
  \bibinfo{year}{1992}\natexlab{}.
\newblock \bibinfo{title}{A Fixpoint Theorem in Linear Logic}.
  (\bibinfo{year}{1992}).
\newblock
\newblock
\shownote{Post to Linear Logic mailing list,
  \url{http://www.seas.upenn.edu/~sweirich/types/archive/1992/msg00030.html}.}


\bibitem[\protect\citeauthoryear{Harper and Morrisett}{Harper and
  Morrisett}{1995}]%
        {Harper95:intensional}
\bibfield{author}{\bibinfo{person}{Robert Harper} {and} \bibinfo{person}{Greg
  Morrisett}.} \bibinfo{year}{1995}\natexlab{}.
\newblock \showarticletitle{Compiling polymorphism using intensional type
  analysis}. In \bibinfo{booktitle}{\emph{POPL}}. \bibinfo{publisher}{ACM
  Press}, \bibinfo{pages}{130--141}.
\newblock


\bibitem[\protect\citeauthoryear{Karachalias, Schrijvers, Vytiniotis, and
  {Peyton Jones}}{Karachalias et~al\mbox{.}}{2015}]%
        {Karachalias15}
\bibfield{author}{\bibinfo{person}{Georgios Karachalias}, \bibinfo{person}{Tom
  Schrijvers}, \bibinfo{person}{Dimitrios Vytiniotis}, {and}
  \bibinfo{person}{Simon {Peyton Jones}}.} \bibinfo{year}{2015}\natexlab{}.
\newblock \showarticletitle{{GADTs} Meet Their Match: pattern-matching warnings
  that account for {GADTs}, guards, and laziness}. In
  \bibinfo{booktitle}{\emph{ICFP}}. \bibinfo{publisher}{ACM Press},
  \bibinfo{pages}{424--436}.
\newblock


\bibitem[\protect\citeauthoryear{Krishnaswami}{Krishnaswami}{2009}]%
        {Krishnaswami09:pattern-matching}
\bibfield{author}{\bibinfo{person}{Neelakantan~R. Krishnaswami}.}
  \bibinfo{year}{2009}\natexlab{}.
\newblock \showarticletitle{Focusing on Pattern Matching}. In
  \bibinfo{booktitle}{\emph{POPL}}. \bibinfo{publisher}{ACM Press},
  \bibinfo{pages}{366--378}.
\newblock


\bibitem[\protect\citeauthoryear{L{\"a}ufer and Odersky}{L{\"a}ufer and
  Odersky}{1994}]%
        {Laufer94:existential-type-inference}
\bibfield{author}{\bibinfo{person}{Konstantin L{\"a}ufer} {and}
  \bibinfo{person}{Martin Odersky}.} \bibinfo{year}{1994}\natexlab{}.
\newblock \showarticletitle{Polymorphic type inference and abstract data
  types}.
\newblock \bibinfo{journal}{\emph{ACM Trans. Prog. Lang. Sys.}}
  \bibinfo{volume}{16}, \bibinfo{number}{5} (\bibinfo{year}{1994}),
  \bibinfo{pages}{1411--1430}.
\newblock


\bibitem[\protect\citeauthoryear{Leijen and Meijer}{Leijen and Meijer}{1999}]%
        {Leijen99:DSEC}
\bibfield{author}{\bibinfo{person}{Daan Leijen} {and} \bibinfo{person}{Erik
  Meijer}.} \bibinfo{year}{1999}\natexlab{}.
\newblock \showarticletitle{Domain specific embedded compilers}. In
  \bibinfo{booktitle}{\emph{{USENIX} Conf. Domain-Specific Languages ({DSL
  '99})}}. \bibinfo{publisher}{ACM Press}, \bibinfo{pages}{109--122}.
\newblock


\bibitem[\protect\citeauthoryear{Odersky, Zenger, and Zenger}{Odersky
  et~al\mbox{.}}{2001}]%
        {Odersky01:ColoredLocal}
\bibfield{author}{\bibinfo{person}{Martin Odersky}, \bibinfo{person}{Matthias
  Zenger}, {and} \bibinfo{person}{Christoph Zenger}.}
  \bibinfo{year}{2001}\natexlab{}.
\newblock \showarticletitle{Colored Local Type Inference}. In
  \bibinfo{booktitle}{\emph{POPL}}. \bibinfo{publisher}{ACM Press},
  \bibinfo{pages}{41--53}.
\newblock


\bibitem[\protect\citeauthoryear{{Peyton Jones}, Vytiniotis, Weirich, and
  Shields}{{Peyton Jones} et~al\mbox{.}}{2007}]%
        {PeytonJones07}
\bibfield{author}{\bibinfo{person}{Simon {Peyton Jones}},
  \bibinfo{person}{Dimitrios Vytiniotis}, \bibinfo{person}{Stephanie Weirich},
  {and} \bibinfo{person}{Mark Shields}.} \bibinfo{year}{2007}\natexlab{}.
\newblock \showarticletitle{Practical type inference for arbitrary-rank types}.
\newblock \bibinfo{journal}{\emph{J. Functional Programming}}
  \bibinfo{volume}{17}, \bibinfo{number}{1} (\bibinfo{year}{2007}),
  \bibinfo{pages}{1--82}.
\newblock


\bibitem[\protect\citeauthoryear{{Peyton Jones}, Vytiniotis, Weirich, and
  Washburn}{{Peyton Jones} et~al\mbox{.}}{2006}]%
        {PeytonJones06:GADTs}
\bibfield{author}{\bibinfo{person}{Simon {Peyton Jones}},
  \bibinfo{person}{Dimitrios Vytiniotis}, \bibinfo{person}{Stephanie Weirich},
  {and} \bibinfo{person}{Geoffrey Washburn}.} \bibinfo{year}{2006}\natexlab{}.
\newblock \showarticletitle{Simple unification-based type inference for
  {GADTs}}. In \bibinfo{booktitle}{\emph{ICFP}}. \bibinfo{publisher}{ACM
  Press}, \bibinfo{pages}{50--61}.
\newblock


\bibitem[\protect\citeauthoryear{Pientka}{Pientka}{2008}]%
        {Pientka08:POPL}
\bibfield{author}{\bibinfo{person}{Brigitte Pientka}.}
  \bibinfo{year}{2008}\natexlab{}.
\newblock \showarticletitle{A type-theoretic foundation for programming with
  higher-order abstract syntax and first-class substitutions}. In
  \bibinfo{booktitle}{\emph{POPL}}. \bibinfo{publisher}{ACM Press},
  \bibinfo{pages}{371--382}.
\newblock


\bibitem[\protect\citeauthoryear{Pierce}{Pierce}{2002}]%
        {Pierce02:TAPL}
\bibfield{author}{\bibinfo{person}{Benjamin~C. Pierce}.}
  \bibinfo{year}{2002}\natexlab{}.
\newblock \bibinfo{booktitle}{\emph{Types and Programming Languages}}.
\newblock \bibinfo{publisher}{MIT Press}.
\newblock


\bibitem[\protect\citeauthoryear{Pierce and Turner}{Pierce and Turner}{2000}]%
        {Pierce00}
\bibfield{author}{\bibinfo{person}{Benjamin~C. Pierce} {and}
  \bibinfo{person}{David~N. Turner}.} \bibinfo{year}{2000}\natexlab{}.
\newblock \showarticletitle{Local Type Inference}.
\newblock \bibinfo{journal}{\emph{ACM Trans. Prog. Lang. Sys.}}
  \bibinfo{volume}{22} (\bibinfo{year}{2000}), \bibinfo{pages}{1--44}.
\newblock


\bibitem[\protect\citeauthoryear{Pottier and R{\'e}gis-Gianas}{Pottier and
  R{\'e}gis-Gianas}{2006}]%
        {Pottier06:stratified}
\bibfield{author}{\bibinfo{person}{Fran{\c{c}}ois Pottier} {and}
  \bibinfo{person}{Yann R{\'e}gis-Gianas}.} \bibinfo{year}{2006}\natexlab{}.
\newblock \showarticletitle{Stratified type inference for generalized algebraic
  data types}. In \bibinfo{booktitle}{\emph{POPL}}. \bibinfo{publisher}{ACM
  Press}, \bibinfo{pages}{232--244}.
\newblock


\bibitem[\protect\citeauthoryear{Rondon, Kawaguchi, and Jhala}{Rondon
  et~al\mbox{.}}{2008}]%
        {Rondon08:LT}
\bibfield{author}{\bibinfo{person}{Patrick Rondon}, \bibinfo{person}{Ming
  Kawaguchi}, {and} \bibinfo{person}{Ranjit Jhala}.}
  \bibinfo{year}{2008}\natexlab{}.
\newblock \showarticletitle{Liquid types}. In \bibinfo{booktitle}{\emph{PLDI}}.
  \bibinfo{publisher}{ACM Press}, \bibinfo{pages}{159--169}.
\newblock


\bibitem[\protect\citeauthoryear{Schroeder-Heister}{Schroeder-Heister}{1994}]%
        {schroeder-heister-equality}
\bibfield{author}{\bibinfo{person}{Peter Schroeder-Heister}.}
  \bibinfo{year}{1994}\natexlab{}.
\newblock \showarticletitle{Definitional reflection and the completion}. In
  \bibinfo{booktitle}{\emph{Extensions of Logic Programming}}
  \emph{(\bibinfo{series}{LNCS})}. \bibinfo{publisher}{Springer},
  \bibinfo{pages}{333--347}.
\newblock


\bibitem[\protect\citeauthoryear{Simonet and Pottier}{Simonet and
  Pottier}{2007}]%
        {Simonet07:constraint}
\bibfield{author}{\bibinfo{person}{Vincent Simonet} {and}
  \bibinfo{person}{Fran{\c{c}}ois Pottier}.} \bibinfo{year}{2007}\natexlab{}.
\newblock \showarticletitle{A constraint-based approach to guarded algebraic
  data types}.
\newblock \bibinfo{journal}{\emph{ACM Transactions on Programming Languages and
  Systems (TOPLAS)}} \bibinfo{volume}{29}, \bibinfo{number}{1}
  (\bibinfo{year}{2007}), \bibinfo{pages}{1}.
\newblock


\bibitem[\protect\citeauthoryear{Tiuryn and Urzyczyn}{Tiuryn and
  Urzyczyn}{1996}]%
        {Tiuryn96}
\bibfield{author}{\bibinfo{person}{Jerzy Tiuryn} {and}
  \bibinfo{person}{Pawe{\l} Urzyczyn}.} \bibinfo{year}{1996}\natexlab{}.
\newblock \showarticletitle{The Subtyping Problem for Second-Order Types is
  Undecidable}. In \bibinfo{booktitle}{\emph{LICS}}. \bibinfo{publisher}{IEEE
  Press}.
\newblock


\bibitem[\protect\citeauthoryear{Vytiniotis, {Peyton Jones}, and
  Schrijvers}{Vytiniotis et~al\mbox{.}}{2010}]%
        {Vytiniotis10}
\bibfield{author}{\bibinfo{person}{Dimitrios Vytiniotis},
  \bibinfo{person}{Simon {Peyton Jones}}, {and} \bibinfo{person}{Tom
  Schrijvers}.} \bibinfo{year}{2010}\natexlab{}.
\newblock \showarticletitle{Let should not be generalised}. In
  \bibinfo{booktitle}{\emph{Workshop on Types in Language Design and Impl.
  (TLDI '10)}}. \bibinfo{publisher}{ACM Press}, \bibinfo{pages}{39--50}.
\newblock


\bibitem[\protect\citeauthoryear{Vytiniotis, {Peyton Jones}, Schrijvers, and
  Sulzmann}{Vytiniotis et~al\mbox{.}}{2011}]%
        {Vytiniotis11}
\bibfield{author}{\bibinfo{person}{Dimitrios Vytiniotis},
  \bibinfo{person}{Simon {Peyton Jones}}, \bibinfo{person}{Tom Schrijvers},
  {and} \bibinfo{person}{Martin Sulzmann}.} \bibinfo{year}{2011}\natexlab{}.
\newblock \showarticletitle{{OutsideIn(X)}: Modular type inference with local
  assumptions}.
\newblock \bibinfo{journal}{\emph{J. Functional Programming}}
  \bibinfo{volume}{21}, \bibinfo{number}{4--5} (\bibinfo{year}{2011}),
  \bibinfo{pages}{333--412}.
\newblock


\bibitem[\protect\citeauthoryear{Watkins, Cervesato, Pfenning, and
  Walker}{Watkins et~al\mbox{.}}{2004}]%
        {Watkins04}
\bibfield{author}{\bibinfo{person}{Kevin Watkins}, \bibinfo{person}{Iliano
  Cervesato}, \bibinfo{person}{Frank Pfenning}, {and} \bibinfo{person}{David
  Walker}.} \bibinfo{year}{2004}\natexlab{}.
\newblock \showarticletitle{A Concurrent Logical Framework: The Propositional
  Fragment}.
\newblock In \bibinfo{booktitle}{\emph{Types for Proofs and Programs}}.
  \bibinfo{publisher}{Springer LNCS 3085}, \bibinfo{pages}{355--377}.
\newblock


\bibitem[\protect\citeauthoryear{Xi, Chen, and Chen}{Xi et~al\mbox{.}}{2003}]%
        {Xi03:guarded}
\bibfield{author}{\bibinfo{person}{Hongwei Xi}, \bibinfo{person}{Chiyan Chen},
  {and} \bibinfo{person}{Gang Chen}.} \bibinfo{year}{2003}\natexlab{}.
\newblock \showarticletitle{Guarded recursive datatype constructors}. In
  \bibinfo{booktitle}{\emph{POPL}}. \bibinfo{publisher}{ACM Press},
  \bibinfo{pages}{224--235}.
\newblock


\bibitem[\protect\citeauthoryear{Xi and Pfenning}{Xi and Pfenning}{1999}]%
        {Xi99popl}
\bibfield{author}{\bibinfo{person}{Hongwei Xi} {and} \bibinfo{person}{Frank
  Pfenning}.} \bibinfo{year}{1999}\natexlab{}.
\newblock \showarticletitle{Dependent Types in Practical Programming}. In
  \bibinfo{booktitle}{\emph{POPL}}. \bibinfo{publisher}{ACM Press},
  \bibinfo{pages}{214--227}.
\newblock


\bibitem[\protect\citeauthoryear{Zeilberger}{Zeilberger}{2009}]%
        {Zeilberger09}
\bibfield{author}{\bibinfo{person}{Noam Zeilberger}.}
  \bibinfo{year}{2009}\natexlab{}.
\newblock \showarticletitle{Refinement Types and Computational Duality}. In
  \bibinfo{booktitle}{\emph{Programming Languages meets Programming
  Verification (PLPV '09)}}. \bibinfo{publisher}{ACM Press},
  \bibinfo{pages}{15--26}.
\newblock


\end{thebibliography}
